\documentclass[11pt]{article}

\usepackage{amsmath,amssymb,amstext,amsthm}
\usepackage{amsfonts}
\usepackage{graphicx,epstopdf,epsfig,multirow,epic,bm}
\usepackage{color}
\usepackage{multicol}
\usepackage{algorithm}
\usepackage{algorithmic}
\usepackage{paralist}

\theoremstyle{definition}
\newtheorem{thm}{Theorem}
\newtheorem{cor}[thm]{Corollary}
\newtheorem{lem}[thm]{Lemma}
\newtheorem{prop}[thm]{Proposition}
\theoremstyle{definition}
\newtheorem{defn}{Definition}
\theoremstyle{definition}
\newtheorem{rem}{Remark}

\theoremstyle{definition}
\newtheorem{problem}{Problem}

\theoremstyle{definition}
\newtheorem{conj}{Conjecture}
\theoremstyle{definition}
\newtheorem{example}{Example}
\theoremstyle{definition}

\DeclareGraphicsExtensions{.eps,.eps.gz}
\normalsize \rm
\usepackage{titlesec}

\newcommand{\qqed}{\hfill $\square$}
\newcommand{\paralled}{\hfill $\parallel$}

\usepackage[nottoc]{tocbibind} 

\begin{document}

\begin{center}
{\Large \textbf{Strings And Colorings Of Topological Coding Towards\\[8pt]
Asymmetric Topology Cryptography}}\\[14pt]
{\large \textbf{Bing YAO, Chao YANG, Xia LIU, Fei MA,\\[4pt]
Jing SU, Hui SUN, Xiaohui ZHANG and Yarong MU}\\[8pt]
}
(\today)
\end{center}

\vskip 1cm

\pagenumbering{roman}
\tableofcontents

\newpage

\setcounter{page}{1}
\pagenumbering{arabic}

\thispagestyle{empty}

\begin{center}
{\Large \textbf{Strings And Colorings Of Topological Coding Towards\\[8pt] Asymmetric Topology Cryptography}}\\[12pt]
{\large \textbf{Bing Yao $^1$, Chao Yang $^2$, Xia Liu $^3$, Fei Ma $^4$\\
Jing Su $^5$, Hui Sun $^6$, Xiaohui Zhang $^7$ and Yarong Mu $^8$}}\\[12pt]
\end{center}
{\small 1. College of Mathematics and Statistics, Northwest Normal University, Lanzhou, 730070, China, yybb918@163.com\\
2. School of Mathematics, Physics and Statistics, Shanghai University of Engineering Science,
Shanghai, 201620, China, jiayouyc1988@163.com\\
3. School of Mathematics and Statistics, Northwestern Polytechnical University, Xi'an, 710072, China, liuxia\_90@163.com\\
4. School of Electronics Engineering and Computer Science, Peking University, Beijing, 100871, China, mafei123987@163.com\\
5. College of Computing Science and Technology, Xi'an University of Science and Technology, Xi'an 710054, China,
1099270659@qq.com\\
6. School of Electronics Engineering and Computer Science, Peking University, Beijing, 100871, China,
18919104606@163.com\\
7. School of Computer, Qinghai Normal University, Xining, 810001, China, 2547851790@qq.com\\
8. Lanzhou College of Information Science and Technology, Lanzhou, 730314, China, 2275005907@qq.com}
\begin{center}
\today
\end{center}

\vskip 0.6cm

\begin{quote}
\textbf{Abstract:} We, for anti-quantum computing, will discuss various number-based strings, such as number-based super-strings, parameterized strings, set-based strings, graph-based strings, integer-partitioned and integer-decomposed strings, Hanzi-based strings, as well as algebraic operations based on number-based strings. Moreover, we introduce number-based string-colorings, magic-constraint colorings, and vector-colorings and set-colorings related with strings. For the technique of encrypting the entire network at once, we propose graphic lattices related with number-based strings, Hanzi-graphic lattices, string groups, all-tree-graphic lattices. We study some topics of asymmetric topology cryptography, such as topological signatures, Key-pair graphs, Key-pair strings, one-encryption one-time and self-certification algorithms. Part of topological techniques and algorithms introduced here are closely related with NP-complete problems or NP-hard problems.\\[6pt]
\end{quote}

\vskip 0.6cm

\section{Introduction}

Using cryptography is a kind of transformation technology to protect information, also is the most original ability of cryptography. However, with the development of information and information technology, modern cryptography is not only used to solve the confidentiality of information, but also used to solve the integrity, availability and controllability of information.

\subsection{Research background}

Asymmetric encryption and certificateless public key encryption are two important branches of cryptography. They are not only hot topics in cryptography research, but also widely used practical encryption techniques.

\textbf{Asymmetric encryption} was introduced by Diffie and Hellman in 1976 \cite{Diffie-and-Hellman-1976}, where each person gets a pair of keys, called the public-key and the private key. Each person's public-key is published while the private key is kept secret. Messages are encrypted using the intended recipient's public-key and can only be decrypted using his private key. This is often used in conjunction with a digital signature. Public-key encryption can be used for authentication, confidentiality, integrity and non-repudiation. RSA encryption is an example of the public-key cryptosystem. Based on the asymmetric encryption algorithm, Rivest, Shamir and Adleman proposed the \emph{public-key cryptography} (PKE) in 1978 \cite{Rivest-Shamir-Adleman-1978}.

Main algorithms for realizing Asymmetric Encryption Algorithm are RSA, Elgamal, Knapsack algorithm, Rabin, HD, ECC (elliptic curve encryption algorithm), DSA (digital signature).

Certificate Authority (CA or ``Trusted Third Party'') is an entity (typically a company) that issues digital certificates to other entities (organizations or individuals) to allow them to prove their identity to others. A Certificate Authority might be an external company such as VeriSign that offers digital certificate services or they might be an internal organizations such as a corporate MIS department. The Certificate Authority's chief function is to verify the identity of entities and issue digital certificates attesting to that identity. The process uses public-key cryptography to create a ``network of trust''. For example, if I want to prove my identity to you, I ask a CA (who you trust to have verified my identity) to encrypt a hash of my signed key with their private key. Then you can use the CA's public-key to decrypt the hash and compare it with a hash you calculate yourself. Hashes are used to decrease the amount of data that needs to be transmitted. The hash function must be cryptographically strong, e.g. MD5.

The main \textbf{advantages}of asymmetric encryption are: (i) Higher security, the public-key is public, the secret key is saved by the user, there is no need to give the private key to others. (ii) The key distribution is simple, and no secret channel or complex protocol is required to transmit the key. Public keys can be distributed to other users based on a public channel, such as a key distribution center, while private keys are kept by users themselves. (iii) Digital signature can be realized.

Some \textbf{disadvantages} of asymmetric encryption are: Compared with the symmetric cryptosystem, the public-key cryptosystem (PKE) with the same security strength requires more key bits, takes longer time to encrypt and decrypt, and is only suitable for encrypting a small amount of data.

\textbf{Certificateless Public Key Cryptography} was proposed by Al-riyami and Paterson in 2003 \cite{Al-Riyami-Paterson-1978}. The certificateless public-key system is no longer an identity based public-key system, because in the certificateless public-key system, the user's public-key is no longer the only identifiable identity of the user, but requires additional public-keys. However, the difference between the Certificateless Public Key System (CL-PKS) and the Certificate-based Public Key System (Cb-PKS) is that in CL-PKS, certificates are no longer required to bind the user's public-key and the user's identity, thus overcoming the certificate management problems in Cb-PKS. Therefore, CL-PKS is regarded as the intermediate product of Cb-PKS and the identity-based public-key system. Some disadvantages of CL-PKC are:

(i) In the certificateless public-key system, since the certification authority is no longer required to generate a certificate for the user to bind the user's public-key and its identity, that is, the user's public-key may be replaced, so the replacement public-key attack must be considered.

(ii) Most efficient certificateless public-key encryption schemes require bilinear pairing on elliptic curves. Although the implementation technology of bilinear pairing has been greatly improved recently, the calculation cost of bilinear pairing operation is still higher than that of ``standard'' modular exponential operation.

Signcryption combines public key encryption with digital signature, which can encrypt and sign messages at the same time in a logical step. The certificate multiple acceptance signcryption scheme was first studied in \cite{Selvi-Vivek-Shukla-Chandrasekaran-2008}. The authors \cite{Sayid-Isaac-Sayid-Jayaprakash-2016} survey and analysis some of the well-known certificateless schemes, and present the generic model of Certificateless Public Key Encryption scheme proposed by comparisons of the certificateless schemes performance and security.

The authors \cite{Zhang-Sun-Zhang-Geng-Li-2011} revisits, analyzes, compares, and briefly reviews some of the main results. Furthermore, they study discusses some existing problems in this research field that deserve further investigation. In particular, in view of the threat of various key leakage attacks, including side-channel attacks, to the security of cryptosystems, and the challenge of modern cryptography by new computing technologies such as quantum computing, we believe that, The research on certificate-free cipher system that is able to resist key leakage attack and secure in quantum computing will be a very meaningful and challenging new topic \cite{Zhang-Sun-Zhang-Geng-Li-2011}.

\vskip 0.4cm

In the coming \textbf{Quantum Computer Era}, we will be facing with huge information security challenges: The Shor algorithm can completely destroy the encryption mechanism based on RSA and elliptic curve cryptography as long as the quantum computer has enough logical qubits to perform operations. As known, Shor algorithm can effectively attack RSA, EIGamal, ECC public-key cryptography and DH key agreement protocols which are widely used at present. This indicates that RSA, EIGamal, ECC public-key cryptography and DH key agreement protocols will no longer be secure in the quantum computing environment. There is also an algorithm called Grover, which can completely weaken AES encryption from 128 bits to 64 bits, and then it can be cracked by ordinary computer algorithms.

In 2016 the National Institute of Standards and Technology has initiated a standardization process for post-quantum cryptosystems. Such cryptosystems are usually based on NP-complete problems for two reasons: NP-complete problems are at least as hard as the hardest problems in NP, but solutions of such problems can be verified efficiently. The main candidates for post-quantum cryptography are:
\begin{asparaenum}[\textrm{\textbf{Reotost}}-1. ]
\item \textbf{Code-based cryptography} is based on the NP-complete problem of decoding a random linear code.
\item \textbf{Lattice-based cryptography} is based on Conjectured security against quantum attacks; Algorithmic simplicity, efficiency, and parallelism; Strong security guarantees from worst-case hardness; NP-complete problems of finding the shortest vector.
\item \textbf{Multivariate cryptography} is based on the NP-complete problem of solving multivariate quadratic equations defined over some finite field.
\item \textbf{Isogeny-based cryptography} is based on finding the isogeny map between two super-singular elliptic curves.
\end{asparaenum}

\vskip 0.4cm

In this article, part of anti-quantum computing foundations of our techniques mentioned here are based on the following works: the problems of integer partition and integer factorization; total graph-colorings; structures of various number-based strings including various number-based strings, integer partitioned and integer decomposed strings, number-based strings generated from indexed-colorings, graphic strings, vector-colorings, set-colorings, graph-based colorings; additive string group of prime order, tree-base graphic lattices; Hanzi-based strings, Hanzi-graphic lattices, topological signatures, TKPDRA-center for networks and self-certification topology algorithms. We will discussion \emph{asymmetric topology encryption} towards the certificateless public key cryptography in the last section.

\subsection{Examples from topological coding}

\subsubsection{Top-EN-DECRYPTION algorithm-I}

\begin{example}\label{exa:top-en-decryption-algorithm-I}
\textbf{The encryption and decryption from the technology of topological coding.} We use the colored graphs $G$, $T$, $J$ and $K_6$ shown in Fig.\ref{fig:introduction-example-11} to explain the application of topological coding, which is a branch of mathematics. ``Topological encryption and decryption algorithm'' is abbreviated by ``Top-EN-DECRYPTION algorithm''.

\vskip 0.4cm

\textbf{Top-EN-DECRYPTION algorithm-I.}

\textbf{Initialization. }A file $d_{oc}$ is encrypted by a number-based string $s_{pub}=135244214255666$ with 15 bytes obtained from the Topcode-matrix $T_{code}(G)$ of the private-key $G$; refer to Fig.\ref{fig:introduction-example-11} (a) and Eq.(\ref{eqa:example-topcode-matrix11}). The encrypted file is denoted as $d^*_{oc}$. \textbf{Do}:

\textbf{Step 1. }Find two \emph{private-key graphs} $T$ and $J$ for a given \emph{public-key graph} $G$.

\textbf{Step 2. }Find a number-based string $s_{pri}=421111235254463$ from the Topcode-matrix $T_{code}(T)$ of the private-key graph $T$; and find another number-based string $s\,'_{pri}=122331311325346$ from the Topcode-matrix $T_{code}(J)$ of the private-key graph $J$; refer to Fig.\ref{fig:introduction-example-11} (b) and (c), and Eq.(\ref{eqa:example-topcode-matrix11}).

\textbf{Step 3. }Complete the \emph{topological authentication} $K_6=G[\odot ]T[\odot ]J$ shown in Fig.\ref{fig:introduction-example-11} (d) by vertex-coinciding the vertices colored the same color in three colored graphs $G$, $T$ and $J$.

\textbf{Step 4. }Decrypt the encrypted file $d^*_{oc}$ by two number-based strings $s_{pri}$ and $s\,'_{pri}$, and get the original file $d_{oc}$.

\vskip 0.4cm

Analysis of the above Top-EN-DECRYPTION algorithm-I.

(i) Finding two \emph{private-key graphs} $T$ and $J$ from $6^{6-2}=1,296$ different colored spanning trees of $K_6$ is not easy, since we will meet the Subgraph Isomorphic Problem, which is a NP-complete problem; refer to the Cayley's formula $\tau(K_n)=n^{n-2}$ in graph theory \cite{Bondy-2008}.

(ii) The topological authentication $K_6=G[\odot ]T[\odot ]J$ holds four edge sets $E(K_6)=E(G)\cup E(T)\cup E(J)$ with $E(G)\cap E(T)=\emptyset $, $E(G)\cap E(J)=\emptyset $ and $E(T)\cap E(J)=\emptyset $ true, since each complete graph $K_{2m}$ with $m\geq 2$ can be vertex-split into $m$ mutually edge-disjoint spanning trees of $2m-1$ edges.

(iii) Each of three Topcode-matrices $T_{code}(G)$, $T_{code}(T)$ and $T_{code}(J)$ can induces $(15)!$ different number-based strings like $s_{pub}$, $s_{pri}$ and $s\,'_{pri}$. So, if the colored graphs $G$, $T$ and $J$ have thousands of edges, it is a terrible job for computer to chose two number-based strings $s_{pri}$ and $s\,'_{pri}$ for decrypting the file encrypted by the number-based string $s_{pub}$.

Summarizing the above analysis, in general, we will \textbf{do}:

(1) \textbf{Find} $m$ mutually edge-disjoint spanning trees $T_1,T_2,\dots ,T_m$ of $2m-1$ edges from the $(2m)^{2m-2}$ colored spanning trees of a complete graph $K_{2m}$ to hold a topological authentication $E(K_{2m})=\bigcup ^m_{i=1}E(T_i)$;

(2) \textbf{find} $m$ number-based strings $s_{k}$ generated from the Topcode-matrices $T_{code}(T_k)$ with $k\in [1,m]$ to complete the encryption and decryption to the digital files, however, each number-based strings $s_{k}$ has to be found from $(6m-3)!$ different number-based strings induced by each Topcode-matrix $T_{code}(T_k)$.\qqed
\end{example}

The following three Topcode-matrices $T_{code}(G)$, $T_{code}(T)$ and $T_{code}(J)$ obtained from three colored graphs $G,T,J$ shown in Fig.\ref{fig:introduction-example-11}.

{\small
\begin{equation}\label{eqa:example-topcode-matrix11}
\centering
{
\begin{split}
T_{code}(G)= \left(
\begin{array}{cccccc}
1 & 3 & 5 & 2 & 4\\
4 & 2 & 1 & 4 & 2\\
5 & 5 & 6 & 6 & 6
\end{array}
\right)
\end{split}}, {
\begin{split}
T_{code}(T)= \left(
\begin{array}{cccccc}
4 & 2 & 1 & 1 & 1\\
1 & 2 & 3 & 5 & 2\\
5 & 4 & 4 & 6 & 3
\end{array}
\right)
\end{split}},
{
\begin{split}
T_{code}(J)= \left(
\begin{array}{cccccc}
1 & 2 & 2 & 3 & 3\\
1 & 3 & 1 & 1 & 3\\
2 & 5 & 3 & 4 & 6
\end{array}
\right)
\end{split}}
\end{equation}
}

\begin{figure}[h]
\centering
\includegraphics[width=16.4cm]{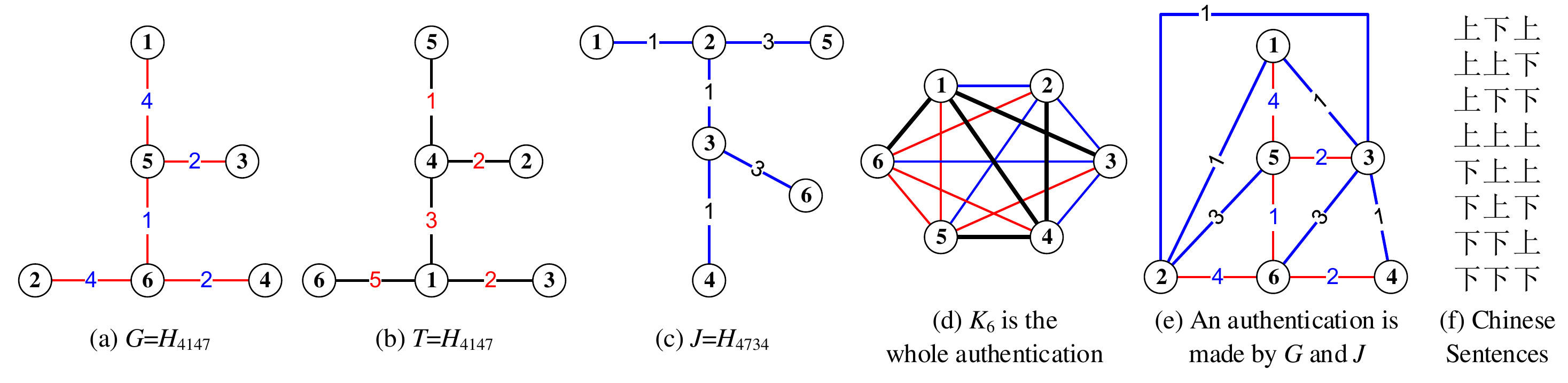}\\
\caption{\label{fig:introduction-example-11}{\small A diagram for Example \ref{exa:top-en-decryption-algorithm-I} and Example \ref{exa:top-en-decryption-algorithm-II}.}}
\end{figure}

\subsubsection{Top-EN-DECRYPTION algorithm-II}

\begin{example}\label{exa:top-en-decryption-algorithm-II}
In Example \ref{exa:top-en-decryption-algorithm-I}, we do not add more constraints to the topological encryption and decryption. Some constraints for the \textbf{Top-EN-DECRYPTION algorithm-II} are as follows:
\begin{asparaenum}[\textbf{\textrm{Cons}}-1.]
\item Decrypt the encrypted file by a number-based string generated from the Topcode-matrix $T_{code}(G[\odot ]T[\odot ]J)$ shown in Eq.(\ref{eqa:example-topcode-matrix33}), it is noticeable, the Topcode-matrix $T_{code}(G[\odot ]T[\odot ]J)$ can produces $(45)!$ number-based strings, which include all permutations of the number-based strings $s_{pub}$, $s_{pri}$ and $s\,'_{pri}$, in total.
\item Each number of $\{1,2,3,4,5,6\}$ in the Topcode-matrices $T_{code}(G)$, $T_{code}(T)$ and $T_{code}(J)$ can be evaluated by other numbers, for example, we use the \emph{assignment symbol} ``$:=$'' to $1:=142857$, $2:=6174$, $3:=0618$, $4:=31415926$, $5:=8128$, $6:=196$, then we get the \emph{assignment string} $s_{pub}:=s^*_{pub}$ as follows
\begin{equation}\label{eqa:evaluated-replace11}
{
\begin{split}
135244214255666=s_{pub}:=&s^*_{pub}\\
=&142857~0618~8128~6174~31415926~31415926~6174~\\
&142857~31415926~6174~8128~8128~196~196~196
\end{split}}
\end{equation} Or the numbers of $\{1,2,3,4,5,6\}$ can be replaced by something else for increasing the cost of code-breakers.
\item The topological authentication is one of eight Chinese Sentences shown in Fig.\ref{fig:introduction-example-11} (f), such that Chinese speakers and writers can use Chinese characters instead of number-based strings with long codes in real application.
\item There are many choices for making \emph{public-key graphs} and \emph{private-key graphs}, such as one-vs-one (see Fig.\ref{fig:introduction-example-11} (e) and Eq.(\ref{eqa:example-topcode-matrix22})), one-vs-more (see Top-EN-DECRYPTION algorithm-I), and more-vs-more.
\item Replace by a connected graph $G$ the complete graph $K_{2m}$, and decompose $G$ into mutually edge-disjoint trees $T_1,T_2,\dots ,T_m$, such that $E(G)=\bigcup^m_{k=1}E(T_k)$, since trees admit many interesting labelings and colorings for making more complex number-based strings.\qqed
\end{asparaenum}
\end{example}

\begin{equation}\label{eqa:example-topcode-matrix22}
\centering
{
\begin{split}
T_{code}(G[\odot ]J)= \left(
\begin{array}{cccccccccccc}
1 & 3 & 5 & 2 & 4 & 1 & 2 & 2 & 3 & 3\\
4 & 2 & 1 & 4 & 2 & 1 & 3 & 1 & 1 & 3\\
5 & 5 & 6 & 6 & 6 & 2 & 5 & 3 & 4 & 6
\end{array}
\right)
\end{split}}
\end{equation}

\begin{equation}\label{eqa:example-topcode-matrix33}
\centering
{
\begin{split}
T_{code}(G[\odot ]T[\odot ]J)= \left(
\begin{array}{cccccccccccccccccc}
1 & 3 & 5 & 2 & 4 & 4 & 2 & 1 & 1 & 1 & 1 & 2 & 2 & 3 & 3\\
4 & 2 & 1 & 4 & 2 & 1 & 2 & 3 & 5 & 2 & 1 & 3 & 1 & 1 & 3\\
5 & 5 & 6 & 6 & 6 & 5 & 4 & 4 & 6 & 3 & 2 & 5 & 3 & 4 & 6
\end{array}
\right)
\end{split}}
\end{equation}

\subsubsection{The PRONBS-problem}

\begin{example}\label{exa:8888888888}
If the number-based string $s^*_{pub}$ shown in Eq.(\ref{eqa:evaluated-replace11}) is as a \emph{public-key string}, then the breakers have to rewrite $s^*_{pub}$ as $s^*_{pub}=c_{1}c_{1}\cdots c_{15}$, where $c_{1}=$142857, $c_{2}=$0618, $c_{3}=$8128, $c_{4}=$6174, $c_{5}=$31415926, $c_{6}=$31415926, $c_{7}=$6174, $c_{8}=$142857, $c_{9}=$31415926, $c_{10}=$6174, $c_{11}=$8128, $c_{12}=$8128, $c_{13}=$196, $c_{14}=$196 and $c_{15}=$196, and they use these 15 sub-number-based strings to form correctly a Topcode-matrix $T_{code}(s^*)$ as follows
\begin{equation}\label{eqa:example-topcode-matrix55}
\centering
{
\begin{split}
T_{code}(s^*)= \left(
\begin{array}{cccccc}
142857 & 0618 & 8128 & 6174 & 31415926\\
31415926 & 6174 & 142857 & 31415926 & 6174\\
8128 & 8128 & 196 & 196 & 196
\end{array}
\right)
\end{split}}
\end{equation} so we have the assignment relationship $T_{code}(G):=T_{code}(s^*)$ for expressing the relationship between two Topcode-matrices $T_{code}(G)$ and $T_{code}(s^*)$.

Next, the breakers use this Topcode-matrix $T_{code}(s^*)$ to find the colored graph $G$ shown in Fig.\ref{fig:introduction-example-11} (a) and its Topcode-matrix $T_{code}(G)$ shown in Eq.(\ref{eqa:example-topcode-matrix11}). Finally, the breakers use the colored graph $G$ and the Topcode-matrix $T_{code}(G)$ to find other two colored graphs (as \emph{private-key graphs}) $T$ and $J$ shown in Fig.\ref{fig:introduction-example-11} (b) and (c), as well as two Topcode-matrices $T_{code}(T)$ and $T_{code}(J)$ to export two number-based strings $s_{pri}$ and $s\,'_{pri}$ to decrypt the encrypted file.

This problem is just the so-called PRONBS-problem.\qqed
\end{example}

\begin{problem}\label{question:PRONBS-problems00}
\cite{Bing-Yao-arXiv:2207-03381} We present a problem, called PRONBS-problem, for the parametric reconstitution of number-based strings as follows:
\begin{quote}
\textbf{PRONBS-problem.} For a given number-based string $s=c_1c_2\cdots c_n$ with $c_i\in [0,9]=\{0,1,\dots ,9\}$, rewrite this number-based string $s$ into $3q$ segments $a_1,a_2,\dots ,a_{3q}$ with $a_j=b_{j,1}b_{j,2}\cdots b_{j,m_j}$ for $j\in [1,3q]=\{1,2,\dots ,3q\}$, such that each number $c_i$ of the number-based string $s$ is in a segment $a_j$ but $c_i$ is not in other segment $a_s$ if $s\neq j$. Find two integers $k_0,d_0\geq 0$, such that each segment $a_j$ with $j\in [1,3q]$ can be expressed as $a_j=\beta_jk_0+\gamma_jd_0$ for integers $\beta_j,\gamma_j\geq 0$ and $j\in [1,3q]$; and moreover \textbf{find} a bipartite $(p,q)$-graph $H$ admitting a \emph{parameterized coloring} $F$ holding its own parameterized Topcode-matrix
\begin{equation}\label{eqa:PRONBS-problem-parameterized-matrix}
P_{(k,d)}(H,F)_{3\times q}=k\cdot I\,^0_{3\times q}+d\cdot T_{code}(H,f)_{3\times q}
\end{equation} where $f$ is a set-ordered $W$-constraint coloring of the bipartite $(p,q)$-graph $H$, as well as the parameterized number-based string $a_1a_2\dots a_{3q}=c\,'_1c\,'_2\cdots c\,'_n$ is just generated from the parameterized Topcode-matrix $P_{(k,d)}(H,F)_{3\times q}$ when $(k,d)=(k_0,d_0)$, where $c\,'_1,c\,'_2,\dots ,c\,'_n$ is a permutation of $c_1,c_2,\dots ,c_n$; refer to Definition \ref{defn:bipartite-parameterized-topcode-matrix}.
\end{quote}
\end{problem}

\begin{example}\label{exa:8888888888}
A bipartite $(5,4)$-graph $H$ colored with a set-ordered coloring $f$ admits a \emph{parameterized coloring} $F$ defined by the following parameterized Topcode-matrix
\begin{equation}\label{eqa:parameterized-topcode-matrix11}
\centering
{
\begin{split}
P_{(k,d)}(H,F)_{3\times 5}=&k\cdot I\,^0_{3\times 5}+d\cdot T_{code}(H,f)_{3\times 5}\\
=&k\cdot \left(
\begin{array}{cccccc}
0 & 0 & 0 & 0 & 0\\
1 & 1 & 1 & 1 & 1\\
1 & 1 & 1 & 1 & 1
\end{array}
\right)+d\cdot \left(
\begin{array}{cccccc}
1 & 2 & 3 & 3 & 3\\
5 & 4 & 3 & 2 & 1\\
6 & 6 & 6 & 5 & 4
\end{array}
\right)\\
=&\left(
\begin{array}{rrrrrr}
d & 2d & 3d & 3d & 3d\\
k+5d & k+4d & k+3d & k+2d & k+d\\
k+6d & k+6d & k+6d & k+5d & k+4d
\end{array}
\right)
\end{split}}
\end{equation} where $I\,^0_{3\times 5}$ is the \emph{unit Topcode-matrix} of order $3\times 5$; refer to Definition \ref{defn:bipartite-parameterized-topcode-matrix}.

Clearly, the above parameterized Topcode-matrix $P_{(k,d)}(H,F)_{3\times 5}$ can distributes us infinite different number-based strings.\qqed
\end{example}

Part of complexities of the PRONBS-problem are as follows:
\begin{asparaenum}[\textrm{\textbf{Pronbs}}-1. ]
\item \textbf{Rewrite} the string $s$ into $3q$ segments $a_1,a_2,\dots ,a_{3q}$ with $a_j=b_{j,1}b_{j,2}\cdots b_{j,m_j}$ for $j\in [1,3q]$, such that each number $c_i$ of the number-based string $s$ is in a segment $a_j$ but $c_i$ is not in other segment $a_s$ if $s\neq j$. In this work, we will meet the \textbf{String Partition Problem}, notice that there are different colored graphs, which induce the same number-based strings.

\item \textbf{Find} two integers $k_0,d_0\geq 0$, such that each segment $a_j$ with $j\in [1,3q]$ can be expressed as $a_j=\beta_jk_0+\gamma_jd_0$ for some integers $\beta_j,\gamma_j\geq 0$ and $j\in [1,3q]$. In this work, we will meet the \textbf{Indefinite Equation Problem}.

\item \textbf{Find} a colored $(p,q)$-graph $H$ admitting a $W$-constraint coloring $f$, such that $H$ admits a parameterized coloring $F$ defined by Eq.(\ref{eqa:PRONBS-problem-parameterized-matrix}). The work of \textbf{finding} the bipartite $(p,q)$-graph $H$ from a huge amount of graphs will meet the \textbf{Subgraph Isomorphic Problem}, which is a NP-complete problem; refer to two numbers $G_{23}$ and $G_{24}$ shown in Eq.(\ref{eqa:number-graphs-23-24-vertices}).
\item \textbf{Finding} the desired $W$-constraint coloring $f$ is related with the \textbf{Indefinite Equation Problem}.
\item Let $G_p$ be the number of graphs of $p$ vertices, Harary and Palmer \cite{Harary-Palmer-1973} computed two graph numbers
\begin{equation}\label{eqa:number-graphs-23-24-vertices}
{
\begin{split}
G_{23}&=559946939699792080597976380819462179812276348458981632\approx 2^{179}\\
G_{24}&=195704906302078447922174862416726256004122075267063365754368\approx 2^{197}
\end{split}}
\end{equation}
\quad So, \textbf{finding} a particular graph from \textbf{hundreds of graphs} with \textbf{hundreds of vertices and edges} is a terrible computational job for supercomputers, or quantum computers.
\item Another difficult problem is to find the particular $W$-constraint coloring $f$ admitted by the graph $H$ from \textbf{thousands of colorings and labelings}, and other difficult problem is that there are many ways for coloring the graph $H$ by the particular $W$-constraint coloring, however, the $W$-constraint coloring $f$ is just one of them.
\item \textbf{Use} this parameterized Topcode-matrix $P_{(k,d)}(H,F)$ to produce just the desired number-based string $a_1a_2\dots a_{3q}=c\,'_1c\,'_2\cdots c\,'_n$ when $(k,d)=(k_0,d_0)$ determined by the solution $k_0,d_0$ of the above indefinite equations $a_j=\beta_jk_0+\gamma_jd_0$ for $j\in [1,3q]$.
\end{asparaenum}

\subsection{Preliminary}

\subsubsection{Terminology and notation}

The articles \cite{Bondy-2008}, \cite{Gallian2021} and \cite{Yao-Wang-2106-15254v1} provide standard terminology, notation, labelings and colorings used here. For the sake of statement, we relist the following terminology and notation:

\begin{asparaenum}[$\bullet$ ]
\item $Z$ is the set of integers, $Z^0$ is the set of non-negative integers
\item An integer set $\{a, a+1, \dots, b\}$ with integers $a, b$ holding $0\leq a<b$ is rewritten as a short notation $[a,b]$, and $[m, n]^o$ denotes an odd-set $\{m, m+2, \dots, n\}$ with odd integers $m, n$ respect to $1\leq m< n$.
\item The symbol $[c,d]^r$ is a set of real numbers for two real numbers $c,d$.
\item $|X|$ is the number of elements of a set $X$.
\item The set of all non-empty subsets of an integer set $[a,b]$ is denoted as $[a,b]^2$, called \emph{power set}. An integer set $[1,4]=\{1,2,3,4\}$ has its own power set $[1,4]^2=\big \{\{1\},\{2\},\{3\}$, $\{4\}$, $\{1,2\}$, $\{1,3\},\{1,4\}$, $\{2,3\}$, $\{2,4\}$, $\{3,4\}$, $\{1,2,3\}$, $\{1,2,4\}$, $\{1,3,4\}$, $\{2,3,4\}$, $\{1,2,3,4\}\big \}$, so the number of subsets is $2^4-1=15$ in total.
\item Let integers $r,k\geq 0$ and $s,d\geq 1$, we have two parameterized sets
\begin{equation}\label{eqa:two-parameterized-sets11}
S_{s,k,r,d}=\{k+rd,k+(r+1)d,\dots ,k+(r+s)d\}
\end{equation} and
\begin{equation}\label{eqa:two-parameterized-sets22}
O_{s,k,r,d}=\{k+[2(r+1)-1]d,k+[2(r+2)-1]d,\dots ,k+[2(r+s)-1]d\}
\end{equation} so $|S_{s,k,r,d}|=s+1$ and $|O_{s,k,r,d}|=s$.

\item The notation $N_{ei}(u)$ is the \emph{neighbor set} of vertices adjacent with a vertex $u$, the number $\deg_G(u)=|N_{ei}(u)|$ is called the \emph{degree} of the vertex $u$. The maximum degree $\Delta(G)$ and the minimum degree $\delta(G)$ of a graph $G$ are defined as follows:
$$\Delta(G)=\max \{\deg_G(u):u\in V(G)\},~\delta(G)=\min \{\deg_G(u):u\in V(G)\}
$$
\item A \emph{leaf} is a vertex $x$ having its degree $\textrm{deg}(x)=1=|N_{ei}(x)|$.
\item A $(p,q)$-graph $G$ having $p$ vertices and $q$ edges, such that its own vertex set $V(G)$ and edge set $E(G)$ satisfy $|V(G)|=p$ and $|E(G)|=q$, respectively.
\item A \emph{tree} has no cycle and any pair of two vertices of the tree is joined by a unique path. A \emph{caterpillar} $T$ is a tree, such that the \emph{leaf-removed graph} $T-L(T)$ is just a path $P=x_1x_2\cdots x_n$, where $L(T)$ is the set of all leaves of $T$. A \emph{lobster} $H$ is a tree, such that leaf-removed graph $H-L(H)$ is just a \emph{caterpillar}.
\item A \emph{bipartite} $(p,q)$-graph $G$ has its own vertex $V(G)=X\cup Y$ with $X\cap Y=\emptyset$, and $x\in X$ and $y\in Y$ for each edge $xy\in E(G)$, also, $X$ and $Y$ are \emph{independent sets} of the graph $G$.
\item Each complete graph contained in a graph $G$ is called \emph{clique}, the symbol $K(G)$ is the maximum vertex number of cliques of $G$.
\end{asparaenum}

\subsubsection{Pan-colorings and pan-strings}

Let $M$ be a set of integers, and let a graph $G$ admit a mapping $f:S\subset V(G)\cup E(G)\rightarrow M$, write the set of colors assigned to the elements of $S$ by $f(S)=\{f(x):x\in S\}$. There are the following basic colorings/labelings in graph theory:
\begin{asparaenum}[\textrm{A}-1.]
\item As $S=V(G)$, $f(S)=f(V(G))$, we call $f$ \emph{vertex coloring}.
\item As $S=E(G)$, $f(S)=f(E(G))$, we call $f$ \emph{edge coloring}.
\item As $S=V(G)\cup E(G)$, $f(S)=f(V(G)\cup E(G))=f(V(G))\cup f(E(G))$, we call $f$ \emph{total coloring}.
\end{asparaenum}
\begin{asparaenum}[\textrm{B}-1. ]
\item As $f$ is a \emph{vertex coloring} and $f(u)\neq f(v)$ for each edge $uv\in E(G)$, we call $f$ \emph{proper vertex coloring}. The number
$$\chi(G)=\min_f\max\{f(w):w\in V(G)\}$$ over all proper vertex colorings of $G$ is called \emph{chromatic number}.
\item As $f$ is an \emph{edge coloring} and $f(uv)\neq f(uw)$ for any two adjacent edges $uv$ and $uw$ of $G$, we call $f$ \emph{proper edge coloring}. The number
$$\chi\,'(G)=\min_f\max\{f(w):w\in E(G)\}$$ over all proper edge colorings of $G$ is called \emph{chromatic index}.
\item As $f$ is a \emph{total coloring} holding $f(u)\neq f(v)$ for each edge $uv$ and $f(uv)\neq f(uw)$ for any two adjacent edges $uv,uw$, we call $f$ \emph{proper total coloring}. The number
$$\chi\,''(G)=\min_f\max\{f(w):w\in V(G)\cup E(G)\}$$ over all proper total colorings of $G$ is called \emph{total chromatic number}.
\item As $f$ is a vertex coloring (resp. an edge coloring, or a total coloring) holding $f(u)\neq f(x)$ for any two vertices $u$ and $x$ (resp. for any two edges, or for any two elements $u,x\in V(G)\cup E(G)$), that is, $|f(V(G))|=|V(G)|$ (resp. $|f(E(G))|=|E(G)|$, or $|f(V(G)\cup E(G))|=|V(G)|+|E(G)|$), we call $f$ \emph{labeling}.
\item If a coloring/labeling $f$ satisfies a mathematical constraint or a group of mathematical constraints, we say $f$ to be a \emph{$W$-constraint coloring/labeling}.
\end{asparaenum}

\vskip 0.4cm

For the sake of definiteness we define:
\begin{asparaenum}[$\ast$ ]
\item A \emph{$n$-rank number-based string} $s=a_1a_2\cdots a_n$ with each number $a_i\in Z^0$, however, we often omit ``$n$-rank'' in discussion.
\item A \emph{$n$-rank $[0,9]$-string} $s=b_1b_2\cdots b_n$ with each number $b_i\in [0,9]$, however, we often omit the word ``$n$-rank'' in discussion.
\item A \emph{Tm-$[0,9]$-string} $s=b_1b_2\cdots b_n$ is made by a Topcode-matrix and holds each number $b_i\in [0,9]$ true; refer to Definition \ref{defn:topcode-matrix-definition}, Definition \ref{defn:colored-topcode-matrix}, and Definition \ref{defn:generalization-colored-topcode-matrix}.
\item A \emph{text-based string} $s_T=a_1a_2\cdots a_m$ has at least one $a_j$ being a letter but number.
\item An \emph{English-text-based string} $s_T=e_1e_2\cdots e_m$ has each $e_j$ to be an \emph{English letter}.
\item A \emph{Hanzi-based string} $s_T=H_{a_1b_1c_1d_1}H_{a_2b_2c_2d_2}\cdots H_{a_mb_mc_md_m}$ has each $H_{a_jb_jc_jd_j}$ to be a \emph{Chinese letter} appeared in \cite{GB2312-80}; see Fig.\ref{fig:Chinese-text-based-string}.
\item A \emph{graph-based string} $s_G=G_1G_2\cdots G_m$ has each $G_j$ to be a \emph{graph}.
\item A \emph{set-based string} $s_S=X_1X_2\cdots X_m$ has each $X_j$ to be a \emph{set}.
\item A \emph{$W$-matrix-based string} $s_M=M_1M_2\cdots M_m$ has each $M_j$ to be a \emph{$W$-type matrix}.
\item A \emph{thing-based string} $s_t=T_1T_2\cdots T_m$ has each $T_j$ to be a \emph{thing} with a particular property or a group of particular properties.\qqed
\end{asparaenum}

\begin{figure}[h]
\centering
\includegraphics[width=13.4cm]{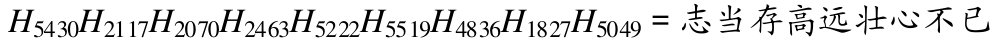}\\
\caption{\label{fig:Chinese-text-based-string}{\small A Hanzi-based string $H_{5430}H_{2117}H_{2070}H_{2463}H_{5222}H_{5519}H_{4836}H_{1827}H_{5049}$.}}
\end{figure}

\begin{rem}\label{rem:333333}
A \emph{number-based string} $s=a_1a_2\cdots a_n$ can correspond to an \emph{integer vector} $(a_1,a_2,\dots $, $a_n)$ in the $n$-dimension space, or an \emph{integer set} $\{a_1,a_2,\dots ,a_n\}$ (no order), or a \emph{vertex-degree} $\textrm{deg}=(a_1,a_2,\dots ,a_n)$ (no order) in graph theory.\paralled
\end{rem}

\begin{defn} \label{defn:all-thing-coloringss}
$^*$ A graph $G$ admits a coloring $F:S\rightarrow U$, where $S\subseteq V(G)\cup E(G)$, such that $F$ holds a group of $W_1$-constraint, $W_2$-constraint, $\dots$, $W_n$-constraint, denoted as $\{W_i\}^n_{i=1}$-constraint with $n\geq 1$. Then
\begin{asparaenum}[\textbf{\textrm{Col}}-1.]
\item $F$ is called \emph{$\{W_i\}^n_{i=1}$-constraint string-coloring} if $U$ is a set of strings.
\item $F$ is called \emph{$\{W_i\}^n_{i=1}$-constraint coloring-coloring} if $U$ is a set of colorings.
\item $F$ is called \emph{$\{W_i\}^n_{i=1}$-constraint set-coloring} if $U$ is a set of sets.
\item $F$ is called \emph{$\{W_i\}^n_{i=1}$-constraint vector-coloring} if $U$ is a set of vectors.
\item $F$ is called \emph{$\{W_i\}^n_{i=1}$-constraint matrix-coloring} if $U$ is a set of matrices.
\item $F$ is called \emph{$\{W_i\}^n_{i=1}$-constraint graph-coloring} if $U$ is a set of graphs.
\item $F$ is called \emph{$\{W_i\}^n_{i=1}$-constraint group-coloring} if $U$ is an every-zero graphic group.
\item $F$ is called \emph{$\{W_i\}^n_{i=1}$-constraint string-group coloring} if $U$ is an every-zero string group.
\item $F$ is called \emph{$\{W_i\}^n_{i=1}$-constraint graphic-group coloring} if $U$ is an every-zero graphic group.
\item $F$ is called \emph{$\{W_i\}^n_{i=1}$-constraint string-lattice coloring} if $U$ is a string lattice.
\item $F$ is called \emph{$\{W_i\}^n_{i=1}$-constraint graphic-lattice coloring} if $U$ is a graphic lattice.
\item $F$ is called \emph{$\{W_i\}^n_{i=1}$-constraint thing-coloring} if $U$ is a set of things having a particular property or a group of particular properties.\qqed
\end{asparaenum}
\end{defn}

\subsubsection{Graph operations}

The sentence ``\emph{adding a leaf $w$ to a graph $G$}'' is an graph operation defined by adding a new vertex $w$ to a graph $G$, and join $w$ with a vertex $x$ of the graph $G$ by an edge $xw$, the resultant graph is denoted as $G+xw$, such that $w$ is a vertex of degree one of the \emph{leaf-added graph} $G+xw$.

Let $G$ and $H$ be two graphs. If $H=G+xy-uv$ for edge $uv\in E(G)$ and edge $xy\not\in E(G)$, then we call $H$ \emph{$\pm e$-dual} of $G$, also, an \emph{adding-edge-subtracting graph homomorphism} $H\rightarrow _{\pm e}G$.

Let $G_i$ and $G_j$ be two disjoint connected graphs in the following argument, we present the following graph operations:

(1) \textbf{Edge-joining operation.} We use a new edge to join a vertex $u$ of $G_i$ with one vertex $x$ of $G_j$ together, so we get a connected graph $G_i[\ominus ]G_j$, called an \emph{edge-joined graph} based on the \emph{graph edge-joining operation} ``$[\ominus]$'', see an example shown in Fig.\ref{fig:22-vertex-join-split-coincide} from (a) to (c).

\begin{figure}[h]
\centering
\includegraphics[width=16cm]{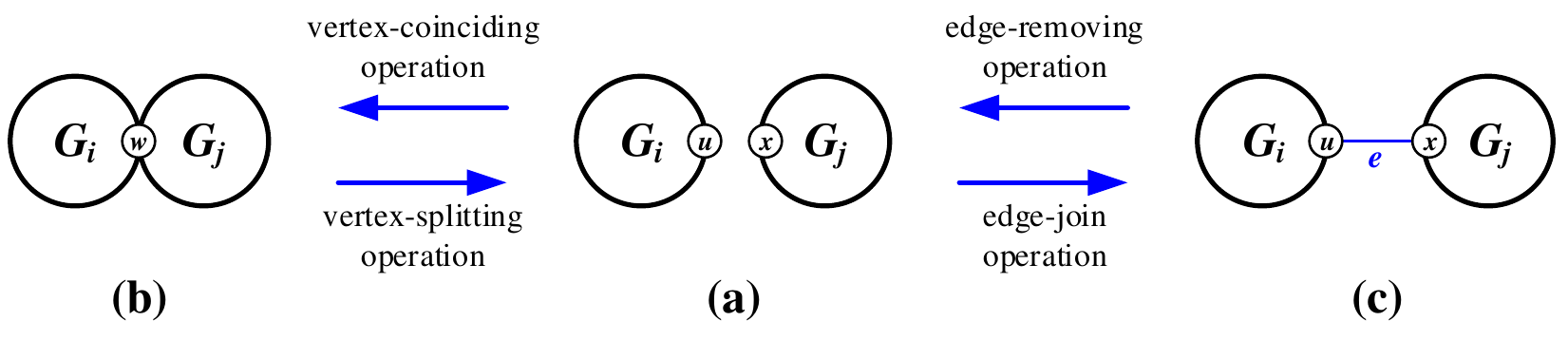}
\caption{\label{fig:22-vertex-join-split-coincide}{\small A diagram for illustrating the vertex-coinciding operation and the edge-joining operation, cited from \cite{Bing-Yao-arXiv:2207-03381}.}}
\end{figure}

(2) \textbf{Vertex-coinciding operation.} We vertex-coincide a vertex $u$ of $G_i$ with one vertex $x$ of $G_j$ into one vertex $w=u\odot x$, the resulting graph is denoted as $G_i[\odot] G_j$, called a \emph{vertex-coincided graph} based on the \emph{graph vertex-coinciding operation} ``$[\odot]$'', such that the vertex $w\in V(G_i[\odot]G_j)$ has the degree $\textrm{deg}_{G_i}(u)+\textrm{deg}_{G_j}(x)$, see an illustration in Fig.\ref{fig:22-vertex-join-split-coincide} from (a) to (b).

(3) \textbf{Vertex-splitting tree-operation. }A \emph{vertex-splitting operation} on a connected graph is defined as: Let $x$ be a vertex of a cycle $C$ of a connected graph $G$, and the neighbor set $N_{ei}(x)=\{x_1,x_2,\dots, x_d\}$ with $d=\textrm{deg}_G(x)\geq 2$ be the neighbor set of the vertex $x$. We vertex-split $x$ into two vertices $x\,'$ and $x\,''$, and let $x\,'$ to join with each of vertices $x_{i_1},x_{i_2},\dots, x_{i_k}$ by edges, and let $x\,''$ to join with each of vertices $x_{i_{k+1}},x_{i_{k+2}},\dots, x_{i_d}$ by edges, the resultant graph is a connected graph, denoted as $G\wedge x$ and called \emph{vertex-split graph}. See an example shown in Fig.\ref{fig:22-vertex-join-split-coincide} from (b) to (a). Clearly, the number of cycles of the vertex-split graph $G\wedge x$ is less than that of the original connected graph $G$.

Fig.\ref{fig:33-from-graph-to-tree} shows us a connected $(5,8)$-graph $G$ that can be vertex-split into two trees $G_4$ and $T_4$ of $8$ edges, we get graph homomorphisms $G_i\rightarrow G$ and $T_i\rightarrow G$ for $i\in [1,4]$.

If $G_1=G\wedge x$ has cycles, we do the vertex-splitting operation to $G_1$. Go on in this way, we get a graph $G_m=G_{m-1}\wedge y$ to be a tree $T$ for some integer $m\geq 1$. The process of vertex-splitting $G$ into a tree $T$ is called the \emph{vertex-splitting tree-operation}, and write the tree $T=\wedge_m(G)$ and the graph $G=\wedge^{-1}_m(T)$, or $G=[\odot_m]T$.

\begin{thm}\label{thm:connected-graph-vs-tree-set}
\cite{Bing-Yao-arXiv:2207-03381} Each non-tree connected $(p,q)$-graph $G$ corresponds to a set $V_{\textrm{split}}(G)$ of trees of $q+1$ vertices, such that $G=\wedge^{-1}_m(T)$, also, $G=[\odot_m]T$ for each tree $T\in V_{\textrm{split}}(G)$.
\end{thm}

\begin{figure}[h]
\centering
\includegraphics[width=15cm]{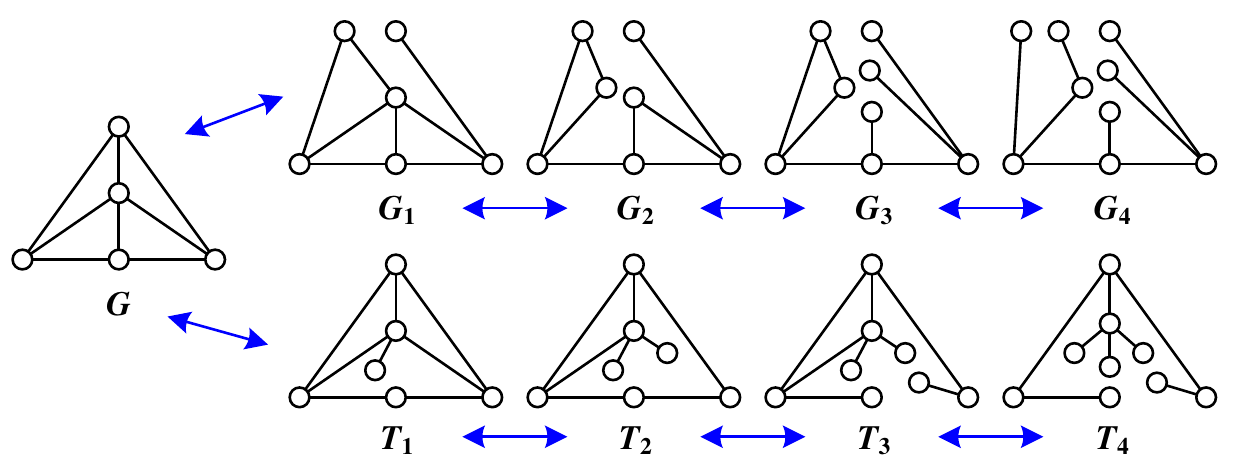}\\
\caption{\label{fig:33-from-graph-to-tree} {\small A connected $(5,8)$-graph $G$ can be vertex-split into trees $G_4$ and $T_4$, cited from \cite{Bing-Yao-arXiv:2207-03381}.}}
\end{figure}

\begin{thm}\label{thm:2-vertex-split-graphs-isomorphic}
\cite{Yao-Su-Sun-Wang-Graph-Operations-2021, Wang-Su-Yao-divided-2020} Suppose that two connected graphs $G$ and $H$ admit a one-one mapping $f:V(G)\rightarrow V(H)$. In general, a vertex-split graph $G\wedge u$ with $\textrm{deg}_G(u)\geq 2$ is not unique, so we have a vertex-split graph set $S_G(u)=\{G\wedge u\}$, similarly, we have another vertex-split graph set $S_H(f(u))=\{H\wedge f(u)\}$. If each vertex-split graph $L\in S_G(u)$ corresponds to another vertex-split graph $T\in S_H(f(u))$ such that $L\cong T$, and vice versa, we write this fact as $G\wedge u\cong H\wedge f(u)$ for each vertex $u\in V(G)$, then we claim that $G$ is \emph{isomorphic} to $H$, that is $G\cong H$.
\end{thm}

\begin{rem}\label{rem:333333}
It does not look easy to prove $G\wedge u\cong H\wedge f(u)$ appeared in Theorem \ref{thm:2-vertex-split-graphs-isomorphic}, although each split-graph $G\wedge u$ maintains all edges of $G$ (accordingly, each split-graph $H\wedge f(u)$ maintains all edges of $H$); refer to Kelly-Ulam's Reconstruction Conjecture proposed in 1942. For determining the vertex-split graph set $S_G(u)=\{G\wedge u\}$ (resp. the vertex-split graph set $S_H(f(u))=\{H\wedge f(u)\}$), we will do the following works:

\textbf{Step 1.} A vertex $u$ of the graph $G$ is vertex-split into $a_i$ vertices $x_{i,1},x_{i,2},\dots ,x_{i,a_i}$.

\textbf{Step 2.} The vertex degree $\textrm{deg}_G(u)$ is partitioned into $\textrm{deg}_G(u)=m_{i,1}+m_{i,2}+\cdots +m_{i,a_i}$ with $m_{i,j}\geq 1$ and $j\in [1,a_i]$.

\textbf{Step 3.} The neighbor set $N_{ei}(u)$ of the vertex $u$ is cut into $a_i$ mutually disjoint subsets $N_{i,1},N_{i,2},\dots ,N_{i,a_i}$ holding $N_{ei}(u)=\bigcup^{a_i}_{j=1} N_{i,j}$ and $m_{i,j}=|N_{i,j}|$ with $j\in [1,a_i]$, such that each split-vertex $x_{i,j}$ is joined with each vertex of $N_{i,j}$ with $j\in [1,a_i]$.

However, we will meet the \textbf{Integer Partition Problem} in Step 2, and the \textbf{Topological Structure-split Problem} in Step 3.\paralled
\end{rem}

\begin{problem}\label{question:444444}
\cite{Yao-Xiaomin-Wang-Ma-Wang-IAEAC-2021} If each spanning tree $T_a$ of a connected $(p, q)$-graph $G_a$ corresponds a spanning tree $T_b$ of another connected
graph $G_b$ such that $T_a\cong T_b$, and vice versa, \textbf{can} we claim $G_a \cong G_b$?
\end{problem}

\begin{problem}\label{qeu:444444}
Vertex-splitting (also, decompose) a connected graph into a group of mutually edge-disjoint trees $H_1,H_2,\dots, H_M$, \textbf{determine} the smallest number $M$.
\end{problem}

\begin{problem}\label{qeu:verious-tree-decompositions}
We vertex-splitting (also, decompose) a connected graph $G$ into mutually edge-disjoint trees $T_1,T_2,\dots, T_m$ with $m\geq 2$, and we put them into a set $T_{ree}(G)$, so $E(G)=\bigcup ^n_{i=1}E(T_i)$. The set $N_{isom}(G)$ collects the non-isomorphic trees $T_{i_1},T_{i_2},\dots, T_{i_k}$ in $T_{ree}(G)$, then we have the following cases for \textbf{characterizing} the graph $G$:

(i) If $k=1$, that is $N_{isom}(G)=\{T\}$ such that $T_{i_j}=T$ for $j\in [1,k]$ and $T_i=T$ for $i\in [1,m]$, then we say that $G$ is \emph{$T$-tree decomposable}, and write the graph $G=[\odot] mT$.

(ii) If $k\geq 2$ and $k$ is the smallest number for the set $N_{isom}(G)$, then the mutually edge-disjoint trees in $T_{ree}(G)$ can be classified into sets $S_{j}(T_{i_j})$ for $j\in [1,k]$, such that any pair of trees $H\,',H\,''\in S_{j}(T_{i_j})$ holds $H\,'\cong H\,''$. So, we rewrite the graph $G$ as $G=[\odot]^k_{j=1}a_jT_{i_j}$, where $a_j=|S_{j}(T_{i_j})|$ for $j\in [1,k]$.

(iii) Let $n_d(G)$ be the number of vertices of degree $d$ in a graph $G$.
\begin{equation}\label{eqa:555555}
\sum _{3\leq d\leq \Delta (T_{i_j})}(d-2)n_d(T_{i_j})=\sum _{3\leq d\leq \Delta (T_{i_s})}(d-2)n_d(T_{i_s}),~j,s\in [1,k]
\end{equation}

(iii) If each tree of $N_{isom}(G)\subseteq T_{ree}(G)$ is a spanning tree of the graph $G$, we call $G$ to be \emph{spanning-tree decomposable}.

(iv) If each tree of $N_{isom}(G)\subseteq T_{ree}(G)$ is a spanning caterpillar (resp. lobster) of the graph $G$, we say $G$ to be \emph{spanning-caterpillar decomposable}, or \emph{spanning-caterpillar pure} (resp. \emph{spanning-lobster decomposable}, or \emph{spanning-lobster pure}).
\end{problem}

\begin{problem} \label{thm:newproblems1}
\cite{Yao-Zhang-Yao-Spanning-2007, Yao-Zhang-Jian-fang-Wang-2010} \textbf{ST-balance set}. Let $G$ be a connected graph, and let $L(G)$ be the set of all leaves of $G$. A subset $S$ of $V(G)$ is called a \textbf{spanning tree balance set} of the connected graph $G$ if for any two spanning trees $T$ and $T\,'$ of the connected graph $G$ the following identity
\begin{equation}\label{eqa:chapter2-controlset00}
n_1(T)-|L(T)\cap S|=n_1(T\,')-|L(T\,')\cap S|
\end{equation} holds true, where $n_1(H)$ is the number of leaves of a tree $H$. \textbf{Determine} a smallest ST-balance set $S^*$ of the connected graph $G$, such that $|S^*|\leq |S|$ for any ST-balance set $S$ of the connected graph $G$.
\end{problem}

\begin{thm}\label{thm:K-2m-spanning-trees-2m-1-edges}
$^*$ For complete graphs, we have

\textbf{Result-1.} Each complete graph $K_{2m}$ with $m\geq 2$ can be vertex-split into $m$ mutually edge-disjoint spanning trees of $2m-1$ edges.

\textbf{Result-2.} Each complete graph $K_{2m+1}$ with $m\geq 2$ can be vertex-split into a star $K_{1,m}$ and a group of $m$ mutually edge-disjoint spanning trees of $2m$ edges.
\end{thm}
\begin{proof} \textbf{Result-1.} Assume that a complete graph $K_{2m}$ can be vertex-split into $m$ mutually edge-disjoint spanning trees of $2m-1$ edges.

Consider the complete graph $K_{2m+2}$ with the vertex set $V(K_{2m+2})=\{x_i:~i\in [1,2m+2]\}$. Since $K_{2m}=K_{2m+2}-\{x_{2m+1},x_{2m+2}\}$, so $K_{2m}$ can be vertex-split into $m$ mutually edge-disjoint spanning trees $T_1,T_2,\dots , T_{m}$ of $2m-1$ edges. Without loss of generality, $x_{i}x_{i+1}\not \in E(T_i)$ for $i\in [1,m]$. Notice that the neighbor set $N_{ei}(x_{2m+1})=V(K_{2m})\cup \{x_{2m+2}\}$ and the neighbor set $N_{ei}(x_{2m+2})=V(K_{2m})\cup \{x_{2m+1}\}$. We have a spanning tree of $K_{2m+2}$, which is a star $K_{1,2m+1}$ with $V(K_{1,2m+1})=\{x_{2m+1}, x_{2m+2}\}\cup V(K_{2m})$, and $E(K_{1,2m+1})=\{x_{2m+1}x_{2m+2},x_{2m+2}x_i:~i\in [1,2m]\}$

We add the vertex $x_{2m+1}$ to $T_i$ to join $x_{2m+1}$ with $x_{i}$ and $x_{i+1}$ of $T_i$ by two edges $x_{2m+1}x_{i}$ and $x_{2m+1}x_{i+1}$ since $x_{i}x_{i+1}\not \in E(T_i)$, then the graph $T\,'_i=T_i+\{x_{2m+1}x_{i},x_{2m+1}x_{i+1}\}$ contains a cycle $C_i$ of length at least $4$, such that the cycle $C_i$ contains the edges $x\,'_{i}x_{i}$, $x_{2m+1}x_{i}$, $x_{2m+1}x_{i+1}$ and $x_{i+1}x\,''_{i+1}$.

(i) If the cycle $C_i=x\,'_{i}x_{i}x_{2m+1}x_{i+1}x\,'_{i}$, that is $x\,'_{i}=x\,''_{i+1}$. We remove the edge $x\,'_{i}x_{i}$ from the graph $T\,'_i$, and add the vertex $x_{2m+2}$ to the graph $T\,'_i-x\,'_{i}x_{i}$ by joining $x\,'_{i}$ and $x_{2m+2}$ together with the edge $x\,'_{i}x_{2m+2}$, such that resultant graph $T\,''_i=T\,'_i-x\,'_{i}x_{i}+x\,'_{i}x_{2m+2}$ is just a spanning tree of $K_{2m+2}$, and the removed-edge-added graph $K_{1,2m+1}-x\,'_{i}x_{2m+2}+x\,'_{i}x_{i}$ is a spanning tree of $K_{2m+2}$ too, since the index of each edge $x_ix_j$ of the complete graph $K_{2m+2}$ is unique.

(ii) If $x\,'_{i}\neq x\,''_{i+1}$ in the cycle $C_i$, we do the process of obtaining two spanning trees $T\,''_i=T\,'_i-x\,'_{i}x_{i}+x\,'_{i}x_{2m+2}$ and $K_{1,2m+1}-x\,'_{i}x_{2m+2}+x\,'_{i}x_{i}$ just like the above process in (i).

\textbf{Result-2.} Since the complete graph $K_{2m}=K_{2m+1}-x_{2m+1}$, and the vertex $x_{2m+1}$ has its own degree $2m$ in the complete graph $K_{2m+1}$, we add the vertex $x_{2m+1}$ to each spanning tree $T_i$ of the complete graph $K_{2m}$ for $i\in [1,m]$ by joining the vertex $x_{2m+1}$ with a vertex $x_i$ of the spanning tree $T_i$, so the resultant graph is just a spanning tree $T_i+x_{2m+1}x_i$ of the complete graph $K_{2m+1}$ holding $|E(T_i+x_{2m+1}x_i)|=2m$, and the remainder edges $x_{m+j}x_{2m+1}$ for $j\in [1,m]$ form a star $K_{1,m}$ of $K_{2m+1}$.

By induction, we have completed the proof of the theorem.
\end{proof}

\begin{cor}\label{cor:complete-graph-vertex-split-m-trees}
$^*$ Let $T_{ree}(2m)$ be the set of mutually disjoint trees of $2m-1$ edges. Then there are $m$ mutually disjoint trees $T_1,T_2,\dots , T_{m}$ of $T_{ree}(2m)$ holding $K_{2m}=[\odot]^m_{i=1} T_i$ true.
\end{cor}

\begin{problem}\label{qeu:444444}
Let $T_{ree}(2m)$ be the set of mutually disjoint trees of $2m-1$ edges; refer to Corollary \ref{cor:complete-graph-vertex-split-m-trees}.

(i) \textbf{Find} all groups of $m$ mutually disjoint trees $T_1,T_2,\dots , T_{m}$ from the set $T_{ree}(2m)$ for the complete graph $K_{2m}$ holding $K_{2m}=[\odot]^m_{i=1} T_i$ true.

(ii) \textbf{Is} there a group of $m$ mutually disjoint trees $T_1,T_2,\dots , T_{m}$ in the set $T_{ree}(2m)$ holding $K_{2m}=[\odot]^m_{i=1} T_i$, and $T_i\cong T_j$ for $i,j\in [1,m]$?

(iii) \textbf{Are} there $m$ mutually disjoint trees $T_1,T_2,\dots , T_{m}$ in the set $T_{ree}(2m)$ holding $K_{2m}=[\odot]^m_{i=1} T_i$, such that two leaf-sets $|L(T_i)|=|L(T_j)|$ for $i,j\in [1,m]$?

(iv) \textbf{Characterize} a connected regular graph $G$ if this graph $G$ can be vertex-split into edge-disjoint trees $H_1,H_2,\dots , H_{n}$ such that $G=[\odot]^n_{i=1} H_i$.
\end{problem}

According to Theorem \ref{thm:connected-graph-vs-tree-set}, we have

\begin{thm}\label{thm:graph-vertex-split-m-tree-set}
$^*$ Each non-tree connected graph $G$ can be vertex-split into mutually edge-disjoint trees $T_1,T_2,\dots, T_m$ with $m\geq 2$, such that the edge set $E(G)=\bigcup ^m_{i=1}E(T_i)$.
\end{thm}

\begin{rem}\label{rem:333333}
Since each non-tree connected $(p,q)$-graph $G$ corresponds to a tree set $V_{\textrm{split}}(G)$ of trees of $q+1$ vertices, it seems to be not easy to characterize the tree set $V_{\textrm{split}}(G)$, although the part of solutions of Problem \ref{qeu:verious-tree-decompositions} are in the tree set $V_{\textrm{split}}(G)$; refer to Corollary \ref{cor:vertex-split-m-tree-set-result}.\paralled
\end{rem}

\begin{cor}\label{cor:vertex-split-m-tree-set-result}
$^*$ If a non-tree connected $(p,q)$-graph $G$ admits a proper vertex coloring $g:V(G)\rightarrow [1,p]$ with $g(u)\neq g(x)$ for distinct $u,x\in V(G)$, and there is a tree $H\in V_{\textrm{split}}(G)$ that can be vertex-split into mutually edge-disjoint trees $H_1,H_,\dots H_s$ holding $|g(V(H_r))|=|V(H_r)|$ for $r\in [1,s]$, then we claim that the mutually edge-disjoint trees $H_1,H_,\dots H_s$ are a tree-decomposition of the graph $G$.
\end{cor}

By Theorem \ref{thm:K-2m-spanning-trees-2m-1-edges}, we propose

\begin{conj}\label{conj:complete-graph-spanning-trees-hamilton-cycle}
$^*$ \textbf{KSTH-conjecture}. Each complete graph $K_{2m}$ of $2m$ vertices can be vertex-split into some group $S_{paning}(K_{2m})$ of $m$ mutually edge-disjoint spanning trees of $2m-1$ edges, such that any pair of spanning trees $T_i$ and $T_j$ in $S_{paning}(K_{2m})$ with $i\neq j$ forms a vertex-coincided graph $G=T_i[\odot_{prop}]T_j$ containing a Hamilton cycle, also, two perfect matchings; refer to Fig.\ref{fig:introduction-example-11} and Fig.\ref{fig:my-conjecture}.
\end{conj}

\begin{figure}[h]
\centering
\includegraphics[width=16.4cm]{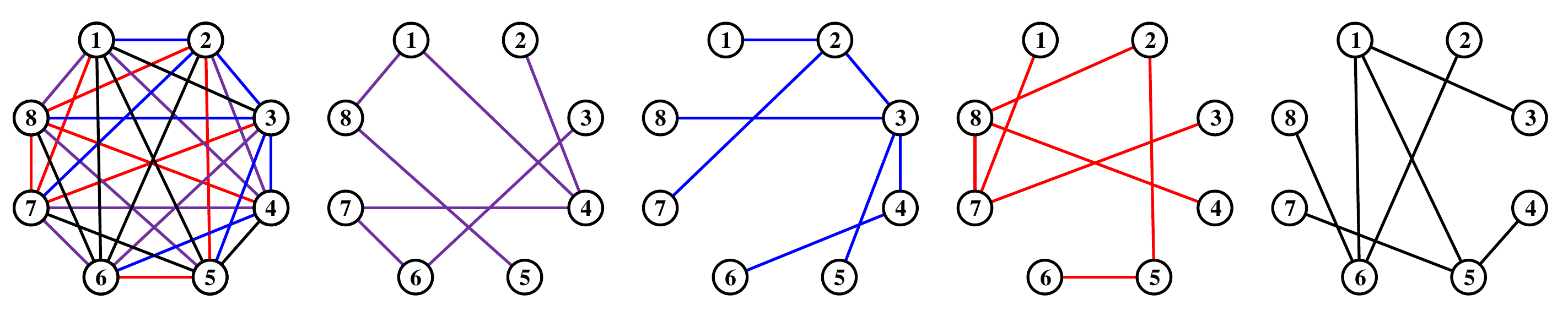}\\
\caption{\label{fig:my-conjecture}{\small A diagram for illustrating Conjecture \ref{conj:complete-graph-spanning-trees-hamilton-cycle}.}}
\end{figure}

\begin{problem}\label{question:444444}
Let $V_{tree}(q)$ be the set of trees of $q$ edges. \textbf{Classify} the set $V_{tree}(q)$ into mutually disjoint subset $V_{split}(G_1),V_{split}(G_2),\dots ,V_{split}(G_m)$ such that $V_{tree}(q)=\bigcup ^m_{i=1}V_{split}(G_i)$, and each subset $V_{split}(G_i)$ is the set of trees obtained by doing the vertex-splitting operation to a certain connected $(p_i,q)$-graph $G_i$ for $i\in [1,m]$. Moreover, \textbf{is} there $G_i\cong G_j$ if $V_{split}(G_i)=V_{split}(G_j)$?
\end{problem}

\begin{problem}\label{question:444444}
Let $V_{graph}(q)$ be the set of connected and non-tree graphs of $q$ edges, and let $V_{graph}(G_i)$ be a subset of $V_{graph}(q)$ such that each connected and non-tree graph of $V_{graph}(G_i)$ is obtained by doing the vertex-splitting operation to a connected $(p_i,q)$-graph $G_i$. So, $V_{graph}(q)=\bigcup ^n_{i=1}V_{graph}(G_i)$, and \textbf{estimate} the number $n$.
\end{problem}

\begin{problem} \label{question:box}
A graph $G$ on $n$ vertices is called \emph{arbitrarily vertex decomposable} if for each sequence $(n_1,n_2,\dots, n_k)$ of positive integers holding $n_1 +n_2 +\cdots+n_k = n$ there exists a partition $(V_1, V_2,\dots, V_k)$ of $V(G)$ such that each subset $V_i$ for $i\in [1,k]$ induces a connected subgraph $G_i$ of the graph $G$ with $n_i=|V(G_i)|$. Obviously, a complete graph $K_n$ is arbitrarily vertex decomposable. The problem of deciding whether a given graph is arbitrarily vertex decomposable is a NP-complete problem (M. Robson, 1998).
\end{problem}

\subsection{Basic colorings and labelings of graphs}

Undefined colorings and labelings mentioned in this article can be found in \cite{Gallian2021, Bondy-2008, Yao-Wang-2106-15254v1}.

\subsubsection{Colorings and labelings with various $W$-constraints}

\begin{defn} \label{defn:basic-W-type-labelings}
\cite{Gallian2021, Yao-Sun-Zhang-Mu-Sun-Wang-Su-Zhang-Yang-Yang-2018arXiv, Bing-Yao-Cheng-Yao-Zhao2009, Zhou-Yao-Chen-Tao2012} Suppose that a connected $(p,q)$-graph $G$ admits a mapping $\theta:V(G)\rightarrow \{0,1,2,\dots \}$. For each edge $xy\in E(G)$, the edge induced-color is defined as $\theta(xy)=|\theta(x)-\theta(y)|$. Write vertex color set by $\theta(V(G))=\{\theta(u):u\in V(G)\}$, and edge color set by
$\theta(E(G))=\{\theta(xy):xy\in E(G)\}$. There are the following constraints:

\textbf{C-1.} $|\theta(V(G))|=p$;

\textbf{C-2.} $\theta(V(G))\subseteq [0,q]$, $\min \theta(V(G))=0$;

\textbf{C-3.} $\theta(V(G))\subset [0,2q-1]$, $\min \theta(V(G))=0$;

\textbf{C-4.} $\theta(E(G))=[1,q]$;

\textbf{C-5.} $\theta(E(G))=[1,2q-1]^o$;

\textbf{C-6.} $G$ is a bipartite graph with the the vertex set $V(G)=X\cup Y$ and $X\cap Y=\emptyset$ such that $\max\{\theta(x):x\in X\}< \min\{\theta(y):y\in Y\}$ ($\theta(X)<\theta(Y)$ for short);

\textbf{C-7.} $G$ is a tree having a perfect matching $M$ such that $\theta(x)+\theta(y)=q$ for each matching edge $xy\in M$;

\textbf{C-8.} $G$ is a tree having a perfect matching $M$ such that $\theta(x)+\theta(y)=2q-1$ for each matching edge $xy\in M$.

\noindent \textbf{Then}:
\begin{asparaenum}[\textbf{\textrm{Lab}}-1.]
\item A \emph{graceful labeling} $\theta$ satisfies the constraints C-1, C-2 and C-4 at the same time.
\item A \emph{set-ordered graceful labeling} $\theta$ holds the constraints C-1, C-2, C-4 and C-6 true.
\item A \emph{strongly graceful labeling} $\theta$ holds the constraints C-1, C-2, C-4 and
C-7 true.
\item A \emph{set-ordered strongly graceful labeling} $\theta$ holds the constraints C-1, C-2, C-4, C-6 and C-7 true.
\item An \emph{odd-graceful labeling} $\theta$ holds the constraints C-1, C-3 and C-5 true.
\item A \emph{set-ordered odd-graceful labeling} $\theta$ abides the constraints C-1, C-3, C-5 and C-6.
\item A \emph{strongly odd-graceful labeling} $\theta$ holds the constraints C-1, C-3, C-5 and C-8 true.
\item A \emph{set-ordered strongly odd-graceful labeling} $\theta$ holds the constraints C-1, C-3, C-5, C-6 and C-8 true.\qqed
\end{asparaenum}
\end{defn}

\begin{figure}[h]
\centering
\includegraphics[width=16.4cm]{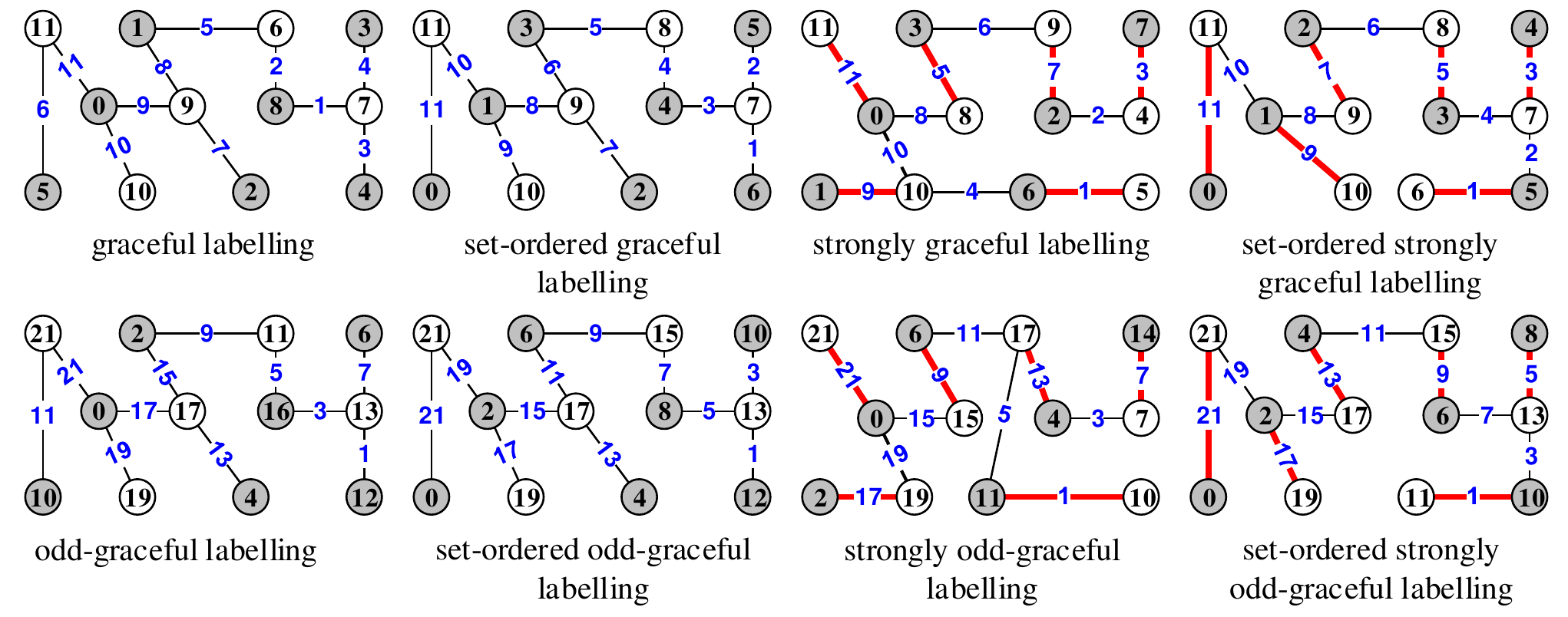}\\
\caption{\label{fig:graceful-odd-graceful}{\small Examples for understanding Definition \ref{defn:basic-W-type-labelings}, the red edges are the matching edges, cited from \cite{Yao-Wang-2106-15254v1}.}}
\end{figure}

\begin{defn} \label{defn:kd-w-type-colorings}
\cite{Yao-Wang-2106-15254v1, Yao-Su-Wang-Hui-Sun-ITAIC2020} Let $G$ be a bipartite and connected $(p,q)$-graph, so its vertex set $V(G)=X\cup Y$ with $X\cap Y=\emptyset$ such that each edge $uv\in E(G)$ holds $u\in X$ and $v\in Y$. If there is a total mapping
\begin{equation}\label{eqa:333333}
f:X\rightarrow S_{m,0,0,d}=\{0,d,\dots ,md\},\quad f:Y\cup E(G)\rightarrow S_{n,k,0,d}=\{k,k+d,\dots ,k+nd\}
\end{equation} here it is allowed $f(x)=f(y)$ for some distinct vertices $x,y\in V(G)$. Let $c$ be a fixed non-negative integer.
\begin{asparaenum}[\textrm{\textbf{Ptol}}-1. ]
\item If $f(uv)=|f(u)-f(v)|$ for each edge $uv\in E(G)$, and the edge color set $f(E(G))=S_{q-1,k,0,d}$, and the vertex color set
$$f(V(G)\cup E(G))\subseteq S_{m,0,0,d}\cup S_{q-1,k,0,d}$$ then $f$ is called a \emph{graceful $(k,d)$-total coloring}; and moreover $f$ is called a \emph{strongly graceful $(k,d)$-total coloring} if $f(x)+f(y)=k+(q-1)d$ for each matching edge $xy$ of a perfect matching $M$ of $G$.
\item If $f(uv)=|f(u)-f(v)|$ for each edge $uv\in E(G)$, and the edge color set $f(E(G))=O_{2q-1,k,d}$, and the vertex color set
$$f(V(G)\cup E(G))\subseteq S_{m,0,0,d}\cup S_{2q-1,k,0,d}$$ then $f$ is called an \emph{odd-graceful $(k,d)$-total coloring}; and moreover $f$ is called a \emph{strongly odd-graceful $(k,d)$-total coloring} if $f(x)+f(y)=k+(2q-1)d$ for each matching edge $xy$ of a perfect matching $M$ of $G$.
\item If the mixed color set $$\{f(u)+f(uv)+f(v):uv\in E(G)\}=\{2k+2ad,2k+2(a+1)d,\dots ,2k+2(a+q-1)d\}$$ with $a\geq 0$ and the total color set
$$f(V(G)\cup E(G))\subseteq S_{m,0,0,d}\cup S_{2(a+q-1),k,a,d}
$$ then $f$ is called an \emph{edge-antimagic $(k,d)$-total coloring}.
\item If $f(uv)=f(u)+f(v)~(\bmod^*qd)$ defined by $f(uv)-k=[f(u)+f(v)-k](\bmod ~qd)$ for each edge $uv\in E(G)$, and the edge color set $f(E(G))=S_{q-1,k,0,d}$, then $f$ is called a \emph{harmonious $(k,d)$-total coloring}.
\item If $f(uv)=f(u)+f(v)~(\bmod^*qd)$ defined by $f(uv)-k=[f(u)+f(v)-k](\bmod ~2qd)$ for each edge $uv\in E(G)$, and the edge color set $f(E(G))=O_{2q-1,k,d}$, then $f$ is called an \emph{odd-elegant $(k,d)$-total coloring}.

--------- \emph{$4$-magic}

\item If there are the edge-magic constraint
\begin{equation}\label{eqa:definition-edge-magic-constraint}
f(u)+f(uv)+f(v)=c,~uv\in E(G)
\end{equation} the edge color set $f(E(G))=S_{q-1,k,0,d}$ and the total color set
$$f((V(G)\cup E(G))\subseteq S_{m,0,0,d}\cup S_{q-1,k,0,d}
$$ then $f$ is called an \emph{edge-magic $(k,d)$-total coloring}. And moreover, $f$ is called a \emph{pseudo-edge-magic $(k,d)$-total coloring} if both $|f(E(G))|< q$ and Eq.(\ref{eqa:definition-edge-magic-constraint}) hold true.
\item If there are the edge-difference constraint
\begin{equation}\label{eqa:definition-edge-difference-constraint}
f(uv)+|f(u)-f(v)|=c,~uv\in E(G)
\end{equation} and the edge color set $f(E(G))=S_{q-1,k,0,d}$, then $f$ is called an \emph{edge-difference $(k,d)$-total coloring}. And moreover $f$ is called a \emph{pseudo-edge-difference $(k,d)$-total coloring} if both $|f(E(G))|< q$ and Eq.(\ref{eqa:definition-edge-difference-constraint}) hold true.
\item If there are the graceful-difference constraint
\begin{equation}\label{eqa:definition-graceful-difference-constraint}
\big ||f(u)-f(v)|-f(uv)\big |=c,~uv\in E(G)
\end{equation} and the edge color set $f(E(G))=S_{q-1,k,0,d}$, then we call $f$ \emph{graceful-difference $(k,d)$-total coloring}. And moreover $f$ is called a \emph{pseudo-graceful-difference $(k,d)$-total coloring} if both $|f(E(G))|< q$ and Eq.(\ref{eqa:definition-graceful-difference-constraint}) hold true.
\item If there are the felicitous-difference constraint
\begin{equation}\label{eqa:definition-felicitous-difference-constraint}
|f(u)+f(v)-f(uv)|=c,~uv\in E(G)
\end{equation} and the edge color set $f(E(G))=S_{q-1,k,0,d}$, then we call $f$ \emph{felicitous-difference $(k,d)$-total coloring}. And moreover $f$ is called a \emph{pseudo-felicitous-difference $(k,d)$-total coloring} if both $|f(E(G))|< q$ and Eq.(\ref{eqa:definition-felicitous-difference-constraint}) hold true.\qqed
\end{asparaenum}
\end{defn}

\begin{figure}[h]
\centering
\includegraphics[width=16.4cm]{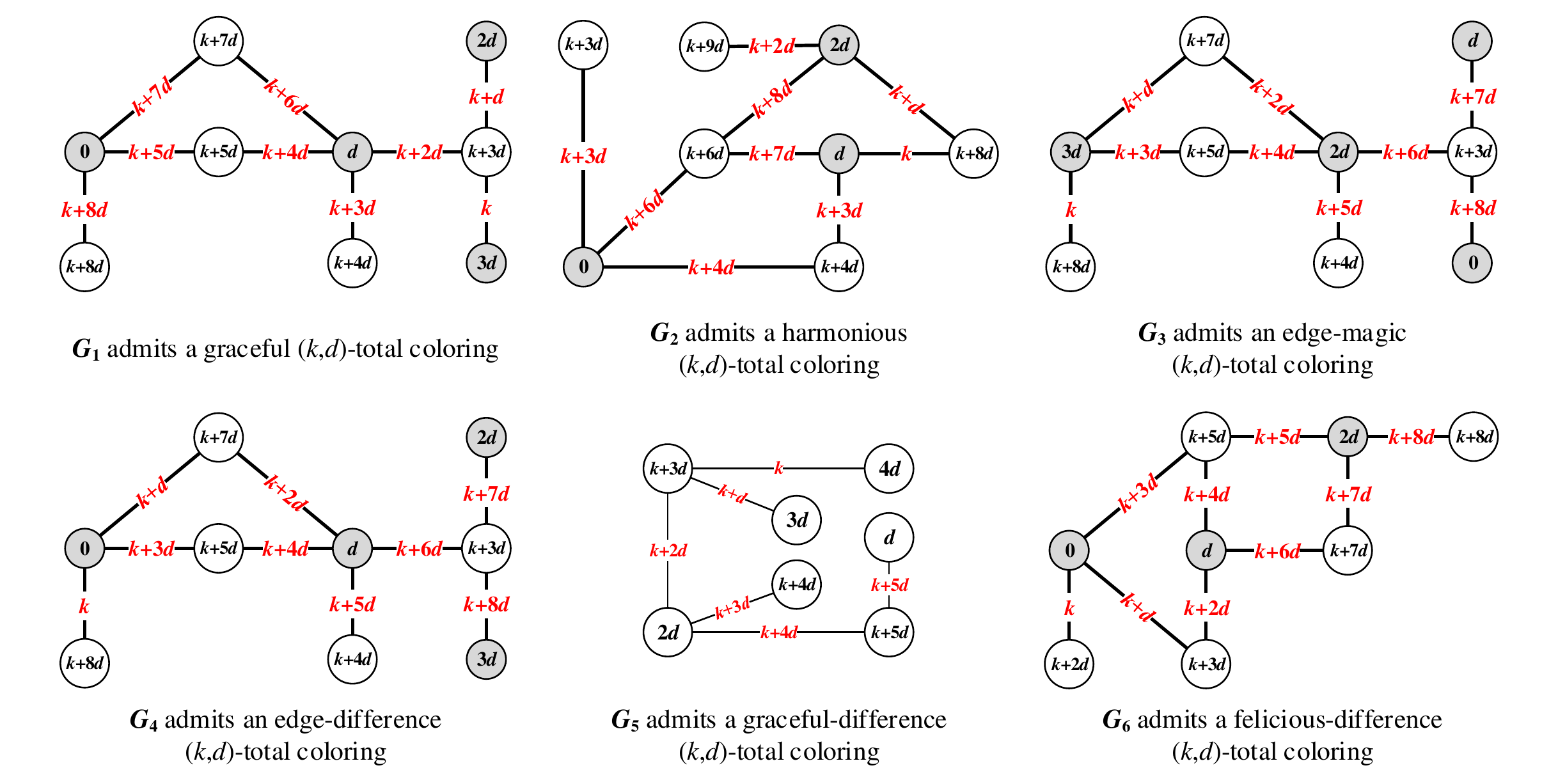}\\
\caption{\label{fig:3-odd-4-magic}{\small A diagram for understanding Definition \ref{defn:kd-w-type-colorings}.}}
\end{figure}

\begin{lem}\label{lem:set-ordered-w-cond-set-co}
\cite{Bing-Yao-arXiv:2207-03381} If a tree admits a \emph{set-ordered $W$-constraint labeling} (refer to Definition \ref{defn:all-thing-coloringss} and Definition \ref{defn:basic-W-type-labelings}), then it admits:

(i) a $W$-constraint proper total set-coloring;

(ii) a $W$-constraint proper $(k,d)$-total coloring; and

(ii) a $W$-constraint proper $(k,d)$-total set-coloring;
\end{lem}

\begin{example}\label{exa:8888888888}
In Fig.\ref{fig:3-odd-4-magic}, we can see:

(i) the graph $G_1$ admits a graceful $(k,d)$-total coloring $f_1$ holding the graceful constraint $f_1(uv)=|f_1(u)-f_1(v)|$ for each edge $uv\in E(G_1)$;

(ii) the graph $G_2$ admits a harmonious $(k,d)$-total coloring $f_2$ holding the harmonious constraint $f_3(uv)=f_3(u)+f_3(v)~(\bmod ~9d)$ for each edge $uv\in E(G_2)$;

(iii) the graph $G_3$ admits an edge-magic $(k,d)$-total coloring $f_3$ holding the edge-magic constraint $$f_3(u)+f_3(uv)+f_3(v)=2k+11d
$$ for each edge $uv\in E(G_3)$;

(iv) the graph $G_4$ admits an edge-difference $(k,d)$-total coloring $f_4$ holding the edge-difference constraint
$$f_4(uv)+|f_4(u)-f_4(v)|=2k+8d
$$ for each edge $uv\in E(G_4)$;

(v) the graph $G_5$ admits a graceful-difference $(k,d)$-total coloring $f_5$ holding the edge-difference constraint
$$\big ||f_5(u)-f_5(v)|-f_5(uv)\big |=d
$$ for each edge $uv\in E(G_5)$;

(vi) the graph $G_6$ admits a graceful-difference $(k,d)$-total coloring $f_6$ holding the edge-difference constraint
$$|f_6(u)+f_6(v)-f_6(uv)|=2d
$$ for each edge $uv\in E(G_6)$.\qqed
\end{example}

\begin{rem}\label{rem:333333}
For the bipartite and connected $(p,q)$-graph $G$ in Definition \ref{defn:kd-w-type-colorings}, we have:
\begin{asparaenum}[\textbf{\textrm{Tradition}}-1.]
\item As $(k,d)=(1,1)$, we get a \emph{graceful total coloring} $f$ wit the edge color set $f(E(G))=[1,q]$; and $f$ is a \emph{strongly graceful total coloring} if $f(x)+f(y)=q$ for each matching edge $xy$ of a perfect matching $M$ of $G$.
\item As $(k,d)=(1,2)$, we get a \emph{odd-graceful total coloring} $f$ wit the edge color set $f(E(G))=[1,2q-1]^o$; and $f$ is a \emph{strongly odd-graceful total coloring} if $f(x)+f(y)=2q-1$ for each matching edge $xy$ of a perfect matching $M$ of $G$.
\item As $(k,d)=(1,1)$, $f(uv)=f(u)+f(v)~(\bmod^*q)$ for each edge $uv\in E(G)$, and the edge color set $f(E(G))=[0.q-1]$, then $f$ is called a \emph{harmonious total coloring}.

\qquad If the vertex color set $|f(V(G))|=|V(G)|=p$, the above $W$-constraint total colorings are called \emph{$W$-constraint total labelings}. Moreover, the graceful total coloring/labeling, the odd-graceful total coloring/labeling and the harmonious total coloring/labeling can be admitted by many non-bipartite graphs. \paralled
\end{asparaenum}
\end{rem}

\begin{conj}\label{conj:strongly-graceful-matching-path-tree}
\cite{Sun-Wang-Yao-2020} Any tree $T$ with a perfect matching can be transformed into a certain path $P$ with a perfect matching through adding-edge-subtracting dual graphs $H_{k}=H_{k-1}+x_ky_k-u_{k-1}v_{k-1}$ for edge $u_{k-1}v_{k-1}\in E(H_{k-1})$ and edge $x_ky_k\not\in E(H_{k-1})$ for $k\in [1,m]$, where $H_{0}=T$ and $H_{m}=P$, such that both $T$ and $P$ admit strongly graceful labelings.
\end{conj}

\begin{defn}\label{defn:odd-edge-W-type-total-labelings-definition}
\cite{Yao-Zhang-Yang-Wang-Odd-Edge-arXiv-02477} Let $G$ be a bipartite and connected $(p,q)$-graph with vertex set $V(G)=X\cup Y$, and $X\cap Y=\emptyset$. There is a total coloring
\begin{equation}\label{eqa:555555}
h:X\rightarrow S_{m,0,0,d},\quad h:Y\cup E(G)\rightarrow S_{n,k,0,d}\cup O_{q,k,0,d}
\end{equation} with integers $k\geq 0$ and $d\geq 1$. Let $c$ be an non-negative integer.

(i) If there are the \emph{edge-magic constraint} $h(u)+h(uv)+h(v)=c$ for each edge $uv\in E(G)$ and the edge color set $h(E(G))=O_{q,k,0,d}$, then $h$ is called \emph{odd-edge edge-magic $(k,d)$-total \textbf{labeling}} when $|h(V(G))|=p$; and moreover $h$ is called \emph{odd-edge edge-magic $(k,d)$-total \textbf{coloring}} when the vertex color set cardinality $|h(V(G))|<p$.

(ii) If there are the \emph{edge-difference constraint} $h(uv)+|h(u)-h(v)|=c$ for each edge $uv\in E(G)$ and the edge color set $h(E(G))=O_{q,k,0,d}$, then $h$ is called \emph{odd-edge edge-difference $(k,d)$-total \textbf{labeling}} when $|h(V(G))|=p$; and furthermore $h$ is called \emph{odd-edge edge-difference $(k,d)$-total \textbf{coloring}} when the vertex color set cardinality $|h(V(G))|<p$.

(iii) If there are the \emph{felicitous-difference constraint} $|h(u)+h(v)-h(uv)|=c$ for each edge $uv\in E(G)$ and the edge color set $h(E(G))=O_{q,k,0,d}$, then $h$ is called \emph{odd-edge felicitous-difference $(k,d)$-total \textbf{labeling}} when $|h(V(G))|=p$; and moreover $h$ is called \emph{odd-edge felicitous-difference $(k,d)$-total \textbf{coloring}} when the vertex color set cardinality $|h(V(G))|<p$.

(iv) If there are the \emph{graceful-difference constraint} $\big ||h(u)-h(v)|-h(uv)\big |=c$ for each edge $uv\in E(G)$ and the edge color set $h(E(G))=O_{q,k,0,d}$, then $h$ is called \emph{odd-edge graceful-difference $(k,d)$-total \textbf{labeling}} when $|h(V(G))|=p$; and furthermore $h$ is called \emph{odd-edge graceful-difference $(k,d)$-total \textbf{coloring}} when the vertex color set cardinality $|h(V(G))|<p$.\qqed
\end{defn}

\begin{figure}[h]
\centering
\includegraphics[width=16.4cm]{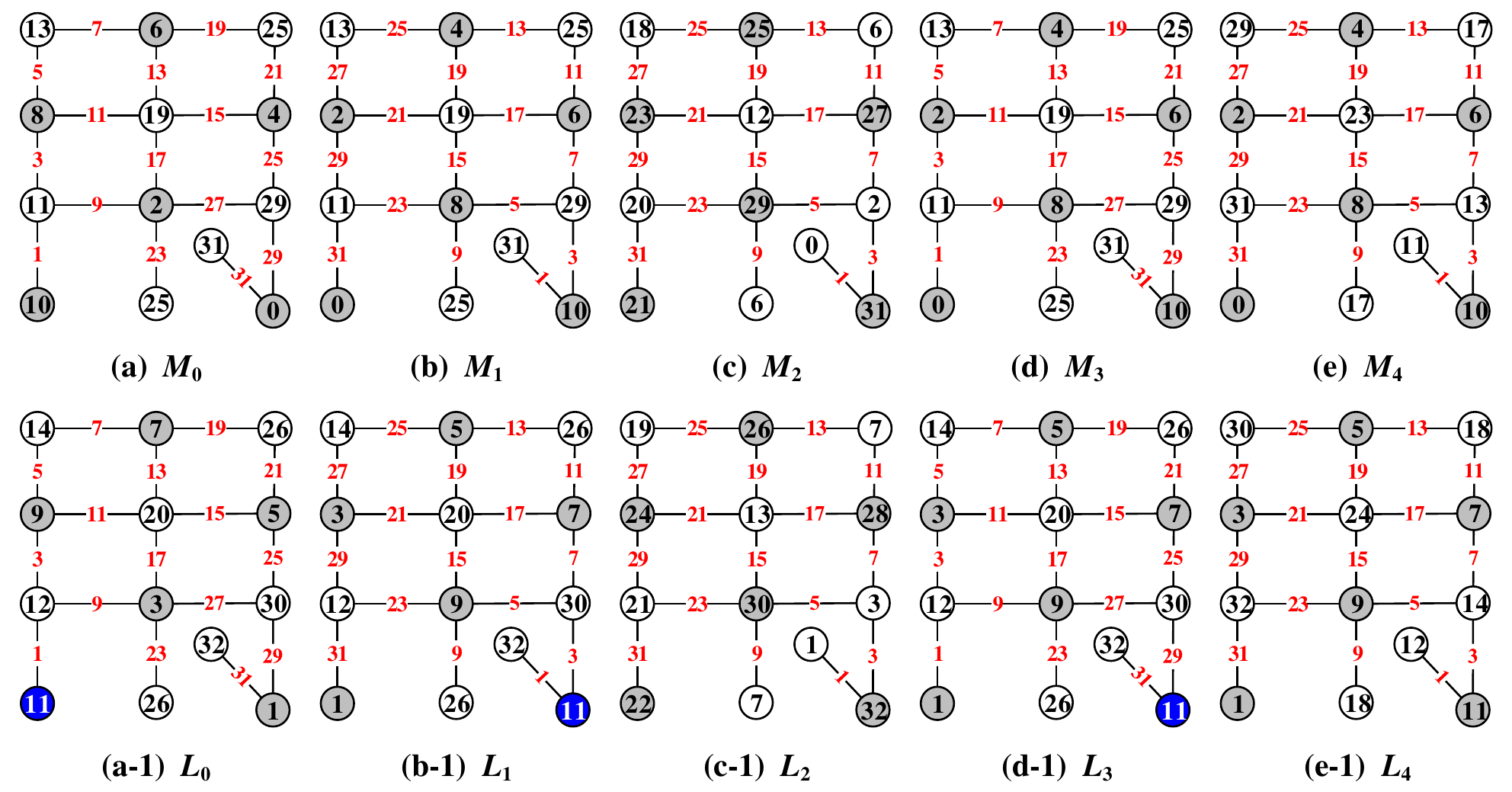}\\
\caption{\label{fig:4magice-total-colorings}{\small A diagram for illustrating Definition \ref{defn:odd-edge-W-type-total-labelings-definition}, cited from \cite{Yao-Zhang-Yang-Wang-Odd-Edge-arXiv-02477}.}}
\end{figure}

\begin{rem} \label{rem:kd-w-tupe-colorings-definition}
\cite{Bing-Yao-arXiv:2207-03381} We call each coloring/labeling $f$ defined in Definition \ref{defn:kd-w-type-colorings} and Definition \ref{defn:odd-edge-W-type-total-labelings-definition} \emph{$W$-constraint $(k,d)$-total coloring} and write the $W$-constraint by $W\langle f(u), f(uv),f(v)\rangle=0$, or $f(uv)=W\langle f(u), f(v)\rangle=0$ in this article.

There are new parameters of graphs based on Definition \ref{defn:kd-w-type-colorings} and Definition \ref{defn:odd-edge-W-type-total-labelings-definition}. We have two $W$-constraint $(k,d)$-total colorings $g_{\min}$ and $g_{\max}$ of a graph $G$ such that
\begin{equation}\label{eqa:w-constraint-coloring-extrem-bounds}
|\emph{g}_{\min}(V(G))|\leq |g(V(G))|\leq |g_{\max}(V(G))|
\end{equation} for each $W$-constraint $(k,d)$-total coloring $g$ of $G$.

As this $W$-constraint $(k,d)$-total coloring is a graceful $(1,1)$-total coloring $g_{\max}$ of a tree $T$ means that the maximal number $|g_{\max}(V(T))|$ is close to a graceful labeling of $T$. \textbf{Graceful Tree Conjecture} \cite{A-Rosa-graceful-1967, Gallian2021} says: \emph{Every tree admits a graceful labeling}, also, $|g_{\max}(V(T))|=|V(T)|$.

As this $W$-constraint $(k,d)$-total coloring is the graceful $(k,d)$-total coloring, then a graceful $(1,2)$-total coloring $g_{\max}$ of a tree $T$ means that the maximal number $|g_{\max}(V(T))|$ is close to an odd-graceful labeling of $T$. In \cite{R-B-Gnanajothi-1991}, Gnanajothi proposed \textbf{Odd-graceful Tree Conjecture}: \emph{Every tree admits an odd-graceful labeling}, also, $|g_{\max}(V(T))|=|V(T)|$.

As this $W$-constraint $(k,d)$-total coloring is a strongly graceful coloring, Wang and Yao, in \cite{Wang-Yao-equivalent-2020}, have proven that two conjectures \textbf{Graceful Tree Conjecture} (see Conjecture \ref{conj:Alexander-Rosa-graceful}) and \textbf{Strongly Graceful Tree Conjecture} (see Conjecture \ref{conj:Broersma-Hoede-strongly-graceful}) are equivalent from each other, here, Broersma and Hoede in 1999 conjectured: \emph{Every tree containing a perfect matching is strongly graceful}.

As this $W$-constraint $(k,d)$-total coloring is an edge-magic $(k,d)$-total labeling: In 1970, Anton Kotzig and Alex Rosa defined an \emph{edge-magic total labeling} of a $(p,q)$-graph $G$ as a bijective mapping $f$ from $V(G)\cup E(G)$ to $[1, p+q]$ such that for any edge $uv$, $f(u)+f(v)+f(uv)=c$, where $c$ is a fixed constant. Moreover, they \textbf{conjectured}: \emph{Every tree admits an edge-magic total labeling}.

A $W$-constraint $(k,d)$-total coloring $f$ defined in Definition \ref{defn:kd-w-type-colorings} is \emph{proper} if $f(u)\neq f(v)$ for each edge $uv\in E(G)$, and $f(uv)\neq f(uw)$ for any two adjacent edges $uv,uw\in E(G)$. So we call $f$ a \emph{$W$-constraint proper $(k,d)$-total coloring} of $G$. Investigating various $W$-constraint proper $(k,d)$-total colorings of graphs is interesting and challenging.

However, determining the upper and lower bounds of Eq.(\ref{eqa:w-constraint-coloring-extrem-bounds}) is difficult, since no polynomial algorithm for determining some $W$-constraint colorings of graphs was reported in our memory.\paralled
\end{rem}

\begin{conj}\label{conj:Alexander-Rosa-graceful}
(Alexander Rosa, 1966) \cite{Alexander-Rosa-graceful-1967} Each tree is graceful.
\end{conj}

Conjecture \ref{conj:Alexander-Rosa-graceful} was computationally verified for all trees of order $\leq 29$ by Deo, Nikoloski and Suraweera in 2002 \cite{Deo-Nikoloski-Suraweera-2002}. After 58 days' computer time, the trueness of Conjecture \ref{conj:Alexander-Rosa-graceful} was verified for all 5,469,566,585 trees on 29 vertices, reported by Horton and Bcomp in their bachelor paper in 2003 \cite{Michael-Horton-Bcomp-2003}.

\begin{conj}\label{conj:Broersma-Hoede-strongly-graceful}
(H. J. Broersma and C. Hoede, 1999) \cite{Broersma-Hoede-strongly-graceful-1999} Every tree containing a perfect matching is strongly graceful.
\end{conj}

In \cite{Bing-Yao-arXiv:2207-03381}, the authors have shown the following results by the polynomial algorithms:

\begin{thm}\label{thm:10-k-d-total-coloringss}
\cite{Bing-Yao-arXiv:2207-03381} Each tree $T$ with diameter $D(T)\geq 3$ admits at least $2^m$ different $W$-constraint $(k,d)$-total colorings for $m+1=\left \lceil \frac{D(T)}{2}\right \rceil $, where $W$-constraint $\in \{$graceful, harmonious, edge-difference, graceful-difference, felicitous-difference, edge-magic, odd-edge edge-difference, odd-edge graceful-difference, odd-edge felicitous-difference, odd-edge edge-magic$\}$; refer to Definition \ref{defn:kd-w-type-colorings} and Definition \ref{defn:odd-edge-W-type-total-labelings-definition}.
\end{thm}

Balakrishnan and Sampathkumar have shown:
\begin{thm}\label{thm:666666}
\cite{Balakrishnan-Sampathkumar-1996} Every graph is a subgraph of a certain graceful graph.
\end{thm}

Acharya, in his paper ``Construction of certain infinite families of graceful graphs from a given graceful graph'' \cite{B-D-Acharya-1982}, investigated the problem of embedding a graph into some graceful graphs.

\begin{rem}\label{rem:333333}
Suppose that $G$ is a connected $(p,q)$-graph. It is not hard to make a labeling $f:V(G)\rightarrow [0,M]$ holding $|f(V(G))|=p$ and $|f(E(G))|=q$. If $f(E(G))=[1,q]$, we are done, otherwise, we make a connected $(p+1,q+1)$-graph $G_1=G+x_1y_1$ for $x_1\not \in V(G)$ and $y_1\in V(G)$, and defined a labeling $f_1$ by coloring the edge $x_1y_1$ with $f_1(x_1y_1)=|f_1(y_1)-f_1(x_1)|\not \in f(E(G))$, and $f_1(w)=f(w)$ for $w\in V(G)$ holding $f_1(uv)=f(uv)$ for $uv\in E(G)$.

Go on in this way, we produce connected $(p+m,q+m)$-graphs $G_m=G_{m-1}+x_my_m$ for $x_m\not \in V(G_{m-1})$ and $y_m\in V(G_{m-1})$, and define a labeling $f_m$ by coloring the edge $x_my_m$ with $f_m(x_my_m)=|f_m(y_m)-f_m(x_m)|\not \in f_{m-1}(E(G_{m-1}))$, and $f_m(w)=f_{m-1}(w)$ for $w\in V(G_{m-1})$ holding $f_m(uv)=f_{m-1}(uv)$ for $uv\in E(G_{m-1})$.

Clearly, there is an integer $B>0$, such that the connected $(p+B,q+B)$-graph $G_B=G_{B-1}+x_By_B$ admits a labeling $f_B$ holding the vertex color set $f_B(V(G_B))\subseteq [0,p+B]$ and the edge color set $f_B(E(G_B))=[1,p+B]$, which means $G_B$ is a graceful graph.

It is meaningful to minimize $|E(G_B)\setminus E(G)|$. \paralled
\end{rem}

\begin{defn}\label{defn:group-definition-twin-total-labelingss}
\cite{Yao-Zhang-Yang-Wang-Odd-Edge-arXiv-02477} Let $G$ be a bipartite $(p,q)$-graph having its own vertex set $V(G)=X_G\cup Y_G$ with $X_G\cap Y_G=\emptyset$, and let $T$ be another bipartite $(p\,',q)$-graph having its own vertex set $V(T)=X_T\cup Y_T$ with $X_T\cap Y_T=\emptyset$. The bipartite $(p,q)$-graph $G$ admits a total labeling $F:V(G)\cup E(G)\rightarrow [0,2q-1]$, and the bipartite $(p\,',q)$-graph $T$ admits a total labeling $F^*:V(T)\cup E(T)\rightarrow [0,2q]$.

(i) \textbf{If}

(i-1) $F$ is a set-ordered odd-edge edge-magic total labeling of $G$;

(i-2) the set-ordered restriction $F^*_{\max}(X_T)<F^*_{\min}(Y_T)$ holds true;

(i-3) the edge color set $F^*(E(T))=[1,2q-1]^o$;

(i-4) there is a positive integer $c_1$, so that each edge $xy\in E(T)$ holds a magic-type restriction $F^*(x)+F^*(xy)+F^*(y)=c_1$; and

(i-5) $F(V(G))\cup F^*(V(T))\subseteq [0,2q]$,\\
then we call $\langle F,F^*\rangle $ a \emph{twin set-ordered odd-edge edge-magic total labeling} of two graphs $G$ and $T$. Especially, we call $\langle F,F^*\rangle $ a \emph{perfect twin set-ordered odd-edge edge-magic total labeling} of $G$ and $T$ if $F(V(G))\cup F^*(V(T))=[0,2q]$.

(ii) \textbf{If}

(ii-1) $F$ is a set-ordered odd-edge edge-difference total labeling of $G$;

(ii-2) the set-ordered restriction $F^*_{\max}(X_T)<F^*_{\min}(Y_T)$ holds true;

(ii-3) the edge color set $F^*(E(T))=[1,2q-1]^o$;

(ii-4) there is a positive integer $c_2$, so that each edge $xy\in E(T)$ holds a magic-type restriction $F^*(xy)+|F^*(y)-F^*(x)|=c_2$; and

(ii-5) $F(V(G))\cup F^*(V(T))\subseteq [0,2q]$,\\
then, $\langle F,F^*\rangle $ is called a \emph{twin set-ordered odd-edge edge-difference total labeling} of two graphs $G$ and $T$. Moreover, we call $\langle F,F^*\rangle $ a \emph{perfect twin set-ordered odd-edge edge-difference total labeling} of $G$ and $T$ if $F(V(G))\cup F^*(V(T))=[0,2q]$.

(iii) \textbf{If}

(iii-1) $F$ is a set-ordered odd-edge felicitous-difference total labeling of $G$;

(iii-2) the set-ordered restriction $F^*_{\max}(X_T)<F^*_{\min}(Y_T)$ holds true;

(iii-3) the edge color set $F^*(E(T))=[1,2q-1]^o$;

(iii-4) there is a non-negative integer $c_3$, so that each edge $xy\in E(T)$ holds a magic-type restriction $|F^*(y)+F^*(x)-F^*(xy)|=c_3$; and

(iii-5) $F(V(G))\cup F^*(V(T))\subseteq [0,2q]$,\\
we call $\langle F,F^*\rangle $ a \emph{twin set-ordered odd-edge felicitous-difference total labeling} of two graphs $G$ and $T$. And we call $\langle F,F^*\rangle $ a \emph{perfect twin set-ordered odd-edge felicitous-difference total labeling} of $G$ and $T$ if $F(V(G))\cup F^*(V(T))=[0,2q]$.

(vi) \textbf{If}

(vi-1) $F$ is a set-ordered odd-edge graceful-difference total labeling;

(vi-2) the set-ordered restriction $F^*_{\max}(X_T)<F^*_{\min}(Y_T)$ holds true;

(vi-3) the edge color set $F^*(E(T))=[1,2q-1]^o$;

(vi-4) there is a non-negative integer $c_4$, so that each edge $xy\in E(T)$ holds a magic-type restriction $\big ||F^*(y)-F^*(x)|-F^*(xy)\big |=c_4$; and

(vi-5) $F(V(G))\cup F^*(V(T))\subseteq [0,2q]$,\\
then, $\langle F,F^*\rangle $ is called a \emph{twin set-ordered odd-edge graceful-difference total labeling} of two graphs $G$ and $T$. Furthermore, we call $\langle F,F^*\rangle $ a \emph{perfect twin set-ordered odd-edge graceful-difference total labeling} of $G$ and $T$ if $F(V(G))\cup F^*(V(T))=[0,2q]$.\qqed
\end{defn}

\begin{defn}\label{defn:twin-odd-edge-w-labelings-coloringsn}
$^*$ Suppose that a bipartite and connected $(p,q)$-graph $G$ admits an odd-edge $W$-constraint $(k,d)$-total labeling (resp. coloring) $h$, and another bipartite and connected $(p\,',q)$-graph $H$ admits an odd-edge $W$-constraint $(k,d)$-total labeling (resp. coloring) $f$, and these two labelings (resp. colorings) $h$ and $f$ are defined in Definition \ref{defn:odd-edge-W-type-total-labelings-definition}. If
\begin{equation}\label{eqa:555555}
h(E(G))=f(E(H)),~h(V(G))\cup f(V(H))\subseteq [0,k+2qd]
\end{equation} we call $\langle h,f\rangle$ a \emph{twin odd-edge $W$-constraint $(k,d)$-total labelings} (resp. \emph{twin odd-edge $W$-constraint $(k,d)$-total colorings}) for each $W$-constraint $\in\{$edge-magic, edge-difference, felicitous-difference, graceful-difference$\}$.\qqed
\end{defn}

\begin{example}\label{exa:8888888888}
\cite{Yao-Zhang-Yang-Wang-Odd-Edge-arXiv-02477} Fig.\ref{fig:4magice-total-colorings} is for illustrating Definition \ref{defn:odd-edge-W-type-total-labelings-definition}, there are the following facts:

(a) The graph $M_0$ admits a set-ordered odd-graceful coloring $f_0$, since there are two vertices colored with 25. And $f_0(E(M_0))=\{f_0(xy)=f_0(y)-f_0(x):xy\in E(M_0)\}=[1,31]^o$.

(b) The graph $M_1$ admits a set-ordered odd-edge edge-magic total coloring $f_1$, since there are two vertices colored with 25. The edge-magic constraint
$$f_1(x)+f_1(xy)+f_1(y)=42
$$ for each edge $xy\in E(M_1)$ holds true.

(c) The graph $M_2$ admits a set-ordered odd-edge edge-difference total coloring $f_2$, since there are two vertices colored with 6. The edge-difference constraint $$f_2(xy)+|f_2(x)-f_2(y)|=f_2(xy)+f_2(y)-f_2(x)=32
$$ for each edge $xy\in E(M_2)$ holds true.

(d) The graph $M_3$ admits a set-ordered odd-edge felicitous-difference total coloring $f_3$, since there are two vertices colored with 25. The felicitous-difference constraint
$$|f_3(x)+f_3(y)-f_3(xy)|=10
$$ for each edge $xy\in E(M_3)$ holds true.

(e) The graph $M_4$ admits a set-ordered odd-edge graceful-difference total coloring $f_4$, since there are two vertices colored with 17. The graceful-difference constraint
$$\big ||f_4(x)-f_4(y)|-f_4(xy)\big |=[f_4(y)-f_4(x)]-f_4(xy)=0
$$ for each edge $xy\in E(M_4)$ holds true.

Each pair of graphs $M_k$ and $L_k$ for $k\in [0,4]$ forms a \emph{twin odd-edge graphs}, they admit a \emph{twin odd-edge $W$-constraint $(k,d)$-total labelings} (resp. \emph{twin odd-edge $W$-constraint $(k,d)$-total colorings}) for each $W$-constraint $\in\{$edge-magic, edge-difference, felicitous-difference, graceful-difference$\}$ defined in Definition \ref{defn:twin-odd-edge-w-labelings-coloringsn}.\qqed
\end{example}

\begin{rem}\label{rem:333333}
In Definition \ref{defn:twin-odd-edge-w-labelings-coloringsn}, the bipartite and connected graph $G$ is as a \emph{public-key graph}, and the bipartite and connected graph $H$ is also a \emph{private-key graph} in real application. However, the public-key graph $G$ may correspond two or more private-key graphs $H_1,H_2,\dots, H_s$ with $s\geq 2$; refer to Example \ref{exa:4magice-total-colorings}.\paralled
\end{rem}

\begin{example}\label{exa:4magice-total-colorings}
Notices the graph $M_0$ shown in Fig.\ref{fig:4magice-total-colorings} admits a set-ordered odd-graceful coloring $f_0$, and each graph $J_i$ shown in Fig.\ref{fig:4magic-twin-colorings} is a \emph{twin odd-graceful graph} of $M_0$, and each odd-graceful coloring $g_i$ admitted by the graph $J_i$ for $i\in [1,6]$ is with the odd-graceful coloring $f_0$ to form a \emph{twin odd-graceful total colorings}. We have

(i) $g_i(E(J_i))=f_0(E(M_0))$ for $i\in [1,6]$;

(ii) $|f_0(V(M_0))\cap g_i(V(J_i))|=1$ for $i\in [1,3]$, and $f_0(V(M_0))\cap g_i(V(J_i))=\emptyset$ for $i\in [4,6]$;

(iii) $f_0(V(M_0))\cup g_6(V(J_6))=[0,32]$, and $f_0(V(M_0))\cup g_i(V(J_i))\subseteq [0,32]$ for $i\in [1,5]$.\qqed
\end{example}

\begin{figure}[h]
\centering
\includegraphics[width=16.4cm]{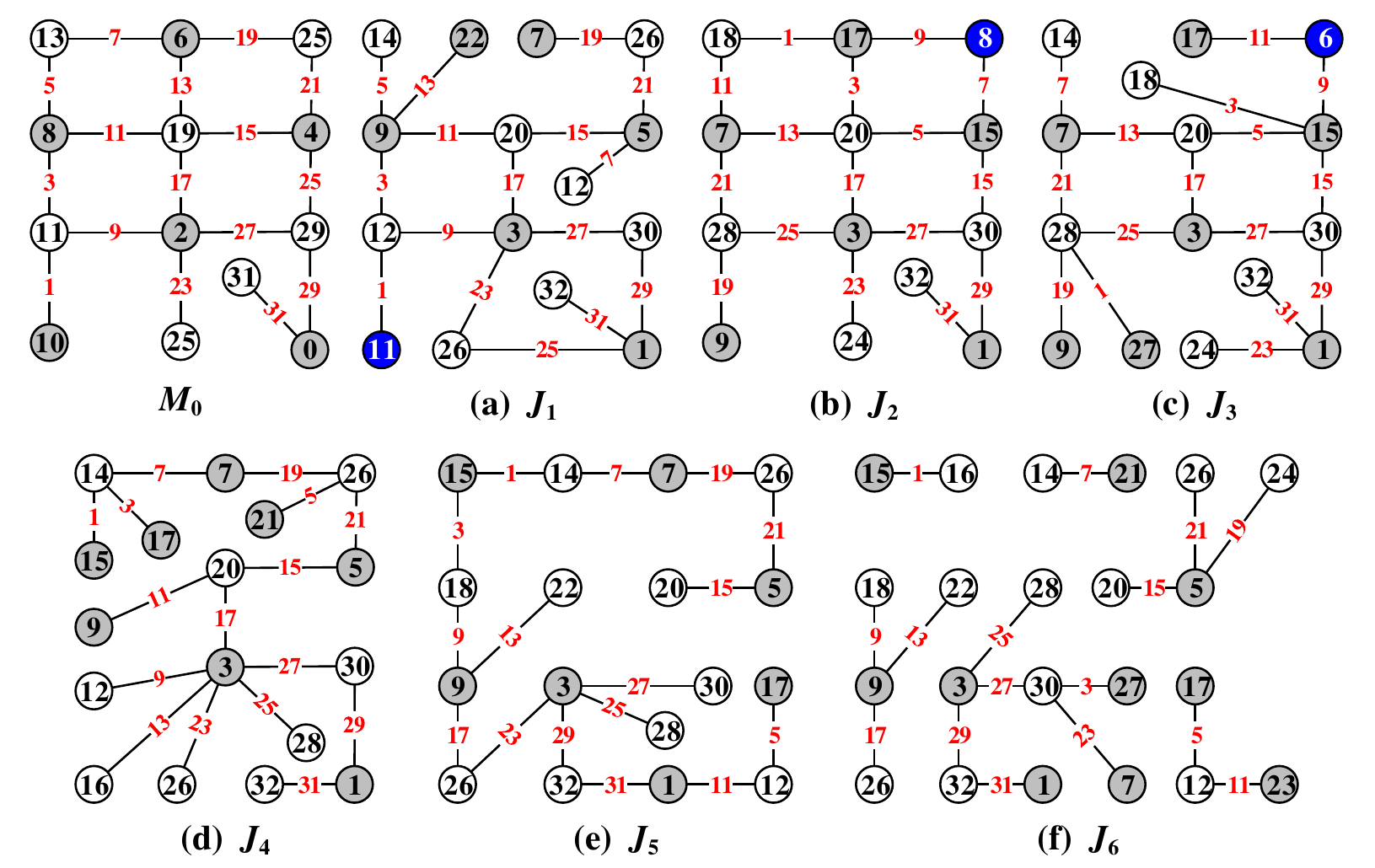}\\
\caption{\label{fig:4magic-twin-colorings}{\small For illustrating Definition \ref{defn:twin-odd-edge-w-labelings-coloringsn}, where each graph $J_i$ admits an odd-graceful coloring $g_i$ for $i\in [1,6]$.}}
\end{figure}

\begin{thm}\label{thm:odd-edge-W-type-kd-total-labelings}
$^*$ If a bipartite and connected $(p,q)$-graph $G$ admits an odd-edge $W$-constraint $(k,d)$-total labelings (colorings) defined in Definition \ref{defn:odd-edge-W-type-total-labelings-definition}, then there is at least a bipartite and connected $(p\,',q)$-graph $H$, such that two graphs $G$ and $H$ admit a \emph{twin odd-edge $W$-constraint $(k,d)$-total labelings} (resp. \emph{twin odd-edge $W$-constraint $(k,d)$-total colorings}) for each $W$-constraint $\in\{$edge-magic, edge-difference, felicitous-difference, graceful-difference$\}$ defined in Definition \ref{defn:twin-odd-edge-w-labelings-coloringsn}.
\end{thm}

\begin{problem}\label{qeu:444444}
If a bipartite and connected $(p,q)$-graph $G$ admits an odd-edge $W$-constraint $(k,d)$-total labelings (colorings) defined in Definition \ref{defn:odd-edge-W-type-total-labelings-definition}, so there is a set $T_{win}(G)$ of graphs such that each graph $H\in T_{win}(G)$ and $G$ together form a twin odd-edge graphs admitting a twin odd-graceful total colorings. \textbf{Characterize} the graph set $T_{win}(G)$ due to the needs of asymmetric topology cryptography.
\end{problem}

\begin{thm} \label{thm:connections-several-labelings}
\cite{Yao-Liu-Yao-2017} Let $T$ be a tree on $p$ vertices, and let $(X,Y)$ be its
bipartition of vertex set $V(T)$. For all values of integers $k\geq 1$ and $d\geq 1$, the following assertions are mutually equivalent:

$(1)$ $T$ admits a set-ordered graceful labeling $f$ with $\max f(X)<\min f(Y)$.

$(2)$ $T$ admits a super felicitous labeling $\alpha$ with $\max \alpha(X)<\min \alpha(Y)$.

$(3)$ $T$ admits a $(k,d)$-graceful labeling $\beta$ with
$\beta(x)<\beta(y)-k+d$ for all $x\in X$ and $y\in Y$.

$(4)$ $T$ admits a super edge-magic total labeling $\gamma$ with $\max \gamma(X)<\min \gamma(Y)$ and a magic constant $|X|+2p+1$.

$(5)$ $T$ admits a super $(|X|+p+3,2)$-edge antimagic total labeling $\theta$ with $\max \theta(X)<\min \theta(Y)$.

$(6)$ $T$ admits an odd-elegant labeling $\eta$ with $\eta(x)+\eta(y)\leq 2p-3$ for every edge $xy\in E(T)$.

$(7)$ $T$ admits a $(k,d)$-arithmetic labeling $\psi$ with $\max \psi(x)<\min \psi(y)-k+d\cdot |X|$ for all $x\in X$ and $y\in Y$.

$(8)$ $T$ admits a harmonious labeling $\varphi$ with $\max \varphi(X)<\min \varphi(Y\setminus \{y_0\})$ and $\varphi(y_0)=0$.
\end{thm}

\begin{thm}\label{thm:graph-admits-6-set-colorings}
\cite{Bing-Yao-arXiv:2207-03381} Each connected graph $G$ admits each one of the following $W$-constraint $(k,d)$-total set-colorings for $W$-constraint $\in \{$graceful, harmonious, edge-difference, edge-magic, felicitous-difference, graceful-difference$\}$.
\end{thm}

\begin{thm}\label{thm:equivalent-k-d-total-colorings}
\cite{Bing-Yao-arXiv:2207-03381} A bipartite and connected $(p,q)$-graph $G$ admits a graceful $(k,d)$-total coloring if and only if $G$ admits each one of edge-magic $(k,d)$-total coloring, graceful-difference $(k,d)$-total coloring, edge-difference $(k,d)$-total coloring, felicitous-difference $(k,d)$-total coloring, harmonious $(k,d)$-total coloring and edge-antimagic $(k,d)$-total coloring.
\end{thm}

\subsubsection{Connections between magic-constraint colorings}

By Definition \ref{defn:kd-w-type-colorings}, we propose several connections between four magic-constraint colorings as follows:

\textbf{Case A.} Suppose that a $(p,q)$-graph $G$ admits an edge-magic total coloring $f$, such that the edge-magic constraint
\begin{equation}\label{eqa:edge-magic-sequence}
f(xy)+f(x)+f(y)=c,~xy\in E(G)
\end{equation} holds true.

\textbf{A-1.} Edge-difference constraint.

A-1.1. If $f(x)<f(y)$ for some two vertices $x$ and $y$ of $V(G)$, we have the edge-difference constraint
$$f(xy)+|f(x)-f(y)|=f(xy)+f(y)-f(x)=f(xy)+f(y)+f(x)-2f(x)=c-2f(x)$$
by Eq.(\ref{eqa:edge-magic-sequence}).

Moreover, if there is the set-ordered constraint $f(X)<f(Y)$ as $G$ is a bipartite graph with the vertex set $V(G)=X\cup Y$ and $X\cap Y=\emptyset$, then the edge-difference constraint color set is
\begin{equation}\label{eqa:555555}
\big \{f(xy)+|f(x)-f(y)|:xy\in E(G)\big \}=\big \{c-2f(x):x\in X\subset V(G)\big \}
\end{equation}

A-1.2. If $f(y)<f(x)$ for some two vertices $x$ and $y$ of $V(G)$, we get the edge-difference constraint
$$f(xy)+|f(x)-f(y)|=f(xy)+f(x)-f(y)=f(xy)+f(y)+f(x)-2f(y)=c-2f(y)$$
according to Eq.(\ref{eqa:edge-magic-sequence}).

\textbf{A-2.} The felicitous-difference constraint
$$|f(y)+f(x)-f(xy)|=|f(y)+f(x)+f(xy)-2f(xy)|=|c-2f(xy)|,~xy\in E(G)
$$ holds true by Eq.(\ref{eqa:edge-magic-sequence}).

\textbf{A-3.} Graceful-difference constraint.

A-3.1. If $f(x)<f(y)$ for some two vertices $x$ and $y$ of $V(G)$, we have the graceful-difference constraint
$$\big ||f(y)-f(x)|-f(xy)\big |=\big |f(y)-f(x)-f(xy)\big |=\big |[f(xy)+f(y)+f(x)]-2f(y)\big |=|c-2f(y)|$$
by Eq.(\ref{eqa:edge-magic-sequence}).

Moreover, if there is the set-ordered constraint $f(X)<f(Y)$ as $G$ is a bipartite graph with the vertex set $V(G)=X\cup Y$ and $X\cap Y=\emptyset$, then the graceful-difference constraint color set is
\begin{equation}\label{eqa:555555}
\left \{\big ||f(y)-f(x)|-f(xy)\big |:xy\in E(G)\right \}=\big \{|c-2f(y)|:y\in Y\subset V(G)\big \}
\end{equation}

A-3.2. If $f(y)<f(x)$ for some two vertices $x$ and $y$ of $V(G)$, then we have the graceful-difference constraint
$$\big ||f(y)-f(x)|-f(xy)\big |=\big |f(x)-f(y)-f(xy)\big |=\big |[f(xy)+f(y)+f(x)]-2f(x)\big |=|c-2f(x)|$$
according to Eq.(\ref{eqa:edge-magic-sequence}).

\vskip 0.2cm

\textbf{Case B.} Suppose that a $(p,q)$-graph $G$ admits an edge-difference total coloring $g$, such that the edge-difference constraint
\begin{equation}\label{eqa:edge-difference-sequence}
g(xy)+|g(x)-g(y)|=c,~xy\in E(G)
\end{equation} holds true.

\textbf{B-1.} The edge-magic constraint.

B-1.1. If $g(x)<g(y)$ for some two vertices $x$ and $y$ of $V(G)$, we have the edge-magic constraint
$$g(xy)+g(x)+g(y)=g(xy)+g(y)-g(x)+2g(x)=c+2g(x)$$
according to Eq.(\ref{eqa:edge-difference-sequence}).

Moreover, if there is the set-ordered constraint $g(X)<g(Y)$ as $G$ is a bipartite graph with the vertex set $V(G)=X\cup Y$ and $X\cap Y=\emptyset$, then the edge-magic constraint color set is
\begin{equation}\label{eqa:555555}
\left \{g(xy)+g(x)+g(y):xy\in E(G)\right \}=\big \{c+2g(x):x\in X\subset V(G)\big \}
\end{equation}

B-1.2. If $g(y)<g(x)$ for some two vertices $x$ and $y$ of $V(G)$ and $xy\in E(G)$, by Eq.(\ref{eqa:edge-difference-sequence}), we have the edge-magic constraint
$$g(xy)+g(x)+g(y)=g(xy)+g(x)-g(y)+2g(y)=c+2g(y)$$

\textbf{B-2.} The felicitous-difference constraint.

B-2.1. If $g(x)<g(y)$ for some two vertices $x$ and $y$ of $V(G)$ and $xy\in E(G)$, we have the felicitous-difference constraint
$${
\begin{split}
|g(y)+g(x)-g(xy)|&=\big |[g(y)-g(x)+g(xy)]+2g(x)-2g(xy)\big |\\
&=|c+2g(x)-2g(xy)-2g(y)+2g(y)|\\
&=|c-2c+2g(y)|\\
&=|c-2g(y)|
\end{split}}
$$ according to Eq.(\ref{eqa:edge-difference-sequence}).

Moreover, if there is the set-ordered constraint $g(X)<g(Y)$ as $G$ is a bipartite graph with the vertex set $V(G)=X\cup Y$ and $X\cap Y=\emptyset$, then the felicitous-difference constraint color set is
\begin{equation}\label{eqa:555555}
\left \{|g(y)+g(x)-g(xy)|:xy\in E(G)\right \}=\big \{|c-2g(y)|:y\in Y\subset V(G)\big \}
\end{equation}

B-2.2. If $g(y)<g(x)$ for some two vertices $x$ and $y$ of $V(G)$ and $xy\in E(G)$, thus, we have the felicitous-difference constraint
$${
\begin{split}
|g(y)+g(x)-g(xy)|&=\big |[g(x)-g(y)+g(xy)]+2g(y)-2g(xy)\big |\\
&=|c+2g(y)-2g(xy)-2g(x)+2g(x)|\\
&=|c-2c+2g(x)|\\
&=|c-2g(x)|
\end{split}}
$$ holding true.

\textbf{B-3.} By Eq.(\ref{eqa:edge-difference-sequence}), the graceful-difference constraint
$$\big ||g(y)-g(x)|-g(xy)\big |=\big ||g(y)-g(x)|+g(xy)-2g(xy)\big |=|c-2g(xy)|,~xy\in E(G)
$$ holds true.

\vskip 0.2cm

\textbf{Case C.} Suppose that a $(p,q)$-graph $G$ admits a felicitous-difference total coloring $h$, such that the felicitous-difference constraint
\begin{equation}\label{eqa:edge-felicitous-sequence}
|h(y)+h(x)-h(xy)|=c,~xy\in E(G)
\end{equation}

\textbf{C-1.} The edge-magic constraint
$$h(xy)+h(x)+h(y)=h(x)+h(y)-h(xy)+2h(xy)=c+2h(xy),\textrm{ or }-c+2h(xy)$$
for each edge $xy\in E(G)$ holds true according to Eq.(\ref{eqa:edge-felicitous-sequence}).

\textbf{C-2.} Edge-difference constraint.

C-2.1. If $h(x)<h(y)$ for some two vertices $x$ and $y$ of $V(G)$ and $xy\in E(G)$, by Eq.(\ref{eqa:edge-felicitous-sequence}), we have the edge-difference constraint
$$h(xy)+|h(x)-h(y)|=h(xy)+h(y)-h(x)=h(xy)-h(x)-h(y)+2h(y)=c+2h(y),\textrm{ or }-c+2h(y)$$

C-2.2. If $h(y)<h(x)$ for some two vertices $x$ and $y$ of $V(G)$, we, by Eq.(\ref{eqa:edge-felicitous-sequence}), have the edge-difference constraint
$$h(xy)+|h(x)-h(y)|=h(xy)+h(x)-h(y)=h(xy)-h(x)-h(y)+2h(x)=c+2h(x),\textrm{ or }-c+2h(x)
$$ holding true.

\textbf{C-3.} Graceful-difference constraint.

C-3.1. If $h(x)<h(y)$ for some two vertices $x$ and $y$ of $V(G)$ and $xy\in E(G)$, and $h(xy)<h(y)+h(x)$, by Eq.(\ref{eqa:edge-felicitous-sequence}), we have the edge-difference constraint
$$\big ||h(y)-h(x)|-h(xy)\big |= |h(y)-h(x)-h(xy)|=|h(y)+h(x)-h(xy)-2h(x)|=|c-2h(x)|$$

Moreover, if there is the set-ordered constraint $h(X)<h(Y)$ as $G$ is a bipartite graph with the vertex set $V(G)=X\cup Y$ and $X\cap Y=\emptyset$, then the graceful-difference constraint color set is
\begin{equation}\label{eqa:555555}
\left \{\big ||h(y)-h(x)|-h(xy)\big |:xy\in E(G)\right \}=\big \{|c-2h(x)|:x\in X\subset V(G)\big \}
\end{equation}

C-3.2. By Eq.(\ref{eqa:edge-felicitous-sequence}), we have the graceful-difference constraint
$$\big ||h(y)-h(x)|-h(xy)\big |= |h(y)-h(x)-h(xy)|=|h(y)+h(x)-h(xy)-2h(x)|=c+2h(x)
$$ if $h(x)<h(y)$ for some two vertices $x$ and $y$ of $V(G)$ and $xy\in E(G)$, and $h(xy)>h(y)+h(x)$.

C-3.3. If $h(x)<h(y)$ for some two vertices $x$ and $y$ of $V(G)$ and $xy\in E(G)$, and $h(xy)<h(y)+h(x)$, we have the graceful-difference constraint
$$\big ||h(y)-h(x)|-h(xy)\big |= |h(x)-h(y)-h(xy)|=|h(y)+h(x)-h(xy)-2h(y)|=|c-2h(y)|
$$ holding true.

C-3.4. By Eq.(\ref{eqa:edge-felicitous-sequence}), we have the graceful-difference constraint
$$\big ||h(y)-h(x)|-h(xy)\big |= |h(x)-h(y)-h(xy)|=|h(y)+h(x)-h(xy)-2h(y)|=c+2h(y)
$$ if $h(x)<h(y)$ for some two vertices $x$ and $y$ of $V(G)$, and $h(xy)>h(y)+h(x)$.

\vskip 0.2cm

\textbf{Case D.} Suppose that a $(p,q)$-graph $G$ admits a graceful-difference total coloring $r$, such that the graceful-difference constraint
\begin{equation}\label{eqa:edge-graceful-sequence}
\big ||r(y)-r(x)|-r(xy)\big |=c,~xy\in E(G)
\end{equation} holds true.

\textbf{D-1.} Edge-magic constraint.

D-1.1. If $r(x)<r(y)$ for some two vertices $x$ and $y$ of $V(G)$, and $r(xy)<r(y)-r(x)$, Eq.(\ref{eqa:edge-graceful-sequence}) enables us to obtain the edge-magic constraint
$$r(xy)+r(x)+r(y)= |r(y)-r(x)-r(xy)-2r(y)|=\big |[r(y)-r(x)-r(xy)]-2r(y)\big |=|c-2r(y)|
$$

D-1.2. According to Eq.(\ref{eqa:edge-graceful-sequence}), we get the edge-magic constraint
$$r(xy)+r(x)+r(y)= |r(y)-r(x)-r(xy)-2r(y)|=\big |[r(y)-r(x)-r(xy)]-2r(y)\big |=c+2r(y)
$$ if $r(x)<r(y)$ for some two vertices $x$ and $y$ of $V(G)$, and $r(xy)>r(y)-r(x)$.

D-1.3. If $r(y)<r(x)$ for some two vertices $x$ and $y$ of $V(G)$, and $r(xy)<r(x)-r(y)$, there is the edge-magic constraint by Eq.(\ref{eqa:edge-graceful-sequence})
$$r(xy)+r(x)+r(y)= |r(x)-r(y)-r(xy)-2r(x)|=\big |[r(x)-r(y)-r(xy)]-2r(x)\big |=|c-2r(x)|
$$

D-1.4. When $r(y)<r(x)$ for some two vertices $x$ and $y$ of $V(G)$, and $r(xy)>r(x)-r(y)$, we have the edge-magic constraint
$$r(xy)+r(x)+r(y)= |r(x)-r(y)-r(xy)-2r(x)|=\big |[r(x)-r(y)-r(xy)]-2r(x)\big |=c+2r(x)$$
according to Eq.(\ref{eqa:edge-graceful-sequence})

\textbf{D-2.} Edge-difference constraint.

D-2.1. When $r(y)<r(x)$ for some two vertices $x$ and $y$ of $V(G)$ and $r(xy)<r(x)-r(y)$, we, by Eq.(\ref{eqa:edge-graceful-sequence}), get the edge-difference constraint
$$r(xy)+|r(x)-r(y)|=r(xy)+r(x)-r(y)=[r(x)-r(y)-r(xy)]+2r(xy)=c+2r(xy)
$$ holding true.

Moreover, if there is the set-ordered constraint $r(X)<r(Y)$ as $G$ is a bipartite graph with the vertex set $V(G)=X\cup Y$ and $X\cap Y=\emptyset$, then the edge-difference constraint color set is
\begin{equation}\label{eqa:555555}
\left \{r(xy)+|r(x)-r(y)|:xy\in E(G)\right \}=\big \{c+2r(xy):xy\in E(G)\big \}
\end{equation}

D-2.2. According to Eq.(\ref{eqa:edge-graceful-sequence}) we get the edge-difference constraint
$$r(xy)+|r(x)-r(y)|=r(xy)+r(x)-r(y)=|[r(x)-r(y)-r(xy)]+2r(xy)|=|c-2r(xy)|$$ when as $r(y)<r(x)$ for some two vertices $x$ and $y$ of $V(G)$, and $r(xy)>r(x)-r(y)$.

D-2.3. If $r(x)<r(y)$ for some two vertices $x$ and $y$ of $V(G)$, and $r(xy)<r(x)-r(y)$, by Eq.(\ref{eqa:edge-graceful-sequence}), there is the edge-difference constraint
$$r(xy)+|r(x)-r(y)|=r(xy)+r(y)-r(x)=[r(y)-r(x)-r(xy)]+2r(xy)=c+2r(xy)
$$ holding true.

D-2.4. When $r(x)<r(y)$ for some two vertices $x$ and $y$ of $V(G)$, and $r(xy)>r(x)-r(y)$, we get the edge-difference constraint
$$r(xy)+|r(x)-r(y)|=r(xy)+r(y)-r(x)=|[r(y)-r(x)-r(xy)]+2r(xy)|=|c-2r(xy)|$$
according to Eq.(\ref{eqa:edge-graceful-sequence}).

\textbf{D-3.} Felicitous-difference constraint.

D-3.1. If $r(x)<r(y)$ for some two vertices $x$ and $y$ of $V(G)$ and $xy\in E(G)$, by Eq.(\ref{eqa:edge-graceful-sequence}), there is the felicitous-difference constraint
$$|r(y)+r(x)-r(xy)|=\big |[r(y)-r(x)-r(xy)]+2r(x)\big |=c+2r(x)
$$ holding true.

Moreover, if there is the set-ordered constraint $r(X)<r(Y)$ as $G$ is a bipartite graph with the vertex set $V(G)=X\cup Y$ and $X\cap Y=\emptyset$, then the felicitous-difference constraint color set is
\begin{equation}\label{eqa:555555}
\left \{|r(y)+r(x)-r(xy)|:xy\in E(G)\right \}=\big \{c+2r(x):x\in X\subset E(G)\big \}
\end{equation}

D-3.2. By means of Eq.(\ref{eqa:edge-graceful-sequence}), there is the felicitous-difference constraint
$$|r(y)+r(x)-r(xy)|=\big |[r(x)-r(y)-r(xy)]+2r(y)\big |=c+2r(y)
$$ if $r(y)<r(x)$ for some two vertices $x$ and $y$ of $V(G)$ and $xy\in E(G)$.

\begin{rem}\label{rem:333333}
The magic-constraint total colorings are useful in randomly growing graph sequences for making more complex number-based strings. See a randomly growing graph sequence shown in Fig.\ref{fig:spider-edge-magic-22}.\paralled
\end{rem}

\begin{figure}[h]
\centering
\includegraphics[width=16.4cm]{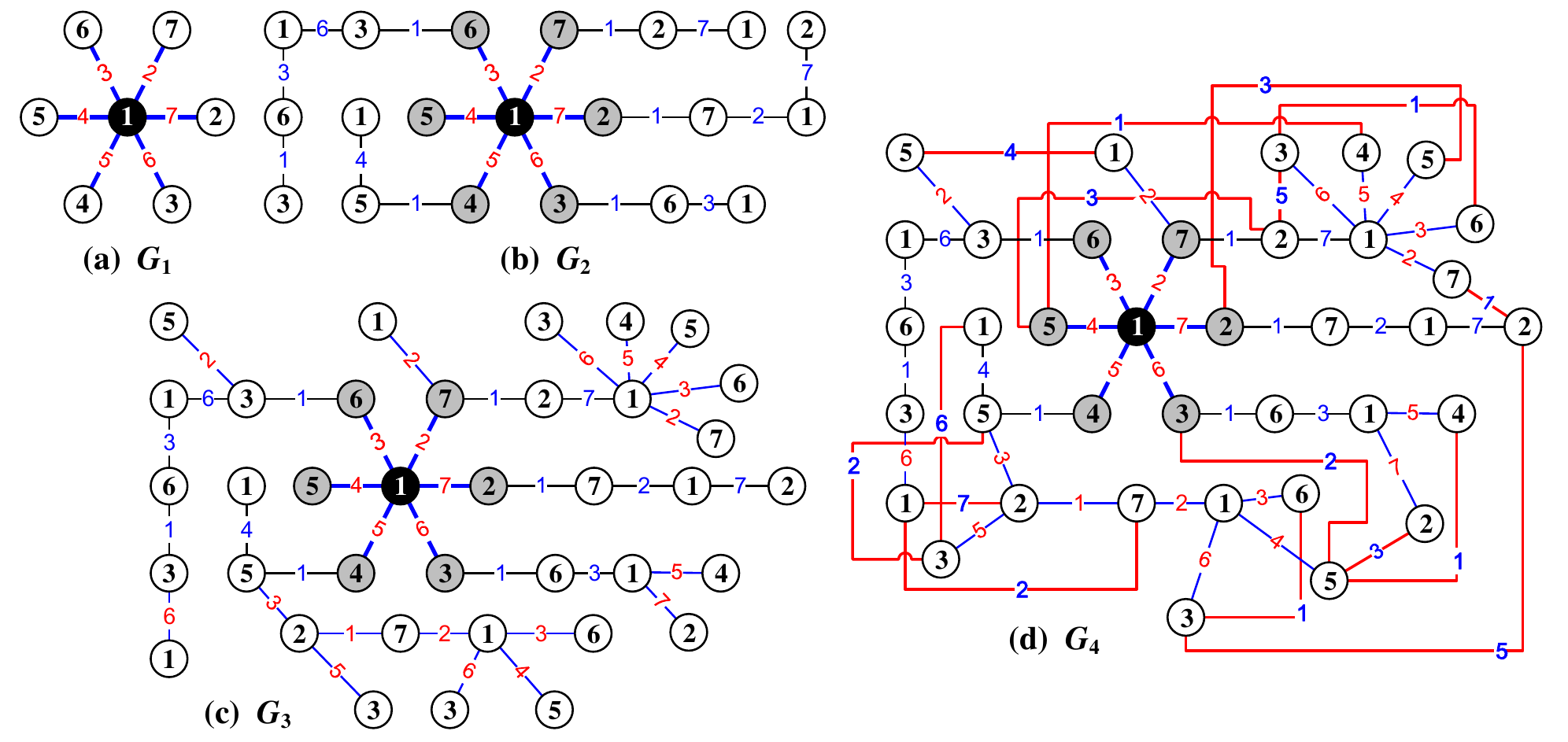}\\
\caption{\label{fig:spider-edge-magic-22}{\small A randomly growing graph sequence: The first four graphs admit the edge-magic proper total colorings holding the edge-magic constraint $f(u)+f(uv)+f(v)=10$, where (a) is a star $K_{1,6}$, called a \emph{root}; (b) is a \emph{rooted spider} $S_{1,3,3,4,3,5}$; (c) is a rooted tree; (d) is a non-planar rooted graph, cited from \cite{Bing-Yao-2020arXiv}.}}
\end{figure}

\subsection{Topcode-matrices}

A graph of graph theory is saved in computer by its \emph{adjacent matrix}, or its \emph{incident matrix}, in general. A colored graph $G$ can be saved in computer by one of its Topcode-matrix $T_{code}(G)$ defined in Definition \ref{defn:topcode-matrix-definition}, \emph{adjacent matrix} $A(G)$ defined in \cite{Bondy-2008}, other two \emph{adjacent e-value matrix} and \emph{adjacent ve-value matrix} are defined in \cite{Yao-Su-Ma-Wang-Yang-arXiv-2202-03993v1}. A colored graph drawn on planar paper can be scanned into computer by modern image recognition technology, and then switch them into various matrices for computation \cite{Yao-Su-Ma-Wang-Yang-arXiv-2202-03993v1}.

\subsubsection{Topcode-matrices}

\begin{defn}\label{defn:topcode-matrix-definition}
\cite{Yao-Sun-Zhao-Li-Yan-ITNEC-2017, Yao-Zhao-Zhang-Mu-Sun-Zhang-Yang-Ma-Su-Wang-Wang-Sun-arXiv2019} A \emph{Topcode-matrix} (or \emph{topological coding matrix}) is defined as
\begin{equation}\label{eqa:Topcode-matrix}
\centering
{
\begin{split}
T_{code}= \left(
\begin{array}{ccccc}
x_{1} & x_{2} & \cdots & x_{q}\\
e_{1} & e_{2} & \cdots & e_{q}\\
y_{1} & y_{2} & \cdots & y_{q}
\end{array}
\right)_{3\times q}=
\left(\begin{array}{c}
X\\
E\\
Y
\end{array} \right)=(X,~E,~Y)^{T}
\end{split}}
\end{equation} where the \emph{v-vector} $X=(x_1, x_2, \dots, x_q)$, the \emph{e-vector} $E=(e_1$, $e_2 $, $ \dots $, $e_q)$, and the \emph{v-vector} $Y=(y_1, y_2, \dots, y_q)$ consist of non-negative integers $e_i$, $x_i$ and $y_i$ for $i\in [1,q]$. We say $T_{code}$ to be \emph{evaluated} if there exists a function $\theta$ such that $e_i=\theta(x_i,y_i)$ for $i\in [1,q]$, and call $x_i$ and $y_i$ to be the \emph{ends} of $e_i$, denoted as $e_i=x_iy_i$, and $q$ the \emph{size} of $T_{code}$, such that $V(G)=X\cup Y$ and $E(G)=E$.\qqed
\end{defn}

\begin{example}\label{exa:8888888888}
In Fig.\ref{fig:introduction-example-11}, three colored graphs $G,T,J$ correspond to three different Topcode-matrices $T_{code}(G)$, $T_{code}(T)$ and $T_{code}(J)$ shown in Eq.(\ref{eqa:example-topcode-matrix11}), however, three colored graphs $G,T,J$ are isomorphic to each other. Five colored graphs $J_1,J_2,J_3,J_4,J_5$ shown in Fig.\ref{fig:topcode-matrices-examples} are not isomorphic to each other, but they correspond to a common Topcode-matrix $T_{code}(G[\odot ]J)$ shown in Eq.(\ref{eqa:example-topcode-matrix22}).\qqed
\end{example}

\begin{figure}[h]
\centering
\includegraphics[width=16.4cm]{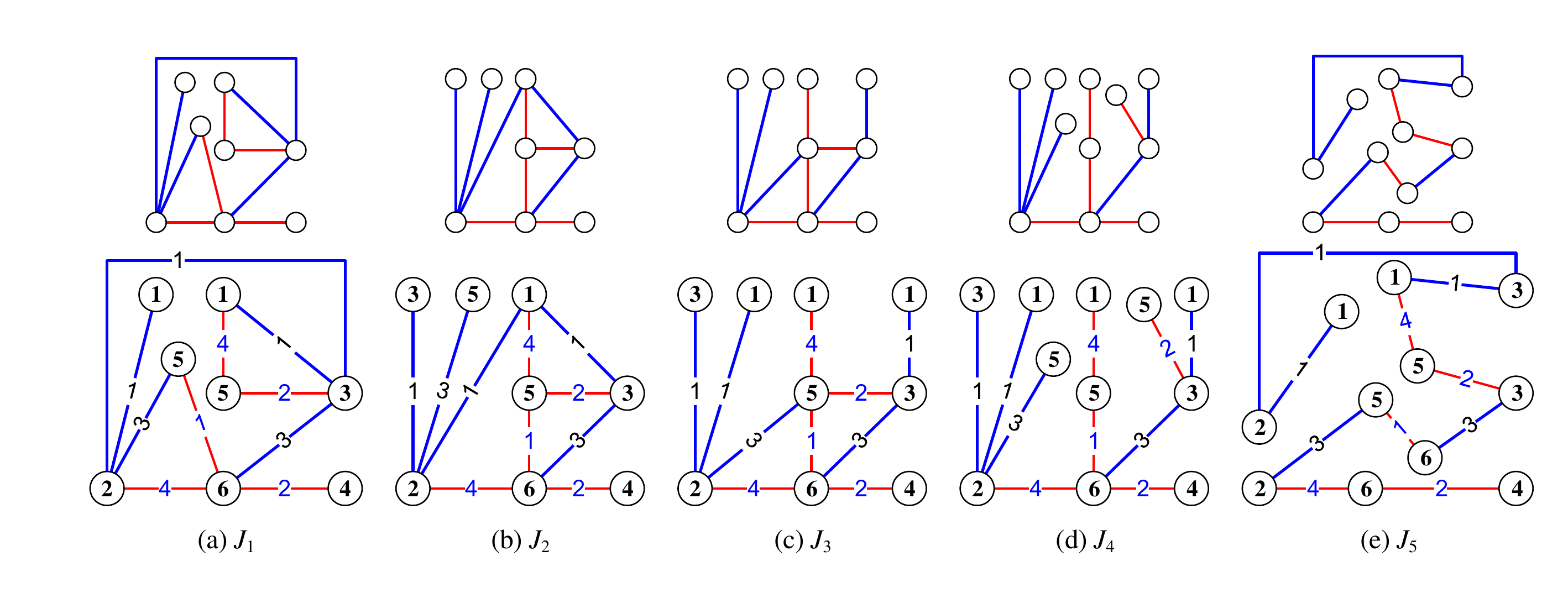}\\
\caption{\label{fig:topcode-matrices-examples}{\small Five colored graphs $J_1,J_2,J_3,J_4$ and $J_5$ are not isomorphic to each other.}}
\end{figure}

\begin{defn}\label{defn:colored-topcode-matrix}
\cite{Bing-Yao-arXiv:2207-03381} Suppose that a $(p,q)$-graph $G$ admits a $W$-constraint total coloring $f:V(G)\cup E(G)\rightarrow [a,b]$. A \emph{colored Topcode-matrix} $T_{code}(G,f)$ of the graph $G$ is defined as
\begin{equation}\label{eqa:basic-colored-Topcode-matrix}
\centering
{
\begin{split}
T_{code}(G,f)= \left(
\begin{array}{ccccc}
f(x_{1}) & f(x_{2}) & \cdots & f(x_{q})\\
f(x_{1}y_{1}) & f(x_{2}y_{2}) & \cdots & f(x_{q}y_{q})\\
f(y_{1}) & f(y_{2}) & \cdots & f(y_{q})
\end{array}
\right)_{3\times q}=\left(
\begin{array}{cccccccccccccc}
X_f\\
E_f\\
Y_f
\end{array}
\right)=(X_f,E_f,Y_f)^{T}
\end{split}}
\end{equation}\\
holding the $W$-constraint $W\langle f(x_{i}),f(x_{i}y_{i}),f(y_{i})\rangle=0$ for $i\in [1,q]$. Moreover, if $G$ is a bipartite graph with the vertex set $V(G)=X^v\cup Y^v$ and $X^v\cap Y^v=\emptyset $, we stipulate $x_{i}\in X^v$ and $y_{i}\in Y^v$ such that $X_f\cap Y_f=\emptyset $ in Eq.(\ref{eqa:basic-colored-Topcode-matrix}), where ``$W$-constraint'' is a mathematical constraint, or a group of mathematical constraints.\qqed
\end{defn}

\begin{defn} \label{defn:generalization-colored-topcode-matrix}
\cite{Bing-Yao-arXiv:2207-03381} If a $(p,q)$-graph $G$ admits a $W$-constraint vertex coloring $g:V(G)\rightarrow [\alpha , \beta]$, we have a colored Topcode-matrix $T_{code}(G,g)$ of $G$ defined as
\begin{equation}\label{eqa:vertex-coloring-Topcode-matrix}
\centering
{
\begin{split}
T^{vert}_{code}(G,g)= \left(
\begin{array}{ccccc}
g(x_{1}) & g(x_{2}) & \cdots & g(x_{q})\\
x_{1}y_{1} & x_{2}y_{2} & \cdots & x_{q}y_{q}\\
g(y_{1}) & g(y_{2}) & \cdots & g(y_{q})
\end{array}
\right)_{3\times q}=\left(
\begin{array}{cccccccccccccc}
X_g\\
E\\
Y_g
\end{array}
\right)=(X_g,E,Y_g)^{T}
\end{split}}
\end{equation} If $G$ is a bipartite graph with the vertex set $V(G)=X^v\cup Y^v$ and $X^v\cap Y^v=\emptyset $, we have $X_g\cap Y_g=\emptyset $, also, $X_g=X^v$ and $Y_g=Y^v$. And, if the $(p,q)$-graph $G$ admits a $W$-constraint edge coloring $h:E(G)\rightarrow [\lambda, \gamma]$, we have a colored Topcode-matrix $T_{code}(G,h)$ of $G$ defined as
\begin{equation}\label{eqa:edge-colored-Topcode-matrix}
\centering
{
\begin{split}
T^{edge}_{code}(G,h)= \left(
\begin{array}{ccccc}
x_{1} & x_{2} & \cdots & x_{q}\\
h(x_{1}y_{1}) & h(x_{2}y_{2}) & \cdots & h(x_{q}y_{q})\\
y_{1} & f(y_{2} & \cdots & y_{q}
\end{array}
\right)_{3\times q}=\left(
\begin{array}{cccccccccccccc}
X\\
E_h\\
Y
\end{array}
\right)=(X,E_h,Y)^{T}
\end{split}}
\end{equation} If $G$ is a bipartite graph with the vertex set $V(G)=X^v\cup Y^v$ and $X^v\cap Y^v=\emptyset $, refer to that in Definition \ref{defn:colored-topcode-matrix}.\qqed
\end{defn}

\begin{rem}\label{rem:generalization-topcode-matrix}
In Definition \ref{defn:topcode-matrix-definition}, if each of elements $x_i,e_j,y_k$ in the v-vectors and the e-vector is a \emph{set} (resp. graph, matrix, string), we call $T_{code}(G)$ \emph{set-type Topcode-matrix} (resp. \emph{graph-type Topcode-matrix}, \emph{matrix-type Topcode-matrix}, \emph{string-type Topcode-matrix}) in this article. In other word, the generalization of a Topcode-matrix is that each of elements in the Topcode-matrix is a \emph{thing} in the world, such that the Topcode-matrix brings these $3q$ things together topologically by a mathematical constraint, or a group of mathematical constraints.

Notice that the Topcode-matrix $T_{code}(G,f)$ can induce other two Topcode-matrices $T^{vert}_{code}(G,g)$ and $T^{edge}_{code}(G,h)$ defined in Definition \ref{defn:generalization-colored-topcode-matrix}, but not necessarily vice versa. Thereby, we can use this property to make asymmetric topology ciphers for real application, for example, the Topcode-matrix $T^{vert}_{code}(G,g)$ is as a \emph{public-key matrix}, the Topcode-matrix $T^{vert}_{code}(G,h)$ is as a \emph{private-key matrix}, and the Topcode-matrix $T_{code}(G,f)$ is an \emph{Topcode-matrix authentication} of the public-key matrix and the private-key matrix.\paralled
\end{rem}

\subsubsection{Pan-Topcode-matrices}

\begin{defn}\label{defn:pan-Topcode-matrix}
\cite{Yao-Wang-Ma-Su-Wang-Sun-2020ITNEC} A \emph{pan-Topcode-matrix} is defined as $P_{code}=(X_{pan}, E_{pan}, Y_{pan}~)^{T}$ with three vectors
$$X_{pan}=(\alpha_1, \alpha_2, \dots , \alpha_q),~E_{pan}=(\gamma_1, \gamma_2, \dots , \gamma_q),~Y_{pan}=(\beta_1, \beta_2, \dots , \beta_q)$$ and $\alpha_j,\beta_j$ are the ends of $\gamma_j$. If there exits a constraint $W$ such that $\gamma_j=W\langle \alpha_j,\beta_j\rangle$ for each $j\in [1,q]$, then the pan-Topcode-matrix $P_{code}$ is \emph{$W$-constraint valued}.\qqed
\end{defn}

\begin{defn} \label{defn:evaluated-topcode-matrix}
$^*$ If $x_i:=\alpha_i$, $e_i:=\gamma_i$ and $y_i:=\beta_i$ in a Topcode-matrix $T_{code}$ defined in Definition \ref{defn:topcode-matrix-definition}, we get another Topcode-matrix $T^{evalu}_{code}=\big (X_{(:)},~E_{(:)},~Y_{(:)}\big )^{T}$ withe three vectors
$$X_{(:)}=(\alpha_1, \alpha_2, \dots, \alpha_q),~E_{(:)}=(\gamma_1, \gamma_2, \dots, \gamma_q)\textrm{ and }Y_{(:)}=(\beta_1, \beta_2, \dots, \beta_q)
$$ We call $T^{evalu}_{code}$ \emph{assignment Topcode-matrix} of $T_{code}$, and denote this face as $T_{code}:=T^{evalu}_{code}$. Moreover, there are Topcode-matrices $T^{evalu}_{code}(1)$, $T^{evalu}_{code}(2)$, $\dots$, $T^{evalu}_{code}(m)$ holding
\begin{equation}\label{eqa:555555}
T^{evalu}_{code}(k):=T^{evalu}_{code}(k+1)
\end{equation} for $k\in [1,m-1]$.\qqed
\end{defn}

\begin{rem}\label{rem:333333}
In Definition \ref{defn:pan-Topcode-matrix}, the elements $\alpha_i,\gamma_i,\beta_i$ of the pan-Topcode-matrix $P_{code}$ are graphs, matrices, vectors, strings, formulae, articles, any things if there are connections be tween them, then the pan-Topcode-matrices show these related things in topological structures.

The generalization $T^{gener}_{code}$ of a Topcode-matrix $T_{code}$ is that each of elements in the Topcode-matrix is a \emph{thing} in the world, such that the Topcode-matrix $T^{gener}_{code}$ brings these $3q$ things together topologically by a mathematical constraint, or a group of mathematical constraints for getting a complete ``mathematical story''.\paralled
\end{rem}

\subsubsection{Connections between matrices}

About Topcode-matrix algebra, it can refer to \cite{Bing-Yao-2020arXiv} and \cite{Yao-Zhao-Zhang-Mu-Sun-Zhang-Yang-Ma-Su-Wang-Wang-Sun-arXiv2019}.

For a graph $G$ of $n$ vertices and a total coloring $f:V(G)\cup E(G)\rightarrow [a,b]$, the graph $G$ corresponds to its own \emph{adjacent matrix} $A(G)=(a_{i,j})_{n\times n}$ defined as
\begin{equation}\label{eqa:555555}
\centering
{
\begin{split}
A(G)= \left(
\begin{array}{cccccc}
 a_{1,1} & a_{1,2} & \cdots & a_{1,n}\\
 a_{2,1} & a_{2,2} & \cdots & a_{2,n}\\
 \cdots & \cdots & \cdots & \cdots\\
 a_{n,1} & a_{n,2} & \cdots & a_{n,n}
\end{array}
\right)=(A_1,A_2,\dots ,A_n)^T_{n\times n}
\end{split}}
\end{equation} where $A_i=(a_{i,1},a_{i,2},\dots ,a_{i,n})$ with $i\in [1,n]$, $a_{i,j}=1$ if $x_{i}x_{j}\in E(G)$, otherwise $a_{i,j}=0$.

We color the elements of $A(G)$ by a total coloring $f$ of the graph $G$, the resultant matrix
\begin{equation}\label{eqa:555555}
\centering
{
\begin{split}
f(A(G))= \left(
\begin{array}{cccccc}
 f(a_{1,1}) & f(a_{1,2}) & \cdots & f(a_{1,n})\\
 f(a_{2,1}) & f(a_{2,2}) & \cdots & f(a_{2,n})\\
 \cdots & \cdots & \cdots & \cdots\\
 f(a_{n,1}) & f(a_{n,2}) & \cdots & f(a_{n,n})
\end{array}
\right)
\end{split}}
\end{equation} is called \emph{colored adjacent matrix} of the adjacent matrix $A(G)$, such that $f(a_{i,j})=f(x_{i}x_{j})$ if $x_{i}x_{j}\in E(G)$ and $a_{i,j}=1$ in $A(G)$, otherwise $f(a_{i,j})=0$. About the colored adjacent matrix $f(A(G))$, we have

(i) The colored adjacent matrix $f(A(G))$ corresponds to the unique graph $G$, however, it has no information of vertex colors of the graph $G$.

(ii) The \emph{characteristic equation} $F(\lambda)=0$ is obtained from $F(\lambda)=|\lambda I-f(A(G))|$, where $I$ is the \emph{identity matrix} of order $n$. Here, a matrix of order $n$ whose elements on the main diagonal are all 1's and all other elements are 0's is called an \emph{identity matrix} of order $n$ (also, \emph{unit matrix}), denoted as $I$. We do not know the theory and application of the colored adjacent matrix $f(A(G))$.

(iii) The colored adjacent matrix $f(A(G))$ is a connection between the adjacent matrix $A(G)$ (corresponds to the unique graph $G$) and the Topcode-matrix $T_{code}(G,f)_{3\times n}$ (corresponds to tow or more different colored graphs, in general).

(iv) Moreover, we have the \emph{colored ve-matrix} $A_{code}(G)=(V_f,X_1,X_2,\dots ,X_n)^T_{(n+1)\times (n+1)}$, where
$$V_f=(0,f(x_1),f(x_2)\dots ,f(x_n)),~x_i\in V(G);~X_i=(f(x_i),f(a_{i,1}),f(a_{i,2}),\dots ,f(a_{i,n})),~i\in [1,n]
$$ and $f(a_{i,j})=f(x_{i}x_{j})$ if $x_{i}x_{j}\in E(G)$ and $a_{i,j}=1$ in $A(G)$, otherwise $f(a_{i,j})=0$.

\begin{equation}\label{eqa:555555}
\centering
{
\begin{split}
A_{code}(G,f)= \left(
\begin{array}{cccccc}
0 & f(x_1) & f(x_2) & \cdots & f(x_n)\\
f(x_1) & f(a_{1,1}) & f(a_{1,2}) & \cdots & f(a_{1,n})\\
f(x_2) & f(a_{2,1}) & f(a_{2,2}) & \cdots & f(a_{2,n})\\
\cdots & \cdots & \cdots & \cdots & \cdots\\
f(x_n) & f(a_{n,1}) & f(a_{n,2}) & \cdots & f(a_{n,n})
\end{array}
\right)_{(n+1)\times (n+1)}
\end{split}}
\end{equation} which is a connection between the adjacent matrix $A(G)$ and the Topcode-matrix $T_{code}(G,f)_{3\times n}$, and contains all colors of vertices and edges of the graph $G$ under the total coloring $f$.

\begin{example}\label{exa:8888888888}
For the colored graph $G=H_{4147}$ shown in Fig.\ref{fig:introduction-example-11} (a), we have the following four matrices of the graph $G=H_{4147}$:

\begin{equation}\label{eqa:example-four-matrices11}
\centering
{
\begin{split}
A(G)= \left(
\begin{array}{ccccccc}
0 & 0 & 0 & 0 & 1 & 0\\
0 & 0 & 0 & 0 & 0 & 1\\
0 & 0 & 0 & 0 & 1 & 0\\
0 & 0 & 0 & 0 & 0 & 1\\
1 & 0 & 1 & 0 & 0 & 1\\
0 & 1 & 0 & 1 & 1 & 0
\end{array}
\right),~f(A(G))= \left(
\begin{array}{ccccccc}
0 & 0 & 0 & 0 & 4 & 0\\
0 & 0 & 0 & 0 & 0 & 4\\
0 & 0 & 0 & 0 & 2 & 0\\
0 & 0 & 0 & 0 & 0 & 2\\
4 & 0 & 2 & 0 & 0 & 1\\
0 & 4 & 0 & 2 & 1 & 0
\end{array}
\right)
\end{split}}
\end{equation}

\begin{equation}\label{eqa:example-four-matrices22}
\centering
{
\begin{split}
A_{code}(G,f)= \left(
\begin{array}{ccccccc}
0 & 1 & 2 & 3 & 4 & 5 & 6\\
1 & 0 & 0 & 0 & 0 & 1 & 0\\
2 & 0 & 0 & 0 & 0 & 0 & 1\\
3 & 0 & 0 & 0 & 0 & 1 & 0\\
4 & 0 & 0 & 0 & 0 & 0 & 1\\
5 & 1 & 0 & 1 & 0 & 0 & 1\\
6 & 0 & 1 & 0 & 1 & 1 & 0
\end{array}
\right),~T_{code}(G,f)= \left(
\begin{array}{cccccc}
1 & 3 & 6 & 6 & 6\\
4 & 2 & 1 & 4 & 2\\
5 & 5 & 5 & 2 & 4
\end{array}
\right)
\end{split}}
\end{equation}Obviously, the Topcode-matrix $T_{code}(G,f)$ in four matrices of the graph $G=H_{4147}$ takes up the least amount of computer storage space, but $T_{code}(G,f)$ corresponds to two or more graphs. However, the other three matrices can induce number-based strings with longer bytes.\qqed
\end{example}

\subsection{The advantages of topological encryption}

Graphs of graph theory are ubiquitous in the real world, represent objects and their relationships such as social networks, e-commerce networks, biology networks and traffic networks and many areas of science such as Deep Learning, Graph Neural Network, Graph Networks (Ref. \cite{Battaglia-27-authors-arXiv1806-01261v2}).

The authors in \cite{Yao-Wang-2106-15254v1} concluded some advantages of Topsnut-gpws of topological coding, where Topsnut-gpws is the abbreviation of the sentence ``Graphical passwords consisted of topological structure and mathematical constraints'', as follows:
\begin{asparaenum}[\textbf{Advan}-1.]
\item \textbf{Related with two or more different mathematical areas}. Topsnut-gpws are made of topological structures and mathematical constrains. Here, ``mathematical constrains'' are related with Number theory, Set theory, Probability theory, Calculus, Algebra, Abstract algebra, Linear algebra, Combinatorics, Discrete mathematics, Graph theory, Cryptography, Computer science and Information theory \emph{etc}. However, ``topological structures'' differ from the above ``hard restrictions'' and are ``mathematical expressions'' belonging to Graph Theory, a mathematical branch. These advantages will produce Asymmetric Topology Cryptosystem of topological coding for resisting the quantum computation and intelligent attacks equipped with quantum computer in future ear of quantum computer.
\item \textbf{Not pictures}. Topsnut-gpws run fast in communication networks because they are saved in computer by popular matrices rather than pictures and photos. With the present image recognition technology, it is easy to scan Topsnut-gpws drawn on the plane into computer, and obtain the adjacent matrices and other matrices of the Topsnut-gpws.
\item \textbf{Lots of graphs, colorings and labelings.} There are enormous numbers of graphs, graph colorings and labelings in graph theory. And new graph colorings (resp. labelings) come into being everyday. The number of one $W$-constrain different labelings (resp. colorings) of a graph maybe large, and no method is reported to find out all of such labelings (resp. colorings) for a graph. Refer to two graph numbers of graphs with 23 vertices and 24 vertices $G_{23}\approx 2^{179}$ and $G_{24}\approx 2^{197}$, and two comprehensive articles \cite{Gallian2021} and \cite{Yao-Wang-2106-15254v1}.
\item \textbf{Diversity of \emph{asymmetric topology cryptography}.} Topological authentications realize the coexistence of two or more labelings and colorings on a graph, which leads to the problem of multiple labeling (resp. multiple coloring) decomposition of graphs, and brings new research objects and new problems to mathematics. One graph can produce more public-key graphs since there are hundreds of colorings and labelings, so there are more private-key graphs. There are: one public-key corresponds two or more private-key graphs, and more public-key graphs correspond one or more private-key graphs for asymmetric encryption algorithm.
\item \textbf{Composability}. There are connections between topological authentications and other type of passwords. For example, small circles in Topsnut-gpws can be equipped with fingerprints and other biological information requirements, and the users' pictures can be embedded in small circles, greatly reflects personalization. Many labelings of trees are convertible from each other \cite{Yao-Liu-Yao-2017}. Tree-type structures can adapt to a large number of labelings and colorings, in addition to graceful labeling, no other labelings reported that were established in those tree-type structures having smaller vertex numbers by computer, almost no computer proof. Because construction methods are complex, this means that using computers to break down Topsnut-gpws will be difficult greatly.
\item \textbf{Simplicity and convenience.} Topsnut-gpws are suitable for people who need not learn new rules and are allowed to use their private knowledge in making Topsnut-gpws for the sake of remembering easily. For example, Chinese characters (Hanzis) are naturally topological structures to produce Hanzi-graphs for topological coding. Chinese people can generate Hanzi-graphs simply by speaking and writing Chinese directly.
\item \textbf{Irreversibility.} Topsnut-gpws can generate quickly number-based (resp. text-based) strings with bytes as long as desired, but these strings can not reconstruct the original Topsnut-gpws. The confusion of number-based (resp. text-based) strings and Topsnut-gpws can't be erased, as our knowledge.
\item \textbf{Computational security.} There are many non-polynomial algorithms in making Topsnut-gpws. For example, drawing non-isomorphic graphs is very difficult and non-polynomial. For a given graph, finding out all possible colorings (resp. labelings) are impossible, since these colorings (resp. labelings) are massive data, and many graph problems have been proven to be NP-complete, or NP-hard.
\item \textbf{Provable security.} There are many longstanding mathematical conjectures (also, \emph{open problems}) in graph labelings and colorings, such as the famous graceful tree conjecture, odd-graceful tree conjecture, total coloring conjecture, and many open problems of topological structures are: Hamilton graph determination, $KT$-conjecture, Kelly-Ulam's Reconstruction Conjecture \emph{etc.}
\end{asparaenum}

\section{$[0,9]$-strings and $[0,9]$-string groups}

A $[0,9]$-\emph{string} is a \emph{number-based string} $s=c_{1}c_{2}\cdots c_{m}$ with $c_{j}\in [0,9]$ for $j\in [1,m]$. We will discuss some algebraic operations on $[0,9]$-strings and $[0,9]$-string groups. Moreover, we will introduce super-strings in this subsection.

\subsection{Algebraic operations based on $[0,9]$-strings}

\begin{defn} \label{defn:0-9-string-4-operation}
$^*$ For two $[0,9]$-strings $s_i=c_{i,1}c_{i,2}\cdots c_{i,m}$ with $c_{i,j}\in [0,9]$ for $j\in [1,m]$ and $i=1,2$, we define

1. The \emph{addition operation} ``$[+]$'' is defined by
$$
s_1[+]s_2=(c_{1,1}+c_{2,1})(c_{1,2}+c_{2,2})\cdots (c_{1,m}+c_{2,m})
$$ and define $s_1[+]s_2~(\bmod~9)=b_1b_2\cdots b_m$ with $b_j=c_{1,j}+c_{2,j}$ if $c_{1,j}+c_{2,j}\leq 9$, otherwise $b_j=c_{1,j}+c_{2,j}~(\bmod~9)$ for $j\in [1,m]$. Obviously,
$s_1[+]s_2\neq s_1[+]s_2~(\bmod~9)$ in general.

2. The \emph{subtraction operation} ``$[-]$'' is defined by
$$
s_1[-]s_2=(c_{1,1}-c_{2,1})(c_{1,2}-c_{2,2})\cdots (c_{1,m}-c_{2,m})
$$ and define $s_1[-]s_2~(\bmod~9)=d_1d_2\cdots d_m$ with $d_j=c_{1,j}-c_{2,j}$ if $0\leq c_{1,j}-c_{2,j}\leq 9$, otherwise $d_j=c_{1,j}-c_{2,j}~(\bmod~9)$ for $j\in [1,m]$. Clearly,
\begin{equation}\label{eqa:555555}
s_1[-]s_2\neq s_2[-]s_1,\quad s_1[-]s_2~(\bmod~9)\neq s_2[-]s_1~(\bmod~9)
\end{equation} in general.

3. The \emph{complementary} of a $[0,9]$-string $s=c_1c_2\cdots c_m$ with $c_i\in [0,9]$ is defined by
$$\overline{s}=(9-c_1)(9-c_2)\cdots (9-c_m)=a_1a_2\cdots a_m$$

4. The \emph{inverse} of a $[0,9]$-string $s=c_1c_2\cdots c_m$ with $c_i\in [0,9]$ is defined by $s^{-1}=c_mc_{m-1}\cdots c_2c_1$.

5. The \emph{number multiplication} is defined by $$k[\bullet ]s~(\bmod~9)=(k\cdot c_1)(k\cdot c_2)\cdots (k\cdot c_m)~(\bmod~9)
$$ with $c_i\in [0,9]$ and integer $k\geq 1$.\qqed
\end{defn}

\begin{prop}\label{thm:666666}
$^*$ For a $[0,9]$-string $s=c_1c_2\cdots c_m$ with $c_i\in [0,9]$, by Definition \ref{defn:0-9-string-4-operation}, we have
\begin{equation}\label{eqa:555555}
\overline{s}^{-1}=(9-c_m)(9-c_{m-1})\cdots (9-c_2)(9-c_1)=a_ma_{m-1}\cdots a_2a_1
\end{equation} and
\begin{equation}\label{eqa:555555}
{
\begin{split}
s[+]s^{-1}&=(c_1+c_m)(c_2+c_{m-1})\cdots (c_{m-1}+c_2)(c_m+c_1),\\
\overline{s}[+]\overline{s}^{-1}&=(a_1+a_m)(a_2+a_{m-1})\cdots (a_{m-1}+a_2)(a_m+a_1),\\
\overline{s}[+]s^{-1}&=(a_1+c_m)(a_2+c_{m-1})\cdots (a_{m-1}+c_2)(a_m+c_1),\\
s[+]\overline{s}^{-1}&=(c_1+a_m)(c_2+a_{m-1})\cdots (c_{m-1}+a_2)(c_m+a_1)
\end{split}}
\end{equation}
\end{prop}

\begin{center}
Table-1. $[0,9]$-strings.\\[6pt]
\begin{tabular}{r|c|c|c|c||c|c|c}
$i$ & $s_i$ & $s^{-1}_i$ & $\overline{s}_i$ & $\overline{s}^{-1}_i$ & $s_i[+]s^{-1}_i$ & $s_i[-]s^{-1}_i$ & $\overline{s}_i[+]\overline{s}^{-1}_i$\\
\hline
1 & 1013412 & 2143101 & 8986587 & 7856898 & 3156513 & 9970311 & 5732375\\
2 & 2124523 & 3254212 & 7875476 & 6745787 & 5378735 & 9970311 & 3510153\\
3 & 3235634 & 4365323 & 6764365 & 5634676 & 7590957 & 9970311 & 1398931\\
4 & 4346745 & 5476434 & 5653254 & 4523565 & 9712179 & 9970311 & 9176719\\
5 & 5457856 & 6587545 & 4542143 & 3412454 & 1934391 & 9970311 & 7954597\\
6 & 6568967 & 7698656 & 3431032 & 2301343 & 3156513 & 9970311 & 5732375\\
7 & 7679078 & 8709767 & 2320921 & 1290232 & 5378735 & 9970311 & 3510153\\
8 & 8780189 & 9810878 & 1219810 & 0189121 & 7590957 & 9970311 & 1398931\\
9 & 9891290 & 0921989 & 0108709 & 9078010 & 9712179 & 9970311 & 9176719\\
10 & 0902301 & 1032090 & 9097698 & 8967909 & 1934391 & 9970311 & 7954597
\end{tabular}
\end{center}

\begin{prop}\label{thm:666666}
$^*$ According to the basic operations on $[0,9]$-strings defined in Definition \ref{defn:0-9-string-4-operation}, there are the following formulae:
\begin{asparaenum}[\textbf{\textrm{Form}}-1. ]
\item $s[+]\overline{s}=(c_1+a_1)(c_2+a_2)\cdots (c_m+a_m)=99\cdots 9$.
\item $s^{-1}[+]\overline{s^{-1}}=99\cdots 9$.
\item $(\overline{s}[-]s)[+](s[-]\overline{s})=99\cdots 9$.
\item $(s^{-1}[-]\overline{s}^{-1})[+](\overline{s}^{-1}[-]s^{-1})=99\cdots 9$.
\item $(s[+]\overline{s})[-]\overline{s}=s$, $(s[+]\overline{s})[-]s=\overline{s}$.
\item $\overline{(\overline{s})}=s$, $\overline{s}^{-1}=\overline{(s^{-1})}$, $\overline{(\overline{s}^{-1})}=s^{-1}$.
\item $(s[-]s^{-1})[+](s^{-1}[-]s)~(\bmod~9)=99\cdots 9$, since

\qquad $s[-]s^{-1}=(c_1-c_m)(c_2-c_{m-1})\cdots (c_{m-1}-c_2)(c_m-c_1)$ and

\qquad $s^{-1}[-]s=[10-(c_1-c_m)][10-(c_2-c_{m-1})]\cdots [10-(c_{m-1}-c_2)][10-(c_m-c_1)]$.
\end{asparaenum}
\end{prop}

\begin{example}\label{exa:8888888888}
A $[0,9]$-string $s=$1013412 has its own inverse $s^{-1}=$2143101, and its own complementary $\overline{s}=(9-1)(9-0)(9-1)(9-3)(9-4)(9-1)(9-2)=$8986587. Moreover, we have the complementary $\overline{s}^{-1}=$7856898. By the additive operation, we get

$s[+]\overline{s}=$9999999, $s[+]s^{-1}=(1+2)(0+1)(1+4)(3+3)(4+1)(1+0)(2+1)=$3156513,

$\overline{s}[+]\overline{s}^{-1}=(8+7)(9+8)(8+5)(6+6)(5+8)(8+9)(7+8)=$15171312131715,\\
however, $\overline{s}[+]\overline{s}^{-1}~(\bmod~9)=6843486$; refer to Table-1.\qqed
\end{example}

\subsection{$[0,9]$-string groups}

\begin{defn} \label{defn:111111}
$^*$ By the algebraic operations on number-based strings, we define a \emph{number-based string sequence}
$$
S_{eq}=\{s_i\}^{9}_{i=1}=\{s_i=c_{i,1}c_{i,2}\cdots c_{i,m}:c_{i,j}\in [0,9]\}^{9}_{i=1}
$$ holds $c_{i,j}=1+c_{i-1,j}~(\bmod~9)$ for $i\in [2,9]$ and $j\in [1,m]$.

Similarly, other \emph{number-based string sequences} are $\overline{S}_{eq}=\{\overline{s}_i\}^{9}_{i=1}$, $S^{-1}_{eq}=\{s^{-1}_i\}^{9}_{i=1}$, $\overline{S}^{-1}_{eq}=\{\overline{s}^{-1}_i\}^{9}_{i=1}$, $S^{[+]}_{eq}=\{s_i[+]s^{-1}_i\}^{9}_{i=1}$ and $\overline{S}^{[+]}_{eq}=\{\overline{s}_i[+]\overline{s}^{-1}_i\}^{9}_{i=1}$ according to the algebraic operations on number-based strings.\qqed
\end{defn}

We have an every-zero $[0,9]$-string group $\{F^+_{10}(X);\oplus\ominus(\bmod~9)\}$ for each $X\in \{S_{eq},\overline{S}_{eq},S^{-1}_{eq},\overline{S}^{-1}_{eq}$, $S^{[+]}_{eq}$, $\overline{S}^{[+]}_{eq}\}$ defined in Definition \ref{defn:0-9-string-groups111}.

\begin{defn} \label{defn:0-9-string-groups111}
$^*$ If a set
$$S_{tring}(m,n)=\{s_i\}^{n}_{i=1}=\{s_{i}=c_{i,1}c_{i,2}\cdots c_{i,m}:c_{i,j}\in [0,9]\}^{n}_{i=1},~n\geq 10
$$ holds
\begin{equation}\label{eqa:555555}
(c_{i,r}+ c_{j,r})- c_{k,r}~(\bmod~9)=c_{\mu}\in [0,9],~r\in [1,m]
\end{equation} true, then we get an operation ``$\oplus \ominus$'' on the set $S_{tring}(m,n)$ as follows
\begin{equation}\label{eqa:555555}
s_{i}[\oplus \ominus_k]s_{j}:=(s_{i}[+] s_{j})[-] s_{k}=s_{\lambda}\in S_{tring}(m,n)
\end{equation} with the index $\lambda=i+j-k~(\bmod~n)$ for any preappointed \emph{zero} $s_{k}\in S_{tring}(m,n)$, we have defined an \emph{every-zero $[0,9]$-string group} $\{F^+_n(S_{tring}(m,n));\oplus\ominus(\bmod~n)\}$ of order $n$.

Another \emph{every-zero $[0,9]$-string group} $\{F^-_n(S_{tring}(m,n));\ominus \oplus(\bmod~n)\}$ of order $n$ is defined as follows: If the set $S_{tring}(m,n)$ holds
\begin{equation}\label{eqa:555555}
c_{k,r}-(c_{i,r}+c_{j,r})~(\bmod~9)=c_{\eta}\in [0,9],~r\in [1,m]
\end{equation} true, then there is an operation ``$\ominus \oplus $'' defined on the set $S_{tring}(m,n)$ as
\begin{equation}\label{eqa:555555}
s_{i}[\ominus \oplus_k]s_{j}:= s_{k}[-](s_{i}[+] s_{j})=s_{\tau}\in S_{tring}(m,n)
\end{equation} with the index $\tau=i+j-k~(\bmod~n)$ for any preappointed \emph{zero} $s_{k}\in S_{tring}(m,n)$.\qqed
\end{defn}

\begin{example}\label{exa:8888888888}
By Table-1, we have

\textbf{1.1.} $(s_{5}[+] s_{6})[-] s_{2}~(\bmod~9)=s_{9}$, since

$(c_{5,r}+ c_{6,r})- c_{2,r}~(\bmod~9)=989(11)(12)9(10)~(\bmod~9)=9891290=s_{9}$.

\textbf{1.2.} $(s_{1}[+] s_{2})[-] s_{8}~(\bmod~9)=s_{5}$, since

$(c_{1,r}+ c_{2,r})- c_{8,r}~(\bmod~9)=(-5)(-6)(-5)78(-5)(-4)~(\bmod~9)=5457856=s_{5}$.

\textbf{1.3.} $(s_{1}[+] s_{9})[-] s_{7}~(\bmod~9)=s_{3}$, since

$(c_{1,r}+ c_{9,r})- c_{7,r}~(\bmod~9)=323(-5)63(-6)~(\bmod~9)=3235634=s_{3}$.

\textbf{2.1.} $(s^{-1}_{5}[+] s^{-1}_{6})[-] s^{-1}_{2}~(\bmod~9)=s^{-1}_{9}$, since

$(c_{5,r}+ c_{6,r})- c_{2,r}~(\bmod~9)=(10)9(12)(11)989~(\bmod~9)=0921989=s^{-1}_{9}$.

\textbf{2.2.} $(s^{-1}_{1}[+] s^{-1}_{9})[-] s^{-1}_{7}~(\bmod~9)=s^{-1}_{3}$, since

$(c_{1,r}+ c_{9,r})- c_{7,r}~(\bmod~9)=(-6)36(-5)323~(\bmod~9)=4365323=s^{-1}_{3}$.

\textbf{2.3.} $(s^{-1}_{1}[+] s^{-1}_{2})[-] s^{-1}_{8}~(\bmod~9)=s^{-1}_{5}$, since

$(c_{1,r}+ c_{2,r})- c_{8,r}~(\bmod~9)=(-4)(-5)87(-5)(-6)(-5)~(\bmod~9)=6587545=s^{-1}_{5}$.

\textbf{3.1.} $(\overline{s}_{5}[+] \overline{s}_{6})[-] \overline{s}_{2}~(\bmod~9)=\overline{s}_{9}$, since

$(c_{5,r}+ c_{6,r})- c_{2,r}~(\bmod~9)=010(-2)(-3)0(-1)~(\bmod~9)=0108709=\overline{s}_{9}$.

\textbf{3.2.} $(\overline{s}_{1}[+] \overline{s}_{2})[-] \overline{s}_{8}~(\bmod~9)=\overline{s}_{5}$, since

$(c_{1,r}+ c_{2,r})- c_{8,r}~(\bmod~9)=(14)(15)(14)21(14)(13)~(\bmod~9)=4542143=\overline{s}_{5}$.

\textbf{3.3.} $(\overline{s}_{1}[+] \overline{s}_{9})[-] \overline{s}_{7}~(\bmod~9)=\overline{s}_{3}$, since

$(c_{1,r}+ c_{9,r})- c_{7,r}~(\bmod~9)=676(14)36(15)~(\bmod~9)=6764365=\overline{s}_{3}$.\qqed
\end{example}

\begin{thm}\label{thm:666666}
$^*$ There are two every-zero $[0,9]$-string groups
$$
\{F^+_n(S_{tring}(m,n));\oplus\ominus(\bmod~9)\}\textrm{ and }\{F^-_n(S_{tring}(m,n));\ominus \oplus(\bmod~9)\}
$$ if the set $S_{tring}(m,n)$ has $10^m$ elements.
\end{thm}

\begin{example}\label{exa:8888888888}
About the $[0,9]$-string $d_1=142857$, we have Table-2 as follows
\begin{center}
Table-2. $[0,9]$-string $d_1=142857$ $(\bmod~9)$.\\[6pt]
\begin{tabular}{r|c|c|c|c||c|c|c|c|c}
$i$ & $d_i$ & $d^{-1}_i$ & $\overline{d}_i$ & $\overline{d}^{-1}_i$ & $d_i[+]d^{-1}_i$ & $\overline{d}_i[+]\overline{d}^{-1}_i$ & $d_i[-]\overline{d}_i$ & $d_i[-]\overline{d}^{-1}_i$ & $d^{-1}_i[-]\overline{d}^{-1}_i$\\
\hline
1 & \underline{142857} & \underline{758241} & \underline{857142} & \underline{241758} & 891198 & 198891 & \underline{284715} & 891198 & \underline{517482}\\
2 & 253968 & 869352 & 746031 & 130647 & 123321 & 876678 & 416937 & 123321 & 739614\\
3 & 364179 & 971463 & 635820 & 028536 & 345543 & 654456 & 638259 & 345543 & 952836\\
4 & \underline{475281} & \underline{182574} & \underline{524718} & \underline{817425} & 567765 & 432234 & \underline{851472} & 567765 & \underline{274158}\\
5 & 586392 & 293685 & 413607 & 706314 & 789987 & 219912 & 173694 & 789987 & 496371\\
6 & 697413 & 314796 & 302586 & 685203 & 912219 & 987789 & 395826 & 921129 & 628593\\
7 & \underline{718524} & \underline{425817} & \underline{281475} & \underline{574182} & 234432 & 765567 & \underline{527148} & 234432 & \underline{841725}\\
8 & 829635 & 536928 & 170364 & 463071 & 456654 & 543345 & 749361 & 456654 & 163947\\
9 & 931746 & 647139 & 068253 & 352860 & 678876 & 321123 & 962853 & 678876 & 385269
\end{tabular}
\end{center}
where there are the operation strings $d_i[+]\overline{d}_i=999999$, $d^{-1}_i[+]\overline{d}^{-1}_i=999999$, $d_i[+]\overline{d}^{-1}_i=383616$, $d_i[-]d^{-1}_i=383616$ and $d^{-1}_i[+]\overline{d}_i=616383$ for $i\in [1,9]$.

\vskip 0.4cm

By Table-2, we can observe
$$d^{-1}_i[-]d_i=\overline{d}_i[+]\overline{d}^{-1}_i,~(d_i[-]d^{-1}_i)^{-1}=d^{-1}_i[+]\overline{d}_i\textrm{ and }(d_i[-]d^{-1}_i)[+](d^{-1}_i[+]\overline{d}_i)=99\cdots 9
$$ Four sets $$\textbf{\textrm{D}}=\big \{d_i:i\in [1,9]\big \},~\textbf{\textrm{D}}^{-1}=\big \{d^{-1}_i:i\in [1,9]\big \},~\overline{\textbf{\textrm{D}}}=\big \{\overline{d}_i:i\in [1,9]\big \} \textrm{ and } \overline{\textbf{\textrm{D}}}^{-1}=\big \{\overline{d}^{-1}_i:i\in [1,9]\big \}
$$ enables us to get other five sets below

$\textbf{\textrm{D}}[+]\textbf{\textrm{D}}^{-1}=\big \{d_i[+]d^{-1}_i:i\in [1,9]\big \}$, $\overline{\textbf{\textrm{D}}}[+]\overline{\textbf{\textrm{D}}}^{-1}=\big \{\overline{d}_i[+]\overline{d}^{-1}_i:i\in [1,9]\big \}$,

$\textbf{\textrm{D}}[-]\overline{\textbf{\textrm{D}}}=\big \{d_i[-]\overline{d}_i:i\in [1,9]\big \}$, $\textbf{\textrm{D}}^{-1}[-]\overline{\textbf{\textrm{D}}}^{-1}=\big \{d^{-1}_i[-]\overline{d}^{-1}_i:i\in [1,9]\big \}$ and

$\textbf{\textrm{D}}[-]\overline{\textbf{\textrm{D}}}^{-1}=\big \{d_i[-]\overline{d}^{-1}_i:i\in [1,9]\big \}$

It is not hard to observe the following facts:

\begin{asparaenum}[(i) ]
\item $d_i[+]d^{-1}_i=d_i[-]\overline{d}^{-1}_i$;
\item The strings $d_i[+]\overline{d}_i$, $d^{-1}_i[+]\overline{d}^{-1}_i$, $d_i[+]d^{-1}_i$, $\overline{d}_i[+]\overline{d}^{-1}_i$, $d_i[+]\overline{d}^{-1}_i$, $d_i[-]d^{-1}_i$ and $d_i[-]\overline{d}^{-1}_i$ are \emph{symmetric strings}, where $d_i[+]\overline{d}_i$ and $d^{-1}_i[+]\overline{d}^{-1}_i$ are \emph{trivial symmetric strings}.

\item Two sets $\textbf{\textrm{D}}[+]\textbf{\textrm{D}}^{-1}$ and $\overline{\textbf{\textrm{D}}}[+]\overline{\textbf{\textrm{D}}}^{-1}$ are \emph{inverse} from each other, that is
$$\big (\textbf{\textrm{D}}[+]\textbf{\textrm{D}}^{-1}\big )[+]\big (\overline{\textbf{\textrm{D}}}[+]\overline{\textbf{\textrm{D}}}^{-1}\big )=\big \{999999:i\in [1,9]\big \}$$

\item Since $\overline{d_i[-]\overline{d}_i}=d^{-1}_i[-]\overline{d}^{-1}_i$, so two sets $\textbf{\textrm{D}}[-]\overline{\textbf{\textrm{D}}}$ and $\textbf{\textrm{D}}^{-1}[-]\overline{\textbf{\textrm{D}}}^{-1}$ are \emph{complementary} from each other, that is $\overline{\textbf{\textrm{D}}[-]\overline{\textbf{\textrm{D}}}}=\textbf{\textrm{D}}^{-1}[-]\overline{\textbf{\textrm{D}}}^{-1}$.

\item Two sets $\textbf{\textrm{D}}$ and $\textbf{\textrm{D}}^{-1}$ form every-zero string groups by adding one to each number of the strings under $(\bmod~9)$;

\item Two sets $\overline{\textbf{\textrm{D}}}$ and $\overline{\textbf{\textrm{D}}}^{-1}$ form every-zero string groups by substituting one to each number of the strings under $(\bmod~9)$;

\item The set $\overline{\textbf{\textrm{D}}}[+]\overline{\textbf{\textrm{D}}}^{-1}$ forms an every-zero string group by substituting two to each number of the strings under $(\bmod~9)$;

\item The sets $\textbf{\textrm{D}}[+]\textbf{\textrm{D}}^{-1}$, $\textbf{\textrm{D}}[-]\overline{\textbf{\textrm{D}}}$, $\textbf{\textrm{D}}[-]\overline{\textbf{\textrm{D}}}$ form $\textbf{\textrm{D}}^{-1}[-]\overline{\textbf{\textrm{D}}}^{-1}$ every-zero string groups by adding two to each number of the strings under $(\bmod~9)$.\qqed
\end{asparaenum}
\end{example}

\begin{problem}\label{qeu:444444}
An \emph{$[a,b]$-string} $s=c_1c_2\cdots c_n$ holds $c_j\in [a,b]$ with two integers $a,b$ subject to $0<a<b$. It is may be interesting to generalize the properties of $[0,9]$-strings to $[a,b]$-strings, and find some connections between $[a_k,b_k]$-strings with integers $0<a_k<b_k$.
\end{problem}

\begin{thm}\label{thm:first-mixed-string-operation}
$^*$ \textbf{The first mixed string operation ``$\oplus \ominus$'' for the every-zero string group.} Considering a string set $X=\{s_i=c_{i,1}c_{i,2}\cdots c_{i,n}:i\in [1,m]\}$, where each string $s_{i}=s_{1}[+](i-1)k$ is defined by $c_{i,j}=c_{1,j}+(i-1)k~(\bmod~A)$ for $i\in [1,m]$ and $j\in [1,n]$ and integer $k\geq 1$. Suppose that two strings $s_{i}\in X$ and $s_{r}\in X$ hold
\begin{equation}\label{eqa:new-k-string-groups00}
{
\begin{split}
&\quad ~ c_{i,j}+c_{r,j}-c_{p,j}~(\bmod~A)\\
&=(c_{1,j}+(i-1)k)+(c_{1,j}+(r-1)k)-(c_{1,j}+(p-1)k)~(\bmod~A)\\
&=c_{1,j}+(i+r-p-1)k~(\bmod~A)\\
&=c_{\lambda,j}
\end{split}}
\end{equation}
with the index $\lambda=i+r-p~(\bmod~A)$ and $j\in [1,n]$, such that $s_{\lambda}\in X$ for any preappointed \emph{zero} $s_{p}\in X$.

The operation ``$s_{i}[\oplus \ominus_p]s_{r}$'' on the string set $X$ is defined by
\begin{equation}\label{eqa:new-k-string-groups11}
s_{i}[\oplus \ominus_p]s_{r}:=s_{i}[+]s_{r}[-]s_{p}~(\bmod~A)=s_{\lambda}\in X
\end{equation} with the index $\lambda=i+r-p~(\bmod~A)$ based on Eq.(\ref{eqa:new-k-string-groups00}) for any preappointed \emph{zero} $s_{p}\in X$, so we get an \emph{every-zero string group} $\{F_{[+]k}(X);\oplus \ominus(\bmod~A)\}$.
\end{thm}
\begin{proof} By Eq.(\ref{eqa:new-k-string-groups00}) and Eq.(\ref{eqa:new-k-string-groups11}), we have

\textbf{First-1.} \textbf{Zero.} Each string $s_{p}\in \{F_{[+]k}(X);\oplus \ominus(\bmod~A)\}$ is the \emph{zero}.

\textbf{First-2.} \textbf{Closure.} Any pair of strings $s_{i},s_{r}\in \{F_{[+]k}(X);\oplus \ominus(\bmod~A)\}$ holds
$$
s_{i}[\oplus \ominus_p]s_{r}:=s_{i}[+]s_{r}[-]s_{p}\in \{F_{[+]k}(X);\oplus \ominus(\bmod~A)\}
$$

\textbf{First-3.} \textbf{Associative law.} Three strings $s_{i},s_{r}, s_{t}\in \{F_{[+]k}(X);\oplus \ominus(\bmod~A)\}$ hold
$$
\big (s_{i}[\oplus \ominus_p]s_{r}\big )[\oplus \ominus_p]s_{t}=s_{i}[\oplus \ominus_p]\big ( s_{r}[\oplus \ominus_p]s_{t}\big )
$$

\textbf{First-4.} \textbf{Commutative law.} Since
$$
s_{i}[\oplus \ominus_p]s_{r}=s_{r}[\oplus \ominus_p]s_{i}
$$ for any pair of strings $s_{i},s_{r}\in \{F_{[+]k}(X);\oplus \ominus(\bmod~A)\}$ holds true.
\end{proof}

\begin{example}\label{exa:8888888888}
Verify an every-zero string group $\{F_{[+]2}(\textbf{\textrm{D}}[+]\textbf{\textrm{D}}^{-1});\oplus \ominus (\bmod~9)\}$ based on the string set
$${
\begin{split}
\textbf{\textrm{D}}[+]\textbf{\textrm{D}}^{-1}=&\{a_1=891198=801108,a_2=123321,a_3=345543,a_4=567765,a_5=789987\\
&=780087,a_6=912219=012210,a_7=234432,a_8=456654,a_9=678876\}
\end{split}}
$$ by adding $k=2$ under $\bmod~9$ in Table-2.

We select some \emph{zeros} for computing $(a_{i}[+]a_{r}[-]a_{p}~(\bmod~9)\big )~[+]2~(\bmod~9)$. Taking a preappointed \emph{zero} $a_9\in \textbf{\textrm{D}}[+]\textbf{\textrm{D}}^{-1}$, we have

(1-1) $a_1[+]a_4[-]a_9=a_5$ with $[1+4-9~(\bmod~9)]=5$, then

$891198[+]567765[-]678876 ~(\bmod~9)=780087=a_5$;

(1-2) $a_3[+]a_2[-]a_9=a_5$ with $[3+2-9~(\bmod~9)]=5$, then

$345543[+]123321[-]678876 ~(\bmod~9)=780081=a_5$;

(1-3) $a_8[+]a_5[-]a_9=a_4$ with $[8+5-9~(\bmod~9)]=4$, then

$456654[+]789987[-]678876 ~(\bmod~9)=567765=a_4$.\\
For another preappointed \emph{zero} $s_2\in \textbf{\textrm{D}}[+]\textbf{\textrm{D}}^{-1}$, we have

(2-1) $a_1[+]a_4[-]a_2=a_3$ with $[1+4-2~(\bmod~9)]=3$, then

$891198[+]567765[-]123321~ (\bmod~9)=567765=a_3$;

(2-2) $a_3[+]a_2[-]a_2=a_3$ with $[3+2-2~(\bmod~9)]=3$, then

$345543[+]123321[-]123321 ~(\bmod~9)=567765=a_3$;

(2-3) $a_8[+]a_5[-]a_2=a_2$ with $[8+5-2~(\bmod~9)]=2$, then

$456654[+]789987[-]123321~ (\bmod~9)=123321=a_2$.\\
However, we verify $(a_{i}[-]a_{r}[+]a_{p}~(\bmod~9)\big )~[+]2=a_{\lambda}\in \textbf{\textrm{D}}[+]\textbf{\textrm{D}}^{-1}$ in the following tests:

(3-1) $a_1[-]a_4[+]a_2=a_8$ with $[1-4+2~(\bmod~9)]=8$, then

$891198[-]567765[+]123321 ~(\bmod~9)=456654=a_8$;

(3-2) $a_3[-]a_2[+]a_2=a_3$ with $[3-2+2~(\bmod~9)]=3$, then

$345543[-]123321[+]123321 ~(\bmod~9)=345543=a_3$;

(3-3) $a_8[-]a_5[+]a_2=a_5$ with $[8-5+2~(\bmod~9)]=5$, then

$456654[-]789987[+]123321~ (\bmod~9)=789987=a_5$.\\
We get another every-zero string group $\{F_{[-]2}(\textbf{\textrm{D}}[+]\textbf{\textrm{D}}^{-1});\ominus \oplus (\bmod~9)\}$.\qqed
\end{example}

\begin{thm}\label{thm:second-mixed-string-operation}
$^*$ \textbf{The second mixed string operation ``$\ominus\oplus $'' for the every-zero string group.} Considering a string set $Y=\{r_i=t_{i,1}t_{i,2}\cdots t_{i,n}:i\in [1,m]\}$, where each string $r_{i,j}=r_{1,j}[-](i-1)k~(\bmod~B)$ is determined by $c_{i,j}=t_{1,j}-(i-1)k~(\bmod~B)$ for $i\in [1,m]$ and integer $k\geq 1$. Suppose that two strings $r_{i}\in Y$ and $r_{r}\in Y$ hold
\begin{equation}\label{eqa:new-k-string-groups}
t_{i,j}-t_{r,j}+t_{p,j}~(\bmod~B)=t_{\lambda,j}
\end{equation}
with the index $\lambda=i-r+p~(\bmod~B)$ and $j\in [1,n]$, such that $r_{\lambda}\in Y$ for any preappointed \emph{zero} $r_{p}\in Y$.

The operation ``$r_{i}[\ominus \oplus_p]r_{r}$'' on the string set $Y$ is defined by
\begin{equation}\label{eqa:555555}
r_{i}[\ominus \oplus _p]r_{r}:=r_{i}[-]r_{r}[+]r_{p}~(\bmod~B)=r_{\lambda}
\end{equation} with the index $\lambda=i-r+p~(\bmod~B)$ based on Eq.(\ref{eqa:new-k-string-groups}) for any preappointed \emph{zero} $r_{p}\in Y$, so we get an \emph{every-zero string group} $\{F_{[-]k}(Y);\ominus \oplus (\bmod~B)\}$.
\end{thm}

\begin{example}\label{exa:8888888888}
Verify an every-zero string group $\{F_{[-]2}(\overline{\textbf{\textrm{D}}}[+]\overline{\textbf{\textrm{D}}}^{-1});\ominus \oplus (\bmod~9)\}$ based on the string set
$${
\begin{split}
\overline{\textbf{\textrm{D}}}[+]\overline{\textbf{\textrm{D}}}^{-1}=&\{b_1=198891=108801,b_2=876678,b_3=654456,b_4=432234,b_5=219912\\
&=210012,b_6=987789=087780,b_7=765567,b_8=543345,b_9=321123\}
\end{split}}
$$ by substituting $k=2$ under $(\bmod~9)$ in Table-2.

For a preappointed \emph{zero} $b_2\in \textbf{\textrm{D}}[+]\textbf{\textrm{D}}^{-1}$, we have

(i) $b_1[-]b_4[+]b_2=b_8$ with $[1-4+2~(\bmod~9)]=8$, then

$198891[-]432234[+]876678~(\bmod~9)=543345=b_8$;

(ii) $b_3[-]b_2[+]b_2=b3$ with $[3-2+2~(\bmod~9)]=3$, then

$654456[-]876678[+]876678~(\bmod~9)=654456=b_3$;

(iii) $b_8[-]b_5[+]b_2=b_5$ with $[8-5+2~(\bmod~9)]=5$, then

$543345[-]219912[+]876678~(\bmod~9)=219912=b_5$.\\
However, we have

(iv) $b_1[+]b_4[-]b_2=b_3$ with $[1+4-2~(\bmod~9)]=3$, then

$198891[+]432234[-]876678~(\bmod~9)=654456=b_3$;

(v) $b_3[+]b_2[-]b_2=b_3$ with $[3+2-2~(\bmod~9)]=3$, then

$654456[+]876678[-]876678~(\bmod~9)=654456=b_3$;

(vi) $b_8[+]b_5[-]b_2=b_2$ with $[8+5-2~(\bmod~9)]=2$, then

$543345[+]219912[-]876678~(\bmod~9)=876678=b_2$.\\
Thereby, we get another every-zero string group $\{F_{[+]2}(\overline{\textbf{\textrm{D}}}[+]\overline{\textbf{\textrm{D}}}^{-1});\oplus \ominus (\bmod~9)\}$.\qqed
\end{example}

\begin{example}\label{exa:two-operations-two-sets}
The following string set
\begin{equation}\label{eqa:there-is-a-string-set}
{
\begin{split}
X_{[+][-]}=&\{c_1=702207=792297,c_2=246642,c_3=681186,c_4=135531,c_5=570075\\
& ~ =579975,c_6=024420=924429,c_7=468864,c_8=813318,c_9=357753\}
\end{split}}
\end{equation} is obtained from the string operation $(d_i[+]d^{-1}_i)[-](\overline{d}_i[+]\overline{d}^{-1}_i)~(\bmod~9)$ in Table-2.

Notice that $c_{i}=c_1-(i-1)5~(\bmod~9)$ in $X_{[+][-]}$, we have

(1) $c_1[-]c_3[+]c_9=c_7$ with $[1-3+9~(\bmod~9)]=7$, then

$702207[-]681186[+]357753~(\bmod~9)=468864=c_7$.

(2) $c_1[+]c_3[-]c_9=c_4$ with $[1+3-9~(\bmod~9)]=4$, then

$702207[+]681186[-]357753~(\bmod~9)=135531=c_4$.

(3) $c_4[-]c_7[+]c_9=c_6$ with $[4-7+9~(\bmod~9)]=6$, then

$135531[-]468864[+]357753~(\bmod~9)=024420=c_6$.

(4) $c_4[+]c_7[-]c_9=c_2$ with $[4+7-9~(\bmod~9)]=2$, then

$135531[+]468864[-]357753~(\bmod~9)=246642=c_2$.

Thereby, we get two every-zero string groups
$$\{F_{[+]5}(X_{[+][-]});\oplus \ominus (\bmod~9)\}\textrm{ and }\{F_{[-]5}(X_{[+][-]});\ominus \oplus (\bmod~9)\}
$$ based on the string set $X_{[+][-]}$ defined in Eq.(\ref{eqa:there-is-a-string-set}).\qqed
\end{example}

\begin{thm}\label{thm:666666}
There are infinite $[0,9]$-string sets $X$ and integers $k\in [0,8]$, such that each $[0,9]$-string set $X$ induces two every-zero string groups
$$\{F_{[+]k}(X);\oplus \ominus(\bmod~9)\}\textrm{ and }\{F_{[-]k}(X);\ominus \oplus (\bmod~9)\}$$
\end{thm}
\begin{proof} Take a $[0,9]$-string $s=c_{1}c_{2}\cdots c_{n}$ with $c_{j}\in [0,9]$, we make $[0,9]$-strings $s(t)=c_{t,1}c_{t,2}\cdots c_{t,n}$ with $c_{t,j}\in [0,9]$, where each $c_{t,j}=c_{j}+t~(\bmod~9)$ for $j\in [1,n]$ and $t\in [0,8]$. Since
\begin{equation}\label{eqa:0-9-string-group11}
{
\begin{split}
[c_{r,j}+c_{l,j}-c_{p,j}~(\bmod~9)]&=(c_{j}+r)+(c_{j}+l)-(c_{j}+p)~(\bmod~9)\\
&=c_{j}+(r+l-p)~(\bmod~9)\\
&=c_{\lambda,j}\in [0,9]
\end{split}}
\end{equation} then the string set $X_1=\{s(t)=c_{t,1}c_{t,2}\cdots c_{t,n}:t\in [0,8]\}$ forms an \emph{every-zero string group} $\{F_{[+]}(X_1);\oplus \ominus (\bmod~9)\}$ by the operation ``$\oplus \ominus$'' defined by $$s(r)[\oplus \ominus_p]s(l):=s(r)[+]s(l)[-]s(p)=s(\lambda)
$$ with the index $\lambda=r+l-p~(\bmod~9)\in [0,8]$ based on Eq.(\ref{eqa:0-9-string-group11}).

Moreover, because of
\begin{equation}\label{eqa:0-9-string-group22}
{
\begin{split}
[c_{r,j}-c_{l,j}+c_{p,j}~(\bmod~9)]&=(c_{j}+r)-(c_{j}+l)+(c_{j}+p)~(\bmod~9)\\
&=c_{j}+(r-l+p)~(\bmod~9)\\
&=c_{\mu,j}\in [0,9]
\end{split}}
\end{equation} so we get another \emph{every-zero string group} $\{F_{[-]}(X_1);\ominus \oplus (\bmod~9)\}$ based on the string set $X_1$ by the operation ``$\ominus \oplus$'' defined by
$$s(r)[\ominus\oplus_p]s(l):=s(p)[-]s(r)[+]s(l)=s(\mu)\in X_1
$$ with the index $\mu=r-l+p~(\bmod~9)\in [0,8]$ based on Eq.(\ref{eqa:0-9-string-group22}).

Since $[0,9]$-strings $s(mk)=c_{mk,1}c_{mk,2}\cdots c_{mk,n}$ with $c_{mk,j}\in [0,9]$ for $j\in [1,n]$, where each $c_{mk,j}=c_{j}+mk~(\bmod~9)$ and $k\in [0,8]$ and $m\in [1,8]$, however, $c_{mk,j}\in [0,9]$, that is $c_{mk,j}\in X_1$. Thereby,
$$X_{mk}=\{s(mk)=c_{mk,1}c_{mk,2}\cdots c_{mk,n}:~k\in [0,8], m\in [1,8]\}=X_1
$$ which implies this theorem holding true.
\end{proof}

\begin{thm}\label{thm:666666}
For any $[0,9]$-string $s=c_{1}c_{2}\cdots c_{n}$ with $c_{j}\in [0,9]$, there is a $[0,9]$-string set $X_k=\{s(mk)=c_{mk,1}c_{mk,2}\cdots c_{mk,n}:c_{mk,j}\in [0,9]\}$ with $c_{mk,j}=c_{j}+mk~(\bmod~9)$ for $m\in [1,A_k]$ and $k\in [0,8]$, then each of two every-zero string groups
$$\{F_{[+]k}(X);\oplus \ominus(\bmod~9)\}\textrm{ and }\{F_{[-]k}(X_k);\ominus \oplus (\bmod~9)\}$$ contains three elements only as $k=3,6$.
\end{thm}

\subsection{Number-based super-strings}

We will do some generalizations of $[0,9]$-strings for obtaining super-strings and some algebraic operation on them in this subsection.

\begin{defn} \label{defn:number-based-super-string-def}
$^*$ A \emph{$(9)$s-number-based super-string} $s_{uper}=C_{1}C_{2}\cdots C_{n}$ is defined by $C_{j}\in [0,(9)_j]$ for $j\in [1,n]$, where each positive integer $(9)_j$ is made up of $a_j$ nines, that is $(9)_j=10^{a_j}-1$, also, $(9)_j+1~(\bmod~(9)_j)=1$. We make new $(9)$s-number-based super-strings $C_{j,m_j}=C_{j}+m_j~(\bmod~(9)_j)$ with $m_j\in [1,(9)_j]$ and $j\in [1,n]$, and get $(9)$s-number-based super-strings $$s_{uper}(\{m_j\}^n_{j=1})=C_{1,m_1}C_{2,m_2}\cdots C_{n,m_n},~m_j\in [1,(9)_j],~j\in [1,n]
$$ and put them into a \emph{$(9)$s-number-based super-string set}
\begin{equation}\label{eqa:555555}
S^n_{uper}(\{(9)_j\}^n_{j=1})=\big \{s_{uper}(\{m_j\}^n_{j=1})=C_{1,m_1}C_{2,m_2}\cdots C_{n,m_n}:m_j\in [1,(9)_j],j\in [1,n]\big \}
\end{equation} with the cardinality $\big |S^n_{uper}(\{(9)_j\}^n_{j=1})\big |=\prod ^n_{j=1}(9)_j$.\qqed
\end{defn}

\begin{defn} \label{defn:uniformly-arithmetic-addition-subtraction}
$^*$ In Definition \ref{defn:number-based-super-string-def}, the $(9)$s-number-based super-string $s_{uper}(\{m_j\}^n_{j=1})$ has its own complementary as
$$
\overline{s}_{uper}(\{m_j\}^n_{j=1})=[(9)_1-C_{1,m_1}][(9)_2-C_{2,m_2}]\cdots [(9)_n-C_{n,m_n}]
$$ and its own inverse
$$
s^{-1}_{uper}(\{m_j\}^n_{j=1})=C_{n,m_n}C_{n-1,m_{n-1}}\cdots C_{2,m_2}C_{1,m_1}
$$ Moreover, we define the \emph{uniformly arithmetic addition} by
\begin{equation}\label{eqa:arithmetic-addition-operation11}
{
\begin{split}
& s_{uper}(\{m_j\}^n_{j=1})[+](ik)~(\bmod~\{m_j\}^n_{j=1})\\
=&[C_{1,m_1}+(ik)~(\bmod~(9)_1)][C_{2,m_2}+(ik)~(\bmod~(9)_2)]\cdots [C_{n,m_n}+(ik)~(\bmod~(9)_n)]
\end{split}}
\end{equation} write $s_{uper}([+](ik);\{m_j\}^n_{j=1})=s_{uper}(\{m_j\}^n_{j=1})[+](ik)~(\bmod~\{m_j\}^n_{j=1})$. Similarly, we define the \emph{uniformly arithmetic subtraction} by
\begin{equation}\label{eqa:arithmetic-subtraction-operation22}
{
\begin{split}
& s_{uper}([-](ik);\{m_j\}^n_{j=1})=s_{uper}(\{m_j\}^n_{j=1})[-](ik)~(\bmod~\{m_j\}^n_{j=1})\\
=&[C_{1,m_1}-(ik)~(\bmod~(9)_1)][C_{2,m_2}-(ik)~(\bmod~(9)_2)]\cdots [C_{n,m_n}-(ik)~(\bmod~(9)_n)]
\end{split}}
\end{equation} We call $s_{uper}([+](ik);\{m_j\}^n_{j=1})$ defined in Eq.(\ref{eqa:arithmetic-addition-operation11}) and $s_{uper}([-](ik);\{m_j\}^n_{j=1})$ defined in Eq.(\ref{eqa:arithmetic-subtraction-operation22}) \emph{uniformly $k$-arithmetic number-based strings}.\qqed
\end{defn}

\begin{rem}\label{rem:333333}
By Definition \ref{defn:uniformly-arithmetic-addition-subtraction}, we have the following operations
$${
\begin{split}
S_{[+]}=&\big (s_{uper}(\{m_j\}^n_{j=1})[+](ik)~(\bmod~\{m_j\}^n_{j=1})\big )[-]\big (s_{uper}(\{m_j\}^n_{j=1})[-](ik)~(\bmod~\{m_j\}^n_{j=1})\big )\\
=&[2ik~(\bmod~(9)_1)][2ik~(\bmod~(9)_2)]\cdots [2ik~(\bmod~(9)_n)]
\end{split}}$$ and
$${
\begin{split}
S_{[-]}=&\big (s_{uper}(\{m_j\}^n_{j=1})[-](ik)~(\bmod~\{m_j\}^n_{j=1})\big )[-]\big (s_{uper}(\{m_j\}^n_{j=1})[+](ik)~(\bmod~\{m_j\}^n_{j=1})\big )\\
=&[-2ik~(\bmod~(9)_1)][-2ik~(\bmod~(9)_2)]\cdots [-2ik~(\bmod~(9)_n)]
\end{split}}$$ Immediately,
\begin{equation}\label{eqa:555555}
S_{[+]}[+]S_{[-]}=(9)_1~(9)_2~\cdots ~(9)_n
\end{equation} so we have the complementary $\overline{S}_{[+]}=S_{[-]}$, or the complementary $S_{[+]}=\overline{S}_{[-]}$.

In application, we have the \emph{non-uniformly arithmetic addition/subtraction} as
\begin{equation}\label{eqa:non-uniformly-arithmetic-addition-operation}
{
\begin{split}
S^*_{[\pm]}=& s_{uper}(\{m_j\}^n_{j=1})[\pm]\{k_j\}^n_{j=1}~(\bmod~\{m_j\}^n_{j=1})\\
=&[C_{1,m_1}\pm k_1~(\bmod~(9)_1)][C_{2,m_2}\pm k_2~(\bmod~(9)_2)]\cdots [C_{n,m_n}\pm k_n~(\bmod~(9)_n)]
\end{split}}
\end{equation} with some $k_r\neq k_s$ for $r\neq s$.\paralled
\end{rem}

\begin{example}\label{exa:8888888888}
Given a number-based string $s_{uper}(\{(9)_j\}^6_{j=1})=6174~123~0~618~3~141$, where $(9)_1=9999$, $(9)_2=999$,$(9)_3=9$, $(9)_4=999$, $(9)_5=9$ and $(9)_6=999$. We have
$${
\begin{split}
s([+]152)=&s_{uper}(\{(9)_j\}^6_{j=1})[+]152~(\bmod~\{(9)_j\}^6_{j=1})\\
=&[6174+152~(\bmod~(9)_1)][123+152~(\bmod~(9)_2)][0+152~(\bmod~(9)_3)]\\
&[618+152~(\bmod~(9)_4)][3+152~(\bmod~(9)_5)][141+152~(\bmod~(9)_6)]\\
=&6326~275~8~770~2~293
\end{split}}
$$
$${
\begin{split}
s([-]152)=&s_{uper}(\{(9)_j\}^6_{j=1})[-]152~(\bmod~\{(9)_j\}^6_{j=1})\\
=&[6174-152~(\bmod~(9)_1)][123-152~(\bmod~(9)_2)][0-152~(\bmod~(9)_3)]\\
&[618-152~(\bmod~(9)_4)][3-152~(\bmod~(9)_5)][141-152~(\bmod~(9)_6)]\\
=&6022~970~1~466~4~988
\end{split}}
$$ We get

$s([+]152)[-]s([-]152)=304~304~7~304~7~304$, $s([-]152)[-]s([+]152)=9695~695~2~695~2~695$,\\
and $\overline{s([+]152)[-]s([-]152)}=s([-]152)[-]s([+]152)$.\qqed
\end{example}

\begin{problem}\label{qeu:444444}
\textbf{The $(9)$s-number-based Super-string Problem.} We rewrite a number-based string $s=c_{1}c_{2}\cdots c_{n}$ as $(9)$s-number-based super-strings $s_i=a_{i,1}a_{i,2}\cdots a_{i,m_i}$, where
$$
a_{i,1}=c_{1}c_{2}\cdots c_{n_1},~a_{i,j}=c_{n_{j-1}+1}c_{n_{j-1}+2}\cdots c_{n_j}\textrm{ with }j\in [1,i_s],~n_{0}=0,~n_{n_{i_s}}=n
$$ and $n_j-(n_{j-1})+1=(9)_{i,j}$, then rewrite $s_i$ as $s_i=s_{uper}(\{(9)_{i,j}\}^{m_i}_{j=1})$, and we put these $(9)$s-number-based super-strings into a string set
$$C_{ut}(s)=\{s_i=s_{uper}(\{(9)_{i,j}\}^{m_i}_{j=1}):i\in [1,M(s)]\}
$$ \textbf{Compute} the exact value of $M(s)$.

Moreover, we make other $(9)$s-number-based super-strings $s^*_k=b_{k,1}b_{k,2}\cdots b_{k,d_k}$ based on the given string $s=c_{1}c_{2}\cdots c_{n}$, such that each $c_j$ is in one $b_{k,r}$, but not in other $b_{k,s}$ if $s\neq r$, where each $b_{k,i}$ is made up of $(9)_{k,i}$ numbers of $[0,9]$, we rewrite $s^*_k$ as $s^*_k=s_{uper}(\{(9)_{k,i}\}^{d_k}_{i=1})$, and get a $(9)$s-number-based super-string set
$$R_{ewrite}(s)=\{s^*_k=s_{uper}(\{(9)_{k,i}\}^{d_k}_{i=1}):k\in [1,M^*(s)]\}
$$ \textbf{Determine} the value of $M^*(s)$.
\end{problem}

\begin{thm}\label{thm:uniformly-super-string-groups}
$^*$ By Definition \ref{defn:number-based-super-string-def}, we do the uniformly arithmetic addition defined in Eq.(\ref{eqa:arithmetic-addition-operation11}) to a \emph{uniformly $(9)$s-number-based super-string}
$$s_{uper}(A)=C_{1,A_1}C_{2,A_2}\cdots C_{n,A_n}
$$ with $A_j=10^A-1$ and $j\in [1,n]$, and get uniformly $k$-arithmetic number-based strings as
\begin{equation}\label{eqa:uniformly-arithmetic-addition-operation}
{
\begin{split}
s_{uper}([+](ik))=&s_{uper}(A)[+](ik)~(\bmod~A)\\
=&[C_{1,A_1}+(ik)~(\bmod~A)][C_{2,A_2}+(ik)~(\bmod~A)]\cdots [C_{n,A_n}+(ik)~(\bmod~A)]
\end{split}}
\end{equation} There is a \emph{uniformly $(9)$s-number-based super-string set}
$$U_{nif}(A)=\big \{s_{uper}([+](ik)):~i\in [1,m],~mk\geq A\big \}
$$ such that each of the following two \emph{every-zero $(9)$s-uniformly super-string groups}
$$\{F_{[+]k}(U_{nif}(A));\oplus \ominus(\bmod~10^A-1)\}\textrm{ and } \{F_{[-]k}(U_{nif}(A));\ominus \oplus (\bmod~10^A-1)\}
$$ contains three elements only when $k=\frac{1}{3}(10^A-1), \frac{2}{3}(10^A-1)$.
\end{thm}

\begin{thm}\label{thm:super-string-groups00}
$^*$ Let $B=\max\{(9)_j:j\in [1,n]\}$ in the $(9)$s-number-based super-string set $S^n_{uper}(\{(9)_j\}^n_{j=1})$ defined in Definition \ref{defn:number-based-super-string-def}. Then there are two \emph{every-zero $(9)$s-uniformly super-string groups} of order $B$ below
$$
\{F_{[+]k}(S^n_{uper}(\{m_j\}^n_{j=1}));\oplus \ominus(\bmod~B)\}\textrm{ and }\{F_{[-]k}(S^n_{uper}(\{m_j\}^n_{j=1}));\ominus \oplus (\bmod~B)\},k\in [1,B]
$$
\end{thm}

\subsection{Self-breeding number-based strings}

We present an algorithm for producing self-breeding number-based strings in this subsection.

\vskip 0.4cm

\textbf{Input:} A number-based string set $S_1=\{s_1(r_0):r_0\in [1,M_1]\}$, where each $[0,9]$-string $s_1(r_0)=c_{1,r,1}c_{1,r,1}\cdots c_{1,r,a_r}$ with $c_{1,r,j}\in [0,9]$ has $b_{\textrm{yte}}(s_1(r_0))$ bytes, so $S_1$ has $$b_{\textrm{yte}}(S_1)=\sum^{M_1}_{r=1}b_{\textrm{yte}}(s_1(r))
$$ bytes, in total.

\textbf{Output:} $t$-rank self-breeding number-based strings for $t\geq 2$.

\textbf{Step 1.} Let $(M_1)!=M_2$. We make a \emph{self-breeding number-based string set} $S_2=\{s_2(r_1):r_1\in [1,M_2]\}$, in which each \emph{self-breeding number-based string}
$$
s_2(r_1)=s_1(r_1,r_0,1)s_1(r_1,r_0,2)\cdots s_1(r_1,r_0,M_1)
$$ has $b_{\textrm{yte}}(s_2(r_1))$ bytes, where $s_2(r_1)$ is a permutation of elements of $S_1$ for $r_1\in [1,M_2]=[1,(M_1)!]$. Thereby, $S_2$ has $b_{\textrm{yte}}(S_2)=b_{\textrm{yte}}(S_1)\cdot M_2=b_{\textrm{yte}}(S_1)\cdot (M_1)!$ bytes, in total.

\textbf{Step 2.} Let $(M_2)!=M_3$. We make a self-breeding number-based string set $S_3=\{s_3(r_2):r_2\in [1,M_3]\}$, each self-breeding number-based string
$$s_3(r_2)=s_2(r_2,r_1,1)s_2(r_2,r_1,2)\cdots s_2(r_2,r_1,M_2)
$$ where $s_3(r_2)$ is a permutation of elements of $S_2$ for $r_2\in [1,M_3]=[1,(M_2)!]$. So, $S_3$ has
$$b_{\textrm{yte}}(S_3)=b_{\textrm{yte}}(S_1)\cdot M_3=b_{\textrm{yte}}(S_1)\cdot (M_2)!\cdot (M_1)!=b_{\textrm{yte}}(S_1)\cdot \prod ^2_{j=1}(M_j)!
$$ bytes, in total.

\textbf{Step $t$.} Let $(M_t)!=M_{t+1}$. We make a self-breeding number-based string set $S_{t+1}=\{s_{t+1}(r_t):r_t\in [1,M_{t+1}]\}$, in which each \emph{self-breeding number-based string}
$$
s_{t+1}(r_t)=s_{t+1}(r_t,r_{t-1},1)s_{t+1}(r_t,r_{t-1},2)\cdots s_{t+1}(r_t,r_{t-1},M_t)
$$ where $s_{t+1}(r_t)$ is a permutation of elements of $S_t$ for $r_{t+1}\in [1,M_{t+1}]=[1,(M_t)!]$. Thereby, the sum of element's bytes of $S_{t+1}$ is denoted as
\begin{equation}\label{eqa:self-breeding-string-bytes}
{
\begin{split}
b_{\textrm{yte}}(S_{t+1})=b_{\textrm{yte}}(S_1)\cdot M_{t+1}=b_{\textrm{yte}}(S_1)\cdot \prod ^t_{j=1}(M_j)!=b_{\textrm{yte}}(S_1)\cdot (M_t)! (M_{t-1})!\cdots (M_2)!(M_1)!
\end{split}}
\end{equation} with $M_k=(M_{k-1})!$ for $k\geq 2$.

\textbf{Step $t+1$.} Return $t$-rank self-breeding number-based strings.

\vskip 0.4cm

For Eq.(\ref{eqa:self-breeding-string-bytes}), we rewrite
\begin{equation}\label{eqa:recucive-string-bytes}
M_{k+1}=(M_{k})!=[(M_{k-1})!]!=[M_{k-1}](!^{2})=[(M_{k-2})!](!^{2})=\cdots =[M_{1}](!^{k})
\end{equation} and let $[M_{1}](!^{1})=(M_1)!$, then we have
\begin{equation}\label{eqa:self-breeding-string-bytes22}
{
\begin{split}
b_{\textrm{yte}}(S_{t+1})=b_{\textrm{yte}}(S_1)\cdot [M_{1}](!^{t}) [M_{1}](!^{t-1})\cdots [M_{1}](!^{2})(M_1)!=b_{\textrm{yte}}(S_1)\cdot \prod ^t_{j=1} [M_{1}](!^{j})
\end{split}}
\end{equation}

\begin{rem}\label{rem:333333}
Each self-breeding number-based string set $S_{t+1}=\{s_{t+1}(r_t):r_t\in [1,M_{t+1}]\}$ can be used to build up an every-zero self-breeding number-based string group $\{F_{M_{t+1}}(S_{t+1});\oplus\ominus\}$ defined by the operation
\begin{equation}\label{eqa:self-breeding-string-group}
s_{t+1}(i)~[\oplus \ominus_k] ~s_{t+1}(j):=s_{t+1}(i)\oplus s_{t+1}(j)\ominus s_{t+1}(k)=s_{t+1}(\lambda)\in S_{t+1}
\end{equation} with the index $\lambda=i+j-k~(\bmod~M_{t+1})$ for a preappointed \emph{zero} $s_{t+1}(k)\in S_{t+1}$.\paralled
\end{rem}

\begin{example}\label{exa:8888888888}
Let $S_1=\{s_1(1)=$214, $s_1(2)=$1001, $s_1(3)=$68$\}$ be a self-breeding number-based string set, we have $b_{\textrm{yte}}(S_1)=\sum^3_{r=1}b_{\textrm{yte}}(s_1(r))=9$ bytes in total. Let $M_1=3!=6$.

We make a self-breeding number-based string set $S_2=\{s_2(r):r\in [1,6]\}$, in which

$s_2(1)=$214100168, $s_2(2)=$214681001, $s_2(3)=$100168214, $s_2(4)=$100121468, $s_2(5)=$682141001, $s_2(6)=$681001214, so there are $b_{\textrm{yte}}(S_2)=\sum^6_{r=1}b_{\textrm{yte}}(s_2(r))=6\cdot 9=54$ bytes in total.

Let $M_2=(M_1)!=6!=720$. We make a self-breeding number-based string set $S_3=\{s_3(r):r\in [1,720]\}$ containing self-breeding number-based strings like the following two number-based strings

\begin{equation}\label{eqa:555555}
{
\begin{split}
s_3(1)&=214100168214681001100168214100121468682141001681001214\\
s_3(2)&=100121468214100168100168214681001214682141001214681001
\end{split}}
\end{equation} with $54$ bytes for each. So, there are
$$
b_{\textrm{yte}}(S_3)=\sum^{720}_{r=1}b_{\textrm{yte}}(s_3(r))=720\cdot 6\cdot 9=38,880
$$ bytes for all elements of $S_3$, in total.

In general, we make a self-breeding number-based string set $S_{t+1}=\{s_{t+1}(r):r\in [1,M_t]\}$ with $M_{t+1}=(M_{t})!$, and there are
\begin{equation}\label{eqa:555555}
b_{\textrm{yte}}(S_{t+1})=\sum^{M_t}_{r=1}b_{\textrm{yte}}(s_{t+1}(r))=9\cdot \prod ^{t-1}_{j=1}(N_j)!
\end{equation} bytes for all number-based strings of the self-breeding number-based string set $S_{t+1}$.\qqed
\end{example}

\section{Number-based string-colorings}

Strings can be used to design new colorings of graphs. There are homogeneous strings and non-homogeneous strings, in general, we will introduce colorings on these strings in this section.

\subsection{Homogeneous number-based string-colorings}

We present several string-colorings of graphs, more or less, which are like $n$-dimension-type total colorings introduced in \cite{Yao-Wang-2106-15254v1, Yao-Sun-Hongyu-Wang-n-dimension-ICIBA2020}.

\begin{defn} \label{defn:homoge-uniformly-string-total-colorings}
$^*$ \textbf{Homogeneous string-coloring.} Let $\textbf{\textrm{S}}_{tring}(n)$ be the set of \emph{$n$-rank number-based strings} $\alpha_{1}\alpha_{2}\cdots \alpha_{n}$ with each number $\alpha_{j}\ge 0$ for $j\in [1,n]$. A $(p,q)$-graph $G$ admits a total string-coloring $f:V(G)\cup E(G)\rightarrow \textbf{\textrm{S}}_{tring}(n)$, such that
\begin{equation}\label{eqa:homogeneous-string-colorings}
f(u_k)=a_{k,1}a_{k,2}\cdots a_{k,n},~f(v_k)=b_{k,1}b_{k,2}\cdots b_{k,n},~f(u_kv_k)=c_{k,1}c_{k,2}\cdots c_{k,n}
\end{equation} for each edge $u_kv_k\in E(G)=\{u_kv_k:k\in [1,q]\}$, and let $\lambda$ be a non-negative integer. There, for each $k\in [1,q]$, are the following constraints:
\begin{asparaenum}[\textbf{\textrm{Res}}-1. ]
\item \label{stringco:adjacent-v} $f(u_k)\neq f(v_k)$ for each edge $u_kv_k\in E(G)$.
\item \label{stringco:adjacent-e} $f(u_kv_k)\neq f(u_kw_k)$ for each pair of two adjacent edges $u_kv_k,u_kw_k\in E(G)$.
\item \label{stringco:incident-v-e} $f(u_k)\neq f(u_kv_k)$ and $f(v_k)\neq f(u_kv_k)$ for each edge $u_kv_k\in E(G)$.

----- \emph{uniformly}

\item \label{stringco:e-graceful} Each $j\in [1,n]$ holds the graceful constraint $c_{k,j}=|a_{k,j}-b_{k,j}|$ true.
\item \label{stringco:each-odd} Each $j\in [1,n]$ holds the odd-graceful constraint $c_{k,j}=|a_{k,j}-b_{k,j}|$ true, and each $c_{k,j}$ for $j\in [1,n]$ is odd.
\item \label{stringco:harmonious} Each $j\in [1,n]$ holds the harmonious constraint $c_{k,j}=a_{k,j}+b_{k,j}~(\bmod~q)$ true.
\item \label{stringco:uniform-graceful} $f(E(G))=\{f(u_kv_k)=c_{k,1}c_{k,2}\cdots c_{k,n}:k\in[1$, $q]\}$ such that $c_{k,r}=k\in[1,q]$ for each $r\in [1,n]$.
\item \label{stringco:uniform-odd-graceful} $f(E(G))=\{f(u_jv_j)=t_{j,1}t_{j,2}\cdots t_{j,n}:j\in[1$, $2q-1]^o\}$ such that $t_{j,r}=j\in[1,2q-1]^o$ for each $r\in [1,n]$.

----- \emph{common-factor, anti-equitable}

\item \label{stringco:each-common-factor} Each $j\in [1,n]$ holds the common-factor constraint $c_{k,j}=\textrm{gcd}(a_{k,j},b_{k,j})$ true.
\item \label{stringco:pairwise-distinct} The anti-equitable constraint $c_{k,i}\neq c_{k,j}$ holds true for each pair of $c_{k,i}$ and $c_{k,j}$ if $i\neq j$.

----- \emph{integer partition, factorization}

\item \label{stringco:integer-partition} Each $i\in [1,n]$ holds that $a_{k,i}$ and $b_{k,i}$ are prime numbers, and $c_{k,i}=a_{k,i}+b_{k,i}$.
\item \label{stringco:integer-factorization} Each $j\in [1,n]$ holds that $a_{k,j}$ and $b_{k,j}$ are prime numbers, and $c_{k,j}=a_{k,j}\cdot b_{k,j}$.
\end{asparaenum}
\noindent \textbf{We call $f$}:

\noindent ----- \emph{proper}

\begin{asparaenum}[\textbf{\textrm{Strco}}-1. ]
\item a \emph{proper total hostring-coloring} if the constraints Res-\ref{stringco:adjacent-v}, Res-\ref{stringco:adjacent-e} and Res-\ref{stringco:incident-v-e} hold true.
\item a \emph{subtraction proper total hostring-coloring} if the constraints Res-\ref{stringco:adjacent-v}, Res-\ref{stringco:adjacent-e}, Res-\ref{stringco:incident-v-e} and Res-\ref{stringco:e-graceful} hold true.
\item an \emph{odd-subtraction proper total hostring-coloring} if the constraints Res-\ref{stringco:adjacent-v}, Res-\ref{stringco:adjacent-e}, Res-\ref{stringco:incident-v-e} and Res-\ref{stringco:each-odd} hold true.
\item a \emph{factor proper total hostring-coloring} if the constraints Res-\ref{stringco:adjacent-v}, Res-\ref{stringco:adjacent-e}, Res-\ref{stringco:incident-v-e} and Res-\ref{stringco:each-common-factor} hold true.
\item a \emph{harmonious proper total hostring-coloring} if the constraints Res-\ref{stringco:adjacent-v}, Res-\ref{stringco:adjacent-e}, Res-\ref{stringco:incident-v-e} and Res-\ref{stringco:harmonious} hold true.

----- \emph{uniformly}

\item an \emph{uniformly graceful proper total hostring-coloring} if the constraints Res-\ref{stringco:adjacent-v}, Res-\ref{stringco:adjacent-e}, Res-\ref{stringco:e-graceful} and Res-\ref{stringco:uniform-graceful} hold true.
\item an \emph{uniformly odd-graceful proper total hostring-coloring} if the constraints Res-\ref{stringco:adjacent-v}, Res-\ref{stringco:adjacent-e}, Res-\ref{stringco:e-graceful} and Res-\ref{stringco:uniform-odd-graceful} hold true.

----- \emph{anti-equitable}

\item an \emph{anti-equitable proper total hostring-coloring} if the constraints Res-\ref{stringco:adjacent-v}, Res-\ref{stringco:adjacent-e}, Res-\ref{stringco:incident-v-e}, Res-\ref{stringco:e-graceful} and Res-\ref{stringco:pairwise-distinct} hold true.

----- \emph{integer partition, factorization}

\item an \emph{uniformly partition proper total hostring-coloring} if the constraints Res-\ref{stringco:adjacent-v}, Res-\ref{stringco:adjacent-e}, Res-\ref{stringco:incident-v-e} and Res-\ref{stringco:integer-partition} hold true.
\item an \emph{uniformly factorization proper total hostring-coloring} if the constraints Res-\ref{stringco:adjacent-v}, Res-\ref{stringco:adjacent-e}, Res-\ref{stringco:incident-v-e} and Res-\ref{stringco:integer-factorization} hold true.\qqed
\end{asparaenum}
\end{defn}

\begin{figure}[h]
\centering
\includegraphics[width=16cm]{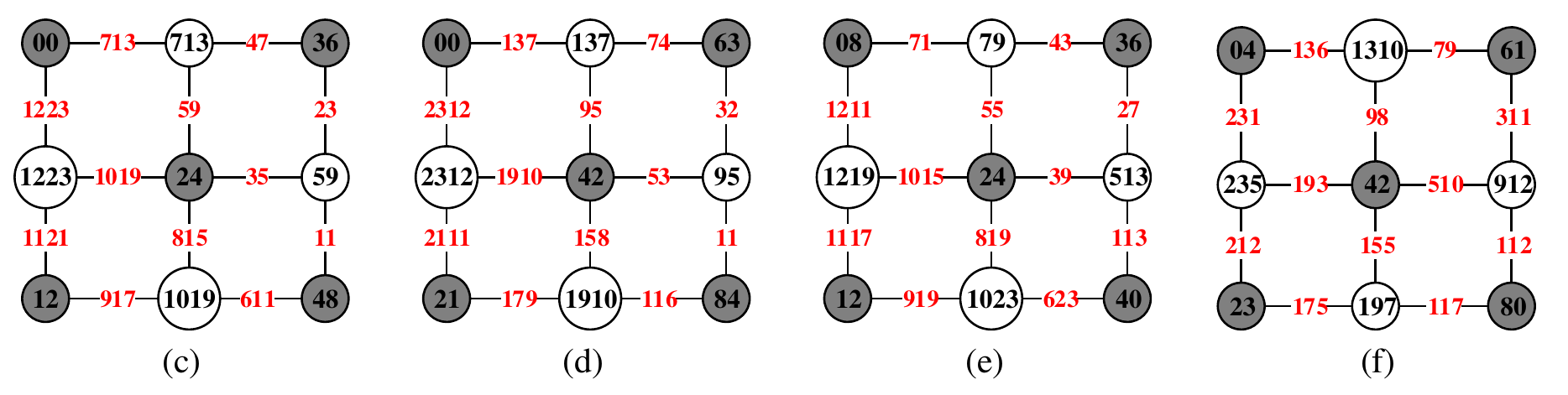}
\caption{\label{fig:n-dimension-11}{\small Four colored graphs for illustrating Definition \ref{defn:homoge-uniformly-string-total-colorings}, cited from \cite{Yao-Wang-2106-15254v1}.}}
\end{figure}

\begin{figure}[h]
\centering
\includegraphics[width=14cm]{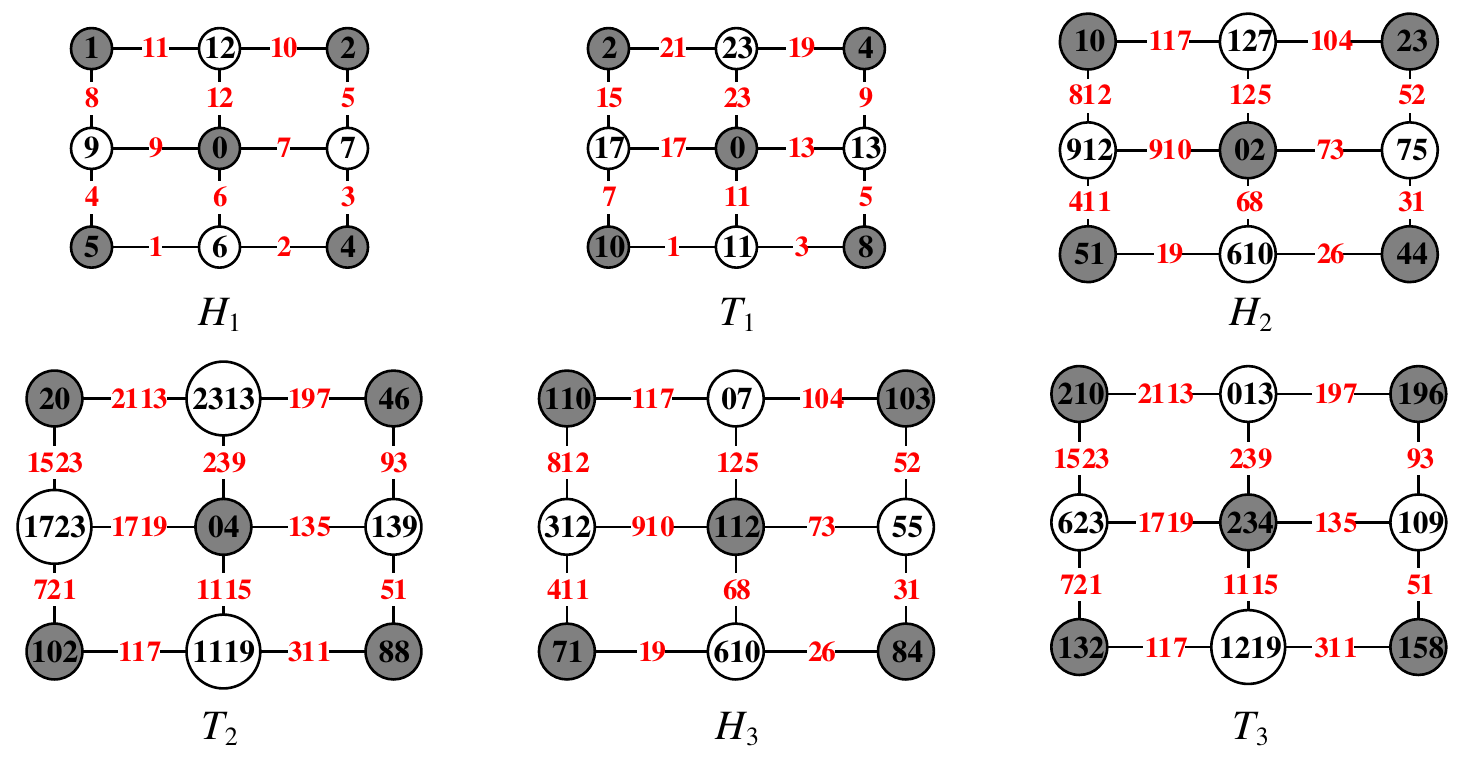}
\caption{\label{fig:n-dimension-22}{\small Six colored graphs: $H_1$ admits a graceful labeling, $T_1$ admits an odd-graceful labeling; both $H_2$ and $H_3$ admit two uniformly graceful proper total string-colorings; both $T_2$ and $T_3$ admit two uniformly odd-graceful proper total string-colorings defined in Definition \ref{defn:homoge-uniformly-string-total-colorings}, cited from \cite{Yao-Wang-2106-15254v1}.}}
\end{figure}

\begin{rem}\label{rem:333333}
By the homogeneous strings
\begin{equation}\label{eqa:66666666}
A_k=a_{k,1}a_{k,2}\cdots a_{k,n},~B_k=b_{k,1}b_{k,2}\cdots b_{k,n},~C_k=c_{k,1}c_{k,2}\cdots c_{k,n}
\end{equation}
and Definition \ref{defn:0-9-string-4-operation}, we have the \emph{string-magic constraints}
\begin{equation}\label{eqa:555555}
A_k[+]B_k[+]C_k=\lambda,~C_k[+]\big |A_k[-]B_k\big |=\mu,~\Big |\big |A_k[-]B_k\big |~[-]C_k\Big |=\gamma,~\big |A_k[+]B_k[-]C_k\big |=\rho
\end{equation}

For constants $a, b , c $, we have the following \emph{$(abc)$-magic string-constraints}
\begin{equation}\label{eqa:555555}
{
\begin{split}
&a\cdot A_k[+]b\cdot B_k[+]c\cdot C_k=\lambda,\quad c\cdot C_k[+]\big |a\cdot A_k[-]b\cdot B_k\big |=\mu,\\
&\Big |\big |a\cdot A_k[-]b\cdot B_k\big |~[-]c\cdot C_k\Big |=\gamma,\quad \big |a\cdot A_k[+]b\cdot B_k[-]c\cdot C_k\big |=\rho
\end{split}}
\end{equation} as \emph{string algebra}.\paralled
\end{rem}

\begin{defn} \label{defn:111111}
$^*$ \textbf{Homogeneous sub-proper total hostring-colorings.} As removing the constraint ``$f(u_k)\neq f(u_kv_k)$ and $f(v_k)\neq f(u_kv_k)$ for each edge $u_kv_k\in E(G)$'' from Definition \ref{defn:homoge-uniformly-string-total-colorings}, we get the following the ``sub-proper'' total hostring-colorings:
\begin{asparaenum}[\textbf{\textrm{Strco}}-1. ]
\item a \emph{sub-proper total hostring-coloring} if the constraints Res-\ref{stringco:adjacent-v} and Res-\ref{stringco:adjacent-e} hold true.
\item a \emph{subtraction sub-proper total hostring-coloring} if the constraints Res-\ref{stringco:adjacent-v}, Res-\ref{stringco:adjacent-e} and Res-\ref{stringco:e-graceful} hold true.
\item an \emph{odd-subtraction sub-proper total hostring-coloring} if the constraints Res-\ref{stringco:adjacent-v}, Res-\ref{stringco:adjacent-e} and Res-\ref{stringco:each-odd} hold true.
\item a \emph{factor sub-proper total hostring-coloring} if the constraints Res-\ref{stringco:adjacent-v}, Res-\ref{stringco:adjacent-e} and Res-\ref{stringco:each-common-factor} hold true.
\item a \emph{harmonious sub-proper total hostring-coloring} if the constraints Res-\ref{stringco:adjacent-v}, Res-\ref{stringco:adjacent-e} and Res-\ref{stringco:harmonious} hold true.
\item an \emph{anti-equitable sub-proper total hostring-coloring} if the constraints Res-\ref{stringco:adjacent-v}, Res-\ref{stringco:adjacent-e}, Res-\ref{stringco:e-graceful} and Res-\ref{stringco:pairwise-distinct} hold true.

\item an \emph{uniformly partition sub-proper total hostring-coloring} if the constraints Res-\ref{stringco:adjacent-v}, Res-\ref{stringco:adjacent-e} and Res-\ref{stringco:integer-partition} hold true.
\item an \emph{uniformly factorization sub-proper total hostring-coloring} if the constraints Res-\ref{stringco:adjacent-v}, Res-\ref{stringco:adjacent-e} and Res-\ref{stringco:integer-factorization} hold true.\qqed
\end{asparaenum}
\end{defn}

\begin{defn} \label{defn:homoge-various-string-total-colorings}
$^*$ \textbf{Weak homogeneous total string-coloring.} Let $\textbf{\textrm{S}}_{tring}(n)$ be the set of \emph{$n$-rank number-based strings} $\alpha_{1}\alpha_{2}\cdots \alpha_{n}$ with each number $\alpha_{j}\ge 0$ for $j\in [1,n]$. Suppose that a $(p,q)$-graph $H$ admits a total string-coloring $h:V(H)\cup E(H)\rightarrow \textbf{\textrm{S}}_{tring}(n)$, such that
\begin{equation}\label{eqa:555555}
h(x_k)=r_{k,1}r_{k,2}\cdots r_{k,n},~h(y_k)=s_{k,1}s_{k,2}\cdots s_{k,n},~h(x_ky_k)=t_{k,1}t_{k,2}\cdots t_{k,n}
\end{equation} for each edge $x_ky_k\in E(H)=\{x_ky_k:k\in [1,q]\}$. Let $\gamma$ be a non-negative integer. There, for each $k\in [1,q]$, are the following constraints:
\begin{asparaenum}[\textbf{\textrm{Nos}}-1. ]
\item \label{substring:adjacent-v} $h(x_k)\neq h(y_k)$ for each edge $x_ky_k\in E(H)$.
\item \label{substring:adjacent-e} $h(x_ky_k)\neq h(x_kw_k)$ for each pair of two adjacent edges $x_ky_k,x_kw_k\in E(H)$.
\item \label{substring:incident-v-e} $h(x_k)\neq h(x_ky_k)$ and $h(y_k)\neq h(x_ky_k)$ for each edge $x_ky_k\in E(H)$.
\item \label{substring:common-factor} Some $j\in [1,n]$ holds $t_{k,j}=\textrm{gcd}(r_{k,j},s_{k,j})$ true, but not all.
\item \label{substring:e-graceful} Some $j\in [1,n]$ holds the graceful constraint $t_{k,j}=|r_{k,j}-s_{k,j}|$ true, but not all.
\item \label{substring:each-odd} Each $t_{k,j}$ is odd, some $j\in [1,n]$ holds the odd-graceful constraint $t_{k,j}=|r_{k,j}-s_{k,j}|$ true, but not all.
\item \label{substring:harmonious} Some $j\in [1,n]$ holds the harmonious constraint $t_{k,j}=r_{k,j}+s_{k,j}~(\bmod~q)$ true, but not all.

----- \emph{common-factor, anti-equitable}

\item \label{substring:part-graceful} $h(E(H))=\{h(x_ky_k)=c_{k,1}c_{k,2}\cdots c_{k,n}:k\in[1$, $q]\}$ such that some $j\in[1,q]$ holds $c_{k,j}=k$ true, but not all.
\item \label{substring:part-odd-graceful} $h(E(H))=\{h(u_jv_j)=c_{j,1}c_{j,2}\cdots c_{j,n}:j\in[1$, $2q-1]^o\}$ such that some $j\in[1,2q-1]^o$ holds $c_{j,r}=j$ true, but not all.

----- \emph{common-factor, anti-equitable}

\item \label{substring:common-factor} Some $j\in [1,n]$ holds $t_{k,j}=\textrm{gcd}(r_{k,j},s_{k,j})$ true, but not all.
\item \label{substring:pairwise-distinct} Some pair of $c_{k,i}$ and $c_{k,j}$ with $i\neq j$ holds $c_{k,i}\neq c_{k,j}$ true, but not all.

----- \emph{weak-partition, weak-factorization}

\item \label{substring:integer-partition} $r_{k,i}$ and $s_{k,i}$ are prime numbers for some $i\in [1,n]$, and $t_{k,i}=r_{k,i}+s_{k,i}$, but not all.
\item \label{substring:integer-factorization} $r_{k,j}$ and $s_{k,j}$ are prime numbers for some $j\in [1,n]$, and $t_{k,j}=r_{k,j}\cdot s_{k,j}$, but not all.
\end{asparaenum}
\textbf{We call $h$}:
\begin{asparaenum}[\textbf{\textrm{Noc}}-1. ]
\item a \emph{weak-subtraction proper total hostring-coloring} if the constraints Nos-\ref{substring:adjacent-v}, Nos-\ref{substring:adjacent-e}, Nos-\ref{substring:incident-v-e} and Nos-\ref{substring:e-graceful} hold true.
\item a \emph{weak-odd-subtraction proper total hostring-coloring} if the constraints Nos-\ref{substring:adjacent-v}, Nos-\ref{substring:adjacent-e}, Nos-\ref{substring:incident-v-e} and Nos-\ref{substring:each-odd} hold true.
\item a \emph{weak-harmonious proper total hostring-coloring} if the constraints Nos-\ref{substring:adjacent-v}, Nos-\ref{substring:adjacent-e}, Nos-\ref{substring:incident-v-e} and Nos-\ref{substring:harmonious} hold true.
\item a \emph{weak-graceful proper total hostring-coloring} if the constraints Nos-\ref{substring:adjacent-v}, Nos-\ref{substring:adjacent-e}, Nos-\ref{substring:incident-v-e}, Nos-\ref{substring:e-graceful} and Nos-\ref{substring:part-graceful} hold true.
\item a \emph{weak-odd-graceful proper total hostring-coloring} if the constraints Nos-\ref{substring:adjacent-v}, Nos-\ref{substring:adjacent-e}, Nos-\ref{substring:incident-v-e}, Nos-\ref{substring:each-odd} and Nos-\ref{substring:part-odd-graceful} hold true.
\item a \emph{weak-factor proper total hostring-coloring} if the constraints Nos-\ref{substring:adjacent-v}, Nos-\ref{substring:adjacent-e}, Nos-\ref{substring:incident-v-e} and Nos-\ref{substring:common-factor} hold true.

----- \emph{anti-equitable}

\item a \emph{weak-anti-equitable proper total hostring-coloring} if the constraints Nos-\ref{substring:adjacent-v}, Nos-\ref{substring:adjacent-e}, Nos-\ref{substring:incident-v-e} and Nos-\ref{substring:pairwise-distinct} hold true.

----- \emph{integer partition, factorization}

\item a \emph{weak-partition proper total hostring-coloring} if the constraints Nos-\ref{substring:adjacent-v}, Nos-\ref{substring:adjacent-e}, Nos-\ref{substring:incident-v-e} and Nos-\ref{substring:integer-partition} hold true.
\item a \emph{weak-factorization proper total hostring-coloring} if the constraints Nos-\ref{substring:adjacent-v}, Nos-\ref{substring:adjacent-e}, Nos-\ref{substring:incident-v-e} and Nos-\ref{substring:integer-factorization} hold true.\qqed
\end{asparaenum}
\end{defn}

\begin{defn} \label{defn:weak-homo-sub-proper-hostring-co}
$^*$ \textbf{Weak homogeneous sub-proper total hostring-colorings.} As removing the constraint ``$h(x_k)\neq h(x_ky_k)$ and $h(y_k)\neq h(x_ky_k)$ for each edge $x_ky_k\in E(H)$'' from Definition \ref{defn:homoge-various-string-total-colorings}, we get the following the ``sub-proper'' total hostring-colorings:
\begin{asparaenum}[\textbf{\textrm{Noc}}-1. ]
\item a \emph{weak-subtraction sub-proper total hostring-coloring} if the constraints Nos-\ref{substring:adjacent-v}, Nos-\ref{substring:adjacent-e} and Nos-\ref{substring:e-graceful} hold true.
\item a \emph{weak-odd-subtraction sub-proper total hostring-coloring} if the constraints Nos-\ref{substring:adjacent-v}, Nos-\ref{substring:adjacent-e} and Nos-\ref{substring:each-odd} hold true.
\item a \emph{weak-harmonious sub-proper total hostring-coloring} if the constraints Nos-\ref{substring:adjacent-v}, Nos-\ref{substring:adjacent-e} and Nos-\ref{substring:harmonious} hold true.
\item a \emph{weak-graceful sub-proper total hostring-coloring} if the constraints Nos-\ref{substring:adjacent-v}, Nos-\ref{substring:adjacent-e}, Nos-\ref{substring:e-graceful} and Nos-\ref{substring:part-graceful} hold true.
\item a \emph{weak-odd-graceful sub-proper total hostring-coloring} if the constraints Nos-\ref{substring:adjacent-v}, Nos-\ref{substring:adjacent-e}, Nos-\ref{substring:each-odd} and Nos-\ref{substring:part-odd-graceful} hold true.
\item a \emph{weak-factor sub-proper total hostring-coloring} if the constraints Nos-\ref{substring:adjacent-v}, Nos-\ref{substring:adjacent-e} and Nos-\ref{substring:common-factor} hold true.
\item a \emph{weak-anti-equitable sub-proper total hostring-coloring} if the constraints Nos-\ref{substring:adjacent-v}, Nos-\ref{substring:adjacent-e} and Nos-\ref{substring:pairwise-distinct} hold true.
\item a \emph{weak-partition sub-proper total hostring-coloring} if the constraints Nos-\ref{substring:adjacent-v}, Nos-\ref{substring:adjacent-e} and Nos-\ref{substring:integer-partition} hold true.
\item a \emph{weak-factorization sub-proper total hostring-coloring} if the constraints Nos-\ref{substring:adjacent-v}, Nos-\ref{substring:adjacent-e} and Nos-\ref{substring:integer-factorization} hold true.\qqed
\end{asparaenum}
\end{defn}

\begin{thm}\label{thm:666666}
Each tree admits every one of weak-type sub-proper total hostring-colorings defined in Definition \ref{defn:weak-homo-sub-proper-hostring-co}.
\end{thm}

\begin{defn} \label{defn:homoge-4-magic-string-colorings}
$^*$ \textbf{Magic-type homogeneous total string-coloring.} Let $\textbf{\textrm{S}}_{tring}(n)$ be the set of \emph{$n$-rank number-based strings} $\alpha_{1}\alpha_{2}\cdots \alpha_{n}$ with each number $\alpha_{j}\ge 0$ for $j\in [1,n]$. A $(p,q)$-graph $G$ admits a total string-coloring $g:V(G)\cup E(G)\rightarrow \textbf{\textrm{S}}_{tring}(n)$, such that
\begin{equation}\label{eqa:555555}
g(w_k)=a_{k,1}a_{k,2}\cdots a_{k,n},~g(z_k)=b_{k,1}b_{k,2}\cdots b_{k,n},~g(w_kz_k)=c_{k,1}c_{k,2}\cdots c_{k,n}
\end{equation} for each edge $w_kz_k\in E(G)=\{w_kz_k:k\in [1,q]\}$, and let $\lambda$ be a non-negative integer. For each $k\in [1,q]$, there are the following constraints:
\begin{asparaenum}[\textbf{\textrm{Fma}}-1. ]
\item \label{4-magic:adjacent-v} $g(w_k)\neq g(z_k)$ for each edge $w_kz_k\in E(G)$.
\item \label{4-magic:adjacent-e} $g(w_kz_k)\neq g(w_ku_k)$ for each pair of two adjacent edges $w_kz_k,w_ku_k\in E(G)$.
\item \label{4-magic:incident-v-e} $g(w_k)\neq g(w_kz_k)$ and $g(z_k)\neq g(w_kz_k)$ for each edge $w_kz_k\in E(G)$.

----- \emph{strong 4-magic}

\item \label{4-magic:uniform-edge-magic} Each $j\in [1,n]$ holds the edge-magic constraint $a_{k,j}+b_{k,j}+c_{k,j}=\lambda$ true, denoted as
\begin{equation}\label{eqa:555555}
g(w_k)[+]g(w_kz_k)[+]g(z_k)=\lambda
\end{equation}
\item \label{4-magic:uniform-edge-difference} Each $j\in [1,n]$ holds the edge-difference constraint $c_{k,j}+|a_{k,j}-b_{k,j}|=\lambda$ true, denoted as
\begin{equation}\label{eqa:555555}
g(w_kz_k)[+]|g(w_k)[-]g(z_k)|=\lambda
\end{equation}
\item \label{4-magic:uniform-graceful-difference} Each $j\in [1,n]$ holds the graceful-difference constraint $\big ||a_{k,j}-b_{k,j}|-c_{k,j}\big |=\lambda$ true, denoted as
\begin{equation}\label{eqa:555555}
\big ||g(w_k)[-]g(z_k)|]-]g(w_kz_k)\big |=\lambda
\end{equation}
\item \label{4-magic:uniform-felicitous-difference} Each $j\in [1,n]$ holds the felicitous-difference constraint $|a_{k,j}+b_{k,j}-c_{k,j}|=\lambda$ true, denoted as
\begin{equation}\label{eqa:555555}
|g(w_k)[+]g(z_k)[-]g(w_kz_k)|=\lambda
\end{equation}

----- \emph{weak 4-magic}

\item \label{4-magic:edge-magic} Some $r\in [1,n]$ holds the edge-magic constraint $a_{k,r}+b_{k,r}+c_{k,r}=\gamma$ true, but not all, denoted as $\partial_r \langle g(w_k)[+]g(w_kz_k)[+]g(z_k)=\gamma \rangle$.
\item \label{4-magic:edge-difference} Some $s\in [1,n]$ holds the edge-difference constraint $c_{k,s}+|a_{k,s}-b_{k,s}|=\gamma$ true, but not all, denoted as $\partial_s \langle g(w_kz_k)[+]|g(w_k)[-]g(z_k)|=\gamma \rangle$.
\item \label{4-magic:graceful-difference} Some $t\in [1,n]$ holds the graceful-difference constraint $\big ||a_{k,t}-b_{k,t}|-c_{k,t}\big |=\gamma$ true, but not all, denoted as $\partial_t \langle \big ||g(w_k)[-]g(z_k)|[-]g(w_kz_k)\big |=\gamma \rangle$.
\item \label{4-magic:felicitous-difference} Some $d\in [1,n]$ holds the felicitous-difference constraint $|a_{k,d}+b_{k,d}-c_{k,d}|=\gamma$ true, but not all, denoted as $\partial_d \langle |g(w_k)[+]g(z_k)[-]g(w_kz_k)|=\gamma \rangle$.

----- \emph{different 4-magic}

\item \label{4-magic:different-edge-magic} For each $j\in [1,n]$, there is a positive integer $\mu_j$ such that the edge-magic constraint $a_{k,j}+b_{k,j}+c_{k,j}=\mu_j$ holds true, denoted as \begin{equation}\label{eqa:555555}
 g(w_k)+g(w_kz_k)+g(z_k):=\{\mu_j\}^n_{j=1}
 \end{equation}
\item \label{4-magic:different-edge-difference} For each $j\in [1,n]$, there is a positive integer $\kappa_j$ such that the edge-difference constraint $c_{k,j}+|a_{k,j}-b_{k,j}|=\kappa_j$ true, denoted as \begin{equation}\label{eqa:555555}
 g(w_kz_k)+|g(w_k)-g(z_k)|:=\{\kappa_j\}^n_{j=1}
 \end{equation}
\item \label{4-magic:different-graceful-difference} For each $j\in [1,n]$, there is a non-negative integer $\eta_j$ such that the graceful-difference constraint $\big ||a_{k,j}-b_{k,j}|-c_{k,j}\big |=\eta_j$ true, denoted as \begin{equation}\label{eqa:555555}
 \big ||g(w_k)-g(z_k)|-g(w_kz_k)\big |:=\{\eta_j\}^n_{j=1}
 \end{equation}
\item \label{4-magic:different-felicitous-difference} For each $j\in [1,n]$, there is a non-negative integer $\rho_j$ such that the felicitous-difference constraint $|a_{k,j}+b_{k,j}-c_{k,j}|=\lambda$ true, denoted as \begin{equation}\label{eqa:555555}
 |g(w_k)+g(z_k)-g(w_kz_k)|:=\{\rho_j\}^n_{j=1}
 \end{equation}
\end{asparaenum}
\textbf{Then we call $g$}

\begin{asparaenum}[\textbf{\textrm{Magco}}-1. ]
\item a \emph{$\lambda$-uniformly edge-magic proper total hostring-coloring} if the constraints Fma-\ref{4-magic:adjacent-v}, Fma-\ref{4-magic:adjacent-e}, Fma-\ref{4-magic:incident-v-e} and Fma-\ref{4-magic:uniform-edge-magic} hold true.
\item a \emph{$\lambda$-uniformly edge-difference proper total hostring-coloring} if the constraints Fma-\ref{4-magic:adjacent-v}, Fma-\ref{4-magic:adjacent-e}, Fma-\ref{4-magic:incident-v-e} and Fma-\ref{4-magic:uniform-edge-difference} hold true.
\item a \emph{$\lambda$-uniformly graceful-difference proper total hostring-coloring} if the constraints Fma-\ref{4-magic:adjacent-v}, Fma-\ref{4-magic:adjacent-e}, Fma-\ref{4-magic:incident-v-e} and Fma-\ref{4-magic:uniform-graceful-difference} hold true.
\item a \emph{$\lambda$-uniformly felicitous-difference proper total hostring-coloring} if the constraints Fma-\ref{4-magic:adjacent-v}, Fma-\ref{4-magic:adjacent-e}, Fma-\ref{4-magic:incident-v-e} and Fma-\ref{4-magic:uniform-felicitous-difference} hold true.

----- \emph{weak 4-magic}

\item a \emph{weak-edge-magic proper total hostring-coloring} if the constraints Fma-\ref{4-magic:adjacent-v}, Fma-\ref{4-magic:adjacent-e}, Fma-\ref{4-magic:incident-v-e} and Fma-\ref{4-magic:edge-magic} hold true.
\item a \emph{weak-edge-difference proper total hostring-coloring} if the constraints Fma-\ref{4-magic:adjacent-v}, Fma-\ref{4-magic:adjacent-e}, Fma-\ref{4-magic:incident-v-e} and Fma-\ref{4-magic:edge-difference} hold true.
\item a \emph{weak-graceful-difference proper total hostring-coloring} if the constraints Fma-\ref{4-magic:adjacent-v}, Fma-\ref{4-magic:adjacent-e}, Fma-\ref{4-magic:incident-v-e} and Fma-\ref{4-magic:graceful-difference} hold true.
\item a \emph{weak-felicitous-difference proper total hostring-coloring} if the constraints Fma-\ref{4-magic:adjacent-v}, Fma-\ref{4-magic:adjacent-e}, Fma-\ref{4-magic:incident-v-e} and Fma-\ref{4-magic:felicitous-difference} hold true.

----- \emph{mixed 4-magic}

\item a \emph{mixed-4-magic proper total hostring-coloring} if the constraints Fma-\ref{4-magic:adjacent-v}, Fma-\ref{4-magic:adjacent-e}, Fma-\ref{4-magic:incident-v-e}, Fma-\ref{4-magic:edge-magic}, Fma-\ref{4-magic:edge-difference}, Fma-\ref{4-magic:graceful-difference} and Fma-\ref{4-magic:felicitous-difference} hold true.

----- \emph{different magic constants}

\item a \emph{$\{\mu_j\}^n_{j=1}$-edge-magic proper total hostring-coloring} if the constraints Fma-\ref{4-magic:adjacent-v}, Fma-\ref{4-magic:adjacent-e}, Fma-\ref{4-magic:incident-v-e} and Fma-\ref{4-magic:different-edge-magic} hold true.
\item a \emph{$\{\kappa_j\}^n_{j=1}$-edge-difference proper total hostring-coloring} if the constraints Fma-\ref{4-magic:adjacent-v}, Fma-\ref{4-magic:adjacent-e}, Fma-\ref{4-magic:incident-v-e} and Fma-\ref{4-magic:different-edge-difference} hold true.
\item a \emph{$\{\eta_j\}^n_{j=1}$-graceful-difference proper total hostring-coloring} if the constraints Fma-\ref{4-magic:adjacent-v}, Fma-\ref{4-magic:adjacent-e}, Fma-\ref{4-magic:incident-v-e} and Fma-\ref{4-magic:different-graceful-difference} hold true.
\item a \emph{$\{\rho_j\}^n_{j=1}$-felicitous-difference proper total hostring-coloring} if the constraints Fma-\ref{4-magic:adjacent-v}, Fma-\ref{4-magic:adjacent-e}, Fma-\ref{4-magic:incident-v-e} and Fma-\ref{4-magic:different-felicitous-difference} hold true.\qqed
\end{asparaenum}
\end{defn}

\begin{defn} \label{defn:magic-type-sub-proper-tring-colorings}
$^*$ \textbf{Magic-type sub-proper total hostring-colorings.} As removing the constraint ``$g(w_k)\neq g(w_kz_k)$ and $g(z_k)\neq g(w_kz_k)$ for each edge $w_kz_k\in E(G)$'' from Definition \ref{defn:homoge-4-magic-string-colorings}, we get the following the ``sub-proper'' total hostring-colorings:
\begin{asparaenum}[\textbf{\textrm{Magco}}-1. ]
\item a \emph{$\lambda$-uniformly edge-magic sub-proper total hostring-coloring} if the constraints Fma-\ref{4-magic:adjacent-v}, Fma-\ref{4-magic:adjacent-e} and Fma-\ref{4-magic:uniform-edge-magic} hold true.
\item a \emph{$\lambda$-uniformly edge-difference sub-proper total hostring-coloring} if the constraints Fma-\ref{4-magic:adjacent-v}, Fma-\ref{4-magic:adjacent-e} and Fma-\ref{4-magic:uniform-edge-difference} hold true.
\item a \emph{$\lambda$-uniformly graceful-difference sub-proper total hostring-coloring} if the constraints Fma-\ref{4-magic:adjacent-v}, Fma-\ref{4-magic:adjacent-e} and Fma-\ref{4-magic:uniform-graceful-difference} hold true.
\item a \emph{$\lambda$-uniformly felicitous-difference sub-proper total hostring-coloring} if the constraints Fma-\ref{4-magic:adjacent-v}, Fma-\ref{4-magic:adjacent-e} and Fma-\ref{4-magic:uniform-felicitous-difference} hold true.
\item a \emph{weak-edge-magic sub-proper total hostring-coloring} if the constraints Fma-\ref{4-magic:adjacent-v}, Fma-\ref{4-magic:adjacent-e} and Fma-\ref{4-magic:edge-magic} hold true.
\item a \emph{weak-edge-difference sub-proper total hostring-coloring} if the constraints Fma-\ref{4-magic:adjacent-v}, Fma-\ref{4-magic:adjacent-e} and Fma-\ref{4-magic:edge-difference} hold true.
\item a \emph{weak-graceful-difference sub-proper total hostring-coloring} if the constraints Fma-\ref{4-magic:adjacent-v}, Fma-\ref{4-magic:adjacent-e} and Fma-\ref{4-magic:graceful-difference} hold true.
\item a \emph{weak-felicitous-difference sub-proper total hostring-coloring} if the constraints Fma-\ref{4-magic:adjacent-v}, Fma-\ref{4-magic:adjacent-e} and Fma-\ref{4-magic:felicitous-difference} hold true.
\item a \emph{mixed-4-magic sub-proper total hostring-coloring} if the constraints Fma-\ref{4-magic:adjacent-v}, Fma-\ref{4-magic:adjacent-e}, Fma-\ref{4-magic:edge-magic}, Fma-\ref{4-magic:edge-difference}, Fma-\ref{4-magic:graceful-difference} and Fma-\ref{4-magic:felicitous-difference} hold true.
\item a \emph{$\{\mu_j\}^n_{j=1}$-edge-magic sub-proper total hostring-coloring} if the constraints Fma-\ref{4-magic:adjacent-v}, Fma-\ref{4-magic:adjacent-e} and Fma-\ref{4-magic:different-edge-magic} hold true.
\item a \emph{$\{\kappa_j\}^n_{j=1}$-edge-difference sub-proper total hostring-coloring} if the constraints Fma-\ref{4-magic:adjacent-v}, Fma-\ref{4-magic:adjacent-e} and Fma-\ref{4-magic:different-edge-difference} hold true.
\item a \emph{$\{\eta_j\}^n_{j=1}$-graceful-difference sub-proper total hostring-coloring} if the constraints Fma-\ref{4-magic:adjacent-v}, Fma-\ref{4-magic:adjacent-e} and Fma-\ref{4-magic:different-graceful-difference} hold true.
\item a \emph{$\{\rho_j\}^n_{j=1}$-felicitous-difference sub-proper total hostring-coloring} if the constraints Fma-\ref{4-magic:adjacent-v}, Fma-\ref{4-magic:adjacent-e} and Fma-\ref{4-magic:different-felicitous-difference} hold true.\qqed
\end{asparaenum}
\end{defn}

By Theorem \ref{thm:10-k-d-total-coloringss}, we have
\begin{thm}\label{thm:666666}
$^*$ Each tree admits every one of the $W$-constraint sub-proper total hostring-colorings defined in Definition \ref{defn:magic-type-sub-proper-tring-colorings}.
\end{thm}

\begin{rem}\label{rem:333333}
Three number-based strings in Eq.(\ref{eqa:homogeneous-string-colorings}) in Definition \ref{defn:homoge-uniformly-string-total-colorings} correspond to the following matrices
\begin{equation}\label{eqa:555555}
X_k=(a_{k,1},a_{k,2},\dots ,a_{k,n}),~E_k=(b_{k,1},b_{k,2},\dots ,b_{k,n}),~Y_k=(c_{k,1},c_{k,2},\dots ,c_{k,n})
\end{equation} for $k\in [1,q]$, which form the Topcode-matrix $T_{code}(H_k,\varphi_k)=(X_k,E_k,Y_k)^T$ of a graph $H$ of $n$ edges under one of the $W$-constraints introduced in Definition \ref{defn:homoge-uniformly-string-total-colorings} and Definition \ref{defn:homoge-4-magic-string-colorings}.

For example, there is a constraint: \emph{Each $j\in [1,n]$ holds the edge-difference constraint $c_{k,j}+|a_{k,j}-b_{k,j}|=\lambda$ true} in Definition \ref{defn:homoge-4-magic-string-colorings}, which meas that the graph $H_k$ admits an edge-difference total coloring $\varphi_k$ induced by the $\lambda$-uniformly edge-difference proper total hostring-coloring $g$ of the $(p,q)$-graph $G$ defined in Definition \ref{defn:homoge-4-magic-string-colorings}. Then the $(p,q)$-graph $G$ admitting a $\lambda$-uniformly edge-difference proper total hostring-coloring $g$ connects $q$ graphs $H_k$ for $k\in [1,q]$ together, through the vertex-coinciding operation, such that two adjacent edges $u_kv_k$ and $u_kw_{k\,'}$ correspond to two graphs $H_k$ and $H_{k\,'}$ holding $H_k\odot H_{k\,'}$ under the vertex-coinciding operation, since $V(H_k)\cap V(H_{k\,'})=\{a_{k,1},a_{k,2},\dots ,a_{k,n}\}$. In other word, the coloring $g$ is a particular \emph{graph-coloring}.

Thereby, replacing the vertices and edges of $G$ by graphs $H_k$ for $k\in [1,q]$ produces the graph, denoted as $G^*=G[\odot ]^q_{k=1}H_k$.

More important, the graph $G$ has its own Topcode-matrix
$$T_{code}(G,g):=(T_{code}(H_1,\varphi_1),T_{code}(H_2,\varphi_2),\dots ,T_{code}(H_q,\varphi_q))
$$ to be a \emph{three-dimensional Topcode-matrix}.\paralled
\end{rem}

\begin{defn} \label{defn:magic-kaleidoscope-t-string-coloring}
$^*$ \textbf{The magic-kaleidoscope total string-colorings.} Suppose that a connected $(p,q)$-graph $G$ admits a total string-coloring $g:V(G)\cup E(G)\rightarrow \textbf{\textrm{S}}_{tring}(n)$ with $g(w_k)\neq g(z_k)$ for each edge $w_kz_k\in E(G)$, and $g(w_kz_k)\neq g(w_ku_k)$ for each pair of two adjacent edges $w_kz_k,w_ku_k\in E(G)$. Let $S_{magic}=\{$Fma-\ref{4-magic:uniform-edge-magic}, Fma-\ref{4-magic:uniform-edge-difference}, Fma-\ref{4-magic:uniform-graceful-difference}, Fma-\ref{4-magic:uniform-felicitous-difference}, Fma-\ref{4-magic:edge-magic}, Fma-\ref{4-magic:edge-difference}, Fma-\ref{4-magic:graceful-difference}, Fma-\ref{4-magic:felicitous-difference}, Fma-\ref{4-magic:different-edge-magic}, Fma-\ref{4-magic:different-edge-difference}, Fma-\ref{4-magic:different-graceful-difference}, Fma-\ref{4-magic:different-felicitous-difference}$\}$ defined in Definition \ref{defn:homoge-4-magic-string-colorings}.

(i) If the edge set $E(G)$ of the graph $G$ can be classified into different subsets, that is $E(G)=\bigcup E_{\varepsilon}$ for each $\varepsilon \in S_{magic}$, such that the colors $g(w_k)$, $g(w_kz_k)$ and $g(z_k)$ for each edge $w_kz_k\in E_{\varepsilon}$ holds the constraint $\varepsilon$ true, then the total string-coloring $g$ is called \emph{magic-kaleidoscope sub-proper total string-coloring}.

(ii) If the edge set $E(G)=\bigcup ^m_{j=1} E_{\varepsilon_j}$ for $\varepsilon_j \in S_{magic}$ and $2\leq m$, such that the colors $g(w_k)$, $g(w_kz_k)$ and $g(z_k)$ for each edge $w_kz_k\in E_{\varepsilon_j}$ holds the constraint $\varepsilon_j$ true, then the total string-coloring $g$ is called \emph{partly magic-kaleidoscope sub-proper total string-coloring}.\qqed
\end{defn}

\begin{problem}\label{qeu:444444}
About Definition \ref{defn:homoge-uniformly-string-total-colorings}, Definition \ref{defn:homoge-various-string-total-colorings} and Definition \ref{defn:homoge-4-magic-string-colorings}, we have

(i) For the set $\textbf{\textrm{S}}_{tring}(1)$ as $n=1$, the \emph{hostring-colorings} defined in Definition \ref{defn:homoge-uniformly-string-total-colorings}, Definition \ref{defn:homoge-various-string-total-colorings} and Definition \ref{defn:homoge-4-magic-string-colorings} are popular colorings of graph theory.

(ii) \textbf{Find} graphs admitting part of the \emph{hostring-colorings} defined in Definition \ref{defn:homoge-uniformly-string-total-colorings}, Definition \ref{defn:homoge-various-string-total-colorings} and Definition \ref{defn:homoge-4-magic-string-colorings}.

(iii) Let $S^{magic}_{W}$ be a $W$-constraint subset of $\textbf{\textrm{S}}_{tring}(n)$. For example, $W$-constraint $\in S_{magic}$, where the set $S_{magic}$ is defined in Definition \ref{defn:magic-kaleidoscope-t-string-coloring}, \textbf{does} $S^{magic}_{W}\neq \emptyset$?
\end{problem}

\subsection{Non-homogeneous number-based string-colorings}

\begin{defn} \label{defn:non-homoge-string-colorings}
$^*$ \textbf{Non-homogeneous magic-constraint total string-colorings.} Let $\textbf{\textrm{S}}_{tring}(\leq n)$ be the set of \emph{$m$-rank number-based strings} $\alpha_{1}\alpha_{2}\cdots \alpha_{m}$ with each number $\alpha_{j}\ge 0$ for $j\in [1,m]$ and $m\leq n$. A $(p,q)$-graph $G$ admits a total string-coloring $g:V(G)\cup E(G)\rightarrow \textbf{\textrm{S}}_{tring}(\leq n)$. For the edge set $E(G)=\{w_kz_k:k\in [1,q]\}$, we have
\begin{equation}\label{eqa:555555}
g(w_k)=a_{k,1}a_{k,2}\cdots a_{k,m(a,k)},~g(z_k)=b_{k,1}b_{k,2}\cdots b_{k,m(b,k)},~g(w_kz_k)=c_{k,1}c_{k,2}\cdots c_{k,m(c,k)}
\end{equation} for each edge $w_kz_k\in E(G)$ with $1\leq m(a,k), m(b,k), m(c,k)\leq n$; and let $\eta$ be a non-negative integer. For each $k\in [1,q]$, we define

\textbf{Statement-ABC-I}: Each $c_{k,j}$ with $j\in [1,m(c,k)]$ corresponds to $a_{k,j\,'}$ and $b_{k,j\,''}$ with $j\,'\in [1,m(a,k)]$ and $j\,''\in [1,m(b,k)]$ holding \textbf{A} true, and each $a_{k,j}$ with $j\in [1,m(a,k)]$ corresponds to $b_{k,i\,'}$ and $c_{k,i\,''}$ with $i\,'\in [1,m(b,k)]$ and $i\,''\in [1,m(c,k)]$ holding \textbf{B} true, as well as each $b_{k,j}$ with $j\in [1,m(b,k)]$ corresponds to $a_{k,s\,'}$ and $c_{k,s\,''}$ with $s\,'\in [1,m(a,k)]$ and $s\,''\in [1,m(c,k)]$ holding \textbf{C} true.

\textbf{Statement-ABC-II}: Each $c_{k,j}$ with $j\in [1,m(c,k)]$ corresponds to $a_{k,j\,'}$ and $b_{k,j\,''}$ with $j\,'\in [1,m(a,k)]$ and $j\,''\in [1,m(b,k)]$ holding \textbf{A-$\mu_j$} with $\mu_j\in \{\mu_j\}^{m(c,k)}_{j=1}$ true, and each $a_{k,i}$ with $i\in [1,m(a,k)]$ corresponds to $b_{k,i\,'}$ and $c_{k,i\,''}$ with $i\,'\in [1,m(b,k)]$ and $i\,''\in [1,m(c,k)]$ holding \textbf{B-$\mu_i$} with $\mu_i\in \{\mu_j\}^{m(c,k)}_{j=1}$ true, as well as each $b_{k,s}$ with $s\in [1,m(b,k)]$ corresponds to $a_{k,s\,'}$ and $c_{k,s\,''}$ with $s\,'\in [1,m(a,k)]$ and $s\,''\in [1,m(c,k)]$ holding \textbf{C-$\mu_s$} with $\mu_s\in \{\mu_j\}^{m(c,k)}_{j=1}$ true.

There are the following constraints:
\begin{asparaenum}[\textbf{\textrm{Nm}}-1. ]
\item \label{nonhomoge:adjacent-v} $g(w_k)\neq g(z_k)$ for each edge $w_kz_k\in E(G)$.
\item \label{nonhomoge:adjacent-e} $g(w_kz_k)\neq g(w_ku_k)$ for each pair of two adjacent edges $w_kz_k,w_ku_k\in E(G)$.
\item \label{nonhomoge:incident-v-e} $g(w_k)\neq g(w_kz_k)$ and $g(z_k)\neq g(w_kz_k)$ for each edge $w_kz_k\in E(G)$.

----- \emph{traditional 4-magic}

\item \label{nonhomoge:tradi-edge-magic} Each number $c_{k,r}$ with $r\in [1,m(c,k)]$ corresponds to $a_{k,r\,'}$ and $b_{k,r\,''}$ with $r\,'\in [1,m(a,k)]$ and $r\,''\in [1,m(b,k)]$ holding the edge-magic constraint $a_{k,r\,'}+b_{k,r\,''}+c_{k,r}=\gamma$ true.
\item \label{nonhomoge:tradi-edge-difference} Each number $c_{k,s}$ with $s\in [1,m(c,k)]$ corresponds to $a_{k,s\,'}$ and $b_{k,s\,''}$ with $s\,'\in [1,m(a,k)]$ and $s\,''\in [1,m(b,k)]$ holding the edge-difference constraint $c_{k,s}+|a_{k,s\,'}-b_{k,s\,''}|=\gamma$ true.
\item \label{nonhomoge:tradi-graceful-difference} Each number $c_{k,t}$ with $t\in [1,m(c,k)]$ corresponds to $a_{k,t\,'}$ and $b_{k,t\,''}$ with $t\,'\in [1,m(a,k)]$ and $t\,''\in [1,m(b,k)]$ holding the graceful-difference constraint $\big ||a_{k,t\,'}-b_{k,t\,''}|-c_{k,t}\big |=\gamma$ true.
\item \label{nonhomoge:tradi-felicitous-difference} Each number $c_{k,d}$ with $d\in [1,m(c,k)]$ corresponds to $a_{k,d\,'}$ and $b_{k,d\,''}$ with $d\,'\in [1,m(a,k)]$ and $d\,''\in [1,m(b,k)]$ holding the felicitous-difference constraint $|a_{k,d\,'}+b_{k,d\,''}-c_{k,d}|=\gamma$ true.

----- \emph{strong 4-magic}

\item \label{nonhomoge:uniform-edge-magic} \textbf{Statement-ABC-I}: \textbf{A}=the edge-magic constraint $a_{k,j\,'}+b_{k,j\,''}+c_{k,j}=\eta$, \textbf{B}=the edge-magic constraint $a_{k,j}+b_{k,i\,'}+c_{k,i\,''}=\eta$, \textbf{C}=the edge-magic constraint $a_{k,s\,'}+b_{k,j}+c_{k,s\,''}=\eta$.
\item \label{nonhomoge:uniform-edge-difference} \textbf{Statement-ABC-I}: \textbf{A}=the edge-difference constraint $c_{k,j}+|a_{k,j\,'}-b_{k,j\,''}|=\eta$, \textbf{B}=the edge-difference constraint $c_{k,i\,''}+|a_{k,j}-b_{k,i\,'}|=\eta$, \textbf{C}=the edge-difference constraint $c_{k,s\,''}+|a_{k,s\,'}-b_{k,j}|=\eta$.
\item \label{nonhomoge:uniform-graceful-difference} \textbf{Statement-ABC-I}: \textbf{A}=the graceful-difference constraint $\big ||a_{k,j\,'}-b_{k,j\,''}|-c_{k,j}\big |=\eta$, \textbf{B}=the graceful-difference constraint $\big ||a_{k,j}-b_{k,i\,'}|-c_{k,i\,''}\big |=\eta$, \textbf{C}=the graceful-difference constraint $\big ||a_{k,s\,'}-b_{k,j}|-c_{k,s\,''}\big |=\eta$.
\item \label{nonhomoge:uniform-felicitous-difference} \textbf{Statement-ABC-I}: \textbf{A}=the felicitous-difference constraint $|a_{k,j\,'}+b_{k,j\,''}-c_{k,j}|=\eta$, \textbf{B}=the felicitous-difference constraint $|a_{k,j}+b_{k,i\,'}-c_{k,i\,''}|=\eta$, \textbf{C}=the felicitous-difference constraint $|a_{k,s\,'}+b_{k,j}-c_{k,s\,''}|=\eta$.

----- \emph{multiple magic constants}

\item \label{nonhomoge:different-edge-magic} \textbf{Statement-ABC-II}: \textbf{A-$\mu_j$}=the edge-magic constraint $a_{k,j\,'}+b_{k,j\,''}+c_{k,j}=\mu_j$ with $\mu_j\in \{\mu_j\}^{m(c,k)}_{j=1}$, and \textbf{B-$\mu_i$}=the edge-magic constraint $a_{k,i}+b_{k,i\,'}+c_{k,i\,''}=\mu_i$ with $\mu_i\in \{\mu_j\}^{m(c,k)}_{j=1}$, as well as \textbf{C-$\mu_s$}=the edge-magic constraint $a_{k,s\,'}+b_{k,s}+c_{k,s\,''}=\mu_s$ with $\mu_s\in \{\mu_j\}^{m(c,k)}_{j=1}$.
\item \label{nonhomoge:different-edge-difference} \textbf{Statement-ABC-II}: \textbf{A-$\mu_j$}=the edge-difference constraint $c_{k,j}+|a_{k,j\,'}-b_{k,j\,''}|=\mu_j$ with $\mu_j\in \{\mu_j\}^{m(c,k)}_{j=1}$, and \textbf{B-$\mu_i$}=the edge-difference constraint $c_{k,i\,''}+|a_{k,i}-b_{k,i\,'}|=\mu_i$ with $\mu_i\in \{\mu_j\}^{m(c,k)}_{j=1}$, as well as \textbf{C-$\mu_s$}=the edge-difference constraint $c_{k,s\,''}+|a_{k,s\,'}-b_{k,s}|=\mu_s$ with $\mu_s\in \{\mu_j\}^{m(c,k)}_{j=1}$.
\item \label{nonhomoge:different-graceful-difference} \textbf{Statement-ABC-II}: \textbf{A-$\mu_j$}=the graceful-difference constraint $\big ||a_{k,j\,'}-b_{k,j\,''}|-c_{k,j}\big |=\mu_j$ with $\mu_j\in \{\mu_j\}^{m(c,k)}_{j=1}$, and \textbf{B-$\mu_i$}=the graceful-difference constraint $\big ||a_{k,i}-b_{k,i\,'}|-c_{k,i\,''}\big |=\mu_i$ with $\mu_i\in \{\mu_j\}^{m(c,k)}_{j=1}$, as well as \textbf{C-$\mu_s$}=the graceful-difference constraint $\big ||a_{k,s\,'}-b_{k,s}|-c_{k,s\,''}\big |=\mu_s$ with $\mu_s\in \{\mu_j\}^{m(c,k)}_{j=1}$.
\item \label{nonhomoge:different-felicitous-difference} \textbf{Statement-ABC-II}: \textbf{A-$\mu_j$}=the felicitous-difference constraint $|a_{k,j\,'}+b_{k,j\,''}-c_{k,j}|=\mu_j$ with $\mu_j\in \{\mu_j\}^{m(c,k)}_{j=1}$, and \textbf{B-$\mu_i$}=the felicitous-difference constraint $|a_{k,i}+b_{k,i\,'}-c_{k,i\,''}|=\mu_i$ with $\mu_i\in \{\mu_j\}^{m(c,k)}_{j=1}$, as well as \textbf{C-$\mu_s$}=the felicitous-difference constraint $|a_{k,s\,'}+b_{k,s}-c_{k,s\,''}|=\mu_s$ with $\mu_s\in \{\mu_j\}^{m(c,k)}_{j=1}$.

----- \emph{weak 4-magic}

\item \label{nonhomoge:edge-magic} Some $c_{k,r}$ with $r\in [1,m(c,k)]$ corresponds to $a_{k,r\,'}$ and $b_{k,r\,''}$ with $r\,'\in [1,m(a,k)]$ and $r\,''\in [1,m(b,k)]$ holding the edge-magic constraint $a_{k,r\,'}+b_{k,r\,''}+c_{k,r}=\gamma$ true, but not all.
\item \label{nonhomoge:edge-difference} Some $c_{k,s}$ with $s\in [1,m(c,k)]$ corresponds to $a_{k,s\,'}$ and $b_{k,s\,''}$ with $s\,'\in [1,m(a,k)]$ and $s\,''\in [1,m(b,k)]$ holding the edge-difference constraint $c_{k,s}+|a_{k,s\,'}-b_{k,s\,''}|=\gamma$ true, but not all.
\item \label{nonhomoge:graceful-difference} Some $c_{k,t}$ with $t\in [1,m(c,k)]$ corresponds to $a_{k,t\,'}$ and $b_{k,t\,''}$ with $t\,'\in [1,m(a,k)]$ and $t\,''\in [1,m(b,k)]$ holding the graceful-difference constraint $\big ||a_{k,t\,'}-b_{k,t\,''}|-c_{k,t}\big |=\gamma$ true, but not all.
\item \label{nonhomoge:felicitous-difference} Some $c_{k,d}$ with $d\in [1,m(c,k)]$ corresponds to $a_{k,d\,'}$ and $b_{k,d\,''}$ with $d\,'\in [1,m(a,k)]$ and $d\,''\in [1,m(b,k)]$ holding the felicitous-difference constraint $|a_{k,d\,'}+b_{k,d\,''}-c_{k,d}|=\gamma$ true, but not all.
\end{asparaenum}
\textbf{Then $g$ is}

\begin{asparaenum}[\textbf{\textrm{Magco}}-1. ]
\item an \emph{edge-magic proper total string-coloring} if the constraints Nm-\ref{nonhomoge:adjacent-v}, Nm-\ref{nonhomoge:adjacent-e}, Nm-\ref{nonhomoge:incident-v-e} and Nm-\ref{nonhomoge:tradi-edge-magic} hold true.
\item an \emph{edge-difference proper total string-coloring} if the constraints Nm-\ref{nonhomoge:adjacent-v}, Nm-\ref{nonhomoge:adjacent-e}, Nm-\ref{nonhomoge:incident-v-e} and Nm-\ref{nonhomoge:tradi-edge-difference} hold true.
\item a \emph{graceful-difference proper total string-coloring} if the constraints Nm-\ref{nonhomoge:adjacent-v}, Nm-\ref{nonhomoge:adjacent-e}, Nm-\ref{nonhomoge:incident-v-e} and Nm-\ref{nonhomoge:tradi-graceful-difference} hold true.
\item a \emph{felicitous-difference proper total string-coloring} if the constraints Nm-\ref{nonhomoge:adjacent-v}, Nm-\ref{nonhomoge:adjacent-e}, Nm-\ref{nonhomoge:incident-v-e} and Nm-\ref{nonhomoge:tradi-felicitous-difference} hold true.

----- \emph{uniformly magic}

\item a \emph{$\eta$-uniformly edge-magic proper total string-coloring} if the constraints Nm-\ref{nonhomoge:adjacent-v}, Nm-\ref{nonhomoge:adjacent-e}, Nm-\ref{nonhomoge:incident-v-e} and Nm-\ref{nonhomoge:uniform-edge-magic} hold true.
\item a \emph{$\eta$-uniformly edge-difference proper total string-coloring} if the constraints Nm-\ref{nonhomoge:adjacent-v}, Nm-\ref{nonhomoge:adjacent-e}, Nm-\ref{nonhomoge:incident-v-e} and Nm-\ref{nonhomoge:uniform-edge-difference} hold true.
\item a \emph{$\eta$-uniformly graceful-difference proper total string-coloring} if the constraints Nm-\ref{nonhomoge:adjacent-v}, Nm-\ref{nonhomoge:adjacent-e}, Nm-\ref{nonhomoge:incident-v-e} and Nm-\ref{nonhomoge:uniform-graceful-difference} hold true.
\item a \emph{$\eta$-uniformly felicitous-difference proper total string-coloring} if the constraints Nm-\ref{nonhomoge:adjacent-v}, Nm-\ref{nonhomoge:adjacent-e}, Nm-\ref{nonhomoge:incident-v-e} and Nm-\ref{nonhomoge:uniform-felicitous-difference} hold true.

----- \emph{multiple magic constants}

\item a \emph{$\{\{\mu_{k,j}\}^{m(c,k)}_{j=1}\}^q_{k=1}$-edge-magic proper total string-coloring} if the constraints Nm-\ref{nonhomoge:adjacent-v}, Nm-\ref{nonhomoge:adjacent-e}, Nm-\ref{nonhomoge:incident-v-e} and Nm-\ref{nonhomoge:different-edge-magic} hold true.
\item a \emph{$\{\{\mu_{k,j}\}^{m(c,k)}_{j=1}\}^q_{k=1}$-edge-difference proper total string-coloring} if the constraints Nm-\ref{nonhomoge:adjacent-v}, Nm-\ref{nonhomoge:adjacent-e}, Nm-\ref{nonhomoge:incident-v-e} and Nm-\ref{nonhomoge:different-edge-difference} hold true.
\item a \emph{$\{\{\mu_{k,j}\}^{m(c,k)}_{j=1}\}^q_{k=1}$-graceful-difference proper total string-coloring} if the constraints Nm-\ref{nonhomoge:adjacent-v}, Nm-\ref{nonhomoge:adjacent-e}, Nm-\ref{nonhomoge:incident-v-e} and Nm-\ref{nonhomoge:different-graceful-difference} hold true.
\item a \emph{$\{\{\mu_{k,j}\}^{m(c,k)}_{j=1}\}^q_{k=1}$-felicitous-difference proper total string-coloring} if the constraints Nm-\ref{nonhomoge:adjacent-v}, Nm-\ref{nonhomoge:adjacent-e}, Nm-\ref{nonhomoge:incident-v-e} and Nm-\ref{nonhomoge:different-felicitous-difference} hold true.

----- \emph{weak 4-magic}

\item a \emph{weak-edge-magic proper total string-coloring} if the constraints Nm-\ref{nonhomoge:adjacent-v}, Nm-\ref{nonhomoge:adjacent-e}, Nm-\ref{nonhomoge:incident-v-e} and Nm-\ref{nonhomoge:edge-magic} hold true.
\item a \emph{weak-edge-difference proper total string-coloring} if the constraints Nm-\ref{nonhomoge:adjacent-v}, Nm-\ref{nonhomoge:adjacent-e}, Nm-\ref{nonhomoge:incident-v-e} and Nm-\ref{nonhomoge:edge-difference} hold true.
\item a \emph{weak-graceful-difference proper total string-coloring} if the constraints Nm-\ref{nonhomoge:adjacent-v}, Nm-\ref{nonhomoge:adjacent-e}, Nm-\ref{nonhomoge:incident-v-e} and Nm-\ref{nonhomoge:graceful-difference} hold true.
\item a \emph{weak-felicitous-difference proper total string-coloring} if the constraints Nm-\ref{nonhomoge:adjacent-v}, Nm-\ref{nonhomoge:adjacent-e}, Nm-\ref{nonhomoge:incident-v-e} and Nm-\ref{nonhomoge:felicitous-difference} hold true.

----- \emph{mixed-4-magic}

\item a \emph{mixed-4-magic proper total string-coloring} if the constraints Nm-\ref{nonhomoge:adjacent-v}, Nm-\ref{nonhomoge:adjacent-e}, Nm-\ref{nonhomoge:incident-v-e}, Nm-\ref{nonhomoge:edge-magic}, Nm-\ref{nonhomoge:edge-difference}, Nm-\ref{nonhomoge:graceful-difference} and Nm-\ref{nonhomoge:felicitous-difference} hold true.\qqed
\end{asparaenum}
\end{defn}

\begin{defn} \label{defn:non-homoge-sub-proper-eachostring-coloring}
$^*$ \textbf{Non-homogeneous magic-constraint sub-proper total string-colorings.} As removing the constraint ``$g(w_k)\neq g(w_kz_k)$ and $g(z_k)\neq g(w_kz_k)$ for each edge $w_kz_k\in E(G)$'' from Definition \ref{defn:non-homoge-string-colorings}, \textbf{then $g$ is}:
\begin{asparaenum}[\textbf{\textrm{Subpro}}-1. ]
\item an \emph{edge-magic sub-proper total string-coloring} if the constraints Nm-\ref{nonhomoge:adjacent-v}, Nm-\ref{nonhomoge:adjacent-e} and Nm-\ref{nonhomoge:tradi-edge-magic} hold true.
\item an \emph{edge-difference sub-proper total string-coloring} if the constraints Nm-\ref{nonhomoge:adjacent-v}, Nm-\ref{nonhomoge:adjacent-e} and Nm-\ref{nonhomoge:tradi-edge-difference} hold true.
\item a \emph{graceful-difference sub-proper total string-coloring} if the constraints Nm-\ref{nonhomoge:adjacent-v}, Nm-\ref{nonhomoge:adjacent-e} and Nm-\ref{nonhomoge:tradi-graceful-difference} hold true.
\item a \emph{felicitous-difference sub-proper total string-coloring} if the constraints Nm-\ref{nonhomoge:adjacent-v}, Nm-\ref{nonhomoge:adjacent-e} and Nm-\ref{nonhomoge:tradi-felicitous-difference} hold true.

\item a \emph{$\eta$-uniformly edge-magic sub-proper total string-coloring} if the constraints Nm-\ref{nonhomoge:adjacent-v}, Nm-\ref{nonhomoge:adjacent-e} and Nm-\ref{nonhomoge:uniform-edge-magic} hold true.
\item a \emph{$\eta$-uniformly edge-difference sub-proper total string-coloring} if the constraints Nm-\ref{nonhomoge:adjacent-v}, Nm-\ref{nonhomoge:adjacent-e} and Nm-\ref{nonhomoge:uniform-edge-difference} hold true.
\item a \emph{$\eta$-uniformly graceful-difference sub-proper total string-coloring} if the constraints Nm-\ref{nonhomoge:adjacent-v}, Nm-\ref{nonhomoge:adjacent-e} and Nm-\ref{nonhomoge:uniform-graceful-difference} hold true.
\item a \emph{$\eta$-uniformly felicitous-difference sub-proper total string-coloring} if the constraints Nm-\ref{nonhomoge:adjacent-v}, Nm-\ref{nonhomoge:adjacent-e} and Nm-\ref{nonhomoge:uniform-felicitous-difference} hold true.

\item a \emph{$\{\{\mu_{k,j}\}^{m(c,k)}_{j=1}\}^q_{k=1}$-edge-magic sub-proper total string-coloring} if the constraints Nm-\ref{nonhomoge:adjacent-v}, Nm-\ref{nonhomoge:adjacent-e} and Nm-\ref{nonhomoge:different-edge-magic} hold true.
\item a \emph{$\{\{\mu_{k,j}\}^{m(c,k)}_{j=1}\}^q_{k=1}$-edge-difference sub-proper total string-coloring} if the constraints Nm-\ref{nonhomoge:adjacent-v}, Nm-\ref{nonhomoge:adjacent-e} and Nm-\ref{nonhomoge:different-edge-difference} hold true.
\item a \emph{$\{\{\mu_{k,j}\}^{m(c,k)}_{j=1}\}^q_{k=1}$-graceful-difference sub-proper total string-coloring} if the constraints Nm-\ref{nonhomoge:adjacent-v}, Nm-\ref{nonhomoge:adjacent-e} and Nm-\ref{nonhomoge:different-graceful-difference} hold true.
\item a \emph{$\{\{\mu_{k,j}\}^{m(c,k)}_{j=1}\}^q_{k=1}$-felicitous-difference sub-proper total string-coloring} if the constraints Nm-\ref{nonhomoge:adjacent-v}, Nm-\ref{nonhomoge:adjacent-e} and Nm-\ref{nonhomoge:different-felicitous-difference} hold true.

\item a \emph{weak-edge-magic sub-proper total string-coloring} if the constraints Nm-\ref{nonhomoge:adjacent-v}, Nm-\ref{nonhomoge:adjacent-e} and Nm-\ref{nonhomoge:edge-magic} hold true.
\item a \emph{weak-edge-difference sub-proper total string-coloring} if the constraints Nm-\ref{nonhomoge:adjacent-v}, Nm-\ref{nonhomoge:adjacent-e} and Nm-\ref{nonhomoge:edge-difference} hold true.
\item a \emph{weak-graceful-difference sub-proper total string-coloring} if the constraints Nm-\ref{nonhomoge:adjacent-v}, Nm-\ref{nonhomoge:adjacent-e} and Nm-\ref{nonhomoge:graceful-difference} hold true.
\item a \emph{weak-felicitous-difference sub-proper total string-coloring} if the constraints Nm-\ref{nonhomoge:adjacent-v}, Nm-\ref{nonhomoge:adjacent-e} and Nm-\ref{nonhomoge:felicitous-difference} hold true.

\item a \emph{mixed-4-magic sub-proper total string-coloring} if the constraints Nm-\ref{nonhomoge:adjacent-v}, Nm-\ref{nonhomoge:adjacent-e}, Nm-\ref{nonhomoge:edge-magic}, Nm-\ref{nonhomoge:edge-difference}, Nm-\ref{nonhomoge:graceful-difference} and Nm-\ref{nonhomoge:felicitous-difference} hold true.\qqed
\end{asparaenum}
\end{defn}

By Theorem \ref{thm:10-k-d-total-coloringss}, we have
\begin{thm}\label{thm:666666}
$^*$ Each tree admits every one of the $W$-constraint sub-proper total string-colorings defined in Definition \ref{defn:non-homoge-sub-proper-eachostring-coloring}.
\end{thm}

\begin{defn} \label{defn:non-homoge-traditional-string-colorings}
$^*$ \textbf{Non-homogeneous traditional total string-colorings.} Let $\textbf{\textrm{S}}_{tring}(\leq n)$ be the set of \emph{$m$-rank number-based strings} $\alpha_{1}\alpha_{2}\cdots \alpha_{m}$ with each number $\alpha_{j}\ge 0$ for $j\in [1,m]$ and $m\leq n$. A $(p,q)$-graph $H$ admits a total string-coloring $\theta :V(H)\cup E(H)\rightarrow \textbf{\textrm{S}}_{tring}(\leq n)$, such that
\begin{equation}\label{eqa:555555}
\theta (x_k)=a_{k,1}a_{k,2}\cdots a_{k,m(a,k)},~\theta (y_k)=b_{k,1}b_{k,2}\cdots b_{k,m(b,k)},~\theta (x_ky_k)=c_{k,1}c_{k,2}\cdots c_{k,m(c,k)}
\end{equation}for each edge $x_ky_k\in E(H)=\{x_ky_k:k\in [1,q]\}$ with $1\leq m(a,k), m(b,k), m(c,k)\leq n$. For each $k\in [1,q]$, there are the following constraints:
\begin{asparaenum}[\textbf{\textrm{Tra}}-1. ]
\item \label{nonhomoge:tradi-adjacent-v} $\theta (x_k)\neq \theta (y_k)$ for each edge $x_ky_k\in E(H)$.
\item \label{nonhomoge:tradi-adjacent-e} $\theta (x_ky_k)\neq \theta (x_ku_k)$ for each pair of two adjacent edges $x_ky_k,x_ku_k\in E(H)$.
\item \label{nonhomoge:tradi-incident-v-e} $\theta (x_k)\neq \theta (x_ky_k)$ and $\theta (y_k)\neq \theta (x_ky_k)$ for each edge $x_ky_k\in E(H)$.

\item \label{nonhomoge:tradi-e-graceful} Each number $c_{k,r}$ with $r\in [1,m(c,k)]$ corresponds to $a_{k,r\,'}$ and $b_{k,r\,''}$ with $r\,'\in [1,m(a,k)]$ and $r\,''\in [1,m(b,k)]$ holding the \emph{graceful constraint} $c_{k,r}=|a_{k,r\,'}-b_{k,r\,''}|$ true.
\item \label{nonhomoge:tradi-each-odd} Each number $c_{k,s}$ with $r\in [1,m(c,k)]$ is odd and corresponds to $a_{k,s\,'}$ and $b_{k,s\,''}$ with $s\,'\in [1,m(a,k)]$ and $s\,''\in [1,m(b,k)]$ holding the \emph{odd-graceful constraint} $c_{k,s}=|a_{k,s\,'}-b_{k,s\,''}|$ true.
\item \label{nonhomoge:tradi-harmonious} Each number $c_{k,j}$ with $j\in [1,m(c,k)]$ corresponds to $a_{k,j\,'}$ and $b_{k,j\,''}$ for $j\,'\in [1,m(a,k)]$ and $j\,''\in [1,m(b,k)]$ holding the \emph{harmonious constraint} $c_{k,j}=a_{k,j\,'}+b_{k,j\,''}~(\bmod~q)$ true.
\item \label{nonhomoge:tradi-graceful-q-edges} $|f(E(H))|=q$.
\item \label{nonhomoge:tradi-common-factor} Each number $c_{k,i}$ with $i\in [1,m(c,k)]$ corresponds to $a_{k,i\,'}$ and $b_{k,i\,''}$ with $i\,'\in [1,m(a,k)]$ and $i\,''\in [1,m(b,k)]$ holding the \emph{common-factor constraint} $c_{k,i}=\textrm{gcd}(a_{k,i\,'},b_{k,i\,''})$ true.
\end{asparaenum}
\textbf{Then $\theta$ is}

\begin{asparaenum}[\textbf{\textrm{Traco}}-1. ]
\item a \emph{proper total string-coloring} if the constraints Tra-\ref{nonhomoge:tradi-adjacent-v}, Tra-\ref{nonhomoge:tradi-adjacent-e} and Tra-\ref{nonhomoge:tradi-incident-v-e} hold true.
\item a \emph{sub-proper total string-coloring} if the constraints Tra-\ref{nonhomoge:tradi-adjacent-v} and Tra-\ref{nonhomoge:tradi-adjacent-e} hold true.
\item a \emph{graceful proper total string-coloring} if the constraints Tra-\ref{nonhomoge:tradi-adjacent-v}, Tra-\ref{nonhomoge:tradi-adjacent-e}, Tra-\ref{nonhomoge:tradi-incident-v-e}, Tra-\ref{nonhomoge:tradi-e-graceful} and Tra-\ref{nonhomoge:tradi-graceful-q-edges} hold true.
\item a \emph{graceful sub-proper total string-coloring} if the constraints Tra-\ref{nonhomoge:tradi-adjacent-v}, Tra-\ref{nonhomoge:tradi-adjacent-e}, Tra-\ref{nonhomoge:tradi-e-graceful} and Tra-\ref{nonhomoge:tradi-graceful-q-edges} hold true.

\item an \emph{odd-graceful proper total string-coloring} if the constraints Tra-\ref{nonhomoge:tradi-adjacent-v}, Tra-\ref{nonhomoge:tradi-adjacent-e}, Tra-\ref{nonhomoge:tradi-incident-v-e}, Tra-\ref{nonhomoge:tradi-each-odd} and Tra-\ref{nonhomoge:tradi-graceful-q-edges} hold true.
\item an \emph{odd-graceful sub-proper total string-coloring} if the constraints Tra-\ref{nonhomoge:tradi-adjacent-v}, Tra-\ref{nonhomoge:tradi-adjacent-e}, Tra-\ref{nonhomoge:tradi-each-odd} and Tra-\ref{nonhomoge:tradi-graceful-q-edges} hold true.

\item a \emph{harmonious proper total string-coloring} if the constraints Tra-\ref{nonhomoge:tradi-adjacent-v}, Tra-\ref{nonhomoge:tradi-adjacent-e}, Tra-\ref{nonhomoge:tradi-incident-v-e}, Tra-\ref{nonhomoge:tradi-harmonious} and Tra-\ref{nonhomoge:tradi-graceful-q-edges} hold true.
\item a \emph{harmonious sub-proper total string-coloring} if the constraints Tra-\ref{nonhomoge:tradi-adjacent-v}, Tra-\ref{nonhomoge:tradi-adjacent-e}, Tra-\ref{nonhomoge:tradi-harmonious} and Tra-\ref{nonhomoge:tradi-graceful-q-edges} hold true.
\item a \emph{common-factor proper total string-coloring} if the constraints Tra-\ref{nonhomoge:tradi-adjacent-v}, Tra-\ref{nonhomoge:tradi-adjacent-e}, Tra-\ref{nonhomoge:tradi-incident-v-e} and Tra-\ref{nonhomoge:tradi-common-factor} hold true.
\item a \emph{common-factor sub-proper total string-coloring} if the constraints Tra-\ref{nonhomoge:tradi-adjacent-v}, Tra-\ref{nonhomoge:tradi-adjacent-e} and Tra-\ref{nonhomoge:tradi-common-factor} hold true.
\end{asparaenum}
\end{defn}

\begin{thm}\label{thm:connected-graph-vs-trees-colorings}
Each connected graph admits some of the various string-type colorings defined in Definition \ref{defn:homoge-uniformly-string-total-colorings}, Definition \ref{defn:homoge-various-string-total-colorings} and Definition \ref{defn:homoge-4-magic-string-colorings}, Definition \ref{defn:non-homoge-string-colorings}, Definition \ref{defn:non-homoge-sub-proper-eachostring-coloring} and Definition \ref{defn:non-homoge-traditional-string-colorings}.
\end{thm}
\begin{proof} We use the \emph{vertex-splitting operation} to a connected $(p,q)$-graph $G$ for producing a tree $T$ of $q+1$ vertices, clearly, we have a graph homomorphism $T\rightarrow G$. Suppose that $T$ admits a $W$-constraint proper total coloring $f$, so the connected $(p,q)$-graph $G$ admits a $W$-constraint set-coloring $F$ induced by $f$ under the graph homomorphism $T\rightarrow G$ as:

(i) $F(w)=\{f(w)\}$ for $w\in V^*\subset V(G)$, where each vertex $w$ does not participate the vertex-splitting operation.

(ii) $F(x)=\{f(x_1),f(x_2),\dots ,f(x_{a_x})\}$ for $x_1,x_2,\dots ,x_{a_x}\in V(T)\setminus V^*$, where $x=x_1\odot x_2\odot\cdots \odot x_{a_x}$ with $x\in V(G)\setminus V^*$, that is the vertices $x_1,x_2,\dots ,x_{a_x}$ of $T$ are the result of vertex-splitting the vertex $x$ of $G$.

(iii) $F(uv)=\{f(uv)\}$ for $uv\in E(T)=E(G)$.

By means of the above set-coloring $F$, we define a $W$-constraint (sub-)proper total string-coloring $\theta$ for the $(p,q)$-graph $G$ in the following way:

(a-1) $\theta (w)=\gamma$ with $\gamma \in F(w)$ for $w\in V(G)\cup E(G)$ if $|F(w)|=1$;

(a-2) If $F(w)=\{f(x_1),f(x_2),\dots ,f(x_{a_x})\}$ with $a_x\geq 2$ for $w\in V(G)\cup E(G)$, then color $w$ with $\theta (w)=\gamma_{i_1}\gamma_{i_2}\cdots \gamma_{i_{a_x}}$, where $\gamma_{i_1},\gamma_{i_2},\dots ,\gamma_{i_{a_x}}$ is a permutation of the elements of the set $F(w)$.

The proof of the theorem is complete.
\end{proof}

\begin{rem}\label{rem:333333}
The computational complexity of Theorem \ref{thm:connected-graph-vs-trees-colorings} is based on the following analysis:

(i) Using the \emph{vertex-splitting operation} to a connected $(p,q)$-graph $G$ will produce many trees of $q+1$ vertices, called \emph{vsplit-trees}, so we put them into a set $V^{tree}_{split}(G)$. Thereby, each tree $T\in V^{tree}_{split}(G)$ is graph homomorphism to $G$, i.e. $T\rightarrow G$. Determining the cardinality of the vsplit-tree set $V^{tree}_{split}(G)$ is not easy, since the \textbf{Subgraph Isomorphic Problem}, a NP-complete problem.

(ii) Determining whether a tree set admits a $W$-constraint labeling/coloring is difficult, refer to Graceful Tree Conjecture and Strongly Graceful Tree Conjecture (Ref. \cite{Bondy-2008} and \cite{Gallian2021}).

(iii) There are two or more $W$-constraint (sub-)proper total string-colorings $\theta$ defined in the proof of Theorem \ref{thm:connected-graph-vs-trees-colorings}, since there are $(a_x)!$ permutations for $$F(w)=\{f(x_1),f(x_2),\dots ,f(x_{a_x})\},~w\in V(G)\cup E(G)
$$ with $a_x\geq 2$ based on a tree $T\in V^{tree}_{split}(G)$.\paralled
\end{rem}

\subsection{Constructing number-based string-colorings}

\begin{thm} \label{thm:homogeneous-total-coloring-to-string-coloring}
$^*$ Let $\textbf{\textrm{S}}_{tring}(\leq n)$ be the set of \emph{$m$-rank number-based strings} $\alpha_{1}\alpha_{2}\cdots \alpha_{m}$ with each number $\alpha_{j}\ge 0$ for $j\in [1,m]$ and $m\leq n$. Suppose that a connected $(p,q)$-graph $G$ admits $n$ different proper total colorings $f_1,f_2$, $\dots $, $f_n$ such that each total coloring $f_i$ obeys a $W_i$-constraint $W_i\langle f_i(u),f_i(uv),f_i(v)\rangle =0$ for each edge $uv\in E(G)$ and $i\in [1,n]$, then this graph $G$ admits a $\{W_i\}^n_{i=1}$-constraint \emph{homogeneous proper total string-coloring} $F:V(G)\cup E(G)\rightarrow \textbf{\textrm{S}}_{tring}(\leq n)$ defined by
\begin{equation}\label{eqa:555555}
F(u)=f_1(u)f_2(u)\cdots f_n(u),~F(uv)=f_1(uv)f_2(uv)\cdots f_n(uv),~F(v)=f_1(v)f_2(v)\cdots f_n(v)
\end{equation} holding

(i) $F(u)\neq F(v)$, $F(u)\neq F(uv)$ and $F(uv)\neq F(v)$ for each edge $uv\in E(H)$;

(ii) $F(uv)\neq F(uw)$ for any pair of adjacent edges $uv,uw\in E(H)$;

(iii) the $W_i$-constraint $W_i\langle f_i(u),f_i(uv),f_i(v)\rangle =0$ for each edge $uv\in E(G)$ and $i\in [1,n]$.\qqed
\end{thm}

\begin{rem}\label{rem:333333}
In Theorem \ref{thm:homogeneous-total-coloring-to-string-coloring}, we have the Topcode-matrix
$$T_{code}(G,F)=T_{code}(G,f_1f_2\cdots f_n)
$$ and the graph $G$ admits $n!$ different homogeneous proper total string-colorings, since there are $n!$ permutations of the colorings $f_1,f_2,\dots ,f_n$.\paralled
\end{rem}

\begin{example}\label{exa:8888888888}
A complete graph $K_q$ admits a proper vertex coloring $f:V(K_q)\rightarrow [1,q]$, such that the vertex color set $f(V(K_q))=[1,q]$. Then there are $q^{q-2}$ different colored spanning trees of $K_q$ by the famous Cayley's formula $\tau(K_q)=q^{q-2}$ (Ref. \cite{Bondy-2008}). We use a set $S_{pan}(K_q)$ to collect all colored spanning trees of $K_q$, and each colored spanning tree $T\in S_{pan}(K_q)$ admits a proper vertex coloring $f_T=f$. Next, we define a total coloring $F^1_T$ for the colored spanning tree $T\in S_{pan}(K_q)$ as:

(i) $F^1_T(x)=f_T(x)$ for each vertex $x\in V(T)=V(K_q)$; and

(ii) each edge $xy\in E(T)$ is colored by $F^1_T(xy)=|F^1_T(x)-F^1_T(y)|$, or $F^1_T(xy)=F^1_T(x)(\bullet)F^1_T(y)$, where ``$(\bullet)$'' is an operation.

Notice that this spanning tree $T$ admits $n-1$ total colorings $F^2_T, F^3_T,\dots ,F^{n}_T$ with $n\geq 11$ according to Theorem \ref{thm:10-k-d-total-coloringss}, where each coloring $F^i_T$ is a $W_i$-constraint total coloring with $2\leq i\leq n$, and some coloring $F^j_T$ with $j\in [1, n]$ is a $W_j$-constraint proper total coloring. We define a total string-coloring $F\,^*_T$ as:

(iii) $F\,^*_T(x)=F^1_T(x)F^2_T(x)F^3_T(x)\cdots F^{n}_T(x)$ for each vertex $x\in V(T)=V(K_q)$; and

(iv) $F\,^*_T(xy)=F^1_T(xy)F^2_T(xy)F^3_T(xy)\cdots F^{n}_T(xy)$ for each edge $xy\in E(T)$.

Thereby, each $F\,^*_T(xy)$ holds $F^i_T(xy)=W_i\langle F^i_T(x),F^i_T(y)\rangle$ for each edge $xy\in E(T)$, then $F\,^*_T$ is really a proper total homogeneous string-coloring of the spanning tree $T$.

The set $S^{n\textrm{-color}}_{pan}(K_q)$ contains each spanning tree $T\in S_{pan}(K_q)$ admitting a proper total homogeneous string-coloring $F\,^*_T$, such that the cardinality $|S^{n\textrm{-color}}_{pan}(K_q)|=q^{q-2}$.\qqed
\end{example}

\begin{thm}\label{thm:infinite-elements-n-segment}
Let $\textbf{\textrm{S}}_{tring}(n)$ be the set of \emph{$n$-rank number-based strings} $\alpha_{1}\alpha_{2}\cdots \alpha_{n}$ with each number $\alpha_{j}\ge 0$ for $j\in [1,n]$. Then the set $\textbf{\textrm{S}}_{tring}(n)$ contains infinite elements.
\end{thm}
\begin{proof}Notice that
\begin{equation}\label{eqa:infinite elementss}
f(u_k)=A_k=a_{k,1}a_{k,2}\cdots a_{k,n},~f(v_k)=B_k=b_{k,1}b_{k,2}\cdots b_{k,n},~f(u_kv_k)=C_k=c_{k,1}c_{k,2}\cdots c_{k,n}
\end{equation} for each edge $u_kv_k\in E(G)=\{u_kv_k:k\in [1,q]\}$ defined in Definition \ref{defn:homoge-uniformly-string-total-colorings}.

(i) Set $c^*_{k,j}=c_{k,j}$, $a^*_{k,j}=a_{k,j}+\beta$ and $b^*_{k,j}=b_{k,j}+\beta$ with $\beta\in Z^0$, then each of the following statements

\begin{asparaenum}[\textbf{\textrm{State}}-1. ]
\item Each $j\in [1,n]$ holds the graceful constraint $c_{k,j}=|a_{k,j}-b_{k,j}|$ true.

\item Each $c_{k,j}$ is odd and holds the odd-graceful constraint $c_{k,j}=|a_{k,j}-b_{k,j}|$ true for $j\in [1,n]$.

\item Each $j\in [1,n]$ holds the edge-magic constraint $a_{k,j}+b_{k,j}+c_{k,j}=\lambda$ true.

\item Each $j\in [1,n]$ holds the edge-difference constraint $c_{k,j}+|a_{k,j}-b_{k,j}|=\lambda$ true.

\item Each $j\in [1,n]$ holds the graceful-difference constraint $\big ||a_{k,j}-b_{k,j}|-c_{k,j}\big |=\lambda$ true.
\end{asparaenum}
keeps no change. We get the algebraic operation on $n$-rank number-based strings as follows:
\begin{equation}\label{eqa:555555}
{
\begin{split}
C^*_k&=C_k\\
A^*_k&=A_k[+]\beta =(a_{k,1}+\beta)(a_{k,2}+\beta)\cdots (a_{k,n}+\beta)=\beta +A_k\\
B^*_k&=B_k[+]\beta =(b_{k,1}+\beta)(b_{k,2}+\beta)\cdots (b_{k,n}+\beta)=\beta +B_k\\
\big |A^*_k[-]B^*_k|&=\big |(a_{k,1}+\beta)-(b_{k,1}+\beta)| \cdot |(a_{k,2}+\beta)-(b_{k,2}+\beta)|\cdot \cdots \cdot |(a_{k,n}+\beta)-(b_{k,n}+\beta)|\\
&=|a_{k,1}-b_{k,1}| \cdot |a_{k,2}-b_{k,2}|\cdot \cdots \cdot |a_{k,n}-b_{k,n}|\\
A^*_k[+]B^*_k&=[(a_{k,1}+\beta)+(b_{k,1}+\beta)] \cdot [(a_{k,2}+\beta)+(b_{k,2}+\beta)]\cdot \cdots \cdot [(a_{k,n}+\beta)+(b_{k,n}+\beta)]\\
&=[a_{k,1}+b_{k,1}+2\beta]\cdot [a_{k,2}+b_{k,2}+2\beta]\cdot \cdots \cdot [a_{k,n}+b_{k,n}+2\beta]\\
&=A_k[+]B_k+2\beta
\end{split}}
\end{equation} Moreover, we write the above facts by
$$
C^*_k=\big |A^*_k[-]B^*_k\big |,~A^*_k[+]B^*_k[+]C^*_k=\lambda,~\big ||A^*_k[-]B^*_k|[-]C^*_k\big |=\lambda
$$ respectively.

(ii) Set $c^*_{k,j}=c_{k,j}+2\beta$, $a^*_{k,j}=a_{k,j}+\beta$ and $b^*_{k,j}=b_{k,j}+\beta$ with $\beta\in Z^0$, then we can keep the statement ``Each $j\in [1,n]$ holds the felicitous-difference constraint $|a_{k,j}+b_{k,j}-c_{k,j}|=\lambda$ true'' no change. So, we have $\big |A^*_k[+]B^*_k[-]C^*_k\big |=\lambda$.

(iii) Set $c^*_{k,j}=c_{k,j}$, $a^*_{k,j}=a_{k,j}+\beta q$ and $b^*_{k,j}=b_{k,j}+\beta q$ with $\beta\in Z^0$, then ``Each $j\in [1,n]$ holds the harmonious constraint $c^*_{k,j}=a^*_{k,j}+b^*_{k,j}~(\bmod~q)$ true''. Then $C^*_k=A^*_k[+]B^*_k~(\bmod~q)$.

(iv) For $a_{k,j}$ and $b_{k,j}$ are prime numbers, and $c_{k,j}=a_{k,j}+b_{k,j}$, we get $c^*_{k,j}=c_{k,j}+4\beta$, two prime numbers $a^*_{k,j}=a_{k,j}+2\beta$ and $b^*_{k,j}=b_{k,j}+2\beta$ with $\beta\in Z^0$, such that $c^*_{k,j}=a^*_{k,j}+b^*_{k,j}$.

(v) For the common-factor constraint $c_{k,j}=\textrm{gcd}(a_{k,j},b_{k,j})$, we get $c^*_{k,j}=\beta\cdot c_{k,j}$, two prime numbers $a^*_{k,j}=\beta\cdot a_{k,j}$ and $b^*_{k,j}=\beta\cdot b_{k,j}$ with $\beta\in Z^0$, then the common-factor constraint $c^*_{k,j}=\textrm{gcd}(a^*_{k,j},b^*_{k,j})$ holds true.
\begin{equation}\label{eqa:555555}
{
\begin{split}
A^*_k&=\beta \cdot A_k=(\beta \cdot a_{k,1})(\beta \cdot a_{k,2})\cdots (\beta \cdot a_{k,n})\\
B^*_k&=\beta \cdot B_k=(\beta \cdot b_{k,1})(\beta \cdot b_{k,2})\cdots (\beta \cdot b_{k,n})\\
C^*_k&=\beta \cdot C_k=(\beta \cdot c_{k,1})(\beta \cdot c_{k,2})\cdots (\beta \cdot c_{k,n})\\
&=\beta \cdot \textrm{gcd}(a_{k,1},b_{k,1})\beta \cdot \textrm{gcd}(a_{k,2},b_{k,2})\cdots \beta \cdot \textrm{gcd}(a_{k,n},b_{k,n})\\
&=\beta ^n\textrm{gcd}(a_{k,1},b_{k,1})\textrm{gcd}(a_{k,2},b_{k,2})\cdots \textrm{gcd}(a_{k,n},b_{k,n})
\end{split}}
\end{equation} and $C^*_k=\textrm{gcd}(A^*_k,B^*_k)$, as well as
\begin{equation}\label{eqa:555555}
\beta \cdot C_k=C^*_k=\textrm{gcd}(A^*_k,B^*_k)=\textrm{gcd}(\beta \cdot A_k,\beta \cdot B_k)=\beta \cdot \textrm{gcd}(A_k,B_k)
\end{equation}

(vi) For two prime numbers $a_{k,j}$ and $b_{k,j}$ holding $c_{k,j}=a_{k,j}+b_{k,j}$, we set
$$c^*_{k,j}=c_{k,j}+2\beta a_{k,j}+2\beta b_{k,j}+4\beta ^2,~a^*_{k,j}=a_{k,j}+\beta ,~b^*_{k,j}=b_{k,j}+2\beta
$$ with $\beta\in Z^0$, then $c^*_{k,j}=a^*_{k,j}\cdot b^*_{k,j}$.

Summary of the above facts, we get the following \emph{string-algebraic operations}
\begin{equation}\label{eqa:555555}
{
\begin{split}
\gamma A_k[+]\beta B_k&=(\gamma a_{k,1}+\beta b_{k,1}) \cdot (\gamma a_{k,2}+\beta b_{k,2})\cdot \cdots \cdot (\gamma a_{k,n}+\beta b_{k,n})\\
\big |\gamma A_k[-]\beta B_k\big |&=|\gamma a_{k,1}-\beta b_{k,1}| \cdot |\gamma a_{k,2}-\beta b_{k,2}|\cdot \cdots \cdot |\gamma a_{k,n}-\beta b_{k,n}|
\end{split}}
\end{equation} based on $n$-rank number-based strings.
\end{proof}

\begin{thm} \label{thm:homoge-more-string-total-coloringss}
$^*$ Let $\textbf{\textrm{S}}_{tring}(\leq n)$ be the set of \emph{$m$-rank number-based strings} $\alpha_{1}\alpha_{2}\cdots \alpha_{m}$ with each number $\alpha_{j}\ge 0$ for $j\in [1,m]$ and $m\leq n$. A connected $(p,q)$-graph $G$ admits $B$ different string-colorings
$$g_j:V(G)\cup E(G)\rightarrow \textbf{\textrm{S}}_{tring}(\leq n),~j\in [1,B]$$ and each string-coloring $g_j$ is one of the string-colorings defined in Definition \ref{defn:homoge-uniformly-string-total-colorings}, Definition \ref{defn:homoge-various-string-total-colorings} and Definition \ref{defn:homoge-4-magic-string-colorings}, Definition \ref{defn:non-homoge-string-colorings}, Definition \ref{defn:non-homoge-sub-proper-eachostring-coloring} and Definition \ref{defn:non-homoge-traditional-string-colorings}. Then $G$ admits a \emph{compounded homogeneous proper total string-coloring} $F$ holding
\begin{equation}\label{eqa:555555}
F(u)=g_1(u)g_2(u)\cdots g_B(u),F(uv)=g_1(uv)g_2(uv)\cdots g_B(uv),F(v)=g_1(v)g_2(v)\cdots g_B(v)
\end{equation} for each edge $uv\in E(G)$. Moreover, if some $g_j$ is proper, so is $F$.\qqed
\end{thm}

\begin{example}\label{exa:5-labelings-string-set-vector}
In Fig.\ref{fig:construct-string-00}, a graph $G$ admits the following labelings:

(i) The graph $G_1$ admits a \emph{graceful labeling} $g_1$ holding $g_1(xy)=|g_1(x)-g_1(y)|$ for each edge $xy\in E(G_1)$, such that the vertex color set $g_1(V(G_1))\subset [0,11]$ and the edge color set $g_1(E(G_1))=[1,11]$;

(ii) The graph $G_2$ admits an \emph{odd-edge graceful labeling} $g_2$ holding
$$g_2(xy)=2|g_2(x)-g_2(y)|-1,~xy\in E(G_2)
$$ such that the vertex color set $g_2(V(G_2))\subset [0,11]$ and the edge color set $g_2(E(G_2))=[1,21]^o$;

(iii) The graph $G_3$ admits an \emph{edge-difference total labeling} $g_3$ holding the edge-difference constraint
$$g_3(xy)+|g_3(x)-g_3(y)|=12,~xy\in E(G_3)
$$ such that the vertex color set $g_3(V(G_3))\subset [0,11]$ and the edge color set $g_3(E(G_3))=[1,11]$;

(iv) The graph $G_4$ admits a \emph{felicitous-difference total labeling} $g_4$ holding the felicitous-difference constraint
$$|g_4(x)+g_4(y)-g_4(xy)|=5,~xy\in E(G_4)
$$ such that the vertex color set $g_4(V(G_4))\subset [0,11]$ and the edge color set $g_4(E(G_4))=[1,11]$;

(v) The graph $G_5$ admits an \emph{edge-magic graceful total labeling} $g_5$ holding the edge-magic constraint
$$g_5(x)+g_5(y)+g_5(xy)=17,~xy\in E(G_5)
$$ such that the vertex color set $g_5(V(G_5))\subset [0,11]$ and the edge color set $g_5(E(G_5))=[1,11]$.
\qqed
\end{example}

\begin{figure}[h]
\centering
\includegraphics[width=16.4cm]{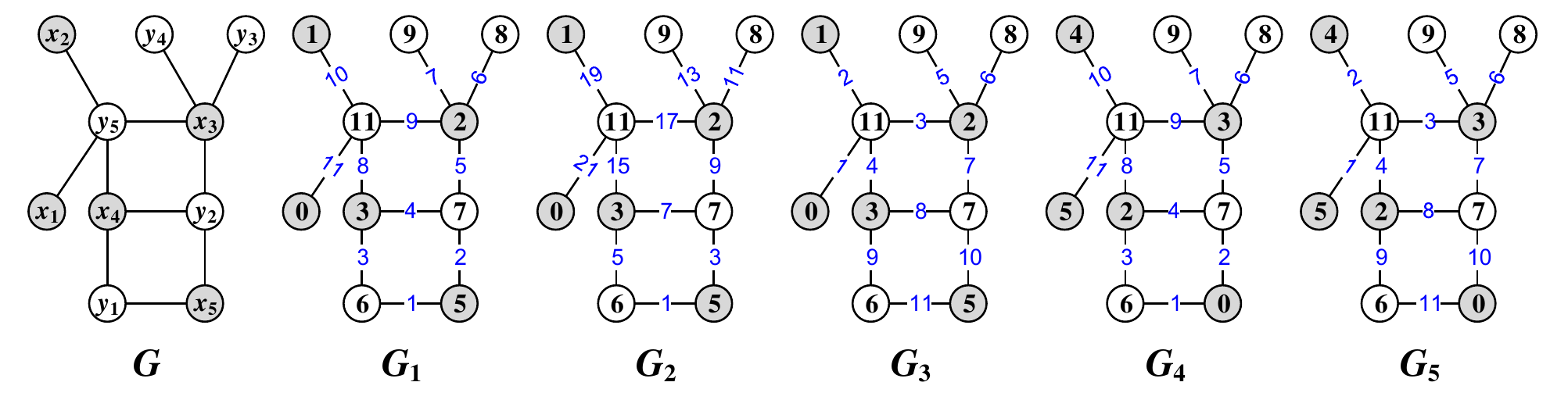}\\
\caption{\label{fig:construct-string-00} {\small A connected $(10,11)$-graph $G$ admits five colorings.}}
\end{figure}

\begin{example}\label{exa:8888888888}
In Fig.\ref{fig:333string-vector-set-coloring}, the connected $(10,11)$-graph $G$ shown in Fig.\ref{fig:construct-string-00} admits a $\{W_i\}^n_{i=1}$-constraint homogeneous proper total string-coloring $F_{stri}$ defined in Theorem \ref{thm:homogeneous-total-coloring-to-string-coloring}, a total vector-coloring $F_{vector}$ and a total set-coloring $F_{set}$ defined in Definition \ref{defn:n-di-vector-colorings-definition}. Three colorings $F$, $F_{vector}$ and $F_{set}$ are induced by the colorings $g_1$, $g_2$, $g_3$, $g_4$ and $g_5$ introduced in Example \ref{exa:5-labelings-string-set-vector}.\qqed
\end{example}

\begin{figure}[h]
\centering
\includegraphics[width=16.4cm]{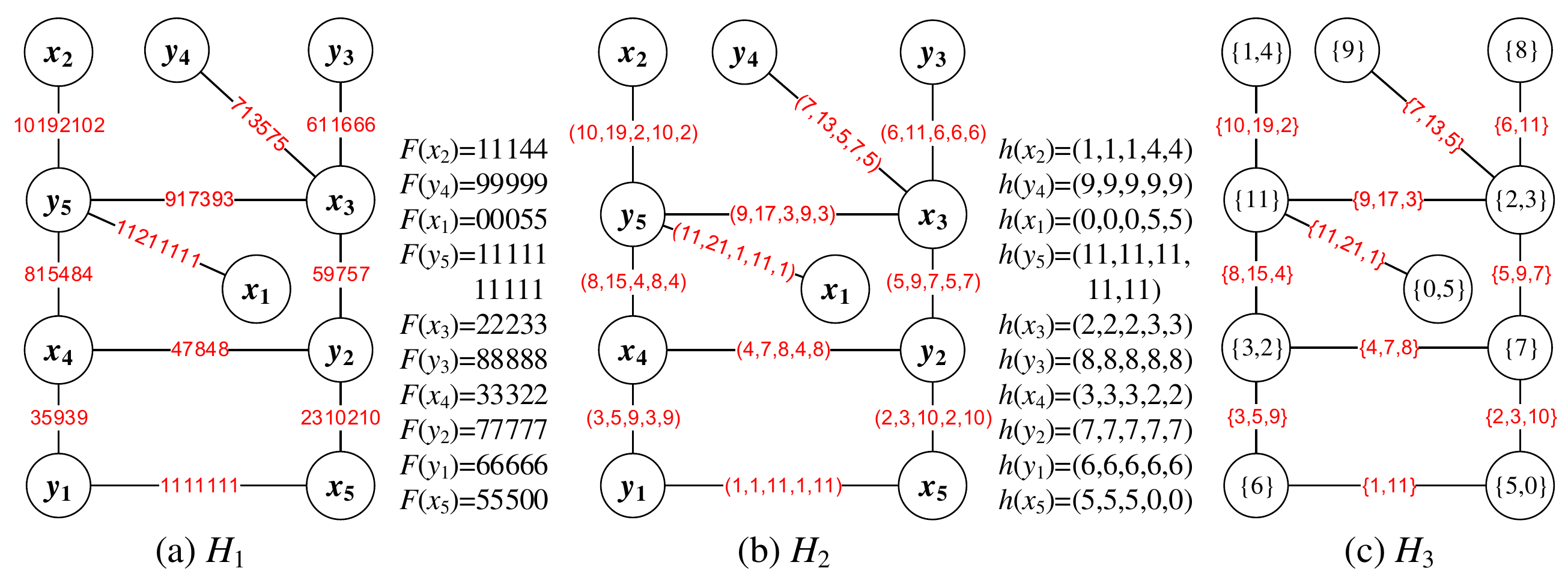}\\
\caption{\label{fig:333string-vector-set-coloring} {\small (a) A connected $(10,11)$-graph $H_1$ admits a string-coloring $F_{stri}$; (b) a connected $(10,11)$-graph $H_2$ admits a vector-coloring $F_{vector}$; (c) a connected $(10,11)$-graph $H_3$ admits a set-coloring $F_{set}$.}}
\end{figure}

\begin{defn} \label{defn:more-k-d-string-total-coloring}
\cite{Bing-Yao-arXiv:2207-03381} \textbf{Parameterized string-colorings and set-colorings.} By Definition \ref{defn:kd-w-type-colorings} and Definition \ref{defn:odd-edge-W-type-total-labelings-definition}, let $G$ be a bipartite $(p,q)$-graph, and its vertex set $V(G)=X\cup Y$ with $X\cap Y=\emptyset$ such that each edge $uv\in E(G)$ holds $u\in X$ and $v\in Y$. There is a proper total coloring
$$
f_s:X\rightarrow S_{m,0,0,d}=\{0,d,\dots ,md\},~f_s:Y\cup E(G)\rightarrow S_{n,k,0,d}=\{k,k+d,\dots ,k+nd\}
$$ here it is allowed $f_s(u)=f_s(w)$ for some distinct vertices $u,w\in V(G)$) for $s\in [1,B]$ with integer $B\geq 2$, such that $f_s\in \{$gracefully $(k_s,d_s)$-total coloring, odd-gracefully $(k_s,d_s)$-total coloring, edge anti-magic $(k_s,d_s)$-total coloring, harmonious $(k_s,d_s)$-total coloring, odd-elegant $(k_s,d_s)$-total coloring, edge-magic $(k_s,d_s)$-total coloring, edge-difference $(k_s,d_s)$-total coloring, felicitous-difference $(k_s,d_s)$-total coloring, graceful-difference $(k_s,d_s)$-total coloring, odd-edge edge-magic $(k_s,d_s)$-total coloring, odd-edge edge-difference $(k_s,d_s)$-total coloring, odd-edge felicitous-difference $(k_s,d_s)$-total coloring, odd-edge graceful-difference $(k_s,d_s)$-total coloring$\}$ with $s\in [1,B]$. Then

(i) The bipartite $(p,q)$-graph $G$ admits a \emph{parameterized compounded homogeneous proper total string-coloring} $F$ holding
$$
F(u)=f_{i_1}(u)f_{i_2}(u)\cdots f_{i_B}(u),F(uv)=f_{j_1}(uv)f_{j_2}(uv)\cdots f_{j_B}(uv),F(v)=f_{s_1}(v)f_{s_2}(v)\cdots f_{s_B}(v)
$$true for each edge $uv\in E(G)$, where the number-based string $f_{i_1}(u)f_{i_2}(u)\cdots f_{i_B}(u)$ is a permutation of $f_1(u),f_2(u),\cdots $, $f_B(u)$, and the number-based string $f_{j_1}(uv)f_{j_2}(uv)\cdots f_{j_B}(uv)$ is a permutation of $f_1(uv),f_2(uv),\cdots ,f_B(uv)$, as well as the number-based string $f_{s_1}(v)f_{s_2}(v)$ $\cdots f_{s_B}(v)$ is a permutation of $f_1(v),f_2(v),\cdots ,f_B(v)$. Moreover, if some $f_{i_j}$ is proper, then $F$ is proper too.

(ii) The bipartite $(p,q)$-graph $G$ admits a \emph{parameterized proper total set-coloring} $\theta$ holding
$${
\begin{split}
\theta(u)&=\{f_1(u),f_2(u),\cdots ,f_B(u)\},~\theta(uv)=\{f_1(uv),f_2(uv),\cdots ,f_B(uv)\},\\
\theta(v)&=\{f_1(v),f_2(v),\cdots ,f_B(v)\}
\end{split}}
$$ true for each edge $uv\in E(G)$.\qqed
\end{defn}

\begin{rem}\label{rem:333333}
In Definition \ref{defn:more-k-d-string-total-coloring}, each sequence $\{(k_s,d_s)\}^B_{s=1}$ will induce random number-based strings generated from the Topcode-matrix of the bipartite $(p,q)$-graph $G$ admitting a parameterized compounded homogeneous total string-coloring; refer to Definition \ref{defn:basic-W-type-labelings}. Similarly, a compounded total string-coloring $F$ of the bipartite $(p,q)$-graph $G$ shown in Theorem \ref{thm:homoge-more-string-total-coloringss} is more complex than each one of the hostring-colorings defined in Definition \ref{defn:homoge-uniformly-string-total-colorings}, Definition \ref{defn:homoge-various-string-total-colorings} and Definition \ref{defn:homoge-4-magic-string-colorings}.\paralled
\end{rem}

By Lemma \ref{lem:set-ordered-w-cond-set-co}, Theorem \ref{thm:10-k-d-total-coloringss} and Theorem \ref{thm:connected-graph-vs-trees-colorings}, we have

\begin{thm}\label{thm:99999}
$^*$ Each connected graphs admits some of the various total string-colorings and total set-colorings defined from Definition \ref{defn:homoge-uniformly-string-total-colorings} to Definition \ref{defn:more-k-d-string-total-coloring}.
\end{thm}

\begin{thm} \label{thm:no-order-homogeneous-proper-total-string-coloringss}
$^*$ Let $\textbf{\textrm{S}}_{tring}(\leq n)$ be the set of \emph{$m$-rank number-based strings} $\alpha_{1}\alpha_{2}\cdots \alpha_{m}$ with each number $\alpha_{j}\ge 0$ for $j\in [1,m]$ and $m\leq n$. Suppose that a connected $(p,q)$-graph $H$ admits $n$ proper total colorings $h_1,h_2$, $\dots $, $h_n$ such that each total coloring $h_i$ obeys a $W_i$-constraint $h_i(uv)=W_i\langle h_i(u),h_i(v)\rangle $ for each edge $uv\in E(H)$ and $i\in [1,n]$, then this graph $G$ admits a $\{W_i\}^n_{i=1}$-constraint \emph{no-order homogeneous proper total string-coloring} $\varphi:V(H)\cup E(H)\rightarrow \textbf{\textrm{S}}_{tring}(\leq n)$ defined by
\begin{equation}\label{eqa:no-order-homogeneous-string-colors}
{
\begin{split}
\varphi(u)&=h_{i_1}(u)h_{i_2}(u)\cdots h_{i_n}(u),~\varphi(uv)=h_{j_1}(uv)h_{j_2}(uv)\cdots h_{j_n}(uv),\\
\varphi(v)&=h_{k_1}(v)h_{k_2}(v)\cdots h_{k_n}(v)
\end{split}}
\end{equation} where the \emph{vertex string} $h_{i_1}(u)h_{i_2}(u)\cdots h_{i_n}(u)$ is a permutation of vertex colors $h_1(u),h_2(u)$, $\dots $, $h_n(u)$, and the \emph{edge string} $h_{j_1}(uv)h_{j_2}(uv)\cdots h_{j_n}(uv)$ is a permutation of edge colors $h_1(uv),h_2(uv)$, $\dots $, $h_n(uv)$, and the \emph{vertex string} $h_{k_1}(v)h_{k_2}(v)\cdots h_{k_n}(v)$ is a permutation of vertex colors $h_1(v)$, $h_2(v)$, $\dots $, $h_n(v)$; such that

(i) $\varphi(u)\neq \varphi(v)$, $\varphi(u)\neq \varphi(uv)$ and $\varphi(uv)\neq \varphi(v)$ for each edge $uv\in E(H)$;

(ii) $\varphi(uv)\neq \varphi(uw)$ for any pair of adjacent edges $uv,uw\in E(H)$;

(iii) each edge number $h_{j_c}(uv)$ corresponds two vertex numbers $h_{i_a}(u)$ and $h_{k_b}(v)$ holding the $W_{j_c}$-constraint $h_{j_c}(uv)=W_{j_c}\langle h_{i_a}(u),h_{k_b}(v)\rangle $ for each edge $uv\in E(H)$;

(iv) each vertex number $h_{i_a}(u)$ corresponds two numbers $h_{j_c}(uv)$ and $h_{k_b}(v)$ holding the $W_{i_a}$-constraint $h_{j_c}(uv)=W_{i_a}\langle h_{i_a}(u),h_{k_b}(v)\rangle $ for each vertex $u\in V(H)$.

(v) each vertex number $h_{k_b}(v)$ corresponds two numbers $h_{j_c}(uv)$ and $h_{i_a}(u)$ holding the $W_{k_b}$-constraint $h_{j_c}(uv)=W_{k_b}\langle h_{i_a}(u),h_{k_b}(v)\rangle $ for each vertex $v\in V(H)$.
\end{thm}

\begin{thm} \label{thm:non-homogeneous-proper-total-string-colorings}
$^*$ Let $\textbf{\textrm{S}}_{tring}(\leq n)$ be the set of \emph{$m$-rank number-based strings} $\alpha_{1}\alpha_{2}\cdots \alpha_{m}$ with each number $\alpha_{j}\ge 0$ for $j\in [1,m]$ and $m\leq n$. Suppose that a connected $(p,q)$-graph $T$ admits $n$ proper total colorings $g_1,g_2$, $\dots $, $g_n$ such that each total coloring $g_i$ obeys a $W_i$-constraint $g_i(uv)=W_i\langle g_i(u),g_i(v)\rangle $ for each edge $uv\in E(T)$ and $i\in [1,n]$, then this graph $G$ admits a $\{W_i\}^n_{i=1}$-constraint \emph{non-homogeneous proper total string-coloring} $\theta:V(T)\cup E(T)\rightarrow \textbf{\textrm{S}}_{tring}(\leq n)$ defined by
\begin{equation}\label{eqa:non-homogeneous-string-colorss}
{
\begin{split}
\theta(u)&=g_{i_1}(u)g_{i_2}(u)\cdots g_{i_A}(u),~\theta (uv)=g_{j_1}(uv)g_{j_2}(uv)\cdots g_{j_B}(uv),\\
\theta (v)&=g_{k_1}(v)g_{k_2}(v)\cdots g_{k_C}(v)
\end{split}}
\end{equation} with $1\leq A,B,C\leq n$, where the \emph{vertex string} $g_{i_1}(u)g_{i_2}(u)\cdots g_{i_A}(u)$ is a sub-permutation of vertex colors $g_1(u),g_2(u)$, $\dots $, $g_n(u)$, and the \emph{edge string} $g_{j_1}(uv)g_{j_2}(uv)\cdots g_{j_B}(uv)$ is a sub-permutation of edge colors $g_1(uv),g_2(uv)$, $\dots $, $g_n(uv)$, and the \emph{vertex string} $g_{k_1}(v)g_{k_2}(v)\cdots g_{k_C}(v)$ is a sub-permutation of vertex colors $g_1(v)$, $g_2(v)$, $\dots $, $g_n(v)$; such that

(i) $\theta (u)\neq \theta (v)$, $\theta (u)\neq \theta (uv)$ and $\theta (uv)\neq \theta (v)$ for each edge $uv\in E(T)$;

(ii) $\theta (uv)\neq \theta (uw)$ for any pair of adjacent edges $uv,uw\in E(T)$;

(iii) each edge number $g_{j_s}(uv)$ corresponds two vertex numbers $g_{i_r}(u)$ and $g_{k_t}(v)$ holding the $W_{j_s}$-constraint $g_{j_s}(uv)=W_{j_s}\langle g_{i_r}(u),g_{k_t}(v)\rangle $ for each edge $uv\in E(T)$;

(iv) each vertex number $g_{i_r}(u)$ corresponds two numbers $g_{j_s}(uv)$ and $g_{k_t}(v)$ holding the $W_{i_r}$-constraint $g_{j_s}(uv)=W_{i_r}\langle g_{i_r}(u),g_{k_t}(v)\rangle $ for each vertex $u\in V(T)$.

(v) each vertex number $g_{k_t}(v)$ corresponds two numbers $g_{j_s}(uv)$ and $g_{i_r}(u)$ holding the $W_{k_t}$-constraint $g_{j_s}(uv)=W_{k_t}\langle g_{i_r}(u),g_{k_t}(v)\rangle $ for each vertex $v\in V(T)$.
\end{thm}

\section{Graph-based strings}

\subsection{Total graph-colorings}

\begin{defn} \label{defn:graphic-topcode-matrix}
\cite{Yao-Wang-2106-15254v1} Let $G_{x,i},G_{e,j}$ and $G_{y,k}$ be colored graphs with $1\leq i,j,k\leq q$. A \emph{graphic Topcode-matrix} $G_{code}$ is defined as
\begin{equation}\label{eqa:graphic-topcode-matrix}
\centering
{
\begin{split}
G_{code}= \left(
\begin{array}{ccccc}
G_{x,1} & G_{x,2} & \cdots & G_{x,q}\\
G_{e,1} & G_{e,2} & \cdots & G_{e,q}\\
G_{y,1} & G_{y,2} & \cdots & G_{y,q}
\end{array}
\right)=
\left(\begin{array}{c}
G_X\\
G_E\\
G_Y
\end{array} \right)=(G_X,~G_E,~G_Y)^{T}
\end{split}}
\end{equation} both colored graphs $G_{x,i}$ and $G_{y,i}$ are the \emph{ends} of the colored graph $G_{e,i}$, and we call
$$G_X=(G_{x,1}, G_{x,2},\dots ,G_{x,q}),~G_E=(G_{e,1}, G_{e,2},\dots ,G_{e,q}),~G_Y=(G_{y,1}, G_{y,2},\dots ,G_{y,q})
$$ \emph{graphic vectors}.

Moreover the graphic Topcode-matrix $G_{code}$ is \emph{constraint valued} if there is a group of constraints $\varphi_1,\varphi_2,\dots ,\varphi_m$ such that $G_{e,i}=\varphi_j\langle G_{x,i},G_{y,i}\rangle$ for some $j\in [1,m]$ and each $i\in [1,q]$.\qqed
\end{defn}

\begin{defn} \label{defn:topological-number-based-string}
$^*$ If a number-based string $s(n)$ can be made by some colored graph, then we call it \emph{topological number-based string}, and the colored graph is a \emph{topological expression} of the number-based string $s(n)$.\qqed
\end{defn}

\begin{rem}\label{rem:333333}
In Definition \ref{defn:topological-number-based-string}, a topological number-based string can be expressed by two or more topological expressions.\paralled
\end{rem}

Let $\textbf{\textrm{T}}_i=\big (T_{i,1},T_{i,2},\dots, T_{i,n_i}\big )$ be a \emph{colored graph base} for $n_i\geq 1$ and $i\in [1,M]$, and each colored graph $T_{i,j}$ admits a total coloring (or a total graph-coloring) $f_{i,j}$ for $j\in [1,n_i]$. Suppose that each colored graph base $\textbf{\textrm{T}}_i$ is an \emph{every-zero additive group} defined in Definition \ref{defn:every-zero-abstract-group}.

\begin{defn} \label{defn:edge-index-graphic-group222}
$^*$ A new Abelian additive operation is defined as
\begin{equation}\label{eqa:edge-index-graphic-group222}
T_{i,r}~[\ominus \oplus _k] ~T_{i,j}:=T_{i,k}\ominus T_{i,r}\oplus T_{i,j}=T_{i,\mu}\in \textbf{\textrm{T}}_i
\end{equation} with the index $\mu=k-(r+j)~(\bmod~M)$ for any preappointed \emph{zero} $T_{i,k}\in \textbf{\textrm{T}}_i$. So, we get an \emph{every-zero edge-index graphic group} $\{F^-_M(\textbf{\textrm{T}}_i);\ominus \oplus\}$ of order $M$ defined by the new Abelian operation $T_{i,r}~[\ominus \oplus _k] ~T_{i,j}$ based on a colored graph base $\textbf{\textrm{T}}_i$.\qqed
\end{defn}

\begin{figure}[h]
\centering
\includegraphics[width=16.4cm]{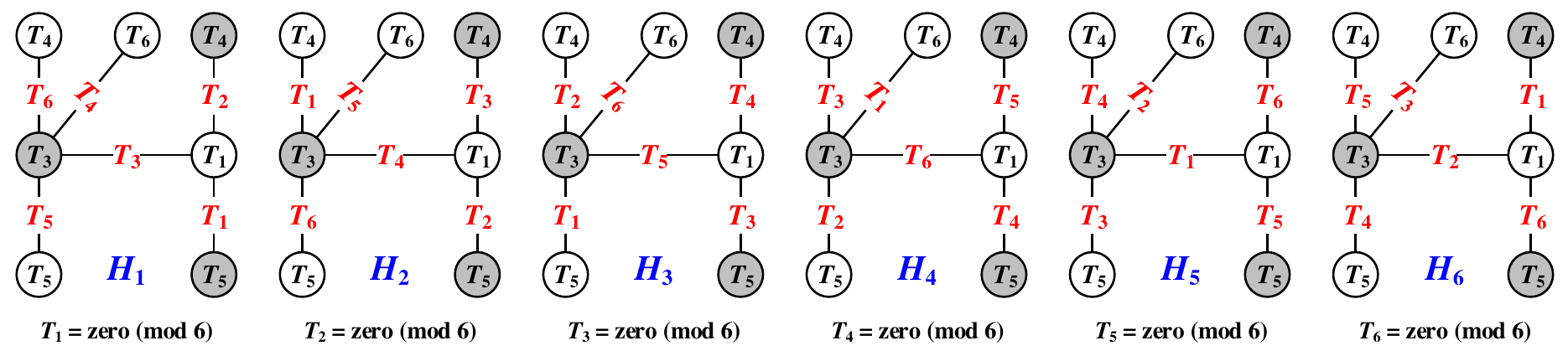}\\
\caption{\label{fig:graphic-group-operation22} {\small An edge-index graphic group $\{F^-_6(G);\ominus \oplus\}$ based on $G\{T_1,T_2,\dots ,T_6\}$ for understanding Definition \ref{defn:edge-index-graphic-group222}.}}
\end{figure}

\begin{defn} \label{defn:edge-index-total-graph-coloring}
$^*$ A $(p,q)$-graph $H$ admits an \emph{edge-index total graph-coloring} $\theta:V(H)\cup E(H)\rightarrow \textbf{\textrm{T}}_i=\big \{T_{i,1},T_{i,2},\dots, T_{i,n_i}\big \}$, such that each edge $uv\in E(H)$ holds $F(u)=T_{i,r}$, $F(v)=T_{i,j}$ and $F(uv)=T_{i,\lambda}$ with the index $\lambda=r+j-k~(\bmod~M)$ under the Abelian additive operation
\begin{equation}\label{eqa:edge-index-graph-coloring}
T_{i,r}~[\oplus \ominus_k] ~T_{i,j}:=T_{i,r}\oplus T_{i,j}\ominus T_{i,k}=T_{i,\lambda}\in \textbf{\textrm{T}}_i
\end{equation} for any preappointed \emph{zero} $T_{i,k}\in \textbf{\textrm{T}}_i$, where $M$ is a positive integer. Let
$$F_{\textrm{e-index}}(E(H))=\{s=r+j-k~(\bmod~M): ~\theta(uv)=T_{i,s},uv\in E(H)\}
$$ be the \emph{edge-index set} for $\theta(u)=T_{i,r}$, $\theta(v)=T_{i,j}$ and $\theta(uv)=T_{i,s}$. Let $\lambda$ be a non-negative integer. \textbf{Then we call $\theta$}:
\begin{asparaenum}[\textbf{Incol}-1.]
\item an \emph{edge-index graceful total graph-coloring} if $F_{\textrm{e-index}}(E(H))=[1,q]$ holds true.
\item an \emph{edge-index odd-graceful total graph-coloring} if $F_{\textrm{e-index}}(E(H))=[1,2q-1]^o$ holds true.
\item an \emph{edge-index edge-difference total graph-coloring} if the edge-difference constraint $s+|r-j|=\lambda$ holds true.
\item an \emph{edge-index felicitous-difference total graph-coloring} if the felicitous-difference constraint $|r+j-s|=\lambda$ holds true.
\item an \emph{edge-index graceful-difference total graph-coloring} if the graceful-difference constraint $\big ||r+j|-s\big |=\lambda$ holds true.
\item an \emph{edge-index edge-magic total graph-coloring} if the edge-magic constraint $s+r+j=\lambda$ holds true.\qqed
\end{asparaenum}
\end{defn}

\begin{example}\label{exa:8888888888}
In Fig.\ref{fig:odd-edge-index-colorings}, $F(uv)=T_{ij}$, $F(u)=T_{i}$ and $F(v)=T_{j}$. So, we have

(a) An odd-graceful edge-index felicitous-difference total graph-coloring holding the felicitous-difference constraint $|i+j-ij|=6~(\bmod ~12)$;

(b) an odd-graceful edge-index graceful-difference total graph-coloring holding the graceful-difference constraint $\big ||i-j|-ij\big |=0~(\bmod ~12)$, and each vertex of (a) is colored the same color with the image vertex of (b);

(c) an odd-graceful edge-index edge-magic total graph-coloring holding the edge-magic constraint $i+j+ij=4~(\bmod ~12)$ under $4-(i+j)=ij~(\bmod ~12)$, and each edge index of (c) (as a \emph{public-key}) matches with edge index of the image edge of (b) (as a \emph{private-key}), such that the sum of two edge indices is equal to 12;

(d) an odd-graceful edge-index felicitous-difference total graph-coloring holding the felicitous-difference constraint $|i+j-ij|=4~(\bmod ~12)$;

(e) an odd-graceful edge-index felicitous-difference total graph-coloring holding the felicitous-difference constraint $|i+j-ij|=4~(\bmod ~12)$, where each vertex index of (b) matches with vertex index of the image edge of (e), such that the sum of two vertex indices is equal to 15;

(f) each edge-graph $F(uv)=T_{ij}$ corresponds to another edge-graph $F(xy)=T_{xy}$, such that $i+j=xy~(\bmod ~12)$.\qqed
\end{example}

\begin{figure}[h]
\centering
\includegraphics[width=16.4cm]{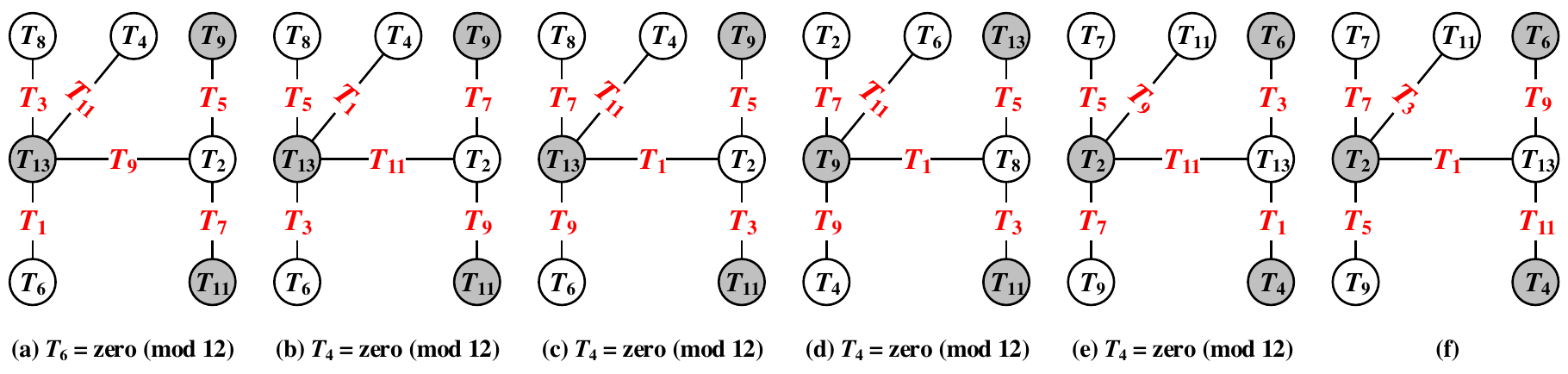}\\
\caption{\label{fig:odd-edge-index-colorings} {\small Six graphs admit six edge-index total graph-colorings; refer to Definition \ref{defn:edge-index-total-graph-coloring}.}}
\end{figure}

\begin{example}\label{exa:8888888888}
In Fig.\ref{fig:edge-index-colorings}, $F(uv)=T_{ij}$, $F(u)=T_{i}$ and $F(v)=T_{j}$, we have

(a) A graceful edge-index felicitous-difference total graph-coloring holding the felicitous-difference constraint $|i+j-ij|=5~(\bmod ~6)$;

(b) a graceful edge-index edge-difference total graph-coloring holding the edge-difference constraint $ij+|i-j|=7~(\bmod ~6)$;

(c) a graceful edge-index felicitous-difference total graph-coloring holding the felicitous-difference constraint $|i+j-ij|=4~(\bmod ~6)$;

(d) an edge-index felicitous-difference total graph-coloring holding the graceful-difference constraint $\big ||i-j|-ij\big |=3~(\bmod ~6)$;

(e) a graceful edge-index felicitous-difference total graph-coloring holding the felicitous-difference constraint $|i+j-ij|=5~(\bmod ~6)$;

(f) a graceful edge-index felicitous-difference total graph-coloring holding the felicitous-difference constraint $|i+j-ij|=2~(\bmod ~6)$;

(g) a graceful edge-index edge-magic total graph-coloring with $i+j+ij=10~(\bmod ~6)$, where each edge-graph $F(uv)=T_{ij}$ corresponds to another edge-graph $F(xy)=T_{xy}$ with $i+j-6~(\bmod ~6)=xy$;

(h) a graceful edge-index felicitous-difference total graph-coloring with $i+j+ij=12~(\bmod ~6)$, where $k-(i+j)=ij~(\bmod ~7)$.\qqed
\end{example}

\begin{figure}[h]
\centering
\includegraphics[width=16.4cm]{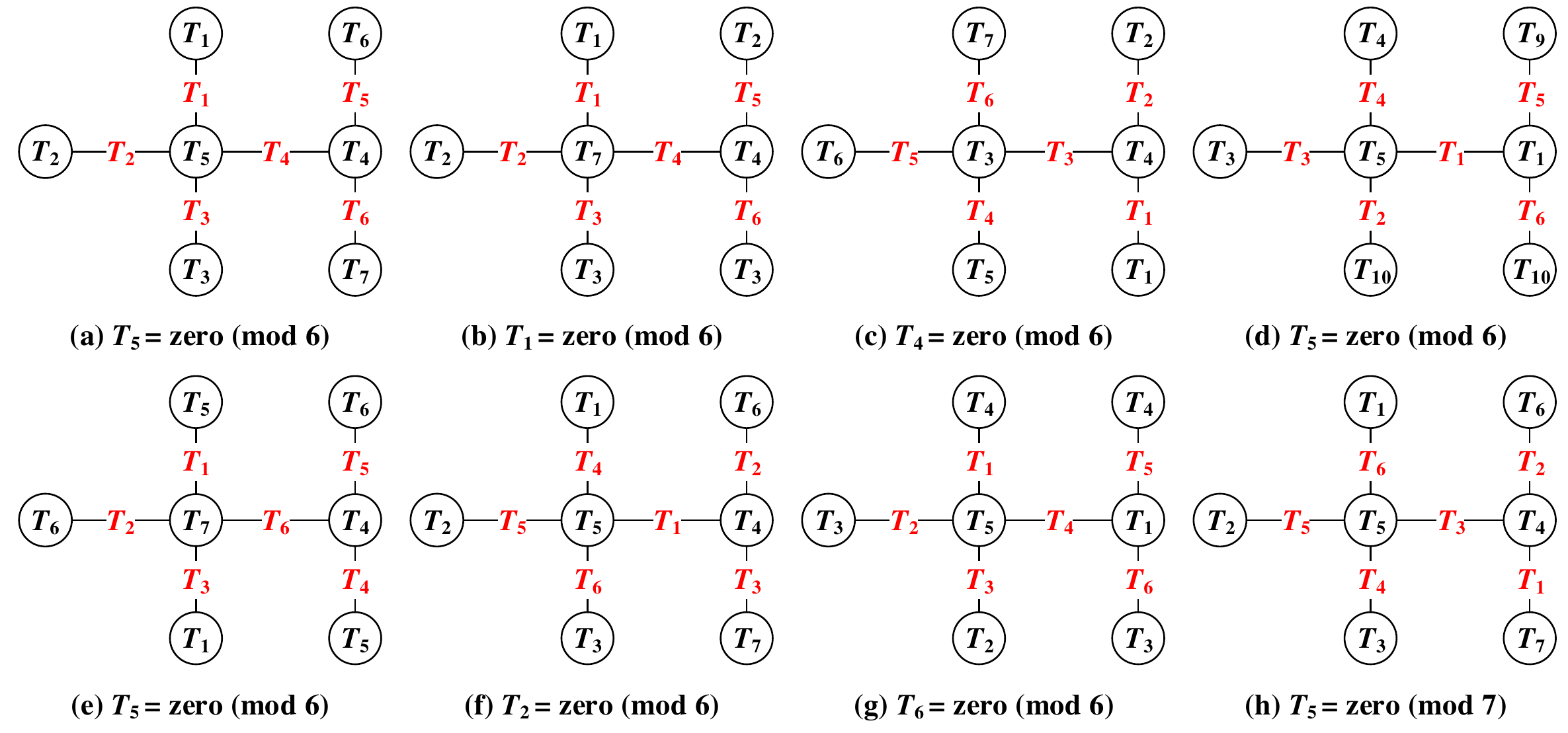}\\
\caption{\label{fig:edge-index-colorings} {\small Examples for understanding Definition \ref{defn:edge-index-total-graph-coloring}.}}
\end{figure}

\subsection{TOTAL-graph-coloring algorithm-I}

We show the TOTAL-graph-coloring algorithms for producing \emph{$t$-rank topological number-based strings} in this subsection.

\vskip 0.4cm

\noindent\textbf{TOTAL-graph-coloring algorithm-I.}

\textbf{Initialization.} By Definition \ref{defn:edge-index-total-graph-coloring}, a $(p,q)$-graph $H$ has its own Topcode-matrix
\begin{equation}\label{eqa:555555}
T_{code}\big (H,F^{1}_0\big )=\big (X^{1}_0,E^{1}_0,Y^{1}_0\big )^T
\end{equation} with the vertex vector $X^{1}_0=\big (F^{1}_0\big (x^{0}_1\big )$, $F^{1}_0\big (x^{0}_2\big ),\dots ,F^{1}_0\big (x^{0}_{q}\big )\big )$, the edge vector $E^{1}_0=\big (F^{1}_0\big (x^{0}_1y^{0}_1\big )$, $F^{1}_0\big (x^{0}_2y^{0}_2\big )$, $\dots ,F^{1}_0\big (x^{0}_{q}y^{0}_{q}\big )\big )$ and the vertex vector $Y^{1}_0=\big (F^{1}_0\big (y^{0}_1\big ),F^{1}_0\big (y^{0}_2\big ),\dots ,F^{1}_0\big (y^{0}_{q}\big )\big )$ under an \emph{edge-index total graph-coloring} $F^1_0$ defined in Definition \ref{defn:edge-index-total-graph-coloring}. For the simplicity of statement, we write
\begin{equation}\label{eqa:n-rank-graph-based-string11}
U^0(H)=X^{1}_0\cup E^{1}_0\cup Y^{1}_0=\big \{F^{1}_0\big (w_1\big ),F^{1}_0\big (w_2\big ),\dots ,F^{1}_0\big (w_{3q}\big )\big \}
\end{equation}
and we get a \emph{graph-based string}
\begin{equation}\label{eqa:graph-based-string-3q-k-0}
S^{3q}_{k_0}(H)=F^{1}_0\big (w_{j_1}\big )F^{1}_0\big (w_{j_2}\big )\cdots F^{1}_0\big (w_{j_{3q}}\big ),~k_0\in [1,(3q)!]
\end{equation} where $F^{1}_0\big (w_{j_1}\big ),F^{1}_0\big (w_{j_2}\big ),\cdots ,F^{1}_0\big (w_{j_{3q}}\big )$ is a permutation of elements of the graph set $U^0(H)$ defined in Eq.(\ref{eqa:n-rank-graph-based-string11}).

Since each graph $F^{1}_0\big (w_{j_s}\big )=T_{1,k_{j_s}}\in \textbf{\textrm{T}}_1$ for $s\in [1,3q]$ and $1\leq k_{j_s}\leq n_1$, then the graph $T_{1,k_{j_s}}$ has its own Topcode-matrix
\begin{equation}\label{eqa:555555}
T_{code}\big (T_{1,k_{j_s}},f_{1,k_{j_s}}\big )=T_{code}\big (F^{1}_0\big (w_{j_s}\big ),f_{1,k_{j_s}}\big )=\big (X_{1,k_{j_s}},E_{1,k_{j_s}},Y_{1,k_{j_s}}\big )^T
\end{equation} under the total coloring

(i) $f_{1,k_{j_s}}:V(T_{1,k_{j_s}})\cup E(T_{1,k_{j_s}})\rightarrow Z_{num}$, where $Z_{num}$ is an integer set; or

(ii) $f_{1,k_{j_s}}:V(T_{1,k_{j_s}})\cup E(T_{1,k_{j_s}})\rightarrow Z_{grd}$, where $Z_{grd}$ is a graph set.

\vskip 0.4cm

\textbf{Step 1.} Notice that $F^{1}_0(w_i)=T_{1,k^{1}_i}\in \textbf{\textrm{T}}_1$ with $i\in [1,3q_1]$ and $1\leq k^{1}_i \leq n_1$, and let $N_{1,i}=|E\big (T_{1,k^{1}_i}\big )|$. There is the Topcode-matrix
\begin{equation}\label{eqa:555555}
T_{code}\big (T_{1,k^{1}_i},f_{1,k^{1}_i}\big )=T_{code}\big (F^{1}_0(w_i),f_{1,k^{1}_i}\big )=\big (X_{1,k^{1}_i},E_{1,k^{1}_i},Y_{1,k^{1}_i}\big )^T
\end{equation} having

its own vertex vector $X_{1,k^{1}_i}=\big (f_{1,k^{1}_i}\big (x^{1}_1\big ),f_{1,k^{1}_i}\big (x^{1}_2\big ),\dots ,f_{1,k^{1}_i}\big (x^{1}_{N_{1,i}}\big )\big )$ with $x^{1}_j\in V(T_{1,k^{1}_i})$ for $j\in [1,N_{1,i}]$,

its own edge vector $E_{1,k^{1}_i}=\big (f_{1,k^{1}_i}\big (x^{1}_1y^{1}_1\big ),f_{1,k^{1}_i}\big (x^{1}_2y^{1}_2\big ),\dots ,f_{1,k^{1}_i}\big (x^{1}_{q_1}y^{1}_{N_{1,i}}\big )\big )$ with $x^{1}_jy^{1}_j\in E(T_{1,k^{1}_i})$ for $j\in [1,N_{1,i}]$, and

its own vertex vector $Y_{1,k^{1}_i}=\big (f_{1,k^{1}_i}\big (y^{1}_1\big ),f_{1,k^{1}_i}\big (y^{1}_2\big ),\dots ,f_{1,k^{1}_i}\big (y^{1}_{N_{1,i}}\big )\big )$ with $y^{1}_j\in V(T_{1,k^{1}_i})$ for $j\in [1,N_{1,i}]$
under the total coloring

(i) $f_{1,k^{1}_i}:V(T_{1,k^{1}_i})\cup E(T_{1,k^{1}_i})\rightarrow Z_{num}$, where $Z_{num}$ is an integer set; or

(ii) $f_{1,k^{1}_i}:V(T_{1,k^{1}_i})\cup E(T_{1,k^{1}_i})\rightarrow Z_{grd}$, where $Z_{grd}$ is a graph set.\\
If the above (ii) holds true, we have a graph set
\begin{equation}\label{eqa:n-rank-graph-based-string22}
U^{1}_{1,k^{1}_i}(T_{1,k^{1}_i})=X_{1,k^{1}_i}\cup E_{1,k^{1}_i}\cup Y_{1,k^{1}_i}=\big \{f_{1,k^{1}_i}\big (w_{k^{1}_i,1}\big ),f_{1,k^{1}_i}\big (w_{k^{1}_i,2}\big ),\dots ,f_{1,k^{1}_i}\big (w_{k^{1}_i,~3N_{1,i}}\big )\big \}
\end{equation} Moreover, we get a graph-based string
\begin{equation}\label{eqa:graph-based-string-N-1-i}
S^{N_{1,i}}_{k_{1,i}}(T_{1,k^{1}_i})=f_{1,k^{1}_i}\big (w_{k^{1}_i,j_1}\big )f_{1,k^{1}_i}\big (w_{k^{1}_i,j_2}\big )\cdots f_{1,k^{1}_i}\big (w_{k^{1}_i,~j_{3N_{1,i}}}\big )=S^{N_{1,j_1}N_{1,j_2}\cdots N_{1,j_{3N_{1,i}}}}_{k_{1,j_1}k_{1,j_2}\cdots k_{1,j_{3N_{1,i}}}}
\end{equation} with $k_{1,i}\in [1,(3N_{1,i})!]$, where $f_{1,k^{1}_i}\big (w_{k^{1}_i,j_1}\big ),f_{1,k^{1}_i}\big (w_{k^{1}_i,j_2}\big ),\dots ,f_{1,k^{1}_i}\big (w_{k^{1}_i,~j_{3N_{1,i}}}\big )$ is a permutation of elements of the graph set $U^{2}_{1,k^{1}_i}(T_{1,k^{1}_i})$ defined in Eq.(\ref{eqa:n-rank-graph-based-string22}).

By Eq.(\ref{eqa:graph-based-string-3q-k-0}) and Eq.(\ref{eqa:graph-based-string-N-1-i}), we have
\begin{equation}\label{eqa:555555}
S^{3q}_{k_0}(H)=F^{1}_0\big (w_{j_1}\big )F^{1}_0\big (w_{j_2}\big )\cdots F^{1}_0\big (w_{j_{3q}}\big )=S^{N_{1,j_1}N_{1,j_2}\cdots N_{1,j_{3q}}}_{k_{1,j_1}k_{1,j_2}\cdots k_{1,j_{3q}}}
\end{equation}

\vskip 0.4cm

\textbf{Step 2.} Since each graph $f_{1,k^{1}_i}\big (w_{k^{1}_i,~i}\big )=T_{2,k^{2}_i}\in \textbf{\textrm{T}}_2$ with $i\in [1,3N_{1,i}]$ and $1\leq k^{2}_i \leq n_2$, and let $N_{2,i}=|E\big (T_{2,k^{2}_i}\big )|$. We get the Topcode-matrix
\begin{equation}\label{eqa:555555}
T_{code}\big (T_{2,k^{2}_i},f_{2,k^{2}_i}\big )=T_{code}\big (f_{1,k^{1}_i}\big (w_{k^{1}_i,~i}\big ),f_{2,k^{2}_i}\big )=\big (X_{2,k^{2}_i},E_{2,k^{2}_i},Y_{2,k^{2}_i}\big )^T
\end{equation} having

its own vertex vector $X_{2,k^{2}_i}=\big (f_{2,k^{2}_i}\big (x^{2}_1\big ),f_{2,k^{2}_i}\big (x^{2}_2\big ),\dots ,f_{2,k^{2}_i}\big (x^{2}_{N_{2,i}}\big )\big )$ with $x^{2}_j\in V(T_{2,k^{2}_i})$ for $j\in [1,N_{2,i}]$,

its own edge vector $E_{2,k^{2}_i}=\big (f_{2,k^{2}_i}\big (x^{2}_1y^{2}_1\big ),f_{2,k^{2}_i}\big (x^{2}_2y^{2}_2\big ),\dots ,f_{2,k^{2}_i}\big (x^{2}_{q_1}y^{2}_{N_{2,i}}\big )\big )$ with $x^{2}_jy^{2}_j\in E(T_{2,k^{2}_i})$ for $j\in [1,N_{2,i}]$,

its own vertex vector $Y_{2,k^{2}_i}=\big (f_{2,k^{2}_i}\big (y^{2}_1\big ),f_{2,k^{2}_i}\big (y^{2}_2\big ),\dots ,f_{2,k^{2}_i}\big (y^{2}_{N_{2,i}}\big )\big )$ with $y^{2}_j\in V(T_{2,k^{2}_i})$ for $j\in [1,N_{2,i}]$
under the total coloring

(i) $f_{2,k^{2}_i}:V(T_{2,k^{2}_i})\cup E(T_{2,k^{2}_i})\rightarrow Z_{num}$, where $Z_{num}$ is an integer set; or

(ii) $f_{2,k^{2}_i}:V(T_{2,k^{2}_i})\cup E(T_{2,k^{2}_i})\rightarrow Z_{grd}$, where $Z_{grd}$ is a graph set.\\
If the above (ii) holds true, by
\begin{equation}\label{eqa:n-rank-graph-based-string33}
U^{2}_{2,k^{2}_i}(T_{2,k^{2}_i})=X_{2,k^{2}_i}\cup E_{2,k^{2}_i}\cup Y_{2,k^{2}_i}=\big \{f_{2,k^{2}_i}\big (w_{k^{2}_i,1}\big ),f_{2,k^{2}_i}\big (w_{k^{2}_i,2}\big ),\dots ,f_{2,k^{2}_i}\big (w_{k^{2}_i,~3N_{2,i}}\big )\big \}
\end{equation} we get a graph-based string
\begin{equation}\label{eqa:graph-based-string-N-2-i}
S^{N_{2,i}}_{k_{2,i}}(T_{2,k^{2}_i})=f_{2,k^{2}_i}\big (w_{k^{2}_i,j_1}\big )f_{2,k^{2}_i}\big (w_{k^{2}_i,j_2}\big )\cdots f_{2,k^{2}_i}\big (w_{k^{2}_i,~j_{3N_{2,i}}}\big )=S^{N_{2,j_1}N_{2,j_2}\cdots N_{2,j_{3N_{2,i}}}}_{k_{2,j_1}k_{2,j_2}\cdots k_{2,j_{3N_{2,i}}}}
\end{equation} with $k_2\in [1,(3N_{2,i})!]$, where $f_{2,k^{2}_i}\big (w_{k^{2}_i,j_1}\big ),f_{2,k^{2}_i}\big (w_{k^{2}_i,j_2}\big ),\cdots ,f_{2,k^{2}_i}\big (w_{k^{2}_i,~j_{3N_{2,i}}}\big )$ is a permutation of elements of the graph set $U^{2}_{2,k^{2}_i}(T_{2,k^{2}_i})$ defined in Eq.(\ref{eqa:n-rank-graph-based-string33}).

\vskip 0.4cm

\textbf{Step $t$.} Since $f_{t-1,k^{t-1}_i}\big (w_{k^{t-1}_i,~i}\big )=T_{t,k^{t}_i}\in \textbf{\textrm{T}}_t$ with $i\in [1,3N_{t-1,i}]$ and $1\leq k^{t}_i \leq n_t$. Let $N_{t,i}=|E\big (T_{t,k^{t}_i}\big )|$, so we have the Topcode-matrix
\begin{equation}\label{eqa:555555}
T_{code}\big (T_{t,k^{t}_i},f_{t,k^{t}_i}\big )=T_{code}\big (f_{t-1,k^{t-1}_i}\big (w_{k^{t-1}_i,i}\big ),f_{t,k^{t}_i}\big )=\big (X_{t,k^{t}_i},E_{t,k^{t}_i},Y_{t,k^{t}_i}\big )^T
\end{equation} having

its own vertex vector $X_{t,k^{t}_i}=\big (f_{t,k^{t}_i}\big (x^{t}_1\big ),f_{t,k^{t}_i}\big (x^{t}_2\big ),\dots ,f_{t,k^{t}_i}\big (x^{t}_{N_{t,i}}\big )\big )$ with $x^{t}_j\in V(T_{t,k^{t}_i})$ for $j\in [1,N_{t,i}]$,

its own edge vector $E_{t,k^{t}_i}=\big (f_{t,k^{t}_i}\big (x^{t}_1y^{t}_1\big ),f_{t,k^{t}_i}\big (x^{t}_2y^{t}_2\big ),\dots ,f_{t,k^{t}_i}\big (x^{t}_{N_{t,i}}y^{t}_{N_{t,i}}\big )\big )$ with $x^{t}_jy^{t}_j\in E(T_{t,k^{t}_i})$ for $j\in [1,N_{t,i}]$,

its own vertex vector $Y_{t,k^{t}_i}=\big (f_{t,k^{t}_i}\big (y^{t}_1\big ),f_{t,k^{t}_i}\big (y^{t}_2\big ),\dots ,f_{t,k^{t}_i}\big (y^{t}_{N_{t,i}}\big )\big )$ with $y^{t}_j\in V(T_{t,k^{t}_i})$ for $j\in [1,N_{t,i}]$
under the total coloring

(i) $f_{t,k^{t}_i}:V(T_{t,k^{t}_i})\cup E(T_{t,k^{t}_i})\rightarrow Z_{num}$, where $Z_{num}$ is an integer set; or

(ii) $f_{t,k^{t}_i}:V(T_{t,k^{t}_i})\cup E(T_{t,k^{t}_i})\rightarrow Z_{grd}$, where $Z_{grd}$ is a graph set.\\
If the above (ii) holds true, thus, we get a graph set
\begin{equation}\label{eqa:n-rank-graph-based-string44}
U^{t}_{t,k^{t}_i}(T_{t,k^{t}_i})=X_{t,k^{t}_i}\cup E_{t,k^{t}_i}\cup Y_{t,k^{t}_i}=\big \{f_{t,k^{t}_i}\big (w_{k^{t}_i,1}\big ),f_{t,k^{t}_i}\big (w_{k^{t}_i,2}\big ),\dots ,f_{t,k^{t}_i} (w_{k^{t}_i,~3N_{t,i}}\big )\big \}
\end{equation}There is a graph-based string
\begin{equation}\label{eqa:555555}
S^{N_{t,i}}_{k_{t,i}}(T_{t,k^{t}_i})=f_{t,k^{t}_i}\big (w_{k^{t}_i,j_1}\big )f_{t,k^{t}_i}\big (w_{k^{t}_i,j_2}\big )\cdots f_{t,k^{t}_i}\big (w_{k^{t}_i,~j_{3N_{t,i}}}\big )=S^{N_{t,j_1}N_{t,j_2}\cdots N_{t,j_{3N_{t,i}}}}_{k_{t,j_1}k_{t,j_2}\cdots k_{t,j_{3N_{t,i}}}}
\end{equation} with $k_{t,i}\in [1,(3N_{t,i})!]$, where $f_{t,k^{t}_i}\big (w_{k^{t}_i,j_1}\big ),f_{t,k^{t}_i}\big (w_{k^{t}_i,j_2}\big ),\dots ,f_{t,k^{t}_i} (w_{k^{t}_i,~j_{3N_{t,i}}}\big )$ is a permutation of elements of the graph set $U^{t}_{t,k^{t}_i}(T_{t,~k^{t}_i})$ defined in Eq.(\ref{eqa:n-rank-graph-based-string44}).

Since each graph $f_{t,k^{t}_i}\big (w_{k^{t}_i,j_s})=T_{t,k_{j_s}}\in \textbf{\textrm{T}}_t$ for $s\in [1,3N_{t,i}]$ and $1\leq k_{j_s}\leq n_t$, then the graph $T_{t,j_s}$ has its own Topcode-matrix
\begin{equation}\label{eqa:555555}
T_{code}\big (T_{t,j_s},f_{t,j_s}\big )=T_{code}\big (f_{t,k^{t}_i}\big (w_{k^{t}_i,j_s}),f_{t,j_s}\big )=\big (X_{t,j_s},E_{t,j_s},Y_{t,j_s}\big )^T
\end{equation} having

its own vertex vector $X_{t,j_s}=\big (f_{t,j_s}(x_1 ),f_{t,j_s}(x_2 ),\dots ,f_{t,j_s}(x_{N_{t,i}})\big )$,

its own edge vector $E_{t,j_s}=\big (f_{t,j_s}(x_1y_1 ),f_{t,j_s}(x_2y_2 ),\dots ,f_{t,j_s}(x_{N_{t,i}}y_{N_{t,i}} )\big )$ and

its own vertex vector $Y_{t,j_s}=\big (f_{t,j_s}(y_1 ),f_{t,j_s}(y_2 ),\dots ,f_{t,j_s}(y_{N_{t,i}} )\big )$.\\
And we have the set
\begin{equation}\label{eqa:last-number-based-string}
U^{t}(T_{t,j_s})=X_{t,j_s}\cup E_{t,j_s}\cup Y_{t,j_s}=\big \{f_{t,j_s}(w_{j_s,1}),f_{t,j_s}(w_{j_s,2}),\dots ,f_{t,j_s}(w_{j_s,N})\big \},~N=3N_{t,i}
\end{equation}

\vskip 0.4cm

\textbf{Step end.} Return $t$-rank topological number-based strings.

\vskip 0.4cm

Suppose that $f_{t,j_s}:V(T_{t,j_s})\cup E(T_{t,j_s})\rightarrow Z_{num}$, where $Z_{num}$ is an integer set, so we have

(i) the $[0,9]$-string $f_{t,j_s}(x_i)=d_{i,1}d_{i,2}\cdots d_{i,a_i}$ with $d_{i,j}\in [0,9]$ and $j\in [1,a_i]$, $i\in [1,N_{t,i}]$;

(ii) the $[0,9]$-string $f_{t,j_s}(x_iy_i)=d\,'_{i,1}d\,'_{i,2}\cdots d\,'_{i,b_i}$ with $d\,'_{i,j}\in [0,9]$ and $j\in [1,b_i]$, $i\in [1,N_{t,i}]$ and

(iii) the $[0,9]$-string $f_{t,j_s}(y_i)=d\,''_{i,1}d\,''_{i,2}\cdots d\,''_{i,c_i}$ with $d\,''_{i,j}\in [0,9]$ and $j\in [1,c_i]$, $i\in [1,N_{t,i}]$.

For each permutation $f_{t,j_s}(w_{j_s,j_1}),f_{t,j_s}(w_{j_s,j_2}),\dots ,f_{t,j_s}(w_{j_s,j_N})$ of elements of the graph set $U^{t}(T_{t,j_s})$ defined in Eq.(\ref{eqa:last-number-based-string}), we have a number-based string
\begin{equation}\label{eqa:555555}
n^{j_s(r)}_{bs}(T_{t,j_s})=f_{t,j_s}(w_{j_s,j_1})f_{t,j_s}(w_{j_s,j_2})\cdots f_{t,j_s}(w_{j_s,j_N})
\end{equation} with $s\in [1,3N_{t,i}]$, $j_s\in [1,3N_{t,i}]$ and $r\in [1,N]$. In total, we have a number-based string
$$
n(T_{t,j_s})=n^{j_s(1)}_{bs}(T_{t,j_s})n^{j_s(2)}_{bs}(T_{t,j_s})\cdots n^{j_s(N)}_{bs}(T_{t,j_s}),~s\in [1,3N_{t,i}]
$$ Moreover, we get a compound number-based string $n^k_{gbs}(T_{t,k^{t}_i}):=n(T_{t,j_1})n(T_{t,j_2})\cdots n(T_{t,j_N})$ with $k\in [1,3N_{t,i}]$. In general, we have
{\small
\begin{equation}\label{eqa:555555}
{
\begin{split}
&S^{3q}_{k_0}(H)=F^{1}_0\big (w_{j_1}\big )F^{1}_0\big (w_{j_2}\big )\cdots F^{1}_0\big (w_{j_{3q}}\big )=S^{N_{1,j_1}N_{1,j_2}\cdots N_{1,j_{3q}}}_{k_{1,j_1}k_{1,j_2}\cdots k_{1,j_{3q}}}\\
&S^{N_{1,j_1}N_{1,j_2}\cdots N_{1,j_{3q}}}_{k_{1,j_1}k_{1,j_2}\cdots k_{1,j_{3q}}}=S^{N_{2,j_1}N_{2,j_2}\cdots N_{2,j_{3N_{1,j_1}}},N_{2,j_1}N_{2,j_2}\cdots N_{2,j_{3N_{1,j_2}}},\cdots \cdots ~, N_{2,j_{3N_{1,j_1}}}N_{2,j_1}N_{2,j_2}\cdots N_{2,j_{3N_{1,j_{3q}}}}}_{k_{2,j_1}k_{2,j_2}\cdots k_{2,j_{3N_{1,j_1}}},k_{2,j_1}k_{2,j_2}\cdots k_{2,j_{3N_{1,j_2}}},\cdots \cdots ~,k_{2,j_1}k_{2,j_2}\cdots k_{2,j_{3N_{1,j_{3q}}}}}\\
& \quad \cdots \cdots \cdots \\
&S^{N_{t,j_1}N_{t,j_2}\cdots N_{t,j_{3q(t)}}}_{k_{t,j_1}k_{t,j_2}\cdots k_{t,j_{3q(t)}}}=S^{N_{t,j_1}N_{t,j_2}\cdots N_{t,j_{3N_{t,j_1}}}N_{t,j_1},N_{t,j_2}\cdots N_{t,j_{3N_{t,j_2}}},\cdots \cdots ~, N_{t,j_{3N_{t,j_1}}}N_{t,j_1}N_{t,j_2}\cdots N_{t,j_{3N_{t,j_{3q(t)}}}}}_{k_{t,j_1}k_{t,j_2}\cdots k_{t,j_{3N_{t,j_1}}},k_{t,j_1}k_{t,j_2}\cdots k_{t,j_{3N_{t,j_2}}},\cdots\cdots ~,k_{t,j_1}k_{t,j_2}\cdots k_{t,j_{3N_{t,j_{3q(t)}}}}}
\end{split}}
\end{equation}
}

By the TOTAL-graph-coloring algorithm-I, we have

\textbf{Level-1.} Each element $w_k\in V(H)\cup E(H)$ with $k\in [1,p+q]$ corresponds to a graph $T^k_{1,i}\in \textbf{\textrm{T}}_1$ with $e_{1,k}=|E(T^k_{1,i})|$ and $v_{1,k}=|V(T^k_{1,i})|$, the Topcode-matrix $T_{code}(T^k_{1,i},f^k_{1,i})=(X^k_{1,i},E^k_{1,i},Y^k_{1,i})^T$ produces $(3e_{1,k})!$ $1$-rank graph-based strings
\begin{equation}\label{eqa:555555}
g(T^k_{1,i},r)=f^k_{1,i}(w^1_{r,1})f^k_{1,i}(w^1_{r,2})\cdots f^k_{1,i}(w^1_{r,3e_{1,k}}),~r\in [1,(3e_{1,k})!],k\in [1,p+q]
\end{equation} There are $N_1$ different $1$-rank graph-based strings like $g(T^k_{1,i},r)$, where $N_1=[(3e_{1,1})!]\cdot [(3e_{1,2})!]\cdot \cdots \cdot [(3e_{1,p+q})!]$.

\textbf{Level-2.} Each element $w^1_k\in V(T^k_{1,i})\cup E(T^k_{1,i})$ with $k\in [1,v_{1,k}+e_{1,k}]$ corresponds to a graph $T^k_{2,i}\in \textbf{\textrm{T}}_2$ with $e_{2,k}=|E(T^k_{2,i})|$ and $v_{2,k}=|V(T^k_{2,i})|$, the Topcode-matrix $T_{code}(T^k_{2,i},f^k_{2,i})=(X^k_{2,i},E^k_{2,i},Y^k_{2,i})^T$ produces $(3e_{2,k})!$ $2$-rank graph-based strings
\begin{equation}\label{eqa:555555}
g(T^k_{2,i},r)=f^k_{2,i}(w^2_{r,1})f^k_{2,i}(w^2_{r,2})\cdots f^k_{2,i}(w^2_{r,3e_{2,k}}),~r\in [1,(3e_2(k))!],k\in [1,v_{1,k}+e_{1,k}]
\end{equation} There are $N_2$ different $2$-rank graph-based strings like $g(T^k_{2,i},r)$, where
$$N_2=[(3e_{2,1})!]\cdot [(3e_{2,2})!]\cdot \cdots \cdot [(3e_{2,v_{1,k}+e_{1,k}})!]$$

\textbf{Level-$t$.} Each element $w^{t-1}\in V(T_{t-1,i})\cup E(T_{t-1,i})$ with $k\in [1,v_{t-1,k}(k)+e_{t-1,k}(k)]$ corresponds to a graph $T^k_{t,i}\in \textbf{\textrm{T}}_t$ with $e_t(k)=|E(T^k_{t,i})|$ and $v_t(k)=|V(T^k_{2,i})|$, the Topcode-matrix $T_{code}(T^k_{t,i},f^k_{t,i})=(X^k_{t,i},E^k_{t,i},Y^k_{t,i})^T$ produces $(3e_t(k))!$ graph-based strings
\begin{equation}\label{eqa:555555}
g(T^k_{t,i},r)=f^k_{t,i}(w^t_{r,1})f^k_{t,i}(w^t_{r,2})\cdots f^k_{t,i}(w^t_{r,3e_t(k)}),~r\in [1,(3e_t(k))!],k\in [1,v_{t-1,k}+e_{t-1,k}]
\end{equation} There are
\begin{equation}\label{eqa:555555}
N_t=[(3e_{t,1})!]\cdot [(3e_{t,2})!]\cdot \cdots \cdot [(3e_{t,v_{t-1,k}+e_{t-1,k}})!]
\end{equation} different $t$-rank graph-based strings like $g(T^k_{2,i},r)$, in total.

At the end of Level-$t$, the $(p,q)$-graph $H$ produces $I(t)$ $t$-rank graph-based strings, in total, where $I(t)=N_1\cdot N_2\cdot \cdots \cdot N_t=\prod ^t_{i=1}N_i$.

\subsection{TOTAL-graph-coloring algorithm-II}

Let $\textbf{\textrm{O}}=(O_1,O_2,\dots ,O_p)$ be a \emph{graph operation base}, where each $O_i$ is a graph operation of graph theory, and let $\textbf{\textrm{T}}_i=\big (T_{i,1},T_{i,2},\dots, T_{i,n_i}\big )$ be a \emph{colored graph base} for $n_i\geq 1$ and $i\in [1,M]$, so that each colored graph $T_{i,j}$ admits a total coloring (or a total graph-coloring) $f_{i,j}$. Suppose that each colored graph base $\textbf{\textrm{T}}_i$ is an every-zero additive group defined in Definition \ref{defn:every-zero-abstract-group}.

\vskip 0.4cm

\noindent\textbf{TOTAL-graph-coloring algorithm-II with graph operation.}

\textbf{Input.} A $(p_1,q_1)$-graph $H_1$ admits an \emph{edge-index total graph-coloring} $F_{1}$ defined in Definition \ref{defn:edge-index-total-graph-coloring}.

\textbf{Output.} $t$-rank topological number-based strings.
\vskip 0.4cm

\textbf{Step-II-1.} A $(p_1,q_1)$-graph $H_1$ admitting an \emph{edge-index total graph-coloring} $F_{1}:V(H_1)\cup E(H_1)\rightarrow \textbf{\textrm{T}}_1$, such that each edge $uv\in E(H_1)$ holds $F_{1}(u)=T_{1,i}$, $F_{1}(v)=T_{1,j}$ and $F_{1}(uv)=T_{1,\lambda}$ with the index $\lambda=i+j-k_1~(\bmod~M_1)$ for a preappointed \emph{zero} $k_1\in [1,n_1]$.

\textbf{Step-II-1.1.} From the edge-index total graph-coloring $F_{1}$ of the graph $H_1$, we get a Topcode-matrix $T_{code}(H_1,F_{1})= (X_{1},E_{1},Y_{1} )^T$ with the vertex vector $X_{1}=(F_{1}(x_{1,1})$, $F_{1}(x_{1,2}),\dots ,F_{1}(x_{1,q_1}))$, the edge vector $E_{1}=(F_{1}(x_{1,1}y_{1,1}),F_{1}(x_{1,2}y_{1,2}),\dots ,F_{1}(x_{1,q_1}y_{1,q_1}))$ and the vertex vector $Y_{1}=(F_{1}(y_{1,1}),F_{1}(y_{1,2}),\dots ,F_{1}(y_{1,q_1}))$. For the simplicity of statement, we write
\begin{equation}\label{eqa:edge-index-graph-based-string11}
U(H_1)=X_{1}\cup E_{1}\cup Y_{1}=\{F_{1}(w_{1,1}),F_{1}(w_{1,2}),\dots ,F_{1}(w_{1,3q_1})\}
\end{equation}
and we get \emph{graph-based strings}
\begin{equation}\label{eqa:graph-based-string-3q-11}
S_{gb}(H_1,r)=F_{1}(w_{1,r,1})F_{1}(w_{1,r,2})\cdots F_{1}(w_{1,r,3q_1})
\end{equation} with $r\in [1,(3q_1)!]$, where $F_{1}(w_{1,r,1}),F_{1}(w_{1,r,2}),\dots ,F_{1}(w_{1,r,3q_1})$ is a permutation of elements of the graph set $U(H_1)$ defined in Eq.(\ref{eqa:edge-index-graph-based-string11}).

\textbf{Step-II-1.2.} Since each graph $F_{1}(w_{1,r,j})=T_{1,(r,j)}\in \textbf{\textrm{T}}_1$ with $j\in [1,3q_1]$, and $T_{1,(r,j)}$ admits a total coloring
$$f_{1,(r,j)}:V(T_{1,(r,j)})\cup E(T_{1,(r,j)})\rightarrow W_{1,(r,j)}
$$ where $W_{1,(r,j)}$ is a non-negative integer set. Let $q(1,r,j)=|E(T_{1,(r,j)})|$. Thereby, we have a Topcode-matrix
$$
T_{code}(T_{1,(r,j)},f_{1,(r,j)})= (X_{1,(r,j)},E_{1,(r,j)},Y_{1,(r,j)})^T
$$ with its own \emph{vertex vectors} and \emph{edge vector} as follows

$X_{1,(r,j)}=(f_{1,(r,j)}(u_{1,(r,j),1})$, $f_{1,(r,j)}(u_{1,(r,j),2}),\dots ,f_{1,(r,j)}(u_{1,(r,j),q(1,r,j)}))$,

$E_{1,(r,j)}=(f_{1,(r,j)}(u_{1,(r,j),1}v_{1,(r,j),1}),f_{1,(r,j)}(u_{1,(r,j),2}v_{1,(r,j),2}),\dots ,f_{1,(r,j)}(u_{1,(r,j),q(1,r,j)}v_{1,(r,j),q(1,r,j)}))$,

$Y_{1,(r,j)}=(f_{1,(r,j)}(v_{1,(r,j),1}),f_{1,(r,j)}(v_{1,(r,j),2}),\dots ,f_{1,(r,j)}(v_{1,(r,j),q(1,r,j)}))$.\\
We get a graph set
\begin{equation}\label{eqa:graph-based-string-base}
{
\begin{split}
U(T_{1,(r,j)})&=X_{1,(r,j)}\cup E_{1,(r,j)}\cup Y_{1,(r,j)}\\
&=\{f_{1,(r,j)}(z_{1,(r,j),1}),f_{1,(r,j)}(z_{1,(r,j),3}),\dots ,f_{1,(r,j)}(z_{1,(r,j),3q(1,r,j)})\}
\end{split}}
\end{equation} with $r\in [1,(3q_1)!]$ and $j\in [1,3q_1]$.

Each graph $F_{1}(w_{1,r,j})=T_{1,(r,j)}$ distributes us a topological number-based string
$$f_{1,(r,j)}(z_{1,(r,j),1})f_{1,(r,j)}(z_{1,(r,j),3})\cdots f_{1,(r,j)}(z_{1,(r,j),3q(1,r,j)})=S^{q(1,r,j)}_{(r,j)}
$$ then each graph-based string $S_{gb}(H_1,r)$ defined in Eq.(\ref{eqa:graph-based-string-3q-11}) distributes us a \emph{$1$-rank topological number-based string}
\begin{equation}\label{eqa:555555}
s(H_1,r)=S^{q(1,r,1)q(1,r,2)\cdots q(1,r,3q_1)}_{(r,1)(r,2)\cdots (r,3q_1)}\leftarrow S_{gb}(H_1,r)=F_{1}(w_{1,r,1})F_{1}(w_{1,r,2})\cdots F_{1}(w_{1,r,3q_1})
\end{equation} with $r\in [1,(3q_1)!]$.

\vskip 0.4cm
\textbf{Step-II-1.3.} A $(p_2,q_2)$-graph $H_2$ is obtained in the following way: For each edge $uv\in E(H_1)$, do a graph operation $O_{s}\in \textbf{\textrm{O}}$ to the graphs $F_{1}(u)$ and $F_{1}(uv)$ for producing a subgraph $F_{1}(u)O_{s}F_{1}(uv) $, and do a graph operation $O_{k}\in \textbf{\textrm{O}}$ to the graphs $F_{1}(v)$ and $F_{1}(uv)$ for producing a subgraph $ F_{1}(v)O_{k}F_{1}(uv) $, then we get a subgraph $F_{1}(u)O_{s}F_{1}(uv)O_{k}F_{1}(v)$, called a \emph{graph-edge}. We write $H_2=H_1[\textbf{\textrm{O}}]^{n_1}_{k=1}a_{1,k}T_{1,k}$ for $a_{1,k}\in Z^0$ and $T_{1,k}\in \textbf{\textrm{T}}_1$, immediately, we get an \emph{operation graphic lattice}
\begin{equation}\label{eqa:edge-index-graphic-lattice11}
\textbf{\textrm{L}}(\textbf{\textrm{Q}}_1[\textbf{\textrm{O}}]\textbf{\textrm{T}}_1)=\left \{H[\textbf{\textrm{O}}]^{n_1}_{k=1}a_{1,k}T_{1,k}:~a_{1,k}\in Z^0,~T_{1,k}\in \textbf{\textrm{T}}_1, H\in \textbf{\textrm{Q}}_1 \right \}
\end{equation} with $\sum ^{n_1}_{k=1}a_{1,k}\geq 1$, where $\textbf{\textrm{Q}}_1$ is a graph set.

\vskip 0.4cm

\textbf{Step-II-2.} Suppose the $(p_2,q_2)$-graph $H_2$ admits an \emph{edge-index total graph-coloring} $F_{2}:V(H_2)\cup E(H_2)\rightarrow \textbf{\textrm{T}}_2$, such that each edge $xy\in E(H_2)$ holds $F_{2}(x)=T_{2,i}$, $F_{2}(y)=T_{2,j}$ and $F_{2}(xy)=T_{2,\lambda}$ with the index $\lambda=i+j-k_2~(\bmod~M_2)$ for a preappointed \emph{zero} $T_{2,k_2}\in \textbf{\textrm{T}}_2$; refer to Definition \ref{defn:edge-index-total-graph-coloring}.

\textbf{Step-II-2.1.} From the edge-index total graph-coloring $F_{2}$ of the graph $H_2$, we get a Topcode-matrix $T_{code}(H_2,F_{2})= (X_{2},E_{2},Y_{2} )^T$ with the vertex vector $X_{2}=(F_{2}(x_{2,1})$, $F_{2}(x_{2,2}),\dots ,F_{2}(x_{2,q_2}))$, the edge vector $E_{2}=(F_{2}(x_{2,1}y_{2,1}),F_{2}(x_{2,2}y_{2,2}),\dots ,F_{2}(x_{2,q_2}y_{2,q_2}))$ and the vertex vector $Y_{2}=(F_{2}(y_{2,1}),F_{2}(y_{2,2}),\dots ,F_{2}(y_{2,q_2}))$. For the simplicity of statement, we write
\begin{equation}\label{eqa:t-rank-graph-based-string22}
U(H_2)=X_{2}\cup E_{2}\cup Y_{2}=\{F_{2}(w_{2,r,1}),F_{1}(w_{2,r,2}),\cdots ,F_{1}(w_{2,r,3q_2})\}
\end{equation}
and we get a \emph{graph-based string}
\begin{equation}\label{eqa:graph-based-string-3q-22}
S_{gb}(H_2,r)=F_{2}(w_{2,r,1})F_{1}(w_{2,r,2})\cdots F_{1}(w_{2,r,3q_2})
\end{equation} with $r\in [1,(3q_2)!]$, where $F_{2}(w_{2,r,1})F_{1}(w_{2,r,2})\cdots F_{1}(w_{2,r,3q_2})$ is a permutation of elements of the graph set $U(H_2)$ defined in Eq.(\ref{eqa:t-rank-graph-based-string22}).

\textbf{Step-II-2.2.} Since each graph $F_{2}(w_{2,r,j})=T_{2,(r,j)}\in \textbf{\textrm{T}}_2$ with $j\in [1,3q_2]$, and $T_{2,(r,j)}$ admits a a total coloring $f_{2,(r,j)}:V(T_{2,(r,j)})\cup E(T_{2,(r,j)})\rightarrow W_{2,(r,j)}$, where $W_{2,(r,j)}$ is a non-negative integer set. Let $q(2,r,j)=|E(T_{2,(r,j)})|$.

So we have a Topcode-matrix $T_{code}(T_{2,(r,j)},f_{2,(r,j)})= (X_{2,(r,j)},E_{2,(r,j)},Y_{2,(r,j)})^T$ with its own vertex vectors and edge vector as follows

$X_{2,(r,j)}=(f_{2,(r,j)}(u_{2,(r,j),1})$, $f_{2,(r,j)}(u_{2,(r,j),2}),\dots ,f_{2,(r,j)}(u_{2,(r,j),q(2,r,j)}))$,

$E_{2,(r,j)}=(f_{2,(r,j)}(u_{2,(r,j),1}v_{2,(r,j),1}),f_{2,(r,j)}(u_{2,(r,j),2}v_{2,(r,j),2}),\dots ,f_{2,(r,j)}(u_{2,(r,j),q(2,r,j)}v_{2,(r,j),q(2,r,j)}))$,

$Y_{2,(r,j)}=(f_{2,(r,j)}(v_{2,(r,j),1}),f_{2,(r,j)}(v_{2,(r,j),2}),\dots ,f_{2,(r,j)}(v_{2,(r,j),q(2,r,j)}))$.\\
We get a graph set
\begin{equation}\label{eqa:graph-based-string-base}
{
\begin{split}
U(T_{2,(r,j)})&=X_{2,(r,j)}\cup E_{2,(r,j)}\cup Y_{2,(r,j)}\\
&=\{f_{2,(r,j)}(z_{2,(r,j),1}),f_{2,(r,j)}(z_{2,(r,j),3}),\dots ,f_{2,(r,j)}(z_{2,(r,j),3q(2,r,j)})\}
\end{split}}
\end{equation} with $r\in [1,(3q_2)!]$ and $j\in [1,3q_2]$.

Each graph $F_{2}(w_{2,r,j})=T_{2,(r,j)}\in \textbf{\textrm{T}}_2$ distributes us a topological number-based string
$$f_{2,(r,j)}(z_{2,(r,j),1})f_{2,(r,j)}(z_{2,(r,j),3})\cdots f_{2,(r,j)}(z_{2,(r,j),3q(2,r,j)})=S^{q(2,r,j)}_{(r,j)}
$$ then each graph-based string $S_{gb}(H_2,r)$ defined in Eq.(\ref{eqa:graph-based-string-3q-22}) distributes us a \emph{$2$-rank topological number-based string} as follows
\begin{equation}\label{eqa:555555}
s(H_2,r)=S^{q(2,r,1)q(2,r,2)\cdots q(2,r,3q_2)}_{(r,1)(r,2)\cdots (r,3q_2)}\leftarrow S_{gb}(H_2,r)=F_{2}(w_{2,r,1})F_{2}(w_{2,r,2})\cdots F_{2}(w_{2,r,3q_2})
\end{equation} with $r\in [1,(3q_2)!]$.

\vskip 0.4cm

\textbf{Step-II-2.3.} A $(p_3,q_3)$-graph $H_3$ is obtained in the following way: For each edge $xy\in E(H_2)$, do a graph operation $O_{xy}\in \textbf{\textrm{O}}$ to the graphs $F_{2}(x)$ and $F_{2}(xy)$ for producing a subgraph $F_{2}(x)O_{xy}F_{2}(xy) $, and do a graph operation $O_{j}\in \textbf{\textrm{O}}$ to the graphs $F_{2}(y)$ and $F_{2}(xy)$ for producing a subgraph $F_{2}(y)O_{j}F_{2}(xy) $, then we get a subgraph (also graph-edge) $F_{2}(x)O_{xy}F_{2}(xy)O_{j}F_{2}(y)$. We write $H_3=H_2[\textbf{\textrm{O}}]^{n_2}_{k=1}a_{2,k}T_{2,k}$ for $a_{2,k}\in Z^0$ and $T_{2,k}\in \textbf{\textrm{T}}_2$.

Let $\textbf{\textrm{Q}}_2=\textbf{\textrm{L}}(\textbf{\textrm{Q}}_1[\textbf{\textrm{O}}]\textbf{\textrm{T}}_1)$ defined in Eq.(\ref{eqa:edge-index-graphic-lattice11}), we get an \emph{operation graphic lattice}
\begin{equation}\label{eqa:edge-index-graphic-lattice22}
\textbf{\textrm{L}}(\textbf{\textrm{Q}}_2[\textbf{\textrm{O}}]\textbf{\textrm{T}}_2)=\big \{H[\textbf{\textrm{O}}]^{n_2}_{k=1}a_{2,k}T_{2,k}:~a_{2,k}\in Z^0,~T_{2,k}\in \textbf{\textrm{T}}_2, H\in \textbf{\textrm{Q}}_2 \big \}
\end{equation} with $\sum ^{n_2}_{k=1}a_{2,k}\geq 1$.

\vskip 0.4cm

\textbf{Step-II-$t$.} Suppose the $(p_t,q_t)$-graph $H_t$ admits an \emph{edge-index total graph-coloring} $F_{t}:V(H_t)\cup E(H_t)\rightarrow \textbf{\textrm{T}}_t$, such that each edge $xy\in E(H_t)$ holds $F_{t}(x)=T_{t,i}$, $F_{t}(y)=T_{t,j}$ and $F_{t}(xy)=T_{t,\lambda}$ with the index $\lambda=i+j-k_t~(\bmod~M_t)$ for a preappointed \emph{zero-index} $T_{t,k_t}\in \textbf{\textrm{T}}_t$; refer to Definition \ref{defn:edge-index-total-graph-coloring}.

\textbf{Step-II-$t$.1.} From the edge-index total graph-coloring $F_{t}$ of the graph $H_t$, we get a Topcode-matrix $T_{code}(H_t,F_{t})= (X_{t},E_{t},Y_{t} )^T$ with the vertex vector $X_{t}=(F_{t}(x_{t,1})$, $F_{t}(x_{t,2}),\dots ,F_{t}(x_{t,q_t}))$, the edge vector $E_{t}=(F_{t}(x_{t,1}y_{t,1}),F_{t}(x_{t,2}y_{t,2}),\dots ,F_{t}(x_{t,q_t}y_{t,q_t}))$ and
the vertex vector $Y_{t}=(F_{t}(y_{t,1})$, $F_{t}(y_{t,2}),\dots ,F_{t}(y_{t,q_t}))$. For the simplicity of statement, we write
\begin{equation}\label{eqa:t-rank-graph-based-stringtt}
U(H_t)=X_{t}\cup E_{t}\cup Y_{t}=\{F_{t}(w_{t,r,1}),F_{t}(w_{t,r,2}),\cdots ,F_{t}(w_{t,r,3q_t})\}
\end{equation}
and we get a \emph{graph-based string}
\begin{equation}\label{eqa:graph-based-string-3q-tt}
S_{gb}(H_t,r)=F_{t}(w_{t,r,1})F_{t}(w_{t,r,2})\cdots F_{t}(w_{t,r,3q_t})
\end{equation} with $r\in [1,(3q_t)!]$, where $F_{t}(w_{t,r,1})F_{t}(w_{t,r,2})\cdots F_{t}(w_{t,r,3q_t})$ is a permutation of elements of the graph set $U(H_t)$ defined in Eq.(\ref{eqa:t-rank-graph-based-stringtt}).

\textbf{Step-II-$t$.2.} Since each graph $F_{t}(w_{t,r,j})=T_{t,(r,j)}\in \textbf{\textrm{T}}_t$ with $j\in [1,3q_t]$, and $T_{t,(r,j)}$ admits a a total coloring $f_{t,(r,j)}:V(T_{t,(r,j)})\cup E(T_{t,(r,j)})\rightarrow W_{t,(r,j)}$, where $W_{t,(r,j)}$ is a non-negative integer set. Let $q(t,r,j)=|E(T_{t,(r,j)})|$. So we have a Topcode-matrix
$$
T_{code}(T_{t,(r,j)},f_{t,(r,j)})= (X_{t,(r,j)},E_{t,(r,j)},Y_{t,(r,j)})^T
$$ with its own vertex vectors and edge vector as follows

$X_{t,(r,j)}=(f_{t,(r,j)}(u_{t,(r,j),1})$, $f_{t,(r,j)}(u_{t,(r,j),2}),\dots ,f_{t,(r,j)}(u_{t,(r,j),q(t,r,j)}))$,

$E_{t,(r,j)}=(f_{t,(r,j)}(u_{t,(r,j),1}v_{t,(r,j),1}),f_{t,(r,j)}(u_{t,(r,j),2}v_{t,(r,j),2}),\dots ,f_{t,(r,j)}(u_{t,(r,j),q(t,r,j)}v_{t,(r,j),q(t,r,j)}))$,

$Y_{t,(r,j)}=(f_{t,(r,j)}(v_{t,(r,j),1}),f_{t,(r,j)}(v_{t,(r,j),2}),\dots ,f_{t,(r,j)}(v_{t,(r,j),q(t,r,j)}))$.\\
We get a number-based string set
\begin{equation}\label{eqa:graph-based-string-base}
{
\begin{split}
U(T_{t,(r,j)})&=X_{t,(r,j)}\cup E_{t,(r,j)}\cup Y_{t,(r,j)}\\
&=\{f_{t,(r,j)}(z_{t,(r,j),1}),f_{t,(r,j)}(z_{t,(r,j),3}),\dots ,f_{t,(r,j)}(z_{t,(r,j),3q(t,r,j)})\}
\end{split}}
\end{equation} with $r\in [1,(3q_t)!]$ and $j\in [1,3q_t]$.

Each graph $F_{t}(w_{t,r,j})=T_{t,(r,j)}$ distributes us a topological number-based string
$$f_{t,(r,j)}(z_{t,(r,j),1})f_{t,(r,j)}(z_{t,(r,j),3})\cdots f_{t,(r,j)}(z_{t,(r,j),3q(t,r,j)})=S^{q(t,r,j)}_{(r,j)}
$$ then each graph-based string $S_{gb}(H_t,r)$ defined in Eq.(\ref{eqa:graph-based-string-3q-tt}) distributes us a \emph{$t$-rank topological number-based string}
\begin{equation}\label{eqa:complex-topological-number-based-string-tt}
s(H_t,r)=S^{q(t,r,1)q(t,r,2)\cdots q(t,r,3q_t)}_{(r,1)(r,2)\cdots (r,3q_t)}\leftarrow S_{gb}(H_t,r)=F_{t}(w_{t,r,1})F_{t}(w_{t,r,2})\cdots F_{t}(w_{t,r,3q_t})
\end{equation} with $r\in [1,(3q_t)!]$.

\vskip 0.4cm

\textbf{Step-II-$t$.3.} A $(p_{t+1},q_{t+1})$-graph $H_{t+1}$ is obtained as: For each edge $xy\in E(H_t)$, do a graph operation $O_{j}\in \textbf{\textrm{O}}$ to the graphs $F_{t}(x)$ and $F_{t}(xy)$ for producing a subgraph $F_{t}(x)O_{j}F_{t}(xy) $, and do a graph operation $O_{j}\in \textbf{\textrm{O}}$ to the graphs $F_{t}(y)$ and $F_{t}(xy)$ for producing a subgraph $F_{t}(y)O_{j}F_{t}(xy) $, then we get a subgraph (also graph-edge) $F_{t}(x)O_{j}F_{t}(xy)O_{j}F_{t}(y)$. Thereby, we obtain an \emph{operation graphic lattice}
\begin{equation}\label{eqa:edge-index-graphic-lattice33}
\textbf{\textrm{L}}(\textbf{\textrm{Q}}_{t+1}[\textbf{\textrm{O}}]\textbf{\textrm{T}}_{t+1})=\big \{H[\textbf{\textrm{O}}]^{n_{t+1}}_{k=1}a_{t+1,k}T_{t+1,k}:~a_{t+1,k}\in Z^0,~T_{t+1,k}\in \textbf{\textrm{T}}_{t+1}, H\in \textbf{\textrm{Q}}_{t+1}\big \}
\end{equation} with $\sum ^{n_{t+1}}_{k=1}a_{t+1,k}\geq 1$, and $\textbf{\textrm{Q}}_{s+1}=\textbf{\textrm{L}}(\textbf{\textrm{Q}}_{s}[\textbf{\textrm{O}}]\textbf{\textrm{T}}_s)$ for $s\geq 1$.

\subsection{Complex number-based strings}

By Eq.(\ref{eqa:complex-topological-number-based-string-tt}), we have a set of $t$-rank topological number-based strings as follows
\begin{equation}\label{eqa:555555}
T_{onbs}(H_t)=\left \{s(H_t,r)=S^{q(t,r,1)q(t,r,2)\cdots q(t,r,3q_t)}_{(r,1)(r,2)\cdots (r,3q_t)}:r\in [1,(3q_t)!]\right \}
\end{equation} defined on the $(p_{t},q_{t})$-graph $H_{t}$ admitting an edge-index total graph-coloring
$$F_{t}:V(H_t)\cup E(H_t)\rightarrow \textbf{\textrm{T}}_t
$$ refer to Definition \ref{defn:edge-index-total-graph-coloring}. We have an operation
\begin{equation}\label{eqa:topological-number-based-string-group}
s(H_t,i)[\oplus \ominus_k] ~s(H_t,j):=s(H_t,i)[\oplus ]s(H_t,j)[\ominus ]s(H_t,k)=s(H_t,\lambda)\in T_{onbs}(H_t)
\end{equation} with the index $\lambda=i+j-k~(\bmod~M)$ for a preappointed \emph{zero} $s(H_t,k)\in T_{onbs}(H_t)$. So, we get a \emph{topological number-based string every-zero group} $\{T_{onbs}(H_t);\oplus\ominus\}$ of order $M$.

Let $N_1=(3q_t)!$, we have a \emph{$1$-rank complex topological number-based string}
$$
s^1_{per}(m)=s(H_t,r_1)s(H_t,r_2)\cdots s(H_t,r_{N_1})
$$ that is a permutation $s(H_t,r_1),s(H_t,r_2),\cdots ,s(H_t,r_{N_1})$ of elements of the group $\{T_{onbs}(H_t);\oplus\ominus\}$, so there are $(N_1)!$ permutations, we put them into a set $T^1_{onbs}(H_t)=\{s^1_{per}(m):~m\in [1,(N_1)!]\}$. Then we have an operation
\begin{equation}\label{eqa:complex-number-based-string-group11}
s^1_{per}(i)[\oplus \ominus_k] ~ s^1_{per}(j):=s^1_{per}(i)[\oplus ]s^1_{per}(j)[\ominus ]s^1_{per}(k)=s^1_{per}(\lambda)\in T^1_{onbs}(H_t)
\end{equation} with the index $\lambda=i+j-k~(\bmod~(N_1)!)$ for a preappointed \emph{zero} $s^1_{per}(k)\in T^1_{onbs}(H_t)$, which has defined a \emph{topological number-based string every-zero group} $\{T^1_{onbs}(H_t);\oplus\ominus\}$ of order $(N_1)!$.

Let $N_2=(N_1)!$. The elements of the group $\left \{T^1_{onbs}(H_t);\oplus\ominus\right \}$ distributes us $(N_2)!$ $2$-rank complex topological number-based strings of the form $s^2_{per}(m)=s^1_{per}(m_1)s^1_{per}(m_2)\cdots s^1_{per}(m_{N_2})$, we get a topological number-based string every-zero group $\{T^2_{onbs}(H_t);\oplus\ominus\}$ based on the number-based string set $T^2_{onbs}(H_t)=\left \{s^2_{per}(m):~m\in [1,(N_2)!]\right \}
$. Go on in this way, we have topological number-based string every-zero groups $\left \{T^k_{onbs}(H_t);\oplus\ominus\right \}$ defined on the number-based string set
\begin{equation}\label{eqa:555555}
T^k_{onbs}(H_t)=\left \{s^k_{per}(m):~m\in [1,(N_k)!]\right \}, ~N_k=(N_{k-1})!, ~k\geq 2
\end{equation} where the set $T^k_{onbs}(H_t)$ contains \emph{$k$-rank complex topological number-based strings} of order $N_k$.

The computational complexity of TOTAL-graph-coloring algorithm-II can be get from Eq.(\ref{eqa:graph-based-string-3q-tt}) and Eq.(\ref{eqa:complex-topological-number-based-string-tt}), we omit it.

\section{Miscellaneous methods for generating strings}

\subsection{Number-based strings with parameters}

\begin{defn}\label{defn:parameterized-topcode-matrix}
\cite{Bing-Yao-arXiv:2207-03381} A \emph{parameterized Topcode-matrix} $P_{(k,d)}=(X_{k,d},E_{k,d},Y_{k,d})^T$ shown in Eq.(\ref{eqa:basic-formula-1}) is defined as

$X_{k,d}=(\alpha_1k+a_1d,\alpha_2k+a_2d,\dots ,\alpha_qk+a_qd)$, $Y_{k,d}=(\beta_1k+b_1d,\beta_2k+b_2d,\dots ,\beta_qk+b_qd)$, and

$E_{k,d}=(\gamma_1k+c_1d,\gamma_2k+c_2d,\dots ,\gamma_qk+c_qd)$\\
for two parameters $k,d\geq 1$. \qqed
\end{defn}

Moreover we can partition a parameterized Topcode-matrix defined in Definition \ref{defn:parameterized-topcode-matrix} into the linear combination of two Topcode-matrices as follows:
\begin{equation}\label{eqa:basic-formula-1}
{
\begin{split}
P_{(k,d)}&=\left(
\begin{array}{ccccc}
\alpha_1k+a_1d & \alpha_2k+a_2d & \cdots & \alpha_qk+a_qd\\
\gamma_1k+c_1d & \gamma_2k+c_2d & \cdots & \gamma_qk+c_qd\\
\beta_1k+b_1d & \beta_2k+b_2d & \cdots & \beta_qk+b_qd
\end{array}
\right)_{3\times q}\\
&=\left(
\begin{array}{ccccc}
\alpha_1k & \alpha_2k & \cdots & \alpha_qk\\
\gamma_1k & \gamma_2k & \cdots & \gamma_qk\\
\beta_1k & \beta_2k & \cdots & \beta_qk
\end{array}
\right)_{3\times q}+\left(
\begin{array}{ccccc}
a_1d & a_2d & \cdots & a_qd\\
c_1d & c_2d & \cdots & c_qd\\
b_1d & b_2d & \cdots & b_qd
\end{array}
\right)_{3\times q}\\
&=k\cdot I_{3\times q}+d\cdot C_{3\times q}
\end{split}}
\end{equation}
where $I_{3\times q}$ is called \emph{constant Topcode-matrix} defined as
\begin{equation}\label{eqa:basic-formula-2}
{
\begin{split}
I_{3\times q}=\left(
\begin{array}{ccccc}
\alpha_1 & \alpha_2 & \cdots & \alpha_q\\
\gamma_1 & \gamma_2 & \cdots & \gamma_q\\
\beta_1 & \beta_2 & \cdots & \beta_q
\end{array}
\right)_{3\times q}=(I_X,I_E,I_Y)^T
\end{split}}
\end{equation} with three vectors $I_X=(\alpha_1,\alpha_2,\dots ,\alpha_q)$, $I_E=(\gamma_1,\gamma_2,\dots ,\gamma_q)$ and $I_Y=(\beta_1,\beta_2,\dots ,\beta_q)$; and $C_{3\times q}$, called \emph{main Topcode-matrix}, is defined as
\begin{equation}\label{eqa:basic-formula-3}
{
\begin{split}
C_{3\times q}=\left(
\begin{array}{ccccc}
a_1 & a_2 & \cdots & a_q\\
c_1 & c_2 & \cdots & c_q\\
b_1 & b_2 & \cdots & b_q
\end{array}
\right)_{3\times q}=(C_X,C_E,C_Y)^T
\end{split}}
\end{equation} for three vectors $C_X=(a_1, a_2, \dots ,a_q)$, $C_E=(c_1, c_2, \dots ,c_q)$, and $C_Y=(b_1, b_2, \dots ,b_q)$. Also, the parameterized Topcode-matrix can be written as
\begin{equation}\label{eqa:basic-formula-4}
{
\begin{split}
P_{(k,d)}=&k\cdot (I_X,I_E,I_Y)^T+d\cdot (C_X,C_E,C_Y)^T\\
=&(k\cdot I_X+d\cdot C_X,k\cdot I_E+d\cdot C_E,k\cdot I_Y+d\cdot C_Y)^T\\
=&k\cdot I_{3\times q}+d\cdot C_{3\times q}.
\end{split}}
\end{equation}

\begin{example}\label{exa:8888888888}
In Eq.(\ref{eqa:basic-formula-1}), if we have an edge-difference constraint
$$(\gamma_ik+c_id)+|(\beta_ik+b_id)-(\alpha_ik+a_id)|=(\gamma_i+\beta_i-\alpha_i)k+(c_i+b_i-a_i)d=k\cdot A_{ed}+d\cdot B_{ed}
$$ for $i\in [1,q]$, we write the above fact as $I_E+|I_Y-I_X|:=A_{ed}$ and $C_E+|C_Y-C_X|:=B_{ed}$, where $A_{ed}$ and $B_{ed}$ are constants.

If there is an edge-magic constraint
$$
(\gamma_ik+c_id)+(\beta_ik+b_id)+(\alpha_ik+a_id=(\gamma_i+\beta_i+\alpha_i)k+(c_i+b_i+a_i)d)=k\cdot A_{em}+d\cdot B_{em}
$$ for $i\in [1,q]$, we write the above fact as $I_X+I_E+I_Y:=A_{em}$ and $C_X+C_E+C_Y:=B_{em}$ for constants $A_{em}$ and $B_{em}$.

If we have a graceful-difference constraint
$$
\big ||(\beta_ik+b_id)-(\alpha_ik+a_id)|-(\gamma_ik+c_id)\big |=(\beta_i-\alpha_i-\gamma_i)k+(b_i-a_i-c_i)d=k\cdot A_{gd}+d\cdot B_{gd}
$$ for $i\in [1,q]$, we write the above fact as $\big ||I_Y-I_X|-I_E\big |:=A_{gd}$ and $\big ||C_Y-C_X|-C_E\big |:=B_{gd}$ for constants $A_{gd}$ and $B_{gd}$.

If we have felicitous-difference constraint
$$
|(\beta_ik+b_id)+(\alpha_ik+a_id)-(\gamma_ik+c_id)|=(\beta_i+\alpha_i-\gamma_i)k+(b_i+a_i-c_i)d=k\cdot A_{fd}+d\cdot B_{fd}
$$ for $i\in [1,q]$, we write the above fact as $|I_Y+I_X-I_E|:=A_{fd}$ and $|C_Y+C_X-C_E|:=B_{fd}$ for constants $A_{fd}$ and $B_{fd}$.\qqed
\end{example}

If there is no confusion, we omit ``order $3\times q$'' in the following discussion, or add a sentence `` the Topcode-matrices $I$, $T_{code}(G)$ and $P_{(k,d)}(G)$ have the same order''. For \emph{bipartite graphs}, especially, we define the \emph{unite Topcode-matrix} as follows
\begin{equation}\label{eqa:unit-Topcode-matrix}
{
\begin{split}
I\,^0=\left(
\begin{array}{ccccc}
0 & 0 & \cdots & 0\\
1 & 1 & \cdots & 1\\
1 & 1 & \cdots & 1
\end{array}
\right)_{3\times q}=(X\,^0,~E\,^0,~Y\,^0)^T
\end{split}}
\end{equation} with two vertex-vectors $X\,^0=(0, 0, \dots ,0)_{1\times q}$ and $Y\,^0=(1, 1, \dots ,1)_{1\times q}$, and the edge-vector $E\,^0=(1, 1, \dots ,1)_{1\times q}$.

\begin{defn} \label{defn:bipartite-parameterized-topcode-matrix}
\cite{Bing-Yao-arXiv:2207-03381} Let $G$ be a bipartite $(p,q)$-graph with $V(G)=X\cup Y$ and $X\cap Y=\emptyset$, and let $k,d$ be integers. If $G$ admits a set-ordered $W$-constraint coloring $f$, that is $\max f(X)<\min f(Y)$, so we get a \emph{parameterized Topcode-matrix} defined by
\begin{equation}\label{eqa:definition-parameterized-topcode-matrix}
{
\begin{split}
P_{(k,d)}(G,F)&=k\cdot I\,^0+d\cdot T_{code}(G,f)\\
&=\left(
\begin{array}{rrrrr}
f(u_1)d & f(u_2)d & \cdots & f(u_q)d\\
k+f(u_1v_1)d & k+f(u_2v_2)d & \cdots & k+f(u_qv_q)d\\
k+f(v_1)d & k+f(v_2)d & \cdots & k+f(v_q)d
\end{array}
\right)
\end{split}}
\end{equation}
where three Topcode-matrices $I\,^0$, $T_{code}(G,f)$ and $P_{(k,d)}(G,F)$ have the same order $3\times q$, and $F$ is a \emph{$W$-constraint parameterized coloring} of $G$, as well as
\begin{equation}\label{eqa:set-type-topcode-matrix}
\centering
{
\begin{split}
T_{code}(G,f)= \left(
\begin{array}{ccccc}
f(u_1) & f(u_2) & \cdots & f(u_q)\\
f(u_1v_1) & f(u_2v_2) & \cdots & f(u_qv_q)\\
f(v_1) & f(v_2) & \cdots & f(v_q)
\end{array}
\right)
\end{split}}
\end{equation} hoding the $W$-constraint $f(u_kv_k)=W\langle f(u_k), f(v_k)\rangle $ for each edge $u_kv_k\in E(G)$ with $u_k\in X$ and $v_k\in Y$.\qqed
\end{defn}

\begin{rem}\label{rem:333333}
\textbf{The assignment Topcode-matrices.} By Definition \ref{defn:evaluated-topcode-matrix}, we have a \emph{assignment Topcode-matrix} $T_{code}(G,f):=P_{(k,d)}(G,F)$ defined in Eq.(\ref{eqa:definition-parameterized-topcode-matrix}), which converts a parameterized number-based string \begin{equation}\label{eqa:strings-plane-kd-coordinates}
s(k,d)=c_{1}(k,d)c_{2}(k,d)\cdots c_{3q}(k,d)
\end{equation} with longer bytes made by $P_{(k,d)}(G,F)$ defined in Definition \ref{defn:bipartite-parameterized-topcode-matrix} to a string $s=c_{1}c_{2}\cdots c_{3q}$ with shorter bytes made by $T_{code}(G,f)$.

\textbf{The fractional strings.} The parameterized Topcode-matrix $P_{(k,d)}(G,F)$ is useful in the discussion of fractional strings. The limitation $(k,d)\rightarrow (k_0,d_0)$ enables us to induce real-valued strings. For example, a \emph{fractional $(k_n,d_n)$-string} $s^*_n=c^*_{n,1}c^*_{n,2}\cdots c^*_{n,m}$ holds:

(i) there is at least one $c^*_{n,j}$ to be a positive fractional number;

(ii) $(k_n,d_n)\rightarrow (k_0,d_0)$ as $n\rightarrow \infty$;

(iii) there is a positive integer $M_n$ for each $n$ holding
$$
M_n[\bullet ]s^*_n=(M_n\cdot c^*_{n,1})(M_n\cdot c^*_{n,2})\cdots (M_n\cdot c^*_{n,m})=s_n
$$ such that $s_n$ is just a \emph{proper number-based string} with positive integer $M_n\cdot c^*_{n,j}$ for $j\in [1,m]$. In other words, the research of fractional strings can be translated into the investigation of proper number-based strings.\paralled
\end{rem}

\begin{defn} \label{defn:plane-coordinate-string-sequence}
$^*$ We call a parameterized number-based string $s(k,d)$ made by the parameterized Topcode-matrix $P_{(k,d)}(G,F)$ defined in Eq.(\ref{eqa:definition-parameterized-topcode-matrix}) \emph{plane-curve-attached string}. Let $pc(x,y)=0$ be a \emph{plane curve} defined on a domain $[\alpha,\beta]^{r}$ for $0\leq \alpha<\beta$. If there are positive integer points $(k_n,d_n)\in [\alpha,\beta]^{r}$ holding $pc(k_n,d_n)=0$ for integers $k_n,d_n\geq 0$ with $n\in [1,m]$, then we get a \emph{plane-curve-attached string sequence} $\{s(k_n,d_n)\}^m_{n=1}$ based on the plane curve $pc(x,y)=0$.\qqed
\end{defn}

\begin{thm}\label{thm:one-encryption-one-time}
$^*$ By Definition \ref{defn:plane-coordinate-string-sequence}, there are infinite plane-curve-attached string sequences based on a parameterized Topcode-matrix $P_{(k,d)}(G,F)$ defined in Eq.(\ref{eqa:definition-parameterized-topcode-matrix}) and infinite plane curves, which provides the theoretical basis for the one-encryption one-time.
\end{thm}

\begin{rem}\label{rem:333333}
For a \emph{public-key graph} $G$ admitting a $W$-constraint parameterized coloring $F$, we use this \emph{parameter-colored graph} $G$ and a plane curve $pc(x,y)=0$ to form a \emph{private-key graph} in a \emph{topological signature authentication}, the private-key graph is denoted as $H=\langle G,F,pc(x,y)=0\rangle$. Since there are infinite real-valued functions and there are infinite integer points in a plane curve, so we can get infinite number-based strings to encrypt or to decrypt a file consisted of many segments in the method of \emph{asymmetric topology cryptography}, and these number-based strings are random since the plane curve are taken randomly in the private-key graphs like as $H=\langle G,F,pc(x,y)=0\rangle$.

If the plane curve $pc(x,y)=0$ is an elliptic curve: $y^2=x^3+ ax^2+bx+c$ defined on a finite field $[0,A_{prime}]$ with a prime number $A_{prime}$, then deciphering the plane-curve-attached string sequence $\{s(k_n,d_n)\}^m_{n=1}$ is even more difficult, even impossible.\paralled
\end{rem}

\subsection{$(abc)$-linear colorings}

\begin{thm}\label{thm:abc-linear-magic-string-colorings}
\cite{Bing-Yao-arXiv:2207-03381} If a bipartite and connected $(p,q)$-graph $G$ admits a \emph{set-ordered} $W$-constraint total coloring $f$ defined in Definition \ref{defn:basic-W-type-labelings}, Definition \ref{defn:kd-w-type-colorings} and Definition \ref{defn:odd-edge-W-type-total-labelings-definition}. For positive integers $a,b,c$ and non-negative integer $\lambda$, if there are

(i) the $(abc)$-edge-magic constraint $af(x)+bf(y)+cf(xy)=\lambda$ for each edge $xy\in E(G)$; or

(ii) the $(abc)$-felicitous-difference constraint $|af(x)+bf(y)-cf(xy)|=\lambda$ for each edge $xy\in E(G)$; or

(iii) the $(abc)$-graceful-difference constraint $\big ||af(x)-bf(y)|-cf(xy)\big |=\lambda$ for each edge $xy\in E(G)$; or

(iv) the $(abc)$-edge-difference constraint $cf(xy)+|af(x)-bf(y)|=c$ for each edge $xy\in E(G)$.\\
Then this bipartite and connected graph $G$ admits a $W$-constraint $(k,d)$-total coloring $F$ holding each one of the above four $(abc)$-constraints.
\end{thm}
\begin{proof} Let $(X,Y)$ be the vertex bipartition of a bipartite $(p,q)$-graph $G$, so a set-ordered $W$-constraint total coloring $f$ of $G$ holds $\max f(X)<\min f(Y)$ true. There is a $W$-constraint $(k,d)$-total coloring $F$ defined by setting $F(x)=f(x)d$, $F(y)=k+f(y)d$ and $F(xy)=k+f(xy)d$ made by the parameterized Topcode-matrix $P_{(k,d)}(G,F)$ defined in Eq.(\ref{eqa:definition-parameterized-topcode-matrix}), so we have:

(i) If the $(abc)$-edge-magic constraint $af(x)+bf(y)+cf(xy)=\lambda$ holds true for each edge $xy\in E(G)$, then the $(abc)$-edge-magic constraint
\begin{equation}\label{eqa:555555}
{
\begin{split}
aF(x)+bF(y)+cF(xy)&=af(x)d+b[k+f(y)d]+c[k+f(xy)d]\\
&=[af(x)+bf(y)+cf(xy)]d+(b +c)k\\
&=d \cdot \lambda+(b +c)k
\end{split}}
\end{equation} holds true for each edge $xy\in E(G)$.

(ii) If the $(abc)$-felicitous-difference constraint $|af(x)+bf(y)-cf(xy)|=\lambda$ holds true for each edge $xy\in E(G)$, then we get the $(abc)$-felicitous-difference constraint
\begin{equation}\label{eqa:555555}
{
\begin{split}
\big |aF(x)+bF(y)-cF(xy)\big |&=\big |af(x)d+b[k+f(y)d]-c[k+f(xy)d]\big |\\
&=\big |[af(x)+bf(y)-cf(xy)]d+(b -c)k\big |\\
&\leq \big |af(x)+bf(y)-cf(xy)\big |d+\big |b -c\big |k\\
&=d \cdot \lambda+|b-c|k
\end{split}}
\end{equation} for each edge $xy\in E(G)$.

(iii) If the $(abc)$-graceful-difference constraint $\big ||af(x)-bf(y)|-cf(xy)\big |=\lambda$ holds true for each edge $xy\in E(G)$, thus, the $(abc)$-graceful-difference constraint
\begin{equation}\label{eqa:555555}
{
\begin{split}
\big ||aF(x)-bF(y)|-cF(xy)\big |&=\big ||af(x)d-b[k+f(y)d]|-c[k+f(xy)d]\big |\\
&\leq\big |bf(y)-af(x)-cf(xy)\big |d+|b-c|k\\
&=d \cdot \lambda+|b -c|k
\end{split}}
\end{equation} or
\begin{equation}\label{eqa:555555}
{
\begin{split}
\big ||aF(x)-bF(y)|-cF(xy)\big |&=\big ||af(x)d-b[k+f(y)d]|-c[k+f(xy)d]\big |\\
&\leq\big |af(x)-bf(y)-cf(xy)\big |d+|b-c|k\\
&=d \cdot \lambda+|b -c|k
\end{split}}
\end{equation} holds true for each edge $xy\in E(G)$.

(iv) If the $(abc)$-edge-difference constraint $cf(xy)+|af(x)-bf(y)|=\lambda$ holds true for each edge $xy\in E(G)$, then we have the $(abc)$-graceful-difference constraint
\begin{equation}\label{eqa:555555}
{
\begin{split}
cF(xy)+\big |aF(x)-bF(y)\big |&=c[k+f(xy)d]+|af(x)d-b[k+f(y)d] |\\
&\leq [cf(xy)+|af(x)-bf(y)|]d+(b+c)k\\
&=d \cdot \lambda+(b+c)k
\end{split}}
\end{equation} for each edge $xy\in E(G)$.

Moreover, we have other two relations as follows:
$${
\begin{split}
|aF(x)-bF(y)|&=|b[k+f(y)d]-af(x)d|\leq |bf(y)-af(x)|d+bk\\
&=c[k+f(xy)d]+(b-c)k\\
&=cF(xy)+(b-c)k
\end{split}}
$$ if $cf(xy)=|bf(y)-af(x)|$ for each edge $xy\in E(G)$; and moreover
$${
\begin{split}
aF(x)+bF(y)&=af(x)d+b[k+f(y)d]=[bf(y)+af(x)]d+bk\\
&=c[k+f(xy)d]+(b-c)k\\
&=cF(xy)+(b-c)k
\end{split}}
$$ if $cf(xy)=af(x)+bf(y)$ for each edge $xy\in E(G)$.

The proof of the theorem is complete.
\end{proof}

\begin{rem}\label{rem:333333}
About Theorem \ref{thm:abc-linear-magic-string-colorings}, trees admit some \emph{set-ordered} $W$-constraint total colorings defined in Definition \ref{defn:basic-W-type-labelings}, Definition \ref{defn:kd-w-type-colorings} and Definition \ref{defn:odd-edge-W-type-total-labelings-definition}; refer to Theorem \ref{thm:10-k-d-total-coloringss}. In Theorem \ref{thm:abc-linear-magic-string-colorings}, the parameterized Topcode-matrix $P_{(a,b,c)}(G,F)$ of the bipartite and connected $(p,q)$-graph $G$ admitting the $W$-constraint $(k,d)$-total coloring $F$ holding each one of four $(abc)$-constraints can provide us more complex number-based strings for real application.

The ABC-conjecture (also, Oesterl\'{e}-Masser conjecture, 1985) involves the equation $a+b=c$ and the relationship between prime numbers. Proving or disproving the ABC-conjecture could impact: Diophantine (polynomial) math problems including Tijdeman's theorem, Vojta's conjecture, Erd\"{o}s-Woods conjecture, Fermat's last theorem, Wieferich prime and Roth's theorem, and so on \cite{Cami-Rosso2017Abc-conjecture}.\paralled
\end{rem}

\begin{defn} \label{defn:n-dimensioncurve-attached-colorings}
$^*$ Suppose that $f$ is a $W$-constraint total coloring of a $(p,q)$-graph $G$, so we have its own Topcode-matrix $T_{code}(G,f)$. We define new colorings by the $W$-constraint total coloring $f$ as follows:

\begin{asparaenum}[\textbf{\textrm{Pacolo}}-1.]
\item Setting a new total coloring $\alpha(w)=f(w)+x+y+z$ for each element $w\in V(G)\cup E(G)$, where $(x,y,z)$ is an integer point of a space curve $F(x,y,z)=0$ with integers $x,y,z\geq 0$.
\item Letting a new total coloring $\beta(w)=s(x,y,z)f(w)+\lambda$ for each element $w\in V(G)\cup E(G)$, where
\begin{equation}\label{eqa:555555}
s(x,y,z)=c_1(x,y,z)c_2(x,y,z)\cdots c_m(x,y,z)
\end{equation} is a parameterized number-based string, and $\lambda=a_1a_2\cdots a_n$ is a number-based string, and $(x,y,z)$ is an integer point of a space curve $F(x,y,z)=0$ with integers $x,y,z\geq 0$.
\item Defining a new total coloring $\gamma(w)=s(x_1,x_2,\dots ,x_n)f(w)+\mu$ for $w\in V(G)\cup E(G)$, where
\begin{equation}\label{eqa:555555}
s(x_1,x_2,\dots ,x_n)=c_1(x_1,x_2,\dots ,x_n)c_2(x_1,x_2,\dots ,x_n)\cdots c_m(x_1,x_2,\dots ,x_n)
\end{equation} is a parameterized number-based string, and $\mu=b_1b_2\cdots b_s$ is a number-based string, as well as $(x_1,x_2$, $\dots $, $x_n)$ is an integer point of a $n$-dimension curve $F(x_1,x_2,\dots ,x_n)=0$ with integers $x_1,x_2,\dots ,x_n\geq 0$.
\end{asparaenum}

Since the total colorings $\alpha,\beta$ and $\gamma$ defined above are based one curves $F(x,y,z)=0$ and $F(x_1,x_2,\dots ,x_n)=0$, as well as the $W$-constraint total coloring $f$ of the $(p,q)$-graph $G$, we call them \emph{$n$-dimension curve-attached $W$-constraint total colorings}, and the number-based strings induced by the parameterized $W$-constraint Topcode-matrices $T_{code}(G,\alpha)$, $T_{code}(G,\beta)$ and $T_{code}(G,\gamma)$ are called \emph{$n$-dimension curve-attached $W$-constraint number-based strings}. \qqed
\end{defn}

\vskip 0.4cm

\begin{example}\label{exa:8888888888}
As the application of Theorem \ref{thm:abc-linear-magic-string-colorings}, we let $(x_k,y_k,z_k)$ be an integer point of a space curve $F(x,y,z)=0$ with non-negative integers $x_k,y_k,z_k\geq 0$ for $k\in [1,M]$. By Definition \ref{defn:n-dimensioncurve-attached-colorings}, we get a $3$-dimension curve-attached $W$-constraint total colorings $\beta_k(w)=s(x_k,y_k,z_k)f(w)+\lambda_k$ for each element $w\in V(G)\cup E(G)$ with
\begin{equation}\label{eqa:555555}
s(x_k,y_k,z_k)=c_1(x_k,y_k,z_k)c_2(x_k,y_k,z_k)\cdots c_m(x_k,y_k,z_k)
\end{equation} and $\lambda_k=a_{k,1}a_{k,2}\cdots a_{k,n_k}$ for $k\in [1,M]$.

The parameterized Topcode-matrix $T_{code}(G,\beta_k)$ with $k\in [1,M]$ induces $M\cdot (3q)!$ number-based strings of form $\beta_k(w_{j_1})\beta_k(w_{j_2})\cdots \beta_k(w_{j_{3q}})$, where $\beta_k(w_{j_1}),\beta_k(w_{j_2}),\dots ,\beta_k(w_{j_{3q}})$ is a permutation of the elements of the parameterized Topcode-matrix $T_{code}(G,\beta_k)$.

Then we have:

(i) The edge-magic constraint
\begin{equation}\label{eqa:555555}
{
\begin{split}
\beta_k(u)+\beta_k(uv)+\beta_k(v)=&s(x_k,y_k,z_k)\big [f(u)+f(uv)+f(v)\big ]+3\lambda_k\\
=&s(x_k,y_k,z_k)\mu _1+3\lambda_k
\end{split}}
\end{equation} if the edge-magic constraint $f(u)+f(uv)+f(v)=\mu_1$ for each edge $uv\in E(G)$ holds true.

(ii) The edge-difference constraint
\begin{equation}\label{eqa:555555}
{
\begin{split}
\beta_k(uv)+|\beta_k(u)-\beta_k(v)|=&s(x_k,y_k,z_k)\big [f(uv)+|f(u)-f(v)|\big ]+\lambda_k\\
=&s(x_k,y_k,z_k)\mu_2 +\lambda_k
\end{split}}
\end{equation} if the edge-difference constraint $f(uv)+|f(u)-f(v)|=\mu_2$ for each edge $uv\in E(G)$ holds true.

(iii) The graceful-difference constraint
\begin{equation}\label{eqa:555555}
{
\begin{split}
\big ||\beta_k(u)-\beta_k(v)|-\beta_k(uv)\big |\leq &s(x_k,y_k,z_k)\big ||f(u)-f(v)|-f(uv)\big |+\lambda_k\\
=&s(x_k,y_k,z_k)\mu_3 +\lambda_k
\end{split}}
\end{equation} if the graceful-difference constraint $\big ||f(u)-f(v)|-f(uv)\big |=\mu_3$ for each edge $uv\in E(G)$ holds true.

(iv) The felicitous-difference constraint
\begin{equation}\label{eqa:555555}
{
\begin{split}
|\beta_k(u)+\beta_k(v)-\beta_k(uv)|\leq &s(x_k,y_k,z_k)\big |f(u)+f(v)-f(uv)\big |+\lambda_k\\
=&s(x_k,y_k,z_k)\mu_4 +\lambda_k
\end{split}}
\end{equation} if the felicitous-difference constraint $\big |f(u)+f(v)-f(uv)\big |=\mu_4$ for each edge $uv\in E(G)$ holds true.\qqed
\end{example}

\begin{thm}\label{thm:one-encryption-one-time11}
Let $f$ be a $W$-constraint total coloring of a $(p,q)$-graph $G$, and $F(x_1,x_2,\dots ,x_n)=0$ be a $n$-dimension curve. Then there are infinite $n$-dimension curve-attached $W$-constraint total colorings $$\gamma(w)=s(x_1,x_2,\dots ,x_n)f(w)+\mu,~w\in V(G)\cup E(G)
$$ defined in Definition \ref{defn:n-dimensioncurve-attached-colorings}, where each $x_i$ in the string $s(x_1,x_2,\dots ,x_n)$ for an integer point $(x_1,x_2,\dots $, $x_n)$ of the $n$-dimension curve is a non-negative integer for $i\in [1,n]$.
\end{thm}

\subsection{Set-based strings}

\begin{defn} \label{defn:together-definition-set-labelings}
\cite{Yao-Sun-Zhang-Li-Yan-Zhang-Wang-ITOEC-2017} Let a $(p,q)$-graph $G$ admit a mapping $F: X\rightarrow S$, where $S$ is a subset of subsets of the power set $[0,pq]^2$, and let $Rs(m)=\{c_1,c_2,\dots, c_m\}$ be a given constraint set. There are the following constraints:
\begin{asparaenum}[(a)]
\item \label{vertex-set} $X=V(G)$.
\item \label{edge-set} $X=E(G)$.
\item \label{total-set} $X=V(G)\cup E(G)$.
\item \label{adjacent-vertex-labeling} $F(u)\not =F(v)$ if $uv\in E(G)$ (it may happen $F(u)\cap F(v)\neq \emptyset$).
\item \label{adjacent-edge-labeling} $F(uv)\not =F(uw)$ for each pair of adjacent edges $uv$ and $uw$ of $G$.
\item \label{vertex-labeling} $F(x)\not =F(y)$ for distinct vertices $x$ and $y$ of $G$.
\item \label{edge-labeling} $F(xy)\not =F(uv)$ for different edges $uv$ and $xy$ of $G$.
\item \label{edge-induced} A mapping $F\,': E(G)\rightarrow S$ is induced by $F$ subject to $Rs(m)$, that is, each edge $uv\in E(G)$ is labelled by a number of the set $F\,'(uv)$ such that each $c\in F\,'(uv)$ is generated by some $a\in F(u)$, $b\in F(v)$ and one constraint or more constraints of $Rs(m)$.
\end{asparaenum}
\textbf{Then }

(1-1) $F$ is called a \emph{vertex-distinguishing set-labeling} of $G$ if the constraints (\ref{vertex-set}) and (\ref{vertex-labeling}) hold true.

(1-2) $F$ is called an \emph{edge-distinguishing set-labeling} of $G$ if the constraints (\ref{edge-set}) and (\ref{edge-labeling}) hold true.

(1-3) $F$ is called a \emph{all-distinguishing set-coloring} of $G$ if (\ref{total-set}), (\ref{vertex-labeling}) and (\ref{edge-labeling}) hold true.

(1-4) $F\,'$ is called an \emph{induced edge-distinguishing edge set-labeling} of $G$ if both constraints (\ref{edge-labeling}) and (\ref{edge-induced}) hold true.

(1-5) $(F,F\,')$ is called an \emph{induced all-distinguishing set-coloring} subject to $Rs(m)$ if the constraints (\ref{vertex-set}), (\ref{vertex-labeling}), (\ref{edge-labeling}) and (\ref{edge-induced}) hold true.

(2-1) $F$ is called a \emph{proper set-labeling} of $G$ if it satisfies the constraints (\ref{vertex-set}) and (\ref{adjacent-vertex-labeling}) simultaneously.

(2-2) $F$ is called a \emph{proper edge set-labeling} of $G$ if it satisfies the constraints (\ref{edge-set}) and (\ref{adjacent-edge-labeling}) simultaneously.

(2-3) $F$ is called a \emph{proper total set-coloring} of $G$ if it satisfies the constraints (\ref{total-set}), (\ref{adjacent-vertex-labeling}) and (\ref{adjacent-edge-labeling}) simultaneously.

(2-4) $(F,F\,')$ is called an \emph{induced proper total set-coloring} subject to $Rs(m)$ if the constraints (\ref{vertex-set}), (\ref{adjacent-vertex-labeling}), (\ref{adjacent-edge-labeling}) and (\ref{edge-induced}) hold true simultaneously.

(3-1) $F$ is called a \emph{pseudo-vertex set-labeling} of $G$ if it holds the constraints (\ref{vertex-set}), but not (\ref{adjacent-vertex-labeling}).

(3-2) $F$ is called a \emph{pseudo-edge set-labeling} of $G$ if it holds the constraints (\ref{edge-set}), but not (\ref{adjacent-edge-labeling}).

(3-3) $F$ is called a \emph{pseudo-total set-coloring} of $G$ if it holds the constraint (\ref{total-set}), but not (\ref{adjacent-vertex-labeling}), or not (\ref{adjacent-edge-labeling}), or not both (\ref{adjacent-vertex-labeling}) and (\ref{adjacent-edge-labeling}).\qqed
\end{defn}

\begin{defn}\label{defn:parameters-set-labelings-colorings}
\cite{Yao-Sun-Zhang-Li-Yan-Zhang-Wang-ITOEC-2017} We have the following set-labelings/set-colorings:

(i) There are two parameters $$v_s=\min \{|F(x)|:x\in V(G)\},~v_l=\max \{|F(y)|:y\in V(G)\}
$$ by Definition \ref{defn:together-definition-set-labelings}. Then the set-labeling/set-coloring $F$ is called an \emph{$\alpha$-uniformly vertex set-labeling/set-coloring} of $G$ if $v_s=v_l=\alpha$. As $\alpha=1$, $(F,F\,')$ is a popular labeling of graph theory.

(ii) There are two parameters $$e_s=\min \{|F\,'(uv)|:~uv\in E(G)\},~e_l=\max \{|F\,'(xy)|:~xy\in E(G)\}
$$ from Definition \ref{defn:together-definition-set-labelings}. Then the edge set-labeling/set-coloring $F\,'$ is called a \emph{$\beta$-uniformly edge set-labeling/set-coloring} of $G$ if $e_s=e_l=\beta$.

(iii) Definition \ref{defn:together-definition-set-labelings} induces two parameters $$t_s=\min \{|F(x)|:x\in V(G)\cup E(G)\},~t_l=\max \{|F(y)|:~y\in V(G)\cup E(G)\}
$$ If $k=t_s=t_l$, then the total set-labeling/set-coloring $F$ is called \emph{$k$-uniformly total set-labeling/set-coloring}.\qqed
\end{defn}

\begin{rem}\label{rem:333333}
The colorings defined in Definition \ref{defn:together-definition-set-labelings} and Definition \ref{defn:parameters-set-labelings-colorings} form set-type Topcode-matrices, which can induce more complex number-based strings.\paralled
\end{rem}

\subsubsection{Distinguishing set-colorings}
\begin{defn} \label{defn:v-e-ve-distinguishing-set-coloringss}
$^*$ Let $H$ be a $(p,q)$-graph, and let $Q$ be a set of subsets with each subset $e\in Q$ holding $e\neq \emptyset$ true. There is a \emph{set-coloring} $\theta: S\rightarrow Q$ for $S\subseteq V(H)\cup E(H)$. There are the following constraints:
\begin{asparaenum}[\textbf{\textrm{Dint}}-1.]
\item \label{intersect:vertex} $S=V(H)$.
\item \label{intersect:edge} $S=E(H)$.
\item \label{intersect:total} $S=V(H)\cup E(H)$.

--- \emph{local distinguishing}

\item \label{intersect:proper-vertex} $\theta(u)\neq \theta(v)$ for each edge $uv\in E(H)$.
\item \label{intersect:proper-edges} $\theta(uv)\neq \theta(uw)$ for $v,w\in N_{ei}(u)$ and each vertex $u\in V(H)$.
\item \label{intersect:incident-edge-vertex} $F(u)\neq F(uv)$ and $F(v)\neq F(uv)$ for each edge $uv\in E(G)$.

--- \emph{local intersected}

\item \label{intersect:only-join-edge} $\theta(u)\cap \theta(v)\neq \emptyset $ for each edge $uv\in E(H)$.
\item \label{intersect:join-edge-r-rank} $\theta(u)\cap \theta(v)\subseteq \theta(uv)$ and $|\theta(u)\cap \theta(v)|\geq r\geq 1$ for each edge $uv\in E(H)$.
\item \label{intersect:join-edge-vertex} $\theta(uv)\cap \theta(u)\neq \emptyset$ and $\theta(uv)\cap \theta(v)\neq \emptyset$ for each edge $uv\in E(H)$.
\item \label{intersect:join-adjacent-edges} $\theta(uv)\cap \theta(uw)\neq \emptyset$ for $v,w\in N_{ei}(u)$ and each vertex $u\in V(H)$.

--- \emph{v-adjacent distinguishing}

\item \label{intersect:adjacent-vertex-union-dis} $\bigcup _{v\in N_{ei}(u)}\theta(v)\neq \bigcup _{z\in N_{ei}(w)}\theta(z)$ for each edge $uw\in V(H)$.
\item \label{intersect:adjacent-vertex-intersect-dis} $\bigcap _{v\in N_{ei}(u)}\theta(v)\neq \bigcap _{z\in N_{ei}(w)}\theta(z)$ for each edge $uw\in V(H)$.
\item \label{intersect:adjacent-all-vertex-union-dis} $\theta(u)\cup \big [\bigcup _{v\in N_{ei}(u)}\theta(v)\big ]\neq \theta(w)\cup \big [\bigcup _{z\in N_{ei}(w)}\theta(z)\big ]$ for each edge $uw\in V(H)$.
\item \label{intersect:adjacent-all-vertex-intersect-dis} $\theta(u)\cap \big [\bigcap _{v\in N_{ei}(u)}\theta(v)\big ]\neq \theta(w)\cap \big [\bigcap _{z\in N_{ei}(w)}\theta(z)\big ]$ for each edge $uw\in V(H)$.

--- \emph{e-adjacent distinguishing}

\item \label{intersect:adjacent-edges-union-dis} $\bigcup _{v\in N_{ei}(u)}\theta(uv)\neq \bigcup _{z\in N_{ei}(w)}\theta(wz)$ for each edge $uw\in V(H)$.
\item \label{intersect:adjacent-edges-intersect-dis} $\bigcap _{v\in N_{ei}(u)}\theta(uv)\neq \bigcap _{z\in N_{ei}(w)}\theta(wz)$ for each edge $uw\in V(H)$.

--- \emph{ve-adjacent distinguishing}

\item \label{intersect:adjacent-vertex-edges-union-dis} $\theta(u)\cup \big [\bigcup _{v\in N_{ei}(u)}\theta(uv)\big ]\neq \theta(w)\cup \big [\bigcup _{z\in N_{ei}(w)}\theta(wz)\big ]$ for each edge $uw\in V(H)$.
\item \label{intersect:adjacent-vertex-edges-intersect-dis} $\theta(u)\cap \big [\bigcap _{v\in N_{ei}(u)}\theta(uv)\big ]\neq \theta(w)\cap \big [\bigcap _{z\in N_{ei}(w)}\theta(wz)\big ]$ for each edge $uw\in V(H)$.
\end{asparaenum}
\textbf{We call $\theta$}

\noindent --- \emph{traditional}

\begin{asparaenum}[\textbf{\textrm{Tingco}}-1.]
\item a \emph{proper set-coloring} if the constraints Dint-\ref{intersect:vertex} and Dint-\ref{intersect:proper-vertex} hold true.
\item a \emph{proper edge set-coloring} if the constraints Dint-\ref{intersect:edge} and Dint-\ref{intersect:proper-edges} hold true.
\item a \emph{proper total set-coloring} if the constraints Dint-\ref{intersect:total}, Dint-\ref{intersect:proper-vertex}, Dint-\ref{intersect:proper-edges} and Dint-\ref{intersect:incident-edge-vertex} hold true.

--- \emph{intersected}

\item a \emph{$v$-intersected proper set-coloring} if the constraints Dint-\ref{intersect:vertex}, Dint-\ref{intersect:proper-vertex} and Dint-\ref{intersect:only-join-edge} hold true.
\item an \emph{$r$-rank $e$-intersected proper total set-coloring} if the constraints Dint-\ref{intersect:total}, Dint-\ref{intersect:proper-vertex}, Dint-\ref{intersect:proper-edges}, Dint-\ref{intersect:incident-edge-vertex} and Dint-\ref{intersect:join-edge-r-rank} hold true.
\item an \emph{$(e,ve)$-intersected proper total set-coloring} if the constraints Dint-\ref{intersect:total}, Dint-\ref{intersect:proper-vertex}, Dint-\ref{intersect:proper-edges}, Dint-\ref{intersect:incident-edge-vertex}, Dint-\ref{intersect:join-edge-r-rank} and Dint-\ref{intersect:join-edge-vertex} hold true.
\item an \emph{$(e,ve,ee)$-intersected proper total set-coloring} if the constraints Dint-\ref{intersect:total}, Dint-\ref{intersect:proper-vertex}, Dint-\ref{intersect:proper-edges}, Dint-\ref{intersect:incident-edge-vertex}, Dint-\ref{intersect:join-edge-r-rank}, Dint-\ref{intersect:join-edge-vertex} and Dint-\ref{intersect:join-adjacent-edges} hold true.

--- \emph{v-distinguishing}

\item a \emph{$v$-union adjacent-$v$ distinguishing proper set-coloring} if the constraints Dint-\ref{intersect:vertex}, Dint-\ref{intersect:proper-vertex} and Dint-\ref{intersect:adjacent-vertex-union-dis} hold true.
\item a \emph{$v$-intersected adjacent-$v$ distinguishing proper set-coloring} if the constraints Dint-\ref{intersect:vertex}, Dint-\ref{intersect:proper-vertex} and Dint-\ref{intersect:adjacent-vertex-intersect-dis} hold true.
\item a \emph{$[v]$-union adjacent-$v$ distinguishing proper set-coloring} if the constraints Dint-\ref{intersect:vertex}, Dint-\ref{intersect:proper-vertex} and Dint-\ref{intersect:adjacent-all-vertex-union-dis} hold true.
\item a \emph{$[v]$-intersected adjacent-$v$ distinguishing proper set-coloring} if the constraints Dint-\ref{intersect:vertex}, Dint-\ref{intersect:proper-vertex} and Dint-\ref{intersect:adjacent-all-vertex-intersect-dis} hold true.

--- \emph{edge, e-distinguishing}

\item an \emph{ $ee$-intersected set-coloring} if the constraints Dint-\ref{intersect:edge}, Dint-\ref{intersect:proper-edges} and Dint-\ref{intersect:join-adjacent-edges} hold true.
\item an \emph{$e$-union adjacent-$v$ distinguishing set-coloring} if the constraints Dint-\ref{intersect:edge}, Dint-\ref{intersect:proper-edges} and Dint-\ref{intersect:adjacent-edges-union-dis} hold true.
\item an \emph{$e$-intersected adjacent-$v$ distinguishing set-coloring} if the constraints Dint-\ref{intersect:edge}, Dint-\ref{intersect:proper-edges} and Dint-\ref{intersect:adjacent-edges-intersect-dis} hold true.

--- \emph{total distinguishing}

\item a \emph{$v$-union adjacent-$v$ distinguishing proper total set-coloring} if the constraints Dint-\ref{intersect:total}, Dint-\ref{intersect:proper-vertex}, Dint-\ref{intersect:proper-edges}, Dint-\ref{intersect:incident-edge-vertex} and Dint-\ref{intersect:adjacent-vertex-union-dis} hold true.
\item a \emph{$v$-intersected adjacent-$v$ distinguishing proper total set-coloring} if the constraints Dint-\ref{intersect:total}, Dint-\ref{intersect:proper-vertex}, Dint-\ref{intersect:proper-edges}, Dint-\ref{intersect:incident-edge-vertex} and Dint-\ref{intersect:adjacent-vertex-intersect-dis} hold true.
\item an \emph{$e$-union adjacent-$v$ distinguishing proper total set-coloring} if the constraints Dint-\ref{intersect:total}, Dint-\ref{intersect:proper-vertex}, Dint-\ref{intersect:proper-edges}, Dint-\ref{intersect:incident-edge-vertex} and Dint-\ref{intersect:adjacent-edges-union-dis} hold true.
\item an \emph{$e$-intersected adjacent-$v$ distinguishing proper total set-coloring} if the constraints Dint-\ref{intersect:total}, Dint-\ref{intersect:proper-vertex}, Dint-\ref{intersect:proper-edges}, Dint-\ref{intersect:incident-edge-vertex} and Dint-\ref{intersect:adjacent-edges-intersect-dis} hold true.

--- \emph{closed total distinguishing}

\item a \emph{$[v]$-union adjacent-$v$ distinguishing proper total set-coloring} if the constraints Dint-\ref{intersect:total}, Dint-\ref{intersect:proper-vertex}, Dint-\ref{intersect:proper-edges}, Dint-\ref{intersect:incident-edge-vertex} and Dint-\ref{intersect:adjacent-all-vertex-union-dis} hold true.
\item a \emph{$[v]$-intersected adjacent-$v$ distinguishing proper total set-coloring} if the constraints Dint-\ref{intersect:total}, Dint-\ref{intersect:proper-vertex}, Dint-\ref{intersect:proper-edges}, Dint-\ref{intersect:incident-edge-vertex} and Dint-\ref{intersect:adjacent-all-vertex-intersect-dis} hold true.

\item a \emph{$[ve]$-union adjacent-$v$ distinguishing proper total set-coloring} if the constraints Dint-\ref{intersect:total}, Dint-\ref{intersect:proper-vertex}, Dint-\ref{intersect:proper-edges}, Dint-\ref{intersect:incident-edge-vertex} and Dint-\ref{intersect:adjacent-vertex-edges-union-dis} hold true.
\item a \emph{$[ve]$-intersected adjacent-$v$ distinguishing proper total set-coloring} if the constraints Dint-\ref{intersect:total}, Dint-\ref{intersect:proper-vertex}, Dint-\ref{intersect:proper-edges}, Dint-\ref{intersect:incident-edge-vertex} and Dint-\ref{intersect:adjacent-vertex-edges-intersect-dis} hold true.

--- \emph{multiple-adjacent distinguishing}

\item an \emph{$(e,ve,ee)$-intersected $v$-union adjacent-$v$ distinguishing proper total set-coloring} if the constraints Dint-\ref{intersect:total}, Dint-\ref{intersect:proper-vertex}, Dint-\ref{intersect:proper-edges}, Dint-\ref{intersect:incident-edge-vertex}, Dint-\ref{intersect:join-edge-r-rank}, Dint-\ref{intersect:join-edge-vertex}, Dint-\ref{intersect:join-adjacent-edges} and Dint-\ref{intersect:adjacent-vertex-union-dis} hold true.
\item an \emph{$(e,ve,ee)$-$v$-intersected adjacent-$v$ distinguishing proper total set-coloring} if the constraints Dint-\ref{intersect:total}, Dint-\ref{intersect:proper-vertex}, Dint-\ref{intersect:proper-edges}, Dint-\ref{intersect:incident-edge-vertex}, Dint-\ref{intersect:join-edge-r-rank}, Dint-\ref{intersect:join-edge-vertex}, Dint-\ref{intersect:join-adjacent-edges} and Dint-\ref{intersect:adjacent-vertex-intersect-dis} hold true.
\item an \emph{$(e,ve,ee)$-intersected $[v]$-union adjacent-$v$ distinguishing proper total set-coloring} if the constraints Dint-\ref{intersect:total}, Dint-\ref{intersect:proper-vertex}, Dint-\ref{intersect:proper-edges}, Dint-\ref{intersect:incident-edge-vertex}, Dint-\ref{intersect:join-edge-r-rank}, Dint-\ref{intersect:join-edge-vertex}, Dint-\ref{intersect:join-adjacent-edges} and Dint-\ref{intersect:adjacent-all-vertex-union-dis} hold true.
\item an \emph{$(e,ve,ee)$-$[v]$-intersected adjacent-$v$ distinguishing proper total set-coloring} if the constraints Dint-\ref{intersect:total}, Dint-\ref{intersect:proper-vertex}, Dint-\ref{intersect:proper-edges}, Dint-\ref{intersect:incident-edge-vertex}, Dint-\ref{intersect:join-edge-r-rank}, Dint-\ref{intersect:join-edge-vertex}, Dint-\ref{intersect:join-adjacent-edges} and Dint-\ref{intersect:adjacent-all-vertex-intersect-dis} hold true.
\item an \emph{$(e,ve,ee)$-intersected $e$-union adjacent-$v$ distinguishing proper total set-coloring} if the constraints Dint-\ref{intersect:total}, Dint-\ref{intersect:proper-vertex}, Dint-\ref{intersect:proper-edges}, Dint-\ref{intersect:incident-edge-vertex}, Dint-\ref{intersect:join-edge-r-rank}, Dint-\ref{intersect:join-edge-vertex}, Dint-\ref{intersect:join-adjacent-edges} and Dint-\ref{intersect:adjacent-edges-union-dis} hold true.
\item an \emph{$(e,ve,ee)$-$e$-intersected adjacent-$v$ distinguishing proper total set-coloring} if the constraints Dint-\ref{intersect:total}, Dint-\ref{intersect:proper-vertex}, Dint-\ref{intersect:proper-edges}, Dint-\ref{intersect:incident-edge-vertex}, Dint-\ref{intersect:join-edge-r-rank}, Dint-\ref{intersect:join-edge-vertex}, Dint-\ref{intersect:join-adjacent-edges} and Dint-\ref{intersect:adjacent-edges-intersect-dis} hold true.
\item an \emph{$(e,ve,ee)$-intersected $[ve]$-union adjacent-$v$ distinguishing proper total set-coloring} if the constraints Dint-\ref{intersect:total}, Dint-\ref{intersect:proper-vertex}, Dint-\ref{intersect:proper-edges}, Dint-\ref{intersect:incident-edge-vertex}, Dint-\ref{intersect:join-edge-r-rank}, Dint-\ref{intersect:join-edge-vertex}, Dint-\ref{intersect:join-adjacent-edges} and Dint-\ref{intersect:adjacent-vertex-edges-union-dis} hold true.
\item an \emph{$(e,ve,ee)$-$[ve]$-intersected adjacent-$v$ distinguishing proper total set-coloring} if the constraints Dint-\ref{intersect:total}, Dint-\ref{intersect:proper-vertex}, Dint-\ref{intersect:proper-edges}, Dint-\ref{intersect:incident-edge-vertex}, Dint-\ref{intersect:join-edge-r-rank}, Dint-\ref{intersect:join-edge-vertex}, Dint-\ref{intersect:join-adjacent-edges} and Dint-\ref{intersect:adjacent-vertex-edges-intersect-dis} hold true.\qqed
\end{asparaenum}
\end{defn}

\subsubsection{Hyperedge-set colorings}

\begin{defn} \label{defn:distinguishing-hyperedge-set-colorings}
$^*$ Let $G$ be a $(p,q)$-graph, and let $\Lambda$ be a finite set of numbers. There is a \emph{hyperedge-set coloring} $F: S\rightarrow \mathcal{E}$, where $\mathcal{E}$ is a \emph{hyperedge set} of the power set $\Lambda^2$ holding $\Lambda=\bigcup _{e\in \mathcal{E}}e$, and $S\subseteq V(G)\cup E(G)$. There are the following constraints:
\begin{asparaenum}[\textbf{\textrm{Hyset}}-1.]
\item \label{hyper:vertex} $S=V(G)$.
\item \label{hyper:edge} $S=E(G)$.
\item \label{hyper:total} $S=V(G)\cup E(G)$.

--- \emph{local distinguishing}

\item \label{hyper:adjacent-vertex} $F(u)\neq F(v)$ for each edge $uv\in E(G)$.
\item \label{hyper:adjacent-edge} $F(uv)\neq F(uw)$ for adjacent edges $uv,uw\in E(G)$ and $u\in V(G)$.
\item \label{hyper:incident-edge-vertex} $F(u)\neq F(uv)$ and $F(v)\neq F(uv)$ for each edge $uv\in E(G)$.

--- \emph{local intersected}

\item \label{hyper:join-oper-verticeice} $F(u)\cap F(v)\neq \emptyset$ for each edge $uv\in E(G)$.
\item \label{hyper:join-oper-adjacent-edges} $F(uv)\cap F(uw)\neq \emptyset$ for adjacent edges $uv,uw\in E(G)$.
\item \label{hyper:join-edge} $F(u)\cap F(v)\subseteq F(uv)$ and $F(u)\cap F(v)\neq \emptyset $ for each edge $uv\in E(H)$.
\item \label{hyper:join-oper-vertex-edge} $F(uv)\cap F(u)\neq \emptyset$ and $F(uv)\cap F(v)\neq \emptyset$ for each edge $uv\in E(G)$.

--- \emph{v-adjacent distinguishing}

\item \label{hyper:adjacent-vertex-union-dis} $\bigcup _{v\in N_{ei}(u)}F(v)\neq \bigcup _{z\in N_{ei}(w)}F(z)$ for each edge $uw\in V(H)$.
\item \label{hyper:adjacent-vertex-intersect-dis} $\bigcap _{v\in N_{ei}(u)}F(v)\neq \bigcap _{z\in N_{ei}(w)}F(z)$ for each edge $uw\in V(H)$.
\item \label{hyper:adjacent-all-vertex-union-dis} $F(u)\cup \big [\bigcup _{v\in N_{ei}(u)}F(v)\big ]\neq F(w)\cup \big [\bigcup _{z\in N_{ei}(w)}F(z)\big ]$ for each edge $uw\in V(H)$.
\item \label{hyper:adjacent-all-vertex-intersect-dis} $F(u)\cap \big [\bigcap _{v\in N_{ei}(u)}F(v)\big ]\neq F(w)\cap \big [\bigcap _{z\in N_{ei}(w)}F(z)\big ]$ for each edge $uw\in V(H)$.

--- \emph{e-adjacent distinguishing}

\item \label{hyper:adjacent-edges-union-dis} $\bigcup _{v\in N_{ei}(u)}F(uv)\neq \bigcup _{z\in N_{ei}(w)}F(wz)$ for each edge $uw\in V(H)$.
\item \label{hyper:adjacent-edges-intersect-dis} $\bigcap _{v\in N_{ei}(u)}F(uv)\neq \bigcap _{z\in N_{ei}(w)}F(wz)$ for each edge $uw\in V(H)$.

--- \emph{ve-adjacent distinguishing}

\item \label{hyper:adjacent-vertex-edges-union-dis} $F(u)\cup \big [\bigcup _{v\in N_{ei}(u)}F(uv)\big ]\neq F(w)\cup \big [\bigcup _{z\in N_{ei}(w)}F(wz)\big ]$ for each edge $uw\in V(H)$.
\item \label{hyper:adjacent-vertex-edges-intersect-dis} $F(u)\cap \big [\bigcap _{v\in N_{ei}(u)}F(uv)\big ]\neq F(w)\cap \big [\bigcap _{z\in N_{ei}(w)}F(wz)\big ]$ for each edge $uw\in V(H)$.
\end{asparaenum}
\textbf{Then}

\noindent --- \emph{hyperedge-set colorings}

\begin{asparaenum}[\textbf{\textrm{Scolo}}-1.]
\item $F$ is called \emph{proper hyperedge-set coloring} if the constraints Hyset-\ref{hyper:vertex}, and Hyset-\ref{hyper:adjacent-vertex} hold true.
\item $F$ is called \emph{proper edge hyperedge-set coloring} if the constraints Hyset-\ref{hyper:edge}, and Hyset-\ref{hyper:adjacent-edge} hold true.
\item $F$ is called \emph{proper total hyperedge-set coloring} if the constraints Hyset-\ref{hyper:total}, Hyset-\ref{hyper:adjacent-vertex}, Hyset-\ref{hyper:adjacent-edge}, and Hyset-\ref{hyper:incident-edge-vertex} hold true.

--- \emph{intersected hyperedge-set colorings}

\item $F$ is called \emph{$v$-intersected proper hyperedge-set coloring} if the constraints Hyset-\ref{hyper:vertex}, Hyset-\ref{hyper:adjacent-vertex} and Hyset-\ref{hyper:join-oper-verticeice} hold true.
\item $F$ is called \emph{$e$-intersected proper edge hyperedge-set coloring} if the constraints Hyset-\ref{hyper:edge}, Hyset-\ref{hyper:adjacent-edge} and Hyset-\ref{hyper:join-oper-adjacent-edges} hold true.
\item $F$ is called \emph{$ee$-intersected proper total hyperedge-set coloring} if the constraints Hyset-\ref{hyper:total}, Hyset-\ref{hyper:adjacent-vertex}, Hyset-\ref{hyper:adjacent-edge}, Hyset-\ref{hyper:incident-edge-vertex}, Hyset-\ref{hyper:join-oper-verticeice}, Hyset-\ref{hyper:join-oper-adjacent-edges} and Hyset-\ref{hyper:join-oper-vertex-edge} hold true.

--- \emph{v-distinguishing}

\item $F$ is called \emph{$v$-union adjacent-$v$ distinguishing proper hyperedge-set coloring} if the constraints Hyset-\ref{hyper:vertex}, Hyset-\ref{hyper:adjacent-vertex} and Hyset-\ref{hyper:adjacent-vertex-union-dis} hold true.
\item $F$ is called \emph{$v$-intersected adjacent-$v$ distinguishing proper hyperedge-set coloring} if the constraints Hyset-\ref{hyper:vertex}, Hyset-\ref{hyper:adjacent-vertex} and Hyset-\ref{hyper:adjacent-vertex-intersect-dis} hold true.
\item $F$ is called \emph{$[v]$-union adjacent-$v$ distinguishing proper hyperedge-set coloring} if the constraints Hyset-\ref{hyper:vertex}, Hyset-\ref{hyper:adjacent-vertex} and Hyset-\ref{hyper:adjacent-all-vertex-union-dis} hold true.
\item $F$ is called \emph{$[v]$-intersected adjacent-$v$ distinguishing proper hyperedge-set coloring} if the constraints Hyset-\ref{hyper:vertex}, Hyset-\ref{hyper:adjacent-vertex} and Hyset-\ref{hyper:adjacent-all-vertex-intersect-dis} hold true.

--- \emph{e-distinguishing}

\item $F$ is called \emph{$v$-union adjacent-$v$ distinguishing proper edge hyperedge-set coloring} if the constraints Hyset-\ref{hyper:edge}, Hyset-\ref{hyper:adjacent-edge} and Hyset-\ref{hyper:adjacent-edges-union-dis} hold true.
\item $F$ is called \emph{$v$-intersected adjacent-$v$ distinguishing proper edge hyperedge-set coloring} if the constraints Hyset-\ref{hyper:edge}, Hyset-\ref{hyper:adjacent-edge} and Hyset-\ref{hyper:adjacent-edges-intersect-dis} hold true.
\item $F$ is called \emph{$(e,v)$-intersected adjacent-$v$ distinguishing proper edge hyperedge-set coloring} if the constraints Hyset-\ref{hyper:edge}, Hyset-\ref{hyper:adjacent-edge}, Hyset-\ref{hyper:join-oper-adjacent-edges} and Hyset-\ref{hyper:adjacent-edges-intersect-dis} hold true.

--- \emph{total-distinguishing}

\item $F$ is called \emph{$v$-union adjacent-$v$ distinguishing proper total hyperedge-set coloring} if the constraints Hyset-\ref{hyper:total}, Hyset-\ref{hyper:adjacent-vertex}, Hyset-\ref{hyper:adjacent-edge}, and Hyset-\ref{hyper:incident-edge-vertex} and Hyset-\ref{hyper:adjacent-vertex-union-dis} hold true.
\item $F$ is called \emph{$v$-intersected adjacent-$v$ distinguishing proper total hyperedge-set coloring} if the constraints Hyset-\ref{hyper:total}, Hyset-\ref{hyper:adjacent-vertex}, Hyset-\ref{hyper:adjacent-edge}, and Hyset-\ref{hyper:incident-edge-vertex} and Hyset-\ref{hyper:adjacent-vertex-intersect-dis} hold true.

\item $F$ is called \emph{$[v]$-union adjacent-$v$ distinguishing proper total hyperedge-set coloring} if the constraints Hyset-\ref{hyper:total}, Hyset-\ref{hyper:adjacent-vertex}, Hyset-\ref{hyper:adjacent-edge}, and Hyset-\ref{hyper:incident-edge-vertex} and Hyset-\ref{hyper:adjacent-all-vertex-union-dis} hold true.
\item $F$ is called \emph{$[v]$-intersected adjacent-$v$ distinguishing proper total hyperedge-set coloring} if the constraints Hyset-\ref{hyper:total}, Hyset-\ref{hyper:adjacent-vertex}, Hyset-\ref{hyper:adjacent-edge}, and Hyset-\ref{hyper:incident-edge-vertex} and Hyset-\ref{hyper:adjacent-all-vertex-intersect-dis} hold true.
\item $F$ is called \emph{$e$-union adjacent-$v$ distinguishing proper total hyperedge-set coloring} if the constraints Hyset-\ref{hyper:total}, Hyset-\ref{hyper:adjacent-vertex}, Hyset-\ref{hyper:adjacent-edge}, and Hyset-\ref{hyper:incident-edge-vertex} and Hyset-\ref{hyper:adjacent-edges-union-dis} hold true.
\item $F$ is called \emph{$e$-intersected adjacent-$v$ distinguishing proper total hyperedge-set coloring} if the constraints Hyset-\ref{hyper:total}, Hyset-\ref{hyper:adjacent-vertex}, Hyset-\ref{hyper:adjacent-edge}, and Hyset-\ref{hyper:incident-edge-vertex} and Hyset-\ref{hyper:adjacent-edges-intersect-dis} hold true.

\item $F$ is called \emph{$[ve]$-union adjacent-$v$ distinguishing proper total hyperedge-set coloring} if the constraints Hyset-\ref{hyper:total}, Hyset-\ref{hyper:adjacent-vertex}, Hyset-\ref{hyper:adjacent-edge}, and Hyset-\ref{hyper:incident-edge-vertex} and Hyset-\ref{hyper:adjacent-vertex-edges-union-dis} hold true.
\item $F$ is called \emph{$[ve]$-intersected adjacent-$v$ distinguishing proper total hyperedge-set coloring} if the constraints Hyset-\ref{hyper:total}, Hyset-\ref{hyper:adjacent-vertex}, Hyset-\ref{hyper:adjacent-edge}, and Hyset-\ref{hyper:incident-edge-vertex} and Hyset-\ref{hyper:adjacent-vertex-edges-intersect-dis} hold true.

\item $F$ is called \emph{$ee$-intersected $v$-union adjacent-$v$ distinguishing proper total hyperedge-set coloring} if the constraints Hyset-\ref{hyper:total}, Hyset-\ref{hyper:adjacent-vertex}, Hyset-\ref{hyper:adjacent-edge}, Hyset-\ref{hyper:incident-edge-vertex}, Hyset-\ref{hyper:join-oper-verticeice}, Hyset-\ref{hyper:join-oper-adjacent-edges}, Hyset-\ref{hyper:join-oper-vertex-edge} and Hyset-\ref{hyper:adjacent-vertex-union-dis} hold true.
\item $F$ is called \emph{$ee$-$v$-intersected adjacent-$v$ distinguishing proper total hyperedge-set coloring} if the constraints Hyset-\ref{hyper:total}, Hyset-\ref{hyper:adjacent-vertex}, Hyset-\ref{hyper:adjacent-edge}, Hyset-\ref{hyper:incident-edge-vertex}, Hyset-\ref{hyper:join-oper-verticeice}, Hyset-\ref{hyper:join-oper-adjacent-edges}, Hyset-\ref{hyper:join-oper-vertex-edge} and Hyset-\ref{hyper:adjacent-vertex-intersect-dis} hold true.
\item $F$ is called \emph{$ee$-intersected $[v]$-union adjacent-$v$ distinguishing proper total hyperedge-set coloring} if the constraints Hyset-\ref{hyper:total}, Hyset-\ref{hyper:adjacent-vertex}, Hyset-\ref{hyper:adjacent-edge}, Hyset-\ref{hyper:incident-edge-vertex}, Hyset-\ref{hyper:join-oper-verticeice}, Hyset-\ref{hyper:join-oper-adjacent-edges}, Hyset-\ref{hyper:join-oper-vertex-edge} and Hyset-\ref{hyper:adjacent-all-vertex-union-dis} hold true.
\item $F$ is called \emph{$ee$-$[v]$-intersected adjacent-$v$ distinguishing proper total hyperedge-set coloring} if the constraints Hyset-\ref{hyper:total}, Hyset-\ref{hyper:adjacent-vertex}, Hyset-\ref{hyper:adjacent-edge}, Hyset-\ref{hyper:incident-edge-vertex}, Hyset-\ref{hyper:join-oper-verticeice}, Hyset-\ref{hyper:join-oper-adjacent-edges}, Hyset-\ref{hyper:join-oper-vertex-edge} and Hyset-\ref{hyper:adjacent-all-vertex-intersect-dis} hold true.
\item $F$ is called \emph{$ee$-intersected $e$-union adjacent-$v$ distinguishing proper total hyperedge-set coloring} if the constraints Hyset-\ref{hyper:total}, Hyset-\ref{hyper:adjacent-vertex}, Hyset-\ref{hyper:adjacent-edge}, Hyset-\ref{hyper:incident-edge-vertex}, Hyset-\ref{hyper:join-oper-verticeice}, Hyset-\ref{hyper:join-oper-adjacent-edges}, Hyset-\ref{hyper:join-oper-vertex-edge} and Hyset-\ref{hyper:adjacent-edges-union-dis} hold true.
\item $F$ is called \emph{$ee$-$e$-intersected adjacent-$v$ distinguishing proper total hyperedge-set coloring} if the constraints Hyset-\ref{hyper:total}, Hyset-\ref{hyper:adjacent-vertex}, Hyset-\ref{hyper:adjacent-edge}, Hyset-\ref{hyper:incident-edge-vertex}, Hyset-\ref{hyper:join-oper-verticeice}, Hyset-\ref{hyper:join-oper-adjacent-edges}, Hyset-\ref{hyper:join-oper-vertex-edge} and Hyset-\ref{hyper:adjacent-edges-intersect-dis} hold true.
\item $F$ is called \emph{$ee$-intersected $[ve]$-union adjacent-$v$ distinguishing proper total hyperedge-set coloring} if the constraints Hyset-\ref{hyper:total}, Hyset-\ref{hyper:adjacent-vertex}, Hyset-\ref{hyper:adjacent-edge}, Hyset-\ref{hyper:incident-edge-vertex}, Hyset-\ref{hyper:join-oper-verticeice}, Hyset-\ref{hyper:join-oper-adjacent-edges}, Hyset-\ref{hyper:join-oper-vertex-edge} and Hyset-\ref{hyper:adjacent-vertex-edges-union-dis} hold true.
\item $F$ is called \emph{$ee$-$[ve]$-intersected adjacent-$v$ distinguishing proper total hyperedge-set coloring} if the constraints Hyset-\ref{hyper:total}, Hyset-\ref{hyper:adjacent-vertex}, Hyset-\ref{hyper:adjacent-edge}, Hyset-\ref{hyper:incident-edge-vertex}, Hyset-\ref{hyper:join-oper-verticeice}, Hyset-\ref{hyper:join-oper-adjacent-edges}, Hyset-\ref{hyper:join-oper-vertex-edge} and Hyset-\ref{hyper:adjacent-vertex-edges-intersect-dis} hold true.\qqed
\end{asparaenum}
\end{defn}

\begin{defn} \label{defn:distinguishing-constraint-hypergraph}
$^*$ For Definition \ref{defn:distinguishing-hyperedge-set-colorings}, we add a \textbf{Feedback-constraint}: Each pair of hyperedges $e,e\,'\in \mathcal{E}$ holding $|e\cap e\,'|\geq 1$ corresponds to an edge $xy\in E(G)$, such that $F(x)=e$ and $F(y)=e\,'$.

Suppose that a graph $G$ admits a $W$-constraint hyperedge-set coloring defined in Definition \ref{defn:distinguishing-hyperedge-set-colorings}. If the graph $G$ satisfies the Feedback-constraint, then we call $G$ the \emph{$W$-constraint graph} of the hypergraph $\mathcal{H}_{yper}=(\Lambda,\mathcal{E})$.\qqed
\end{defn}

\begin{defn} \label{defn:intersected-graph-vs-hypergraph}
Let $\Lambda$ be a set of finite elements. A graph $G$ has its own vertex set $V(G)=\mathcal{E}$, where $\mathcal{E}$ is a \emph{hyperedge set} holding $\Lambda=\bigcup_{e\in \mathcal{E}}e$, and each edge $e_xe_y\in E(G)$ means $e_x\cap e_y\neq \emptyset$, then we call $G$ \emph{intersected graph} of the hypergraph $\mathcal{H}_{yper}=(\Lambda,\mathcal{E})$.
\end{defn}

\begin{rem}\label{rem:333333}
Definition \ref{defn:distinguishing-constraint-hypergraph} enables us to know more complex hypergraphs through $W$-constraint graphs, for example, a hypergraph $\mathcal{H}_{yper}=(\Lambda,\mathcal{E})$ with its own $ee$-$[ve]$-intersected adjacent-$v$ distinguishing hyperedge set $\mathcal{E}$ defined in Definition \ref{defn:distinguishing-hyperedge-set-colorings}.

Clearly, some topological structures of the intersected graph $G$ are just what the hypergraph $\mathcal{H}_{yper}=(\Lambda,\mathcal{E})$ has by Definition \ref{defn:intersected-graph-vs-hypergraph}, so a hypergraph is equivalent to its intersected graph. Thereby, intersected graphs help us to know hypergraphs, and moreover some properties of the intersected graphs are just what corresponding hypergraphs have; refer to \cite{Yao-Ma-arXiv-2201-13354v1}.\paralled
\end{rem}

\begin{rem}\label{rem:333333}
\textbf{The edge-set-coloring 1-2-3-conjecture.} If there is an \emph{edge set-coloring} $\theta: E(H)\rightarrow Q$ for a connected graph $H$, such that the vertex induced color holding
\begin{equation}\label{eqa:edge-set-coloring-123-conjecture}
\sum _{v\in N_{ei}(u)}|\theta(uv)|=: \theta(u)\neq \theta(w):=\sum _{z\in N_{ei}(w)}|\theta(wz)|
\end{equation} for each edge $uw\in V(H)$, then we call $\theta$ a \emph{proper vertex set-coloring} of the connected graph $H$, where $Q$ is a set of subsets with different cardinalities, or $Q$ is a hyperedge set. Motivated from the 1-2-3-conjecture, we propose the \textbf{edge-set-coloring 1-2-3-conjecture}: \emph{The cardinality $|Q|\leq 3$ is enough for every connected graph $H$ to admits an edge set-coloring $\theta$ based on $Q$ such that $\theta$ induces a proper vertex set-coloring} of $H$.\paralled
\end{rem}

\begin{thm}\label{thm:666666}
$^*$ If the intersected graph $G$ of a hypergraph $\mathcal{H}_{yper}=(\Lambda,\mathcal{E})$ supports the \textbf{1-2-3-conjecture}, then $\mathcal{H}_{yper}=(\Lambda,\mathcal{E})$ supports the \textbf{edge-set-coloring 1-2-3-conjecture} too.
\end{thm}

\begin{rem}\label{rem:333333}
For vector $\overrightarrow{d}_i=(a_i,b_i,c_i)$, the \textbf{edge-vector-coloring 1-2-3-conjecture}: Each connected graph $G$ admits an \emph{edge vector-coloring} $f:E(G)\rightarrow V_{ector}=\{\overrightarrow{d}_i=(a_i,b_i,c_i):i\in [1,3]\}$ such that $f$ colors properly the vertices of $G$, that is
\begin{equation}\label{eqa:vector-coloring-123-conjecture}
\sum _{v\in N_{ei}(u)}\| f(uv)\|=:f(u)\neq f(w):=\sum _{z\in N_{ei}(w)}\| f(wz)\|
\end{equation} for each edge $uw\in E(G)$.\paralled
\end{rem}

\subsection{Integer partition into magic-constraints}

Since number-based strings were induced from the colorings of graphs, we present some theorems for determining graphs admitting colorings with the magic-constraints in this subsection.

\begin{thm}\label{thm:partition-integer-edge-magic-cons}
The set $S_{ema}(m)$ contains each group of mutually distinct positive integers $a,b,c$ holding the $ema$-constraint $m=a+b+c$ for a fixed integer $m\geq 5$. There is at least a connected graph $G$ admitting a proper total coloring $f:V(G)\cup E(G)\rightarrow S_{ema}(m)$, such that each edge $uv\in E(G)$ holds the edge-magic constraint $f(u)+f(v)+f(uv)=m$ with $f(u)=a$, $f(v)=b$ and $f(uv)=c$, and $\Delta(G)\leq m-5$, or $\Delta(G)\leq m-4$, and $G$ contains a path $P$ of length being greater than any given positive integer, as well as the edge number $|E(G)|$ can be as big as it wants.
\end{thm}
\begin{proof}
$G=K_{1,m-4}$, $V(K_{1,m-4})=\{u,v_i:i\in [1,m-4]\}$, $E(K_{1,m-4})=\{uv_i:i\in [1,m-4]\}$; $f(u)=1$, $f(uv_i)=i+1$ and $f(v_i)=m-2-i$ for $i\in [1,m-4]$. Clearly, $f(u)+f(uv_i)+f(v_i)=m$ and $\Delta(K_{1,m-4})\leq m-4$; see Fig.\ref{fig:integer-edge-magic} (a).

If $i+1=m-2-i$, or $2i=m-3$, we get another connected graph $K_{1,m-4}-uv_i+v_1v_i$, and add $f(v_1v_i)=i$, here $\Delta(K_{1,m-4}-uv_i+v_1v_i)\leq m-5$; see Fig.\ref{fig:integer-edge-magic} (b).

We omit the proofs about the path length and the edge number, since they can be obtained by the polynomial construction methods; see Fig.\ref{fig:integer-edge-magic} (c).
\end{proof}

\begin{figure}[h]
\centering
\includegraphics[width=16.4cm]{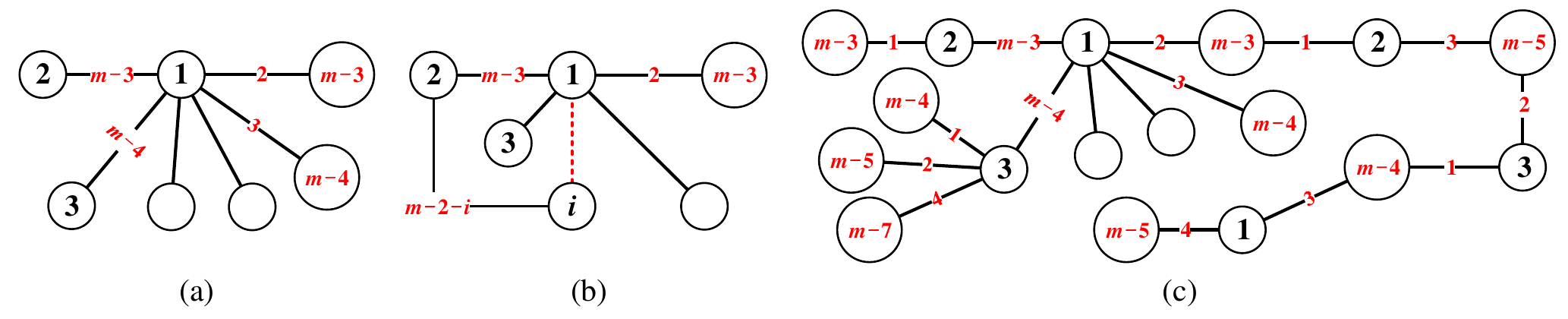}\\
\caption{\label{fig:integer-edge-magic}{\small Examples for the integer partition under the edge-magic constraint.}}
\end{figure}

\begin{thm}\label{thm:edge-difference-integer-partition}
The set $S_{edi}(m)$ for a fixed integer $m\geq 5$ contains each group of mutually distinct positive integers $a,b,c$ holding the $edi$-constraint $m=edi \langle a,b,c\rangle=c+|a-b|$. For given integers $d\geq 2$ and $r\geq 2$, there is a connected graph $H$ admitting a proper total coloring $g:V(H)\cup E(H)\rightarrow S_{edi}(m)$, such that $g(u)=a$, $g(v)=b$ and $g(uv)=c$, and

(i) each edge $uv\in E(H)$ holds the edge-difference constraint $g(uv)+|g(u)-g(v)|=m$ true;

(ii) the maximal degree $\Delta(H)<m$;

(iii) the edge number $|E(H)|\geq r$;

(iii) the graph $H$ contains a path $P$ of length no less than $d$.
\end{thm}

\begin{figure}[h]
\centering
\includegraphics[width=16.4cm]{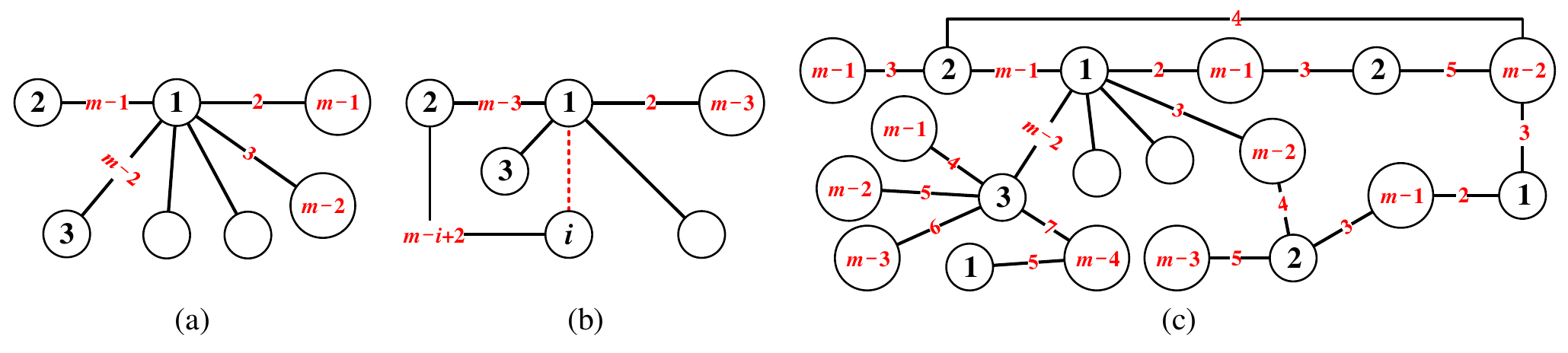}\\
\caption{\label{fig:integer-edge-difference}{\small Examples for illustrating Theorem \ref{thm:edge-difference-integer-partition} under the edge-difference constraint.}}
\end{figure}

\begin{thm}\label{thm:felicitous-difference-integer-partition}
The set $S_{fdi}(m)$ for a fixed integer $m\geq 5$ contains each group of mutually distinct positive integers $a,b,c$ holding the $fdi$-constraint $m=fdi \langle a,b,c\rangle=|a+b-c|$. Then there are infinite elements in $S_{fdi}(m)$, and for given integers $\Delta\geq 2$, $d\geq 2$ and $r\geq 2$, there is a connected graph $T$ admitting a proper total coloring $h:V(T)\cup E(T)\rightarrow S_{fdi}(m)$, such that $h(u)=a$, $h(v)=b$ and $h(uv)=c$, and

(i) the felicitous-difference constraint $|h(u)+h(v)-h(uv)|=m$ for each edge $uv\in E(T)$ holds true;

(ii) the maximal degree $\Delta(T)=\Delta$;

(iii) the edge number $|E(T)|\geq r$;

(iii) the graph $T$ contains a path $P$ of length no less than $d$.
\end{thm}

\begin{figure}[h]
\centering
\includegraphics[width=16.4cm]{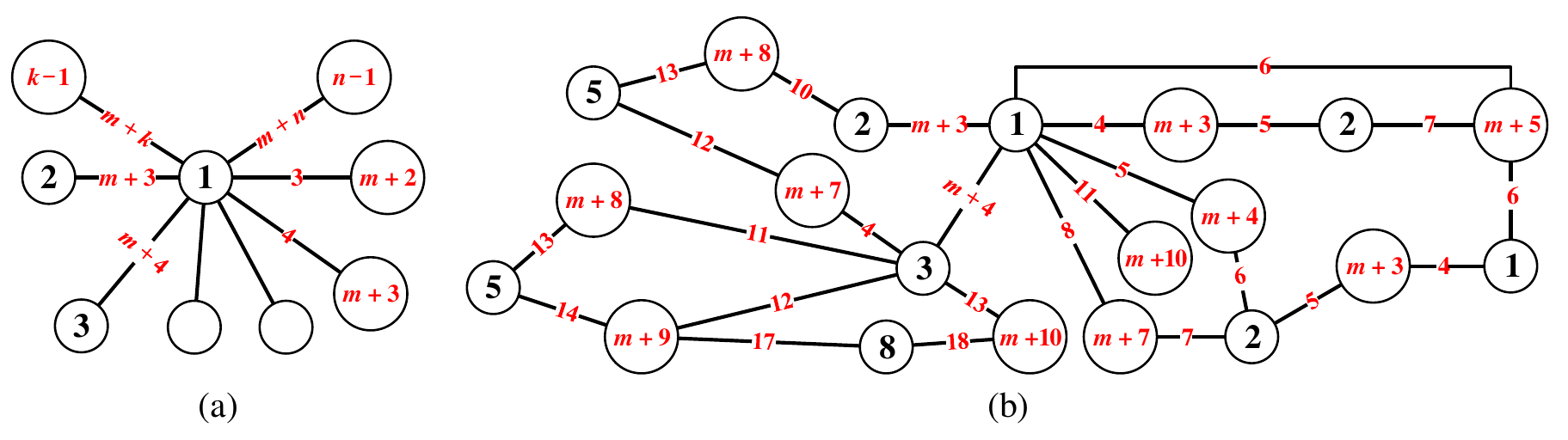}\\
\caption{\label{fig:integer-felicitous-difference}{\small A diagram for illustrating Theorem \ref{thm:felicitous-difference-integer-partition} under the felicitous-difference constraint.}}
\end{figure}

\begin{thm}\label{thm:graceful-difference-integer-partition}
The set $S_{gdi}(m)$ for a fixed integer $m\geq 0$ contains each group of mutually distinct positive integers $a,b,c$ holding the $gdi$-constraint $m=gdi \langle a,b,c\rangle=\big ||a-b|-c\big |$. Then there are infinite elements in $S_{gdi}(m)$, and for given integers $\Delta\geq 2$, $d\geq 2$ and $r\geq 2$, there is a connected graph $J$ admitting a proper total coloring $\theta:V(J)\cup E(J)\rightarrow S_{gdi}(m)$, such that $\theta(u)=a$, $\theta(v)=b$ and $\theta(uv)=c$, and

(i) the graceful-difference constraint $\big ||\theta(u)-\theta(v)|-\theta(uv)\big |=m$ for each edge $uv\in E(J)$ holds true;

(ii) the maximal degree $\Delta(J)=\Delta$;

(iii) the edge number $|E(J)|\geq r$;

(iii) the graph $J$ contains a path $P$ of length no less than $d$.
\end{thm}

\begin{figure}[h]
\centering
\includegraphics[width=16.4cm]{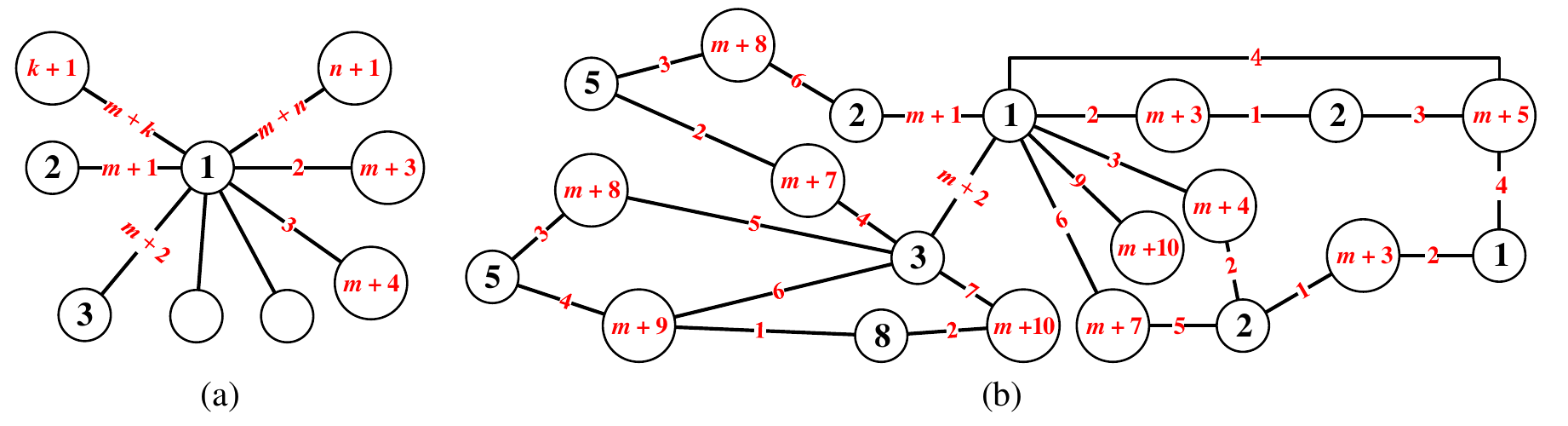}\\
\caption{\label{fig:integer-graceful-difference}{\small Examples for illustrating Theorem \ref{thm:graceful-difference-integer-partition} under the graceful-difference constraint.}}
\end{figure}

\begin{problem}\label{problem:integer-partitions-other-4}
Let $\varepsilon \in \{ema, edi, fdi, gdi\}$, and let $G^{~\varepsilon}_{raph}(m)$ be the set of graphs, such that each graph $G\in G^{~\varepsilon}_{raph}(m)$ admits a proper total coloring
\begin{equation}\label{eqproblem:integer-partitions-other-4}
f:V(G)\cup E(G)\rightarrow S_{\varepsilon}(m)
\end{equation} holding the $\varepsilon$-constraint $m=\varepsilon \langle a,b,c\rangle$ for each edge $uv\in E(G)$ with $f(u)=a$, $f(v)=b$ and $f(uv)=c$ about four sets $S_{ema}(m)$, $S_{edi}(m)$, $S_{fdi}(m)$ and $S_{gdi}(m)$ defined in Theorem \ref{thm:partition-integer-edge-magic-cons}, Theorem \ref{thm:edge-difference-integer-partition}, Theorem \ref{thm:felicitous-difference-integer-partition} and Theorem \ref{thm:graceful-difference-integer-partition}, respectively.
\begin{asparaenum}[\textbf{\textrm{Four}}-1. ]
\item \textbf{Characterize} each set $S_{\varepsilon}(m)$ for each $m\geq 5$ and $\varepsilon \in \{ema, edi, fdi, gdi\}$.
\item \textbf{Characterize} some properties of each graph set $G^{~\varepsilon}_{raph}(m)$ for $\varepsilon \in \{ema, edi, fdi, gdi\}$.
\item For each group of the $\varepsilon$-constraints $m=\varepsilon \langle a_i,b_i,c_i\rangle $ with $i\in [1,A]$ in $S_{\varepsilon}(m)$, \textbf{does} there exist a connected $(p,A)$-graph $G\in G^{~\varepsilon}_{raph}(m)$ admitting a proper total coloring $\varphi:V(G)\cup E(G)\rightarrow S_{\varepsilon}(m)$, such that each edge $x_iy_i\in E(G)$ with $i\in [1,A]$ just holding $\varphi(x_i)=a_i$, $\varphi(y_i)=b_i$ and $\varphi(x_iy_i)=c_i$ true? Moreover, considering the following particular cases:

\qquad (i) $a_i,b_i$ and $c_i$ are prime numbers;

\qquad (ii) $a_i,b_i$ and $c_i$ are odd numbers;

\qquad (iii) $c_i=\textrm{gcd}(a_i,b_i)$.
\item \textbf{Find} a connected graph $H\in G^{~\varepsilon}_{raph}(m)$ holding
$$\frac{|E(H)|}{|V(H)|}\geq \frac{|E(T)|}{|V(T)|}
$$ for any connected graph $T\in G^{~\varepsilon}_{raph}(m)$ with $\varepsilon \in \{ema, edi, fdi, gdi\}$.
\item For each graph $G\in G^{~\varepsilon}_{raph}(m)$, \textbf{find} two proper total colorings $f^{v}_{\min}$ and $f^{v}_{\max}$ (resp. $f^{e}_{\min}$ and $f^{e}_{\max}$) defined in Eq.(\ref{eqproblem:integer-partitions-other-4}), such that cardinalities
\begin{equation}\label{eqa:555555}
|f^{v}_{\min}(V(G))|\leq |f(V(G))|\leq |f^{v}_{\max}(V(G))|~(\textrm{resp.}|f^{e}_{\min}(E(G))|\leq |f(E(G))|\leq |f^{e}_{\max}(E(G))|)
\end{equation} over all proper total colorings $f$ of $G$ defined in Eq.(\ref{eqproblem:integer-partitions-other-4}).
\item \textbf{Find} a proper total coloring $f^{ve}$ for each graph $G\in G^{~\varepsilon}_{raph}(m)$ such that
$$\frac{|f(E(G))|}{|f(V(G))|}\leq \frac{|f^{ve}(E(G))|}{|f^{ve}(V(G))|}
$$ over all proper total colorings $f$ of $G$ defined in Eq.(\ref{eqproblem:integer-partitions-other-4}).
\item \textbf{Find} some connections between $S_{\varepsilon}(m)$ and $S_{\varepsilon}(n)$ (resp. $G^{~\varepsilon}_{raph}(m)$ and $G^{~\varepsilon}_{raph}(n)$) for $m\neq n$ and $\varepsilon \in \{ema, edi, fdi, gdi\}$.
\item \textbf{Find} non-empty subsets $X^{sub}_{graph}\subseteq G^{~\varepsilon}_{raph}(m)\cap G^{~\gamma}_{raph}(m)$ for distinct $\varepsilon,\gamma \in \{ema$, $edi$, $fdi$, $gdi\}$, because it is guaranteed by Theorem \ref{thm:10-k-d-total-coloringss} (Ref. \cite{Bing-Yao-arXiv:2207-03381}).
\end{asparaenum}
\end{problem}

\begin{thm}\label{thm:666666}
$^*$ By the graph sets $G^{~\varepsilon}_{raph}(m)$ with $\varepsilon \in \{ema, edi, fdi, gdi\}$ defined in Problem \ref{problem:integer-partitions-other-4}, we have:

(i) If a connected non-tree $(p,q)$-graph $G^*\in G^{~\varepsilon}_{raph}(m)$ admits a proper total coloring $\pi:V(G^*)\cup E(G^*)\rightarrow S_{\varepsilon}(m)$, such that each edge $uv\in E(G^*)$ holds the $\varepsilon$-constraint $m=\varepsilon \langle a,b,c\rangle $ true, then there is another connected $(p\,',q)$-graph $H\in G^{~\varepsilon}_{raph}(m)$ admitting a proper total coloring $\varphi:V(H)\cup E(H)\rightarrow \pi(V(G^*)\cup E(G^*))$ holding the $\varepsilon$-constraint $m=\varepsilon \langle a,b,c\rangle $ true.

(ii) Each graph $H\in G^{~\varepsilon}_{raph}(m)$ holding the $\varepsilon$-constraint $m=\varepsilon \langle a,b,c\rangle $ corresponds to another graph $H^*\in G^{~\varepsilon}_{raph}(km)$ with $k\geq 2$, since the $\varepsilon$-constraint $m=\varepsilon \langle a,b,c\rangle $ corresponds to the $\varepsilon$-constraint $km=\varepsilon \langle ka,kb,kc\rangle $.

(iii) Each graph $H\in G^{~\varepsilon}_{raph}(m)$ holding the $\varepsilon$-constraint $m=\varepsilon \langle a,b,c\rangle $ corresponds to another graph $H^*\in G^{~\varepsilon}_{raph}(m^*)$ admitting a proper total coloring $g:V(H^*)\cup E(H^*)\rightarrow S_{\varepsilon}(m^*)$ holding the $\varepsilon$-constraint
$$m^*=m+k=\varepsilon \langle a+k,b+k,c+k\rangle\textrm{ or }m^*=|m-k|=\varepsilon \langle a+k,b+k,c+k\rangle $$
\end{thm}

\begin{rem}\label{rem:333333}
Theorem \ref{thm:10-k-d-total-coloringss} shows that a connected graph $G$ may admit two or more $\varepsilon$-constraint proper total colorings with $\varepsilon$-constraint $\in \{$$ema$-constraint $m=ema \langle a,b,c\rangle=a+b+c$, $edi$-constraint $m=edi \langle a,b,c\rangle=c+|a-b|$, $fdi$-constraint $m=fdi \langle a,b,c\rangle=|a+b-c|$, $gdi$-constraint $m=gdi \langle a,b,c\rangle=\big ||a-b|-c\big |$$\}$.\paralled
\end{rem}

\begin{thm}\label{thm:2-magic-coloring-in-one}
$^*$ Suppose that a bipartite connected graph $G$ admits a proper total coloring $f$ holding the $ema$-constraint $m_1=ema \langle a,b,c\rangle=a+b+c$ and another proper total coloring $g$ holding the $edi$-constraint $m_2=edi \langle a,b,c\rangle=c+|a-b|$, such that $f(X)<f(Y)$ and $g(X)<g(Y)$ with $V(G)=X\cup Y$ and $X\cap Y=\emptyset$. We make a new proper total coloring $F=f+g$ for the connected graph $G$, then
\begin{equation}\label{eqa:2-magic-coloring-in-one11}
F(x)+F(y)+F(xy)+F(xy)+|F(x)-F(y)|\leq 2(m_1+m_2)
\end{equation} for each edge $xy\in E(G)$.
\end{thm}
\begin{proof} By Case A shown in the subsection ``Connections between magic-constraint colorings of graphs'', we have
\begin{equation}\label{eqa:2-magic-coloring-in-one22}
{
\begin{split}
F(x)+F(y)+F(xy)&=f(x)+f(y)+f(xy)+g(x)+g(y)+g(xy)\\
&=m_1+g(x)+g(y)+g(xy)\\
&=m_1+m_2+2g(x)
\end{split}}
\end{equation}

By Case B shown in the subsection ``Connections between magic-constraint colorings of graphs'', we have
\begin{equation}\label{eqa:2-magic-coloring-in-one33}
{
\begin{split}
F(xy)+|F(x)-F(y)|&=(f(xy)+|f(x)-f(y)|)+(g(xy)+|g(x)-g(y)|)\\
&=f(xy)+|f(x)-f(y)|+m_2\\
&=m_1-2g(x)+m_2
\end{split}}
\end{equation} So, Eq.(\ref{eqa:2-magic-coloring-in-one11}) follows Eq.(\ref{eqa:2-magic-coloring-in-one22}) and Eq.(\ref{eqa:2-magic-coloring-in-one33}).
\end{proof}

\begin{rem}\label{rem:333333}
About Theorem \ref{thm:2-magic-coloring-in-one} we have other cases as follows:

\textbf{Case (i)} The proper total coloring $f$ holds the $ema$-constraint $m_1=ema \langle a,b,c\rangle=a+b+c$ and, the proper total coloring $g$ holds the $fdi$-constraint $m_3=fdi \langle a,b,c\rangle=|a+b-c|$, and $f(X)<f(Y)$ and $g(X)<g(Y)$ with $V(G)=X\cup Y$ and $X\cap Y=\emptyset$.

By Case A shown in the subsection ``Connections between magic-constraint colorings of graphs'', we have
\begin{equation}\label{eqa:2-magic-coloring-in-one44}
{
\begin{split}
F(x)+F(y)+F(xy)&=f(x)+f(y)+f(xy)+g(x)+g(y)+g(xy)\\
&=m_1+g(x)+g(y)+g(xy)\\
&=m_1+m_3-2g(xy)
\end{split}}
\end{equation}

By Case C shown in the subsection ``Connections between magic-constraint colorings of graphs'', we have
\begin{equation}\label{eqa:2-magic-coloring-in-one55}
{
\begin{split}
|F(x)+F(y)-F(xy)|&=|(f(x)+f(y)-f(xy))+(g(x)+g(y)-g(xy))|\\
&\leq |f(x)+f(y)-f(xy)|+m_3\\
&=m_3+m_1+2g(xy)
\end{split}}
\end{equation} We get
$$F(x)+F(y)+F(xy)+|F(x)+F(y)-F(xy)|\leq 2(m_1+m_3)
$$
 for each edge $xy\in E(G)$.

\textbf{Case (ii)} The proper total coloring $f$ holds the $ema$-constraint $m_1=ema \langle a,b,c\rangle=a+b+c$ and, the proper total coloring $g$ holds the $gdi$-constraint $m_4=gdi \langle a,b,c\rangle=\big ||a-b|-c\big |$, and $f(X)<f(Y)$ and $g(X)<g(Y)$ with $V(G)=X\cup Y$ and $X\cap Y=\emptyset$.

By Case A shown in the subsection ``Connections between magic-constraint colorings of graphs'', we have
\begin{equation}\label{eqa:2-magic-coloring-in-one66}
{
\begin{split}
F(x)+F(y)+F(xy)&=f(x)+f(y)+f(xy)+g(x)+g(y)+g(xy)\\
&=m_1+g(x)+g(y)+g(xy)\\
&=m_1+m_4-2g(y)
\end{split}}
\end{equation}

By Case D shown in the subsection ``Connections between magic-constraint colorings of graphs'', we have
\begin{equation}\label{eqa:2-magic-coloring-in-one77}
{
\begin{split}
\big ||F(x)-F(y)|-F(xy)\big |&\leq \big ||f(x)-f(y)|-f(xy)\big |+\big ||g(x)-g(y)|-g(xy)\big |\\
&=\big ||f(x)-f(y)|-f(xy)\big |+m_4\\
&=m_4+m_1-2g(y)
\end{split}}
\end{equation} We obtain
$$F(x)+F(y)+F(xy)+\big ||F(x)-F(y)|-F(xy)\big |\leq 2(m_1+m_4)
$$ for each edge $xy\in E(G)$.

\textbf{Case (iii)} The proper total coloring $f$ holds the $edi$-constraint $m_2=edi \langle a,b,c\rangle=c+|a-b|$ and, the proper total coloring $g$ holds the $fdi$-constraint $m_3=fdi \langle a,b,c\rangle=|a+b-c|$, and $f(X)<f(Y)$ and $g(X)<g(Y)$ with $V(G)=X\cup Y$ and $X\cap Y=\emptyset$. Since
$${
\begin{split}
F(xy)+|F(y)-F(x)|&=f(xy)+g(xy)+|[f(x)-f(y)]+[g(x)-g(y)]|\\
&\leq f(xy)+g(xy)+|[f(x)-f(y)]|+|[g(x)-g(y)]|\\
&=m_2+g(xy)+|g(x)-g(y)|\\
&=m_2+m_3+2g(y)
\end{split}}
$$
and
$${
\begin{split}
|F(y)+F(x)-F(xy)|&=|[f(x)+f(y)-f(xy)]+[g(x)+g(y)-g(xy)]|\\
&\leq |f(x)+f(y)-f(xy)|+|g(x)+g(y)-g(xy)|\\
&=|f(x)+f(y)-f(xy)|+m_3\\
&=m_2-2f(xy)+m_3
\end{split}}
$$
so we obtain
$${
\begin{split}
F(xy)+|F(y)-F(x)|+|F(y)+F(x)-F(xy)|\leq 2(m_2+m_3)
\end{split}}
$$ for each edge $xy\in E(G)$.

\textbf{Case (iv)} The proper total coloring $f$ holds the $edi$-constraint $m_2=edi \langle a,b,c\rangle=c+|a-b|$ and, the proper total coloring $g$ holds the $gdi$-constraint $m_4=gdi \langle a,b,c\rangle=\big ||a-b|-c\big |$, and $f(X)<f(Y)$ and $g(X)<g(Y)$ with $V(G)=X\cup Y$ and $X\cap Y=\emptyset$. Because of
$${
\begin{split}
F(xy)+|F(y)-F(x)|&=f(xy)+g(xy)+|[f(x)-f(y)]+[g(x)-g(y)]|\\
&\leq m_2+g(xy)+|g(x)-g(y)|\\
&=m_2+m_4+2g(xy)
\end{split}}
$$
and
$${
\begin{split}
\big ||F(x)-F(y)|-F(xy)\big |&\leq \big ||f(x)-f(y)|-f(xy)\big |+\big ||g(x)-g(y)|-g(xy)\big |\\
&=\big ||f(x)-f(y)|-f(xy)\big |+m_4\\
&=m_2-2f(xy)+m_4
\end{split}}
$$
so we obtain
$$F(xy)+|F(y)-F(x)|+\big ||F(x)-F(y)|-F(xy)\big |\leq 2(m_2+m_4)
$$ for each edge $xy\in E(G)$.

\textbf{Case (v)} The proper total coloring $f$ holds the $fdi$-constraint $m_3=fdi \langle a,b,c\rangle=|a+b-c|$ and, the proper total coloring $g$ holds the $gdi$-constraint $m_4=gdi \langle a,b,c\rangle=\big ||a-b|-c\big |$, and $f(X)<f(Y)$ and $g(X)<g(Y)$ with $V(G)=X\cup Y$ and $X\cap Y=\emptyset$. Notice
$${
\begin{split}
|F(y)+F(x)-F(xy)|&\leq |f(y)+f(x)-f(xy)|+|g(y)+g(x)-g(xy)|\\
&=m_3+|g(y)+g(x)-g(xy)|\\
&=m_3+m_4+2g(x)
\end{split}}
$$
and
$${
\begin{split}
\big ||F(x)-F(y)|-F(xy)\big |&\leq \big ||f(x)-f(y)|-f(xy)\big |+\big ||g(x)-g(y)|-g(xy)\big |\\
&=\big ||f(x)-f(y)|-f(xy)\big |+m_4\\
&=m_3-2f(x)+m_4
\end{split}}
$$
then we get
$$|F(y)+F(x)-F(xy)|+\big ||F(x)-F(y)|-F(xy)\big |\leq 2(m_3+m_4)
$$ for each edge $xy\in E(G)$. \paralled
\end{rem}

\subsection{Integer-partitioned and integer-decomposed strings}

\begin{defn}\label{defn:integer-partitioned-strings}
$^*$ Suppose that a positive integer $m$ can be partitioned into a sum
\begin{equation}\label{eqa:intege-m-partitioned-00}
m=m_{i,1}+m_{i,2}+\cdots +m_{i,a_i},~m_{i,j}>0,~a_i\geq 2
\end{equation} for $i\in [1,P_{art}(m)]$, where $P_{art}(m)$ is the number of all possible partitions of the positive integer $m$. Then we have:

(i) Directly, there are \emph{$m$-partitioned number-based strings}
\begin{equation}\label{eqa:intege-m-partitioned-11}
s^*_i=m_{i,1}m_{i,2}\dots m_{i,a_i},~i\in [1,P_{art}(m)]
\end{equation}

(ii) If a number $m_{i,j}$ appeared in Eq.(\ref{eqa:intege-m-partitioned-00}) can be partitioned into a sum
\begin{equation}\label{eqa:m-partitioned-twin-numbers11}
m_{i,j}=b_{i,j,1}+b_{i,j,2}+\cdots +b_{i,j,s(i,j)},~b_{i,j,s}>0,~s(i,j)\geq 2
\end{equation} then we get a \emph{compound $m$-partitioned number-based string}
\begin{equation}\label{eqa:intege-m-partitioned-22}
s^j_i=m_{i,1}m_{i,2}\dots m_{i,j-1}~b_{i,j,1}b_{i,j,2}\dots b_{i,j,s(i,j)}~m_{i,j+1}\dots m_{i,a_i}
\end{equation}

(iii) Let $m_{i,j,1},m_{i,j,2},\dots ,m_{i,j,a_i}$ be the $j$th permutation of the numbers $m_{i,1},m_{i,2},\dots ,m_{i,a_i}$ for $j\in[1,(a_i)!]$ and $i\in [1,P_{art}(m)]$. We have $m$-partitioned number-based strings
\begin{equation}\label{eqa:intege-m-partitioned-33}
s_{i,j}=m_{i,j,1}m_{i,j,2}\cdots m_{i,j,a_i},~j\in[1,(a_i)!],~i\in [1,P_{art}(m)]
\end{equation}

There are $\prod^{P_{art}(m)}_{i=1}(a_i)!$ different $m$-partitioned number-based strings of form $s_{i,j}$ defined in Eq.(\ref{eqa:intege-m-partitioned-33}).\qqed
\end{defn}

\begin{problem}\label{qeu:444444}
Given a number-based string $s=c_1c_2\cdots c_n$ with $c_i\in [0,9]$ and $c_1\geq 1$, \textbf{find} a positive integer $m$, such that $m$ can be expressed as a sum $m=m_{i,1}+m_{i,2}+\cdots +m_{i,a_i}$ defined in Eq.(\ref{eqa:intege-m-partitioned-00}), and there is $s=s^*_i$ defined in Eq.(\ref{eqa:intege-m-partitioned-11}), or $s=s^j_i$ defined in Eq.(\ref{eqa:intege-m-partitioned-22}), or $s=s_{i,j}$ defined in Eq.(\ref{eqa:intege-m-partitioned-33}).
\end{problem}

\begin{defn}\label{defn:integer-decomposed-strings}
$^*$ Suppose that a positive integer $m$ can be decomposed into a product
\begin{equation}\label{eqa:intege-n-decomposed-00}
n=p_{i,1}\cdot p_{i,2}\cdot\cdots \cdot p_{i,b_i},~p_{i,j}\geq 3,~b_i\geq 2,~i\in [1,D_{com}(n)]
\end{equation} where $D_{com}(n)$ is the number of all possible decompositions of the positive integer $n$. We have:

(i) there are \emph{$n$-decomposed number-based strings}
\begin{equation}\label{eqa:intege-n-decomposed-11}
r^*_i=p_{i,1}\cdot p_{i,2}\cdot \dots \cdot p_{i,b_i},~i\in [1,D_{com}(n)]
\end{equation}

(ii) If a number $p_{i,j}$ appeared in Eq.(\ref{eqa:intege-n-decomposed-00}) can be decomposed into a product
\begin{equation}\label{eqa:intege-n-decomposed-public-key}
p_{i,j}=c_{i,j,1}\cdot c_{i,j,2}\cdot \cdots \cdot c_{i,j,t(i,j)},~c_{i,j,s}>0,~t(i,j)\geq 2
\end{equation} then we get a \emph{compound $n$-decomposed number-based string}
\begin{equation}\label{eqa:intege-n-decomposed-22}
r^j_i=p_{i,1}\cdot p_{i,2}\cdot \dots \cdot p_{i,j-1}~\cdot c_{i,j,1}\cdot c_{i,j,2}\cdot \dots \cdot c_{i,j,s(i,j)}\cdot ~p_{i,j+1}\dots \cdot p_{i,b_i}
\end{equation}

(iii) Let $p_{i,j,1},p_{i,j,2},\dots ,p_{i,j,b_i}$ be the $j$th permutation of the numbers $p_{i,1},p_{i,2},\dots ,p_{i,b_i}$ for $j\in[1,(b_i)!]$ and $i\in [1,D_{com}(n)]$. We have $n$-decomposed number-based strings
\begin{equation}\label{eqa:intege-n-decomposed-33}
r_{i,j}=p_{i,j,1}\cdot p_{i,j,2}\cdot \cdots \cdot p_{i,j,b_i},~j\in[1,(b_i)!],~i\in [1,D_{com}(n)]
\end{equation}

In general, there are $\prod^{D_{com}(n)}_{i=1}(b_i)!$ different $m$-decomposed number-based strings of form $r_{i,j}$ defined in Eq.(\ref{eqa:intege-n-decomposed-33}).\qqed
\end{defn}

\begin{problem}\label{qeu:444444}
Given a number-based string $s^*=a_1a_2\cdots a_m$ with $a_i\in [0,9]$ and $a_1\geq 1$, \textbf{find} a positive integer $n$, such that $n$ can be expressed as a product $n=p_{i,1} p_{i,2}\cdots p_{i,b_i}$ defined in Eq.(\ref{eqa:intege-n-decomposed-00}), and there is $s^*=r^*_i$ defined in Eq.(\ref{eqa:intege-n-decomposed-11}), or $s^*=r^j_i$ defined in Eq.(\ref{eqa:intege-n-decomposed-22}), or $s^*=r_{i,j}$ defined in Eq.(\ref{eqa:intege-n-decomposed-33}).
\end{problem}

\begin{example}\label{exa:8888888888}
We show some number-based strings as \emph{public-key strings} and \emph{private-key strings} made by the integer partition and the integer decomposition.

(i) \textbf{The twin-sum string authentication}. In Eq.(\ref{eqa:intege-m-partitioned-11}), two numbers $m_{i,j}$ and $m_{i,j+1}$ are called \emph{twin numbers}, by Eq.(\ref{eqa:m-partitioned-twin-numbers11}) and the following Eq.(\ref{eqa:m-partitioned-twin-numbers22})
\begin{equation}\label{eqa:m-partitioned-twin-numbers22}
m_{i,j+1}=b_{i,j+1,1}+b_{i,j+1,2}+\cdots +b_{i,j+1,s(i,j+1)},~b_{i,j+1,s}>0,~s(i,j+1)\geq 2
\end{equation} we have a \emph{public-key string} $s^j_i$ shown in Eq.(\ref{eqa:intege-m-partitioned-22}) and the following \emph{private-key string}
\begin{equation}\label{eqa:m-partitioned-twin-private-key}
s^{j+1}_i=m_{i,1}m_{i,2}\cdots m_{i,j}~b_{i,j+1,1}b_{i,j+1,2}\cdots b_{i,j+1,s(i,j+1)}~m_{i,j+2}\cdots m_{i,a_i}
\end{equation} Finally, the following number-based string
\begin{equation}\label{eqa:m-partitioned-twin-authentication}
{
\begin{split}
m_{i,1}m_{i,2}\cdots m_{i,j-1}~b_{i,j,1}b_{i,j,2}\cdots b_{i,j,s(i,j)}~b_{i,j+1,1}b_{i,j+1,2}\cdots b_{i,j+1,s(i,j+1)}~m_{i,j+2}\cdots m_{i,a_i}
\end{split}}
\end{equation} is called \emph{twin-sum string authentication}, denoted as $A^{sum}_{uth}\langle s^j_i, s^{j+1}_i\rangle$.

(ii) \textbf{The twin-product string authentication}. In Eq.(\ref{eqa:intege-n-decomposed-00}), two numbers $p_{i,j}$ and $p_{i,j+1}$ are called \emph{twin numbers}, like Eq.(\ref{eqa:intege-n-decomposed-public-key}), we have
\begin{equation}\label{eqa:intege-n-decomposed-private-key}
p_{i,j+1}=c_{i,j+1,1} c_{i,j+1,2} \cdots c_{i,j+1,t(i,j+1)},~c_{i,j+1,s}>0,~t(i,j+1)\geq 2
\end{equation} Immediately, we get the following number-based string
\begin{equation}\label{eqa:intege-n-decomposed-public-private}
{
\begin{split}
A^{pro}_{uth}\langle r^j_i, r^{j+1}_i\rangle=&p_{i,1} p_{i,2} \cdots p_{i,j-1}~ c_{i,j,1} c_{i,j,2} \cdots c_{i,j,s(i,j)} \\
&c_{i,j+1,1} c_{i,j+1,2} \cdots c_{i,j+1,t(i,j+1)}~p_{i,j+1}\cdots p_{i,b_i}
\end{split}}
\end{equation} called \emph{twin-product string authentication}.

(iii) \textbf{The twin-mixed string authentication}. By Eq.(\ref{eqa:m-partitioned-twin-numbers11}), Eq.(\ref{eqa:m-partitioned-twin-numbers22}), Eq.(\ref{eqa:intege-n-decomposed-public-key}) and Eq.(\ref{eqa:intege-n-decomposed-private-key}), we have a \emph{twin-mixed string authentication}
\begin{equation}\label{eqa:mixed-partitioned-public-private}
{
\begin{split}
A^{mix}_{uth}\langle r^j_i, s^{j+1}_i\rangle=&m_{i,1}m_{i,2}\cdots m_{i,j-1}~c_{i,j,1} c_{i,j,2} \cdots c_{i,j,s(i,j)}~m_{i,j+1}\\
&b_{i,j+1,1}b_{i,j+1,2}\cdots b_{i,j+1,s(i,j+1)}~m_{i,j+2}\cdots m_{i,a_i}
\end{split}}
\end{equation}
where $r^j_i$ is as a \emph{public-key string} and $s^{j+1}_i$ is as a \emph{private-key string}; and moreover the following number-based string
\begin{equation}\label{eqa:mixed-decomposed-public-private}
{
\begin{split}
A^{mix}_{uth}\langle s^j_i, r^{j+1}_i\rangle=&p_{i,1} p_{i,2} \cdots p_{i,j-1}~b_{i,j,1}b_{i,j,2}\cdots b_{i,j,s(i,j)} \\
&c_{i,j+1,1} c_{i,j+1,2} \cdots c_{i,j+1,t(i,j+1)}~p_{i,j+1}\cdots p_{i,b_i}
\end{split}}
\end{equation} is called \emph{twin-mixed string authentication}, where $s^j_i$ is as a \emph{public-key string} and $r^{j+1}_i$ is as a \emph{private-key string}.

(iv) \textbf{The $A_m$-sum $B_n$-product complementary key-pairs}. We take $A$ numbers $m_{i,j_t}$ with $t\in [1,A]$ from a sum $m=m_{i,1}+m_{i,2}+\cdots +m_{i,a_i}$, and take $B$ numbers $p_{i,j_s}$ with $s\in [1,B]$ from a product $n=p_{i,1} p_{i,2}\cdots p_{i,b_i}$, and use them to make a \emph{public-key string}
$$s_{pub}(A_m,B_n)=\langle m_{i,j_t}|^A_{t=1},p_{i,j_s}|^B_{s=1}\rangle
$$ so we have two sets

$\overline{S}_m=\{m_{i,1},m_{i,2},\dots ,m_{i,a_i}\}\setminus\{m_{i,j_t}:t\in [1,A]\}$ and

$\overline{P}_n=\{p_{i,1}, p_{i,2},\dots ,p_{i,b_i}\}\setminus \{p_{i,j_s}:s\in [1,B]\}$,\\
the elements of them produce a \emph{private-key string} $s_{pri}(\overline{A}_m,\overline{B}_n)=\langle \overline{S}_m,\overline{P}_n\rangle$. We call $s_{pub}(A_m,B_n)$ and $s_{pri}(\overline{A}_m,\overline{B}_n)$ as a \emph{$A_m$-sum $B_n$-product complementary Key-pair}.

\vskip 0.4cm

It is meaningful to \textbf{do}:

(do-1) Analysis the computational complexity of the number-based strings defined above.

(do-2) Make more complex number-based strings by sums $m=m_{i,1}+m_{i,2}+\cdots +m_{i,a_i}$ and products $n=p_{i,1} p_{i,2}\cdots p_{i,b_i}$ for Key-pairs.\qqed
\end{example}

\begin{rem}\label{rem:333333}
\textbf{The computational complexity of the Integer Partition Problem.} We partition a positive integer $m$ into a group of $a_i$ parts $m_{i,1},m_{i,2},\dots ,m_{i,a_i}$ holding $m=m_{i,1}+m_{i,2}+\cdots +m_{i,a_i}$ with each $m_{i,j}>0$ and $a_i\geq 2$; refer to Eq.(\ref{eqa:intege-m-partitioned-00}). Suppose there is $P_{art}(m)$ groups of such $a_i$ parts. Computing the number $P_{art}(m)$ can be transformed into finding the number $A(m,a_i)$ of solutions of \emph{Diophantine equation} $m=\sum ^k_{i=1}ix_i$. There is a recursive formula
\begin{equation}\label{eqa:c3xxxxx}
A(m,a_i)=A(m,a_i-1)+A(m-a_i,a_i)
\end{equation}
with $0 \leq a_i\leq m$. It is not easy to compute the exact value of $A(m,a_i)$, for example, the authors in \cite{Shuhong-Wu-Accurate-2007} and \cite{WU-Qi-qi-2001} computed exactly
$${
\begin{split}
A(m,6)=&\biggr\lfloor \frac{1}{1036800}(12m^5 +270m^4+1520m^3-1350m^2-19190m-9081)+\\
&\frac{(-1)^m(m^2+9m+7)}{768}+\frac{1}{81}\left[(m+5)\cos \frac{2m\pi}{3}\right ]\biggr\rfloor
\end{split}}
$$

On the other hands, for any odd integer $m\geq 7$, it was conjectured $m=p_1+p_2+p_3$ with three primes $p_1,p_2,p_3$ from the famous Goldbach's conjecture: ``\emph{Every even integer, greater than 2, can be expressed as the sum of two primes}.'' In other word, determining $A(m,3)$ is difficult, also, it is difficult to express an odd integer $m=\sum^{3n}_{k=1} p\,'_k$ with each $p\,'_k$ is a prime number.

Clearly, it is difficult to partition an integer $m$ into $m=m_{i,1}+m_{i,2}+\cdots +m_{i,a_i}$ defined in Eq.(\ref{eqa:intege-m-partitioned-00}), or decompose an integer $n$ as $n=p_{i,1}\cdot p_{i,2}\cdot\cdots \cdot p_{i,b_i}$ defined in Eq.(\ref{eqa:intege-n-decomposed-00}).\paralled
\end{rem}

\subsection{Graphic lattices for making number-based strings}

\subsubsection{A series of operation graphic lattices}

Let $\textbf{\textrm{O}}=(O_1,O_2,\dots ,O_p)$ be a \emph{graph operation base}, where each $O_i$ is a graph operation of graph theory, and let $\textbf{\textrm{T}}_i=\big (T_{i,1},T_{i,2},\dots, T_{i,n_i}\big )$ be a \emph{colored graph base} for $n_i\geq 1$ and $i\in [1,M]$, so that each colored graph $T_{i,j}$ admits a total coloring (or a total graph-coloring) $f_{i,j}$. Suppose that each colored graph base $\textbf{\textrm{T}}_i$ is an every-zero additive group defined in Definition \ref{defn:every-zero-abstract-group}. By Eq.(\ref{eqa:edge-index-graphic-lattice11}), Eq.(\ref{eqa:edge-index-graphic-lattice11}) and Eq.(\ref{eqa:edge-index-graphic-lattice11}) in the TOTAL-graph-coloring algorithm-II with graph operation, there is a series of \emph{operation graphic lattices}.

First operation graphic lattice is
\begin{equation}\label{eqa:operation-edge-index-graphic-lattice11}
\textbf{\textrm{L}}(\textbf{\textrm{Q}}_1[\textbf{\textrm{O}}]\textbf{\textrm{T}}_1)=\left \{H[\textbf{\textrm{O}}]^{n_1}_{k=1}a_{1,k}T_{1,k}:~a_{1,k}\in Z^0,~T_{1,k}\in \textbf{\textrm{T}}_1, H\in \textbf{\textrm{Q}}_1 \right \}
\end{equation} with $\sum ^{n_1}_{k=1}a_{1,k}\geq 1$, where $\textbf{\textrm{Q}}_1$ is a graph set.

Let $\textbf{\textrm{Q}}_2=\textbf{\textrm{L}}(\textbf{\textrm{Q}}_1[\textbf{\textrm{O}}]\textbf{\textrm{T}}_1)$, the second operation graphic lattice is
\begin{equation}\label{eqa:operation-edge-index-graphic-lattice22}
\textbf{\textrm{L}}(\textbf{\textrm{Q}}_2[\textbf{\textrm{O}}]\textbf{\textrm{T}}_2)=\big \{H[\textbf{\textrm{O}}]^{n_2}_{k=1}a_{2,k}T_{2,k}:~a_{2,k}\in Z^0,~T_{2,k}\in \textbf{\textrm{T}}_2, H\in \textbf{\textrm{Q}}_2 \big \}
\end{equation} with $\sum ^{n_2}_{k=1}a_{2,k}\geq 1$.

In general, the $s$th operation graphic lattice is
\begin{equation}\label{eqa:operation-edge-index-graphic-lattice33}
\textbf{\textrm{L}}(\textbf{\textrm{Q}}_{s+1}[\textbf{\textrm{O}}]\textbf{\textrm{T}}_{s+1})=\big \{H[\textbf{\textrm{O}}]^{n_{s+1}}_{k=1}a_{s+1,k}T_{s+1,k}:~a_{s+1,k}\in Z^0,~T_{s+1,k}\in \textbf{\textrm{T}}_{s+1}, H\in \textbf{\textrm{Q}}_{s+1}\big \}
\end{equation} with $\sum ^{n_{s+1}}_{k=1}a_{s+1,k}\geq 1$, and $\textbf{\textrm{Q}}_{s+1}=\textbf{\textrm{L}}(\textbf{\textrm{Q}}_{s}[\textbf{\textrm{O}}]\textbf{\textrm{T}}_s)$ for $s\geq 1$.

\subsubsection{Tree-base graphic lattices}

In \cite{Bing-Yao-arXiv:2207-03381}, a \emph{tree base} $\textbf{\textrm{T}}=(T_1,T_2,\dots ,T_m)$ consists of mutually vertex-disjoint uncolored trees $T_1,T_2,\dots ,T_m$ holding $T_i\not \cong T_j$ if $i\neq j$. We use the edges of an edge set $E^*$ to join these mutually disjoint trees together, the resulting graph is denoted as $\textbf{\textrm{T}}+E^*$. Notice that one tree $T_i$ may be used two or more times in $\textbf{\textrm{T}}+E^*$, so we write $E^*[\ominus]^n_{k=1} \alpha_kT_k$ with $\sum ^m_{k=1}\alpha_k\geq 1$ for replacing $\textbf{\textrm{T}}+E^*$ by the \emph{edge-join operation} ``$[\ominus ]$''. We have a \emph{tree-base graphic lattice} as follows:
\begin{equation}\label{eqa:tree-edge-join-tree-lattice}
\textbf{\textrm{L}}(\textbf{\textrm{E}}[\ominus ]Z^0\textbf{\textrm{T}})=\big \{E^*[\ominus ]^m_{k=1}\mu_kT_k:~\mu_k\in Z^0,T_k\in \textbf{\textrm{T}},E^*\in \textbf{\textrm{E}}\big \}
\end{equation} with $\sum ^m_{k=1}\mu_k\geq 1$.

\vskip 0.4cm

Let $J_{i_1},J_{i_2},\dots ,J_{i_M}$ be a permutation of the trees $\lambda_1T_1$, $\lambda_2T_2$, $\dots $, $\lambda_mT_m$ for $M=\sum ^m_{k=1}\lambda_k\geq 1$. By the vertex-coinciding operation ``$[\odot ]$'', we vertex-coincide some vertices $x_{i_j,1},x_{i_j,2},\dots ,x_{i_j,s}$ of the tree $J_{i_j}$ with some vertices $y_{i_{j+1},1},y_{i_{j+1},2},\dots ,y_{i_{j+1},t}$ of the tree $J_{i_{j+1}}$ into
$$\{x_{i_j,1},x_{i_j,2},\dots ,x_{i_j,s}\}\odot \{y_{i_{j+1},1},y_{i_{j+1},2},\dots ,y_{i_{j+1},t}\},~j\in [1,M-1]
$$ the resultant graph is denoted as
\begin{equation}\label{eqa:555555}
J_{i_1}[\odot ]J_{i_2}[\odot ]\dots [\odot ]J_{i_M}=[\odot ]^M_{j=1}J_{i_j}=[\odot ]^m_{k=1}a_kT_k
\end{equation}
which induces a \emph{tree-base graphic lattice}
\begin{equation}\label{eqa:tree-vertex-odot-tree-lattice}
\textbf{\textrm{L}}([\odot ]Z^0\textbf{\textrm{T}})=\big \{[\odot ]^m_{k=1}\lambda_kT_k:~\lambda_k\in Z^0,T_k\in \textbf{\textrm{T}}\big \}
\end{equation} with $\sum ^m_{k=1}\lambda_k\geq 1$.

\begin{problem}\label{qeu:444444}
About tree-base graphic lattices defined in Eq.(\ref{eqa:tree-vertex-odot-tree-lattice}), \textbf{is} there
$$\textbf{\textrm{L}}([\odot ]Z^0\textbf{\textrm{T}})\cap \textbf{\textrm{L}}([\odot ]Z^0\textbf{\textrm{H}})\neq \emptyset$$
for different tree-bases $\textbf{\textrm{T}}$ and $\textbf{\textrm{H}}$?
\end{problem}

We do the \emph{edge-join operation} ``$[\ominus ]$'' and the vertex-coinciding operation ``$[\odot ]$'' to a permutation $H_{i_1},H_{i_2},\dots ,H_{i_M}$ of the trees $\gamma_1T_1$, $\gamma_2T_2$, $\dots $, $\gamma_mT_m$ for $M=\sum ^m_{k=1}\gamma_k\geq 1$, and write the resultant graphs as $[\ominus \odot]^m_{k=1}\gamma_kT_k$, immediately, we get a \emph{tree-base graphic lattice}
\begin{equation}\label{eqa:tree-vertex-edge-mixed-tree-lattice}
\textbf{\textrm{L}}([\ominus \odot]Z^0\textbf{\textrm{T}})=\big \{[\ominus \odot]^m_{k=1}\gamma_kT_k:~\gamma_k\in Z^0,T_k\in \textbf{\textrm{T}}\big \}
\end{equation} with $\sum ^m_{k=1}\gamma_k\geq 1$.

\vskip 0.2cm

\textbf{All-tree-graphic lattice.} For considering the coloring closure of graphic lattices, we let $\textbf{\textrm{L}}_{\textrm{tree}}(\textbf{\textrm{E}}[\ominus ]Z^0\textbf{\textrm{T}})$ be the set of all trees in a tree-base graphic lattice $\textbf{\textrm{L}}(\textbf{\textrm{E}}[\ominus ]Z^0\textbf{\textrm{T}})$ defined in Eq.(\ref{eqa:tree-edge-join-tree-lattice}), and we call it \emph{all-tree-graphic lattice}. Similarly, we have other two all-tree-graphic lattices
$$\textbf{\textrm{L}}_{\textrm{tree}}([\odot]Z^0\textbf{\textrm{T}})\subset \textbf{\textrm{L}}([\odot]Z^0\textbf{\textrm{T}}), ~\textbf{\textrm{L}}_{\textrm{tree}}([\ominus \odot]Z^0\textbf{\textrm{T}})\subset\textbf{\textrm{L}}([\ominus \odot]Z^0\textbf{\textrm{T}})
$$ defined in Eq.(\ref{eqa:tree-vertex-odot-tree-lattice}) and Eq.(\ref{eqa:tree-vertex-edge-mixed-tree-lattice}), respectively.

\vskip 0.4cm

By Theorem \ref{thm:10-k-d-total-coloringss}, these three all-tree-graphic lattices are \emph{coloring closure} to the $(k,d)$-total colorings stated in the following theorem \ref{thm:coloring-closure-kd-total-colorings}:

\begin{thm}\label{thm:coloring-closure-kd-total-colorings}
\cite{Bing-Yao-arXiv:2207-03381} Each tree contained in $\textbf{\textrm{L}}_{\textrm{tree}}(\textbf{\textrm{E}}[\ominus ]Z^0\textbf{\textrm{T}})$, $\textbf{\textrm{L}}_{\textrm{tree}}([\odot]Z^0\textbf{\textrm{T}})$ and $\textbf{\textrm{L}}_{\textrm{tree}}([\ominus \odot]Z^0\textbf{\textrm{T}})$ admits the following $(k,d)$-total colorings: graceful $(k,d)$-total coloring, harmonious $(k,d)$-total coloring, (odd-edge) edge-difference $(k,d)$-total coloring, (odd-edge) graceful-difference $(k,d)$-total coloring, (odd-edge) felicitous-difference $(k,d)$-total coloring, (odd-edge) edge-magic $(k,d)$-total coloring.
\end{thm}

\subsubsection{Tree-base $[\odot_{\textrm{prop}}]$-operation graphic lattices}

Let $\textbf{\textrm{T}}=(T_1,T_2,\dots ,T_m)$ be a \emph{tree base}, where the mutually vertex-disjoint uncolored trees $T_1,T_2$, $\dots $, $T_m$ holds $T_i\not \cong T_j$ if $i\neq j$.

\textbf{Tree-base $[\odot_{\textrm{prop}}]$-operation graphic lattice.} A \emph{proper vertex-coinciding operation} ``$[\odot_{\textrm{prop}}]$'' to the trees of the tree base $\textbf{\textrm{T}}$ produces graphs $H=[\odot_{\textrm{prop}}]^n_{j=1}T_j$ holding:

(a-1) there is a graph homomorphism $\big \langle \bigcup^m_{j=1}T_j\big \rangle \rightarrow H$;

(a-2) $H$ is connected;

(a-3) $H$ has no multiple-edge, and $E(H)= \bigcup^m_{j=1}E(T_j)$.

\vskip 0.4cm

Moreover, we take a permutation $L_{i_1},L_{i_2},\dots ,L_{i_A}$ of the trees $b_1T_1$, $b_2T_2$, $\dots $, $b_mT_m$ from the tree base $\textbf{\textrm{T}}$ with $A=\sum ^m_{k=1}b_k\geq 1$. Then we get graphs
\begin{equation}\label{eqa:proper-operation-tree-base-lattice111}
G=[\odot_{\textrm{prop}}]^A_{j=1}L_{i_j}=[\odot_{\textrm{prop}}]^m_{k=1}b_kT_k
\end{equation} after doing the \emph{proper vertex-coinciding operation} ``$[\odot_{\textrm{prop}}]$'' to the permutation $L_{i_1},L_{i_2},\dots ,L_{i_A}$, such that

(b-1) there is a graph homomorphism $\big \langle \bigcup^A_{j=1}L_{i_j}\big \rangle \rightarrow G$;

(b-2) $G$ is connected;

(b-3) $G$ has no multiple-edge, and $E(G)= \bigcup^A_{j=1}E(L_{i_j})$.\\
Thereby, we get a \emph{tree-base $[\odot_{\textrm{prop}}]$-operation graphic lattice}
\begin{equation}\label{eqa:proper-operation-tree-base-lattice}
\textbf{\textrm{L}}([\odot_{\textrm{prop}}]Z^0\textbf{\textrm{T}})=\big \{[\odot_{\textrm{prop}}]^m_{k=1}b_kT_k:~b_k\in Z^0,T_k\in \textbf{\textrm{T}}\big \}
\end{equation} with $\sum ^m_{k=1}b_k\geq 1$.

\begin{problem}\label{qeu:444444}
\textbf{Is} there another tree base $\textbf{\textrm{H}}=(H_1,H_2,\dots ,H_n)$ with $H_i\not \cong H_j$ if $i\neq j$, such that two tree-base $[\odot_{\textrm{prop}}]$-operation graphic lattices hold $\textbf{\textrm{L}}([\odot_{\textrm{prop}}]Z^0\textbf{\textrm{H}})=\textbf{\textrm{L}}([\odot_{\textrm{prop}}]Z^0\textbf{\textrm{T}})$ true?
\end{problem}

\begin{rem}\label{rem:333333}
Clearly, each connected graph $G\in \textbf{\textrm{L}}([\odot_{\textrm{prop}}]Z^0\textbf{\textrm{T}})$ defined in Eq.(\ref{eqa:proper-operation-tree-base-lattice}) can be vertex-split (also, decompose) into the trees $b_1T_1$, $b_2T_2$, $\dots $, $b_mT_m$ from the tree base $\textbf{\textrm{T}}$ with $\sum ^m_{k=1}b_k\geq 1$; refer to Problem \ref{qeu:verious-tree-decompositions}. There are the following particular cases about the lattices $\textbf{\textrm{L}}([\odot_{\textrm{prop}}]Z^0\textbf{\textrm{T}})$:

\textbf{Case-1.} If each tree $T_i$ in the tree base $\textbf{\textrm{T}}$ holds $|V(T_i)|=p$ true, then the connected graph $G\in \textbf{\textrm{L}}([\odot_{\textrm{prop}}]Z^0\textbf{\textrm{T}})$ is \emph{spanning-tree decomposable} as $|V(G)|=p$.

\textbf{Case-2.} If each tree $T_i$ in the tree base $\textbf{\textrm{T}}$ is a \emph{caterpillar} (resp. \emph{lobster}), then we say that the connected graph $G\in \textbf{\textrm{L}}([\odot_{\textrm{prop}}]Z^0\textbf{\textrm{T}})$ is \emph{caterpillar decomposable} (resp. \emph{lobster decomposable}).

\textbf{Case-3.} \textbf{The extremum problems.} Let $\textbf{\textrm{L}}(b_1,b_2,\dots, b_m)$ be the set of graphs in $\textbf{\textrm{L}}([\odot_{\textrm{prop}}]Z^0\textbf{\textrm{T}})$ for a fixed group of non-negative integers $b_1,b_2,\dots, b_m$ such that each connected graph $G\in \textbf{\textrm{L}}(b_1,b_2,\dots, b_m)$ holding $G=[\odot_{\textrm{prop}}]^m_{k=1}b_kT_k$. Motivated from some traditional lattice problems, we propose the following extremum problems: \textbf{Find} a connected graph $H^*\in \textbf{\textrm{L}}(b_1,b_2,\dots, b_m)$, such that

\begin{asparaenum}[(i) ]
\item the vertex numbers $|V(H^*)|\leq |V(G)|$ for any connected graph $G\in \textbf{\textrm{L}}(b_1,b_2,\dots, b_m)$.

\item the diameters $D(H^*)\leq D(G)$ for any connected graph $G\in \textbf{\textrm{L}}(b_1,b_2,\dots, b_m)$.

\item the maximal degrees $\Delta(H^*)\geq \Delta(G)$ for any connected graph $G\in \textbf{\textrm{L}}(b_1,b_2,\dots, b_m)$.

\item $H^*$ contains a path $P^*$ holding $|V(P^*)|\geq |V(P)|$ for any path $P$ contained in any connected graph $G\in \textbf{\textrm{L}}(b_1,b_2,\dots, b_m)$.

\item the number of cycles of $H^*$ is largest in $\textbf{\textrm{L}}(b_1,b_2,\dots, b_m)$.

\item the combinatorics of the above extremum cases.

\item \textbf{determine} the cardinality $|\textbf{\textrm{L}}(b_1,b_2,\dots, b_m)|$.
\end{asparaenum}

\textbf{Case-4.} \textbf{The set-colorings of connected graphs in $\textbf{\textrm{L}}([\odot_{\textrm{prop}}]Z^0\textbf{\textrm{T}})$.} Since each tree $T_i$ with $i\in [1,m]$ admits a group of colorings $h_1,h_2,\dots ,h_n$, then each connected $(p,q)$-graph $G\in \textbf{\textrm{L}}([\odot_{\textrm{prop}}]Z^0\textbf{\textrm{T}})$ admits a set-coloring $F=\langle h_1,h_2,\dots ,h_n\rangle $, consider the following questions:
\begin{asparaenum}[\textbf{\textrm{Scola}}-1.]
\item If each $h_i$ is a set-ordered graceful labeling, then the connected $(p,q)$-graph $G$ admits a set-coloring $F=\langle h_1,h_2,\dots ,h_n\rangle $ such that the edge color set $F(E(G))=[1,q]$.
\item If each $h_i$ is a set-ordered odd-graceful labeling, then the connected $(p,q)$-graph $G$ admits a set-coloring $F=\langle h_1,h_2,\dots ,h_n\rangle $ such that the edge color set $F(E(G))=[1,2q-1]^o$.
\item If each $h_i$ is one coloring/labeling defined in Definition \ref{defn:basic-W-type-labelings}, Definition \ref{defn:kd-w-type-colorings} and Definition \ref{defn:odd-edge-W-type-total-labelings-definition}, what property does the connected $(p,q)$-graph $G$ have such that $G$ admits a total set-coloring $F$ defined by $h_1$, $h_2$, $\dots $, $h_n$?\paralled
\end{asparaenum}
\end{rem}

\begin{problem}\label{qeu:444444}
For a fixed group of non-negative integers $b_1,b_2,\dots, b_m$ with $\sum ^m_{k=1}b_k\geq 1$, let $\textbf{\textrm{L}}(b_1,b_2,\dots, b_m)$ be the set of graphs of the graphic lattice $\textbf{\textrm{L}}([\odot_{\textrm{prop}}]Z^0\textbf{\textrm{T}})$ defined in Eq.(\ref{eqa:proper-operation-tree-base-lattice}).

(i) If $K_{n}=[\odot_{\textrm{prop}}]^m_{i=1}T_i$ with $n=|V(T_i)|$ for $i\in [1,m]$ in the tree base $\textbf{\textrm{T}}$, then the set $\textbf{\textrm{L}}(b_1,b_2,\dots, b_m)$ contains graphs with cliques $K_{n}$. \textbf{Determine} a connected graph $G\in \textbf{\textrm{L}}(b_1,b_2,\dots$, $b_m)$ such that $G$ contains the largest number of cliques like $K_{n}$.

(ii) If $K_{m}=[\odot_{\textrm{prop}}]^m_{k=1}T_k$ with the vertex number $|V(T_i)|=k$ for $k\in [1,m]$ in the tree base $\textbf{\textrm{T}}$, that is $\langle T_1,T_2,\dots$, $ T_m\mid K_m\rangle$; refer to Conjecture \ref{conj:c2-KT-conjecture}. \textbf{Find} a connected graph $G\in \textbf{\textrm{L}}(b_1,b_2,\dots, b_m)$ such that $G$ contains the most cliques of form $K_{m}$.
\end{problem}

\begin{conj} \label{conj:c2-KT-conjecture}
$K$-$T$ \textbf{conjecture} (Gy\'{a}r\'{a}s and Lehel, 1978; B\'{e}la Bollob\'{a}s, 1995): For integer $n\geq 3$, given $n$ mutually disjoint trees $T_k$ of $k$ vertices with respect to $1\leq k\leq n$. Then the complete graph $K_n$ can be decomposed into the union of $n$ mutually edge-disjoint trees $H_k$, namely $K_n=\bigcup ^{n}_{k=1}H_k$, such that $T_k\cong H_k$ whenever $1\leq k\leq n$. Also, write this case as $\langle T_1,T_2,\dots, T_n\mid K_n\rangle$.
\end{conj}

\begin{defn} \label{defn:totally-graceful-split-tree-groups}
\cite{Yao-Su-Ma-Wang-Yang-arXiv-2202-03993v1} Suppose that a complete graph $K_n$ admits a total coloring $f:V(K_n)\cup E(K_n)\rightarrow [1,n]$ holding $|f(V(K_n))|=n$ and $f(E(K_n))=\{f(uv)=|f(u)-f(v)|:uv\in E(K_n)\}$. Doing the vertex-split operation to $K_n$ produces a group of trees $T_1,T_{2},T_{3},\dots ,T_{n}$, where each tree $T_{k}$ has just $k$ vertices for $k\in [1,n]$, such that $K_n=\odot|^n_{k=1}T_{k}$, also $\langle T_1,T_{2},T_{3},\dots ,T_{n}\mid K_n\rangle $ (refer to Conjecture \ref{conj:c2-KT-conjecture}). We call the set $S=\{T_1,T_{2},T_{3},\dots ,T_{n}\}$ a \emph{vertex-split tree-group} of $K_n$. Let $f_k$ be the total coloring of each tree $T_{k}$, so $f_k(x)\neq f_k(y)$ for $x,y\in V(T_{k})$, if
$$f_k(E(T_{k}))=\big \{f_k(uv)=|f_k(u)-f_k(v)|:uv\in E(T_{k})\big \}=[1,k-1]
$$ we call this set $S$ \emph{totally graceful vertex-split tree-group} of $K_n$.
\qqed
\end{defn}

\begin{figure}[h]
\centering
\includegraphics[width=16.4cm]{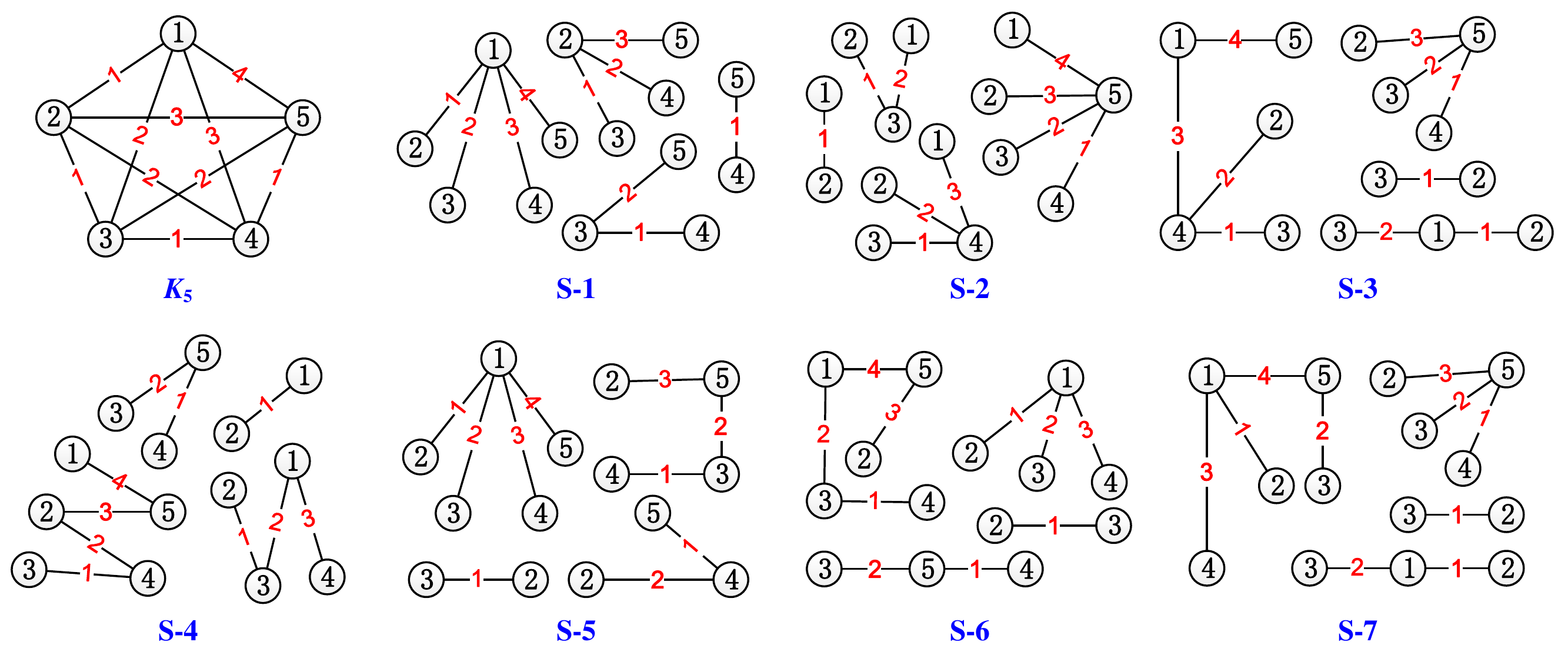}\\
\caption{\label{fig:total-graceful-group}{\small Seven totally graceful vertex-split tree-groups of $K_5$ for understanding Definition \ref{defn:totally-graceful-split-tree-groups}, cited from \cite{Yao-Su-Ma-Wang-Yang-arXiv-2202-03993v1}.}}
\end{figure}

\begin{problem}\label{qeu:444444}
\textbf{Does} each complete graph $K_n$ have a totally graceful vertex-split tree-group defined in Definition \ref{defn:totally-graceful-split-tree-groups}?
\end{problem}

\begin{defn} \label{defn:55-v-set-e-proper-more-labelings}
\cite{Yao-Wang-2106-15254v1} A \emph{$W$-constraint v-set e-proper total coloring} of a connected graph $G$ is a mapping $f:V(G)\cup E(G)\rightarrow \Omega$, where $\Omega$ consists of numbers and sets, such that

(i) $f(u)$ for each vertex $u\in V(G)$ is a set;

(ii) $f(xy)$ for each edge $xy\in E(G)$ is a number $c$; and

(iii) there are some $a\in f(u)$ and $b\in f(v)$ holding the $W$-constraint $c=W\langle a,b\rangle $.\qqed
\end{defn}

By Theorem \ref{thm:K-2m-spanning-trees-2m-1-edges} and Theorem \ref{thm:10-k-d-total-coloringss}, we have

\begin{thm}\label{thm:666666}
$^*$ Each complete graph $K_{2m}$ of $2m$ vertices admits a graceful v-set e-proper total coloring $f$ such that the edge set $E(K_{2m})=\bigcup^m_{k=1}E_k$, and there are some $a_u\in f(u)$ and $b_v\in f(v)$, and $f(uv)=k=|a_u-b_v|$ for each edge $uv\in E_k$ with $k\in [1,m]$.
\end{thm}

\begin{thm}\label{thm:graceful-v-set-e-proper-co}
$^*$ For a fixed group of non-negative integers $b_1,b_2,\dots, b_m$ with $B=\sum ^m_{k=1}b_k\geq 1$, let $\textbf{\textrm{L}}(b_1,b_2,\dots, b_m)$ be the set of graphs in the lattice $\textbf{\textrm{L}}([\odot_{\textrm{prop}}]Z^0\textbf{\textrm{T}})$ defined in Eq.(\ref{eqa:proper-operation-tree-base-lattice}) such that each connected $(p,q)$-graph $G\in \textbf{\textrm{L}}(b_1,b_2,\dots, b_m)$ holds $G=[\odot_{\textrm{prop}}]^m_{k=1}b_kT_k$ true. If each tree $T_i$ of the lattice base $\textbf{\textrm{T}}$ admits a \emph{graceful $(k,d)$-total coloring} $f_i$ for $i\in [1,m]$, then each connected $(p,q)$-graph $G\in \textbf{\textrm{L}}(b_1,b_2,\dots, b_m)$ admits at least $B!$ different \emph{graceful v-set e-proper colorings} $F_k$ such that each edge color set $F_k(E(G))=[1,q]$ for $k\in [1,B!]$.
\end{thm}
\begin{proof} By Theorem \ref{thm:10-k-d-total-coloringss}, each tree $T$ with diameter $D(T)\geq 3$ admits at least $2^m$ different graceful $(k,d)$-total colorings for $m+1=\Big \lceil \frac{D(T)}{2}\Big \rceil $.

Notice that each tree $T_i\in \textbf{\textrm{T}}$ is a bipartite and connected graph with its own vertex set $V(T_i)=X_i\cup Y_i$ and $X_i\cap Y_i=\emptyset $. According to Definition \ref{defn:basic-W-type-labelings}, a graceful $(k,d)$-total coloring $f_i$ of each tree $T_i$ is defined as
\begin{equation}\label{eqa:graceful-v-set-e-proper-00}
{
\begin{split}
&f_i:X_i\rightarrow S_{m,0,0,d}=\{0,d,\dots ,md\}, \\
&f_i:Y_i\cup E(T_i)\rightarrow S_{q_i-1,k,0,d}=\{k,k+d,\dots ,k+(q_i-1)d\}
\end{split}}
\end{equation} where $q_i=|E(T_i)|$, and it is allowed $f_i(u)=f_i(w)$ for some distinct vertices $u,w\in V(T_i)$, such that the edge color set
\begin{equation}\label{eqa:graceful-v-set-e-proper-11}
f(E(T_i))=\{f(uv)=|f(u)-f(v)|:uv\in E(T_i)\}=\{k,k+d,\dots ,k+(q_i-1)d\}
\end{equation}

For a permutation $H_{i_1},H_{i_2},\dots ,H_{i_B}$ of the trees $b_1T_1$, $b_2T_2$, $\dots $, $b_mT_m$ from the tree base $\textbf{\textrm{T}}$ with $B=\sum ^m_{k=1}b_k\geq 1$, we have connected graphs
\begin{equation}\label{eqa:proper-operation-tree-base-lattice-coloring}
G=[\odot_{\textrm{prop}}]^B_{j=1}H_{i_j}=[\odot_{\textrm{prop}}]^m_{k=1}b_kT_k
\end{equation} and each tree $H_{i_j}$ admits a graceful $(k_{i_j},d_{i_j})$-total coloring $g_{i_j}$ defined as
\begin{equation}\label{eqa:graceful-v-set-e-proper-22}
{
\begin{split}
&g_{i_j}:X_{i_j}\rightarrow S_{m_{i_j},0,0,d_{i_j}}=\big \{0,d_{i_j},\dots ,m_{i_j}d_{i_j}\big \}, \\
&g_{i_j}:Y_{i_j}\cup E(H_{i_j})\rightarrow S_{q_{i_j}-1,k_{i_j},0,d_{i_j}}=\big \{k_{i_j},k_{i_j}+d_{i_j},\dots ,k_{i_j}+(q_{i_j}-1)d_{i_j}\big \}
\end{split}}
\end{equation} where $q_{i_j}=|E(H_{i_j})|$ and the vertex set $V(H_{i_j})=X_{i_j}\cup Y_{i_j}$ with $X_{i_j}\cap Y_{i_j}=\emptyset $, and it is allowed $g_{i_j}(x)=g_{i_j}(u)$ for some distinct vertices $x,u\in V(T_i)$, such that the edge color set
\begin{equation}\label{eqa:graceful-v-set-e-proper-33}
{
\begin{split}
g_{i_j}(E(H_{i_j}))&=\big \{g_{i_j}(uv)=|g_{i_j}(u)-g_{i_j}(v)|:uv\in E(H_{i_j})\big \}\\
&=\big \{k_{i_j},k_{i_j}+d_{i_j},\dots ,k_{i_j}+(q_{i_j}-1)d_{i_j}\big \}
\end{split}}
\end{equation}

Let $\Sigma (j)=\sum ^{j}_{k=1}q_{i_k}$ for $j\in [1,B]$. For $(k_{i_1},d_{i_1})=(1,1)$, we get the edge color set
$$g_{i_1}(E(H_{i_1}))=\big \{1,2,\dots ,q_{i_1}\big \}=[1,q_{i_1}]=[1,\Sigma (1)]
$$ As $(k_{i_2},d_{i_2})=(\Sigma (1)+1,1)$, we get the edge color set
$$g_{i_2}(E(H_{i_2}))=\big \{\Sigma (1)+1,\Sigma (1)+2,\dots ,\Sigma (1)+1+(q_{i_2}-1)\big \}=[\Sigma (1)+1,\Sigma (2)]
$$

In general, we take $(k_{i_{j+1}},d_{i_{j+1}})=(\Sigma (j)+1,1)$, then we get the edge color set
$$g_{i_{j+1}}(E(H_{i_{j+1}}))=\big \{\Sigma (j)+1,\Sigma (j)+2,\dots ,\Sigma (j)+1+(q_{i_{j+1}}-1)\big \}=\left [\Sigma (j)+1,\Sigma (j+1)\right ]
$$
For the last $(k_{i_B},d_{i_B})=(\Sigma (B-1)+1,1)$, we have the edge color set
$$g_{i_B}(E(H_{i_B}))=\big \{\Sigma (B-1)+1,\Sigma (B-1)+2,\dots ,\Sigma (B-1)+1+(q_{i_B}-1)\big \}=\left [\Sigma (B-1)+1,\Sigma (B)\right ]
$$

By Eq.(\ref{eqa:proper-operation-tree-base-lattice-coloring}) and the permutation $H_{i_1},H_{i_2},\dots ,H_{i_B}$ of the trees $b_1T_1,b_2T_2,\dots ,b_mT_m$, the connected graph $G$ admits a graceful v-set e-proper coloring $F_i=\langle g_{i_1},g_{i_2},\dots $, $g_{i_B}\rangle$ defined by setting $g_{i_j}(x)\subseteq F_i(x)$, $g_{i_j}(y)\subseteq F_i(y)$ and $F_i(xy)=g_{i_j}(xy)$ if edge $xy\in E(H_{i_j})\subseteq E(G)$, such that the edge color set
$$
F_i(E(G))=\big \{F_i(uv)=|F_i(u)-F_i(v)|:uv\in E(G)\big \}=\bigcup ^B_{j=1}g_{i_j}(E(H_{i_j}))=[1,q]
$$ where $F_i(uv)=|F_i(u)-F_i(v)|$ stands for $g_{i_j}(uv)=|g_{i_j}(u)-g_{i_j}(v)|$ if edge $uv\in E(H_{i_j})$.

Since there are $B!$ permutations $H_{i_1},H_{i_2},\dots ,H_{i_B}$ of the trees $b_1T_1,b_2T_2,\dots ,b_mT_m$ for $i\in [1,B!]$ with $B=\sum ^m_{k=1}b_k\geq 1$, we have shown that the connected graph $G\in \textbf{\textrm{L}}(b_1,b_2,\dots, b_m)$ admits $B!$ different graceful v-set e-proper colorings $f$ with the edge color set $f(E(G))=[1,q]$.

The proof of the theorem is complete.
\end{proof}

\begin{cor}\label{cor:others-v-set-e-proper-co}
$^*$ For a fixed group of non-negative integers $a_1,a_2,\dots, a_m$ with $A=\sum ^m_{k=1}a_k\geq 1$, let $\textbf{\textrm{L}}(a_1,a_2,\dots, a_m)$ be the set of graphs in the lattice $\textbf{\textrm{L}}([\odot_{\textrm{prop}}]Z^0\textbf{\textrm{T}})$ defined in Eq.(\ref{eqa:proper-operation-tree-base-lattice}) such that each connected $(p,q)$-graph $G\in \textbf{\textrm{L}}(a_1,a_2,\dots, a_m)$ holding $G=[\odot_{\textrm{prop}}]^m_{k=1}a_kT_k$. If each tree $T_i$ in the lattice base $\textbf{\textrm{T}}$ admits a \emph{graceful $(k,d)$-total coloring} $\varphi_i$ for $i\in [1,m]$, then each connected $(p,q)$-graph $G\in \textbf{\textrm{L}}(a_1,a_2,\dots, a_m)$ admits at least $A!$ different \emph{$W$-constraint v-set e-proper colorings} $\theta_k=\langle \varphi_1,\varphi_2,\dots ,\varphi_m\rangle $ such that

(i) As $W$-constraint = \emph{harmonious}, so each edge color set $\theta_k(E(G))=[1,q]$ for $k\in [1,A!]$.

(ii) As $W$-constraint = \emph{odd-graceful}, then each edge color set $\theta_k(E(G))=[1,2q-1]^o$ for $k\in [1,A!]$.

(iii) As $W$-constraint = \emph{odd-elegant}, thus each edge color set $\theta_k(E(G))=[1,2q-1]^o$ for $k\in [1,A!]$.
\end{cor}

\begin{thm}\label{thm:k-d-graceful-v-set-e-proper-co}
$^*$ For a fixed group of non-negative integers $a_1,a_2,\dots, a_m$ with $A=\sum ^m_{k=1}a_k\geq 1$, let $\textbf{\textrm{L}}(a_1,a_2,\dots, a_m)$ be the set of graphs in the lattice $\textbf{\textrm{L}}([\odot_{\textrm{prop}}]Z^0\textbf{\textrm{T}})$ defined in Eq.(\ref{eqa:proper-operation-tree-base-lattice}) such that each connected $(p,q)$-graph $G\in \textbf{\textrm{L}}(a_1,a_2,\dots, a_m)$ holding $G=[\odot_{\textrm{prop}}]^m_{k=1}a_kT_k$. If each tree $T_i$ in the lattice base $\textbf{\textrm{T}}$ admits a \emph{graceful $(k,d)$-total coloring} $g_i$ for $i\in [1,m]$, then each connected $(p,q)$-graph $G\in \textbf{\textrm{L}}(a_1,a_2,\dots, a_m)$ admits at least $A!$ different \emph{$(k,d)$-graceful v-set e-proper colorings} $\pi_i$ with the edge color set
\begin{equation}\label{eqa:555555}
\pi_i(E(G))=\big \{k,k+d,\dots ,k+(q-1)d\big \}=[k,k+(q-1)d]_{+d},~i\in [1,A!]
\end{equation}
\end{thm}
\textbf{The main idea of the proof of Theorem \ref{thm:k-d-graceful-v-set-e-proper-co}.} By Eq.(\ref{eqa:graceful-v-set-e-proper-00}), Eq.(\ref{eqa:graceful-v-set-e-proper-11}), Eq.(\ref{eqa:graceful-v-set-e-proper-22}) and Eq.(\ref{eqa:graceful-v-set-e-proper-33}) shown in the proof of Theorem \ref{thm:graceful-v-set-e-proper-co}, we present the main proof process as follows:

Let $\Sigma _j(k,d)=k+d\sum ^{j}_{k=1}q_{i_k}$ for $j\in [0,A]$, where $\Sigma _0(k,d)=k$.

For $(k_{i_1},d_{i_1})=(k,d)$, we get the edge color set
$$g_{i_1}(E(H_{i_1}))=\big \{k,k+d,\dots ,k+(q_{i_1}-1)d\big \}=[\Sigma _0(k,d),\Sigma _1(k,d)-d]_{+d}
$$ As $(k_{i_2},d_{i_2})=(k+(q_{i_1}-1)d+d,d)=(\Sigma _1(k,d),d)$, we get the edge color set
$${
\begin{split}
g_{i_2}(E(H_{i_2}))&=\big \{k+q_{i_1}d,k+q_{i_1}d+d,\dots ,k+q_{i_1}d+(q_{i_2}-1)d\big \}\\
&=\big \{\Sigma _1(k,d),\Sigma _1(k,d)+d,\dots ,\Sigma _1(k,d)+(q_{i_2}-1)d\big \}\\
&=[\Sigma _1(k,d),\Sigma _2(k,d)-d]_{+d}
\end{split}}
$$

In general, we take $(k_{i_{j+1}},d_{i_{j+1}})=(\Sigma _j(k,d),d)$, then we get the edge color set
$${
\begin{split}
g_{i_{j+1}}(E(H_{i_{j+1}}))&=\big \{\Sigma _j(k,d),\Sigma _j(k,d)+d,\dots ,\Sigma _j(k,d)+(q_{i_{j+1}}-1)d\big \}\\
&= [\Sigma _j(k,d),\Sigma _{j+1}(k,d)-d ]_{+d}
\end{split}}
$$
For the last $(k_{i_A},d_{i_A})=(\Sigma _{A-1}(k,d),d)$, we have the edge color set
$${
\begin{split}
g_{i_A}(E(H_{i_A}))&=\big \{\Sigma _{A-1}(k,d),\Sigma _{A-1}(k,d)+d,\dots ,\Sigma _{A-1}(k,d)+(q_{i_A}-1)d\big \}\\
&= [\Sigma _{A-1}(k,d),\Sigma _{A}(k,d)-d ]_{+d}
\end{split}}
$$ Thereby, based on the permutation $H_{i_1},H_{i_2},\dots ,H_{i_A}$ of the trees $b_1T_1,b_2T_2,\dots ,b_mT_m$, the connected graph $G$ admits a graceful v-set e-proper coloring $\pi_i=\langle g_{i_1},g_{i_2},\dots ,g_{i_A}\rangle$ defined by setting $g_{i_j}(x)\subseteq F_i(x)$, $g_{i_j}(y)\subseteq F_i(y)$ and $F_i(xy)=g_{i_j}(xy)$ if edge $xy\in E(H_{i_j})\subseteq E(G)$, such that the edge color set
$${
\begin{split}
\pi_i(E(G))&=\big \{\pi_i(uv)=|\pi_i(u)-\pi_i(v)|:uv\in E(G)\big \}=\bigcup ^A_{j=1}g_{i_j}(E(H_{i_j}))\\
&=\bigcup ^A_{j=0}[\Sigma _j(k,d),\Sigma _{j+1}(k,d)-d ]_{+d}\\
&=[k,k+(q-1)d]_{+d}
\end{split}}
$$ where $\pi_i(uv)=|\pi_i(u)-\pi_i(v)|$ stands for $g_{i_j}(uv)=|g_{i_j}(u)-g_{i_j}(v)|$ if edge $uv\in E(H_{i_j})$. Refer to the proof of Theorem \ref{thm:graceful-v-set-e-proper-co} for the complete proof of Theorem \ref{thm:k-d-graceful-v-set-e-proper-co}.\qqed

\begin{rem}\label{rem:333333}
\textbf{The analysis of computational complexity} about Theorem \ref{thm:graceful-v-set-e-proper-co}, Corollary \ref{cor:others-v-set-e-proper-co}, and Theorem \ref{thm:k-d-graceful-v-set-e-proper-co}.

Suppose that each tree $T_i$ in the tree base $\textbf{\textrm{T}}=(T_1,T_2,\dots ,T_m)$ has its own diameter $D(T_i)\geq 3$. Let $m_i+1=\Big \lceil \frac{D(T_i)}{2}\Big \rceil $ for $i\in [1,m]$, then the each tree $T_i$ admits at least $2^{m_i}$ different \emph{graceful $(k,d)$-total coloring} $g_{i_j}$ for $j\in [1,2^{m_i}]$.

For a permutation $J_{i_1},J_{i_2},\dots ,J_{i_A}$ of the trees $a_1T_1$, $a_2T_2$, $\dots $, $a_mT_m$ from the tree base $\textbf{\textrm{T}}$ with $A=\sum ^m_{k=1}a_k\geq 1$, then the each tree $J_{i_j}$ admits at least $2^{m_{i_j}}$ different \emph{graceful $(k,d)$-total coloring} $g_{i_j}$ for $j\in [1,2^{m_{i_j}}]$, where $m_{i_j}+1=\Big \lceil \frac{D(T_{i_j})}{2} \Big \rceil $. Let
\begin{equation}\label{eqa:analysis-complexity-123}
n_{um}(A,m_{i_j})=A!\prod ^A_{j=1}2^{m_{i_j}}
\end{equation}

(i) In Theorem \ref{thm:graceful-v-set-e-proper-co}, each connected $(p,q)$-graph $G\in \textbf{\textrm{L}}(b_1,b_2,\dots, b_m)$ admits at least $n_{um}(A,m_{i_j})$ \emph{graceful v-set e-proper colorings} $F_k$ such that each edge color set $F_k(E(G))=[1,q]$.

Since each Topcode-matrix $T_{code}(G,F_k)$ produces $(3q)!$ super-strings, then this graph $G$ distributes us $n_{um}(A,m_{i_j})\cdot (3q)!$ different number-based strings, in total.

\vskip 0.2cm

(ii) In Corollary \ref{cor:others-v-set-e-proper-co}, each connected $(p,q)$-graph $G\in \textbf{\textrm{L}}(a_1,a_2,\dots, a_m)$ admits at least $n_{um}(A,m_{i_j})$ \emph{$W$-constraint v-set e-proper colorings} for $W$-constraint $\in \{$harmonious, odd-graceful, odd-elegant$\}$.

\vskip 0.2cm

(iii) In Theorem \ref{thm:k-d-graceful-v-set-e-proper-co}, each connected $(p,q)$-graph $G\in \textbf{\textrm{L}}(a_1,a_2,\dots, a_m)$ admits at least $n_{um}(A,m_{i_j})$ \emph{$(k,d)$-graceful v-set e-proper colorings} $\pi_i$ with the edge color set $\pi_i(E(G))=[k,k+(q-1)d]_{+d}$. When this graph $G$ is bipartite, we have a parameterized Topcode-matrix
$$P_{(k,d)}(G,\pi_i)=k\cdot I\,^0+d\cdot T_{code}(G,F_i)
$$ for a graceful v-set e-proper colorings $F_i$ for $i\in [1,n_{um}(A,m_{i_j})]$.\paralled
\end{rem}

\begin{thm}\label{thm:666666}
By the hypothesis of Theorem \ref{thm:graceful-v-set-e-proper-co} and the number $n_{um}(A,m_{i_j})$ defined in Eq.(\ref{eqa:analysis-complexity-123}), each connected $(p,q)$-graph $G\in \textbf{\textrm{L}}(b_1,b_2,\dots, b_m)$ holding $G=[\odot_{\textrm{prop}}]^m_{k=1}b_kT_k$ admits at least $n_{um}(A,m_{i_j})$

(i) \emph{different graceful v-set e-proper colorings} $F_k$ such that each edge color set $F_k(E(G))=[1,q]$ with $k\in [1, n_{um}(A,m_{i_j})]$;

(ii) different $W$-constraint v-set e-proper colorings for each $W$-constraint $\in \{$harmonious, odd-graceful, odd-elegant$\}$;

(iii) different $(k,d)$-graceful v-set e-proper colorings $\pi_i$ with the edge color set
$$\pi_i(E(G))=[k,k+(q-1)d]_{+d}=\{k,k+d,\dots ,k+(q-1)d\},~i\in [1, n_{um}(A,m_{i_j})]$$
\end{thm}

\begin{thm}\label{thm:v-set-e-proper-4-magic}
$^*$ For a fixed group of non-negative integers $d_1,d_2,\dots, d_m$ with $\Phi=\sum ^m_{k=1}d_k\geq 1$, let $\textbf{\textrm{L}}(d_1,d_2,\dots, d_m)$ be the set of graphs in the lattice $\textbf{\textrm{L}}([\odot_{\textrm{prop}}]Z^0\textbf{\textrm{T}})$ defined in Eq.(\ref{eqa:proper-operation-tree-base-lattice}) such that each connected $(p,q)$-graph $G\in \textbf{\textrm{L}}(d_1,d_2,\dots, d_m)$ holding $G=[\odot_{\textrm{prop}}]^m_{k=1}d_kT_k$. If each tree $T_i$ of the lattice base $\textbf{\textrm{T}}$ admits a \emph{graceful $(k,d)$-total coloring} $g_i$ for $i\in [1,m]$, then each connected $(p,q)$-graph $G\in \textbf{\textrm{L}}(d_1,d_2,\dots, d_m)$ admits at least $\Phi !$ different \emph{magic-constraint v-set e-proper colorings} $F_i$ with $i\in [1,\Phi!]$.
\end{thm}
\begin{proof} By Theorem \ref{thm:10-k-d-total-coloringss}, each tree $T$ with diameter $D(T)\geq 3$ admits at least $2^m$ different magic-constraint $(k,d)$-total colorings for magic-constraint$\{$edge-magic, edge-difference, felicitous-difference, graceful-difference$\}$, where $m+1=\Big \lceil \frac{D(T)}{2}\Big \rceil $.

Taking the $i$th permutation $G_{i_1},G_{i_2},\dots ,G_{i_\Phi}$ of the trees $d_1T_1$, $d_2T_2$, $\dots $, $d_mT_m$ from the tree base $\textbf{\textrm{T}}$ with $\Phi=\sum ^m_{k=1}d_k\geq 1$, we have connected graphs
\begin{equation}\label{eqa:v-set-e-proper-4-magic00}
G=[\odot_{\textrm{prop}}]^{\Phi}_{k=1}G_{i_k}=[\odot_{\textrm{prop}}]^m_{k=1}d_kT_k
\end{equation} and each tree $G_{i_j}$ admits a magic-constraint total coloring $\eta_{i_j}$ related a constant $\lambda_{i_j}$.

Let $\beta_i=\Pi^{\Phi}_{k=1}\lambda_{i_k}$, and $\beta_{i_j}=\frac{\beta_i}{\lambda_{i_j}}$ for $j\in [1,\Phi]$.

We define a new total coloring $h_{i_j}=\beta_{i_j} \eta_{i_j}$. For each edge $xy\in E(G_{i_j})$, then we have the edge-magic constraint
\begin{equation}\label{eqa:v-set-e-proper-4-magic11}
h_{i_j}(x)+h_{i_j}(y)+h_{i_j}(xy)=\beta_{i_j}[\eta_{i_j}(x)+\eta_{i_j}(y)+\eta_{i_j}(xy)]=\beta_{i_j}\lambda_{i_j}=\beta_i
\end{equation} when the edge-magic constraint $\eta_{i_j}(x)+\eta_{i_j}(y)+\eta_{i_j}(xy)=\lambda_{i_j}$ holds true; and we have the edge-difference constraint
\begin{equation}\label{eqa:v-set-e-proper-4-magic22}
h_{i_j}(xy)+|h_{i_j}(x)-h_{i_j}(y)|=\beta_{i_j}[\eta_{i_j}(xy)+|\eta_{i_j}(x)-\eta_{i_j}(y)|]=\beta_{i_j}\lambda_{i_j}=\beta_i
\end{equation} when as the edge-difference constraint $\eta_{i_j}(xy)+|h_{i_j}(x)-\eta_{i_j}(y)|=\lambda_{i_j}$ holds true; and we have the felicitous-difference constraint
\begin{equation}\label{eqa:v-set-e-proper-4-magic33}
|h_{i_j}(x)+h_{i_j}(y)-h_{i_j}(xy)|=\beta_{i_j}|\eta_{i_j}(x)+\eta_{i_j}(y)-\eta_{i_j}(xy)|=\beta_{i_j}\lambda_{i_j}=\beta_i
\end{equation} if the felicitous-difference constraint $|\eta_{i_j}(x)+\eta_{i_j}(y)-\eta_{i_j}(xy)|=\lambda_{i_j}$ holds true; and we have the graceful-difference constraint
\begin{equation}\label{eqa:v-set-e-proper-4-magic44}
\big ||h_{i_j}(x)-h_{i_j}(y)|-h_{i_j}(xy)\big |=\beta_{i_j}\big ||\eta_{i_j}(x)-\eta_{i_j}(y)|-\eta_{i_j}(xy)\big |=\beta_{i_j}\lambda_{i_j}=\beta_i
\end{equation} as the graceful-difference constraint $\big ||\eta_{i_j}(x)-\eta_{i_j}(y)|-\eta_{i_j}(xy)\big |=\lambda_{i_j}$ holds true.

We define a \emph{magic-constraint v-set e-proper coloring} $F_i=\langle h_{i_1},h_{i_2},\dots ,h_{i_\Phi}\rangle $ with the constant $\beta_i$ by setting $h_{i_j}(x)\subseteq F_i(x)$, $h_{i_j}(y)\subseteq F_i(y)$ and $F_i(xy)=h_{i_j}(xy)$ if edge $xy\in E(G_{i_j})$. Thereby, the theorem follows Eq.(\ref{eqa:v-set-e-proper-4-magic11}), Eq.(\ref{eqa:v-set-e-proper-4-magic22}), Eq.(\ref{eqa:v-set-e-proper-4-magic33}) and Eq.(\ref{eqa:v-set-e-proper-4-magic44}) for $i\in [1,\Phi!]$, since there are $\Phi!$ permutations of the trees $d_1T_1$, $d_2T_2$, $\dots $, $d_mT_m$ from the tree base $\textbf{\textrm{T}}$ with $\Phi=\sum ^m_{k=1}d_k\geq 1$.
\end{proof}

\begin{cor}\label{cor:uniformly-v-set-e-proper-4-magic}
$^*$ By the hypothesis and the proof of Theorem \ref{thm:v-set-e-proper-4-magic} and each magic-constraint v-set e-proper coloring $F_i=\langle h_{i_1},h_{i_2},\dots ,h_{i_\Phi}\rangle $ with the constant $\beta_i$ is defined by $h_{i_j}(x)\subseteq F_i(x)$, $h_{i_j}(y)\subseteq F_i(y)$ and $F_i(xy)=h_{i_j}(xy)$ for edge $xy\in E(G_{i_j})$ and $i\in [1,\Phi!]$, since each $F_i(x)$ is a set, we use the \emph{assignment symbol} ``$:=$'' in the following particular cases:
\begin{asparaenum}[\textbf{\textrm{Parti}}-1.]
\item If each tree $G_{i_j}$ for $j\in [1,\Phi]$ holds the edge-magic constraint $\eta_{i_j}(x)+\eta_{i_j}(y)+\eta_{i_j}(xy)=\lambda_{i_j}$ for each edge $xy\in E(G_{i_j})$, then the connected graph $G$ defined in Eq.(\ref{eqa:v-set-e-proper-4-magic00}) admits an \emph{edge-magic v-set e-proper coloring} $F_i=\langle h_{i_1},h_{i_2},\dots ,h_{i_\Phi}\rangle $ holding the edge-magic constraint
\begin{equation}\label{eqa:555555}
F_i(x)+F_i(xy)+F_i(y):=h_{i_j}(x)+h_{i_j}(y)+h_{i_j}(xy)=\beta_i
\end{equation} for each edge $xy\in E(G_{i_j})\subset E(G)$ according to Eq.(\ref{eqa:v-set-e-proper-4-magic11}).
\item If each tree $G_{i_j}$ for $j\in [1,\Phi]$ holds the edge-difference constraint $\eta_{i_j}(xy)+|h_{i_j}(x)-\eta_{i_j}(y)|=\lambda_{i_j}$ for each edge $xy\in E(G_{i_j})$, then the connected graph $G$ defined in Eq.(\ref{eqa:v-set-e-proper-4-magic00}) admits an \emph{edge-difference v-set e-proper coloring} $F_i=\langle h_{i_1},h_{i_2},\dots ,h_{i_\Phi}\rangle $ holding the edge-difference constraint
\begin{equation}\label{eqa:555555}
F_i(xy)+|F_i(x)-F_i(y)|:=h_{i_j}(xy)+|h_{i_j}(x)-h_{i_j}(y)|=\beta_i
\end{equation} for each edge $xy\in E(G_{i_j})\subset E(G)$ according to Eq.(\ref{eqa:v-set-e-proper-4-magic22}).
\item If each tree $G_{i_j}$ for $j\in [1,\Phi]$ holds the felicitous-difference constraint $|\eta_{i_j}(x)+\eta_{i_j}(y)-\eta_{i_j}(xy)|=\lambda_{i_j}$ for each edge $xy\in E(G_{i_j})$, then the connected graph $G$ defined in Eq.(\ref{eqa:v-set-e-proper-4-magic00}) admits a \emph{felicitous-difference v-set e-proper coloring} $F_i=\langle h_{i_1},h_{i_2},\dots ,h_{i_\Phi}\rangle $ holding the felicitous-difference constraint
\begin{equation}\label{eqa:555555}
|F_i(x)+F_i(y)-F_i(xy)|:=|h_{i_j}(x)+h_{i_j}(y)-h_{i_j}(xy)|=\beta_i
\end{equation} for each edge $xy\in E(G_{i_j})\subset E(G)$ according to Eq.(\ref{eqa:v-set-e-proper-4-magic33}).
\item If each tree $G_{i_j}$ for $j\in [1,\Phi]$ holds the graceful-difference constraint $\big ||\eta_{i_j}(x)-\eta_{i_j}(y)|-\eta_{i_j}(xy)\big |=\lambda_{i_j}$ for each edge $xy\in E(G_{i_j})$, then the connected graph $G$ defined in Eq.(\ref{eqa:v-set-e-proper-4-magic00}) admits a \emph{graceful-difference v-set e-proper coloring} $F_i=\langle h_{i_1},h_{i_2},\dots ,h_{i_\Phi}\rangle $ holding the graceful-difference constraint
\begin{equation}\label{eqa:555555}
\big ||F_i(x)-F_i(y)|-F_i(xy)\big |:=\big ||h_{i_j}(x)-h_{i_j}(y)|-h_{i_j}(xy)\big |=\beta_i
\end{equation} for each edge $xy\in E(G_{i_j})\subset E(G)$ according to Eq.(\ref{eqa:v-set-e-proper-4-magic44}).
\end{asparaenum}
\end{cor}

\begin{thm}\label{thm:666666}
By the hypothesis and the proof of Theorem \ref{thm:v-set-e-proper-4-magic}, there is a number-based string
\begin{equation}\label{eqa:555555}
s(\beta)=\beta_1\beta_2\cdots \beta_M
\end{equation} where $\beta_i=\lambda_{i_1}\lambda_{i_2}\cdots \lambda_{i_\Phi}$, $M=\Phi!$ with $\Phi=\sum ^m_{k=1}d_k\geq 1$. Moreover, there are:

(i) different number-based strings $s_{permu}(k)=\gamma_{k_1}\gamma_{k_2}\cdots \gamma_{k_M}$ for $k\in [1,M!]$ induced from all permutations $\gamma_{k_1},\gamma_{k_2},\cdots ,\gamma_{k_M}$ made by the numbers $\beta_1,\beta_2,\dots ,\beta_M$.

(ii) different number-based strings $\beta_{i,t_i}=\rho^{t_i}_{j_1}\rho^{t_i}_{j_2}\cdots \rho^{t_i}_{j_\Phi}$ with $t_i\in [1,\Phi!]$ made by all permutations $\rho^{t_i}_{j_1},\rho^{t_i}_{j_2},\dots ,\rho^{t_i}_{j_\Phi}$ of the numbers $\lambda_{i_1},\lambda_{i_2},\dots ,\lambda_{i_\Phi}$. Thereby, we get number-based strings
\begin{equation}\label{eqa:555555}
s(\beta_t)=\beta_{1,t_1}\beta_{2,t_2}\cdots \beta_{M,t_M}
\end{equation} with more complex structure.

(iii) different number-based strings made by the combination of the above two cases.
\end{thm}

\begin{example}\label{exa:8888888888}
Let $\textbf{\textrm{L}}(4am)$ be the set of graphs in the lattice $\textbf{\textrm{L}}([\odot_{\textrm{prop}}]Z^0\textbf{\textrm{T}})$ defined in Eq.(\ref{eqa:proper-operation-tree-base-lattice}) such that each connected $(p,q)$-graph $G\in \textbf{\textrm{L}}(4am)$ holding $G=[\odot_{\textrm{prop}}]^m_{k=1}4aT_k$. So, we have tree sets $\textbf{\textrm{T}}_{r,s}(4am)=\{T_{j,r,s}: j\in [1,m]\}$ for $r\in [1,4]$ and $s\in [1,a]$ holding $T_{j,r,s}=T_j\in \textbf{\textrm{T}}$ for $j\in [1,m]$. Suppose that the trees in $\textbf{\textrm{L}}_{r,s}(4am)$ hold the following facts:

(i) Each tree $T_{j,1,s}\in \textbf{\textrm{T}}_{1,s}(4am)$ admits an edge-magic total coloring $f_{j,1,s}$, and the connected graph $G\in \textbf{\textrm{L}}(4am)$ admits \emph{edge-magic v-set e-proper colorings} $F_{1,s}$ for $s\in [1,a]$ holding the edge-magic constraint
\begin{equation}\label{eqa:555555}
F_{1,s}(x)+F_{1,s}(xy)+F_{1,s}(y):=\big \{f_{j,1,s}(x)+f_{j,1,s}(y)+f_{j,1,s}(xy)=\beta_{j,1,s}: j\in [1,m]\big \}
\end{equation} for each edge $xy\in E(T_{j,1,s})\subset E(G)$.

(ii) Each tree $T_{j,2,s}\in \textbf{\textrm{T}}_{2,s}(4am)$ admits an edge-difference total coloring $f_{j,2,s}$, and the connected graph $G\in \textbf{\textrm{L}}(4am)$ admits \emph{edge-difference v-set e-proper colorings} $F_{2,s}$ for $k\in [1,a]$ holding the edge-difference constraint
\begin{equation}\label{eqa:555555}
F_{2,s}(xy)+|F_{2,s}(x)-F_{2,s}(y)|:=\big \{f_{j,2,s}(xy)+|f_{j,2,s}(x)-f_{j,2,s}(y)|=\beta_{j,2,s}: j\in [1,m]\big \}
\end{equation} for each edge $xy\in E(T_{j,2,s})\subset E(G)$.

(iii) Each tree $T_{j,3,s}\in \textbf{\textrm{T}}_{3,s}(4am)$ admits a felicitous-difference total coloring $f_{j,3,s}$, and the connected graph $G\in \textbf{\textrm{L}}(4am)$ admits \emph{felicitous-difference v-set e-proper colorings} $F_{3,s}$ for $k\in [1,a]$ holding the felicitous-difference constraint
\begin{equation}\label{eqa:555555}
|F_{3,s}(x)+F_{3,s}(y)-F_i(xy)|:=\big \{|f_{j,3,s}(x)+f_{j,3,s}(y)-f_{j,3,s}(xy)|=\beta_{j,3,s}: j\in [1,m]\big \}
\end{equation} for each edge $xy\in E(T_{j,3,s})\subset E(G)$.

(iv) Each tree $T_{j,4,s}\in \textbf{\textrm{T}}_{4,s}(4am)$ admits a graceful-difference total coloring $f_{j,4,s}$, and the connected graph $G\in \textbf{\textrm{L}}(4am)$ admits \emph{graceful-difference v-set e-proper colorings} $F_{4,s}$ for $k\in [1,a]$ holding the graceful-difference constraint
\begin{equation}\label{eqa:555555}
\big ||F_{4,s}(x)-F_{4,s}(y)|-F_{4,s}(xy)\big |:=\big \{\big ||f_{j,4,s}(x)-f_{j,4,s}(y)|-f_{j,4,s}(xy)\big |=\beta_{j,4,s}: j\in [1,m]\big \}
\end{equation} for each edge $xy\in E(T_{j,4,s})\subset E(G)$.

Thereby, we get number-based strings $s_{j,r}=\beta_{j,r,1}\beta_{j,r,2}\cdots \beta_{j,r,a}$ and $s^k_{j,r}=\beta_{j,r,k_1}\beta_{j,r,k_2}\cdots \beta_{j,r,k_a}$ with $r\in [1,4]$, where $\beta_{j,r,k_1},\beta_{j,r,k_2},\dots .\beta_{j,r,k_a}$ is a permutation of $\beta_{j,r,1},\beta_{j,r,2},\dots ,\beta_{j,r,a}$ with $k\in [1,a!]$ and $r\in [1,4]$.

Each connected graph $G\in \textbf{\textrm{L}}(4am)$ admits a compound v-set e-proper coloring induced by the $W$-constraint v-set e-proper colorings $F_{1,s},F_{2,s},F_{3,s},F_{4,s}$ with $s\in [1,a]$. Obviously, finding a particular connected graph $G$ and the number-based strings $s_{j,1}s_{j,2}s_{j,3}s_{j,4}$ made by the colorings admitted by the connected graph $G$ seems to be quite difficult, even it is almost impossible.\qqed
\end{example}

\subsection{Number-based strings generated from indexed-colorings}

In \cite{Yao-Su-Ma-Wang-Yang-arXiv-2202-03993v1} the authors use maximal planar graphs to make various topological signature authentications. Here, we will use them to make number-based strings, since maximal planar graphs have many nice properties. In Wikipedia, the enumeration of planar graphs is as:

(i) The asymptotic for the number of \emph{(labeled) planar graphs} on $n$ vertices is $g\cdot n^{-7/2}\cdot n!\cdot \gamma ^{n}$, where $\gamma \approx 27.22687$ and $g\approx 0.43\times 10^{-5}$.

(ii) Almost all planar graphs have an exponential number of automorphisms.

(iii) The number of \emph{unlabeled (non-isomorphic) planar graphs} on $n$ vertices is between $27.2^{n}$ and $30.06^{n}$, roughly, between $2^{4.7n}$ and $2^{6n}$.

\subsubsection{Maximal planar graphic lattices}

A \emph{maximal planar graph} is a connected planar graph with each face being a \emph{triangle}.
\begin{defn} \label{defn:maximal-planar-graph-embedding}
For two vertex-disjoint maximal planar graphs $G$ and $H$, we overlap a triangular face $\Delta_{abc}$ of $G$ with another triangular face $\Delta_{xyz}$ of $H$ into one triangular face
$$\Delta(abc\overline{\ominus} xyz)=\Delta_{abc}[\overline{\ominus}^{cyc}_{3}]\Delta_{xyz}
$$ by edge-coinciding

(i) two edges $ab,xy$ into one edge $ab\overline{\ominus} xy$ with $a\odot x$ and $b\odot y$;

(ii) two edges $bc, yz$ into one edge $bc\overline{\ominus} yz$ with $b\odot y$; and

(iii) two edges $ca, zx$ into one edge $ca\overline{\ominus} zx$ with $c\odot z$ and $x\odot a$. \\
The resultant graph is a maximal planar graph too, denoted as $G[\overline{\ominus}^{cyc}_{3}]H$, and we call the process of obtaining $G[\overline{\ominus}^{cyc}_{3}]H$ \emph{triangular-face embedding operation}, and the maximal planar graph $G[\overline{\ominus}^{cyc}_{3}]H$ is the result of $G$ embedding $H$ (resp. $H$ embedding $G$).\qqed
\end{defn}

\begin{rem}\label{remark:maximal-planar-graph-embedding}
As the generalization of Definition \ref{defn:maximal-planar-graph-embedding}, we split a cycle $C$ of $k$ vertices in a maximal planar graph $G$ in to two cycles $C\,'$ and $C\,''$, such that the cycle-split graph $G\wedge C$ has just two vertex disjoint components $G^C_{out}$ and $G^C_{in}$, called \emph{semi-maximal planar graphs} (also \emph{configurations} \cite{Jin-Xu-Maximal-Science-Press-2019, Jin-Xu-55-56-configurations-arXiv-2107-05454v1}), where $G^C_{out}$ contains the cycle $C\,''$ and is in the \emph{infinite plane}, and $G^C_{in}$ contains the cycle $C\,'$ and is inside of $G$. So, we write $G=G^C_{out}[\overline{\ominus}^{cyc}_k]G^C_{in}$ as the result of doing the cycle-coinciding operation $[\overline{\ominus}^{cyc}_k]$ to two semi-maximal planar graphs $G^C_{out}$ and $G^C_{in}$.

We call a semi-maximal planar graph $G^C$ \emph{configuration}, and the set $U^*$ contains all configurations. Moreover, if any maximal planar graph $H$ contains at least a configuration in $U^*$, we say that $U^*$ is \emph{unavoidable}. Assume that a maximal planar graph $H^*$ of $n$ vertices does not admit a proper 4-coloring, and any maximal planar graph $G$ with $|V(G)|<|V(H^*)|$ is 4-colorable, then $H^*$ is called a smallest counterexample of the Four-coloring Conjecture, \emph{smallest FCC-counterexample}. If any smallest FCC-counterexample does not contain a configuration $G^C$, we say that the configuration $G^C$ is \emph{reducible}. If any maximal planar graph contains at least a configuration in $U^*$, and each configuration in $U^*$ is reducible, then $U^*$ is called a \emph{unavoidable and reducible configuration set}.

As two semi-maximal planar graphs $G^C_{out}$ and $G^C_{in}$ in a maximal planar graph $G=G^C_{out}[\overline{\ominus}^{cyc}_k]G^C_{in}$ approach to two infinite planes, then $G$ becomes a sphere with an equatorial line $C$ and each face to be a triangular face. \paralled
\end{rem}

\begin{cor}\label{cor:99999}
$^*$ Let $T^{C_{2k}}_{1},T^{C_{2k}}_{2},\dots ,T^{C_{2k}}_{m}$ be semi-maximal planar graphs with the cycle $C_{2k}$ on $2k$ vertices, then doing the cycle-coinciding operation $[\overline{\ominus}^{cyc}_{2k}]$ to them produces a \emph{semi-maximal planar graph-book}
$$[\overline{\ominus}^{cyc}_{2k}]^m_{i=1}T^{C_{2k}}_{i}=T^{C_{2k}}_{1}[\overline{\ominus}^{cyc}_{2k}]T^{C_{2k}}_{2}[\overline{\ominus}^{cyc}_{2k}]T^{C_{2k}}_{2}\cdots [\overline{\ominus}^{cyc}_{2k}]T^{C_{2k}}_{m}
$$ with each \emph{book-page} $T^{C_{2k}}_{i}$ and the \emph{book-spine} $C_{2k}$. Moreover, if each book-page $T^{C_{2k}}_{i}$ admits a proper vertex $4$-coloring such that its own cycle $C_{2k}$ is properly colored with two colors, then the semi-maximal planar graph-book $[\overline{\ominus}^{cyc}_{2k}]^m_{i=1}T^{C_{2k}}_{i}$ admits a proper vertex $4$-coloring too.
\end{cor}

\begin{example}\label{exa:8888888888}
In Fig.\ref{fig:planar-graph-c-operation}, the graphs $P$, $B_1$, $B_2$ and $B_3$ are four semi-maximal planar graphs, such that each maximal planar graph $A_i=P[\overline{\ominus}^{cyc}_6]B_i$ for $1\leq i\leq 3$. Conversely, we split the cycle $C$ of $k$ vertices in each maximal planar graph $A_i$, the resultant graph $A_i\wedge C$ has just two vertex disjoint semi-maximal planar graphs $P~(=G^C_{out}$, called \emph{out-planar graph}) and $B_i~(=G^C_{in}$, called \emph{inner-planar graph}) for $1\leq i\leq 3$.

In real application, we use the semi-maximal planar graph $P$ as a \emph{public-key graph}, and three semi-maximal planar graphs $B_1$, $B_2$ and $B_3$ are \emph{private-key graphs}, thus, each maximal planar graph $A_i=P[\overline{\ominus}^{cyc}_6]B_i$ for $1\leq i\leq 3$ is just a \emph{topological signature authentication} $A_i=A_{uth}\langle P,B_i\rangle $. Moreover, the Topcode-matrix $T_{code}(A_i,f_i)$ of each maximal planar graph $A_i$ admitting a $4$-coloring $f_i$ can provide us $(3q_i)!$ different number-based strings, where $q_i=|E(A_i)|$ for $1\leq i\leq 3$. \qqed
\end{example}

\begin{figure}[h]
\centering
\includegraphics[width=16.4cm]{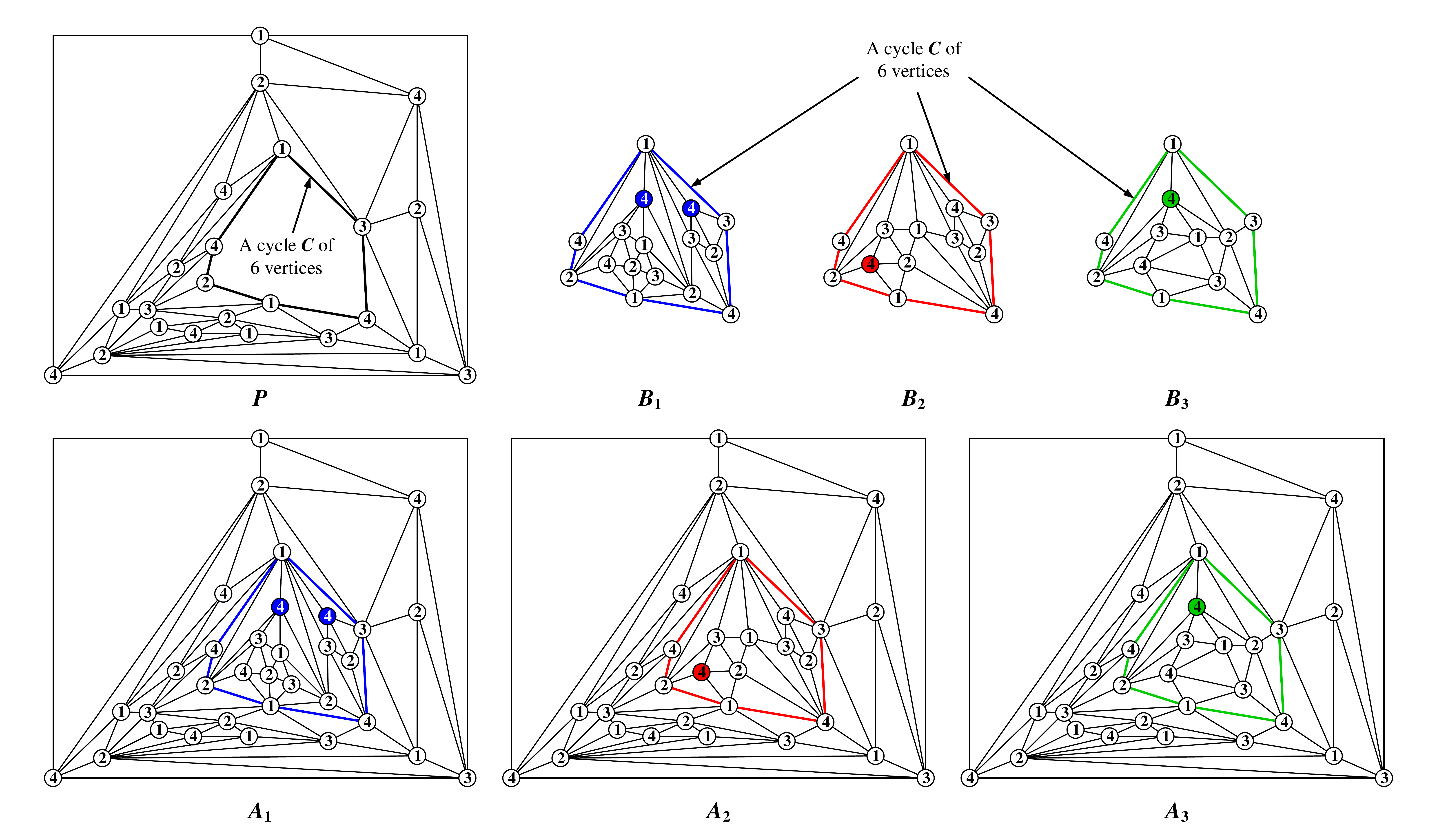}\\
\caption{\label{fig:planar-graph-c-operation}{\small Examples for the cycle-coinciding operation $[\overline{\ominus}^{cyc}_k]$, cited from \cite{Yao-Su-Ma-Wang-Yang-arXiv-2202-03993v1}.}}
\end{figure}

\textbf{1.Uncolored maximal planar graphic lattices.} Suppose that $\textbf{\textrm{P}}=(P_1,P_2,\dots ,P_n)$ is a \emph{maximal planar graph base} with mutually vertex-disjoint uncolored maximal planar graphs $P_1,P_2,\dots ,P_n$ holding $P_i\not \cong P_j$ if $i\neq j$. Each graph $G=[\overline{\ominus}^{cyc}_{3}]^n_{k=1}a_kP_k$ is a maximal planar graph obtained by doing the triangular-face embedding operation defined in Definition \ref{defn:maximal-planar-graph-embedding} on the maximal planar graphs $P_1,P_2,\dots ,P_n$ of the maximal planar graph base $\textbf{\textrm{P}}$, we get a \emph{maximal planar graphic lattice} as follows
\begin{equation}\label{eqa:maximal-planar-graph-lattice}
\textbf{\textrm{L}}([\overline{\ominus}^{cyc}_{3}]Z^0\textbf{\textrm{P}})=\big \{[\overline{\ominus}^{cyc}_{3}]^n_{k=1}a_kP_k: a_k\in Z^0,P_k\in \textbf{\textrm{P}}\big\}
\end{equation} with $\sum^n_{k=1}a_k\geq 1$.

However, there are many maximal planar graphs out of the maximal planar graphic lattices $\textbf{\textrm{L}}([\overline{\ominus}^{cyc}_{3}]Z^0\textbf{\textrm{P}})$. Let $T_{i_1},T_{i_2},\dots ,T_{i_M}$ be a permutation of the maximal planar graphs $c_1P_1,c_2P_2$, $\dots $, $c_nP_n$ of the maximal planar graph base $\textbf{\textrm{P}}$, where with $M=\sum^n_{k=1}c_k\geq 1$.

For a maximal planar graph $H$, we do the triangular-face embedding operation to $H_0=H$ and $T_{i_1}$, the resultant maximal planar graph is denoted as $H_{1}=H_{0}[\overline{\ominus}^{cyc}_{3}]T_{i_1}$, next we get another maximal planar graph is denoted as $H_{2}=H_{1}[\overline{\ominus}^{cyc}_{3}]T_{i_2}$, in general, we have maximal planar graphs
\begin{equation}\label{eqa:555555}
H_{k}=H_{k-1}[\overline{\ominus}^{cyc}_{3}]T_{i_k},~k\in [1,M]
\end{equation} We write $H_{M}=H_{M-1}[\overline{\ominus}^{cyc}_{3}]T_{i_M}=H[\overline{\ominus}^{cyc}_{3}]^n_{i=1}c_iP_i$, and get a maximal planar graphic lattice
\begin{equation}\label{eqa:graphs-maximal-planar-lattice}
\textbf{\textrm{L}}(S_{mpg}[\overline{\ominus}^{cyc}_{3}]Z^0\textbf{\textrm{P}})=\big \{H[\overline{\ominus}^{cyc}_{3}]^n_{i=1}c_iP_i: c_i\in Z^0,P_i\in \textbf{\textrm{P}},H\in S_{mpg}\big\}
\end{equation} with $\sum^n_{k=1}c_k\geq 1$, where $S_{mpg}$ is a set of maximal planar graphs.

\begin{problem}\label{qeu:444444}
About the lattices $\textbf{\textrm{L}}([\overline{\ominus}^{cyc}_{3}]Z^0\textbf{\textrm{P}})$ defined in Eq.(\ref{eqa:maximal-planar-graph-lattice}), $\textbf{\textrm{L}}([\overline{\ominus}^{cyc}_{3}]Z^0\textbf{\textrm{P}}^c)$ defined in Eq.(\ref{eqa:colored-maximal-planar-lattice}), $\textbf{\textrm{L}}(S_{mpg}[\overline{\ominus}^{cyc}_{3}]Z^0\textbf{\textrm{P}})$ defined in Eq.(\ref{eqa:graphs-maximal-planar-lattice}) and $\textbf{\textrm{L}}(S^c_{mpg}[\overline{\ominus}^{cyc}_{3}]Z^0\textbf{\textrm{P}}^c)$ defined in Eq.(\ref{eqa:colored-graphs-maximal-planar-lattice}), we have the following questions:

(i) In a colored maximal planar graph $J^c=G^c[\overline{\ominus}^{cyc}_{3}]H^c$, if $H^c$ is not $4$-colorable, so is $J^c$. If there is a maximal planar graph that cannot be colored properly with 4 colors, then there are infinite maximal planar graphs that are not colored properly by 4 colors.

(ii) Suppose that a maximal planar graph base $\textbf{\textrm{P}}$ contains all unavoidable and reducible configurations of maximal planar graphs. Is there a maximal planar graph $G^*=G[\overline{\ominus}^{cyc}_{3}]T$ with a maximal planar graph $T$ and $G\in \textbf{\textrm{L}}([\overline{\ominus}^{cyc}_{3}]Z^0\textbf{\textrm{P}})$, or $G\in \textbf{\textrm{L}}(S_{mpg}[\overline{\ominus}^{cyc}_{3}]Z^0\textbf{\textrm{P}})$, such that $G^*$ contains no any one of the known unavoidable and reducible configurations of maximal planar graphs? For example, using the triangular-face embedding operation or other technique to destroy as many configurations as we can along the way in the process of making $G^*$.

(iii) For a \emph{maximal planar graph sequence} $G_{n+1}=G_{n}[\overline{\ominus}^{cyc}_{3}]H_n$ with $n\geq 1$, each $H_n$ is a maximal planar graph and $G_1\in \textbf{\textrm{L}}([\overline{\ominus}^{cyc}_{3}]Z^0\textbf{\textrm{P}})$, or $G_1\in \textbf{\textrm{L}}(S_{mpg}[\overline{\ominus}^{cyc}_{3}]Z^0\textbf{\textrm{P}})$. As $n\rightarrow \infty$, the maximal planar graph sequence $\{G_{n+1}\}$ induces a $xOy$-plane titled by maximal planar graphs, we call it \emph{triangular-plane} $\Delta_{xOy}$. In this construction, the \emph{boundary} of the triangular-plane $\Delta_{xOy}$ is a \emph{triangle}. Does the triangular-plane $\Delta_{xOy}$ contains all configurations of maximal planar graphs?
\end{problem}

\textbf{2. Colored maximal planar graphic lattices.} If a colored maximal planar graph $G^c$ admits a proper vertex coloring $f$ and another colored maximal planar graph $H^c$ admits a proper vertex coloring $g$, and $f(a)=g(x)$, $f(b)=g(y)$ and $f(c)=g(z)$ in tow triangular faces $\Delta_{abc}\subset G^c$ and $\Delta_{xyz}\subset H^c$, we have a colored maximal planar graph $G^c[\overline{\ominus}^{cyc}_{3}]H^c$ admitting a proper vertex coloring $F=f[\overline{\ominus}^{cyc}_{3}]g$ by the triangular-face embedding operation defined in Definition \ref{defn:maximal-planar-graph-embedding}, distinguishingly, we call the process of obtaining the colored maximal planar graph $G^c[\overline{\ominus}^{cyc}_{3}]H^c$ \emph{colored triangular-face embedding operation}.

Let $\textbf{\textrm{P}}^c=(P^c_1,P^c_2,\dots ,P^c_n)$ be a \emph{colored maximal planar graph base} with mutually disjoint maximal planar graphs $P^c_1,P^c_2,\dots ,P^c_n$ holding $P^c_i\not \cong P^c_j$ if $i\neq j$, and each maximal planar graph $P^c_j$ admits a proper vertex coloring $g_j$ with $j\in [1,n]$. Doing the colored triangular-face embedding operation defined above to the mutually disjoint maximal planar graphs $b_1P^c_1,b_2P^c_2,\dots ,b_nP^c_n$ for $\sum^n_{k=1}b_k\geq 1$, then we get colored maximal planar graphs $G=[\overline{\ominus}^{cyc}_{3}]^n_{k=1}b_kP^c_k$ admitting proper vertex coloring induced by the proper vertex colorings $g_1,g_2\dots, g_n$, immediately, the following set
\begin{equation}\label{eqa:colored-maximal-planar-lattice}
\textbf{\textrm{L}}([\overline{\ominus}^{cyc}_{3}]Z^0\textbf{\textrm{P}}^c)=\big \{ [\overline{\ominus}^{cyc}_{3}]^n_{k=1}b_kP^c_k: b_k\in Z^0,~P^c_k\in \textbf{\textrm{P}}^c\big\}
\end{equation}
with $\sum^n_{k=1}b_k\geq 1$ is called \emph{colored maximal planar graphic lattice}.

\begin{thm}\label{thm:colored-mpg-lattice-Coloring-closure}
$^*$ \textbf{Coloring closure.} If each maximal planar graph $P^c_j$ of a colored maximal planar graph base $\textbf{\textrm{P}}^c$ admits a proper vertex $4$-coloring, then each maximal planar graph $G\in \textbf{\textrm{L}}([\overline{\ominus}^{cyc}_{3}]Z^0\textbf{\textrm{P}}^c)$ defined in Eq.(\ref{eqa:colored-maximal-planar-lattice}) admits a proper vertex $4$-coloring too.
\end{thm}

Let $S^c_{mpg}$ be a set of colored maximal planar graphs. By the way similarly with that of obtaining the maximal planar graphic lattice $\textbf{\textrm{L}}(S_{mpg}[\overline{\ominus}^{cyc}_{3}]Z^0\textbf{\textrm{P}})$, we get a colored maximal planar graphic lattice
\begin{equation}\label{eqa:colored-graphs-maximal-planar-lattice}
\textbf{\textrm{L}}(S^c_{mpg}[\overline{\ominus}^{cyc}_{3}]Z^0\textbf{\textrm{P}}^c)=\big \{H^c[\overline{\ominus}^{cyc}_{3}]^n_{i=1}d_iP^c_i: d_i\in Z^0,~P^c_i\in \textbf{\textrm{P}}^c,~H^c\in S^c_{mpg}\big\}
\end{equation} with $\sum^n_{i=1}d_i\geq 1$.

\begin{thm}\label{thm:colored-mpg-lattice-Coloring-closure22}
$^*$ \textbf{Coloring closure.} If each maximal planar graph $P^c_j$ in a colored maximal planar graph base $\textbf{\textrm{P}}^c$ admits a proper vertex $4$-coloring, and each maximal planar graph $H\in S^c_{mpg}$ admits a proper vertex $4$-coloring, then each maximal planar graph of the colored maximal planar graphic lattice $\textbf{\textrm{L}}(S^c_{mpg}[\overline{\ominus}^{cyc}_{3}]Z^0\textbf{\textrm{P}}^c)$ defined in Eq.(\ref{eqa:colored-graphs-maximal-planar-lattice}) admits a proper vertex $4$-coloring too.
\end{thm}

\begin{rem}\label{rem:333333}
In \cite{Jin-Xu-Maximal-Science-Press-2019}, Prof. Xu mentioned: There are many variations of the Euler's formula, which play an essential role in the development of discharging, the key approach for exploiting the computer-assisted proof of the Four-Color Conjecture by investigating the unavoidability and reducibility of some configurations of maximal planar graphs \cite{A-Soifer-Springer-2009}. Haken and Appel spent seven years investigating configurations in more details and eventually in 1976 (with the help of Koch and about 1200 hour of fast mainframe computer) gave a computer-based proof of the Four-Color Conjecture (Ref. \cite{K-Appel-W-Haken-discharging-1977} and \cite{K-Appel-W-Haken-reducibility-1977}), where the number of discharging rules and the number of \emph{unavoidable configurations} they used are 487 and 1936, respectively.

In the research process of the Four-Color Conjecture, many challenges are encountered inevitably, and as a result new conjectures are proposed accordingly, such as Uniquely Four Chromatic planar graphs conjecture.\paralled
\end{rem}

Let $M_{PG}(k)$ be a maximal planar graph containing the $k$th unavoidable configuration of 1936 unavoidable configurations \cite{Jin-Xu-Maximal-Science-Press-2019}, and the vertex number $|V(M_{PG}(k))|$ is the smallest one in all maximal planar graphs containing the $k$th unavoidable configuration. Then we have a \emph{MPG-configuration graphic lattice}
\begin{equation}\label{eqa:555555}
\textbf{\textrm{L}}([\overline{\ominus}]Z^0\textbf{\textrm{C}}_{\textrm{onfi}})=\big \{[\overline{\ominus}^{cyc}_3]^{1936}_{k=1}a_kM_{PG}(k):~a_k\in Z^0, ~M_{PG}(k)\in \textbf{\textrm{C}}_{\textrm{onfi}}\big \}
\end{equation} with the \emph{configuration graph base} $\textbf{\textrm{C}}_{\textrm{onfi}}=(M_{PG}(1),M_{PG}(2),\dots ,M_{PG}(1936))$.

Let $S_{mpg}$ be the set of maximal planar graphs. Do cycle-coinciding operation to a inner face and a maximal planar graph $M_{PG}(k)\in S_{mpg}$, the resultant maximal planar graph is denoted as $G[\overline{\ominus}^{cyc}_3]^{1936}_{k=1}a_kM_{PG}(k)$, then we have a \emph{MPG-configuration graphic lattice}
\begin{equation}\label{eqa:555555}
\textbf{\textrm{L}}(S_{mpg}[\overline{\ominus}]Z^0\textbf{\textrm{C}}_{\textrm{onfi}})=\big \{G[\overline{\ominus}^{cyc}_3]^{1936}_{k=1}b_kM_{PG}(k):~G\in S_{mpg}, ~b_k\in Z^0, ~M_{PG}(k)\in \textbf{\textrm{C}}_{\textrm{onfi}}\big \}
\end{equation}

It is necessary to prove that there are only 1936 unavoidable configurations on the maximum planar graphs for determining whether the graphs in two MPG-configuration graphic lattices $\textbf{\textrm{L}}([\overline{\ominus}]Z^0\textbf{\textrm{C}}_{\textrm{onfi}})$ and $\textbf{\textrm{L}}(S_{mpg}[\overline{\ominus}]Z^0\textbf{\textrm{C}}_{\textrm{onfi}})$ are 4-color colorable.

\begin{thm}\label{thm:666666}
Using maximal planar graphs titles the $xOy$-plane, such that the boundary of the triangular-plane $\Delta_{xOy}$ is a triangle.
\end{thm}

\begin{rem}\label{rem:not-necessarily-mpgs}
If it is allowed that the results of the operation $[\overline{\ominus}^{cyc}_{3}]^n_{k=1}b_kP^c_k$ in $\textbf{\textrm{L}}([\overline{\ominus}^{cyc}_{3}]Z^0\textbf{\textrm{P}}^c)$ defined in Eq.(\ref{eqa:colored-maximal-planar-lattice}) are not necessarily maximal planar graphs, then we get a graphic lattice $\textbf{\textrm{L}}^*([\overline{\ominus}^{cyc}_{3}]Z^0\textbf{\textrm{P}}^c)$ holding
$\textbf{\textrm{L}}([\overline{\ominus}^{cyc}_{3}]Z^0\textbf{\textrm{P}}^c)\subset \textbf{\textrm{L}}^*([\overline{\ominus}^{cyc}_{3}]Z^0\textbf{\textrm{P}}^c)$.

Similarly, the results of the operation $H^c[\overline{\ominus}^{cyc}_{3}]^n_{i=1}d_iP^c_i$ in $\textbf{\textrm{L}}(S^c_{mpg}[\overline{\ominus}^{cyc}_{3}]Z^0\textbf{\textrm{P}}^c)$ defined in Eq.(\ref{eqa:colored-graphs-maximal-planar-lattice}) are not necessarily maximal planar graphs, then we get a graphic lattice $\textbf{\textrm{L}}^*(S^c_{mpg}[\overline{\ominus}^{cyc}_{3}]Z^0\textbf{\textrm{P}}^c)$ holding $\textbf{\textrm{L}}(S^c_{mpg}[\overline{\ominus}^{cyc}_{3}]Z^0\textbf{\textrm{P}}^c)\subset \textbf{\textrm{L}}^*(S^c_{mpg}[\overline{\ominus}^{cyc}_{3}]Z^0\textbf{\textrm{P}}^c)$.

Notice that \textbf{determining} whether a graph $G\in \textbf{\textrm{L}}^*([\overline{\ominus}^{cyc}_{3}]Z^0\textbf{\textrm{P}}^c)$ or $G\in \textbf{\textrm{L}}^*(S^c_{mpg}[\overline{\ominus}^{cyc}_{3}]Z^0\textbf{\textrm{P}}^c)$ to be a maximal planar graph is NP-complete. \paralled
\end{rem}

\begin{thm}\label{thm:666666}
$^*$ Each graph of two graphic lattices $\textbf{\textrm{L}}^*([\overline{\ominus}^{cyc}_{3}]Z^0\textbf{\textrm{P}}^c)$ and $\textbf{\textrm{L}}^*(S^c_{mpg}[\overline{\ominus}^{cyc}_{3}]Z^0\textbf{\textrm{P}}^c)$ defined in Remark \ref{rem:not-necessarily-mpgs} admits a proper vertex $4$-coloring under the hypothesis of Theorem \ref{thm:colored-mpg-lattice-Coloring-closure} and Theorem \ref{thm:colored-mpg-lattice-Coloring-closure22}.
\end{thm}

\subsubsection{Indexed-colorings}

For converting traditional colorings to string-colorings, we introduce the following indexed-colorings.

\begin{defn}\label{defn:general-graphs-colorings-redefine}
\cite{Yao-Sun-Wang-Su-Maximal-Planar-Graphs-2021} There are three \emph{indexed-colorings} of graphs as follows:

1. Suppose that a graph $G$ admits a proper vertex $k$-coloring $f:V(G)\rightarrow [1,k]$, such that the vertex set $V(G)=\bigcup^k_{i=1} V_i(G)$, and each vertex $u_{i,j}\in V_i(G)$ is colored with color $f(u_{i,j})=i_j$ for $j\in [1,n_i]$ and $i\in [1,k]$, where $n_i=|V_i(G)|$, we have vertex color sets
$$C_i(G)=\{f(x_{i,j}):x_{i,j}\in V_i(G)\}=\{i_1,i_2,\dots ,i_{n_i}\},~i\in [1,k]
$$ where the color $i_j$ is called \emph{$j$th color $i$}. For each edge $uv\in E(G)$, the \emph{edge induced-color} $f^*(uv)=F\langle f(u),f(v)\rangle$ defined by one of \textbf{addition} $a_i(+)b_j=(a+b)_{i+j}=F\langle f(u),f(v)\rangle$, \textbf{multiplication} $a_i(\cdot ) b_j=(a\cdot b)_{ij}=F\langle f(u),f(v)\rangle$ and \textbf{subtraction} $a_i(-)b_j=|a-b|_{|i-j|}=F\langle f(u),f(v)\rangle$, where two vertex colors $f(u)=a_i$ and $f(v)=b_j$. We call $f$ a \emph{proper vertex indexed-coloring}, and $\langle f,f^*\rangle$ an \emph{indexed proper total coloring}.

2. Suppose that a graph $G$ admits a proper edge $k$-coloring $g:E(G)\rightarrow [1,k]$, such that the edge set $E(G)=\bigcup^k_{r=1} E_r(G)$, and each edge $u_{r,j}v_{r,j}\in E_r(G)$ is colored with color $r_j$ for $j\in [1,s_r]$ and $r\in [1,k]$, where $s_r=|E_r(G)|$, we have edge color sets
$$
EC_r(G)=\{g(u_{r,j}v_{r,j}):u_{r,j}v_{r,j}\in E_r(G)\}=\{r_1,r_2,\dots ,r_{s_r}\},~r\in [1,k]
$$ For each vertex $x$ of $G$ with its neighbor set $N_{ei}(x)=\{y_1,y_2,\dots ,y_d\}$, the \emph{vertex induced-color} is as
\begin{equation}\label{eqa:indexed-proper-edge-coloring}
g^*(x)=A[r_1,r_2,\dots ,r_d]_{B[j(r_1),j(r_2),\dots ,j(r_d)]}
\end{equation} for $g(xy_{r_i,j})=(r_i)_{j(r_i)}$ with $xy_{r_i,j}\in E_{r_i}(G)$ and $r_i\in [1,k]$, where $A[r_1,r_2,\dots ,r_d]$ is a function of colors $r_i$ with $r_i\in [1,k]$, and $B[j(r_1),j(r_2),\dots ,j(r_d)]$ is a function of $j(r_i)$th color $r_i$ with $i\in [1,d]$, as well as degree $d=\textrm{deg}(x)$. We call $g$ a \emph{proper edge indexed-coloring}, and $\langle g,g^*\rangle$ a \emph{proper total indexed-coloring}.

3. Suppose that a $(p,q)$-graph $G$ admits a proper total $M$-coloring $h:V(G)\cup E(G)\rightarrow [1,M]$, such that the total set $V(G)\cup E(G)=\bigcup^M_{j=1} S_j(G)$, and each element $w_{j,k}\in S_j(G)$ is colored with color $j_k$ for $k\in [1,t_j]$ and $j\in [1,M]$, where $t_j=|S_j(G)|$ and $p+q=\sum ^M_{j=1} t_j$, we have the edge color sets
$$TC_j(G)=\{h(w_{j,k}):w_{j,k}\in S_j(G)\}=\{j_1,j_2,\dots ,j_{t_j}\},~j\in [1,M]
$$ We call $h$ a \emph{proper total indexed-coloring}.\qqed
\end{defn}

\begin{rem}\label{rem:333333}
\cite{Yao-Sun-Wang-Su-Maximal-Planar-Graphs-2021} In the proper edge $k$-coloring of Definition \ref{defn:general-graphs-colorings-redefine}, the vertex induced-color $g^*(x)$ defined in Eq.(\ref{eqa:indexed-proper-edge-coloring}) can be as

$A[r_1,r_2,\dots ,r_d]=\sum ^d_{t=1}r_t$ and $B[j(r_1),j(r_2),\dots ,j(r_d)]=\sum ^d_{t=1}j(r_t)$, or

$A[r_1,r_2,\dots ,r_d]=\prod ^d_{t=1}r_t$ and $B[j(r_1),j(r_2),\dots ,j(r_d)]=\prod ^d_{t=1}j(r_t)$, or

$A[r_1,r_2,\dots ,r_d]=\sum ^d_{t=1}r_t$ and $B[j(r_1),j(r_2),\dots ,j(r_d)]=\prod ^d_{t=1}j(r_t)$, or

$A[r_1,r_2,\dots ,r_d]=\prod ^d_{t=1}r_t$ and $B[j(r_1),j(r_2),\dots ,j(r_d)]=\sum ^d_{t=1}j(r_t)$, and so on.

There are the following characteristics of Definition \ref{defn:general-graphs-colorings-redefine} for the computational complexity of topological coding:

(i) The colors in each of color sets $C_i(G)$, $EC_r(G)$ and $TC_j(G)$ differ from each other, since their indexes are different from each other. Thereby, a $W$-constraint indexed-coloring induces more $W$-constraint indexed-colorings, where $W$-type$\in \{$proper vertex, proper edge, proper total$\}$.

(ii) The proper vertex indexed-coloring and the proper edge indexed-coloring both produce proper total indexed-colorings of $G$ by the vertex induced-colors and the edge induced-colors. There are many ways to produce the vertex induced-color $g^*(x)=A[r_1,r_2,\dots ,r_d]_{B[j(r_1),j(r_2),\dots ,j(r_d)]}$ and the edge induced-color $f^*(uv)=F(f(u),f(v))$.

(iii) See an indexed proper total $M$-coloring $h$ of a $(p,q)$-graph $G$ defined in Definition \ref{defn:general-graphs-colorings-redefine}. Let $a_{j_1},a_{j_2},\dots ,a_{j_{t_j}}$ be a permutation of colors $j_1,j_2,\dots ,j_{t_j}$ of $TC_j(G)$, so the elements of $S_j(G)$ defined in Definition \ref{defn:general-graphs-colorings-redefine} can be colored by $(t_j)!$ different ways. And moreover, there are $\Omega$ different indexed-colorings for the $(p,q)$-graph $G$ based on an indexed proper total $M$-coloring $h$, where $\Omega=\prod ^M_{j=1}(t_j)!$. Suppose that a $(p,q)$-graph $G$ admits $N_{tc}$ different proper total indexed-colorings, then we get $\Omega\cdot N_{tc}$ different proper total indexed-colorings of $G$. It will help us resist the attack of quantum computation by producing huge amount of topological codes.

(iv) Many of $\Omega\cdot N_{tc}$ different proper total indexed-colorings of $G$ are distinguishing colorings, so are those proper total indexed-colorings induced by the proper vertex indexed-coloring and the proper edge indexed-coloring.

(v) As $k=\chi(G), \chi\,'(G)$ and $M=\chi\,''(G)$, determining the colorings defined in Definition \ref{defn:general-graphs-colorings-redefine} is not easy, even difficult, since there are two longstanding conjectures:
$${
\begin{split}
&\textrm{Reed's conjecture: }~\chi(G)\leq \left \lceil \frac{\Delta(G)+1+K(G)}{2} \right \rceil, \\
&\textrm{Behzad and Vizing's conjecture: } ~\chi\,''(G)\leq \Delta(G)+2
\end{split}}
$$ proposed by Bruce Reed (1998), Behzad (1965), Vizing (1964), respectively. \paralled
\end{rem}

\begin{defn} \label{defn:equitable-total-coloring-def}
An $i$-\emph{color set} $S_i=V_i\cup E_i$ consists of $V_i=\{u: f(u)=i,u\in V(G)\}$ and $E_i=\{xy: f(xy)=i,xy\in E(G)\}$ foe a graph $G$. If $\big ||S_i|-|S_j|\big |\leq 1$ for any pair of color sets $S_i$ and $S_j$ $(1\leq i,j\leq k)$, we say that $f$ is a $k$-\emph{equitable total coloring} of the graph $G$.
Straightly, the number $$\chi\,''_{e}(G)=\min\{k: \text{ over all
$k$-equitable total colorings of }G \}$$ is called \emph{equitable total chromatic number} of the graph $G$.\qqed
\end{defn}

\begin{conj} \label{conj:c4-Weifan-Wang-2002}
(Weifan Wang, 2002) For every graph $G$, then $\chi\,''_e(G)\leq \Delta(G)+2$.
\end{conj}

\begin{rem}\label{rem:total-coloring-conjecture-00}
Let $\textbf{\textrm{H}}_{total}=(H_1,H_2,\dots ,H_n)$ be a \emph{graph base}, where each connected graph $H_i$ with $\chi\,''(H_i)\leq \Delta(H_i)+2$ has a property $P_{rop}(H_i)$ for $i\in [1,m]$, and each connected graph $H_j$ with a property $P_{rop}(H_j)$ holds $\chi\,''(H_j)\leq \Delta(H_j)+2$ for $j\in [m+1,n]$. Each connected graph $H_r$ admits a proper total coloring $f_r:V(H_r)\cup E(H_r)\rightarrow [1,\chi\,''(H_r)]$ for $r\in [1,n]$.

Let $G_{i,1},G_{i,2},\dots ,G_{i,M}$ be the $i$th permutation of graphs $a_1H_1,a_2H_2,\dots ,a_nH_n$ with $M=\sum^n_{k=1}a_k$ and $i\in [1,M!]$.

We do the vertex-coinciding operation to each graph $G_{i,j}$ for $j\in [1,M]$ and $i\in [1,M!]$ in the following way: Let $u_r\in V(G_{i,r})$ be a \emph{maximal degree vertex} with $\textrm{deg}_{G_{i,r}}(u_r)=\Delta(G_{i,r})$ for $r\in [1,M]$, and without loss of generality, $f_r(u_r)=1$. We vertex-coincide a maximal degree vertex $u_r$ with a maximal degree vertex $u_{r+1}$ into one vertex $u_r\odot u_{r+1}$, define a new total coloring $f^*$ for the graph $G_{i,r}[\odot]G_{i,r+1}$ as: $f^*(u_{r})=f^*(u_{r+1})=1$, and recolor edges $u_{r+1}v_{r+1,s}$ with $f^*(u_{r+1}v_{r+1,s})=\Delta(H_{r})+2+s$ for $s\in [1,\Delta(H_{r+1})]$ if $\textrm{deg}_{H_r}(u_r)\geq \textrm{deg}_{H_{r+1}}(u_{r+1})$, so $$f^*(u_{r+1}v_{r+1,\Delta(G_{i,r+1})})=\Delta(G_{i,r})+2+\Delta(G_{i,r+1})=\Delta(G_{i,r}[\odot]G_{i,r+1})+2
$$ and keep the colors of other vertices and edges in the graph $G_{i,r}[\odot]G_{i,r+1}$ as they were. The vertex-coincided graph $G$ is denoted as
$$G=G_{i,1}[\odot]G_{i,2}[\odot]\cdots [\odot]G_{i,M}=[\odot]^n_{k=1}a_kH_k
$$ and
\begin{equation}\label{eqa:total-coloring-conjecture}
E(G)=\bigcup ^n_{k=1}a_kE(H_k)
\end{equation} It is allowed that the maximal degree vertex $u_1$ is vertex-coincided with each maximal degree vertex $u_r$ for $r\in [2,n]$ such that the vertex-coincided graph $G=[\odot]^n_{k=1}a_kH_k$ has
$\Delta(G)=\sum^n_{k=1} \Delta(G_{i,k})$. Thereby, the vertex-coincided graph $G=[\odot]^n_{k=1}a_kH_k$ is simple and connected, and admits a proper total coloring $f^*$ holding
\begin{equation}\label{eqa:total-coloring-conjecture11}
\max \{f^*(w):w\in V(G)\cup E(G)\}\leq \Delta(G)+2
\end{equation} and holds Eq.(\ref{eqa:total-coloring-conjecture}). Then we get a \emph{graphic lattice}
\begin{equation}\label{eqa:total-coloring-conjecture-lattice}
\textbf{\textrm{L}}([\odot]Z^0\textbf{\textrm{H}}_{total})=\Big \{[\odot]^n_{k=1}a_kH_k:a_k\in Z^0,H_k\in \textbf{\textrm{H}}_{total}\Big \}
\end{equation} with $M=\sum^n_{k=1}a_k\geq 1$. As $a_i\geq 1$ and $a_j\geq 1$ for $i\neq j$, the connected graph $G\in \textbf{\textrm{L}}([\odot]Z^0\textbf{\textrm{H}}_{total})$ does not hold two properties $P_{rop}(H_i)$ and $P_{rop}(H_j)$ at the same time. In other word, if there are no necessary and sufficient conditions for Behzad and Vizing's conjecture, then the number of connected graphs in the graph base $\textbf{\textrm{H}}_{total}$ will get bigger and bigger.\paralled
\end{rem}

\begin{problem}\label{question:444444}
By Remark \ref{rem:total-coloring-conjecture-00}, it is natural to guess: ``A connected graph $G$ holds $\chi\,''(G)\leq \Delta(G)+2$ if and only if $G$ can be vertex-split into mutually edge-disjoint connected graphs $G_1,G_2,\dots ,G_m$ such that each connected graph $G_i$ holds $\chi\,''(G_i)\leq \Delta(G_i)+2$ for $i\in [1,m]$''.

We, also, can conjecture: ``A connected graph $G$ holds $\chi\,''(G)\leq \Delta(G)+2$ if and only if there is a graph base $\textbf{\textrm{H}}_{total}=(H_1,H_2,\dots ,H_n)$ such that the connected graph $G$ belongs to a graphic lattice $\textbf{\textrm{L}}([\odot]Z^0\textbf{\textrm{H}}_{total})$ defined in Eq.(\ref{eqa:total-coloring-conjecture-lattice})''.
\end{problem}

\subsubsection{Indexed-colorings of planar graphs}

\begin{defn}\label{defn:planar-graphs-color-sets-4-operations}
\cite{Yao-Sun-Wang-Su-Maximal-Planar-Graphs-2021} Suppose that a planar graph $T$ admits a proper vertex $4$-coloring $g$, so its vertex set $V(T)$ can be divided into four subsets, namely, $V(T)=\bigcup^4_{k=1} V_k(T)$ such that each vertex $x_{k,j}\in V_k(T)$ is colored with color $k_j$ for $j\in [1,m_k]$ and $k\in [1,4]$, where $m_k=|V_k(T)|$. So we have four vertex color sets
$$C_k(T)=\{g(x_{k,j}):x_{k,j}\in V_k(T)\}=\{k_1,k_2,\dots ,k_{m_k}\},~k\in [1,4]
$$ And for two ends of each edge $uv\in E(G)$ colored with $g(u)=a_i$ and $g(v)=b_j$, we define the edge color $g(uv)$ with one of \emph{indexed operations}: \textbf{indexed addition} $a_i(+)b_j=(a+b)_{(i+j)}$, \textbf{indexed multiplication} $a_i(\cdot ) b_j=ab_{ij}$ and \textbf{indexed subtraction} $a_i(-)b_j=|a-b|_{|i-j|}$. Then the proper vertex $4$-coloring $g$ and the indexed operations induce a \emph{total indexed-coloring} of $T$.\qqed
\end{defn}

\begin{example}\label{exa:8888888888}
In Fig.\ref{fig:mpg-4-color-authen-new}, a \emph{public-key graph} $G$ corresponds a \emph{private-key graph} $H_1$, and a topological signature authentication $T_1=\langle G, H_1\rangle$; another topological signature authentication $T_2=\langle G, H_2\rangle$ made by the public-key graph $G$ and the private-key graph $H_2$. There is a cycle $C_6$ in $G$, $H_1$, $T_1$, $H_2$ and $T_2$. We have

(i) the public-key graph $G$ admits a proper vertex 4-coloring $f$;

(ii) each private-key graph $H_i$ admits a proper vertex 4-coloring $g_i$ for $i=1,2$;

(iii) each topological signature authentication $T_i$ admits a proper vertex 4-coloring $F_i$ for $i=1,2$;

(iv) the cycle $C_6$ admits a proper vertex 4-coloring $h$.\\
Moreover, we have

$V(G)=\bigcup^4_{i=1} V_i(G)$, then the vertex color sets $f(V_1(G))=\{1_1,1_2,1_3,1_4\}$, $f(V_2(G))=\{2_1\}$, $f(V_3(G))=\{3_1,3_2\}$ and $f(V_4(G))=\{4_1,4_2,4_3\}$.

$V(H_1)=\bigcup^4_{i=1} V_i(H_1)$, then the vertex color sets $g_1(V_1(H_1))=\{1_2,1_3,1_4\}$, $g_1(V_2(H_1))=\{2_2\}$, $g_1(V_3(H_1))=\{3_2,3_3\}$ and $g_1(V_4(H_1))=\{4_2,4_3\}$.

$V(T_1)=\bigcup^4_{i=1} V_i(T_1)$, then the vertex color sets $F(V_1(T_1))=\{1_1,1_2,1_3,1_4\}$, $F(V_2(T_1))=\{2_1,2_2\}$, $F(V_3(T_1))=\{3_1,3_2,3_3\}$ and $F(V_4(T_1))=\{4_1,4_2,4_3\}$.

\begin{figure}[h]
\centering
\includegraphics[width=16.4cm]{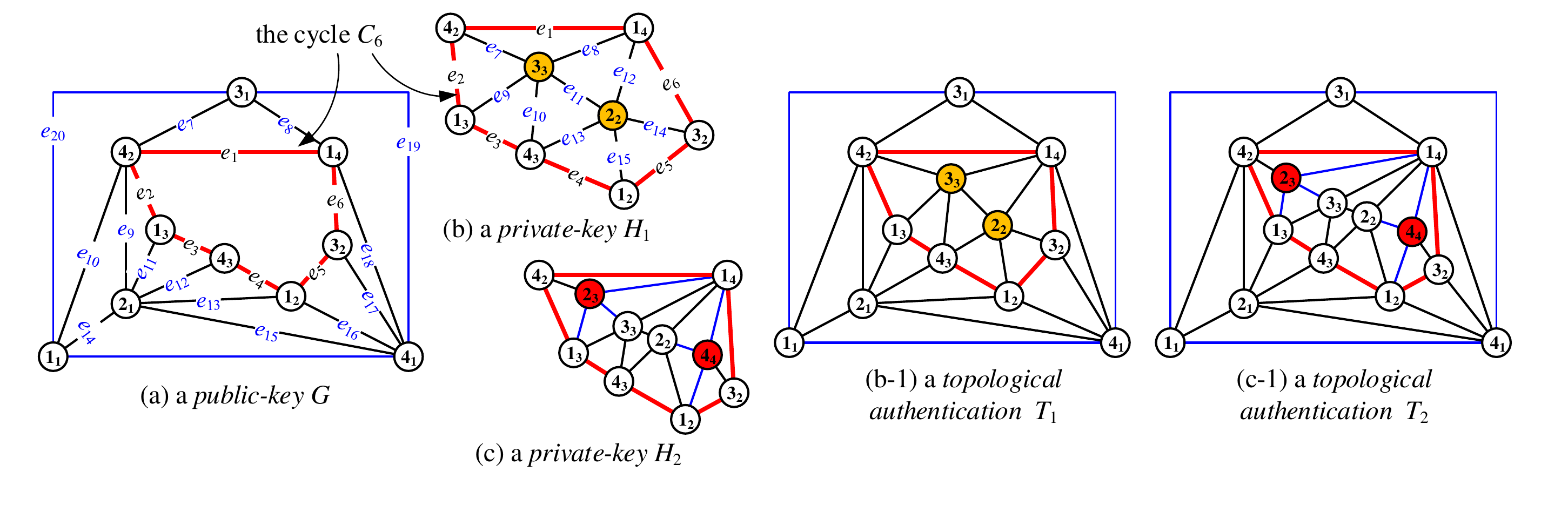}\\
\caption{\label{fig:mpg-4-color-authen-new}{\small Examples for the topological signature authentication based on maximal planar graphs, cited from \cite{Yao-Sun-Wang-Su-Maximal-Planar-Graphs-2021}.}}
\end{figure}

\textbf{Topological matrices.} We get the following Topcode-matrices $T_{code}(G,f)_{3\times 20}$ shown in Eq.(\ref{eqa:public-key-G-matrix}), $T_{code}(H_1,g_1)_{3\times 15}$ shown in Eq.(\ref{eqa:private-key-H1-topcode-matrix}), $T_{code}(C_6,h)_{3\times 6}$ shown in Eq.(\ref{eqa:common-cycle-C6-topcode-matrix}), and the Topcode-matrix $T_{code}(T_1,F_1)_{3\times 29}$ shown in Eq.(\ref{eqa:44-topcode-matrices-connections}), and the following connections
\begin{equation}\label{eqa:44-topcode-matrices-connections}
{
\begin{split}
&T_{code}(G,f)\cap T_{code}(H_1,g_1)=T_{code}(C_6,h)_{3\times 6}\\
&T_{code}(T_1,F_1)_{3\times 29}=T_{code}(G,f)\cup \big [T_{code}(H_1,g_1)\setminus T_{code}(C_6,h)\big ]
\end{split}}
\end{equation}

{\footnotesize
\begin{equation}\label{eqa:public-key-G-matrix}
\centering
{
\begin{split}
T_{code}(G,f)= \left(
\begin{array}{lcccccccccccccccc}
1_{4}~~4_{2}~~1_{3}~~1_{2}~~3_{2}~~3_{2}~~4_{2}~~1_{4}~~4_{2}~~4_{2} ~~~1_{3} ~~~4_{3}~~~1_{2}~~~1_{1}~~~4_{1}~~~4_{1}~~~3_{2}~~~1_{4}~~~3_{1}~~~3_{1}\\
e_{1}~~e_{2}~~e_{3}~~e_{4}~~e_{5}~~e_{6}~~e_{7}~~e_{8}~~e_{9}~~e_{10}~~e_{11}~~e_{12}~~e_{13}~~e_{14}~~e_{15}~~e_{16}~~e_{17}~~e_{18}~~e_{19}~~e_{20}\\
4_{2}~~1_{3}~~4_{3}~~4_{3}~~1_{2}~~1_{4}~~3_{1}~~3_{1}~~2_{1}~~1_{1} ~~~2_{1} ~~~2_{1}~ ~~2_{1}~~2_{1}~~~2_{1}~~~1_{2}~~~~4_{1}~~~4_{1}~~~~4_{1}~~~1_{1}
\end{array}
\right)
\end{split}}
\end{equation}
}

\begin{equation}\label{eqa:private-key-H1-topcode-matrix}
\centering
{
\begin{split}
T_{code}(H_1,g_1)= \left(
\begin{array}{ccccccccccccccccc}
1_{4} & 4_{2} & 1_{3} & 1_{2} & 3_{2} & 3_{2} & 4_{2} & 1_{4} & 3_{3} & 3_{3} & 3_{3} & 1_{4} & 4_{3} & 3_{2} & 1_{2}\\
e_{1} & e_{2} & e_{3} & e_{4} & e_{5} & e_{6} & e_{7} & e_{8} & e_{9} & e_{10} & e_{11} & e_{12} & e_{13} & e_{14} & e_{15}\\
4_{2} & 1_{3} & 4_{3} & 4_{3} & 1_{2} & 1_{4} & 3_{3} & 3_{3} & 1_{3} & 4_{3} & 2_{2} & 2_{2} & 2_{2} & 2_{2} & 2_{2}
\end{array}
\right)
\end{split}}
\end{equation}

\begin{equation}\label{eqa:common-cycle-C6-topcode-matrix}
\centering
{
\begin{split}
T_{code}(C_6,h)= \left(
\begin{array}{ccccccccccccccccc}
1_{4} & 4_{2} & 1_{3} & 1_{2} & 3_{2} & 3_{2}\\
e_{1} & e_{2} & e_{3} & e_{4} & e_{5} & e_{6} \\
4_{2} & 1_{3} & 4_{3} & 4_{3} & 1_{2} & 1_{4}
\end{array}
\right)
\end{split}}
\end{equation}

\textbf{Computational complexity.} The Topcode-matrix $T_{code}(G,f)_{3\times 20}$ produces $(60)!$ different number-based string; the Topcode-matrix $T_{code}(H_1,g_1)_{3\times 15}$ can distribute us $(45)!$ different number-based string; the Topcode-matrix $T_{code}(C_6,h)_{3\times 6}$ can make $(18)!$ different number-based string; and the Topcode-matrix $T_{code}(T_1,F_1)_{3\times 29}$ can distribute us $(87)!$ different number-based string.

\textbf{Total colorings.} By Definition \ref{defn:planar-graphs-color-sets-4-operations}, the proper 4-coloring $h$ admitted by the cycle $C_6$ induces the edge color of $C_6$ for forming total colorings $h(+_+)$, $h(-_-)$, $h(\times_{\times})$ and $h(\times_+)$ shown in Eq.(\ref{eqa:common-cycle-C6-topcode-matrix-2add}), where four Topcode-matrices are as follows:
$$A_1=T_{code}(C_6,h(+_+)),A_2=T_{code}(C_6,h(-_-)),A_3=T_{code}(C_6,h(\times_{\times})),A_4=T_{code}(C_6,h(\times_{+}))$$

\begin{equation}\label{eqa:common-cycle-C6-topcode-matrix-2add}
\centering
{
\begin{split}
&A_1= \left(
\begin{array}{ccccccccccccccccc}
1_{4} & 4_{2} & 1_{3} & 1_{2} & 3_{2} & 3_{2}\\
5_{6} & 5_{5} & 5_{6} & 5_{5} & 4_{4} & 4_{6} \\
4_{2} & 1_{3} & 4_{3} & 4_{3} & 1_{2} & 1_{4}
\end{array}
\right),~A_2= \left(
\begin{array}{ccccccccccccccccc}
1_{4} & 4_{2} & 1_{3} & 1_{2} & 3_{2} & 3_{2}\\
3_{2} & 3_{1} & 3_{0} & 3_{1} & 2_{0} & 2_{2} \\
4_{2} & 1_{3} & 4_{3} & 4_{3} & 1_{2} & 1_{4}
\end{array}
\right),\\
&A_3= \left(
\begin{array}{ccccccccccccccccc}
1_{4} & 4_{2} & 1_{3} & 1_{2} & 3_{2} & 3_{2}\\
4_{8} & 4_{6} & 4_{9} & 4_{6} & 3_{4} & 3_{8} \\
4_{2} & 1_{3} & 4_{3} & 4_{3} & 1_{2} & 1_{4}
\end{array}
\right),~A_4= \left(
\begin{array}{ccccccccccccccccc}
1_{4} & 4_{2} & 1_{3} & 1_{2} & 3_{2} & 3_{2}\\
4_{6} & 4_{5} & 4_{6} & 4_{5} & 3_{4} & 3_{6} \\
4_{2} & 1_{3} & 4_{3} & 4_{3} & 1_{2} & 1_{4}
\end{array}
\right)
\end{split}}
\end{equation}

\textbf{Number-based strings.} We show some number-based strings shown in Eq.(\ref{eqa:number-based-five-strings-matrix}):
\begin{equation}\label{eqa:number-based-five-strings-matrix}
\centering
{
\begin{split}
\begin{array}{ll}
&s_{C_6}(24)=144213123232421343431214\\
&s^{+_+}_{C_6}(36)=144213123232565556554446421343431214\\
&s^{-_-}_{C_6}(36)=144213123232323130312022421343431214\\
&s^{\times_\times}_{C_6}(36)=144213123232484649463438421343431214\\
&s^{\times_+}_{C_6}(36)=144213123232464546453436421343431214
\end{array}
\end{split}}
\end{equation} where

$T_{code}(C_6,h)$ produces the number-based string $s_{C_6}(24)$ with 24 bytes;

$T_{code}(C_6,h(+_+))$ produces the number-based string $s^{+_+}_{C_6}(36)$ with 36 bytes;

$T_{code}(C_6,h(-_-))$ produces the number-based string $s^{-_-}_{C_6}(36)$ with 36 bytes;

$T_{code}(C_6,h(\times_{\times}))$ produces the number-based string $s^{\times_\times}_{C_6}(36)$ with 36 bytes;

$T_{code}(C_6,h(\times_{+}))$ produces the number-based string $s^{\times_+}_{C_6}(36)$ with 36 byte.\qqed
\end{example}

\begin{thm}\label{thm:planar-graph-M-coloring plans}
$^*$ Since there are $(m_k)!$ permutations of $k_1,k_2,\dots ,k_{m_k}$ to color the vertices of $V_k(T)$ of a planar graph $T$ in Definition \ref{defn:planar-graphs-color-sets-4-operations}, so we have $\prod^4_{k=1} (m_k)!$ different coloring plans under a proper vertex $4$-coloring $g$. Suppose that the planar graph $T$ admits $|C^0_4(T)|$ different proper vertex $4$-colorings, then we have $M(T)$ different coloring plans for $T$, where the number
$$
M(T)=|C^0_4(T)|\cdot \prod^4_{k=1} (m_k)!
$$ refer to \cite{Jin-Xu-55-56-configurations-arXiv-2107-05454v1} and \cite{Jin-Xu-Maximal-Science-Press-2019}.
\end{thm}

\begin{rem}\label{rem:333333}
\cite{Yao-Su-Ma-Wang-Yang-arXiv-2202-03993v1} In Definition \ref{defn:planar-graphs-color-sets-4-operations}, we have defined the edge color of an edge $uv$ with two ends colors $g(u)=a_i$ and $g(v)=b_j$ under a proper vertex $4$-coloring $g$. Furthermore, we have
\begin{asparaenum}[\textbf{Oper}-1. ]
\item \textbf{Three proper operations}: Addition $a_i(+)b_j=(a+b)_{(i+j)}$, multiplication $a_i(\cdot ) b_j=ab_{ij}$, subtraction $a_i(-)b_j=|a-b|_{|i-j|}$.
\item \textbf{Mixed operation 1}: $a_i(+\times )b_j=(a+b)_{ij}$, $a_i(+|-|)b_j=(a+b)_{|i-j|}$.
\item \textbf{Mixed operation 2}: $a_i(\times +)b_j=(a\cdot b)_{i+j}$, $a_i(\times |-|)b_j=(a\cdot b)_{|i-j|}$.
\item \textbf{Mixed operation 3}: $a_i(|-|\times )b_j=|a-b|_{ij}$, $a_i(|-|+)b_j=|a-b|_{i+j}$.
\item \textbf{Operation based on Klein four-group}: Let $0:=1$, $a:=2$, $b:=3$, $c:=4$ in the table $K^+_{\textrm{lein}}$ shown in Eq.(\ref{eqa:Klein-field-four-group}). We define $(K^+_{\textrm{lein}})$ addition with the commutative law $r_i(K^+_{\textrm{lein}})s_j=s_j(K^+_{\textrm{lein}})r_i$ in the following computations: $1_i(K^+_{\textrm{lein}})1_j=1_{i+j}$, $1_i(K^+_{\textrm{lein}})2_j=2_{i+j}$, $1_i(K^+_{\textrm{lein}})3_j=3_{i+j}$, $1_i(K^+_{\textrm{lein}})4_j=4_{i+j}$; $2_i(K^+_{\textrm{lein}})2_j=2_{i+j}$, $2_i(K^+_{\textrm{lein}})3_j=4_{i+j}$, $2_i(K^+_{\textrm{lein}})4_j=3_{i+j}$; $3_i(K^+_{\textrm{lein}})3_j=1_{i+j}$, $3_i(K^+_{\textrm{lein}})4_j=2_{i+j}$; and $4_i(K^+_{\textrm{lein}})4_j=1_{i+j}$; refer to Fig.\ref{fig:Klein-addition-multiplication} (a) and (a-1).

\quad We define $(K^{\times}_{\textrm{lein}})$ multiplication with the commutative law $r_i(K^{\times}_{\textrm{lein}})s_j=s_j(K^{\times}_{\textrm{lein}})r_i$ as follows: $1_i(K^{\times}_{\textrm{lein}})1_j=1_{ij}$, $1_i(K^{\times}_{\textrm{lein}})2_j=1_{ij}$, $1_i(K^{\times}_{\textrm{lein}})3_j=1_{ij}$, $1_i(K^{\times}_{\textrm{lein}})4_j=1_{ij}$; $2_i(K^{\times}_{\textrm{lein}})2_j=2_{ij}$, $2_i(K^{\times}_{\textrm{lein}})3_j=3_{ij}$, $2_i(K^{\times}_{\textrm{lein}})4_j=4_{ij}$; $3_i(K^{\times}_{\textrm{lein}})3_j=3_{ij}$, $3_i(K^{\times}_{\textrm{lein}})4_j=2_{ij}$; and $4_i(K^{\times}_{\textrm{lein}})4_j=3_{ij}$; refer to Fig.\ref{fig:Klein-addition-multiplication} (b) and (b-1).\paralled
\end{asparaenum}
\end{rem}

\begin{rem}\label{rem:333333}
\cite{Yao-Su-Ma-Wang-Yang-arXiv-2202-03993v1} Let $K_{\textrm{lein}}= \{0, a, b, c\}$ be \emph{Klein four-group} with addition ``$+$'' in $K^+_{\textrm{lein}}$ shown in Eq.(\ref{eqa:Klein-field-four-group}). The Klein four-group can be extended to a \emph{finite field}, called the \emph{Klein field}, where multiplication is added as a second operation, with $0$ as the zero element and $a$ as the identity element. The multiplication table is $K^{\times}_{\textrm{lein}}$ shown in Eq.(\ref{eqa:Klein-field-four-group}). Multiplication ``$\times$'' and addition ``$+$'' obey the distributive law.

\begin{equation}\label{eqa:Klein-field-four-group}
\centering
K^+_{\textrm{lein}}=
\begin{array}{c|cccc|c}
+ & 0 & a & b & c & \textrm{substitution}\\
\hline
0 & 0 & a & b & c\\
a & a & 0 & c & b & (1,2)(3,4)\\
b & b & c & 0 & a & (1,3)(2,4) \\
c & c & b & a & 0 & (1,4)(2,3)
\end{array}
\qquad K^{\times}_{\textrm{lein}}=
\begin{array}{c|ccccc}
\times & 0 & a & b & c \\
\hline
0 & 0 & 0 & 0 & 0 \\
a & 0 & a & b & c \\
b & 0 & b & c & a \\
c & 0 & c & a & b
\end{array}
\end{equation}

A Klein four-group, also, is a normal subgroup of the \emph{alternating group} $A_4$ and of the \emph{symmetric group} $S_4$ over four letters.\paralled
\end{rem}

\begin{figure}[h]
\centering
\includegraphics[width=16.4cm]{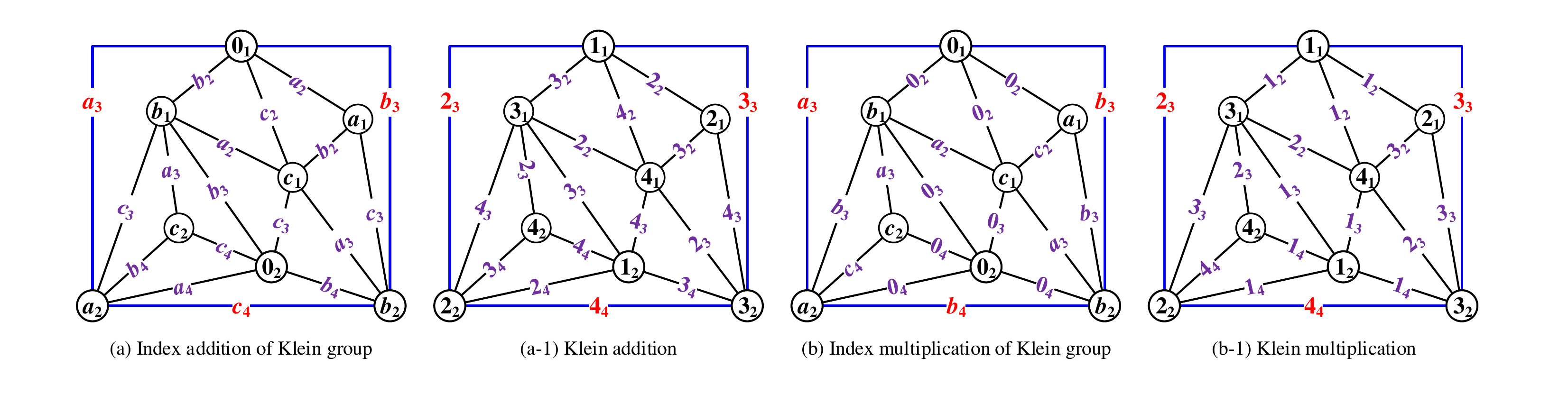}\\
\caption{\label{fig:Klein-addition-multiplication}{\small Examples for the addition and multiplication based on two Klein groups $K^+_{\textrm{lein}}$ and $K^{\times}_{\textrm{lein}}$ shown in Eq.(\ref{eqa:Klein-field-four-group}).}}
\end{figure}

\subsection{Hanzi-based strings}

Since Hanzi-graphs admit many colorings/labelings defined in this article, applying them to the field of information will greatly help people of using Chinese characters to be unimpeded in digital finance and personal privacy protection, on the other hands, Hanzi-graphs can be feeded into a computer by speaking, writing and keyboard inputs in Chinese \cite{Yao-Wang-2106-15254v1, Yao-Su-Ma-Wang-Yang-arXiv-2202-03993v1, Yao-Mu-Sun-Sun-Zhang-Wang-Su-Zhang-Yang-Zhao-Wang-Ma-Yao-Yang-Xie2019, Yarong-Mu-Bing-Yao-Hanzi-graphs-2018, Zhang-Yang-Mu-Zhao-Sun-Yao-2019}. Hanzi-graphs are planar graphs, so they are easy to be imputed into computer by many techniques.

\subsubsection{Planar graphs and Hanzi-based strings}

Let $C_{sent}=H_{a_1b_1c_1d_1}H_{a_2b_2c_2d_2}\cdots H_{a_nb_nc_nd_n}$ be a Hanzi-based string, where each Hanzi-graph $H_{a_ib_ic_id_i}$ for $i\in [1,n]$ is defined in \cite{GB2312-80}. Suppose that each Hanzi-graph $H_{a_ib_ic_id_i}$ for $i\in [1,n]$ admits $m_i$ colorings $f_{i,1}, f_{i,2},\dots ,f_{i,m_i}$ defined on sets of numbers, then we have \emph{Hanzi-graph Topcode-matrices}
\begin{equation}\label{eqa:555555}
T_{code}(H_{a_ib_ic_id_i},f_{i,j})_{3\times q_i}=(X^j_{a_ib_ic_id_i},E^j_{a_ib_ic_id_i},Y^j_{a_ib_ic_id_i})^T
\end{equation} where the edge number $q_i=|E(H_{a_ib_ic_id_i})|$ for $j\in [1,m_i]$ and $i\in [1,n]$.

Each Hanzi-graph Topcode-matrix $T_{code}(H_{a_ib_ic_id_i},f_{i,j})_{3\times q_i}$ produces $(3q_i)!$ different Hanzi-based strings in total, then each Hanzi-graph $H_{a_ib_ic_id_i}$ produces $m_i(3q_i)!$ different number-based strings.

Thereby, the Chinese sentence $C_{sent}$ provides us
$$n_{string}(C_{sent})=\sum^n_{i=1}m_i(3q_i)!$$ different Hanzi-based strings in total.

\begin{figure}[h]
\centering
\includegraphics[width=16.4cm]{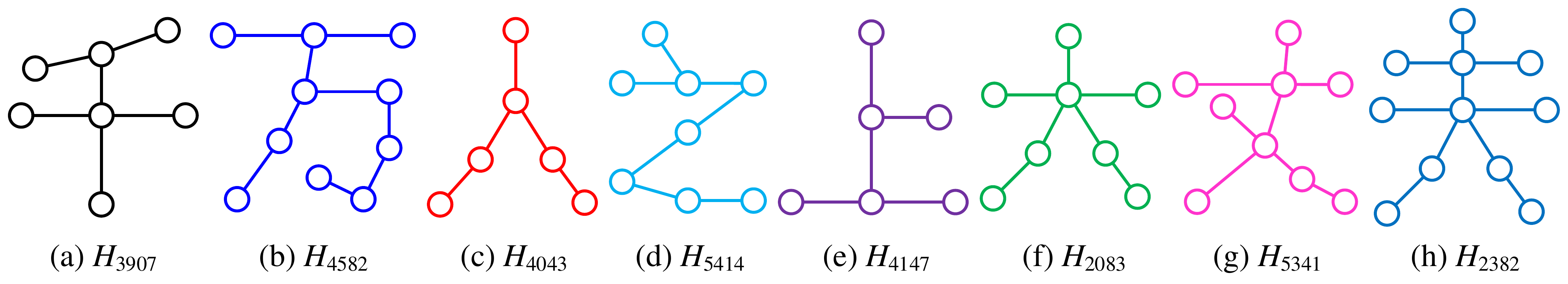}\\
\caption{\label{fig:hundreds-people}{\small A Hanzi-based string $H_{an}=H_{3907}H_{4582}H_{4043}H_{5414}H_{4147}H_{2083}H_{5341}H_{2382}$.}}
\end{figure}

There are $40320=8!$ sentences made by the Chinese letters $H_{3907}$, $H_{4582}$, $H_{4043}$, $H_{5414}$, $H_{4147}$, $H_{2083}$, $H_{5341}$, and $H_{2382}$ shown in Fig.\ref{fig:hundreds-people}, in total; part of them are shown in Fig.\ref{fig:Hanzi-combinators}.

\begin{example}\label{exa:8888888888}
Fig.\ref{fig:Hanzi-centence-colorings} shows us some colorings and labelings of the Hanzi-graphs $H_{3907}$, $H_{4582}$, $H_{4043}$, $H_{5414}$, $H_{4147}$, $H_{2083}$, $H_{5341}$ and $H_{2382}$ shown in Fig.\ref{fig:hundreds-people} as follows:

(H-1) the Hanzi-graph $H_{3907}$ admits an \emph{edge-magic total labeling} $f_1$ holding the edge-magic constraint
$$f_1(u)+f_1(uv)+f_1(v)=10$$ for each edge $uv\in E(H_{3907})$, such that the vertex color set $f_1(V(H_{3907}))=[0,6]$ and the edge color set $f_1(E(H_{3907}))=[1,6]$;

(H-2) the Hanzi-graph $H_{4582}$ admits a non-set-ordered odd-graceful labeling $f_2$ holding each odd number $f_2(uv)=|f_2(u)-f_2(v)|$ for each edge $uv\in E(H_{4582})$, such that the vertex color set $f_2(V(H_{4582}))\subset [0,17]$ and the edge color set $f_2(E(H_{4582}))=[1,17]^o$;

(H-3) the Hanzi-graph $H_{4043}$ admits an edge-difference total labeling $f_3$ holding the edge-difference constraint
$$f_3(uv)+|f_3(u)-f_3(v)|=6$$ for each edge $uv\in E(H_{4043})$, such that the vertex color set $f_3(V(H_{4043}))=[0,5]$ and the edge color set $f_3(E(H_{4043}))=[1,5]$;

(H-4) the Hanzi-graph $H_{5414}$ admits a felicitous-difference total labeling $f_4$ holding the felicitous-difference constraint
$$|f_4(u)+f_4(v)-f_4(uv)|=2$$ for each edge $uv\in E(H_{5414})$, such that the vertex color set $f_4(V(H_{5414}))=[0,7]$ and the edge color set $f_4(E(H_{5414}))=[1,7]$;

(H-5) the Hanzi-graph $H_{4147}$ admits an all-odd set-ordered graceful-difference total labeling $f_5$ holding the graceful-difference constraint
$$\big ||f_5(v)-f_5(u)|-f_5(uv)\big |=1$$ for each edge $uv\in E(H_{4147})$, such that the vertex color set $f_5(V(H_{4147}))=[1,9]^o$ and the edge color set $f_5(E(H_{4147}))=[1,11]^o$;

(H-6) the Hanzi-graph $H_{2083}$ admits an all-odd set-ordered graceful-difference total labeling $f_6$ holding $f_6(uv)=f_6(u)+f_6(v)~(\bmod~7)$ for each edge $uv\in E(H_{2083})$, such that the vertex color set $f_6(V(H_{2083}))=[1,9]^o$ and the edge color set $f_6(E(H_{2083}))=[1,11]^o$;

(H-7) the Hanzi-graph $H_{5341}$ admits an odd-even-separated set-ordered graceful-difference total labeling $f_7$ holding each even number $f_7(uv)=|f_7(u)-f_7(v)|$ for each edge $uv\in E(H_{5341})$, such that the vertex color set $f_7(V(H_{5341}))=[1,17]^o$ and the edge color set $f_7(E(H_{5341}))=[2,16]^e$;

(H-8) the Hanzi-graph $H_{2382}$ admits a non-set-ordered graceful labeling $f_8$ holding $f_8(uv)=|f_8(u)-f_8(v)|$ for each edge $uv\in E(H_{2382})$, such that the vertex color set $f_8(V(H_{2382}))=[0,10]$ and the edge color set $f_8(E(H_{2382}))=[1,10]$.\qqed
\end{example}

\begin{figure}[h]
\centering
\includegraphics[width=12cm]{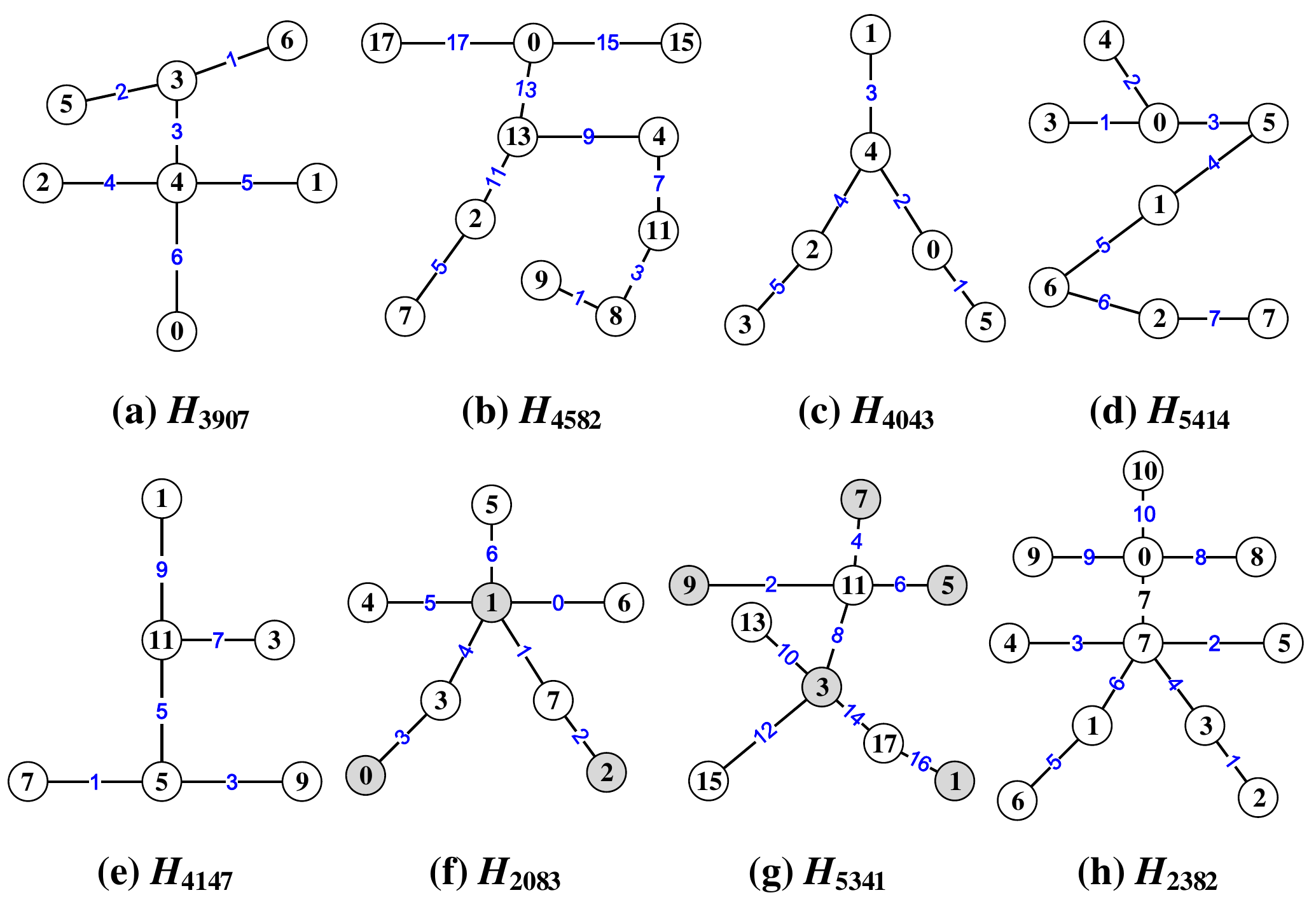}\\
\caption{\label{fig:Hanzi-centence-colorings}{\small Some colorings and labelings of the Hanzi-graphs $H_{3907},H_{4582},H_{4043},H_{5414},H_{4147},H_{2083},H_{5341}$ and $H_{2382}$.}}
\end{figure}

About planar graphs and Hanzi-graphs, we have the following jobs to do:

(i) \textbf{Hanzi-graphs $\rightarrow$ planar graphs}. Let
\begin{equation}\label{eqa:Hanzi-graph-base11}
\textbf{\textrm{H}}_{an}=(H_{a_1b_1c_1d_1},H_{a_2b_2c_2d_2},\dots ,H_{a_nb_nc_nd_n})
\end{equation}be a \emph{Hanzi-graph base} made by mutually disjoint Hanzi-graphs $H_{a_1b_1c_1d_1}$, $H_{a_2b_2c_2d_2}$, $\dots $, $H_{a_nb_nc_nd_n}$, where each Hanzi-graph $H_{a_ib_ic_id_i}$ for $i\in [1,n]$ is defined in \cite{GB2312-80}. Each graph $G=[\odot_{\textrm{plan}}]^n_{i=1}H_{a_ib_ic_id_i}$ is a planar graph; see examples shown in Fig.\ref{fig:Hanzi-vs-planar-graphs}. \textbf{Find} $\min\{|V(G)|:~G=[\odot_{\textrm{plan}}]^n_{i=1}H_{a_ib_ic_id_i}\}$.

\begin{figure}[h]
\centering
\includegraphics[width=16.4cm]{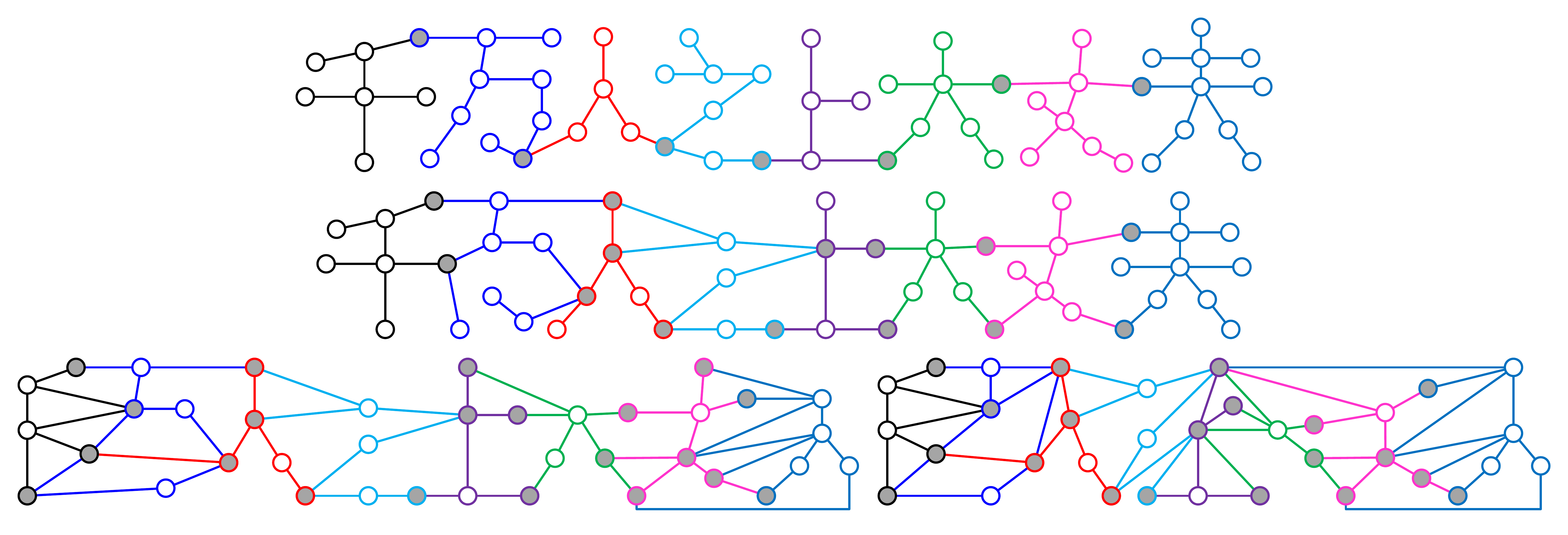}\\
\caption{\label{fig:Hanzi-vs-planar-graphs}{\small Different planar graphs obtained by vertex-coinciding operation to the Hanzi-based string $H_{an}$ shown in Fig.\ref{fig:hundreds-people}.}}
\end{figure}

Finding another planar graph $G^*$ forms a semi-planar graph
$$G^{semi}=G^*[\odot_{\textrm{plan}}]G=G^*[\odot_{\textrm{plan}}]\Big ([\odot_{\textrm{plan}}]^n_{i=1}H_{a_ib_ic_id_i}\Big )
$$ such that the vertex number $|V(G^*)|$ is as smaller as possible.

Conversely, a semi-planar graph $J$ can be vertex-split into different groups of Hanzi-graphs $G_{i,1}$, $ G_{i,2}$, $\dots $, $G_{i,n_i}$, and put them into a set $C_{split}(J)$, \textbf{determine} the set $C_{split}(J)$ and the number $\min\{n_i\}$.

(ii) \textbf{Color} planar graphs $G=[\odot_{\textrm{plan}}]^n_{i=1}H_{a_ib_ic_id_i}$ based on a Hanzi-graph base $\textbf{\textrm{H}}_{an}$ by the indexed-colorings defined in Definition \ref{defn:planar-graphs-color-sets-4-operations}.

(iii) We have a \emph{$[\odot_{\textrm{plan}}]$-graphic lattice} defined as
\begin{equation}\label{eqa:555555}
\textbf{\textrm{L}}([\odot_{\textrm{plan}}]Z^0\textbf{\textrm{H}}_{an})=\Big \{[\odot_{\textrm{plan}}]^n_{i=1}\lambda_iH_{a_ib_ic_id_i}:\lambda_i\in Z^0,H_{a_ib_ic_id_i}\in \textbf{\textrm{H}}_{an}\Big \}
\end{equation} based on the Hanzi-graph base $\textbf{\textrm{H}}_{an}$ shown in Eq.(\ref{eqa:Hanzi-graph-base11}), where $\sum ^n_{i=1}\lambda_i\geq 1$.

A process of starting from a semi-planar graph $J$ and coming back to a semi-planar graph $G^{semi}$ is as follows:
\begin{equation}\label{eqa:555555}
J \rightarrow _{vsplt} \textbf{\textrm{H}}_{an} \rightarrow \textbf{\textrm{L}}([\odot_{\textrm{plan}}]Z^0\textbf{\textrm{H}}_{an}) \rightarrow G \rightarrow G^*[\odot_{\textrm{plan}}]G=G^{semi}
\end{equation}

\subsubsection{Hanzi-Lattices}

In \cite{Wang-Yao-Su-Wanjia-Zhang-2021-IMCEC}, the authors point out that graphic lattices not only have their important applications in information security and mathematics, but also have their potential in other fields. As known, there are 250 ``Pianpang'' and ``Bushou'' in Chinese characters (Hanzis) according to ``Chinese Character GB18030-2000'' containing 27484 Chinese characters, where ``Pianpang'' is the right and left part of a Chinese character, and ``Bushou'' can arrange Chinese characters; we have 299 ``Dutizi''; and there are 16 punctuation being the marks to clarify meaning by indicating separation of words into sentences and clauses and phrases. Let $H_{an}(1),H_{an}(2),\dots ,H_{an}(565)$ indicate all of ``Pianpang'', ``Bushou'', ``Dutizi'' and punctuation in Chinese characters, so we call $\textbf{\textrm{H}}_{an}=(H_{an}(1),H_{an}(2),\dots ,H_{an}(565))$ a \emph{Hanzi-Pianpang base}. Thereby, we can build up a \emph{Hanzi-literary lattice} $\textbf{\textrm{L}}(Z^0\textbf{\textrm{H}}_{an})$ of various Chinese writings by
\begin{equation}\label{eqa:GB18030-2000-Hanzi-lattice}
\textbf{\textrm{L}}(Z^0\textbf{\textrm{H}}_{an})=\left \{C_{para}=\bigcup ^{565}_{k=1}x_kH_{an}(k):x_k\in Z^0,H_{an}(k)\in \textbf{\textrm{H}}_{an}\right \}
\end{equation}
with $\sum ^{565}_{k=1}x_k\geq 1$. Let $C_{para}=\sum ^{565}_{k=1}x_kH_{an}(k)$ in $\textbf{\textrm{L}}(Z^0\textbf{\textrm{H}}_{an})$. We can confirm the following facts:

\begin{problem}\label{qeu:444444}
We need to consider the following problems:
\begin{asparaenum}[\textrm{\textbf{Prob}}-1. ]
\item \textbf{List} all Chinese paragraphs in $\textbf{\textrm{L}}(Z^0\textbf{\textrm{H}}_{an})$ with $\sum ^{565}_{k=1}x_k\leq M$ for a fixed integer $M\in Z^0\setminus \{0\}$.
\item \textbf{Judge} whether each $C_{para}$ defined in Eq.(\ref{eqa:GB18030-2000-Hanzi-lattice}) is meaningful or meaningless in Chinese.
\end{asparaenum}

Considering the above \textrm{Prob}-1 and \textrm{Prob}-2, we can build another \emph{Hanzi-literary lattice}
\begin{equation}\label{eqa:Hanzi-Pianpang-123}
\textbf{\textrm{L}}(Z^0\textbf{\textrm{C}}_{cha})=\left \{\bigcup ^{27500}_{k=1}y_kC_{cha}(k):y_k\in Z^0,C_{cha}(k)\in \textbf{\textrm{C}}_{cha}\right \}
\end{equation}
based on a \emph{Hanzi-graph base} $\textbf{\textrm{C}}_{cha}=(C_{cha}(1),C_{cha}(2),\dots ,C_{cha}(27500))$ and $\sum ^{27500}_{k=17}y_k\geq 1$, where each $C_{cha}(i)$ with $i\in [1,16]$ is a punctuation, and each $C_{cha}(j)$ for $j\in [17,27500]$ is a Chinese character in Chinese Character GB18030-2000.
\end{problem}

\begin{thm}\label{thm:Hanzi-literary-lattice}
\cite{Wang-Yao-Su-Wanjia-Zhang-2021-IMCEC} (1) Chinese paragraphs $C_{para}$ defined in Eq.(\ref{eqa:GB18030-2000-Hanzi-lattice}) can be used in text-based passwords to make more complicated \emph{public-key graphs} and \emph{private-key graphs} for increasing the cost of decryption and attackers.

(2) Each Chinese writing limited in Chinese Character GB18030-2000 has been contained in the Hanzi-literary lattices $\textbf{\textrm{L}}(Z^0\textbf{\textrm{H}}_{an})$ defined in Eq.(\ref{eqa:GB18030-2000-Hanzi-lattice}) and $\textbf{\textrm{L}}(Z^0\textbf{\textrm{C}}_{cha})$ defined in Eq.(\ref{eqa:Hanzi-Pianpang-123}), such as poems, novels, essay, proses, reports, news \emph{et al.}
\end{thm}

\begin{defn}\label{defn:flawed-odd-graceful-labeling}
\cite{Yao-Mu-Sun-Sun-Zhang-Wang-Su-Zhang-Yang-Zhao-Wang-Ma-Yao-Yang-Xie2019} Let $H=E^*+G$ be a connected graph, where $E^*$ is a non-empty set of edges and $G=\bigcup^m_{i=1}G_i$ is a disconnected graph, where $G_1,G_2,\dots, G_m$ are connected graphs and vertex-disjoint from each other. If $H$ admits a (set-ordered) graceful labeling (resp. a (set-ordered) odd-graceful labeling) $f$, then we call $f$ a \emph{flawed (set-ordered) graceful labeling} (resp. a \emph{flawed (set-ordered) odd-graceful labeling}) of $G$.\qqed
\end{defn}

\begin{defn}\label{defn:flawed-labeling}
\cite{Yao-Mu-Sun-Sun-Zhang-Wang-Su-Zhang-Yang-Zhao-Wang-Ma-Yao-Yang-Xie2019} Suppose that $G=\bigcup^m_{i=1}G_i$ is a disconnected graph, where $G_1,G_2,\dots $, $G_m$ are connected graphs and vertex-disjoint from each other. We have a connected graph $G+E^*$ obtained by adding the edges of an edge set $E^*$ to $G$. If $G+E^*$ admits a $W$-constraint coloring $f:S\subseteq V(G)\cup E(G)\cup E^*\rightarrow [a,b]$, then we say that the disconnected graph $G$ admits a \emph{flawed $W$-constraint coloring} $f: S\setminus E^* \rightarrow [a,b]$.\qqed
\end{defn}

\begin{thm}\label{thm:colosed-flawed-graceful-Pianpang}
\cite{Wang-Yao-Su-Wanjia-Zhang-2021-IMCEC} If each Pianpang $H\,'_{an}(k)$ of a Hanzi-Pianpang base $\textbf{\textrm{H}}\,'_{an}=(H\,'_{an}(1),H\,'_{an}(2),\dots $, $H\,'_{an}(m))$ admits a set-ordered graceful labeling for $k\in [1,m]$, then each disconnected graph of the lattice
\begin{equation}\label{eqa:555555}
\textbf{\textrm{L}}(Z^0\textbf{\textrm{H}}\,'_{an})=\left \{\bigcup ^{m}_{k=1}x_kH\,'_{an}(k):x_k\in Z^0,H\,'_{an}(k)\in \textbf{\textrm{H}}\,'_{an}\right \}
\end{equation} admits a \emph{flawed graceful labeling}.
\end{thm}

\subsubsection{Maximal planar graphs and Hanzi-based strings}

Using maximal planar graphs and Hanzi-graphs of topological coding to make asymmetric topology ciphers is as follows:

(i) For a semi-maximal planar graph $G$ as a \emph{public-key graph}, its own adjacent matrix $A(G)$ is unique under the matrix similarity operation, that is, if $B(G)$ is an adjacent matrix of $G$, then there is a matrix $P$ such that $B(G)=PA(G)P^{-1}$; and for another semi-maximal planar graph $H$ as a \emph{private-key graph}, its own adjacent matrix $A(H)$ is unique under the matrix similarity operation in graph theory. Thus, these two semi-maximal planar graphs $G$ and $H$ can be distinguished in the computer.

(ii) There is a topological signature authentication $J=G[\overline{\ominus}^{cyc}_k]H$, since there is a common cycle $C_k$ in two semi-maximal planar graphs $G$ and $H$.

(iii) Two Topcode-matrices $T_{code}(G,f)$ and $T_{code}(H,g)$ of two semi-maximal planar graphs $G$ and $H$ may correspond to other two or more graphs. In other words, these two graphs $G$ and $H$ are not easy to be determined by the Topcode-matrices, since it will meet the Subgraph Isomorphic Problem. Fig.\ref{fig:maximal-planar-hanzi} shows us a topological signature authentication $J=G[\overline{\ominus}^{cyc}_6]H$ that is a maximal planar graph admitting a coloring $F$ induced by two colorings $f$ and $g$.

(iv) We vertex-split the semi-maximal planar graph $H$ (as a \emph{private-key}) into mutually edge-disjoint graphs $H^*,H_1,H_2,\dots ,H_k$, such that each graph $H_i$ for $i\in [1,k]$ is a Hanzi-graph, and $H_1H_2\dots H_k$ is a Hanzi-based string, but $H^*$ is not a Hanzi-graph, sometimes, or $H^*$ does not exist; refer to Fig.\ref{fig:conincide-hanzi-graphs}.

(v) Recolor the topological signature authentication $J=G[\overline{\ominus}^{cyc}_k]H$ by a new proper 4-coloring defined in Definition \ref{defn:planar-graphs-color-sets-4-operations}, and induces number-based strings as desired.

\begin{figure}[h]
\centering
\includegraphics[width=16.4cm]{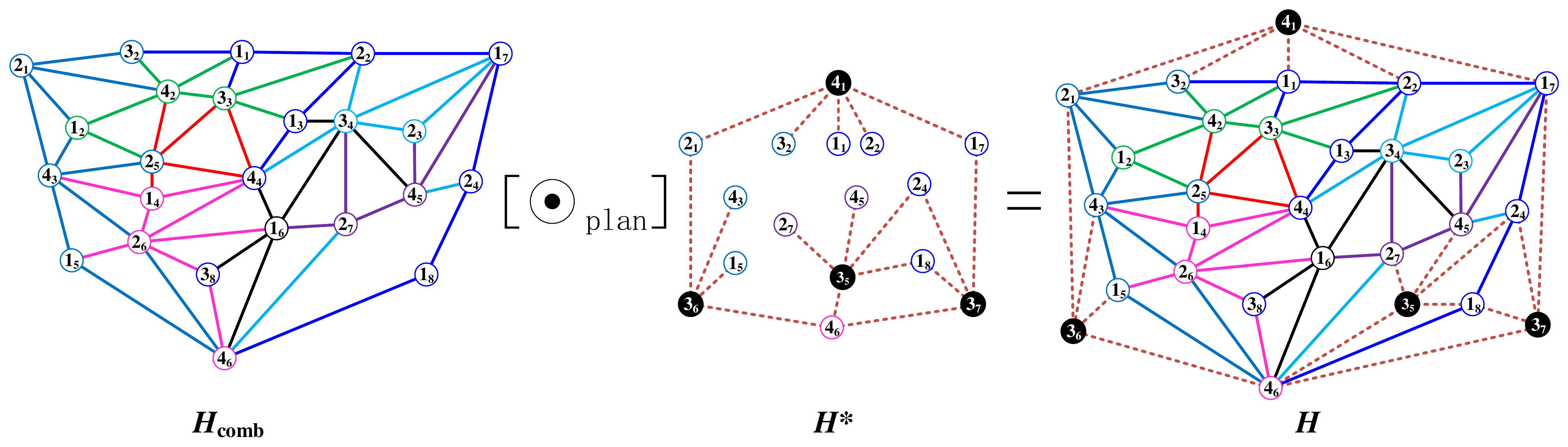}\\
\caption{\label{fig:conincide-hanzi-graphs}{\small A semi-maximal planar graph $H=H_{\textrm{comb}}[\odot_{\textrm{plan}}]H^*$, where $H_{\textrm{comb}}$ is obtained by doing the vertex-coinciding operation to the Hanzi-based string $H_{an}$ shown in Fig.\ref{fig:hundreds-people}.}}
\end{figure}

\begin{figure}[h]
\centering
\includegraphics[width=16.4cm]{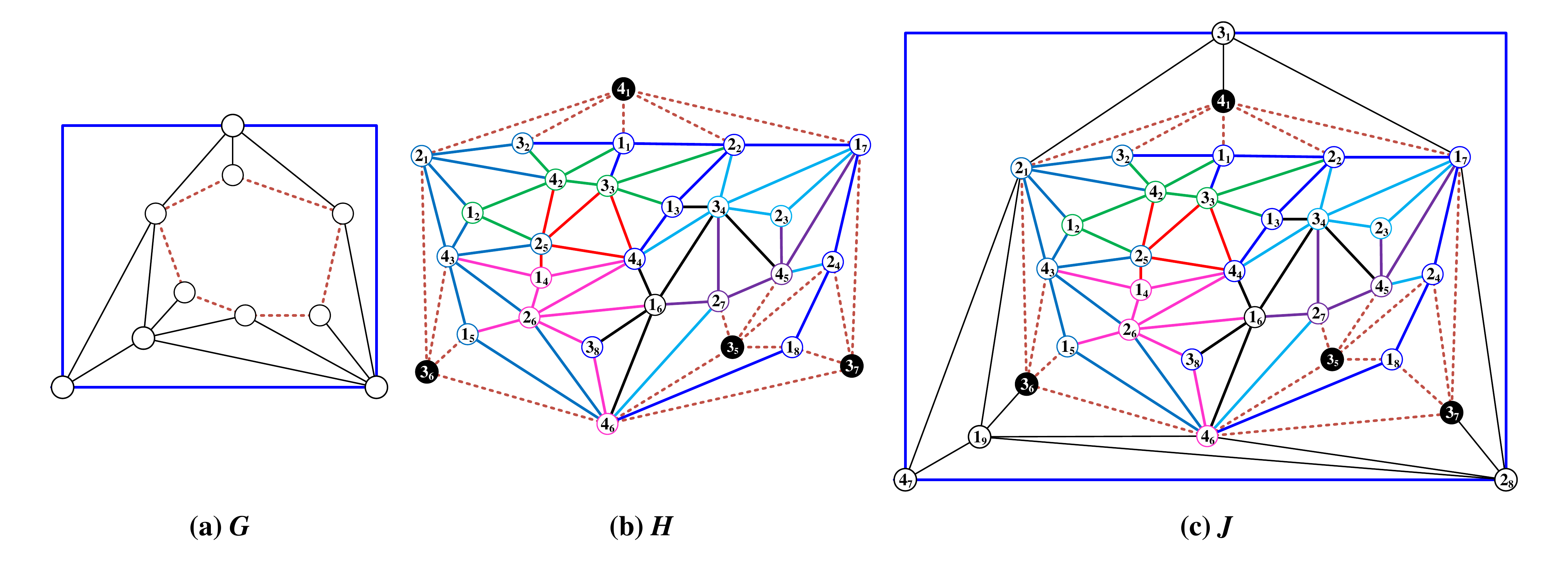}\\
\caption{\label{fig:maximal-planar-hanzi}{\small A topological signature authentication made by a maximal planar graph $J=G[\overline{\ominus}^{cyc}_6]H$, where $G$ is a \emph{public-key graph}, and the \emph{private-key graph} $H$ is shown in Fig.\ref{fig:conincide-hanzi-graphs}.}}
\end{figure}

\begin{problem}\label{qeu:444444}
Vertex-splitting a semi-maximal planar graph $H$ into a group of mutually edge-disjoint Hanzi-graphs $H_1,H_2,\dots ,H_k$ and a non-Hanzi-graph $H^*$ is not a slight job, and there are many Hanzi-based strings $H_{i_1}H_{i_2}\dots H_{i_k}$ obtained from the permutations of the mutually edge-disjoint Hanzi-graphs $H_1,H_2,\dots ,H_k$; see Fig.\ref{fig:Hanzi-combinators}. We may obtain other group of edge-disjoint Hanzi-graphs $G_1,G_2,\dots ,G_m$ after vertex-splitting the semi-maximal planar graph $H$.

Thereby, we get a set $H_{vspl} (H)$ of groups of mutually edge-disjoint Hanzi-graphs $H_{i,1},H_{i,2},\dots ,H_{i,m_i}$ after vertex-splitting the semi-maximal planar graph $H$, that is
\begin{equation}\label{eqa:groups-disjoint-Hanzi-graphs}
H_{vspl} (H)=\Big \{V^{split}_i=\{H_{i,1},H_{i,2},\dots ,H_{i,m_i}\}:i\in [1,n_{vspl}(H)]\Big \}
\end{equation} \textbf{Determine} the set $H_{vspl} (H)$ and the number $n_{vspl}(H)$.
\end{problem}

\subsubsection{The $[\odot_{\textrm{plan}}]$-Hanzi-graphic lattices}

The number-based string $a_jb_jc_jd_j$ of each Hanzi-graph $H_{a_jb_jc_jd_j}$ of a \emph{Hanzi-graph base}
$$\textbf{\textrm{H}}_{GB18}=(H_{a_1b_1c_1d_1},H_{a_2b_2c_2d_2},\dots ,H_{a_{27484}b_{27484}c_{27484}d_{27484}})
$$ is defined in \cite{GB2312-80}, such that $H_{a_ib_ic_id_i}\not \cong H_{a_jb_jc_jd_j}$ if $i\neq j$.

For a permutation $G_{i_1},G_{i_2},\dots ,G_{i_M}$ of mutually disjoint Hanzi-graphs $\lambda_jH_{a_jb_jc_jd_j}$ for $j\in [1,27484]$ and $\lambda_j\in Z^0$, where $M=\sum ^{27484}_{j=1}\lambda_j\geq 1$, we do the vertex-coinciding operation ``$[\odot_{\textrm{plan}}]$'' to the permutation $G_{i_1},G_{i_2},\dots ,G_{i_M}$ as: $T_{1}=G_{i_1}[\odot_{\textrm{plan}}]G_{i_2}$, where $T_1$ is a planar graph, and then we get the second planar graph $T_{2}=T_{1}[\odot_{\textrm{plan}}]G_{i_3}$, go on in this way, we get planar graphs
$$T_{n+1}=T_{n}[\odot_{\textrm{plan}}]G_{i_{n+2}},~n\in [2,27482]
$$ we write the last planar graph as
\begin{equation}\label{eqa:555555}
T_{27483}=T_{27482}[\odot_{\textrm{plan}}]G_{i_{27484}}=[\odot_{\textrm{plan}}]^{27484}_{j=1}\lambda_j H_{a_jb_jc_jd_j}
\end{equation}
We get a \emph{vertex-coincided Hanzi-graphic lattice}
\begin{equation}\label{eqa:graphic-lattice-GB2312-80}
\textbf{\textrm{L}}([\odot_{\textrm{plan}}]Z^0\textbf{\textrm{H}}_{GB18})=\Big \{[\odot_{\textrm{plan}}]^{27484}_{j=1}\lambda_j H_{a_jb_jc_jd_j}:\lambda_j\in Z^0,H_{a_jb_jc_jd_j}\in \textbf{\textrm{H}}_{GB18}\Big \}
\end{equation}
based on the \emph{Hanzi-graph base} $\textbf{\textrm{H}}_{GB18}$.

\begin{thm}\label{thm:666666}
$^*$ Since each graph $G$ of the Hanzi-graphic lattice $\textbf{\textrm{L}}([\odot_{\textrm{plan}}]Z^0\textbf{\textrm{H}}_{GB18})$ defined in Eq.(\ref{eqa:graphic-lattice-GB2312-80}) is a planar graph, then the planar graph $G$ admits a total indexed-coloring defined in Definition \ref{defn:planar-graphs-color-sets-4-operations} and can be decomposed into mutually edge-disjoint Hanzi-graphs.
\end{thm}

\section{Number-based strings towards application}

For encrypting the whole network and assigning passwords to each node of the network one time, we design several string groups, string lattices, graph homomorphisms based on the strings introduced in the previous sections.

\subsection{Multi-level multi-rank strings}

\subsubsection{Definition}

\begin{defn} \label{defn:m-level-more-rank-string-sets}
$^*$ An \emph{$m$-level $\{n_i\}^m_{i=1}$-rank string-set} contains its own elements being strings $s_{tri}(m)=A^{m}_1A^{m}_2\cdots A^{m}_{n_m}$, where
$A^{m}_{j}=A^{m-1}_{j,1}A^{m-1}_{j,2}\cdots A^{m-1}_{j,n_{m-1}}$ for $j\in [1,n_{m}]$, and
\begin{equation}\label{eqa:555555}
A^{m-s+1}_{j}=A^{m-s}_{j,1}A^{m-s}_{j,2}\cdots A^{m-s}_{j,n_{m-s}},~j\in [1,n_{m-s+1}],~s\in [1,m-1]
\end{equation}Finally, each $A^{2}_{j}=A^{1}_{j,1}A^{1}_{j,2}\cdots A^{1}_{j,n_{1}}$ is a $[0,9]$-string with $A^{1}_{j,t}\in [0,9]$ and $t\in [1,n_{1}]$ and $j\in [1,n_{2}]$.

If $n_i=n$ for $i\in [1,m]$, we call $s_{tri}(m)$ \emph{$m$-level uniformly $n$-rank string}.\qqed
\end{defn}

\begin{example}\label{exa:8888888888}
For illustrating Definition \ref{defn:m-level-more-rank-string-sets} we take $m=3$ with $n_3=4$, $n_2=3$, $n_1=7$, then $s_{tri}(3)=A^{3}_1A^{3}_2\cdots A^{3}_{4}$ having

(i) $A^{3}_{j}=A^{2}_{j,1}A^{2}_{j,2}A^{2}_{j,3}$ for $j\in [1,4]$;

(ii) $A^{2}_{j}=A^{1}_{j,1}A^{1}_{j,2}\cdots A^{1}_{j,7}$ with $A^{1}_{j,t}\in [0,9]$ and $t\in [1,7]$ and $j\in [1,3]$;

(iii) $A^{2}_{j,r}\in B_{asic}(7)=\{6174314,~1123061,~8142857\}$ with $r\in [1,3]$ and $j\in [1,4]$.

We get

$A^{3}_{1}=A^{2}_{1,1}A^{2}_{1,2}A^{2}_{1,3}$ with $A^{2}_{1,1}=6174314$, $A^{2}_{1,2}=1123061$, $A^{2}_{1,3}=8142857$

$A^{3}_{2}=A^{2}_{2,1}A^{2}_{2,2}A^{2}_{2,3}$ with $A^{2}_{2,1}=1123061$, $A^{2}_{2,2}=8142857$, $A^{2}_{2,3}=6174314$,

$A^{3}_{3}=A^{2}_{3,1}A^{2}_{3,2}A^{2}_{3,3}$ with $A^{2}_{3,1}=8142857$, $A^{2}_{3,2}=1123061$, $A^{2}_{3,3}=6174314$,

$A^{3}_{4}=A^{2}_{4,1}A^{2}_{4,2}A^{2}_{4,3}$ with $A^{2}_{4,1}=6174314$, $A^{2}_{4,2}=8142857$, $A^{2}_{4,3}=1123061$,

Finally, we get a $3$-level $\{7,3,4\}$-rank string-set
$${
\begin{split}
s_{tri}(3)=&A^{3}_1A^{3}_2A^{3}_{3}A^{3}_{4}=A^{2}_{1,1}A^{2}_{1,2}A^{2}_{1,3}A^{2}_{2,1}A^{2}_{2,2}A^{2}_{2,3}A^{2}_{3,1}A^{2}_{3,2}A^{2}_{3,3}A^{2}_{4,1}A^{2}_{4,2}A^{2}_{4,3}\\
=&6174314~1123061~8142857~~1123061~8142857~6174314\\
&8142857~1123061~6174314~~6174314~8142857~1123061
\end{split}}
$$ with $84$ bytes.

The set $B_{asic}(7)$ can contain more $[0,9]$-strings with $7$ numbers belonging to $[0,9]$.\qqed
\end{example}

\subsubsection{A new representation of Topcode-matrices}

Suppose that a $(p,q)$-graph $G$ admits a total string-coloring $h:V(G)\cup E(G)\rightarrow \textbf{\textrm{S}}_{tring}(n)$, such that
\begin{equation}\label{eqa:555555}
h(u_k)=r_{k,1}r_{k,2}\cdots r_{k,n},~h(v_k)=s_{k,1}s_{k,2}\cdots s_{k,n},~h(u_kv_k)=t_{k,1}t_{k,2}\cdots t_{k,n}
\end{equation} for each edge $u_kv_k\in E(G)=\{u_kv_k:k\in [1,q]\}$. Then the $(p,q)$-graph $G$ has its own Topcode-matrix
\begin{equation}\label{eqa:5555555555555555}
\centering
{
\begin{split}
T_{code}(G,h)= \left(
\begin{array}{ccccc}
h(u_{1}) & h(u_{2}) & \cdots & h(u_q)\\
h(u_{1}v_{1}) & h(u_{2}v_{2}) & \cdots & h(u_qv_q)\\
h(v_{1}) & h(v_{2}) & \cdots & h(v_q)
\end{array}
\right)=\left(
\begin{array}{cccccccccccccc}
X_h\\
E_h\\
Y_h
\end{array}
\right)=(X_h,E_h,Y_h)^{T}
\end{split}}
\end{equation} where $u_kv_k\in E(G)$ for $k\in [1,q]$.

We rewrite each group of colors $h(u_k),h(u_kv_k)$ and $h(v_k)$ into three vectors, immediately, we get the Topcode-matrix of each graph $G_{u_kv_k}$ as follows:
\begin{equation}\label{eqa:edge-vector-graphs-matices}
\centering
{
\begin{split}
T_{code}(G_{u_kv_k},h)= \left(
\begin{array}{ccccc}
r_{k,1} &r_{k,2} &\cdots &r_{k,n}\\
t_{k,1} &t_{k,2} &\cdots &t_{k,n}\\
s_{k,1} &s_{k,2} &\cdots &s_{k,n}
\end{array}
\right)=\left(
\begin{array}{cccccccccccccc}
h(u_k)\\
h(u_kv_k)\\
h(v_k)
\end{array}
\right)=(h(u_k),h(u_kv_k),h(v_k))^{T}
\end{split}}
\end{equation}
with edge number $|E(G_{u_kv_k})|=n$ for $k\in [1,q]$. We give the Topcode-matrix $T_{code}(G,h)$ a new representation as:
\begin{equation}\label{eqa:topcode-matrix-another-expression}
T_{code}(G,h):=(T_{code}(G_{u_1v_1},h),T_{code}(G_{u_2v_2},h),\dots ,T_{code}(G_{u_qv_q},h))
\end{equation} which differs from that defined in Definition \ref{defn:graphic-topcode-matrix}.

\begin{example}\label{exa:8888888888}
In Fig.\ref{fig:matrix-in-matrix-11}, the multiple-edge graph $J_1$ is made by the multiple-edge graphs $T_1,T_2$, $\dots $, $T_{10}$ shown in Fig.\ref{fig:matrix-in-matrix-00}, where each multiple-edge graph $T_i$ is obtained the vertex color and edge color of an edge of the connected $(10,11)$-graph $H_1$ shown in Fig.\ref{fig:333string-vector-set-coloring} (a). So, the Topcode-matrix of $H_1$ is as follows
$$T_{code}(H_1,F_{stri}):=(T_{code}(T_1,f_1), T_{code}(T_2,f_2),\dots ,T_{code}(T_{11},f_{11}))
$$ where each coloring $f_i$ of $T_i$ for $i\in [1,11]$ is shown in Fig.\ref{fig:matrix-in-matrix-00}.\qqed
\end{example}

\begin{figure}[h]
\centering
\includegraphics[width=16.4cm]{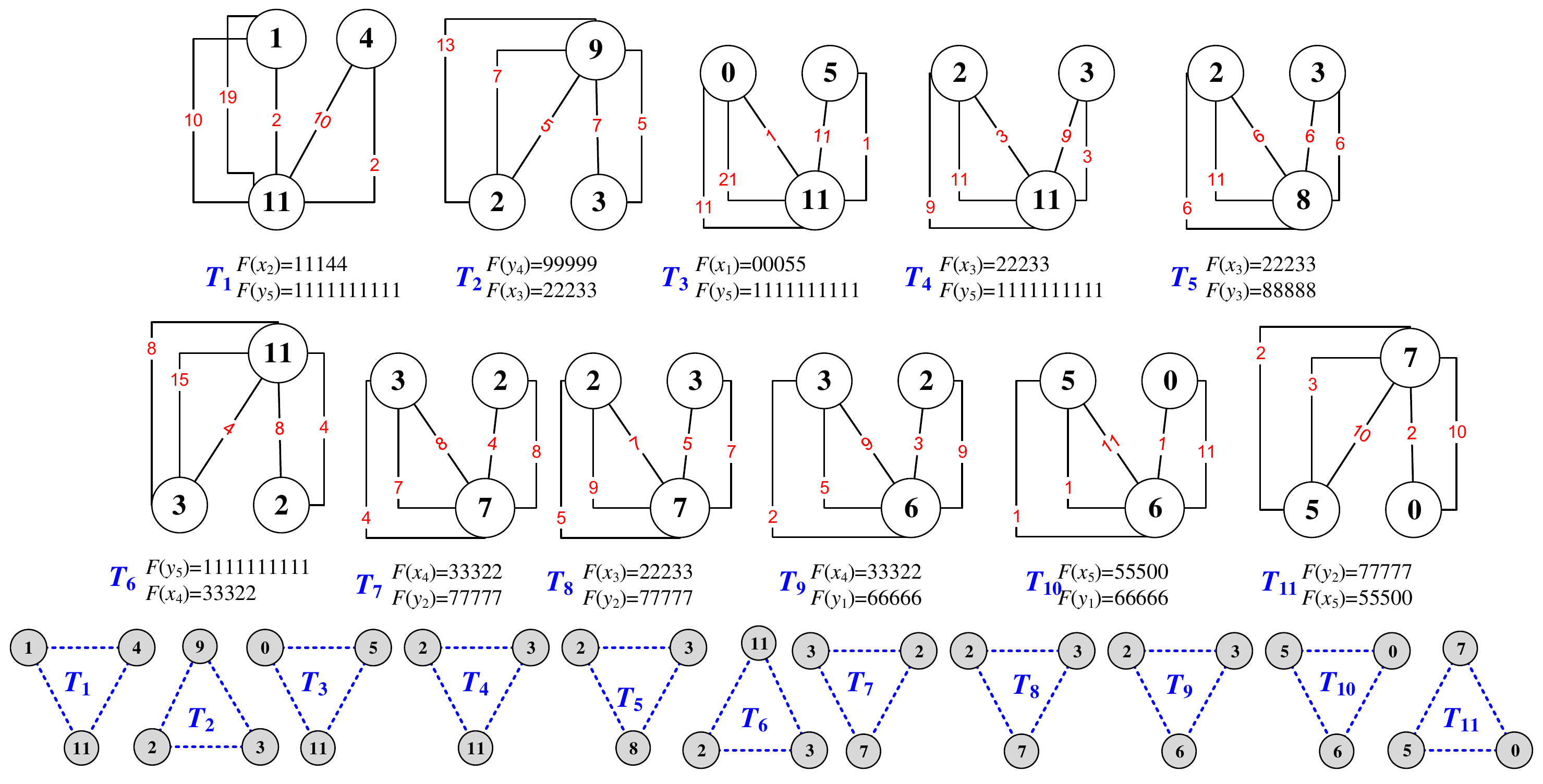}\\
\caption{\label{fig:matrix-in-matrix-00}{\small The multiple-edge graphs $T_1,T_2,\dots ,T_{11}$ induced from the string-coloring $F_{stri}$ of the connected $(10,11)$-graph $H_1$ shown in Fig.\ref{fig:333string-vector-set-coloring} (a).}}
\end{figure}

\begin{problem}\label{qeu:444444}
$^*$ It is not hard to verify that three multiple-edge graphs $J_1$, $J_2$ and $J_3$ shown in Fig.\ref{fig:matrix-in-matrix-11} are not isomorphic from each other. All multiple-edge graphs made by the multiple-edge graphs $T_1,T_2,\dots ,T_{10}$ are collected into a set $G_{\textrm{mu-edge}}(H_1,F_{stri})$, we want to \textbf{characterize} the graphs in the set $G_{\textrm{mu-edge}}(H_1,F_{stri})$.

For a connected $(p,q)$-graph $Q$ admits a total string-coloring $\theta_{stri}:V(Q)\cup E(Q)\rightarrow \textbf{\textrm{S}}_{tring}(n)$, such that each edge $u_kv_k\in E(Q)$ holds
\begin{equation}\label{eqa:555555}
\theta_{stri}(u_k)=a_{k,1}a_{k,2}\cdots a_{k,n},~\theta_{stri}(v_k)=b_{k,1}b_{k,2}\cdots b_{k,n},~\theta_{stri}(u_kv_k)=c_{k,1}c_{k,2}\cdots c_{k,n}
\end{equation} which induce a total set-coloring $\varphi_{set}:V(Q)\cup E(Q)\rightarrow \textbf{\textrm{S}}_{set}(n)$
\begin{equation}\label{eqa:edge-vector-graphs-matices-problem}
{
\begin{split}
&\varphi_{set}(u_k)=(a_{k,1},a_{k,2},\dots ,a_{k,n}),~\varphi_{set}(v_k)=(b_{k,1},b_{k,2},\dots ,b_{k,n}),\\
&\varphi_{set}(u_kv_k)=(c_{k,1},c_{k,2},\dots ,c_{k,n})
\end{split}}
\end{equation} Eq.(\ref{eqa:edge-vector-graphs-matices-problem}) enables us to get a Topcode-matrix $T_{code}(Q_{u_kv_k},\varphi_{set})$ defined in Eq.(\ref{eqa:edge-vector-graphs-matices}), where each $Q_{u_kv_k}$ for $k\in [1,q]$ is a graph or a multiple-edge graph, such that
$$T_{code}(Q,\theta_{stri}):=(T_{code}(Q_{u_1v_1},\varphi_{set}), T_{code}(Q_{u_2v_2},\varphi_{set}),\dots ,T_{code}(Q_{u_qv_q},\varphi_{set}))
$$ We get a set $G_{\textrm{mu-edge}}(Q,\theta_{stri})$ of (multiple-edge) graphs made by these (multiple-edge) graphs $Q_{u_1v_1}, Q_{u_2v_2},\dots , Q_{u_qv_q}$, and moreover we have a set
$$
M_{\textrm{u-edge}}(Q)=\{G_{\textrm{mu-edge}}(Q,\theta_{stri}):\textrm{ each total string-coloring }\theta_{stri}\}
$$ \textbf{Characterize} the graphs in the sets $G_{\textrm{mu-edge}}(Q,\theta_{stri})$ and $M_{\textrm{u-edge}}(Q)$.
\end{problem}

\begin{figure}[h]
\centering
\includegraphics[width=16.4cm]{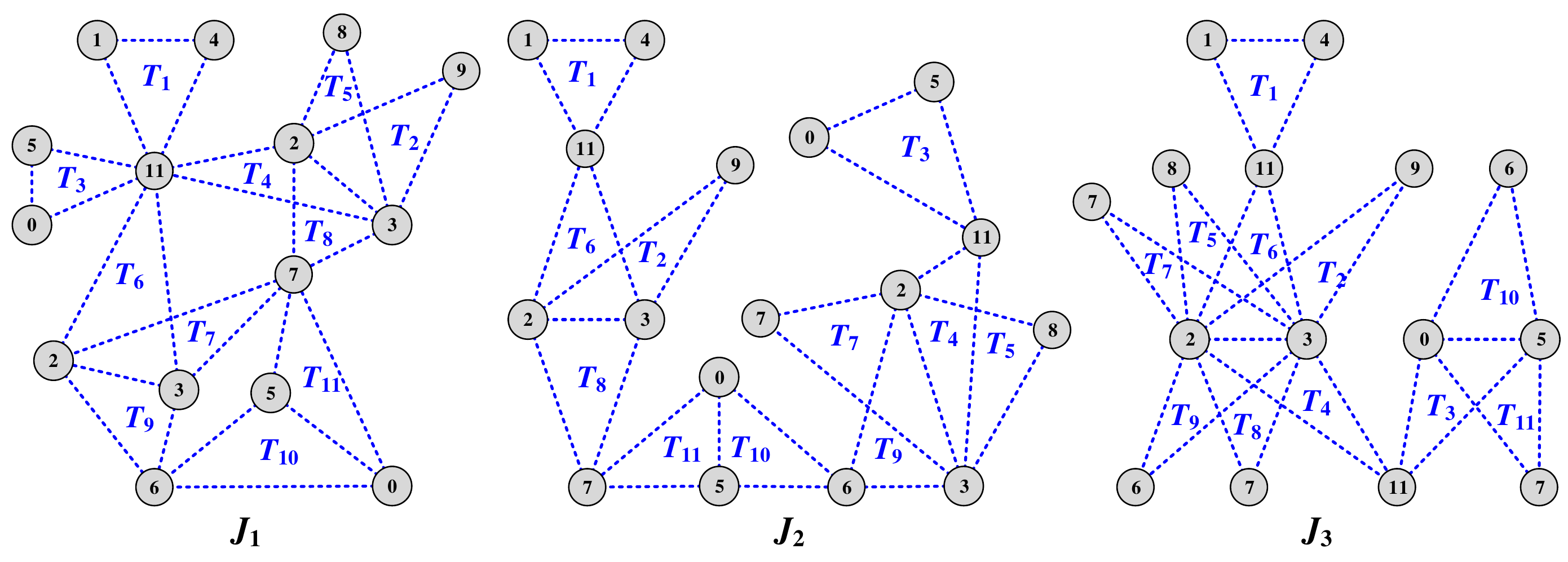}\\
\caption{\label{fig:matrix-in-matrix-11}{\small The multiple-edge graph $J_1$ is obtained from the connected $(10,11)$-graph $H_1$ shown in Fig.\ref{fig:333string-vector-set-coloring} (a). The multiple-edge graphs $J_1$, $J_2$ and $J_3$ are not isomorphic from each other.}}
\end{figure}

\subsection{Graph-colorings based on graphic group}

\subsubsection{Infinite graphic groups}

\textbf{INFINITEGraphic-group Algorithm} \cite{Yao-Wang-2106-15254v1}. Let $f: V(G)\cup E(G)\rightarrow [1,M]$ be a $W$-constraint proper total coloring of a graph $G$ such that two color sets $f(V(G))=\{f(x):x\in V(G)\}$ and $f(E(G))=\{f(uv):uv\in E(G)\}$ hold a collection of constraints true. We define a $W$-constraint proper total coloring $g_{s,k}$ by setting
$$g_{s,k}(x)=f(x)+s~(\bmod~p), x\in V(G)$$ and
$$g_{s,k}(uv)=f(uv)+k~(\bmod~q),~uv\in E(G)$$ Let $F_f(G)$ be the set of graphs $G_{s,k}$ admitting $W$-constraint proper total colorings $g_{s,k}$ defined above, and each graph $G_{s,k}\cong G$ in topological structure. We define an additive operation ``$[\oplus \ominus]$'' on the elements of $F_f(G)$ in the following way: Take arbitrarily an element $G_{a,b}\in F_f(G)$ as \emph{zero}, and
\begin{equation}\label{eqa:555555}
G_{s,k}[\oplus \ominus_{a,b}] G_{i,j}:=G_{s,k}[\oplus] G_{i,j}[\ominus] G_{a,b}
\end{equation} is defined by the following computation
\begin{equation}\label{eqa:mixed-graphic-group}
[g_{s,k}(w)+g_{i,j}(w)-g_{a,b}(w)]~(\bmod~\varepsilon)=g_{\lambda,\mu}(w)
\end{equation}
for each element $w\in V(G)\cup E(G)$, where the index $\lambda=s+i-a~(\bmod~p)$ and the index $\mu=k+j-b~(\bmod~q)$. As $w=x\in V(G)$, the form (\ref{eqa:mixed-graphic-group}) is just equal to
\begin{equation}\label{eqa:mixed-graphic-group11}
[g_{s,k}(x)+g_{i,j}(x)-g_{a,b}(x)]~(\bmod~p)=g_{\lambda,\mu}(x)
 \end{equation}and as $w=uv\in E(G)$, the form (\ref{eqa:mixed-graphic-group}) is defined as follows:
 \begin{equation}\label{eqa:mixed-graphic-group22}
[g_{s,k}(uv)+g_{i,j}(uv)-g_{a,b}(uv)]~(\bmod~q)=g_{\lambda,\mu}(uv)
\end{equation}
Especially, as $s=i=a=\alpha$, we have $\bmod~\varepsilon=\bmod~q$ in Eq.(\ref{eqa:mixed-graphic-group}), and
\begin{equation}\label{eqa:mixed-graphic-group-edge}
[g_{\alpha,k}(uv)+g_{\alpha,j}(uv)-g_{\alpha,b}(uv)]~(\bmod~q)=g_{\alpha,\mu}(uv)
\end{equation} for $uv\in E(G)$; and when $k=j=b=\beta$, so $\bmod~\varepsilon=\bmod~p$ in Eq.(\ref{eqa:mixed-graphic-group}), we have
\begin{equation}\label{eqa:mixed-graphic-group-vertex}
[g_{s,\beta}(x)+g_{i,\beta}(x)-g_{a,\beta}(x)]~(\bmod~p)=g_{\lambda,\beta}(x),~x\in V(G)
\end{equation}

\begin{defn} \label{defn:111111}
\cite{Yao-Wang-2106-15254v1} Sice the graph set $F_f(G)$ made by the INFINITEGraphic-group Algorithm holds:
\begin{asparaenum}[(1) ]
\item \emph{Zero.} Each graph $G_{a,b}\in F_f(G)$ can be determined as a \emph{preappointed zero} such that
$$G_{s,k}[\oplus \ominus_{a,b}] G_{i,j}:=G_{s,k}[\oplus] G_{a,b}[\ominus] G_{a,b}=G_{s,k}
$$

\item \emph{Uniqueness.} If
$$G_{s,k}[\oplus] G_{i,j}[\ominus] G_{a,b}=G_{c,d}\in F_f(G)\textrm{ and }G_{s,k}[\oplus] G_{i,j}[\ominus] G_{a,b}=G_{r,t}\in F_f(G)$$ then the index $c=s+i-a~(\bmod~p)=r$ and the index $d=k+j-b~(\bmod~q)=t$ under any \emph{preappointed zero} $G_{a,b}$.

\item \emph{Inverse.} Each graph $G_{s,k}\in F_f(G)$ has its own \emph{inverse} $G_{s',k'}\in F_f(G)$ such that
$$G_{s,k}[\oplus \ominus_{a,b}] G_{s',k'}:=G_{s,k}[\oplus] G_{s',k'}[\ominus] G_{a,b}=G_{a,b}
$$ determined by $[g_{s,k}(w)+g_{i,j}(w)]~(\bmod~\varepsilon)=2g_{a,b}(w)$ for each element $w\in V(G)\cup E(G)$.
\item \emph{Associative law.} Each triple $G_{s,k},G_{i,j},G_{c,d}\in F_f(G)$ holds
$$G_{s,k}[\oplus] \big ([G_{i,j}[\oplus] G_{c,d}[\ominus] G_{a,b}]\big )[\ominus] G_{a,b}=\big ([G_{s,k}[\oplus] G_{i,j}[\ominus] G_{a,b}]\big )[\oplus] G_{c,d}[\ominus] G_{a,b}
$$ true, that is
$$\big (G_{s,k}[\oplus \ominus_{a,b}] G_{i,j}\big )[\oplus \ominus_{a,b}]G_{c,d}=G_{s,k}[\oplus \ominus_{a,b}] \big (G_{i,j}[\oplus \ominus_{a,b}]G_{c,d}\big )$$
\item \emph{Commutative law.} $G_{s,k}[\oplus] G_{i,j}[\ominus] G_{a,b}=G_{i,j}[\oplus] G_{s,k}[\ominus] G_{a,b}$, also
$$G_{s,k}[\oplus \ominus_{a,b}] G_{i,j}= G_{i,j}[\oplus \ominus_{a,b}]G_{s,k}$$
\end{asparaenum}

Thereby, we call $F_f(G)=\{G_{s,k}:s\in [0,p],k\in [0,q]\}$ an \emph{every-zero graphic group} based on the additive operation ``$[\oplus \ominus]$'' defined in Eq.(\ref{eqa:mixed-graphic-group}), and write this group by $\textbf{\textrm{G}}=\{F_f(G);\oplus \ominus\}$.

There are $pq$ graphs in the every-zero graphic group $\textbf{\textrm{G}}$. There are two particular \emph{every-zero graphic subgroups} $\{F_{v}(G);\oplus \ominus\}\subset \textbf{\textrm{G}}$ and $\{F_{e}(G);\oplus \ominus\}\subset \textbf{\textrm{G}}$, where $F_{v}(G)=\{G_{s,0}:s\in [0,p]\}$ and $F_{e}(G)=\{G_{0,k}:k\in [0,q]\}$. In fact, $\textbf{\textrm{G}}$ contains at least $(p+q)$ different every-zero graphic subgroups.\qqed
\end{defn}

\begin{defn}\label{defn:every-zero-graphic-group-homomorphism}
\cite{Bing-Yao-Hongyu-Wang-graph-homomorphisms-2020} For two every-zero graphic groups $\{F_f(G);\oplus \ominus\}$ based on a graph set $F_f(G)=\{G_i\}^m_1$ and $\{F_h(H);\oplus \ominus\}$ based on a graph set $F_h(H)=\{H_i\}^m_1$, suppose that there are graph homomorphisms $G_i\rightarrow H_i$ defined by $\theta_i:V(G_i)\rightarrow V(H_i)$ with $i\in [1,m]$. We define $\theta=\bigcup^m_{i=1}\theta_i$, and have an \emph{every-zero graphic group homomorphism}
$$\{F_f(G);\oplus \ominus\}\rightarrow \{F_h(H);\oplus \ominus\}
$$ from a graph set $F_f(G)$ to another graph set $F_h(H)$.\qqed
\end{defn}

\textbf{Infinite graphic-sequence groups.} Infinite graphic-sequence groups have been introduced in \cite{yao-sun-su-wang-matching-groups-zhao-2020}. Suppose that a connected $(p,q)$-graph $G$ admits a $W$-constraint total coloring $f$, we define $W$-constraint total colorings $f_{s,k}$ by setting $f_{s,k}(x)=f(x)+s$ for every vertex $x\in V(G)$, and $f_{s,k}(uv)=f(uv)+k$ for each edge $uv\in E(G)$ as two integers $s,k$ belong to the set $Z$ of integers. So, we have each connected $(p,q)$-graph $G_{s,k}\cong G$ admits a $W$-constraint total coloring $f_{s,k}$ defined above, immediately, we get an infinite graphic-sequence $\{\{G_{s,k}\}^{+\infty}_{-\infty}\}^{+\infty}_{-\infty}$. We take a graph $G_{a,b}\in \{\{G_{s,k}\}^{+\infty}_{-\infty}\}^{+\infty}_{-\infty}$ as a \emph{preappointed zero} and any two $G_{s,k}$ and $G_{i,j}$ in $\{\{G_{s,k}\}^{+\infty}_{-\infty}\}^{+\infty}_{-\infty}$, and do the additive computation ``$G_{s,k}[\oplus \ominus_{a,b}] G_{i,j}:=G_{s,k}[\oplus] G_{i,j}[\ominus] G_{a,b}$'' in the following way: For each edge $uv\in E(G)$,
\begin{equation}\label{eqa:edge-graphic-group}
[f_{s,k}(uv)+f_{i,j}(uv)-f_{a,b}(uv)]~(\bmod~ q_W)=f_{\lambda,\mu}(uv).
\end{equation} with the index $\mu=k+j-b~(\bmod~ q_W)$; and for each vertex $x\in V(G)$,
\begin{equation}\label{eqa:vertex-graphic-group}
[f_{s,k}(x)+f_{i,j}(x)-f_{a,b}(x)]~(\bmod~ p_W)=f_{\lambda,\mu}(x)
\end{equation} with the index $\lambda=s+i-a~(\bmod~ p_W)$.

Here, $p_W=|V(G)|$ and $q_W=|E(G)|$ if the $W$-constraint total coloring $f$ is a gracefully total coloring; and $p_W=q_W=2|E(G)|$ if the $W$-constraint total coloring $f$ is an odd-gracefully total coloring (see examples shown in Fig.\ref{fig:odd-graceful-group-vertex} and Fig.\ref{fig:odd-graceful-group-edge}).

\begin{figure}[h]
\centering
\includegraphics[width=16cm]{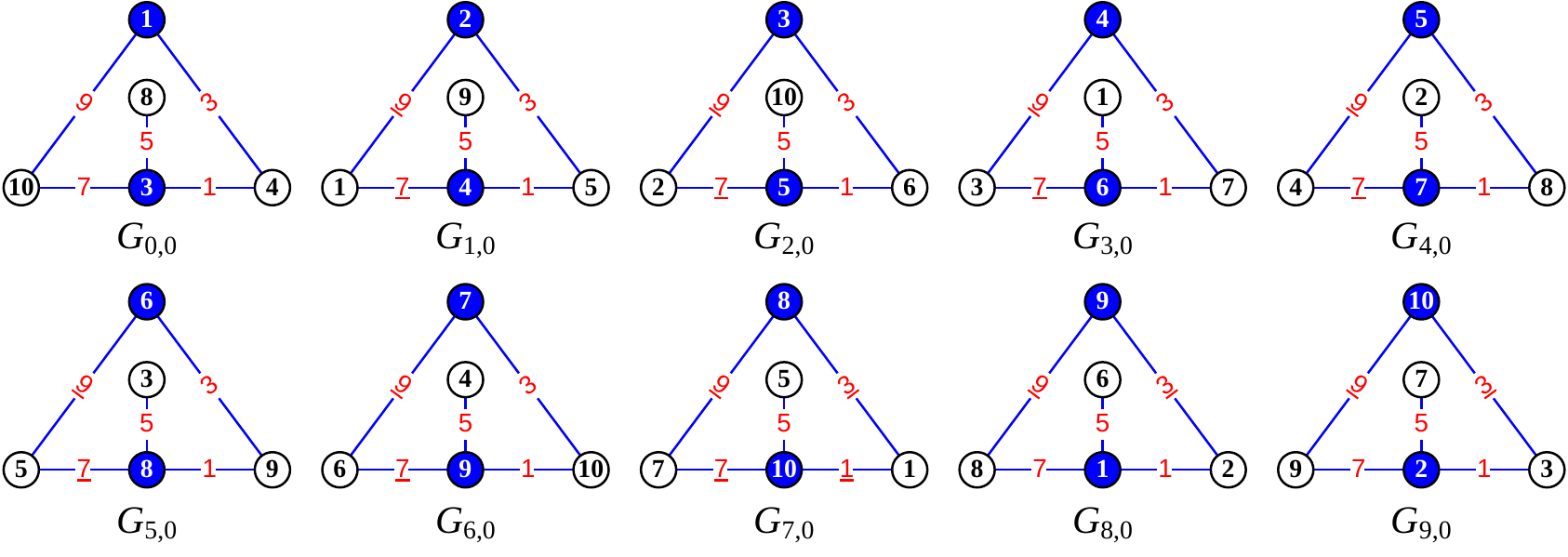}\\
\caption{\label{fig:odd-graceful-group-vertex} {\small An example for $p_W=2|E(G)|=10$ and holding Eq.(\ref{eqa:vertex-graphic-group}), cited from \cite{Yao-Sun-Su-Wang-Zhao-ICIBA-2020}.}}
\end{figure}

\begin{figure}[h]
\centering
\includegraphics[width=16cm]{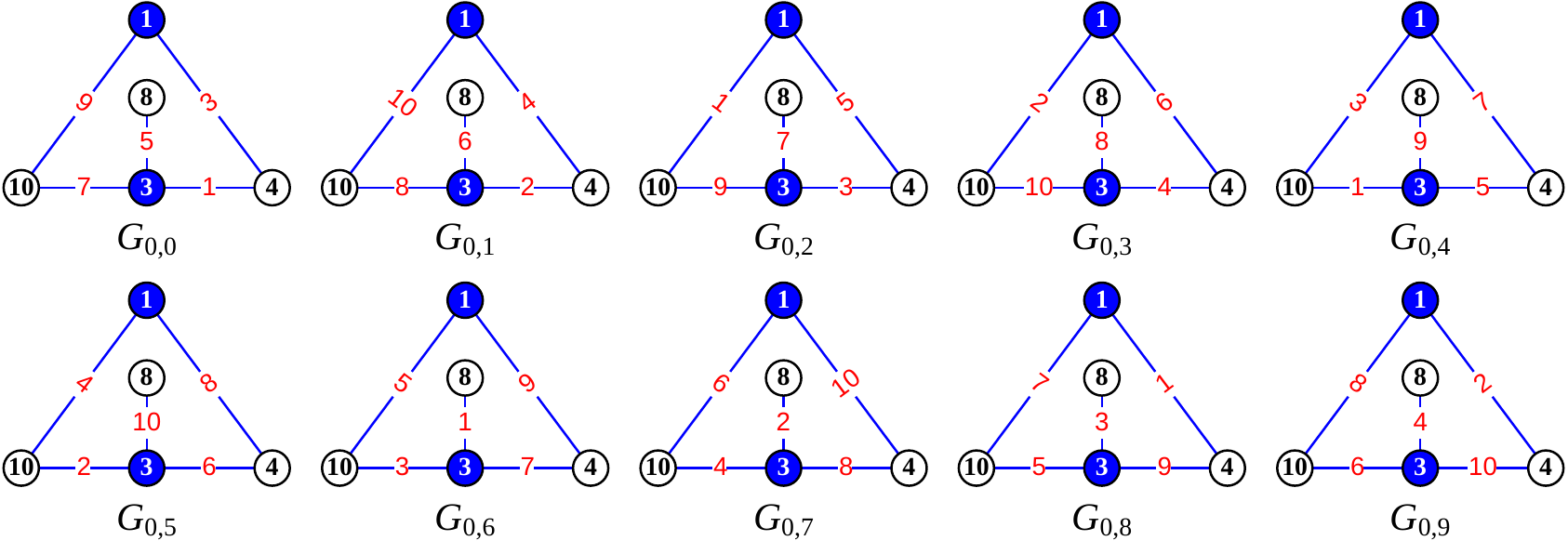}\\
\caption{\label{fig:odd-graceful-group-edge} {\small An example for $q_W=2|E(G)|=10$ and holding Eq.(\ref{eqa:edge-graphic-group}), cited from \cite{Yao-Sun-Su-Wang-Zhao-ICIBA-2020}.}}
\end{figure}

Especially, for an edge subsequence $G_{s,k}, G_{s,k+1},\dots ,G_{s,k+q_W}$, and a vertex subsequence $G_{s,k}$, $G_{s+1,k}$, $\dots $, $G_{s+p_W,k}$, we have two sets $F(\{G_{s,k+j}\}^{q_W}_{j=1};\oplus \ominus;(G,f))$ and $F(\{G_{s+i,k}\}^{p_W}_{i=1};\oplus \ominus;(G,f))$. By the operation ``$[\oplus \ominus]$'' defined in Eq.(\ref{eqa:edge-graphic-group}), we claim that $F(\{G_{s,k+j}\}^{q_W}_{j=1};\oplus \ominus;(G,f))$ is an \emph{every-zero edge-graphic group}, since there are the following facts:

(i) \emph{Zero}. Every graph $G_{s,k+j}$ of $F(\{G_{s,k+j}\}^q_{j=1};\oplus \ominus)$ is as a \emph{preappointed zero} such that
$$G_{s,k+r}[\oplus\ominus_{s,k+j}] G_{s,k+j}:=G_{s,k+r}[\oplus] G_{s,k+j}[\ominus]G_{s,k+j}=G_{s,k+r}$$ for any graph $G_{s,k+r}\in F(\{G_{s,k+j}\}^q_{j=1};\oplus \ominus)$.

(ii) \emph{Closure law}. For the index $r=i+j-j_0~(\bmod~ q_W)$,
$$G_{s,k+r}[\oplus\ominus_{s,k+j_0}] G_{s,k+j}:=G_{s,k+i}[\oplus] G_{s,k+j}[\ominus]G_{s,k+j_0}=G_{s,k+r}\in F(\{G_{s,k+j}\}^q_{j=1};\oplus \ominus)
$$ under any \emph{preappointed zero} $G_{s,k+j_0}$, and $G_{s,k+r}, G_{s,k+j}\in F(\{G_{s,k+j}\}^q_{j=1};\oplus \ominus)$.

(iii) \emph{Inverse.} For the index $i+j=2j_0~(\bmod~ q_W)$, we have
$$G_{s,k+i}[\oplus\ominus_{s,k+j_0}] G_{s,k+j}:=G_{s,k+i}[\oplus] G_{s,k+j}[\ominus]G_{s,k+j}=G_{s,k+j_0}
$$ under any \emph{preappointed zero} $G_{s,k+j_0}$, and $G_{s,k+i}, G_{s,k+j}\in F(\{G_{s,k+j}\}^q_{j=1};\oplus \ominus)$.

(iv) \emph{Associative law}. We have
$$G_{s,k+i}[\oplus\ominus_{s,k+j_0}]\big (G_{s,k+j}[\oplus\ominus_{s,k+j_0}] G_{s,k+r}\big )=\big (G_{s,k+i}[\oplus\ominus_{s,k+j_0}] G_{s,k+j}\big )[\oplus\ominus_{s,k+j_0}] G_{s,k+r}
$$ under any \emph{preappointed zero} $G_{s,k+j_0}$, and $G_{s,k+i}, G_{s,k+j}, G_{s,k+r}\in F(\{G_{s,k+j}\}^q_{j=1};\oplus \ominus)$.

(v) \emph{Commutative law}. For any \emph{preappointed zero} $G_{s,k+j_0}$, there is
$$G_{s,k+i}[\oplus\ominus_{s,k+j_0}] G_{s,k+j}=G_{s,k+j}[\oplus\ominus_{s,k+j_0}] G_{s,k+i}
$$ for $G_{s,k+i}, G_{s,k+j}\in F(\{G_{s,k+j}\}^q_{j=1};\oplus \ominus)$.

\vskip 0.4cm

Similarly, $F(\{G_{s+i,k}\}^{p_W}_{i=1};\oplus \ominus;(G,f))$ is an \emph{every-zero vertex-graphic group} by the operation ``$[\oplus \ominus]$'' defined in Eq.(\ref{eqa:vertex-graphic-group}).

As considering some graphs arbitrarily selected from the sequence $\{\{G_{s,k}\}^{+\infty}_{-\infty}\}^{+\infty}_{-\infty}$, we have
\begin{equation}\label{eqa:mixed-infinite-graphic-group}
[f_{s,k}(w)+f_{i,j}(w)-f_{a,b}(w)]~(\bmod~ p_W, q_W)=f_{\lambda,\mu}(w).
\end{equation} with the indices $\lambda=s+i-a~(\bmod~ p_W)$ and $\mu=k+j-b~(\bmod~ q_W)$ for each element $w\in V(G)\cup E(G)$. See examples shown in Fig.\ref{fig:infinite-graphic-group} and Fig.\ref{fig:odd-graceful-group-mixed}.

Thereby, we call the set $F(\{\{G_{s,k}\}^{+\infty}_{-\infty}\}^{+\infty}_{-\infty};\oplus \ominus;(G,f))$ an \emph{every-zero infinite graphic-sequence group} under the additive operation ``$[\oplus \ominus]$'' based on two modules $p_W$ and $q_W$ and a connected $(p,q)$-graph $G$ admitting a $W$-constraint total coloring, since it possesses the properties of Zero, Closure law, Inverse, Associative law and Commutative law.

\begin{figure}[h]
\centering
\includegraphics[width=16cm]{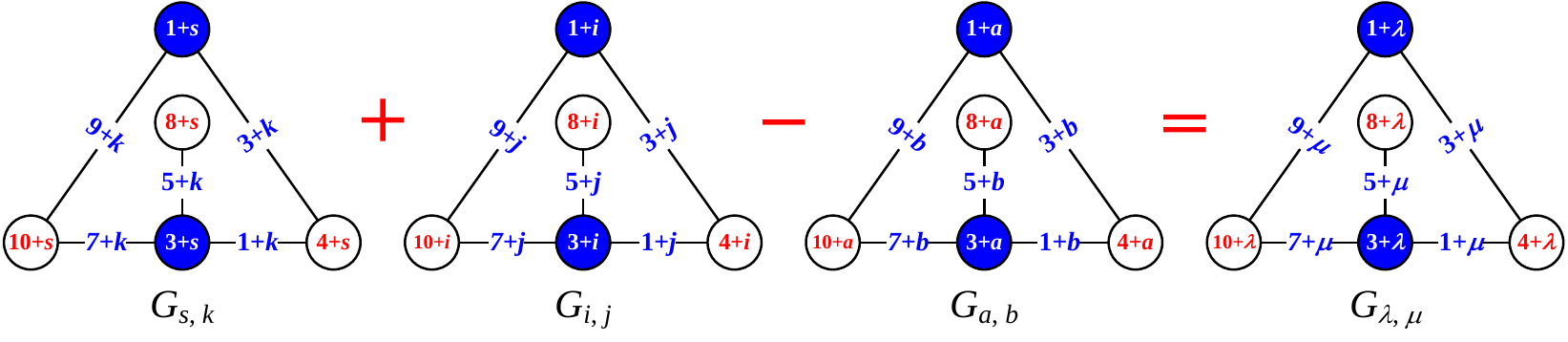}\\
\caption{\label{fig:infinite-graphic-group} {\small A graphic scheme for illustrating the formula Eq.(\ref{eqa:mixed-infinite-graphic-group}), cited from \cite{Yao-Sun-Su-Wang-Zhao-ICIBA-2020}.}}
\end{figure}

\begin{figure}[h]
\centering
\includegraphics[width=16cm]{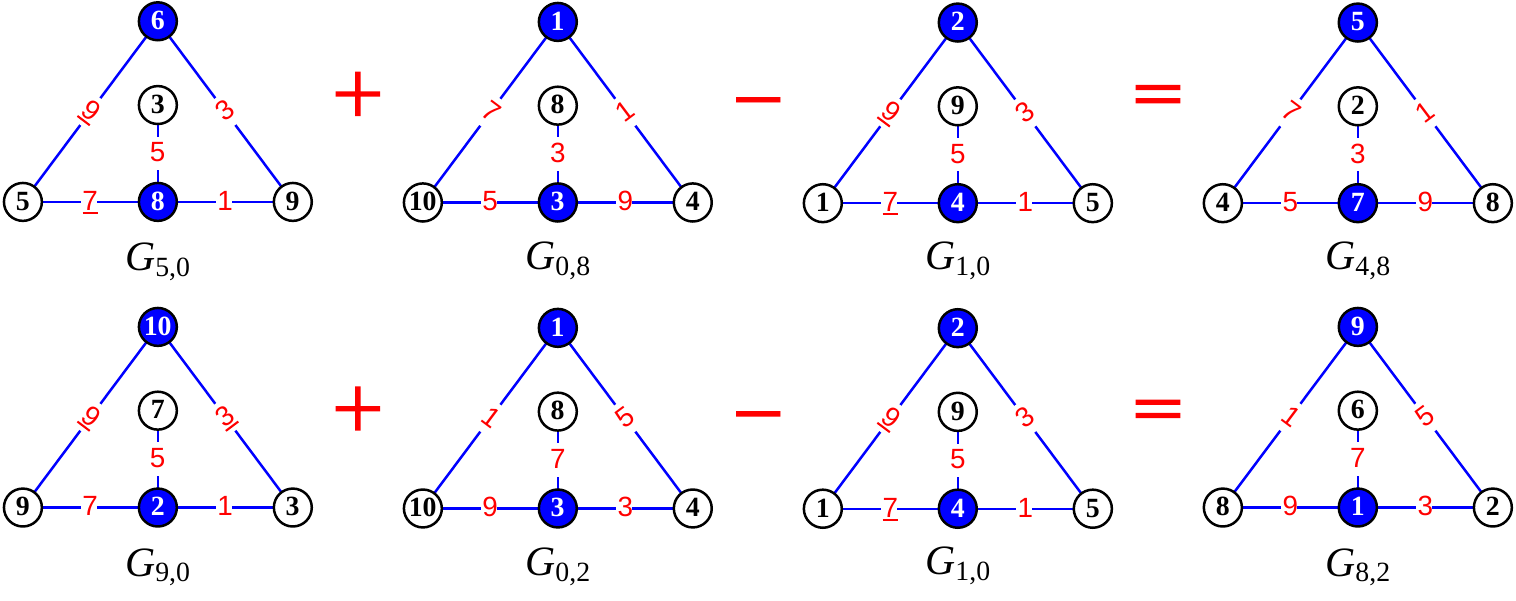}\\
\caption{\label{fig:odd-graceful-group-mixed} {\small Two examples for illustrating the formula Eq.(\ref{eqa:mixed-infinite-graphic-group}), cited from \cite{Yao-Sun-Su-Wang-Zhao-ICIBA-2020}.}}
\end{figure}

\begin{rem}\label{rem:infinite-graphic-sequence-group}
\cite{Yao-Sun-Su-Wang-Zhao-ICIBA-2020} Let $F^*(G,f)=F(\{\{G_{s,k}\}^{+\infty}_{-\infty}\}^{+\infty}_{-\infty};\oplus \ominus;(G,f))$ be an every-zero infinite graphic-sequence group. The elements of the every-zero infinite graphic-sequence group $F^*(G,f)$ can tile fully each point $(x,y)$ of $xOy$-plane. And moreover, $F^*(G,f)$ contains infinite every-zero graphic groups having finite elements, such as $F(\{G_{s+i,k}\}^{p_W}_{i=1}$; $[\oplus \ominus];(G,f))$ and $F(\{G_{s,k+j}\}^{q_W}_{i=1};\oplus \ominus;(G,f))$. Also, $F^*(G,f)$ contains infinite every-zero graphic groups of infinite elements.

Clearly, particular every-zero graphic groups having infinite elements or finite elements can be used easily to encrypt randomly the nodes of networks. Suppose that the coloring $f$ of $G$ in $F^*(G,f)$ is equivalent to another $W_g$-constraint total coloring $g$ of $G$. Then we get another every-zero infinite graphic-sequence group $F^*(G,g)=F(\{\{G_{i,j}\}^{+\infty}_{-\infty}\}^{+\infty}_{-\infty};\oplus \ominus;(G,g))$ with $G\cong G_{i,j}$. Thereby, the every-zero infinite graphic-sequence group $F^*(G,f)$ is a \emph{public-key graphic-sequence group}, the every-zero infinite graphic-sequence group $F^*(G,g)$ is a \emph{private-key graphic-sequence group} accordingly.

Since there exists a mapping $\varphi: V(G)\cup E(G)\rightarrow V(G)\cup E(G)$ such that $g(w)=\varphi(f(w))$ for $w\in V(G)\cup E(G)$, we claim that $F^*(G,f)$ admits an every-zero graphic-sequence homomorphism to $F^*(G,g)$, and moreover
\begin{equation}\label{eqa:555555}
F^*(G,f)\leftrightarrow F^*(G,g),
\end{equation} is a pair of \emph{homomorphically equivalent every-zero graphic-sequence homomorphisms}.\paralled
\end{rem}

\subsubsection{Graphic groups in encrypting dynamic networks}

In \cite{Yao-Sun-Su-Wang-Zhao-ICIBA-2020}, the author have designed the \emph{every-zero infinite graphic groups} with no modular: For a bipartite connected $(p,q)$-graph $H$ admitting a set-ordered $W_g$-constraint total coloring $g$ such that $\max g(X)<\min g(Y)$ for the bipartition $(X,Y)$ of vertices of $H$, we define a graph $H_{i,j}$ holding $H_{i,j}\cong H$ and admitting a $W_g$-constraint total coloring $h_{i,j}$ defined as: Let $i,j\in Z$, $h_{i,j}(u)=g(u)+i$ for each vertex $u\in X$, and $h_{i,j}(v)=g(v)+j$ for each vertex $v\in Y$, and $h_{i,j}(uv)=g(uv)+i+j$ for each edge $uv\in E(H)$. Next, the operation ``$H_{i,j}[\oplus \ominus _{a,b}]H_{k,s}:=H_{i,j}[\oplus ]H_{k,s}[\ominus ]H_{a,b}$'' is defined in the following: For a fixed $H_{a,b}$ admitting a $W_g$-constraint total coloring $h_{a,b}$, we set
\begin{equation}\label{eqa:new-vertex-graphic-group-basic-x}
h_{i,j}(u)+h_{k,s}(u)-h_{a,b}(u)=h_{\lambda,\mu}(u)+\lambda,~u\in X
\end{equation} with the index $\lambda=i+k-a$, and
\begin{equation}\label{eqa:new-vertex-graphic-group-basic-y}
h_{i,j}(v)+h_{k,s}(v)-h_{a,b}(v)=h_{\lambda,\mu}(v)+\mu,~v\in Y
 \end{equation} with the index $\mu=j+s-b$, and for $uv\in E(H)$, thus, we have
\begin{equation}\label{eqa:new-vertex-graphic-group-basic-uv}
{
\begin{split}
h_{i,j}(uv)+h_{k,s}(uv)-h_{a,b}(uv)=h_{\pi,\tau}(uv)+\pi+\tau
\end{split}}
\end{equation} with the indices $\pi=i+k-a$ and $\tau=j+s-b$.

Thereby, we get an infinite graph set $(\textbf{\textrm{H}}_{total})^{+\infty}_{-\infty}|^2=\{H_{s,k}: s,k\in Z\}$, which is an \emph{every-zero infinite graphic group} based on Eq.(\ref{eqa:new-vertex-graphic-group-basic-x}), Eq.(\ref{eqa:new-vertex-graphic-group-basic-y}) and Eq.(\ref{eqa:new-vertex-graphic-group-basic-uv}), as well as

$2h_{a,b}(u)+h_{i,j}(u)=h_{2a+i,2b+j}(u)=g(u)+2a+i$ for $u\in X$,

$2h_{a,b}(v)+h_{i,j}(v)=h_{2a+i,2b+j}(v)=g(v)+2b+j$ for $v\in Y$ and

$2h_{a,b}(uv)+h_{i,j}(uv)=h_{2a+i,2b+j}(uv)=g(uv)+2(a+b)+i+j$ for $uv\in E(H)$.

\vskip 0.4cm

\noindent \textbf{MULTIPLE-JOIN algorithm} is based on an every-zero infinite graphic group $(\textbf{\textrm{H}}_{total})^{+\infty}_{-\infty}|^2$.

\textbf{Initialization.} Select arbitrarily an element $H_{s_0,k_0}\in (\textbf{\textrm{H}}_{total})^{+\infty}_{-\infty}|^2$; two ends $i$ and $j$ of each edge $ij$ of the network model $N(0)$ are encrypted by $H_{s_i,k_i}$ and $H_{s_j,k_j}$ and the edge $ij$ is encrypted by $h_{\pi,\tau}(ij)+\pi+\tau$ with $\pi=s_i+s_j-s_0$ and $\tau=k_i+k_j-k_0$. We call $H_{s_0,k_0}$ an \emph{encryption base} of the network $N(0)$.

\textbf{Iteration.} Assume that a network model $N(t)$ is encrypted well by the every-zero infinite graphic group $(\textbf{\textrm{H}}_{total})^{+\infty}_{-\infty}|^2$ with the encryption base $H_{s_t,k_t}$. Add a new vertex $u$ to the network $N(t)$, and join it with vertices $x_1,x_2,\dots ,x_d$ of the network model $N(t)$ for produce a new network model $N(t+1)$, where each $x_r$ was encrypted by $H_{s_{x_r},k_{x_r}}$ with $r\in [1,d]$; take arbitrarily an element $H_{s_{t+1},k_{t+1}}$ as an encryption base of the network model $N(t+1)$ and encrypt $u$ by selecting randomly $H_{s_u,k_u}\in (\textbf{\textrm{H}}_{total})^{+\infty}_{-\infty}|^2$, each edge $ux_r$ is encrypted by $h_{\pi,\tau}(ux_r)+\pi+\tau$ with the indices $\pi=s_u+s_{x_r}-s_{t+1}$ and $\tau=k_u+k_{x_r}-k_{t+1}$.

\vskip 0.3cm

Since, in the MULTIPLE-JOIN algorithm, the vertices and edges of the network model $N(t)$ are encrypted by the graphs of $(\textbf{\textrm{H}}_{total})^{+\infty}_{-\infty}|^2$, we construct another network model $R(t)$ by replacing each vertex of $N(t)$ with the graphs of $(\textbf{\textrm{H}}_{total})^{+\infty}_{-\infty}|^2$, each edge $uv$ of $N(t)$ is replaced by $H_{s_{uv},k_{uv}}=H_{s_u,k_u}[\oplus ]H_{s_v,k_v}[\ominus ]H_{a,b}$, and use new edges to join $H_{s_u,k_u}$ with $H_{s_{uv},k_{uv}}$ together, and to join $H_{s_{uv},k_{uv}}$ with $H_{s_v,k_v}$ together. We write
$$R(t)=N(t)\big [\overline{\odot}_{uv\in E(N(t))}\big ] H_{s_{uv},k_{uv}}\langle H_{s_u,k_u}, H_{s_v,k_v}\rangle,$$
and moreover we get an \emph{infinite graphic group lattice}
\begin{equation}\label{eqa:c3xxxxx}
{
\begin{split}
\textbf{\textrm{L}}(\textbf{\textrm{F}}\overline{\odot }\textbf{\textrm{H}}_{total})=\big \{R(t):N(t)\in \textbf{\textrm{F}}(n_v,n_e)(t); H_{s_{uv},k_{uv}}, H_{s_u,k_u}, H_{s_v,k_v}\in (\textbf{\textrm{H}}_{total})^{+\infty}_{-\infty}|^2\big \}
\end{split}}
\end{equation} under the infinite base $(\textbf{\textrm{H}}_{total})^{+\infty}_{-\infty}|^2$, where $\textbf{\textrm{F}}(n_v,n_e)(t)$ is a set of networks having vertex numbers $\leq n_v(t)$ and edge numbers $\leq n_e(t)$.

\begin{rem} \label{rem:multiple-join-advantage}
The MULTIPLE-JOIN algorithm has the following advantages:
\begin{asparaenum}[\textrm{\textbf{Advantage}}-1. ]
\item In the process of forming $N(t+1)$, it is allowed that a new vertex joins $d~(\geq 2)$ vertices of $N(t)$, so it is easy to produce particular network models to meet the needs of the \emph{growth} and the \emph{preferential attachment} of BA-models \cite{Barabasi-Reka-Albert-1999}, and we have $n_v(t+1)=n_v(t)+1$ and $n_e(t+1)=n_e(t)+d$.
\item It is allowed to \emph{add new edges} to joining vertices of $N(t)$, and \emph{remove old edges} of $N(t)$ in the process of forming $N(t+1)$ for approximating real networks.
\item Since an every-zero infinite graphic group $(\textbf{\textrm{H}}_{total})^{+\infty}_{-\infty}|^2$ contains infinite elements, and $N(t_i)$ and $N(t_j)$ have different encryption bases at two distinct time steps $t_i,t_j$, so it increases the cost of deciphering $N(t)$.
\item Suppose that each graph $H_{s,k}$ admits a graph homomorphism to a graph $G_{s,k}$, that is $H_{s,k}\rightarrow G_{s,k}$, then we can get another a graph homomorphism group $(\textbf{\textrm{G}}_{total})^{+\infty}_{-\infty}|^2=\{G_{s,k}: s,k\in Z\}$, so we get an \emph{every-zero infinite graphic group homomorphism} $(\textbf{\textrm{H}}_{total})^{+\infty}_{-\infty}|^2\rightarrow (\textbf{\textrm{G}}_{total})^{+\infty}_{-\infty}|^2$.
\item In $(\textbf{\textrm{H}}_{total})^{+\infty}_{-\infty}|^2$, if the bipartite graph $H$ is a tree admitting a set-ordered graceful coloring $g$, then there are many colorings equivalent with $g$ proven in \cite{Gallian2021}, so we get infinite graphs $T^r_{s,k}$ admitting total coloring $\theta^r_{i,j}$ equivalent with $h_{i,j}$ defined in $(\textbf{\textrm{H}}_{total})^{+\infty}_{-\infty}|^2$. Hence, we have many every-zero infinite graphic groups $(\textbf{\textrm{T}}^r_{total})^{+\infty}_{-\infty}|^2=\{T^r_{s,k}: s,k\in Z\}$ with $r\geq 1$. We write this mutually equivalent every-zero infinite graphic groups by $(\textbf{\textrm{H}}_{total})^{+\infty}_{-\infty}|^2~\sim ~(\textbf{\textrm{T}}^r_{total})^{+\infty}_{-\infty}|^2$.
\end{asparaenum}
\end{rem}

\subsubsection{Multi-level graphic group graph-coloring algorithm}

Using a graphic group colors a network overall, where each graph in the graphic group is also colored by another graphic group, we call this process \emph{multi-level graphic group graph-coloring algorithm} (MLGGG-coloring algorithm).

\vskip 0.2cm

\noindent \textbf{Multi-level graphic group graph-coloring algorithm.}

\textbf{Initialization.} Let $G_m=(G_{m,1},G_{m,2},\dots ,G_{m,n_m})$ be a \emph{graph base} with $G_{m,i}\not \cong G_{m,j}$ if $i\neq j$. A connected graph $H$ admits a proper total graph-coloring $\varphi:V(H)\cup E(H)\rightarrow \{F(G_m);\oplus \ominus\}$, where $\{F(G_m);\oplus \ominus\}$ is an \emph{every-zero graphic group}, such that
$$\varphi(xy):=\varphi(x)[\oplus \ominus_p]\varphi(y)
$$ for a preappointed \emph{zero} $G_{m,p}\in G_m$, where $\varphi(x)=G_{m,i}\in G_m$, $\varphi(y)=G_{m,j}\in G_m$ and $\varphi(xy)=G_{m,\lambda}\in G_m$ with the index $\lambda=i+j-p~(\bmod~n_m)$. We get a graphic Topcode-matrix
\begin{equation}\label{eqa:5555555555555555}
\centering
{
\begin{split}
T_{code}(H,\varphi)= \left(
\begin{array}{ccccc}
\varphi(x_{1}) & \varphi(x_{2}) & \cdots & \varphi(x_{q})\\
\varphi(x_{1}y_{1}) & \varphi(x_{2}y_{2}) & \cdots & \varphi(x_{q}y_{q})\\
\varphi(y_{1}) & \varphi(y_{2}) & \cdots & \varphi(y_{q})
\end{array}
\right)=\left(
\begin{array}{cccccccccccccc}
X_\varphi\\
E_\varphi\\
Y_\varphi
\end{array}
\right)=(X_\varphi,E_\varphi,Y_\varphi)^{T}
\end{split}}
\end{equation} where each edge $x_{i}y_{i}\in E(H)$; refer to Definition \ref{defn:graphic-topcode-matrix}.

\textbf{Step 1.} Each graph $G_{m,j}$ for $j\in [1,n_m]$ in the graph base $G_m$ admits a proper total graph-coloring
$$\varphi_{m,j}:V(G_{m,j})\cup E(G_{m,j})\rightarrow \{F(G_{m-1});\oplus \ominus\}
$$ based on a graph base $G_{m-1}=(G_{m-1,1},G_{m-1,2},\dots ,G_{m-1,n_{m-1}})$.

\textbf{Step $k$.} In each graph base $G_{m-k}=(G_{m-k,1},G_{m-k,2},\dots ,G_{m-k,n_{m-k}})$, each graph $G_{m-k,j}$ for $j\in [1,n_{m-k}]$ admits a proper total graph-coloring
$$\varphi_{m-k,j}:V(G_{m-k,j})\cup E(G_{m-k,j})\rightarrow \{F(G_{m-k-1});\oplus \ominus\}
$$ based on a graph base $G_{m-k-1}=(G_{m-k-1,1},G_{m-k-1,2},\dots ,G_{m-k-1,n_{m-k-1}})$ for $k\in [0,m-1]$.

\textbf{Step $m-1$.} Each graph $G_{1,j}$ for $j\in [1,n_{1}]$ in the graph base $G_{1}=(G_{1,1},G_{1,2},\dots ,G_{1,n_{1}})$ admits a proper total graph-coloring
$$\varphi_{1,j}:V(G_{1,j})\cup E(G_{1,j})\rightarrow \{F(G_{0});\oplus \ominus\}
$$ based on a graph base $G_{0}=(G_{0,1},G_{0,2},\dots ,G_{0,n_{0}})$.

\textbf{Step $m$.} Finally, each graph $G_{0,j}$ for $j\in [1,n_{0}]$ in the graph base $G_{0}$ admits a proper total string-coloring
$$\theta_{0,j}:V(G_{0,j})\cup E(G_{0,j})\rightarrow S_{[0,9]\textrm{-string}}
$$ where the $[0,9]$-string set $S_{[0,9]\textrm{-string}}=\{s_i=c_{i,1}c_{i,2}\cdots c_{i,n}:c_{i,j}\in [0,9]\}$.

\begin{rem}\label{rem:333333}
The above every-zero graphic groups
$$\{F(G_m);\oplus \ominus\},\{F(G_{m-1});\oplus \ominus\},\dots,\{F(G_1);\oplus \ominus\},\{F(G_0);\oplus \ominus\}
$$ and a $[0,9]$-string set $S_{[0,9]\textrm{-string}}$ form a multiple protection of a network, and it greatly improves the security of the whole network encryption. In the era of quantum computer, the time consumed by the multiple encryption mentioned here can be completely ignored.\paralled
\end{rem}

\subsection{Vector-colorings made by strings}

Since a number-based string $s=a_1a_2\cdots a_n$ can correspond to a vector $(a_1,a_2,\dots ,a_n)$, or a set $\{a_1,a_2,\dots ,a_n\}$ (no order), or a vertex-degree $\textrm{deg}=(a_1,a_2,\dots ,a_n)$ (no order), we consider \emph{vector groups} based on the following \emph{vector set}
\begin{equation}\label{eqa:string-vector-set}
V^n_{ector}\big (\{(9)_j\}^n_{j=1}\big )=\big \{v_{ector}(\{m_j\}^n_{j=1})=(C_{1,m_1},C_{2,m_2},\dots ,C_{n,m_n}):m_j\in [1,(9)_j],j\in [1,n]\big \}
\end{equation} with notations $C_{1,m_1},C_{2,m_2},\dots ,C_{n,m_n}$ and $(9)_j$ defined in Definition \ref{defn:number-based-super-string-def}; and consider \emph{vertex-degree groups} based on the following \emph{vertex-degree set}
\begin{equation}\label{eqa:string-vertex-degree-set}
D^n_{egree}\big (\{(9)_j\}^n_{j=1}\big )=\big \{\textrm{deg}(\{m_j\}^n_{j=1})=(C_{1,m_1},C_{2,m_2},\dots ,C_{n,m_n}):m_j\in [1,(9)_j],j\in [1,n]\big \}
\end{equation} with notations $C_{1,m_1},C_{2,m_2},\dots ,C_{n,m_n}$ and $(9)_j$ defined in Definition \ref{defn:number-based-super-string-def}.

Let $\textbf{\textrm{S}}_{tring}(n)$ be the set of \emph{$n$-rank number-based strings} $\alpha_{1}\alpha_{2}\cdots \alpha_{n}$ with each number $\alpha_{j}\ge 0$ for $j\in [1,n]$, and a $(p,q)$-graph $G$ admits a total string-coloring $f:V(G)\cup E(G)\rightarrow \textbf{\textrm{S}}_{tring}(n)$, such that
\begin{equation}\label{eqa:homogeneous-vector-colorings11}
f(u_k)=a_{k,1}a_{k,2}\cdots a_{k,n},~f(v_k)=b_{k,1}b_{k,2}\cdots b_{k,n},~f(u_kv_k)=c_{k,1}c_{k,2}\cdots c_{k,n}
\end{equation} for each edge $u_kv_k\in E(G)=\{u_kv_k:k\in [1,q]\}$. Immediately, we get three \emph{vector colorings}
\begin{equation}\label{eqa:homogeneous-vector-colorings22}
\pi(u_k)=(a_{k,1},a_{k,2},\dots ,a_{k,n}),~\pi(v_k)=(b_{k,1},b_{k,2},\dots ,b_{k,n}),~\pi(u_kv_k)=(c_{k,1},c_{k,2},\dots ,c_{k,n})
\end{equation} for each edge $u_kv_k\in E(G)=\{u_kv_k:k\in [1,q]\}$; and moreover three \emph{set-colorings}
\begin{equation}\label{eqa:homogeneous-vector-colorings22}
\psi(u_k)=\{a_{k,1},a_{k,2},\dots ,a_{k,n}\},~\psi(v_k)=\{b_{k,1},b_{k,2},\dots ,b_{k,n}\},~\psi(u_kv_k)=\{c_{k,1},c_{k,2},\dots ,c_{k,n}\}
\end{equation} for each edge $u_kv_k\in E(G)=\{u_kv_k:k\in [1,q]\}$.

\begin{defn} \label{defn:n-di-vector-colorings-definition}
$^*$ \textbf{Homogeneous $(abc)$-magic vector-colorings.} Let $\textbf{\textrm{V}}_{ector}(n)$ be the set of \emph{$n$-rank vectors} $(\alpha_{1},\alpha_{2},\dots ,\alpha_{n})$ with each coordinate $\alpha_{j}\in Z^0$ for $j\in [1,n]$. A $(p,q)$-graph $G$ admits a \emph{$W$-constraint total vector-coloring} $\overrightarrow{\pi} : V(G)\cup E(G)\rightarrow \textbf{\textrm{V}}_{ector}(n)$, such that each edge $u_kv_k\in E(G)=\{u_kv_k:k\in [1,q]\}$ holds
\begin{equation}\label{eqa:n-di-vector-coloringss}
\overrightarrow{\pi}(u_k)=(a_{k,1},a_{k,2},\dots ,a_{k,n}),~\overrightarrow{\pi}(v_k)=(b_{k,1},b_{k,2},\dots ,b_{k,n}),~\overrightarrow{\pi}(u_kv_k)=(c_{k,1},c_{k,2},\dots ,c_{k,n})
\end{equation} subject to the $W$-constraint $\overrightarrow{\pi}(u_kv_k)=W\langle \overrightarrow{\pi}(u_k),\overrightarrow{\pi}(v_k)\rangle$. Let $\lambda$ and $\gamma$ be constants, there are the following $(abc)$-magic vector-constraints:
\begin{asparaenum}[\textbf{\textrm{Vect}}-1. ]
\item \label{4vector-magic:uniform-edge-magic} Each $j\in [1,n]$ holds the edge-magic constraint $a_{k,j}+b_{k,j}+c_{k,j}=\lambda$ true, denoted as
\begin{equation}\label{eqa:555555}
 \overrightarrow{\pi}(w_k)[+]\overrightarrow{\pi}(w_kz_k)[+]\overrightarrow{\pi}(z_k)=\lambda
 \end{equation}
\item \label{4vector-magic:uniform-edge-difference} Each $j\in [1,n]$ holds the edge-difference constraint $c_{k,j}+|a_{k,j}-b_{k,j}|=\lambda$ true, denoted as \begin{equation}\label{eqa:555555}
 \overrightarrow{\pi}(w_kz_k)[+]|\overrightarrow{\pi}(w_k)[-]\overrightarrow{\pi}(z_k)|=\lambda
\end{equation}
\item \label{4vector-magic:uniform-graceful-difference} Each $j\in [1,n]$ holds the graceful-difference constraint $\big ||a_{k,j}-b_{k,j}|-c_{k,j}\big |=\lambda$ true, denoted as
\begin{equation}\label{eqa:555555}
\big ||\overrightarrow{\pi}(w_k)[-]\overrightarrow{\pi}(z_k)|[-]\overrightarrow{\pi}(w_kz_k)\big |=\lambda
\end{equation}
\item \label{4vector-magic:uniform-felicitous-difference} Each $j\in [1,n]$ holds the felicitous-difference constraint $|a_{k,j}+b_{k,j}-c_{k,j}|=\lambda$ true, denoted as
 \begin{equation}\label{eqa:555555}|\overrightarrow{\pi}(w_k)[+]\overrightarrow{\pi}(z_k)[-]\overrightarrow{\pi}(w_kz_k)|=\lambda
\end{equation}
\item \label{4vector-magic:weak-edge-magic} Some $r\in [1,n]$ holds the edge-magic constraint $a_{k,r}+b_{k,r}+c_{k,r}=\gamma$ true, but not all, denoted as $\partial_r \langle \overrightarrow{\pi}(w_k)[+]\overrightarrow{\pi}(w_kz_k)[+]\overrightarrow{\pi}(z_k)=\gamma \rangle$.
\item \label{4vector-magic:weak-edge-difference} Some $s\in [1,n]$ holds the edge-difference constraint $c_{k,s}+|a_{k,s}-b_{k,s}|=\gamma$ true, but not all, denoted as $\partial_s \langle \overrightarrow{\pi}(w_kz_k)[+]|\overrightarrow{\pi}(w_k)[-]\overrightarrow{\pi}(z_k)|=\gamma \rangle$.
\item \label{4vector-magic:weak-graceful-difference} Some $t\in [1,n]$ holds the graceful-difference constraint $\big ||a_{k,t}-b_{k,t}|-c_{k,t}\big |=\gamma$ true, but not all, denoted as $\partial_t \langle \big ||\overrightarrow{\pi}(w_k)[-]\overrightarrow{\pi}(z_k)|[-]\overrightarrow{\pi}(w_kz_k)\big |=\gamma \rangle$.
\item \label{4vector-magic:weak-felicitous-difference} Some $d\in [1,n]$ holds the felicitous-difference constraint $|a_{k,d}+b_{k,d}-c_{k,d}|=\gamma$ true, but not all, denoted as $\partial_d \langle |\overrightarrow{\pi}(w_k)[+]\overrightarrow{\pi}(z_k)[-]\overrightarrow{\pi}(w_kz_k)|=\gamma \rangle$.
\end{asparaenum}
\textbf{We call the total vector-coloring $\overrightarrow{\pi}$}
\begin{asparaenum}[\textbf{\textrm{Vectabc}}-1. ]
\item a \emph{component edge-magic total vector-coloring} if Vect-\ref{4vector-magic:uniform-edge-magic} holds true.
\item a \emph{component edge-difference total vector-coloring} if Vect-\ref{4vector-magic:uniform-edge-difference} is true.
\item a \emph{component graceful-difference total vector-coloring} if Vect-\ref{4vector-magic:uniform-graceful-difference} holds true.
\item a \emph{component felicitous-difference total vector-coloring} if Vect-\ref{4vector-magic:uniform-felicitous-difference} holds true.
\item a \emph{weak-component edge-magic total vector-coloring} if Vect-\ref{4vector-magic:weak-edge-magic} is true.
\item a \emph{weak-component edge-difference total vector-coloring} if Vect-\ref{4vector-magic:weak-edge-difference} holds true.
\item a \emph{weak-component graceful-difference total vector-coloring} if Vect-\ref{4vector-magic:weak-graceful-difference} holds true.
\item a \emph{weak-component felicitous-difference total vector-coloring} if Vect-\ref{4vector-magic:weak-felicitous-difference} is true.\qqed
\end{asparaenum}
\end{defn}

\begin{example}\label{exa:8888888888}
In Definition \ref{defn:n-di-vector-colorings-definition}, a $W$-constraint total vector-coloring $\overrightarrow{\pi} : V(G)\cup E(G)\rightarrow \textbf{\textrm{V}}_{ector}(n)$ holds one $W$-constraint of the following vector-operation constraints
$$
\overrightarrow{\pi}(u_k)+\overrightarrow{\pi}(v_k)=\overrightarrow{\pi}(u_kv_k),~\overrightarrow{\pi}(u_kv_k)+\overrightarrow{\pi}(v_k)=\overrightarrow{\pi}(u_k),~ \overrightarrow{\pi}(u_k)+\overrightarrow{\pi}(u_kv_k)=\overrightarrow{\pi}(v_k)
$$ or $\overrightarrow{\pi}(u_kv_k)=\overrightarrow{\pi}(u_k)\cdot \overrightarrow{\pi}(v_k)$, or $\overrightarrow{\pi}(u_kv_k)=\overrightarrow{\pi}(u_k)\times \overrightarrow{\pi}(v_k)$, or $\lambda=[\overrightarrow{\pi}(u_k)\times \overrightarrow{\pi}(v_k)]\cdot \overrightarrow{\pi}(u_kv_k)$ for each edge $u_kv_k\in E(G)=\{u_kv_k:k\in [1,q]\}$.\qqed
\end{example}

\begin{prop}\label{prop:99999}
$^*$ For a fixed vector $\overrightarrow{\lambda}$, there are more groups of different vectors $\overrightarrow{a},\overrightarrow{b},\overrightarrow{c}$ hold

(i) The \emph{vector-edge-magic constraint} $\overrightarrow{a}+\overrightarrow{b}+\overrightarrow{c}=\overrightarrow{\lambda}$.

(ii) The \emph{vector-edge-difference constraint} $\overrightarrow{c}+\overrightarrow{a}-\overrightarrow{b}=\overrightarrow{\lambda}$.

(iii) The \emph{vector-felicitous-difference constraint} $\overrightarrow{a}+\overrightarrow{b}-\overrightarrow{c}=\overrightarrow{\lambda}$, or $\overrightarrow{c}-(\overrightarrow{a}+\overrightarrow{b})=\overrightarrow{\lambda}$.

(iv) The \emph{vector-graceful-difference constraint} $(\overrightarrow{a}-\overrightarrow{b})-\overrightarrow{c}=\overrightarrow{\lambda}$, or $\overrightarrow{c}-(\overrightarrow{a}-\overrightarrow{b})=\overrightarrow{\lambda}$.
\end{prop}

\begin{defn} \label{defn:n-di-vector-colorings-definition22}
$^*$ \textbf{Non-homogeneous magic vector-colorings.} Let $V_{ector}$ be a set of vectors. Suppose that a graph $G$ admits a total vector-coloring $F:V(G)\cup E(G)\rightarrow V_{ector}$, such that $F(u)=\overrightarrow{v}_{u}$, $F(v)=\overrightarrow{v}_{v}$, and $F(uv)=\overrightarrow{v}_{uv}$ for each edge $uv\in E(G)$.

(i) If there is a fixed vector $\overrightarrow{\lambda}$, such that each edge $uv\in E(G)$ holds the \emph{vector-edge-magic constraint} $\overrightarrow{v}_{u}+ \overrightarrow{v}_{v}+ \overrightarrow{v}_{uv}=\overrightarrow{\lambda}$ true, we call $F$ \emph{edge-magic total vector-coloring}.

(ii) If there is a fixed vector $\overrightarrow{\lambda}$, such that each edge $uv\in E(G)$ holds one of \emph{vector-edge-difference constraints} $\overrightarrow{v}_{uv}+ (\overrightarrow{v}_{u}- \overrightarrow{v}_{v})=\overrightarrow{\lambda}$ and $\overrightarrow{v}_{uv}+ (\overrightarrow{v}_{v}- \overrightarrow{v}_{u})=\overrightarrow{\lambda}$ true, we call $F$ \emph{edge-difference total vector-coloring}.

(iii) If there is a fixed vector $\overrightarrow{\lambda}$, such that each edge $uv\in E(G)$ holds one of \emph{vector-felicitous-difference constraints} $\overrightarrow{v}_{uv}-(\overrightarrow{v}_{u}+ \overrightarrow{v}_{v}) =\overrightarrow{\lambda}$ and $(\overrightarrow{v}_{v}+ \overrightarrow{v}_{u})- \overrightarrow{v}_{uv}=\overrightarrow{\lambda}$ true, we call $F$ \emph{felicitous-difference total vector-coloring}.

(iv) If there is a fixed vector $\overrightarrow{\lambda}$, such that each edge $uv\in E(G)$ holds one of \emph{vector-graceful-difference constraints} $(\overrightarrow{v}_{u}-\overrightarrow{v}_{v})-\overrightarrow{v}_{uv}=\overrightarrow{\lambda}$, $(\overrightarrow{v}_{v}-\overrightarrow{v}_{u})-\overrightarrow{v}_{uv}=\overrightarrow{\lambda}$, $\overrightarrow{v}_{uv}-(\overrightarrow{v}_{u}-\overrightarrow{v}_{v})=\overrightarrow{\lambda}$ and $\overrightarrow{v}_{uv}-(\overrightarrow{v}_{v}-\overrightarrow{v}_{u})=\overrightarrow{\lambda}$ true, we call $F$ \emph{graceful-difference total vector-coloring}.\qqed
\end{defn}

\begin{rem}\label{rem:333333}
The $W$-constraint $\overrightarrow{\pi}(u_kv_k)=W\langle \overrightarrow{\pi}(u_k),\overrightarrow{\pi}(v_k)\rangle$ appeared in Definition \ref{defn:n-di-vector-colorings-definition} and Definition \ref{defn:n-di-vector-colorings-definition22} is a group of constraints like that defined in Definition \ref{defn:homoge-uniformly-string-total-colorings} (Homogeneous string-coloring), Definition \ref{defn:homoge-various-string-total-colorings} (Weak homogeneous string-coloring) and Definition \ref{defn:homoge-4-magic-string-colorings} (Magic-type homogeneous string-coloring).\paralled
\end{rem}

\begin{defn} \label{defn:integer-string-lattice00}
$^*$ Suppose that there are vectors $\textbf{S}_i=(a_{i,1},a_{i,2},\dots ,a_{i,m})$ made by number-based strings $s_i=a_{i,1}a_{i,2}\cdots a_{i,m}$ with $i\in [1,n]$. If a vector base $\textbf{V}_{string}=(\textbf{S}_1, \textbf{S}_2,\dots , \textbf{S}_n)$ is consisted of $n$ linearly independent vectors $\textbf{S}_1, \textbf{S}_2,\dots , \textbf{S}_n$, then we call the following set
\begin{equation}\label{eqa:ssssss}
\textrm{\textbf{L}}(\textbf{V}_{string}) =\{x_1\textbf{\textrm{S}}_1+x_2\textbf{\textrm{S}}_2+\cdots +x_n\textbf{\textrm{S}}_n : x_i \in Z\}
\end{equation} \emph{integer-string lattice} based on the number-based strings $s_1,s_2,\dots ,s_n$.\qqed
\end{defn}
\begin{rem}\label{rem:333333}
In Definition \ref{defn:integer-string-lattice00}, a group of number-based strings $s_1,s_2,\dots ,s_n$ can induces $(m!)^n$ different vector bases like the vector base $\textbf{V}_{string}$ defined in Definition \ref{defn:integer-string-lattice00}.

The authors in \cite{Zhang-Yang-Yao-Frontiers-Computer-2021} study some relationships between traditional lattices and graph lattices of topological coding. \paralled
\end{rem}

\subsection{Set-colorings made by strings}

\begin{defn} \label{defn:n-di-set-colorings-definition}
$^*$ \textbf{Homogeneous $(abc)$-magic set-colorings.} Let $\textbf{\textrm{S}}_{et}(\leq n)$ be the set of sets $\{\alpha_{1},\alpha_{2},\dots ,\alpha_{m}\}$ with each element $\alpha_{j}\in Z^0$ for $j\in [1,m]$ and $m\leq n$. A $(p,q)$-graph $G$ admits a \emph{$\{W_i\}^A_{i=1}$-constraint total set-coloring} $\psi : V(G)\cup E(G)\rightarrow \textbf{\textrm{S}}_{et}(\leq n)$, such that each edge $u_kv_k\in E(G)=\{u_kv_k:k\in [1,q]\}$ holds
\begin{equation}\label{eqa:n-di-set-coloringss}
{
\begin{split}
&\psi(u_k)=\{a_{k,1},a_{k,2},\dots ,a_{k,n}\},~\psi(v_k)=\{b_{k,1},b_{k,2},\dots ,b_{k,n}\},\\
&\psi(u_kv_k)=\{c_{k,1},c_{k,2},\dots ,c_{k,n}\}
\end{split}}
\end{equation} subject to the $W_i$-constraint $W_i\langle \psi(u_k),\psi(u_kv_k),\psi(v_k)\rangle=0$ for some $i\in [1,A]$. Let $\lambda$ and $\gamma$ be constants, there are the following $(abc)$-magic set-constraints:

\begin{asparaenum}[\textbf{\textrm{Set}}-1. ]
\item \label{4set-magic:uniform-edge-magic} Each $j\in [1,n]$ holds the edge-magic constraint $a_{k,j}+b_{k,j}+c_{k,j}=\lambda$ true, denoted as
\begin{equation}\label{eqa:555555}
\psi(w_k)[+]\psi(w_kz_k)[+]\psi(z_k)=\lambda
\end{equation}
\item \label{4set-magic:uniform-edge-difference} Each $j\in [1,n]$ holds the edge-difference constraint $c_{k,j}+|a_{k,j}-b_{k,j}|=\lambda$ true, denoted as
\begin{equation}\label{eqa:555555}
\psi(w_kz_k)[+]|\psi(w_k)[-]\psi(z_k)|=\lambda
\end{equation}
\item \label{4set-magic:uniform-graceful-difference} Each $j\in [1,n]$ holds the graceful-difference constraint $\big ||a_{k,j}-b_{k,j}|-c_{k,j}\big |=\lambda$ true, denoted as
\begin{equation}\label{eqa:555555}
\big ||\psi(w_k)[-]\psi(z_k)|[-]\psi(w_kz_k)\big |=\lambda
\end{equation}
\item \label{4set-magic:uniform-felicitous-difference} Each $j\in [1,n]$ holds the felicitous-difference constraint $|a_{k,j}+b_{k,j}-c_{k,j}|=\lambda$ true, denoted as
\begin{equation}\label{eqa:555555}
|\psi(w_k)[+]\psi(z_k)[-]\psi(w_kz_k)|=\lambda
\end{equation}
\item \label{4set-magic:weak-edge-magic} Some $r\in [1,n]$ holds the edge-magic constraint $a_{k,r}+b_{k,r}+c_{k,r}=\gamma$ true, but not all, denoted as $\partial_r \langle \psi(w_k)[+]\psi(w_kz_k)[+]\psi(z_k)=\gamma \rangle$.
\item \label{4set-magic:weak-edge-difference} Some $s\in [1,n]$ holds the edge-difference constraint $c_{k,s}+|a_{k,s}-b_{k,s}|=\gamma$ true, but not all, denoted as $\partial_s \langle \psi(w_kz_k)[+]|\psi(w_k)[-]\psi(z_k)|=\gamma \rangle$.
\item \label{4set-magic:weak-graceful-difference} Some $t\in [1,n]$ holds the graceful-difference constraint $\big ||a_{k,t}-b_{k,t}|-c_{k,t}\big |=\gamma$ true, but not all, denoted as $\partial_t \langle \big ||\psi(w_k)[-]\psi(z_k)|[-]\psi(w_kz_k)\big |=\gamma \rangle$.
\item \label{4set-magic:weak-felicitous-difference} Some $d\in [1,n]$ holds the felicitous-difference constraint $|a_{k,d}+b_{k,d}-c_{k,d}|=\gamma$ true, but not all, denoted as $\partial_d \langle |\psi(w_k)[+]\psi(z_k)[-]\psi(w_kz_k)|=\gamma \rangle$.
\end{asparaenum}
\textbf{We call the total set-coloring $\psi$}
\begin{asparaenum}[\textbf{\textrm{Setabc}}-1. ]
\item a \emph{component edge-magic total set-coloring} if it holds Set-\ref{4set-magic:uniform-edge-magic} true.
\item a \emph{component edge-difference total set-coloring} if it holds Set-\ref{4set-magic:uniform-edge-difference} true.
\item a \emph{component graceful-difference total set-coloring} if it holds Set-\ref{4set-magic:uniform-graceful-difference} true.
\item a \emph{component felicitous-difference total set-coloring} if it holds Set-\ref{4set-magic:uniform-felicitous-difference} true.
\item a \emph{weak-component edge-magic total set-coloring} if it holds Set-\ref{4set-magic:weak-edge-magic} true.
\item a \emph{weak-component edge-difference total set-coloring} if it holds Set-\ref{4set-magic:weak-edge-difference} true.
\item a \emph{weak-component graceful-difference total set-coloring} if it holds Set-\ref{4set-magic:weak-graceful-difference} true.
\item a \emph{weak-component felicitous-difference total set-coloring} if it holds Set-\ref{4set-magic:weak-felicitous-difference} true.\qqed
\end{asparaenum}
\end{defn}

\begin{rem}\label{rem:333333}
The $W$-constraint $W_i\langle \psi(u_k),\psi(u_kv_k),\psi(v_k)\rangle=0$ in Definition \ref{defn:n-di-set-colorings-definition} is a group of constraints like that defined in Definition \ref{defn:homoge-uniformly-string-total-colorings} (Homogeneous string-coloring), Definition \ref{defn:homoge-various-string-total-colorings} (Weak homogeneous string-coloring) and Definition \ref{defn:homoge-4-magic-string-colorings} (Magic-type homogeneous string-coloring).

Moreover, we can set the colors of vertices and edges as
\begin{equation}\label{eqa:different-dimendion-set-coloringss}
{
\begin{split}
&\psi(u_k)=\{a_{k,1},a_{k,2},\dots ,a_{k,n(k,r)}\},~\psi(v_k)=\{b_{k,1},b_{k,2},\dots ,b_{k,n(k,s)}\},\\
&\psi(u_kv_k)=\{c_{k,1},c_{k,2},\dots ,c_{k,n(k,t)}\}
\end{split}}
\end{equation} for each edge $u_kv_k\in E(G)=\{u_kv_k:k\in [1,q]\}$ under a \emph{$\{W_i\}^A_{i=1}$-constraint total set-coloring} $\psi : V(G)\cup E(G)\rightarrow \textbf{\textrm{S}}_{et}(\leq n)$ of a $(p,q)$-graph $G$. We modify the conditions of Definition \ref{defn:n-di-set-colorings-definition} slightly, and then get the same set-colorings defined in Definition \ref{defn:n-di-set-colorings-definition}.

For example, we set: If each number $c_{k,j}\in \psi(u_kv_k)$ corresponds to some $a_{k,r}\in \psi(u_k)$ and some $b_{k,s}\in \psi(v_k)$ holding the edge-magic constraint $a_{k,r}+b_{k,s}+c_{k,j}=\lambda$ true; each number $a_{k,r}\in \psi(u_k)$ corresponds to some $c_{k,j}\in \psi(u_kv_k)$ and some $b_{k,s}\in \psi(v_k)$ holding the edge-magic constraint $a_{k,r}+b_{k,s}+c_{k,j}=\lambda$ true; and each number $b_{k,s}\in \psi(v_k)$ corresponds to some $a_{k,r}\in \psi(u_k)$ and some $c_{k,j}\in \psi(u_kv_k)$ holding the edge-magic constraint $a_{k,r}+b_{k,s}+c_{k,j}=\lambda$ true. Then we call the total set-coloring $\psi$ \emph{component edge-magic total set-coloring}.

The set-colorings defined in Definition \ref{defn:n-di-set-colorings-definition} can be related with the intersected-graphs of hypergraphs as the set $\textbf{\textrm{S}}_{et}(\leq n)$ appeared in Definition \ref{defn:n-di-set-colorings-definition} is a \emph{hyperedge set} $\mathcal{E}$ holding $\Lambda=\bigcup_{e\in \mathcal{E}}e$, where $\Lambda$ is a set of finite numbers.\paralled
\end{rem}

\begin{prop}\label{prop:99999}
$^*$ For a fixed set $U$, there are more groups of different sets $S_a,S_b,S_c$ hold

(i) The \emph{set-edge-magic constraint} $S_a\cup S_b\cup S_c=U$.

(ii) The \emph{set-edge-difference constraint} $S_c\cup (S_a\setminus S_b)=U$.

(iii) The \emph{set-felicitous-difference constraint} $(S_a\cup S_b)\setminus S_c=U$.

(iv) The \emph{set-graceful-difference constraint} $(S_a\setminus S_b)\setminus S_c=U$.
\end{prop}

\begin{defn} \label{defn:111111}
$^*$ Let $S_{set}$ be a set of sets. Suppose that a graph $G$ admits a total set-coloring $F:V(G)\cup E(G)\rightarrow S_{set}$, such that $F(u)=S_{u}$, $F(v)=S_{v}$, and $F(uv)=S_{uv}$ for each edge $uv\in E(G)$.

(i) If there is a fixed set $U$, such that each edge $uv\in E(G)$ holds the \emph{set-edge-magic constraint} $S_{u}\cup S_{v}\cup S_{uv}=U$ true, we say $F$ \emph{set-edge-magic total set-coloring}.

(ii) If there is a fixed set $U$, such that each edge $uv\in E(G)$ holds one of \emph{set-edge-difference constraints} $S_{uv}\cup (S_{u}\setminus S_{v})=U$ and $S_{uv}\cup (S_{v}\setminus S_{u})=U$ true, we say $F$ \emph{set-edge-difference total set-coloring}.

(iii) If there is a fixed set $U$, such that each edge $uv\in E(G)$ holds one of \emph{set-felicitous-difference constraints} $(S_{u}\cup S_{v})\setminus S_{uv}=U$ and $S_{uv}\setminus (S_{v}\cup S_{u})=U$ true, we say $F$ \emph{set-felicitous-difference total set-coloring}.

(iv) If there is a fixed set $U$, such that each edge $uv\in E(G)$ holds one of \emph{set-graceful-difference constraints} $(S_{u}\setminus S_{v})\setminus S_{uv}=U$, $(S_{v}\setminus S_{u})\setminus S_{uv}=U$, $S_{uv}\setminus (S_{u}\setminus S_{v})=U$ and $S_{uv}\setminus (S_{v}\setminus S_{u})=U$ true, we say $F$ \emph{set-graceful-difference total set-coloring}.\qqed
\end{defn}

\begin{thm}\label{thm:666666}
$^*$ Each connected $(p,q)$-graph $G$ admits a \emph{proper total string-coloring}
$$f:V(G)\cup E(G)\rightarrow \{a_{i}b_{i}:~a_{i},b_{i}\in [1,\Delta(G)-k]\}
$$ with $k\leq \Delta(G)-\sqrt{3[1+\Delta(G)]}$ if $\Delta(G)\geq 6$.
\end{thm}
\begin{proof} As known, the chromatic index $\chi\,'(G)\leq \Delta(G)+1$ (Vadim G. Vizing, 1964) and the chromatic number $\chi(G)\leq \Delta(G)+1+K(G)$ (Bruce Reed, 1998). Let $M=\Delta(G)-k$, so we have $M(M-1)+M=M^2$ different number-based strings $a_{i}b_{i}$, and then we take $\Delta(G)+1$ different number-based strings $a_{i}b_{i}$ with $i\in [1,\Delta(G)+1]$ for coloring properly the edges of $G$ holding $\chi\,'(G)\leq \Delta(G)+1$ true, the remainder $N^*$ different number-based strings are enough for properly coloring the vertices of $G$, by $k\leq \Delta(G)-\sqrt{3[1+\Delta(G)]}$, we have
$$N^*=M^2-\Delta(G)-1\geq 2\Delta(G)+2\geq \Delta(G)+1+K(G)\geq \chi(G)
$$ where $K(G)$ is the number of vertices of a largest clique of the graph $G$, we are done.
\end{proof}

There are several set-colorings and set-labelings introduced in \cite{Yao-Ma-arXiv-2201-13354v1}.

\begin{defn}\label{defn:graceful-intersection}
\cite{Yao-Zhang-Sun-Mu-Sun-Wang-Wang-Ma-Su-Yang-Yang-Zhang-2018arXiv} A $(p,q)$-graph $G$ admits a set-labeling $F:V(G)\rightarrow [1,q]^2$~(resp. $[1,2q-1]^2)$, and induces an edge set-color $F(uv)=F(u)\cap F(v)$ for each edge $uv \in E(G)$. If we can select a \emph{representative} $a_{uv}\in F(uv)$ for each edge color set $F(uv)$ such that $\{a_{uv}:~uv\in E(G)\}=[1,q]$ (resp. $[1,2q-1]^o$), then $F$ is called a \emph{graceful-intersection (resp. an odd-graceful-intersection) total set-labeling} of $G$.\qqed
\end{defn}

\begin{thm}\label{thm:graceful-total-set-labelings}
\cite{Yao-Zhang-Sun-Mu-Sun-Wang-Wang-Ma-Su-Yang-Yang-Zhang-2018arXiv} Each tree $T$ admits a \emph{graceful-intersection (resp. an odd-graceful-intersection) total set-labeling}.
\end{thm}

We define a \emph{regular rainbow set-sequence} $\{R_k\}^{q}_1$ as: $R_k=[1,k]$ with $k\in [1,q]$, where $[1,1]=\{1\}$.

\begin{thm}\label{thm:rainbow-total-set-labelings}
\cite{Yao-Zhang-Sun-Mu-Sun-Wang-Wang-Ma-Su-Yang-Yang-Zhang-2018arXiv} Each tree $T$ of $q$ edges admits a \emph{regular rainbow intersection total set-labeling} based on a \emph{regular rainbow set-sequence} $\{[1,k]\}^{q}_{k=1}$.
\end{thm}
\begin{proof} Suppose $x$ is a leaf of a tree $T$ of $q$ edges, so $T-x$ is a tree of $(q-1)$ edges. Assume that $T-x$ admits a regular rainbow set-sequence $\{R_k\}^{q-1}_1$ total set-labeling $f$. Let $y$ be adjacent with $x$ in $T$. We define a labeling $g$ of $T$ in this way: $g(w)=f(w)$ for $w\in V(T)\setminus \{y,x\}$, $g(y)=R_{q+1}=[1,q+1]$ and $g(x)=R_q=[1,q]$. Therefore, we have $g(u_iv_j)=g(u_i)\cap g(v_j)=[1,i]\cap [1,j]$ for $u_iv_j\in E(T)\setminus \{xy\}$, and $g(xy)=g(x)\cap g(y)=[1,q]$, and $g(s)\neq g(t)$ for any pair of vertices $s$ and $t$. We claim that $g$ is a regular rainbow intersection total set-labeling of $T$ by the hypothesis of induction.
\end{proof}

\begin{rem}\label{rem:333333}
Each tree admits a regular odd-rainbow intersection total set-labeling based on a \emph{regular odd-rainbow set-sequence} $\{R_k\}^{q}_1$ defined as: $R_k=[1,2k-1]$ with $k\in [1,q]$, where $[1,1]=\{1\}$. Moreover, we can define a \emph{regular Fibonacci-rainbow set-sequence} $\{R_k\}^{q}_1$ by $R_1=[1,1]$, $R_2=[1,1]$, and $R_{k+1}=R_{k-1}\cup R_{k}$ with $k\in [2,q]$; or a $\tau$-term Fibonacci-rainbow set-sequence $\{\tau,R_i\}^{q}_1$ holds: $R_i=[1,a_i]$ with $a_i>1$ and $i\in [1,q]$, and $R_k=\sum ^{k-1}_{i=k-\tau}R_i$ with $k>\tau$ \cite{Ma-Wang-Wang-Yao-Theoretical-Computer-Science-2018}. It may be an interesting research on various rainbow set-sequences for non-tree graphs.\paralled
\end{rem}

\begin{example}\label{exa:8888888888}
Fig.\ref{fig:0-graceful-intersection} is for understanding Theorem \ref{thm:rainbow-total-set-labelings}, and we have a number-based string
$${
\begin{split}
s(T)=&\underline{1234}1112121231231234\underline{123456789}\underline{1234}\underline{123456789}123456789\underline{12345678910}\\
&12345678910\underline{123456789101112}123456789101112\underline{1234567891011}1234567891011\\
&123456789101112\underline{12345678910111213}12345678\underline{12345678}\underline{1234567}1234567\\
&123456\underline{123456}12345\underline{12345}.
\end{split}}$$ made by the totally-colored tree shown in Fig.\ref{fig:0-graceful-intersection} (b). \qqed
\end{example}

\begin{figure}[h]
\centering
\includegraphics[width=16cm]{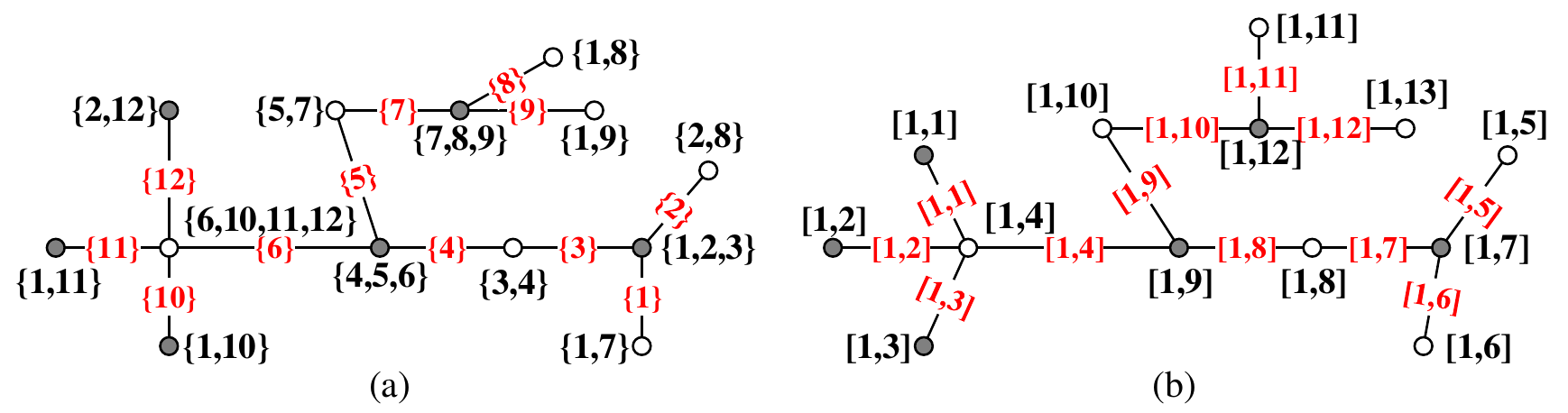}\\
\caption{\label{fig:0-graceful-intersection}{\small For illustrating Theorem \ref{thm:graceful-total-set-labelings} and Theorem \ref{thm:rainbow-total-set-labelings}: Left tree (a) admits a graceful-intersection total set-labeling, and Right tree (b) admits a regular rainbow intersection total set-labeling, cited from \cite{Yao-Zhang-Sun-Mu-Sun-Wang-Wang-Ma-Su-Yang-Yang-Zhang-2018arXiv}.}}
\end{figure}

\subsection{Topcode-matrices with non-number elements}

\subsubsection{Topcode-matrices with the elements of strings, vectors and sets}

\begin{defn} \label{defn:stringic-topcode-matrix}
\cite{Yao-Su-Ma-Wang-Yang-arXiv-2202-03993v1} Let $S_{x,i},S_{e,j}$ and $S_{y,k}$ be number-based strings with $1\leq i,j,k\leq q$, and put them into a set $G_{string}$. A \emph{string-type Topcode-matrix} $S_{code}(G_{string})$ is defined by
\begin{equation}\label{eqa:stringic-topcode-matrix}
\centering
{
\begin{split}
S_{code}(G_{string})= \left(
\begin{array}{ccccc}
S_{x,1} & S_{x,2} & \cdots & S_{x,q}\\
S_{e,1} & S_{e,2} & \cdots & S_{e,q}\\
S_{y,1} & S_{y,2} & \cdots & S_{y,q}
\end{array}
\right)=
\left(\begin{array}{c}
S_X\\
S_E\\
S_Y
\end{array} \right)=(S_X,~S_E,~S_Y)^{T}
\end{split}}
\end{equation} and $S_{x,i}$ and $S_{y,i}$ are called the \emph{ends} of $S_{e,i}$, and we call
$$S_X=(S_{x,1}, S_{x,2},\dots ,S_{x,q}),~S_E=(S_{e,1}, S_{e,2},\dots,S_{e,q}),~S_Y=(S_{y,1}, S_{y,2},\dots ,S_{y,q})
$$ \emph{string vectors}. Moreover the string-type Topcode-matrix $S_{code}$ is \emph{constraint valued} if there is a group of constraints $F_1,F_2,\dots ,F_n$ such that $S_{e,i}=F_k\langle S_{x,i},S_{y,i}\rangle $ for some $k\in [1,n]$ and each $i\in [1,q]$.\qqed
\end{defn}

\begin{rem}\label{rem:333333}
About the string-type Topcode-matrix $S_{code}=(S_X,~S_E,~S_Y)^{T}$ defined in Definition \ref{defn:stringic-topcode-matrix}, we have

(i) Each \emph{graph-based string} $D_k(3q)=H_{k,1}H_{k,2}\cdots H_{k,3q}$ with $k\in [1,(3q)!]$ generated from a graphic Topcode-matrix $G_{code}$ defined in Definition \ref{defn:graphic-topcode-matrix} holds each colored graph $H_{k,i}\in \{H_{x,i},H_{e,i}$, $H_{y,i}:i\in [1,q]\}$ and $H_{k,i}\neq H_{k,j}$ for $i\neq j$. Since the Topcode-matrix $T_{code}(H_{k,i})$ of each $H_{k,i}$ produces a number-based string $s(n_{k,i})$ with length $n_{k,i}$, then we get a \emph{number-based hyper-string} $D_k=s(n_{k,1})s(n_{k,2})\cdots s(n_{k,3q})$ with length $L_{\textrm{ength}}(D_k)=\sum ^{3q}_{i=1} n_{k,i}$ for $k\in [1,(3q)!]$.

(ii) We have compressed a number-based hyper-string $D_k$ with longer bytes into a graph-based string $D_k(3q)$ with smaller bytes, more or less like the effect of Hash function.

(iii) The elements of a string-type Topcode-matrix $S_{code}$ defined in Definition \ref{defn:stringic-topcode-matrix} may belong to an \emph{every-zero string group} $\{S(G);\oplus \ominus\}$.\paralled
\end{rem}

By Definition \ref{defn:n-di-vector-colorings-definition}, a $(p,q)$-graph $G$ admits a \emph{$W$-constraint total vector-coloring} $\overrightarrow{\pi} : V(G)\cup E(G)\rightarrow \textbf{\textrm{V}}_{ector}(n)$, such that each edge $u_kv_k\in E(G)=\{u_kv_k:k\in [1,q]\}$ holds
\begin{equation}\label{eqa:33333333333}
\overrightarrow{\pi}(u_k)=(a_{k,1},a_{k,2},\dots ,a_{k,n}),~\overrightarrow{\pi}(v_k)=(b_{k,1},b_{k,2},\dots ,b_{k,n}),~\overrightarrow{\pi}(u_kv_k)=(c_{k,1},c_{k,2},\dots ,c_{k,n})
\end{equation} subject to the $W$-constraint $\overrightarrow{\pi}(u_kv_k)=W\langle \overrightarrow{\pi}(u_k),\overrightarrow{\pi}(v_k)\rangle$. We define the \emph{vector-type Topcode-matrix} of the $(p,q)$-graph $G$ as
\begin{equation}\label{eqa:vector-type-topcode-matrix}
\centering
{
\begin{split}
T_{code}(G,\overrightarrow{\pi})= \left(
\begin{array}{ccccc}
\overrightarrow{\pi}(u_1) & \overrightarrow{\pi}(u_2) & \cdots & \overrightarrow{\pi}(u_q)\\
\overrightarrow{\pi}(u_1v_1) & \overrightarrow{\pi}(u_2v_2) & \cdots & \overrightarrow{\pi}(u_qv_q)\\
\overrightarrow{\pi}(v_1) & \overrightarrow{\pi}(v_2) & \cdots & \overrightarrow{\pi}(v_q)
\end{array}
\right)=
\left(\begin{array}{c}
\overrightarrow{X}\\
\overrightarrow{E}\\
\overrightarrow{Y}
\end{array} \right)=(\overrightarrow{X},\overrightarrow{E},\overrightarrow{Y})^{T}
\end{split}}
\end{equation} with three vector-vectors $\overrightarrow{X}=(\overrightarrow{\pi}(u_1),\overrightarrow{\pi}(u_2),\cdots ,\overrightarrow{\pi}(u_q))$, $\overrightarrow{E}=(\overrightarrow{\pi}(u_1v_1), \overrightarrow{\pi}(u_2v_2), \cdots ,\overrightarrow{\pi}(u_qv_q))$, and $\overrightarrow{Y}=(\overrightarrow{\pi}(v_1),\overrightarrow{\pi}(v_2),\cdots ,\overrightarrow{\pi}(v_q))$.

By Definition \ref{defn:n-di-set-colorings-definition}, a $(p,q)$-graph $G$ admits a \emph{$\{W_i\}^A_{i=1}$-constraint total set-coloring} $\psi : V(G)\cup E(G)\rightarrow \textbf{\textrm{S}}_{et}(\leq n)$, such that each edge $u_kv_k\in E(G)=\{u_kv_k:k\in [1,q]\}$ holds
\begin{equation}\label{eqa:444444444444}
{
\begin{split}
&\psi(u_k)=\{a_{k,1},a_{k,2},\dots ,a_{k,m_r}\},~\psi(v_k)=\{b_{k,1},b_{k,2},\dots ,b_{k,m_s}\},\\
&\psi(u_kv_k)=\{c_{k,1},c_{k,2},\dots ,c_{k,m_t}\}
\end{split}}
\end{equation} subject to the $W_i$-constraint $W_i\langle \psi(u_k),\psi(u_kv_k),\psi(v_k)\rangle=0$ for some $i\in [1,A]$. We define the \emph{set-type Topcode-matrix} of the $(p,q)$-graph $G$ as follows
\begin{equation}\label{eqa:set-type-topcode-matrix}
\centering
{
\begin{split}
T_{code}(G,\psi)= \left(
\begin{array}{ccccc}
\psi(u_1) & \psi(u_2) & \cdots & \psi(u_q)\\
\psi(u_1v_1) & \psi(u_2v_2) & \cdots & \psi(u_qv_q)\\
\psi(v_1) & \psi(v_2) & \cdots & \psi(v_q)
\end{array}
\right)=
\left(\begin{array}{c}
X_{set}\\
E_{set}\\
Y_{set}
\end{array} \right)=(X_{set},E_{set},Y_{set})^{T}
\end{split}}
\end{equation} with three set-vectors $X_{set}=(\psi(u_1), \psi(u_2), \cdots , \psi(u_q))$, $E_{set}=(\psi(u_1v_1), \psi(u_2v_2), \cdots , \psi(u_qv_q))$, and $Y_{set}=(\psi(v_1), \psi(v_2), \cdots , \psi(v_q))$.

\subsubsection{Non-number-type Topcode-matrices with parameters}

For a set $S\,^i_X=\{a\,^i_1,a\,^i_2,\dots , a\,^i_{s(i)}\}$ with integer $s(i)\geq 1$ and $i\in [1,n]$ and integers $k,d\geq 0$, we define the following two operations:
\begin{equation}\label{eqa:set-type-parameterized-operation11}
{
\begin{split}
&d\cdot S\,^i_X =\{d\cdot a\,^i_1,d\cdot a\,^i_2,\dots , d\cdot a\,^i_{s(i)}\}\\
&k\cdot I^*+d\cdot S\,^i_X=\{k+d\cdot a\,^i_1,~k+d\cdot a\,^i_2,~\dots ,~ k+d\cdot a\,^i_{s(i)}\}
\end{split}}
\end{equation}
where $I^*=(1,1,\dots ,1)$ is a unite vector.
\begin{defn} \label{defn:normai-kd-type-set-coloring}
\cite{Bing-Yao-arXiv:2207-03381} \textbf{The set-type parameterized Topcode-matrix}. A connected bipartite $(p,q)$-graph $G$ has its own vertex set bipartition $V(G)=X\cup Y$ and admits a $W$-constraint set-coloring $F$, then $G$ has its own set-type Topcode-matrix $T_{code}(G,F)=(X_S,E_S,Y_S)^T$ with $X_S\cap Y_S=\emptyset$, where each element in three \emph{set vectors} $X_S,E_S$ and $Y_S$ is a \emph{set}. According to Eq.(\ref{eqa:set-type-parameterized-operation11}), we get a \emph{set-type parameterized Topcode-matrix}
\begin{equation}\label{eqa:66666666666}
\centering
{
\begin{split}
P^{set}_{ara}(G,\Phi)=k\cdot I\,^0+d\cdot T_{code}(G,F)=(d\cdot X_S,~k\cdot I^*+d\cdot E_S,~k\cdot I^*+d\cdot Y_S)^T
\end{split}}
\end{equation}which defines a \emph{$(k,d)$-type set-coloring} $\Phi$ of the connected bipartite $(p,q)$-graph $G$.\qqed
\end{defn}

By a string-type Topcode-matrix defined in Definition \ref{defn:stringic-topcode-matrix}, we have a \emph{string-type parameterized Topcode-matrix}
\begin{equation}\label{eqa:66666666666}
\centering
{
\begin{split}
P^{string}_{ara}(G_{string},\Upsilon)=k\cdot I\,^0+d\cdot S_{code}(G_{string})=(d\cdot S_X,~k\cdot I^*+d\cdot S_E,~k\cdot I^*+d\cdot S_Y)^T
\end{split}}
\end{equation}

By a vector-type Topcode-matrix defined in Eq.(\ref{eqa:vector-type-topcode-matrix}), we have a \emph{vector-type parameterized Topcode-matrix}
\begin{equation}\label{eqa:66666666666}
\centering
{
\begin{split}
P^{vector}_{ara}(G,\Psi)=k\cdot I\,^0+d\cdot T_{code}(G,\overrightarrow{\pi})=(d\cdot \overrightarrow{X},~k\cdot I^*+d\cdot \overrightarrow{E},~k\cdot I^*+d\cdot \overrightarrow{Y})^T
\end{split}}
\end{equation}

\subsection{Every-zero number-based string groups}

We introduce the methods for producing every-zero number-based string groups in this subsection.

\subsubsection{GROUP-compound algorithm}

\textbf{Way-1.} From graphic-group $\{F_m(G,f);\oplus \ominus\}$ to Topcode-matrix group $\{F_m(T_{code}(G),f);\oplus \ominus\}$, and then from Topcode-matrix group $\{F_m(T_{code}(G),f);\oplus \ominus\}$ to string-group $\{F_m(S_r);\oplus \ominus\}$.

\vskip 0.2cm

\noindent \textbf{$^*$ GROUP-compound algorithm for the graphic-to-Topcode-matrix-to-string groups.}

\textbf{Step 1.1.} An \emph{every-zero graphic group} $\{F_m(G,f);\oplus \ominus\}$ is based on a $(p,q)$-graph $G$ and a graph set $F_m(G,f)=\{G_i:i\in [1,m]\}$, where $G=G_1$ admits a $W$-constraint total coloring $f_1=f$ defined on a number set, such that each graph $G_i$ with $i\in [1,m]$ holds $G_i\cong G_1$ and admits a $W$-constraint coloring $f_i$ induced by $f_i(u)=f_1(u)+i$ $(\bmod~m)$ for $u\in V(G)=V(G_i)$, as well as $G_{m+j}~(\bmod~m)=G_j$.

Notice that each $G_i$ has its own Topcode-matrix $T_{code}(G_i)=(X_i,E_i,Y_i)^T$ (refer to Definition \ref{defn:topcode-matrix-definition}), where two \emph{vertex vectors} $X_i=(f_i(x_1),f_i(x_2),\dots ,f_i(x_q))$ and $Y_i=(f_i(y_1),f_i(y_2),\dots ,f_i(y_q))$, and an \emph{edge vector} $E_i=(f_i(x_1y_1),f_i(x_2y_2),\dots ,f_i(x_qy_q))$, each edge $x_ty_t$ for $t\in [1,q]$ is colored by $f_i(x_ty_t)=\varphi(f_i(x_t),f_i(y_t))$ based on the $W$-constraint total coloring $f$ of the $(p,q)$-graph $G$.

\textbf{Step 1.2.} Each Topcode-matrix $T_{code}(G_i)$ distributes us $(3q)!$ different number-based strings $s^i_r=\beta^i_{r,1}\beta^i_{r,2}\cdots \beta^i_{r,3q}$ for $r\in [1,(3q)!]$, so we have the algorithm-$r$ with $r\in [1,(3q)!]$ in order to generate number-based strings $s^i_r$ from the Topcode-matrix $T_{code}(G_i)$ with $i\in [1,m]$. By the properties of a $W$-constraint total coloring $f$, number-based strings $s^i_r\in N^+_{\textrm{string}}(n)$.

\textbf{Step 1.3.} Since
\begin{equation}\label{eqa:555555}
G_i[\oplus \ominus_k] G_j:=G_i[\oplus ]G_j[\ominus ]G_k=G_{\lambda}\in \{F_m(G,f);\oplus \ominus\}
\end{equation} with the index $\lambda=i+j-k~(\bmod~m)$ for any preappointed \emph{zero} $G_k\in \{F_m(G,f);\oplus \ominus\}$, so we have an \emph{every-zero Topcode-matrix group} $\{F_m(T_{code}(G),f);\oplus \ominus\}$ defined as
\begin{equation}\label{eqa:555555}
{
\begin{split}
T_{code}(G_i)[\oplus \ominus_k] T_{code}(G_j):=&T_{code}(G_i)[\oplus ] T_{code}(G_j)[\ominus ]T_{code}(G_k)\\
=&T_{code}(G_{\lambda})\in \{F_m(T_{code}(G),f);\oplus \ominus\}
\end{split}}
\end{equation} with the index $\lambda=i+j-k~(\bmod~m)$ for any preappointed \emph{zero} $T_{code}(G_k)\in \{F_m(T_{code}(G),f);\oplus \ominus\}$,

\textbf{Step 1.4.} An algorithm-$r$ produces a number-based string $s^i_r=\beta^i_{r,1}\beta^i_{r,2}\cdots \beta^i_{r,3q}$ for a fixed $r\in [1,(3q)!]$, so we have $\beta^i_{r,j}=f_i(w^i_{r,j})$ with $j\in [1,3q]$ and $i\in [1,m]$, where $w^i_{r,1},w^i_{r,2},\dots ,w^i_{r,3q}$ is a permutation of elements of the vertex vectors $X_i$ and $Y_i$ and the edge vector $E_i$ of Topcode-matrix $T_{code}(G_i)$ under the algorithm-$r$.

The number-based string set $S_r=\{s^i_r:i\in [1,m]\}$ for a fixed $r\in [1,(3q)!]$ and the operation ``$\oplus \ominus$'' form an \emph{every-zero number-based string group} $\{F_m(S_r);\oplus \ominus\}$ defined as
\begin{equation}\label{eqa:number-based-string-authentication}
s^t_r[\oplus \ominus_k]s^s_r:=s^t_r[\oplus ]s^s_r[\ominus]s^k_r=s^{\lambda}_r\in \{F_m(s^i_r);\oplus \ominus\}
\end{equation} with the index $\lambda=t+s-k~(\bmod~m)$ for any preappointed \emph{zero} $s^k_r\in \{F_m(S_r);\oplus \ominus\}$.

\vskip 0.2cm
So, we can use every-zero number-based string groups $\{F_m(S_r);\oplus \ominus\}$ with $r\in [1,(3q)!]$ to encrypt networks, for each preappointed \emph{zero} $s^k_r\in \{F_m(S_r);\oplus \ominus\}$, we have a \emph{public-key string} $s^t_r$ and a \emph{private-key string} $s^s_r$ and a \emph{number-based string authentication} $s^{\lambda}_r$ defined in Eq.(\ref{eqa:number-based-string-authentication}).

\vskip 0.2cm

\textbf{The computational complexity of the GROUP-compound algorithm for the graphic-to-Topcode-matrix-to-string groups.} There are $m$ preappointed \emph{zeros} $G_k\in F_m(G,f)=\{G_i:i\in [1,m]\}$ to set up $m$ every-zero graphic groups $\{F_m(G,f);\oplus \ominus\}$, such that each every-zero graphic group induces an every-zero Topcode-matrix group $\{F_m(T_{code}(G),f);\oplus \ominus\}$, which induces $(3q)!$ different every-zero number-based string groups $\{F_m(S_r);\oplus \ominus\}$ for $r\in [1,(3q)!]$. Finally, we get $m\cdot (3q)!$ different every-zero number-based string groups based on the set $F_m(G,f)$, in total. Moreover, assume that this $(p,q)$-graph $G$ admits $n_{color}(G)$ colorings like the coloring $f$ appeared in the set $F_m(G,f)$, then the $(p,q)$-graph $G$ distributes us $n_{color}(G) \cdot m\cdot (3q)!$ different every-zero number-based string groups. Notice that if the $(p,q)$-graph $G$ admits another coloring $g$, both $f$ and $g$ do not obey the same $W$-constraint, then
$$
F_m(G,f)=\{G_i:i\in [1,m]\}\neq \{G_i:i\in [1,r]\}=F_r(G,g)
$$ also, $\{F_m(G,f);\oplus \ominus\}\neq \{F_r(G,g);\oplus \ominus\}$, in general.

\vskip 0.4cm

\textbf{Way-2.} String groups made by any number-based strings $N^+_{\textrm{string}}(n)\subset N_{\textrm{string}}(n)$, where
$$
N_{\textrm{string}}(n)=\{s=c_1c_2\cdots c_n:c_i\in [0,9]\},~N^+_{\textrm{string}}(n)=\{s\,'=a_1a_2\cdots a_n:a_i\in [0,9]\textrm{ and }a_1\geq 1\}
$$

\begin{defn} \label{defn:any-number-based-string-groups}
$^*$ An \emph{every-zero number-based string group} $\{F_M(S^*);\oplus \ominus\}$ is made by taking arbitrarily a number-based string $s\in N^+_{\textrm{string}}(n)\subset N_{\textrm{string}}(n)$, where $s=\beta_1\beta_2\cdots \beta_n$ having each $d_{i,1}\geq 1$ in each $[0,9]$-string $\beta_i=d_{i,1}d_{i,2}\cdots d_{i,m_i}$ with $d_{i,j}\in [0,9]$ and $m_i\geq 1$ for $i\in [1,n]$. Let $s=s(1)$, new number-based strings are defined as follows
\begin{equation}\label{eqa:555555}
s(t)=(\beta_1+t-1)(\beta_2+t-1)\cdots (\beta_n+t-1)~(\bmod~M),t\in [1,M]
\end{equation} such that $s(M+t)~(\bmod~M)=s(t)$. For a preappointed \emph{zero} $s(k)\in S^*=\{s(t):t\in [1,M]\}$, since
\begin{equation}\label{eqa:any-number-based-string-groups}
(\beta_p+i-1)+(\beta_p+j-1)-(\beta_p+k-1)=\beta_p+\lambda-1\in S^*,~p\in [1,n]
\end{equation}
with the index $\lambda=i+j-k~(\bmod~M)$, then it defines an Abelian additive operation
\begin{equation}\label{eqa:555555}
s(i)[\oplus \ominus_k] s(j):=s(i)[\oplus] s(j)[\ominus]s(k)=s(\lambda)\in S^*,~\lambda=i+j-k~(\bmod~M)
\end{equation} on number-based strings $s(i),s(j)\in S^*$.\qqed
\end{defn}

We show the following facts about an every-zero number-based string group $\{F_M(S^*);\oplus \ominus\}$:

(1) \textbf{Zero}. Clearly, each string $s(t)\in \{F_M(S^*);\oplus \ominus\}$ can be as the \emph{zero}.

(2) \textbf{Inverse}. Since $s(i)[\oplus] s(2k-i)[\ominus]s(k)=s(k)$, each string $s(t)\in \{F_M(S^*);\oplus \ominus\}$ has its own \emph{inverse} $s(2k-i)$.

(3) \textbf{Uniqueness and Closure}. If $s(i)[\oplus] s(j)[\ominus]s(k)=s(\lambda)$ and $s(i)[\oplus] s(j)[\ominus]s(k)=s(\mu)$, we have the indices $\lambda=i+j-k~(\bmod~M)$ and $\mu=i+j-k~(\bmod~M)$, that is $\lambda=\mu$.

(4) \textbf{Associative law}. We have $s(i)[\oplus \ominus_k] s(j)=s(j)[\oplus \ominus_k] s(i)$.

(5) \textbf{Commutative law}. Notice that $\big (s(i)[\oplus \ominus_k] s(j)\big )[\oplus \ominus_k] s(m)=s(i)[\oplus \ominus_k] \big (s(j)[\oplus \ominus_k] s(m)\big )$.

\subsubsection{Number-based sub-string groups}

\begin{defn} \label{defn:2-level-n-rank-string-sets}
$^*$ A \emph{2-level $n$-rank string-set} of strings is denoted as $N_{string}(n)$, such that each string $\alpha\in N_{string}(n)$ is formed as $\alpha=\beta_1\beta_2\cdots \beta_n$ with each $[0,9]$-string $\beta_i=d_{i,1}d_{i,2}\cdots d_{i,m_i}$ holding $d_{i,j}\in [0,9]$ and $m_i\geq 1$ for $i\in [1,n]$. $N^+_{string}(n)$ is a subset of $N_{string}(n)$, such that each string $\alpha=\beta_1\beta_2\cdots \beta_n$ having each $[0,9]$-string $\beta_i=d_{i,1}d_{i,2}\cdots d_{i,m_i}$ for $d_{i,j}\in [1,9]$ and $m_i\geq 1$ for $i\in [1,n]$.

If $m_i=n\geq 2$ for $i\in [1,n]$, then we get $\beta_i=d_{i,1}d_{i,2}\cdots d_{i,n}$ for $i\in [1,n]$, we call $\alpha$ \emph{2-level uniformly $n$-rank number-based string}. \qqed
\end{defn}

Since each number-based string $\alpha\in N^+_{\textrm{string}}(n)\subset N_{\textrm{string}}(n)$ (refer to Definition \ref{defn:2-level-n-rank-string-sets}) is expressed as $\alpha=\beta_1\beta_2\cdots \beta_n$ having each integer $d_{i,1}\geq 1$ in each $[0,9]$-string $\beta_i=d_{i,1}d_{i,2}\cdots d_{i,m_i}$ with $d_{i,j}\in [0,9]$ and $m_i\geq 1$ for $i\in [1,n]$. We set new number-based strings
$$\alpha_i(t)=\beta_1\beta_2\cdots \beta_{i-1}(\beta_i+t-1)\beta_{i+1}\cdots \beta_n~(\bmod~M),~t\in [1,M]
$$ and get a number-based string set $\Phi=\{\alpha_i(t):t\in [1,M]\}$; and for any preappointed \emph{zero} $\alpha_i(r)\in \Phi$, we get
$$(\beta_i+s-1)+(\beta_i+t-1)-(\beta_i+r-1)=\beta_i+\lambda-1\in \Phi$$
with the index $\lambda=s+t-r~(\bmod~M)$ and $t\in [1,M]$, such that $\alpha_i(M+t)~(\bmod~M)=\alpha_i(t)$. Thereby, we get an Abelian additive operation
\begin{equation}\label{eqa:555555}
\alpha_i(s)[\oplus \ominus_r] \alpha_i(t):=\alpha_i(s)[\oplus] \alpha_i(t)[\ominus]\alpha_i(r)=\alpha_i(\lambda)\in \Phi,~\lambda=s+t-r~(\bmod~M)
\end{equation} which enables us to defined an \emph{every-zero number-based $1$-sub-string group} $\{F_M(\Phi);\oplus \ominus\}$.

\begin{defn} \label{defn:m-sub-string-groups}
$^*$ An \emph{every-zero number-based $m$-sub-string group} $\{F_M(\Phi_{[1,m]});\oplus \ominus\}$ is obtained as: Take a group of numbers $\beta_{j_1}, \beta_{j_2}, \dots,\beta_{j_m}$ form the number-based string $\alpha$ for $2\leq m\leq n-1$, and define new number-based strings as follows
\begin{equation}\label{eqa:555555}
{
\begin{split}
\alpha_{[1,m]}(t)=&\beta_1\beta_2\cdots \beta_{j_1-1}(\beta_{j_1}+t-1)\beta_{j_1+1}\cdots \beta_{j_2-1}(\beta_{j_2}+t-1)\beta_{j_2+1}\cdots \\
&\beta_{j_m-1}(\beta_{j_m}+t-1)\beta_{j_m+1}\cdots \beta_n~(\bmod~M)
\end{split}}
\end{equation} with $\alpha_{[1,m]}(M+t)~(\bmod~M)=\alpha_{[1,m]}(t)$. The set $\Phi_{[1,m]}=\{\alpha_{[1,m]}(t):t\in [1,M]\}$ holds
\begin{equation}\label{eqa:555555}
{
\begin{split}
\alpha_{[1,m]}(s)[\oplus \ominus_r] \alpha_{[1,m]}(t):=&\alpha_{[1,m]}(s)[\oplus] \alpha_{[1,m]}(t)[\ominus]\alpha_{[1,m]}(r)\\
=&\alpha_{[1,m]}(\lambda)\in \Phi,~\lambda=s+t-r~(\bmod~M)
\end{split}}
\end{equation} for any preappointed \emph{zero} $\alpha_{[1,m]}(r)\in \Phi_{[1,m]}$.\qqed
\end{defn}

\begin{thm}\label{thm:666666}
$^*$ A connected graph $H$ admits a proper total coloring $F$ based on each every-zero number-based string group $\{F_m(S_r);\oplus \ominus\}$ for $r\in [1,(3q)!]$ (refer to the GROUP-compound algorithm for the graphic-to-Topcode-matrix-to-string groups) when $\Delta(H)+1\leq m$, that is $F:V(H)\cup E(H)\rightarrow \{F_m(S_r);\oplus \ominus\}$ such that $F(uv)=F(u)[\oplus \ominus_r] F(v)~(\bmod ~m)$ for a preappointed \emph{zero} $s^k_r\in \{F_m(S_r);\oplus \ominus\}$.
\end{thm}

\begin{thm}\label{thm:666666}
$^*$ A connected graph $H$ admits a proper total coloring $F$ based on an every-zero number-based string group $\{F_M(S^*);\oplus \ominus\}$ (refer to Definition \ref{defn:any-number-based-string-groups}) when $\Delta(H)+1\leq M$, that is $F:V(H)\cup E(H)\rightarrow \{F_M(S^*);\oplus \ominus\}$ such that $F(uv):=F(u)[\oplus \ominus_k] F(v)~(\bmod ~M)$ for a preappointed \emph{zero} $s(k)\in S^*$.
\end{thm}

\begin{thm}\label{thm:666666}
$^*$ A connected graph $H$ admits a proper total coloring $F$ based on an every-zero number-based $m$-sub-string group $\{F_M(\Phi_{[1,m]});\oplus \ominus\}$ (refer to Definition \ref{defn:m-sub-string-groups}) when $\Delta(H)+1\leq M$, that is $F:V(H)\cup E(H)\rightarrow \{F_M(\Phi_{[1,m]});\oplus \ominus\}$ such that $F(uv)=F(u)[\oplus \ominus_r] F(v)~(\bmod ~M)$ for a preappointed \emph{zero} $\alpha_{[1,m]}(r)\in \Phi_{[1,m]}$.
\end{thm}

\subsubsection{Compound number-based string groups}

Suppose that a connected $(p,q)$-graph $G$ admits $n$ colorings $f_1,f_2,\dots,f_n$ with $n\geq 2$, and there is an equivalent transformation $O_{i,j}$ for each pair of colorings $f_i$ and $f_j$ for $i\neq j$ and $1\leq i,j\leq n$.

We get $n$ every-zero graphic groups $\{F_{m_i}(G,f_i);\oplus \ominus\}$ with $i\in [1,n]$, where each $\{F_{m_i}(G,f_i);\oplus \ominus\}$ is defined on the graph set $F_{m_i}(G,f_i)=\{G_{i,r}:r\in [1,m_i]\}$, and each graph $G_{i,r}$ holds
\begin{equation}\label{eqa:all-graphs-is-one}
G_{i,r}\cong G
\end{equation} true, and admits a coloring defined by $f_{i,r}(w)=f_{i,1}(w)+r-1~(\bmod ~m_i)$ for $w\in V(G_{r,i})\cup E(G_{r,i})=V(G)\cup E(G)$ with $r\in [1,m_i]$, where $f_{i,1}=f_{i}$, and $m_1\leq m_2\leq \cdots \leq m_n$. Since the equivalent transformations hold

$O_{i,n}(G_{i,a})=G_{n,i_a}$ from $G_{i,a}\in \{F_{m_i}(G,f_i);\oplus \ominus\}$ to $G_{n,i_a}\in \{F_{m_n}(G,f_n);\oplus \ominus\}$,

$O_{j,n}(G_{j,b})=G_{n,j_b}$ from $G_{j,b}\in \{F_{m_j}(G,f_j);\oplus \ominus\}$ to $G_{n,j_b}\in \{F_{m_n}(G,f_n);\oplus \ominus\}$, and

$O_{k,n}(G_{k,c})=G_{n,k_c}$ from $G_{k,c}\in \{F_{m_k}(G,f_k);\oplus \ominus\}$ to $G_{n,k_c}\in \{F_{m_n}(G,f_n);\oplus \ominus\}$,\\
then we have the \emph{complex Abelian additive operation}
\begin{equation}\label{eqa:complex-Abelian-additive-operation}
{
\begin{split}
O_{i,n}(G_{i,a})[\oplus\ominus_{k_c}] O_{j,n}(G_{j,b}):=&O_{i,n}(G_{i,a})[\oplus] O_{j,n}(G_{j,b})[\ominus] O_{k,n}(G_{k,c})\\
=&G_{n,i_a}[\oplus] G_{n,j_b}[\ominus] G_{n,k_c}=G_{n,\lambda}
\end{split}}
\end{equation} with $G_{n,\lambda}\in \{F_{m_n}(G,f_n);\oplus \ominus\}$ as $\lambda=i_a+j_b-k_c~(\bmod ~m_n)$, where $G_{k,c}$ is a preappointed \emph{zero}, and the Abelian additive operation ``$O_{i,n}(G_{i,a})[\oplus \ominus_{k_c}] O_{j,n}(G_{j,b})$'' is defined by
\begin{equation}\label{eqa:555555complex-graphic-group}
f_{i,i_a}(w)+f_{i,j_b}(w)-f_{i,k_c}(w)=f_{i,\lambda}(w),~ w\in V(G)\cup E(G)
\end{equation} with $\lambda=i_a+j_b-k_c~(\bmod ~m_n)$ according to Eq.(\ref{eqa:all-graphs-is-one}). We get a \emph{complex graphic group} defined as follows
\begin{equation}\label{eqa:complex-graphic-group}
C_{om}(G[\oplus \ominus]\{f_{i}\}^n_{i=1})=\bigcup ^n_{i=1}\{F_{m_i}(G,f_i);\oplus \ominus\}
\end{equation} for any preappointed \emph{zero} $G_{k,c}$ under the complex Abelian additive operation.

\vskip 0.4cm

\textbf{Every-zero number-based string groups.} Since an every-zero graphic group $\{F_{m_i}(G,f_i)$; $\oplus \ominus\}$ is based on a colored graph set $F_{m_i}(G,f_i)=\{G_{i,r}:r\in [1,m_i]\}$, where each graph $G_{i,r}\in F_{m_i}(G,f_i)$ holds $G_{i,r}\cong G$ and admits a coloring $f_{i,r}$, so $G_{i,r}$ corresponds to a Topcode-matrix $T_{code}(G_{i,r})=(X_{i,r},~E_{i,r},~Y_{i,r})^T$ with

$X_{i,r}=(f_{i,r}(x_1),~f_{i,r}(x_2),\dots ,f_{i,r}(x_q))$, $E_{i,r}=(f_{i,r}(e_1),~f_{i,r}(e_2),\dots ,f_{i,r}(e_q))$ and

$Y_{i,r}=(f_{i,r}(y_1),~f_{i,r}(y_2),\dots ,f_{i,r}(y_q))$, $e_i=x_iy_i\in E(G_{i,r})=E(G)$.\\
We have an \emph{every-zero Topcode-matrix group}
$$\{T_{m_i}(T_{code}(G),f_i);\oplus \ominus\}=\{T_{code}(G_{i,r}):r\in [1,m_i]\}$$

There are $(3q)!$ different algorithms for producing $(3q)!$ different number-based strings for each Topcode-matrix $T_{code}(G_{i,r})$. For a fixed algorithm-$\gamma$ with $\gamma\in [1,(3q)!]$, each Topcode-matrix $T_{code}(G_{i,r})$ produces a number-based string $s(T_{code}(G_{i,r}),\gamma)$, then each every-zero graphic group $\{F_{m_i}(G,f_i);\oplus \ominus\}$ produces a set $S_{m_i}(T_{code}(G_{i,r}),\gamma)$ of $m_i$ number-based strings, where $$S_{m_i}(T_{code}(G_{i,r}),\gamma)=\{s(T_{code}(G_{i,r}),\gamma):r\in [1,m_i]\}$$

Notice that the operation ``$G_{i,r}[\oplus \ominus_{i,k}] G_{i,t}:=G_{i,r}[\oplus ] G_{i,t}[\ominus]G_{i,k}$'' for each every-zero graphic group $\{F_{m_i}(G,f_i);\oplus \ominus\}$ and any preappointed \emph{zero} $G_{i,k}\in \{F_{m_i}(G,f_i);\oplus \ominus\}$ is defined by
\begin{equation}\label{eqa:graphic-group-Topcode-matrix}
f_{i,r}(w)+f_{i,t}(w)-f_{i,k}(w)=f_{i,\lambda}(w)
\end{equation} and $G_{i,\lambda}\in \{F_{m_i}(G,f_i);\oplus \ominus\}$ with the index $\lambda=r+t-k~(\bmod~m_i)$ for $w\in V(G_{i,r})\cup E(G_{i,r})=V(G)\cup E(G)$ with $r\in [1,m_i]$. By the complex Abelian additive operation defined in Eq.(\ref{eqa:complex-Abelian-additive-operation}) and Eq.(\ref{eqa:graphic-group-Topcode-matrix}), we have defined the Abelian additive operation
\begin{equation}\label{eqa:555555}
s(T_{code}(G_{i,r}),\gamma)[\oplus] s(T_{code}(G_{i,t}),\gamma)[\ominus] s(T_{code}(G_{i,k}),\gamma)=s(T_{code}(G_{i,\lambda}),\gamma)
\end{equation} with the index $\lambda=i+t-k~(\bmod~m^*)$ based on the number-based string set $S_{m_i}(T_{code}(G_{i,r}),\gamma)$ and any preappointed \emph{zero} $s(T_{code}(G_{i,k}),\gamma)$ of $S_{m_i}(T_{code}(G_{i,r}),\gamma)$, and defined an \emph{every-zero number-based string group} $\{S_{m_i}(T_{code}(G_{i,r}),\gamma);\oplus \ominus\}$ for $\gamma\in [1,(3q)!]$.

\vskip 0.4cm

\textbf{Compound number-based string groups.} We have shown a connection between three every-zero graphic group, every-zero Topcode-matrix group and every-zero number-based string group as follows
\begin{equation}\label{eqa:compound-string-groups22}
\{F_{m_i}(G,f_i);\oplus \ominus\}\Rightarrow \{T_{m_i}(T_{code}(G),f_i);\oplus \ominus\}\Rightarrow \{S_{m_i}(T_{code}(G_{i,r}),\gamma);\oplus \ominus\}
\end{equation} By means of Eq.(\ref{eqa:complex-graphic-group}), each set
\begin{equation}\label{eqa:compound-string-groups22}
C_{om}(FTS[\oplus \ominus]\gamma)=\bigcup^n_{i=1} S_{m_i}(T_{code}(G_{i,r}),\gamma),~\gamma\in [1,(3q)!]
\end{equation} is called \emph{compound number-based string group}. Finally, we have $n\cdot (3q)!$ different compound number-based string groups in total.

\subsubsection{Partial every-zero compound number-based string groups}

\begin{defn} \label{defn:partialevery-zero-nbs-group}
$^*$ Suppose that $C_{nbs}(m)=\{s_{uper}(i):i\in [1,m]\}$ is a set of compound number-based strings, each compound number-based string $s_{uper}(i)=a_{i,1}a_{i,2}\cdots a_{i,n}$ with each $a_{i,j}=b_{i,j,1}b_{i,j,2}\cdots b_{i,j,m_{j}}$ is a non-negative integer with $m_j\geq 10$, such that $s_{uper}(i+1)$ holds
\begin{equation}\label{eqa:555555}
s_{uper}(i+1)=a_{i+1,1}a_{i+1,2}\cdots a_{i+1,n},~a_{i+1,j}=1+b_{i,j,1}b_{i,j,2}\cdots b_{i,j,m_{j}}~(\bmod ~m_{j}),~j\in [1,n]
\end{equation} For a fixed integer $j\in [1,n]$, we define a \emph{partial Abelian additive operation} for $C_{nbs}(m)$ as follows
\begin{equation}\label{eqa:partial-additive-operation111}
a_{r,j}+a_{s,j}-a_{k,j}=a_{\lambda,j},~\lambda=r+s-k~(\bmod ~m_{j})
\end{equation} for a preappointed \emph{zero} $s_{uper}(k)\in C_{nbs}(m)$, and denote this operation by
\begin{equation}\label{eqa:555555}
{
\begin{split}
\partial_j\{s_{uper}(r)[\oplus \ominus_k]s_{uper}(s)\}:=\partial_j\{s_{uper}(r)[\oplus ]s_{uper}(s)[\ominus]s_{uper}(k)\}=s_{uper}(\lambda)\in C_{nbs}(m)
\end{split}}
\end{equation} with the index $\lambda=r+s-k~(\bmod ~m_{j})$ for a preappointed \emph{zero} $s_{uper}(k)\in C_{nbs}(m)$, which distributes us a \emph{partial every-zero compound number-based group} $\{C_{nbs}(m);\partial_j [\oplus\ominus]\}$.\qqed
\end{defn}

\begin{rem}\label{rem:333333}
By Definition \ref{defn:partialevery-zero-nbs-group}, we can use the compound number-based string set $C_{nbs}(m)$ to encrypt a network $N(t)$ such that two nodes of each community $c_{ommunity}(j)$ of the network $N(t)$ is encrypted by the partial every-zero compound number-based group $\{C_{nbs}(m);\partial_j [\oplus\ominus]\}$. By Eq.(\ref{eqa:partial-additive-operation111}), the authentication $a_{\lambda,j}$ between two nodes $a_{r,j}$ and $a_{s,j}$ of the network $N(t)$ holds the index $\lambda=r+s-k~(\bmod ~m_{j})$ true for a preappointed \emph{zero} $a_{k,j}\in C_{nbs}(m)$.

\textbf{Application.} The authentication between two communities $c_{ommunity}(a)$ and $c_{ommunity}(b)$ is based on the partial Abelian additive operation $\partial_{(a,b)}$ to be the combination of $$\partial_a\{s_{uper}(r)[\oplus \ominus_{k_{(a,b)}}]s_{uper}(s)\},~ \partial_b\{s_{uper}(r)[\oplus \ominus_{k_{(a,b)}}]s_{uper}(s)\}
$$ defined as follows
\begin{equation}\label{eqa:partial-additive-operation222}
\begin{cases}
a_{r,j}+a_{s,j}-a_{k_{(a,b)},j}=a_{\eta,j},&\eta=r+s-k_{(a,b)}~(\bmod ~m_{a})\\
a_{r,j}+a_{s,j}-a_{k_{(a,b)},j}=a_{\mu,j},&\mu=r+s-k_{(a,b)}~(\bmod ~m_{b})
\end{cases}
\end{equation} for the preappointed \emph{zero} $a_{k_{(a,b)},j}\in C_{nbs}(m)$.\paralled
\end{rem}

\begin{defn} \label{defn:every-zero-abstract-group}
$^*$ Let $S_{thing}=\{t_1,t_2,\dots ,t_m\}$ be a set of $m$ particular things. If there is an operation ``$[\oplus \ominus_k]$'' defined by
\begin{equation}\label{eqa:general-group}
t_i~[\oplus \ominus_k] ~t_j:=t_i\oplus t_j\ominus t_k=t_{\lambda}\in S_{thing}
\end{equation} with the index $\lambda=i+j-k~(\bmod~m)$ for any two elements $t_i, t_j$ of $S_{thing}$ and any preappointed \emph{zero} $t_k\in S_{thing}$, so $S_{thing}$ is called \emph{every-zero thing-index group}, denoted as $\{F^+_m(S_{thing});\oplus \ominus\}$.\qqed
\end{defn}

\begin{prop}\label{prop:99999}
An every-zero thing-index group $\{F^+_m(S_{thing});\oplus \ominus\}$ holds the following laws:

(1) \textbf{Zero}. Each element $t_k\in S_{thing}$ is a zero under the operation ``$\oplus \ominus$''.

(2) \textbf{Inverse}. Each element $t_i\in S_{thing}$ has its own inverse $t_{i^{-1}}$ with $i^{-1}=2k-i$.

(3) \textbf{Uniqueness and Closure}. If two operations $t_i\oplus t_j\ominus t_k=t_{\lambda}$ and $t_i\oplus t_j\ominus t_k=t_{\mu}$, then $t_{\lambda}=t_{\mu}$ since the index $i+j-k~(\bmod~m)=\lambda=\mu=i+j-k~(\bmod~m)$.

(4) \textbf{Associative law}. The equation
$$(t_i~[\oplus \ominus_k] ~ t_j)~[\oplus \ominus_k] ~ t_s=t_i~[\oplus \ominus_k] ~ (t_j~[\oplus \ominus_k] ~ t_s)
$$ holds true based on the following facts:

$t_i~[\oplus \ominus_k] ~ t_j:=t_i\oplus t_j\ominus t_k=t_{\lambda}$ with the index $\lambda=i+j-k~(\bmod~m)$,

$t_{\lambda}~[\oplus \ominus_k] ~ t_s:=t_{\lambda}\oplus t_s\ominus t_k=t_{\eta}$ with the index $\eta=\lambda+s-k~(\bmod~m)$,

$t_j~[\oplus \ominus_k] ~ t_s:=t_j\oplus t_s\ominus t_k=t_{\mu}$ with the index $\mu=j+s-k~(\bmod~m)$, and

$t_i~[\oplus \ominus_k] ~ t_{\mu}:=t_i\oplus t_{\mu}\ominus t_k=t_{\tau}$ with the index $\tau=i+\mu-k~(\bmod~m)$.\\
Since the indices $\eta=\lambda+s-k~(\bmod~m)=i+j-k+s-k~(\bmod~m)$ and $\tau=i+\mu-k~(\bmod~m)=i+j+s-k-k~(\bmod~m)$, then the index $\eta=\tau~(\bmod~m)$.

(5) \textbf{Commutative law}. There is
$$t_i[\oplus ] t_j[\ominus] t_k:=t_i[\oplus \ominus_k] t_j=t_j[\oplus \ominus_k] t_i:=t_j[\oplus ]t_i[\ominus ]t_k
$$ since the index $\lambda=i+j-k~(\bmod~m)=j+i-k~(\bmod~m)$.
\end{prop}

\begin{figure}[h]
\centering
\includegraphics[width=16.4cm]{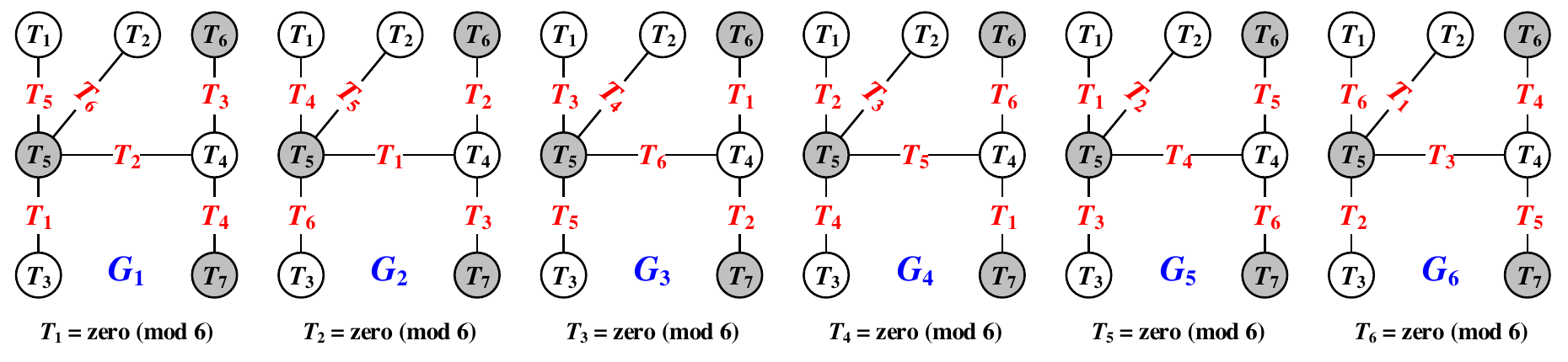}\\
\caption{\label{fig:graphic-group-operation11} {\small A thing-index graphic group $\{F^+_6(G);\oplus\ominus\}$ under the operation ``$G_i[\oplus \ominus_k] G_j$'' for understanding Definition \ref{defn:every-zero-abstract-group}.}}
\end{figure}


\section{Asymmetric topology encryption}

The authors in \cite{Yao-Su-Ma-Wang-Yang-arXiv-2202-03993v1} propose firstly \emph{asymmetric topology encryption} of topological encryption, and use complete graphs, maximal planar graphs and trees to make various topological signature authentications of asymmetric topology encryption; refer to Fig.\ref{fig:Topological-signature-11}. The ``topological structure'' in topological encryption is the natural ``topological signature'', and the ``mathematical constraint'' is the ``key generator'', which perfectly interprets the function of mathematics.

\begin{figure}[h]
\centering
\includegraphics[width=16.4cm]{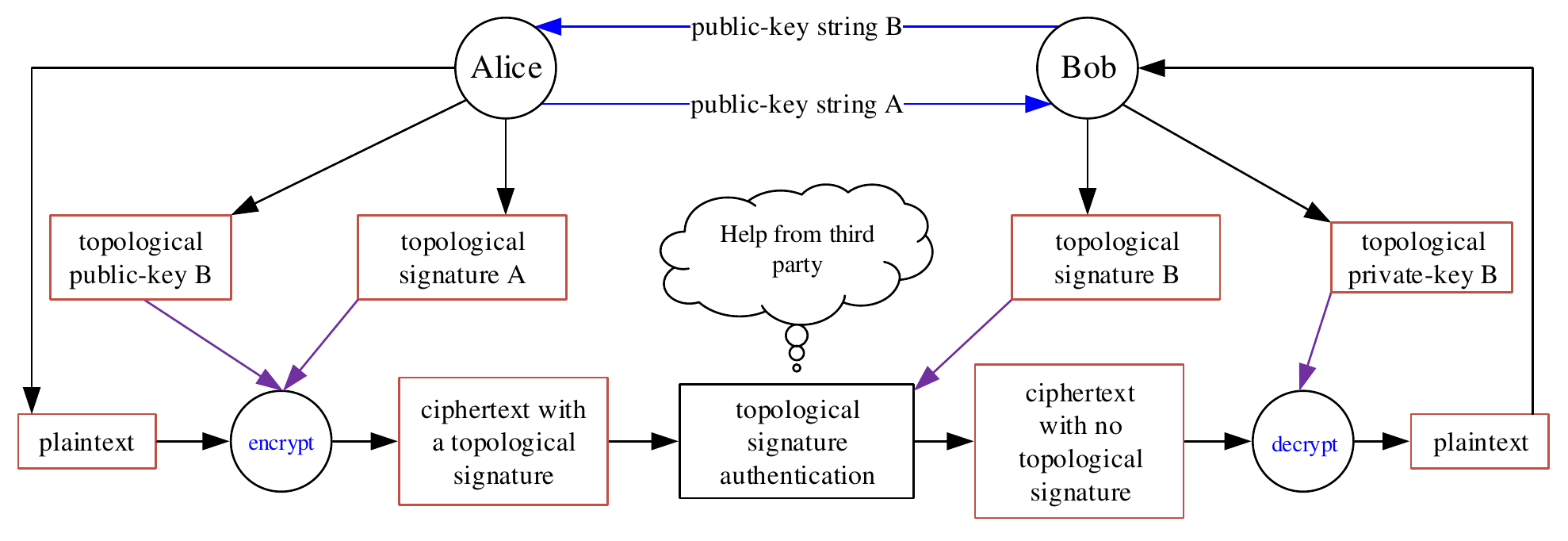}\\
\caption{\label{fig:Topological-signature-11}{\small A diagram for the asymmetric topology cryptography.}}
\end{figure}

\subsection{Topological signatures and Key-pairs}

\subsubsection{Key-pairs of topological signatures and strings}

\begin{defn} \label{defn:topological-signatures-key-graphs}
$^*$ A graph and its own \emph{dual} or \emph{complimentary} in graph theory form a pair of \emph{public-key graph} and \emph{private-key graph}. And \emph{topological signatures} are \emph{colored graphs} of topological coding, or pairs of colored public-key graphs and colored private-key graphs. \qqed
\end{defn}

Sometimes, public-key graphs and colored private-key graphs are called \emph{public-signatures} and \emph{private-signatures}, respectively. Chinese letters (Hanzis) are \emph{natural topology structures}, or \emph{topological fingerprints} in information security. There are dual-type, twin-type, image-type and matching-type labelings/colorings of graphs for making topological signatures.

Topological signatures have the basic functions: \emph{the identity authentication, generating various forms of strings}. The identity authentication is done between colored graphs, or various strings generated from colored graphs. Generating various forms of strings is done by various matrices of topological coding.

\vskip 0.4cm

We set up the center of topological key production, distribution, replacement and topological signature authentication (TKPDRA-center) for a local area network. We will use the following terminology and notation:
\begin{asparaenum}[\textbf{\textrm{Term-}}1.]
\item Alice's \emph{topological signature} $T_{sigA}=\langle G_{Apub},G_{Apri}\rangle $, where $G_{Apub}$ is Alice's \emph{public-key graph}, and $G_{Apri}$ is Alice's \emph{private-key graph}.
\item Bob's topological signature $T_{sigB}=\langle G_{Bpub},G_{Bpri}\rangle $ with Bob's public-key graph $G_{Bpub}$ and Bob's private-key graph $G_{Bpri}$.
\item Thereby, $A_{uth}\langle T_{sigA},T_{sigB}\rangle$ is a \emph{topological signature authentication} based on Alice's topological signature and Bob's topological signature.
\item Alice's \emph{topological string authentication} $S_{tringA}=A_{uth}\langle s_{Apub},s_{Apri}\rangle $, where $s_{Apub}$ is Alice's \emph{public-key string}, and $s_{Apri}$ is Alice's \emph{private-key string}.
\item Bob's topological string authentication $S_{tringB}=A_{uth}\langle s_{Bpub},s_{Bpri}\rangle $ with Bob's public-key string $s_{Bpub}$ and Bob's private-key string $s_{Bpri}$.
\item $A_{uth}\langle S_{tringA},S_{tringB}\rangle$ is a \emph{topological string authentication} based on Alice's topological string authentication and Bob's topological string authentication.
\end{asparaenum}

\vskip 0.2cm

For transforming public-key graphs and public-key strings between people in a network, we design the following plans:

\vskip 0.2cm

\noindent \textbf{Key-pair-Plan-I (Node-to-node)} is based on that individuals make their own Key-pair signatures and Key-pair strings.
\begin{asparaenum}[\textbf{\textrm{Send-I-}}1.]
\item Alice sends to Bob for gaining Bob's provisional public-key graph $G^*_{Bpub}$ and Bob's provisional public-key string $s^*_{Bpub}$.
\item Bob sends to Alice a package including Bob's provisional public-key graph $G^*_{Bpub}$ and Bob's provisional public-key string $s^*_{Bpub}$.
\item Alice uses Bob's provisional public-key graph $G^*_{Bpub}$ and Bob's provisional public-key string $s^*_{Bpub}$ to make an encrypted file $F_{Alice}$ containing Alice's public-key graph $G_{Apub}$ and Alice's public-key string $s_{Apub}$.
\item After received the encrypted file $F_{Alice}$ sent by Alice, Bob deletes his provisional public-key graph $G^*_{Bpub}$ and Bob's provisional public-key string $s^*_{Bpub}$, and makes his public-key graph $G_{Bpub}$ and public-key string $s_{Bpub}$.
\item Bob uses Alice's public-key graph $G_{Apub}$ and Alice's public-key string $s_{Apub}$ to encrypt a file containing his public-key graph $G_{Bpub}$ and public-key string $s_{Bpub}$, and then sends the encrypted file to Alice.
\end{asparaenum}

\vskip 0.4cm

\noindent \textbf{Key-pair-Plan-II (TKPDRA-center and common zeros)} is based on the topological signatures and Key-pair strings come from one graphic group in TKPDRA-center.
\begin{asparaenum}[\textbf{\textrm{Send-II-}}1.]
\item Alice requires TKPDRA-center to create the Key-pair signatures $G_{Apub}$ and $G_{Apri}$ and the Key-pair strings $s_{Apub}$ and $s_{Apri}$ for herself.
\item TKPDRA-center asks Alice to randomly select a group $G_{roup}$, and two elements $G_{Apub}$ and $G_{Apri}$ in this graphic group $G_{roup}$, such that $G_{Apub}$ is Alice's public-key graph and $G_{Apri}$ is Alice's private-key graph.
\item Alice receives her topological signature $T_{sigA}=\langle G_{Apub},G_{Apri}\rangle $ obtained by the operation $G_{Apub}[\oplus]G_{Apri}[\ominus]G_{zero}=T_{sigA}\in G_{roup}$, where $G_{zero}\in G_{roup}$ is a common \emph{zero}.
\item TKPDRA-center produces Alice's public-key string $s_{Apub}$ from her public-key graph and $G_{Apub}$ and Alice's private-key string $s_{Apri}$ from her private-key graph $G_{Apri}$.
\item Alice receives her Key-pair string authentication $S_{tringA}=A_{uth}\langle s_{Apub},s_{Apri}\rangle $ obtained by the operation $s_{Apub}[\oplus]s_{Apri}[\ominus]s_{zero}=S_{tringA}\in S_{roup}$, where the string group $S_{roup}$ is based on the graphic group $G_{roup}$, and $s_{zero}\in S_{roup}$ is a common \emph{zero}.
\end{asparaenum}

\vskip 0.4cm

\noindent \textbf{Key-pair-Plan-III (TKPDRA-center and personalized zeros)} is based on the topological signatures and Key-pair strings come from different graphic groups and different string groups of TKPDRA-center, and the custom-made \emph{zeros} for one person.
\begin{asparaenum}[\textbf{\textrm{Send-III-}}1.]
\item Alice requires TKPDRA-center to create the Key-pair signatures $G_{Apub}$ and $G_{Apri}$ and the Key-pair strings $s_{Apub}$ and $s_{Apri}$ for herself.
\item TKPDRA-center asks Alice to randomly select a graphic group $G_{roup}$ of TKPDRA-center, and two elements $G_{Apub}$ and $G_{Apri}$ in this graphic group $G_{roup}$.
\item TKPDRA-center makes a particular \emph{zero} $G^A_{zero}$ for Alice only.
\item Alice receives her topological signature $T_{sigA}=\langle G_{Apub},G_{Apri}\rangle $ based on the result of the operation $G_{Apub}[\oplus]G_{Apri}[\ominus]G^A_{zero}=T_{sigA}\in G_{roup}$.
\item TKPDRA-center produces randomly Alice's public-key string $s_{Apub}$ and Alice's private-key string $s_{Apri}$ from the string group $S^*_{roup}$ selected randomly by Alice.
\item Alice receives her Key-pair string authentication $S_{tringA}=A_{uth}\langle s_{Apub},s_{Apri}\rangle $ based on the result of the operation $s_{Apub}[\oplus]s_{Apri}[\ominus]s^A_{zero}=S_{tringA}\in S^*_{roup}$, where the \emph{zero} $s^A_{zero}$ is specially customized for Alice only.
\end{asparaenum}

\vskip 0.4cm

\noindent \textbf{Key-pair-Plan-IV (TKPDRA-center and personalized zeros)} is based on the topological signatures and Key-pair strings come from different graphic groups and different string groups of TKPDRA-center, and the custom-made \emph{zeros} for two persons.
\begin{asparaenum}[\textbf{\textrm{Send-IV-}}1.]
\item Alice selects randomly a graphic group $G_{roup}$ of TKPDRA-center, and two elements $G_{Apub}$ and $G_{Apri}$ of the graphic group $G_{roup}$; and Bob selects randomly two elements $G_{Bpub}$ and $G_{Bpri}$ of the graphic group $G_{roup}$ too.
\item TKPDRA-center makes a particular \emph{zero} $G^{A,B}_{zero}$ for Alice and Bob.
\item Alice receives her topological signature $T_{sigA}=\langle G_{Apub},G_{Apri}\rangle $ obtained by the operation $G_{Apub}[\oplus]G_{Apri}[\ominus]G^{A,B}_{zero}=T_{sigA}\in G_{roup}$.
\item Bob receives his topological signature $T_{sigB}=\langle G_{Bpub},G_{Bpri}\rangle $ obtained by the operation $G_{Bpub}[\oplus]G_{Bpri}[\ominus]G^{A,B}_{zero}=T_{sigB}\in G_{roup}$.
\item Alice selects randomly her public-key string $s_{Apub}$ and private-key string $s_{Apri}$ from the string group $S^*_{roup}$ selected randomly by both Alice and Bob.
\item Bob selects randomly his public-key string $s_{Bpub}$ and private-key string $s_{Bpri}$ from the string group $S^*_{roup}$.
\item Alice receives her Key-pair string authentication $S_{tringA}=A_{uth}\langle s_{Apub},s_{Apri}\rangle $ obtained by the operation $s_{Apub}[\oplus]s_{Apri}[\ominus]s^{A,B}_{zero}=S_{tringA}\in S^*_{roup}$, where the \emph{zero} $s^{A,B}_{zero}$ is specially customized for both Alice and Bob.
\item Bob receives his Key-pair string authentication $S_{tringB}=A_{uth}\langle s_{Bpub},s_{Bpri}\rangle $ made by the operation $s_{Bpub}[\oplus]s_{Bpri}[\ominus]s^{A,B}_{zero}=S_{tringB}\in S^*_{roup}$.
\end{asparaenum}

\subsubsection{Twin-type graphs for topological signatures}

A phenomenon about twin labelings was proposed and discussed in \cite{Wang-Xu-Yao-2017-Twin}, that is, the twin odd-graceful labelings are natural-inspired as keys and locks. In fact, each type of twin labelings can be considered as a matching. We have other twin labelings, such as image-labelings, inverse labelings, \emph{etc}. We view many examples for twin labelings, and want to discover that twin labelings have some properties like quantum entanglement \cite{Yao-Zhang-Sun-Mu-Sun-Wang-Wang-Ma-Su-Yang-Yang-Zhang-2018arXiv}.

Twin-type graphs are determined by twin-type colorings/labelings in general; refer to Definition \ref{defn:group-definition-twin-total-labelingss}, Definition \ref{defn:twin-odd-edge-w-labelings-coloringsn} and Theorem \ref{thm:odd-edge-W-type-kd-total-labelings}. Twin-graphic lattices are based on twin-type graphs \cite{Yao-Zhang-Yang-Wang-Odd-Edge-arXiv-02477}.

\begin{defn}\label{defn:22-twin-odd-graceful-labeling}
\cite{Wang-Xu-Yao-2017-Twin} For two connected $(p_i,q)$-graphs $G_i$ with $i=1,2$, if a vertex-coincided graph $G=G_1[\odot] G_2$ admits a vertex labeling $f$: $V(G)\rightarrow [0, 2q]$ such that

(i) $f$ is just an odd-graceful labeling of $G_1$, so
$$f(E(G_1))=\{f(uv)=|f(u)-f(v)|: uv\in E(G_1)\}=[1, 2q-1]^o
$$

(ii) $f(E(G_2))=\{f(uv)=|f(u)-f(v)|: uv\in E(G_2)\}=[1,2q-1]^o$; and

(iii) $|f(V(G_1))\cap f(V(G_2))|=k\geq 0$ and $f(V(G_1))\cup f(V(G_2))\subseteq [0, 2q]$.\\
Then $f$ is called a \emph{twin odd-graceful labeling} of the vertex-coincided graph $G$.\qqed
\end{defn}

\begin{thm}\label{thm:666666}
$^*$ A tree $T$ of $q$ edges admits a set-ordered odd-graceful labeling $f$, then $T$ admits another set-ordered odd-graceful labeling $h$, such that $f(E(T))=h(E(T))=[1,2q-1]^o$, $f(V(T))\cup h(V(T))=[0,2q]$ and $|f(V(T))\cap h(V(T))|\leq 1$, and $\langle f,h \rangle$ forms a twin odd-graceful labeling.
\end{thm}
\begin{proof}Let the vertex set $V(T)=X\cup Y$ with $X\cap Y=\emptyset$. By the hypothesis of the theorem, we have $\max\{f(x):x\in X\}=f(X)<f(Y)=\min\{f(y):y\in Y\}$, and moreover each $f(x)$ for $x\in X$ is even, each $f(x)$ for $x\in X$ is odd according to $f(E(T))=[1,2q-1]^o$. Because of $f(u)\neq f(w)$ for nay pair of vertices $u,w\in V(T)$, without loss of generality, we have
$$0=f(x_1)<f(x_2)<\cdots <f(x_s)<f(y_1)<f(y_2)<\cdots <f(y_t)=2q-1
$$ for $x_i\in X$ and $y_j\in Y$ with $q+1=s+t$.

We set a new labeling $h$ for the tree $T$ as: $h(w)=f(w)+1$ for $w\in V(T)$, so $h(uv)=f(uv)$ for each edge $uv\in E(T)$, and
$$1=h(x_1)<h(x_2)<\cdots <h(x_s)<h(y_1)<h(y_2)<\cdots <h(y_t)=2q
$$ for $x_i\in X$ and $y_j\in Y$. Since $1=f(y)-f(x)$ for an edge $xy\in E(T)$, then we have $1=h(y)-h(x)$, and $h(x)=f(y)$, immediately, $|f(V(T))\cap h(V(T))|\leq 1$.
\end{proof}

\begin{example}\label{exa:8888888888}
In Fig.\ref{fig:top-signature-11}, the graphs $T,H_1,H_2,H_3$ and $H_4$ are topological signatures, in which $T$ admits an odd-graceful labeling $f$ with $f(E(T))=[1,19]^o$, and $H_i$ admits an odd-graceful labeling $f_i$ with $f_i(E(H_i))=[1,19]^o$ and $i\in [1,4]$, and moreover we can see each pair $\langle f,f_i\rangle $ forms a \emph{twin odd-graceful labeling} holding
$$f(V(T)\cup E(T))\bigcup f(V(H_i)\cup E(H_i))=[0,20],~f(V(T))\cap f(V(H_i))=\{9\},~i\in [1,4]
$$ However, $H_i\not \cong H_j$ for $i\neq j$.

\begin{figure}[h]
\centering
\includegraphics[width=16.4cm]{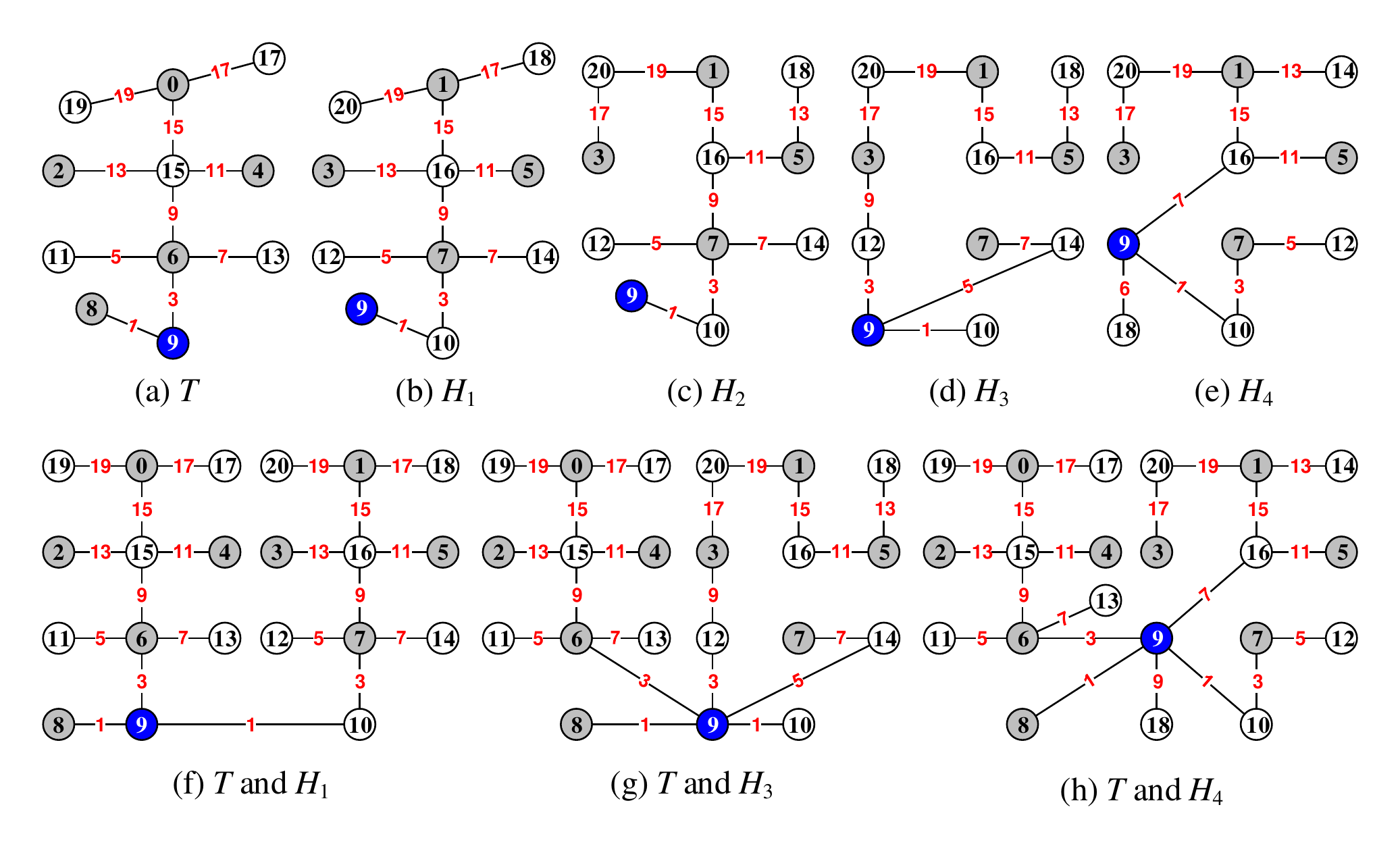}\\
\caption{\label{fig:top-signature-11}{\small Three topological signature authentications: (f) $T[\odot]H_1$; (g) $T[\odot]H_3$; (h) $T[\odot]H_4$.}}
\end{figure}

Each of Topcode-matrices $T_{code}(T,f)$ and $T_{code}(H_i,f_i)$ for $i\in [1,4]$ induces $(30)!$ different number-based strings.

\begin{equation}\label{eqa:vertex-coloring-Topcode-matrix00}
\centering
{
\begin{split}
T_{code}(T,f)= \left(
\begin{array}{ccccccccccc}
8 & 6 & 6 & 6 & 6 & 4 & 2 & 0 & 0 & 0\\
1 & 3 & 5 & 7 & 9 & 11 & 13 & 15 & 17 & 19\\
9 & 9 & 11 & 13 & 15 & 15 & 15 & 15 & 17 & 19
\end{array}
\right)
\end{split}}
\end{equation}

\begin{equation}\label{eqa:vertex-coloring-Topcode-matrix11}
\centering
{
\begin{split}
T_{code}(H_1,f_1)= \left(
\begin{array}{ccccccccccc}
9 & 7 & 7 & 7 & 7 & 5 & 3 & 1 & 1 & 1\\
1 & 3 & 5 & 7 & 9 & 11 & 13 & 15 & 17 & 19\\
10 & 10 & 12 & 14 & 16 & 16 & 16 & 16 & 18 & 20
\end{array}
\right)
\end{split}}
\end{equation}

\begin{equation}\label{eqa:vertex-coloring-Topcode-matrix22}
\centering
{
\begin{split}
T_{code}(H_2,f_2)= \left(
\begin{array}{ccccccccccc}
9 & 7 & 7 & 7 & 7 & 5 & 5 & 1 & 3 & 1\\
1 & 3 & 5 & 7 & 9 & 11 & 13 & 15 & 17 & 19\\
10 & 10 & 12 & 14 & 16 & 16 & 18 & 16 & 20 & 20
\end{array}
\right)
\end{split}}
\end{equation}

\begin{equation}\label{eqa:vertex-coloring-Topcode-matrix33}
\centering
{
\begin{split}
T_{code}(H_3,f_3)= \left(
\begin{array}{ccccccccccc}
9 & 9 & 9 & 7 & 3 & 5 & 5 & 1 & 3 & 1\\
1 & 3 & 5 & 7 & 9 & 11 & 13 & 15 & 17 & 19\\
10 & 12 & 14 & 14 & 12 & 16 & 18 & 16 & 20 & 20
\end{array}
\right)
\end{split}}
\end{equation}

\begin{equation}\label{eqa:vertex-coloring-Topcode-matrix44}
\centering
{
\begin{split}
T_{code}(H_4,f_4)= \left(
\begin{array}{ccccccccccc}
9 & 7 & 7 & 9 & 9 & 5 & 1 & 1 & 3 & 1\\
1 & 3 & 5 & 7 & 9 & 11 & 13 & 15 & 17 & 19\\
10 & 10 & 12 & 16 & 18 & 16 & 14 & 16 & 20 & 20
\end{array}
\right)
\end{split}}
\end{equation}

The Topcode-matrices of the topological signature authentications based on the graphs $T,H_1,H_3$ and $H_4$ shown in Fig.\ref{fig:top-signature-11} are as follows
$$
T_{code}(T[\odot]H_i,g_i)=T_{code}(T,f)\cup T_{code}(H_i,f_i)
$$ for $i=1,3,4$.\qqed
\end{example}

\begin{thm}\label{them:twin-w-type-total-colorings}
\cite{Yao-Zhang-Yang-Wang-Odd-Edge-arXiv-02477} If a connected bipartite $(p,q)$-graph $G$ admits an odd-edge graceful-difference total coloring $f^*$, then there exists a bipartite graph $G^*$ admitting an odd-edge graceful-difference total coloring $g^*$, such that $\langle f^*, g^*\rangle $ forms a \emph{twin set-ordered odd-edge graceful-difference total coloring} of the $(p,q)$-graph $G$ and the graph $G^*$.
\end{thm}

The authors \cite{Wang-Xu-Yao-2017} have defined a \emph{twin odd-elegant labeling} $\langle f,g\rangle $ for a vertex-coincided graph $(p+p\,'-2,2q)$-graph $G[\odot]H$ by

(i) \cite{Zhou-Yao-Chen2013} An \emph{odd-elegant labeling} $f$ of a $(p,q)$-graph $G$ holds $f(V(G))\subset [0,2q-1]$, $f(u)\neq f(v)$ for distinct $u,v\in V(G)$, and
$$f(E(G))=\{f(uv)=f(u)+f(v)~(\bmod~2q):uv\in E(G)\}=[1,2q-1]^o$$

(ii) Another $(p\,',q)$-graph $H$ admits an odd-elegant labeling $g:V(H))\rightarrow [0,2q-1]$, holding $g(u)\neq g(v)$ for distinct $u,v\in V(H)$, and
$$g(E(H))=\{g(uv)=g(u)+g(v)~(\bmod~2q):uv\in E(H)\}=[1,2q-1]^o$$

\begin{cor}\label{thm:666666}
$^*$ A tree $T$ of $q$ edges admits a set-ordered odd-elegant labeling $f$, then $T$ admits another set-ordered odd-elegant labeling $g$, such that
$$f(E(T))=g(E(T))=[1,2q-1]^o,~f(V(T))\cup g(V(T))=[0,2q],~|f(V(T))\cap g(V(T))|\leq 1
$$ that is $\langle f,g \rangle$ forms a twin odd-elegant labeling.
\end{cor}

\begin{lem}\label{thm:adding-edge-subtracting-dual-graph}
$^*$ Suppose that graph $G$ admits an odd-graceful labeling $g$, and $g(uv)=g(xy)$ for two edges $xy\not \in E(G)$ with $x,y\in V(G)$ and $uv\in E(G)$, then the adding-edge-subtracting dual graph $H=G+xy-uv$ admits an odd-graceful labeling $g^*$ induced by the odd-graceful labeling $g$.
\end{lem}

\begin{lem}\label{thm:one-graph-twin-odd-graceful-labeling}
$^*$ If a connected bipartite $(p,q)$-graph $G$ admits an odd-graceful labeling $f$, then $G$ admits another odd-graceful labeling $f^*$, such that $\langle f,f^* \rangle$ forms a twin odd-graceful labeling holding
$$f(E(G))=f^*(E(G))=[1,2q-1]^o,~f(V(G))\cup f^*(V(G))\subseteq [0,2q]
$$
\end{lem}

\begin{defn} \label{defn:edge-separably-uniformly-coloring}
$^{*}$ Suppose that each graph $G_i$ with $i\in [1,n]$ is a proper subgraph of a $(p,q)$-graph $G$ holding $E(G)=\bigcup ^n_{i=1}E(G_i)$ with $E(G_i)\cap E(G_j)=\emptyset$ if $i\neq j$. If the graph $G$ admits a $W$-constraint total coloring $F:V(G)\cup E(G)\rightarrow [0,M]$, such that $F(u_{i,r}v_{i,s})=W\langle F(u_{i,r}),F(v_{i,s})\rangle $ for each edge $u_{i,r}v_{i,s}\in E(G_i)$ with $i\in [1,n]$,
\begin{equation}\label{eqa:555555}
F(V(G))=\bigcup^n_{i=1}F(V(G_i)),~F(E(G))=\bigcup^n_{i=1}F(E(G_i))
\end{equation} then we call $F$

(i) \emph{edge-separably $W$-constraint total coloring} of the graph $G$ if $F(E(G_i))\cap F(E(G_j))=\emptyset$ for $i\neq j$.

(ii) \emph{edge-uniformly $W$-constraint total coloring} of the graph $G$ if $F(E(G_i))=F(E(G_j))$ for any pair of $i,j\in [1,n]$.\qqed
\end{defn}

\begin{defn} \label{defn:separably-uniformly-colorings-labelings}
$^{*}$ By the hypothesis of Definition \ref{defn:edge-separably-uniformly-coloring}, we have the following definitions:
\begin{asparaenum}[\textbf{\textrm{Subdefi}}-1.]
\item If the edge color set
$$F(E(G))=\{F(uv)=|F(u)-F(v)|:uv\in E(G)\}=[1,q]
$$ and the vertex color sets $F(V(G_i))\neq F(V(G_j))$ and the edge color sets $F(E(G_i))\cap F(E(G_j))=\emptyset$ for $i\neq j$, then $F$ is called \emph{edge-separably graceful total coloring}, and $F$ is called \emph{edge-separably graceful total labeling} if $|F(V(G))|=p$.
\item If the edge color set
$$F(E(G))=\{F(uv)=|F(u)-F(v)|:uv\in E(G)\}=[1,2q-1]^o
$$ and the vertex color sets $F(V(G_i))\neq F(V(G_j))$ for $i\neq j$, and the edge color sets $F(E(G_i))\cap F(E(G_j))=\emptyset$ for $i\neq j$, then $F$ is called \emph{edge-separably odd-graceful total coloring}, and moreover $F$ is called \emph{edge-separably odd-graceful total labeling} if $|F(V(G))|=p$.
\item If $F(E(G))=\{F(uv)=|F(u)-F(v)|:uv\in E(G)\}$, and the edge color sets
$$F(E(G))=F(E(G_i))=[1,2q^*-1]^o,~q^*=|F(E(G))|=|F(E(G_i))|
$$ with $i\in [1,n]$ hold true, and the vertex color sets $F(V(G_i))\neq F(V(G_j))$ for $i\neq j$, then $F$ is called \emph{edge-uniformly odd-graceful total coloring}, and moreover $F$ is called \emph{edge-uniformly odd-graceful total labeling} if $|F(V(G))|=p$.
\item If the edge-magic constraint $F(u_i)+F(u_iv_i)+F(v_i)=k_i$ for each edge $u_iv_i\in E(G_i)$ holds true, where $k_i$ is a positive integer, and the edge color set $F(E(G_i))=[1,2q_i-1]^o$ with $q_i=|E(G_i)|$, then $F$ is called \emph{odd-edge-uniformly edge-magic total coloring}.
\item If the edge-difference constraint $F(u_iv_i)+|F(u_i)-F(v_i)|=k_i$ for each edge $u_iv_i\in E(G_i)$ holds true, where $k_i$ is a positive integer, and the edge color set $F(E(G_i))=[1,2q_i-1]^o$ with $q_i=|E(G_i)|$, then $F$ is called \emph{odd-edge-uniformly edge-difference total coloring}.
\item If the felicitous-difference constraint $|F(u_i)+F(v_i)-F(u_iv_i)|=k_i$ for each edge $u_iv_i\in E(G_i)$ holds true, where integer $k_i\geq 0$, and the edge color set $F(E(G_i))=[1,2q_i-1]^o$ with $q_i=|E(G_i)|$, then $F$ is called \emph{odd-edge-uniformly felicitous-difference total coloring}.
\item If the graceful-difference constraint $\big ||F(u_i)-F(v_i)|-F(u_iv_i)\big |=k_i$ for each edge $u_iv_i\in E(G_i)$ holds true, where integer $k_i\geq 0$, and the edge color set $F(E(G_i))=[1,2q_i-1]^o$ with $q_i=|E(G_i)|$, then $F$ is called \emph{odd-edge-uniformly graceful-difference total coloring}.\qqed
\end{asparaenum}
\end{defn}

\begin{figure}[h]
\centering
\includegraphics[width=16.4cm]{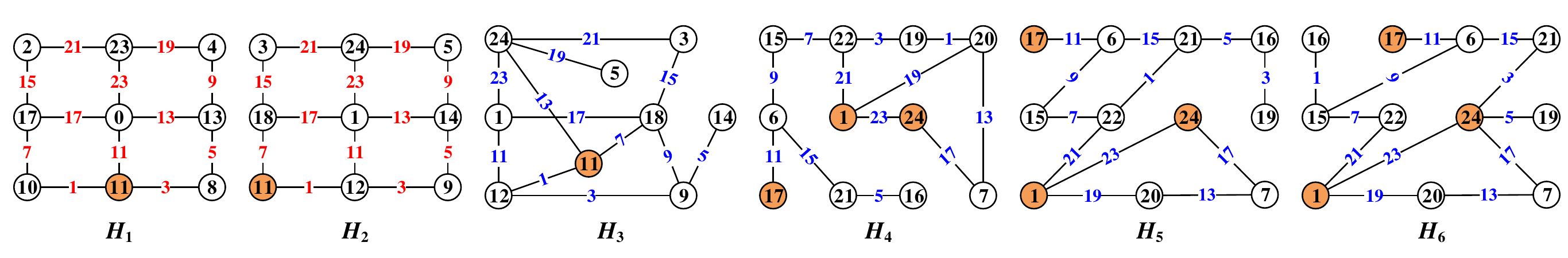}\\
\caption{\label{fig:multi-odd-graceful-matching}{\small Examples for understanding Definition \ref{defn:separably-uniformly-colorings-labelings}.}}
\end{figure}

\begin{example}\label{exa:8888888888}
In Fig.\ref{fig:multi-odd-graceful-matching}, each graph of $T_{12j}=H_1\cup H_2\cup H_j$ for $j\in [4,6]$ and $T_{13k}=H_1\cup H_3\cup H_k$ for $k\in [4,6]$ admits an edge-uniformly odd-graceful total coloring defined in Definition \ref{defn:separably-uniformly-colorings-labelings}. And we have the \emph{adding-edge-subtracting graph homomorphisms} $H_3\rightarrow^e_{\pm} H_2$, $H_5\rightarrow^e_{\pm} H_4$, $H_6\rightarrow^e_{\pm} H_4$ and $H_6\rightarrow^e_{\pm} H_5$.

We have a graph $G=\bigcup^6_{i=1}G_i$ shown in Fig.\ref{fig:multi-odd-graceful-group}, such that $G$ admits an edge-uniformly odd-graceful total coloring $F$ holding $F(E(G))=F(E(G_i))=[1,23]^o$ for $i\in [1,6]$, and $F(V(G_i))\neq F(V(G_j))$ for $i\neq j$, as well as $F(V(G))=[0,28]$.

In Fig.\ref{fig:multi-edge-difference-group}, the graphs $I_{1,2,i}=L_1\cup L_2\cup L_i$ for $i\in [5,7]$ and $I_{1,3,j}=L_1\cup L_3\cup L_j$ for $j\in [5,7]$ admit odd-edge-uniformly edge-difference total colorings defined in Definition \ref{defn:separably-uniformly-colorings-labelings}.\qqed
\end{example}

\begin{figure}[h]
\centering
\includegraphics[width=16.4cm]{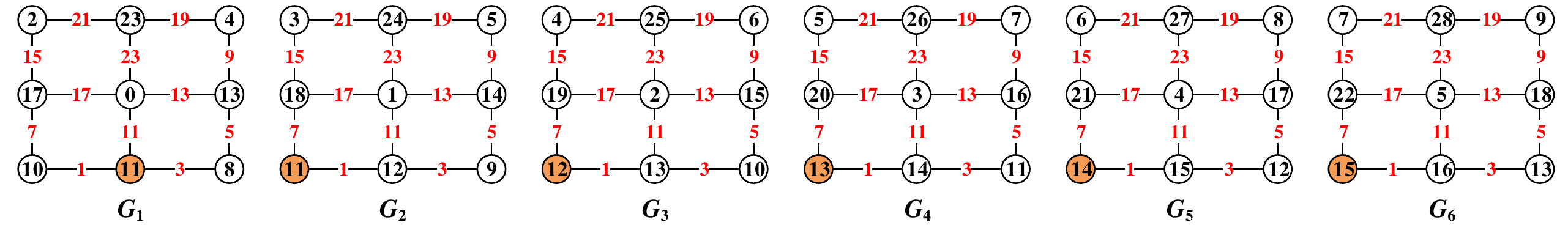}\\
\caption{\label{fig:multi-odd-graceful-group}{\small An odd-graceful graphic group for understanding the edge-uniformly odd-graceful total coloring defined in Definition \ref{defn:separably-uniformly-colorings-labelings}.}}
\end{figure}

\begin{figure}[h]
\centering
\includegraphics[width=16.4cm]{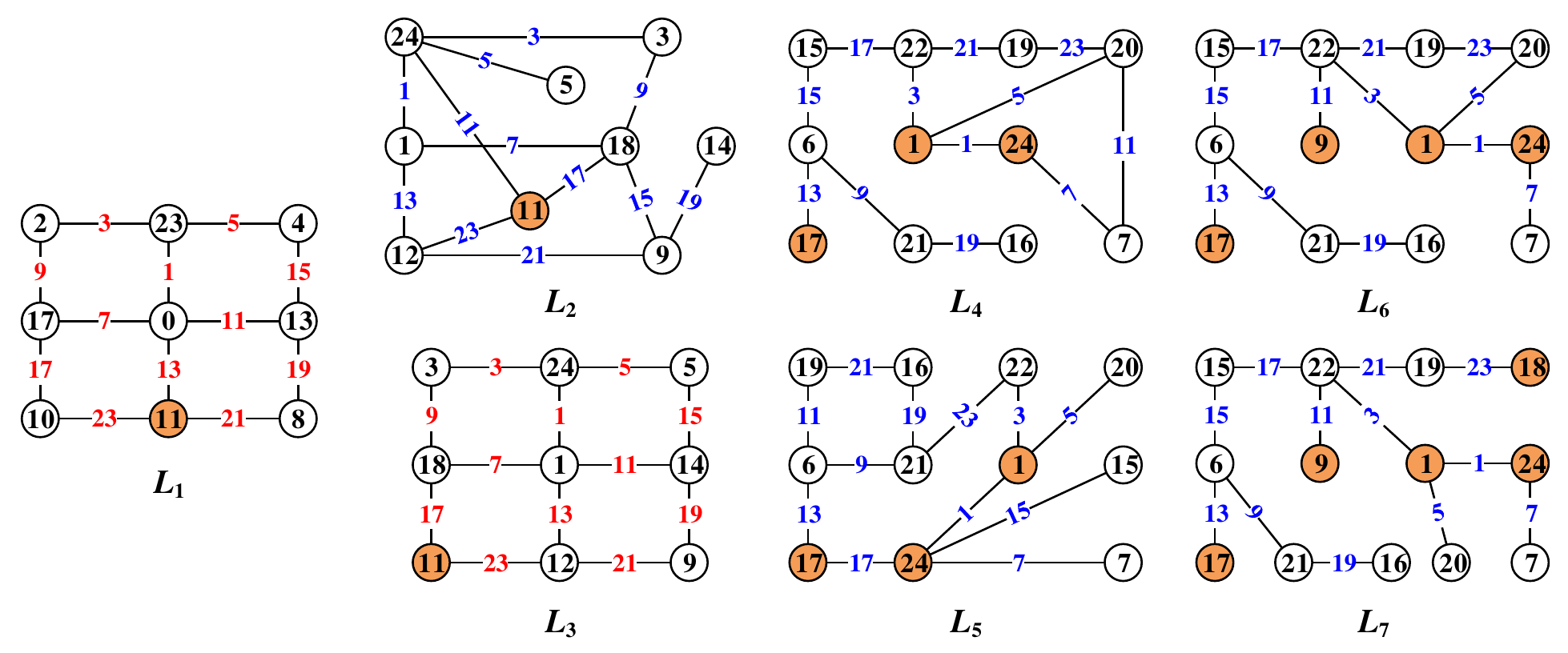}\\
\caption{\label{fig:multi-edge-difference-group}{\small Examples for the odd-edge-uniformly edge-difference total coloring defined in Definition \ref{defn:separably-uniformly-colorings-labelings}.}}
\end{figure}

By Definition \ref{defn:edge-separably-uniformly-coloring} and Definition \ref{defn:separably-uniformly-colorings-labelings}, Theorem \ref{thm:10-k-d-total-coloringss}, Lemma \ref{thm:adding-edge-subtracting-dual-graph} and Lemma \ref{thm:one-graph-twin-odd-graceful-labeling}, we have
\begin{thm}\label{thm:666666}
$^{*}$ Each tree $T$ can be vertex-split into mutually edge-disjoint subtrees $T_1,T_2,\dots ,T_m$, such that $T=[\odot ]^m_{i=1}T_i$ with $|V(T_i)\cap V(T_j)|\leq 1$. Then $T$ admits an edge-separably graceful total coloring and an edge-separably odd-graceful total coloring defined in Definition \ref{defn:separably-uniformly-colorings-labelings}.
\end{thm}

\begin{thm}\label{thm:666666}
$^{*}$ Let $H_1,H_2,\dots ,H_n$ be a group of mutually edge-disjoint trees, and let each tree $H_i$ have just $q$ edges for $i\in [1,n]$. Then the graph $H=\bigcup ^n_{i=1}H_i$ admits each coloring defined in Definition \ref{defn:separably-uniformly-colorings-labelings}.
\end{thm}

\subsubsection{Image-type graphs for topological signatures}

\begin{defn}\label{defn:image-labeling2}
\cite{Jing-Su-Wang-Yao-Image-labellings-2020, Jing-Su-Wang-Yao-Mathematics-2021} Let $\alpha:V(G)\cup E(G)\rightarrow [a,b]$ and $\beta:V(G)\cup E(G)\rightarrow [c,d]$ be two labelings of a $(p,q)$-graph $G$, where integers $a,c ,b,d$ subject to $0\leq a<b$ and $0\leq c<d$.

(1) For each vertex $v\in V(G)$, $\alpha(v)+\beta(v)=k_{v}$ holds true, and $k_{v}$ is a positive integer, we call $\langle \alpha,\beta\rangle $ \emph{matching} of vertex image-labelings, abbreviated as \emph{v-image-labelings};

(2) If equation $\alpha(uv)+\beta(uv)=k_{e}$ for every edge $uv\in E(G)$ holds true, and $k_{e}$ is a positive integer, $\langle \alpha,\beta\rangle $ is called a \emph{matching} of edge image-labelings, abbreviated as \emph{e-image-labelings};

(3) If both equations $\alpha(v)+\beta(v)=k_{v}$ for $v\in V(G)$ and $\alpha(xy)+\beta(xy)=k_{e}$ for $xy\in E(G)$ hold true, we call $\langle \alpha,\beta\rangle $ \emph{matching} of vertex-edge image-labelings (\emph{ve-image-labelings}), where $k_{v}$ and $k_{e}$ are called \emph{vertex-image coefficient} and \emph{edge-image coefficient}, respectively.\qqed
\end{defn}

\begin{thm} \label{thm:1}
\cite{Jing-Su-Wang-Yao-Image-labellings-2020} Let $T$ be a tree with $p$ vertices. If the tree $T$ admits a set-ordered graceful labeling, then the following assertions are mutually equivalent:

$(1)$ The tree $T$ admits a matching of graceful v-image-labelings with $k_{v}=p-1$.

$(2)$ The tree $T$ admits a matching of odd-graceful v-image-labelings with $k_{v}=2p-3$.

$(3)$ The tree $T$ admits a matching of felicitous v-image-labelings with $k_{v}=p-1$.

$(4)$ The tree $T$ admits a matching of odd-elegant v-image-labelings with $k_{v}=p$.

$(5)$ The tree $T$ admits a matching of super edge-magic total v-image-labelings with $k_{v}=p+1$.

$(6)$ The tree $T$ admits a matching of set-ordered $(k,d)$-graceful v-image-labelings with $k_{v}=k+(p-2)d$.

$(7)$ The tree $T$ admits a super $(s+p+3,2)$-edge antimagic total labeling and a super $(s+p+4,2)$-edge antimagic total labeling, they are a matching of v-image-labelings with $k_{v}=p+1$.

$(8)$ The tree $T$ admits a $(k,2)$-arithmetic total labeling and a $(k+2,2)$-arithmetic total labeling, they are a matching of v-image-labelings with $k_{v}=k+2(t-1)$.
\end{thm}

\begin{cor} \label{cor:connections-several-labelings}
\cite{Jing-Su-Wang-Yao-Image-labellings-2020} If a $(p,q)$-tree $T$ admits a set-ordered graceful labeling, then the following assertions hold and are equivalent to each other.

$(1)$ The tree $T$ admits a matching of edge-magic total ve-image-labelings with $k_{v}=p+1$ and $k_{e}=2p+q+1$.

$(2)$ The tree $T$ admits a $(s+p+3,2)$-edge antimagic total labeling and a $(s+p+4,2)$-edge antimagic total labeling such that they are a matching of ve-image-labelings with $k_{v}=p+1$ and $k_{e}=2p+q+1$.

$(3)$ The tree $T$ admits a $(k,2)$-arithmetic total labeling and a $(k+2,2)$- arithmetic total labeling such that they are a matching of ve-image-labelings with $k_{v}=k+2(t-1)$ and $k_{e}=2(k+2t-2)$.
\end{cor}

\begin{cor} \label{cor:connections-several-labelings}
\cite{Jing-Su-Wang-Yao-Image-labellings-2020} Suppose that a $(p,q)$-tree $T$ admits three labelings $f$, $g$ and $h$ such that $\langle f,g\rangle $ is a matching of v-image-labelings, $\langle f,h\rangle $ is a matching of e-image-labelings, then we have the following propositions:

(1) $\{f(v)+g(v)+h(v): v\in V(T)\}$ is an arithmetic sequence if and only if $\{f(uv)+g(uv)+h(uv): uv\in E(T)\}$ is an arithmetic sequence too, and the tolerances are equal.

(2) $\{f(v)+g(v)+h(v): v\in V(T)\}$ contains two arithmetic sequences if and only if $\{f(uv)+g(uv)+h(uv): uv\in E(T)\}$ is an arithmetic sequences with a tolerance 2.
\end{cor}

\begin{thm} \label{thm:3}
\cite{Jing-Su-Wang-Yao-Image-labellings-2020} If a $(p,q)$-tree $T$ admits a set-ordered graceful labeling $f$, then it admits an edge-magic total labeling $g$, so that $\langle f,g\rangle $ is a matching of e-image-labelings with $k_{e}=p+q+1$.
\end{thm}

\begin{thm} \label{thm:4}
\cite{Jing-Su-Wang-Yao-Image-labellings-2020} If a $(p,q)$-tree $T$ admits a set-ordered graceful labeling $f$, then it admits an set-ordered $(s+p+3,2)$-edge antimagic total labeling $g$, so that $\langle f,g\rangle $ is a matching of e-image-labelings with $k_{e}=p+q+1$.
\end{thm}

\begin{defn}\label{defn:image-labeling}
\cite{Yao-Zhang-Sun-Mu-Sun-Wang-Wang-Ma-Su-Yang-Yang-Zhang-2018arXiv} Let $f_i:V(G)\rightarrow [a,b]$ be a labeling of a $(p,q)$-graph $G$ and let each edge $uv\in E(G)$ have its own label as $f_i(uv)=|f_i(u)-f_i(v)|$ with $i=1,2$. If each edge $uv\in E(G)$ holds $f_1(uv)+f_2(uv)=k$ true, where $k$ is a positive constant, we call $\langle f_1,f_2\rangle $ \emph{matching of image-labelings}, and $f_i$ a \emph{mirror-image} of $f_{3-i}$ with $i=1,2$.\qqed
\end{defn}

\begin{defn} \label{defn:twin-k-d-harmonious-labelings}
\cite{Yao-Zhang-Sun-Mu-Sun-Wang-Wang-Ma-Su-Yang-Yang-Zhang-2018arXiv} A $(p,q)$-graph $G$ admits two $(k,d)$-harmonious labelings $f_i:V(G)\rightarrow X_0\cup X_{k,d}$ with $i=1,2$, where
$$X_0=\{0,d,2d, \dots ,(q-1)d\},~X_{k,d}=\{k,k+d,k+2d, \dots ,k+(q-1)d\}
$$ such that each edge $uv\in E(G)$ is labelled as $f_i(uv)-k=[f_i(u)+f_i(v)-k~(\textrm{mod}~qd)]$ with $i=1,2$. If $f_1(uv)+f_2(uv)=2k+(q-1)d$ for each edge $uv\in E(G)$, so we call $\langle f_1,f_2\rangle $ \emph{matching of $(k,d)$-harmonious image-labelings} of $G$.\qqed
\end{defn}

\begin{defn} \label{defn:twin-k-d-harmonious-labelings}
\cite{Yao-Zhang-Sun-Mu-Sun-Wang-Wang-Ma-Su-Yang-Yang-Zhang-2018arXiv} A $(p,q)$-graph $G$ admits a $(k,d)$-labeling $f$, and another $(p',q')$-graph $H$ admits another $(k,d)$-labeling $g$. If $(X_0\cup X_{k,d})\setminus f(V(G)\cup E(G))=g(V(H)\cup E(H))$ with
$$X_0=\{0,d,2d, \dots ,(q-1)d\},~X_{k,d}=\{k,k+d,k+2d, \dots ,k+(q-1)d\}
$$ then $g$ is called a \emph{complementary $(k,d)$-labeling} of $f$, and $\langle f,g\rangle $ is called \emph{twin $(k,d)$-labelings} of the $(p,q)$-graph $G$.\qqed
\end{defn}

\begin{lem}\label{thm:graceful-image-labeling}
\cite{Yao-Zhang-Sun-Mu-Sun-Wang-Wang-Ma-Su-Yang-Yang-Zhang-2018arXiv} If a tree $T$ admits a set-ordered graceful labeling $f$, then $T$ admits another set-ordered graceful labeling $g$ such that $f$ and $g$ are a matching of image-labelings.
\end{lem}

\begin{thm}\label{thm:10-image-labelings}
\cite{Yao-Zhang-Sun-Mu-Sun-Wang-Wang-Ma-Su-Yang-Yang-Zhang-2018arXiv} If a tree $T$ admits a set-ordered graceful labeling, then $T$ admits a matching of $W$-constraint image-labelings, where $W$-constraint $\in \{$set-ordered graceful, set-ordered odd-graceful, edge-magic graceful, set-ordered felicitous, set-ordered odd-elegant, super set-ordered edge-magic total, super set-ordered edge-antimagic total, set-ordered $(k,d)$-graceful, $(k,d)$-edge antimagic total, $(k,d)$-arithmetic total, harmonious, $(k,d)$-harmonious$\}$.
\end{thm}

\subsubsection{Set-dual graphs for topological signatures}

\begin{defn} \label{defn:set-key-dual-coloring-graphs}
$^*$ For a $(p,q)$-graph $G$ (as a \emph{public-key graph}) admitting a coloring $f$, and another $(p\,',q\,')$-graph $H$ (as a \emph{private-key graph}) admitting a coloring $g$, if there exists a constant $c$ holding $$f(w)+g[\pi(w)]=c,~w\in S\subset V(G)\cup E(G)
$$ under a mapping $\pi:S \rightarrow S^*$ for a subset $S^*\subset V(H)\cup E(H)$, then we call $\langle G,H\rangle $ \emph{set-dual graphs}, and $f$ (resp. $g$) \emph{set-dual coloring} of $g$ (resp. $f$).\qqed
\end{defn}

By Definition \ref{defn:set-key-dual-coloring-graphs}, there are the following particular cases:

(i) $S=V(G)$, $S^*=V(H)$.

(ii) $S=E(G)$, $S^*=E(H)$.

(iii) $S=V(G)\cup E(G)$, $S^*=V(H)\cup E(H)$.

(iv) $G=H$ with $S=S^*$.

\subsubsection{Symmetric graphs as topological signatures}

There are \emph{edge-symmetric graphs} $G[\ominus] G\,'$ obtained by using a new edge to join a vertex $u$ of a graph $G$ with its image vertex $u\,'$ of the image graph $G\,'$ with $G\cong G\,'$, and there are \emph{vertex-symmetric graphs} $G[\odot] G\,'$ obtained by vertex-coinciding a vertex $w$ of a graph $G$ with its image vertex $w\,'$ of the image graph $G\,'$ with $G\cong G\,'$ into one vertex $w\odot w\,'$.

\begin{defn}\label{defn:mf-graceful-mf-odd-graceful}
\cite{Yao-Mu-Sun-Sun-Zhang-Wang-Su-Zhang-Yang-Zhao-Wang-Ma-Yao-Yang-Xie2019} For any vertex $u$ of a connected and bipartite $(p,q)$-graph $G$, there exist a vertex labeling $f:V(G)\rightarrow [0,q]$ (or $f:V(G)\rightarrow [0,2q-1]$) such that

(i) $f(u)=0$;

(ii) $f(E(G))=\{f(xy)=|f(x)-f(y)|: ~xy\in E(G)\}=[1,q]$ (or $f(E(G))=[1,2q-1]^o$);

(iii) the bipartition $(X,Y)$ of $V(G)$ holds $\max f(X)<\min f(Y)$.\\
Then we say $G$ admits a \emph{$0$-rotatable set-ordered system of (odd-)graceful labelings}, abbreviated as \emph{$0$-rso-graceful system} (\emph{$0$-rso-odd-graceful system}).\qqed
\end{defn}

\begin{lem}\label{thm:symmetric-tree}
\cite{Yao-Mu-Sun-Sun-Zhang-Wang-Su-Zhang-Yang-Zhao-Wang-Ma-Yao-Yang-Xie2019} If a tree $T$ admits a $0$-rotatable system of (odd-)graceful labelings, then its symmetric tree $T[\ominus] T\,'$ admits a $0$-rotatable set-ordered system of (odd-)graceful labelings.
\end{lem}

\begin{thm}\label{thm:0-rotatable-set-ordered-system00}
\cite{Yao-Mu-Sun-Sun-Zhang-Wang-Su-Zhang-Yang-Zhao-Wang-Ma-Yao-Yang-Xie2019} Suppose that a connected and bipartite $(p,q)$-graph $G$ admits a $0$-rotatable set-ordered system of (odd-)graceful labelings. Then the edge symmetric graph $G[\ominus] G\,'$ admits a $0$-rotatable set-ordered system of (odd-)graceful labelings too.
\end{thm}

\begin{thm}\label{thm:0-rotatable-set-ordered-system11}
\cite{Yao-Mu-Sun-Sun-Zhang-Wang-Su-Zhang-Yang-Zhao-Wang-Ma-Yao-Yang-Xie2019} There are infinite graphs admit $0$-rotatable set-ordered systems of (odd-)graceful labelings.
\end{thm}

\subsubsection{Other techniques for topological signatures}

\quad ~\textbf{1. Matching-pair topological signatures.}

\begin{asparaenum}[\textbf{\textrm{Matching}}-1.]
\item Equivalent colorings/labelings. Suppose that a graph $G$ admits a coloring $f$ and another graph $H$ admits a coloring $g$. If $G\cong H$ and $f$ is equivalent with $g$, so two graphs $G$ and $H$ form a matching-pair of topological signatures.
\item Twin-type colorings/labelings.
\item Dual-type graphs are determined by dual-type colorings/labelings.
\item Magic-constraint colorings are defined by magic-constraint $\in\{$edge-magic, edge-difference, felicitous-difference, graceful-difference$\}$.
\item Dual-magic-constraint colorings are: vertex-dual magic-constraint coloring, edge-dual magic-constraint coloring, all-dual magic-constraint coloring for magic-constraint $\in\{$edge-magic, edge-difference, felicitous-difference, graceful-difference$\}$.
\item Matching-type graphic lattices.
\end{asparaenum}

\vskip 0.4cm

\textbf{2. Graph operations for topological signatures.}
\begin{asparaenum}[\textbf{\textrm{Operation}}-1.]
\item Graph homomorphism for topological signatures.
\item Graph-operation homomorphism; refer to Definition \ref{defn:v-e-colored-graph-homomorphisms} for the colored graph homomorphism.
\item The vertex-splitting operation. For example, we use the \emph{vertex-splitting operation} to a connected $(p,q)$-graph $G$ for producing a tree $T$ of $q+1$ vertices, so we get a graph homomorphism $T\rightarrow G$; refer to the vsplit-tree set $V^{tree}_{split}(G)$, each tree $T\in V^{tree}_{split}(G)$ is graph homomorphism to $G$.
\item Graph-operation homomorphisms: the $\pm e$-operation, also adding-edge-subtracting operation, which produces the adding-edge-subtracting graph homomorphism; refer to Problem \ref{question:techniques-KSTrees} and Problem \ref{question:techniques-KSTrees22}, two graph-operation homomorphisms $T^{pri}_{i}\rightarrow _{\pm e}T^{pub}_{i}$ and $H^{pri}_{i}\rightarrow _{\pm e}H^{pub}_{i}$. For example, since $H=G+xy-uv$ for edge $uv\in E(G)$ and edge $xy\not\in E(G)$, then we get an \emph{adding-edge-subtracting graph homomorphism} $H\rightarrow ^e_{\pm}G$.
\item Complement. If $H=G-E(T)$ and $V(H)=V(G)=V(T)$, then $H$ (resp. $T$) is the \emph{complement} of $T$ (resp. $H$) based on the graph $G$; refer to MPG-Key-pairs: $G=G^C_{out}[\overline{\ominus}^C_k]G^C_{in}$ is a maximal planar graph, where $G^C_{out}$ and $G^C_{in}$ are two topological signatures.
\end{asparaenum}

\vskip 0.4cm

\textbf{3. Sets for topological signatures.}

\begin{asparaenum}[\textbf{\textrm{Groulat}}-1.]
\item Public-key graphic lattices and private-key graphic lattices. For example, the graphic lattice $\textbf{\textrm{L}}([\odot_{\textrm{prop}}]Z^0\textbf{\textrm{T}})$ defined in Eq.(\ref{eqa:public-key-graphic-lattice}) is as a \emph{public-key graphic lattice}, and the graphic lattice $\textbf{\textrm{L}}([\odot_{\textrm{prop}}]Z^0\textbf{\textrm{H}})$ defined in Eq.(\ref{eqa:private-key-graphic-lattice}) is as a \emph{private-key graphic lattice}, where $\textbf{\textrm{T}}$ is a \emph{public-key tree base} and $\textbf{\textrm{H}}$ is a \emph{private-key tree base}.
\item In ATE-CGH-problem, also Problem \ref{problem:ATE-CGH-problem00}, the set $S_e(G)$ is as a \emph{public-key set}, and the set $S_v(G)$ is as a \emph{private-key set}.
\item In Problem \ref{question:Pan-ATE-CGH-problem}, also Pan-ATE-CGH-problem, the set $S_e(q)$ is as a \emph{public-key set}, and the set $S_v(q)$ is as a \emph{private-key set}.
\item Infinite graphic-sequences $\{\{G_{s,k}\}^{+\infty}_{-\infty}\}^{+\infty}_{-\infty}$.
\item In Definition \ref{defn:every-zero-graphic-group-homomorphism}, we have an every-zero graphic group homomorphisms
$$\{F_f(G);\oplus \ominus\}\rightarrow \{F_h(H);\oplus \ominus\}
$$ defined by two every-zero graphic groups $\{F_f(G);\oplus \ominus\}$ based on a graph set $F_f(G)=\{G_i\}^m_1$ and $\{F_h(H);\oplus \ominus\}$ based on a graph set $F_h(H)=\{H_i\}^m_1$, when as there are graph homomorphisms $G_i\rightarrow H_i$ defined by $\theta_i:V(G_i)\rightarrow V(H_i)$ with $i\in [1,m]$.
\end{asparaenum}

\subsection{Topological Key-pairs}

Notice that colored graphs are consisted of topological structures and mathematical constraints. Our public-key graphs and private-key graphs, naturally, contain topological structures and mathematical constraints, so are topological signature authentications. We use ``Key-encryption = cipher code'', and ``Key-pair = a public-key and a private-key'' hereafter.

We will design the following basic \emph{topological public-key models}:

(i) A public-key string is a number-based string (StringKey-only).

(ii) A public-key graph is a colored graph (GraphKey-only).

(iii) A public-key graph and a public-key number-based string (GraphString-Key).

In this subsection, the sentence ``public-key number-based string'' is abbreviated as ``public-key string'', and the sentence ``private-key number-based string'' is abbreviated as ``private-key string''.

\subsubsection{The StringKey-only algorithm}

\textbf{Topological public-key model-I (The StringKey-only algorithm).}
\begin{asparaenum}[\textbf{\textrm{Step}} 1.]
\item Alice sends to Bob her public-key string $s_{Apub}$.
\item Bob encrypts a plaintext $P_{lan}$ by Alice's public-key string $s_{Apub}$ and Bob's topological signature $T_{sigB}$ (as the identity authentication), the encrypted file is denoted as $S_{ocum}=\langle P_{lan},s_{Apub},T_{sigB}\rangle$.
\item Alice uses her own topological signature $T_{sigA}$ to decrypt the topological signature protection of the encrypted file $S_{ocum}$ by the topological signature authentication $A_{uth}\langle T_{sigA},T_{sigB}\rangle$, the resultant file is denoted as $S^*_{ocum}$.
\item Alice uses her private-key string $s_{Apri}$ to decrypt safely the file $S^*_{ocum}$ having Alice's public-key string protection for reading the original plaintext $P_{lan}$ without worry about the files being tampered.
\end{asparaenum}

\vskip 0.4cm

For the security of the StringKey-only algorithm we show part of computational complexities of the algorithm as follows:
\begin{asparaenum}[\textrm{\textbf{ComSKA}}-1. ]
\item Attacking the Alice's public-key string $s_{Apub}$ will meet PRONBS-problem defined in Problem \ref{question:PRONBS-problems00}.
\item There are hundreds of graphs for making two topological signatures $T_{sigA}$ and $T_{sigB}$.
\item There are hundreds of colorings and labelings for finding the particular colorings admitted by the public-key graphs $T_{sigA}$ and $T_{sigB}$.
\item Attacking the topological signature authentication $A_{uth}\langle T_{sigA},T_{sigB}\rangle$ will meet the Subgraph Isomorphic Problem, which is a NP-complete problem.
\end{asparaenum}

\subsubsection{The GraphKey-only algorithm}

\textbf{Topological public-key model-II (The GraphKey-only algorithm).}
\begin{asparaenum}[\textbf{\textrm{Gtep}} 1.]
\item Alice sends to Bob a public-key graph $G_{Apub}$.
\item Bob encrypts a plaintext $P_{lan}$ by Alice's public-key graph $G_{Apub}$ and Bob's topological signature $T_{sigB}$ (as the identity authentication), denoted as $G_{ocum}=\langle P_{lan},G_{Apub},T_{sigB}\rangle$.
\item Alice uses her own topological signature $T_{sigA}$ to decrypt the topological signature protection of the encrypted file $G_{ocum}$, as the topological signature authentication $A_{uth}\langle T_{sigA},T_{sigB}\rangle$ passed, the resultant file is denoted as $G^*_{ocum}$.
\item Alice uses her private-key graph $G_{Apri}$ to decrypt safely the file $G^*_{ocum}$ only protected by Alice's public-key graph $G_{Apub}$, and then gets the original plaintext $P_{lan}$.
\end{asparaenum}

\vskip 0.4cm

For the security of the GraphKey-only algorithm we show part of computational complexities of the algorithm as follows:
\begin{asparaenum}[\textrm{\textbf{ComGKA}}-1. ]
\item There are hundreds of graphs for making Alice's public-key graph $G_{Apub}$ and Bob's topological signature $T_{sigB}$.
\item There are hundreds of colorings and labelings for finding the particular colorings admitted by the public-key graph $G_{Apub}$ and Bob's topological signature $T_{sigB}$.
\item Attacking the topological signature authentication $A_{uth}\langle T_{sigA},T_{sigB}\rangle$ by deciphering violently will meet the Subgraph Isomorphic Problem, which is a NP-complete problem.
\end{asparaenum}

\subsubsection{The GraphString-Key algorithm}

\textbf{Topological public-key model-III (The GraphString-Key algorithm).}
\begin{asparaenum}[\textbf{\textrm{GStep}} 1.]
\item Alice sends to Bob a key-package-A including her public-key graph $G_{Apub}$ and her public-key string $s_{Apub}$.
\item Bob encrypts a plaintext $P_{lan}$ by $G_{Apub}$, $s_{Apub}$ and Bob's topological signature $T_{sigB}$ (as the identity authentication), denoted as $D_{ocum}=\langle P_{lan},G_{Apub},s_{Apub},T_{sigB}\rangle$.
\item First, Alice decrypts the encrypted file $D_{ocum}$ by her private-key graph $G_{Apri}$ as the topological authentication $A_{uth}\langle T_{sigA},T_{sigB}\rangle$ agreed by both Alice and Bob in advance, or the third party, where Alice's topological signature $T_{sigA}=\langle G_{Apub},G_{Apri}\rangle $. Similarly, Bob's topological signature signature $T_{sigB}=\langle G_{Bpub},G_{Bpri}\rangle $.
\item If $A_{uth}\langle T_{sigA},T_{sigB}\rangle$ decrypts the encrypted file $D_{ocum}$ out the topological structure protection, the resultant file is denoted as $D^*_{ocum}$.
\item Alice uses her private-key string $s_{Apri}$ to decipher $D^*_{ocum}$ protected only by Alice's public-key string $s_{Apub}$.
\end{asparaenum}

\vskip 0.4cm

Part of advantages of the GraphString-Key algorithm are as follows:
\begin{asparaenum}[\textrm{\textbf{ADV-GSKA}}-1. ]
\item Alice's topological signature $T_{sigA}=\langle G_{Apub},G_{Apri}\rangle $ may be a colored graph $T$ admits a $W$-constraint coloring $f$, which induces the Topcode-matrix $T_{code}(T,f)_{3\times q}$ of the colored graph $T$. Sometimes, Alice's topological signature $T_{sigA}=\langle G_{Apub},G_{Apri}\rangle $ consists of two colored graphs or more colored graphs.
\item There are $(3q)!$ different public-key strings produced by the Topcode-matrix $T_{code}(T,f)_{3\times q}$. In other word, there are enough number-based strings to make the Key-pairs for larger edge number $q$.
\item A colored graph $T$ is easy to be expressed by the vertex-coincided graph of two vertex-disjoint graphs $H$ and $H\,'$ by the vertex-coinciding operation, that is $T=H[\odot]H\,'$, if $T$ is a tree.
\item The topological signature authentication $A_{uth}\langle T_{sigA},T_{sigB}\rangle=\langle Q,J\rangle$ in the Graph-String algorithm is consisted of two vertex-coincided colored graphs $Q=T_{sigA}$ and $J=T_{sigB}$, such that $T_{sigA}=A_{uth}[\textrm{dig}]J$ is for Alice, anther one $T_{sigB}=A_{uth}[\textrm{dig}] Q$ is for Bob, where $A_{uth}[\textrm{dig}] J$ and $A_{uth}[\textrm{dig}] Q$ are two semi-topological signature authentications.
\end{asparaenum}

\vskip 0.4cm

For the security of the GraphString-Key algorithm we show part of computational complexities of the algorithm as follows:
\begin{asparaenum}[\textrm{\textbf{ComGSKA}}-1. ]
\item There are $(3q)!$ different public-key strings produced by the Topcode-matrix $T_{code}(T,f)_{3\times q}$.
\item There are many graphs for finding topological signatures used here, however there is no construction algorithm with the function of distinguishing isomorphic graphs.
\item There are hundreds of colorings and labelings, however it is extremely difficulty for determining the colorings admitted by the topological signatures mentioned here.
\item The public-key graph $T=H[\odot]H\,'$ is related with Subgraph Isomorphic Problem, a NP-complete problem.
\item Deciphering the public-key string $s_{Apub}$ is related with PRONBS-problem defined in Problem \ref{question:PRONBS-problems00}.
\end{asparaenum}

\subsection{Techniques for making topological Key-pairs}

In \cite{Yao-Su-Ma-Wang-Yang-arXiv-2202-03993v1}, the authors introduced topological signature authentication techniques based on complete graphs, planar graphs and trees for the application of topological coding. We will present some techniques for making public-key graphs and private-key graphs in this subsection.

\subsubsection{MPG-Key-pairs}

By Remark \ref{remark:maximal-planar-graph-embedding}, if $W$ is a cycle $C$ of $k$ vertices in a maximal planar graph $G$, the cycle-split graph $G\wedge C$ has just two vertex disjoint components $G^C_{out}$ and $G^C_{in}$, called \emph{semi-maximal planar graphs}, where $G^C_{out}$ is in the \emph{infinite plane}, and $G^C_{in}$ is inside of $G$ (Ref. \cite{Jin-Xu-Maximal-Science-Press-2019}), then the maximal planar graph $G$ can be written as $G=G^C_{out}[\overline{\ominus}^C_k]G^C_{in}$, also, as a topological authentication.

In Fig.\ref{fig:mpg-public-private-key}, one public-key graph $H$ corresponds three private-key graphs $T_1,T_2,T_3$.

\begin{figure}[h]
\centering
\includegraphics[width=16.4cm]{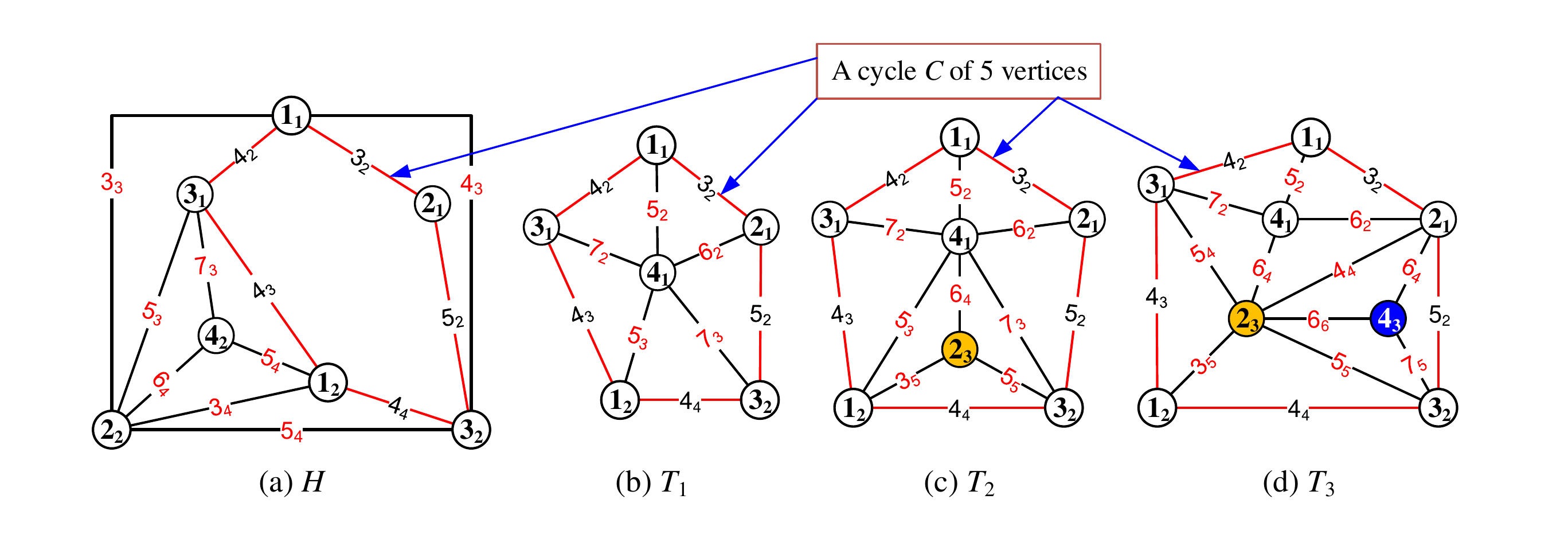}\\
\caption{\label{fig:mpg-public-private-key}{\small (a) $H$ is a public-key graph; (b), (c) and (d) are three private-key graphs. Each maximal planar graph $H[\overline{\ominus}^C_5]T_i$ for $i=1,2,3$ makes a topological signature authentication.}}
\end{figure}

For the security and computational complexity of the MPG-Key-pairs we have:

(i) Since any algorithm for determining the planarity of graphs is a NP-complete problem \cite{Borodin-O-Kotzig-1989}, so the MPG-Key-pairs are guaranteed to be safe.

(ii) On the other hands, the number of \emph{unlabeled (non-isomorphic) planar graphs} on $n$ vertices is between $27.2^{n}$ and $30.06^{n}$, roughly, between $2^{4.7n}$ and $2^{6n}$.

(iii) The following theorem is also the security guarantee of MPG-Key-pairs:
\begin{thm}\label{thm:666666}
$^*$ If a \emph{public-key graph} is a semi-maximal planar graph $G^C_{out}$, then it corresponds infinite private-key graphs, in which each one is a semi-maximal planar graph $G^C_{in}$ to form a maximal planar graph $G=G^C_{out}[\overline{\ominus}^C_k]G^C_{in}$.
\end{thm}

\begin{figure}[h]
\centering
\includegraphics[width=16.4cm]{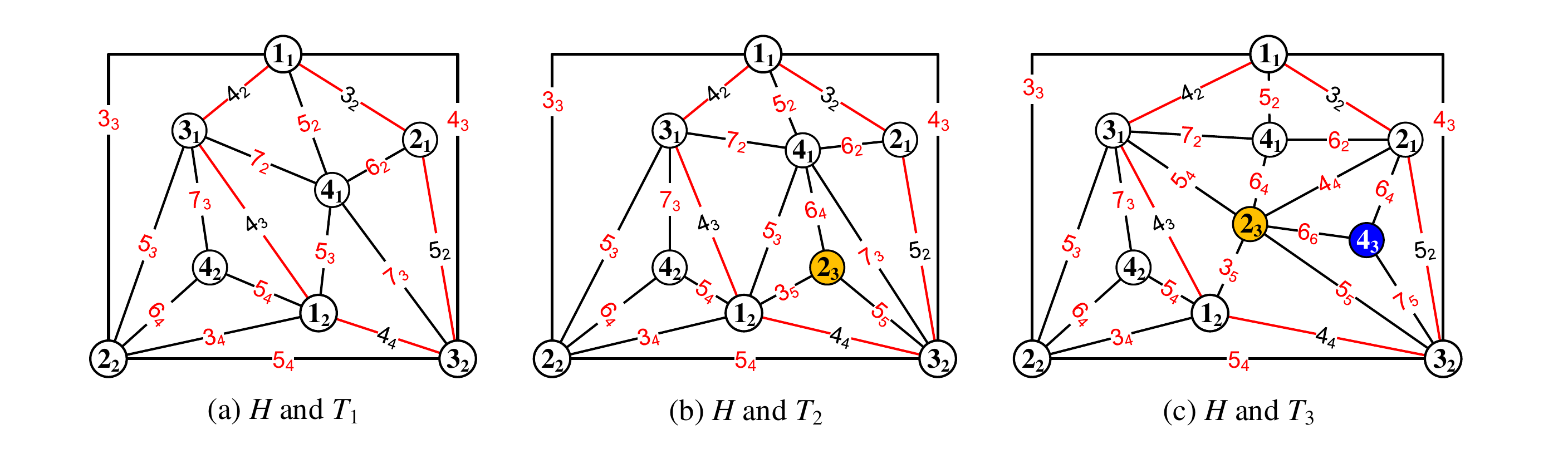}\\
\caption{\label{fig:pub-pri-in-one}{\small Three topological signature authentications made by the semi-maximal planar graphs shown in Fig.\ref{fig:mpg-public-private-key}.}}
\end{figure}

Part of advantages and disadvantages of using maximal planar graph as topology signature and making Key-pairs are as follows:

\begin{asparaenum}[\textrm{\textbf{Adv}}-1. ]
\item A maximal planar graph can be easily scanned into the computer, similar to the two-dimensional code.
\item It is easy to make one-vs-more or more-vs-more topological signature pairs for topological signature authentication $G=G^C_{out}[\overline{\ominus}^C_k]G^C_{in}$.
\item There is no formula and polynomial algorithm for the number of maximal planar graphs, which brings huge overhead to the topological signature of attacking Maximal Planar Graphs.
\item Except for the special maximal plane graphs, the traditional colorings of maximal plane graphs are almost determined.
\item Maximal planar graphs admitting 4-colorings are associated with colored 3-regular planar graphs, and moreover maximal planar graphs correspond to 3-regular planar graphs one by one.
\item Special maximal plane graphs are as topological signatures. For example, if the degree of the vertex of a simple graph $G$ of $p$ is one of $k, l$ and $m$, then $G$ is said to be $(k,l,m)$-regular, and moreover if its number of edges $|E(G)|= 3p-6$, then we call graph $G$ $(k,l,m)$-regular maximal planar graph.
\end{asparaenum}

\subsubsection{KSTree Key-pairs}

\begin{lem}\label{lem:complete-graph-spanning-trees}
\cite{Bondy-2008} The famous Cayley's formula $\tau(K_n)=n^{n-2}$ tells us: Each complete graph $K_{n}$ admitting a labeling $f:V(K_n)\rightarrow [1,n]$ distributes $n^{n-2}$ different spanning trees, such that each spanning tree $T$ holds $f(V(K_n))=[1,n]=f(V(T))$ true.
\end{lem}

Theorem \ref{thm:K-2m-spanning-trees-2m-1-edges} tells us: Each complete graph $K_{2m}$ with $m\geq 2$ can be vertex-split into $m$ mutually edge-disjoint spanning trees $T_1,T_2,\dots ,T_m$ of $2m-1$ edges, such that the topological signature authentication $E(K_{2m})=\bigcup ^m_{i=1}E(T_i)$ holds true.

\begin{problem}\label{question:techniques-KSTrees}
Let $S_{tree}(K_{2m})$ be the set of $(2m)^{2m-2}$ different spanning trees of a colored complete graph $K_{2m}$; refer to Lemma \ref{lem:complete-graph-spanning-trees}. \textbf{Find} all Key-pairs of mutually edge-disjoint spanning trees $T^{pub}_{i}$ and $T^{pri}_{i}$ holding one or more of the following properties:
\begin{asparaenum}[\textbf{\textrm{KSpairs}}-1.]
\item The topological signature authentication $G=T^{pub}_{i}[\odot_{prop}]T^{pri}_{i}$ has its own maximal degree $\Delta(G)=k$.
\item The topological signature authentication $G=T^{pub}_{i}[\odot_{prop}]T^{pri}_{i}$ contains a Hamilton cycle, or is a regular graph, or is a bipartite graph.
\item $T^{pub}_{i}\cong T^{pri}_{i}$ although $f(V\big (T^{pub}_{i}\big ))=[1,n]=f(V\big (T^{pri}_{i}\big ))$.
\item Let $n_d(G)$ be the number of vertices of degree $d$ in a graph $G$. Then
\begin{equation}\label{eqa:555555}
\sum _{3\leq d\leq \Delta \big (T^{pub}_{i}\big )}(d-2)n_d\big (T^{pub}_{i}\big )=\sum _{3\leq d\leq \Delta \big (T^{pri}_{i}\big )}(d-2)n_d\big (T^{pri}_{i}\big )
\end{equation}

\item Two degree sequences of two spanning trees $T^{pub}_{i}$ and $T^{pri}_{i}$
$$
\textrm{deg}\big (T^{pub}_{i}\big )=(a_{i,1},a_{i,2},\dots, a_{i,2m}),~\textrm{deg}\big (T^{pri}_{i}\big )=(b_{i,1},b_{i,2},\dots, b_{i,2m})
$$ hold $a_{i,k}=b_{i,k}$ for $k\in [1,2m]$ true.
\item Two diameters $D_{iam}\big (T^{pub}_{i}\big )=D_{iam}\big (T^{pri}_{i}\big )$.
\item Two spanning trees $T^{pub}_{i}$ and $T^{pri}_{i}$ are (i) two caterpillars; (ii) two lobsters; (iii) one caterpillar and one lobster.
\item Two spanning trees $T^{pub}_{i}$ and $T^{pri}_{i}$ holds $T^{pri}_{i}=T^{pub}_{i}+xy-uv$ with $xy\not \in E\big (T^{pub}_{i}\big )$ and $uv\in E\big (T^{pub}_{i}\big )$, that is a \emph{graph-operation homomorphism} $T^{pri}_{i}\rightarrow _{\pm e}T^{pub}_{i}$.
\end{asparaenum}
\end{problem}

For the security and computational complexity of Problem \ref{question:techniques-KSTrees}, it refers to Table-3, and part of properties of Problem \ref{question:techniques-KSTrees} are related with the Subgraph Isomorphic Problem, a NP-complete problem.

\begin{center}
\textbf{Table-3.} The numbers of trees of $p$ vertices \cite{Harary-Palmer-1973}.
\end{center}
\begin{center}
\begin{tabular}{c|rr}
$p$&$t_p$&$T_p$\\
\hline
7&11&48\\
8&23&115\\
9&47&286\\
10&106&719\\
11&235&1,842\\
12&551&4,766\\
13&1,301&12,486\\
14&3,159&32,973\\
15&7,741&87,811\\
16&19,320&235,381\\
\end{tabular}\qquad
\begin{tabular}{c|rr}
$p$&$t_p$&$T_p$\\
\hline
17&48,629&634,847\\
18&123,867&1,721,159\\
19&317,955&4,688,676\\
20&823,065&12,826,228\\
21&2,144,505&35,221,832\\
22&5,623,756&97,055,181\\
23&14,828,074&268,282,855\\
24&39,299,897&743,724,984\\
25&104,636,890&2,067,174,645\\
26&279,793,450&5,759,636,510\\
\end{tabular}
\end{center}
where $t_p$ is the number of \emph{non-isomorphic trees} of $p$ vertices, and $T_p$ is the number of \emph{non-isomorphic rooted trees} of $p$ vertices, and $t_{29}=$5,469,566,585.

\subsubsection{Key-pairs based on complete bipartite graphs}

\begin{thm}\label{thm:666666}
$^*$ Every bipartite graph is a subgraph of a certain complete bipartite graph admitting a certain $W$-constraint labeling.
\end{thm}

\begin{lem}\label{lem:444445555}
\cite{Bing-Yao-arXiv:2207-03381} If a bipartite and connected $(p,q)$-graph $G$ admits a set-ordered graceful labeling, then the graph $G$ admits each one of the following labelings:
\begin{asparaenum}[\textbf{Coloring}-1. ]
\item graceful-intersection total set-labeling, graceful group-labeling.
\item odd-graceful labeling, set-ordered odd-graceful labeling, edge-odd-graceful total labeling, odd-graceful-intersection total set-labeling, odd-graceful group-labeling, perfect odd-graceful labeling.
\item elegant labeling, odd-elegant labeling.
\item edge-magic total labeling, super edge-magic total labeling, super set-ordered edge-magic total labeling, edge-magic total graceful labeling.
\item relaxed edge-magic total labeling.
\item odd-edge-magic matching labeling, ee-difference odd-edge-magic matching labeling.
\item 6C-labeling, odd-6C-labeling.
\item ee-difference graceful-magic matching labeling.
\item difference-sum labeling, felicitous-sum labeling.
\item multiple edge-meaning vertex labeling.
\item perfect $\varepsilon$-labeling.
\item $(k,d)$-edge antimagic total labeling, $(k, d)$-arithmetic.
\item image-labeling, $(k,d)$-harmonious image-labeling.
\item twin $(k,d)$-labeling, twin Fibonacci-type graph-labeling, twin odd-graceful labeling.
\end{asparaenum}
\end{lem}

\begin{thm}\label{thm:equivalent-k-d-total-colorings}
\cite{Bing-Yao-arXiv:2207-03381} A bipartite and connected $(p,q)$-graph $G$ admits a graceful $(k,d)$-total coloring if and only if this graph $G$ admits each one of edge-magic $(k,d)$-total coloring, graceful-difference $(k,d)$-total coloring, edge-difference $(k,d)$-total coloring, felicitous-difference $(k,d)$-total coloring, harmonious $(k,d)$-total coloring and edge-antimagic $(k,d)$-total coloring.
\end{thm}

\begin{thm}\label{thm:666666}
\cite{Yao-Zhang-Sun-Mu-Sun-Wang-Wang-Ma-Su-Yang-Yang-Zhang-2018arXiv} Each complete bipartite graph $K_{m,n}$ admits a
proper total coloring
$$g : V (K_{m,n})\cup E(K_{m,n})\rightarrow [1,\chi\,''(K_{m,n})]$$
such that the edge-magic color set
$$g^*(E(K_{m,n})) = \{g(u) + g(uv) + g(v) : u\in E(K_{m,n})\}$$
is a consecutive integer set $[a, b]$.
\end{thm}

\begin{thm}\label{thm:666666}
For a complete bipartite graph $K_{m,n}$, we have
\begin{asparaenum}[\textbf{Combi}-1. ]
\item $\chi\,'(K_{m,n})=\Delta(K_{m,n})$.
\item $\chi_{st}(K_{m,n})=2m+n$ for $m\geq n\geq 1$.
\item $\chi\,'_s(K_{m,m})=m+2$ for $m\geq 2$; $\chi\,'_s(K_{m,n})=m+1$ for $m>n\geq 2$.
\item \cite{Bing-Yao-2020arXiv} $\chi\,''_{fdt}(K_{m,m}) =3m$.
\item \cite{Bing-Yao-2020arXiv} $\chi\,''_{fdt}(K_{m,n}) =2m+n$.
\item \cite{Yao-Wang-2106-15254v1} $\chi\,''_{ves}(K_{m,m}) \leq m+2$.
\item \cite{Yao-Wang-2106-15254v1} $\chi\,''_{ved}(K_{m,m}) \leq n+\lceil \frac{n}{2}\rceil $ with $m\leq n$.
\item \cite{Yang-Han-Ren-Yao-2016} $\chi\,''_{ast}(K_{m,m})=m+1$ if $m-n \geq 2$ or$ m > 2$ and $m > n = 1$; $\chi\,''_{ast}(K_{m,m})=m+2$ if $m-n =1$; $\chi\,''_{ast}(K_{m,m})=m+3$ if $m=n$. $\chi\,''_{(4)as}(K_{m,m})=m+4$
\item A bipartite graph $G$ holds $\chi\,''(G)\leq \chi(G)+\chi\,'(G) \leq\Delta(G)+2$.
\end{asparaenum}
\end{thm}

\begin{defn} \label{defn:111111}
\cite{Bing-Yao-arXiv:2207-03381} For a subgraph $G\subset K_{m,n}$, since the complete bipartite graph $K_{m,n}$ admits a graceful $(k,d)$-total coloring, also, $f$ is a $(k,d)$-total coloring of $G$, such that $f(E(G))\subset f(E(K_{m,n}))$, we call $f$ a \emph{fragmentary graceful $(k,d)$-total coloring}.\qqed
\end{defn}

\begin{rem}\label{rem:333333}
$^*$ A subgraph $G\subset K_{m,n}$ holds $f(V(G))= f(V(K_{m,n}))$, another graph $\overline{G}_{bip}=K_{m,n}-E(G)$ is called \emph{complementary} of the subgraph $G\subset K_{m,n}$ based on the complete bipartite graph. Clearly, $E(K_{m,n})=E(G)\cup E(\overline{G}_{bip})$ with $E(G)\cap E(\overline{G}_{bip})=\emptyset$ and $V(K_{m,n})=V(G)= V(\overline{G}_{bip})$. We call $\langle G,\overline{G}_{bip} \rangle $ a \emph{Key-pair based on a complete bipartite graph}, $G$ \emph{public-key graph}, and $\overline{G}_{bip}$ \emph{private-key graph}.\paralled
\end{rem}

\begin{lem}\label{lem:complete-bipartite-graph-spanning-trees}
Suppose that a complete bipartite graph $K_{m,n}$ admits a vertex coloring $f:V(K_{m,n})\rightarrow [1,m+n]$, such that $f(x)\neq f(y)$ for any pair of vertices $x,y\in V(K_{m,n})$, then the number of all spanning trees of this colored complete bipartite graph $K_{m,n}$ is $\tau (K_{m,n}) = m^{n-1}n^{m-1}$.
\end{lem}

\begin{problem}\label{question:techniques-KSTrees22}
Similarly with Problem \ref{question:techniques-KSTrees}, let $S_{tree}(K_{m,n})$ be the set of $m^{n-1}n^{m-1}$ different spanning trees of a colored complete bipartite graph $K_{m,n}$; refer to Lemma \ref{lem:complete-bipartite-graph-spanning-trees}. \textbf{Find} all Key-pairs of mutually edge-disjoint spanning trees $H^{pub}_{i}$ and $H^{pri}_{i}$ holding one or more of the following properties:
\begin{asparaenum}[\textbf{\textrm{Keypair}}-1.]
\item The topological signature authentication $H^{pub}_{i}[\odot_{prop}]H^{pri}_{i}\subset K_{m,n}$ contains a Hamilton cycle.
\item $H^{pub}_{i}\cong H^{pri}_{i}$ although $f(V\big (H^{pub}_{i}\big ))=[1,n]=f(V\big (H^{pri}_{i}\big ))$.
\item Let $n_d(G)$ be the number of vertices of degree $d$ in a graph $G$. Then
\begin{equation}\label{eqa:555555}
\sum _{3\leq d\leq \Delta \big (H^{pub}_{i}\big )}(d-2)n_d\big (H^{pub}_{i}\big )=\sum _{3\leq d\leq \Delta \big (H^{pri}_{i}\big )}(d-2)n_d\big (H^{pri}_{i}\big )
\end{equation}
\item Two degree sequences of spanning trees $H^{pub}_{i}$ and $H^{pri}_{i}$
$$
\textrm{deg}\big (H^{pub}_{i}\big )=(a_{i,1},a_{i,2},\dots, a_{i,m+n}),~\textrm{deg}\big (H^{pri}_{i}\big )=(b_{i,1},b_{i,2},\dots, b_{i,m+n})
$$ hold $a_{i,k}=b_{i,k}$ for $k\in [1,m+n]$ true.
\item Two diameters $D_{iam}\big (H^{pub}_{i}\big )=D_{iam}\big (H^{pri}_{i}\big )$.
\item Two spanning trees $H^{pub}_{i}$ and $H^{pri}_{i}$ are (i) two caterpillars; (ii) two lobsters; (iii) one caterpillar and one lobster.
\item Two spanning trees $H^{pub}_{i}$ and $H^{pri}_{i}$ holds $H^{pri}_{i}=H^{pub}_{i}+xy-uv$ with $xy\not \in E\big (H^{pub}_{i}\big )$ and $uv\in E\big (H^{pub}_{i}\big )$, that is a \emph{graph-operation homomorphism} $H^{pri}_{i}\rightarrow _{\pm e}H^{pub}_{i}$.
\end{asparaenum}
\end{problem}

\subsubsection{Key-pairs based on graph homomorphisms}

Graph homomorphism can be thought of as encrypting the topological structures appeared in our topological designment.

\begin{defn}\label{defn:definition-graph-homomorphism}
\cite{Bondy-2008} \textbf{The uncolored graph homomorphism}. A \emph{graph homomorphism} $G\rightarrow H$ from a graph $G$ into another graph $H$ is a mapping $f: V(G) \rightarrow V(H)$ such that $f(u)f(v)\in E(H)$ for each edge $uv \in E(G)$. \qqed
\end{defn}

\begin{defn} \label{defn:v-e-colored-graph-homomorphisms}
$^*$ \textbf{The colored graph homomorphism}. Suppose that a $(p,q)$-graph $G$ admits a proper total coloring $f$, and another graph $H_v$ admits a proper vertex-coloring $g_v$, and a graph $H_e$ admits a proper edge-coloring $g_e$, as well as a graph $H_{ve}$ admits a proper total coloring $g_{ve}$.

(i) The \emph{vertex-colored graph homomorphism} $H_v \rightarrow ^v_{prop}G$ is a vertex mapping $\alpha:V(H_v)\rightarrow V(G)$, such that $\alpha(u)\alpha(v)\in E(G)$ if and only if $uv\in E(H_v)$, and
$$g_v(u)=\alpha(u)\neq \alpha(v)=g_v(v)$$ for each edge $uv\in E(H_v)$.

(ii) The \emph{edge-colored graph homomorphism} $H_e \rightarrow ^e_{prop}G$ is a vertex mapping $\beta:V(H_e)\rightarrow V(G)$, so that $\beta(u)\beta(v)\in E(G)$ if and only if $uv\in E(H_e)$, and
$$g_e(uv)=\beta(u)\beta(v)\neq \beta(u)\beta(w)=g_e(uw)
$$ for any pair of adjacent edges $uv,uw\in E(H_e)$.

(iii) The \emph{totally colored graph homomorphism} $H_{ve} \rightarrow ^{ve}_{prop}G$ is a vertex mapping $\gamma:V(H_{ve})\rightarrow V(G)$, such that $\gamma(u)\gamma(v)\in E(G)$ if and only if $uv\in E(H_{ve})$, and
$$g_{ve}(u)=\gamma(u)\neq \gamma(v)=g_{ve}(v),~g_{ve}(uv)=\gamma(u)\gamma(v)\neq \gamma(u)\gamma(w)=g_{ve}(uw)
$$ for each edge $uv\in E(H_{ve})$ and any pair of adjacent edges $uv,uw\in E(H_{ve})$.\qqed
\end{defn}

\begin{figure}[h]
\centering
\includegraphics[width=16.4cm]{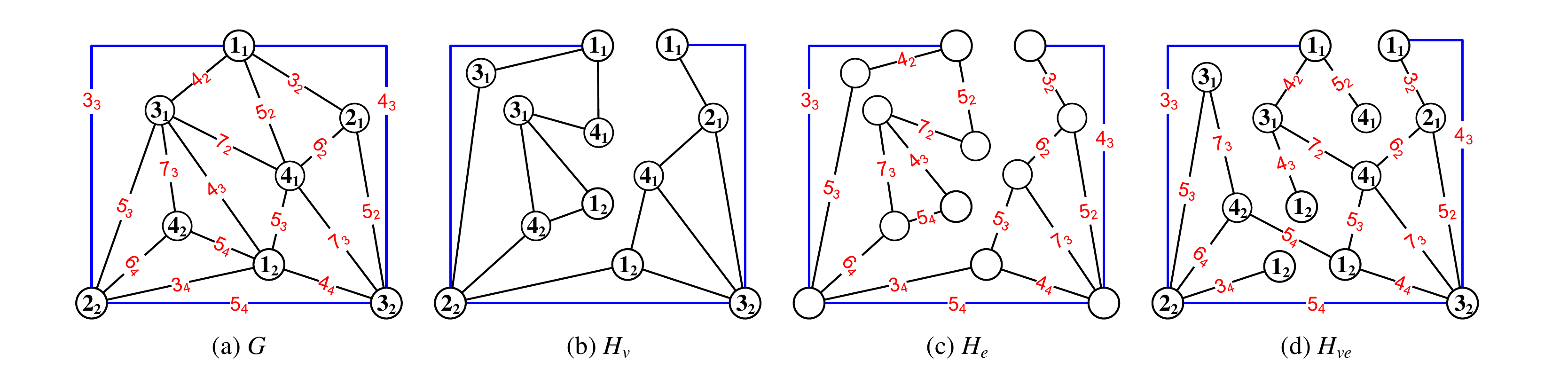}\\
\caption{\label{fig:colored-graph-homomorphism}{\small The examples $H_v \rightarrow ^v_{prop}G$, $H_e \rightarrow ^e_{prop}G$ and $H_{ve} \rightarrow ^{ve}_{prop}G$ for illustrating Definition \ref{defn:v-e-colored-graph-homomorphisms}.}}
\end{figure}

\begin{problem}\label{problem:ATE-CGH-problem00}
\textbf{ATE-CGH-problem (the asymmetric topology encryption related with colored graph homomorphisms)}.
\begin{quote} For a connected $(p,q)$-graph $G$, let $S_e(G)$ be the set of all graphs colored with proper edge-colorings and holding the colored graph homomorphism to $G$, that is, each graph $H_e\in S_e(G)$ holds $H_e \rightarrow ^e_{prop}G$ defined in Definition \ref{defn:v-e-colored-graph-homomorphisms}; and let $S_v(G)$ be the set of all graphs admitting proper vertex-colorings and holding the colored graph homomorphism to $G$, also, each graph $H_v\in S_v(G)$ holds $H_v \rightarrow ^v_{prop}G$ defined in Definition \ref{defn:v-e-colored-graph-homomorphisms}.

\textbf{Find} a graph $H_e\in S_e(G)$ admitting a proper edge-coloring $g_e$, and again \textbf{find} another graph $H_v\in S_v(G)$ admitting a proper vertex-coloring $g_v$, such that $H_e \rightarrow ^e_{prop}G$ and $H_v \rightarrow ^v_{prop}G$, and the connected $(p,q)$-graph $G$ admits a proper total coloring $f$ just induced by two total colorings $g_e$ and $g_v$.
\end{quote}
\end{problem}

\begin{problem}\label{problem:ATE-CGH-problem11}
In ATE-CGH-problem, the set $S_e(G)$ is as a \emph{public-key set}, and the set $S_v(G)$ is as a \emph{private-key set}. We have other questions for ATE-CGH-problem as follows:
\begin{asparaenum}[\textbf{\textrm{Pert}}-1.]
\item If we are asked for $|V(G)|< |V(H_e)|$ and $|V(G)|< |V(H_v)|$, the job of finding the public-key graph $H_e\in S_e(G)$ and the private-key graph $H_v\in S_v(G)$ will meet the Subgraph Isomorphic Problem, which is a NP-complete problem.
\item If we are asked for the proper total coloring $f$ of $G$ holding $f(V(G)\cup E(G))=[1,\chi\,''(G)]$, it may be a NP-type problem for searching the public-key graph $H_e\in S_e(G)$ and the private-key graph $H_v\in S_v(G)$, since it was conjectured: $\chi\,''(G)\leq \Delta(G)+2$ in 1965.
\item Suppose that $\max\{f(w):w\in V(G)\cup E(G)\}\leq \Delta(G)+2$. Let $V_i=\{u: f(u)=i,u\in V(G)\}$, $E_i=\{xy: f(xy)=i,xy\in E(G)\}$ and $S_i=V_i\cup E_i$ for $i\in [1,\Delta(G)+2]$. If $\big ||S_i|-|S_j|\big |\leq 1$ for $i,j\in [1,\Delta(G)+2]$, then $f$ is an \emph{equitable total coloring} defined in Definition \ref{defn:equitable-total-coloring-def}; refer to Conjecture \ref{conj:c4-Weifan-Wang-2002}.
\item If the proper total coloring $f$ is an edge-magic total labeling holding $f(V(G)\cup E(G))=[1,p+q]$ and the edge-magic constraint $f(u)+f(uv)+f(v)=\lambda$ for each edge $uv\in E(G)$, however, \textbf{find} the public-key graph $H_e\in S_e(G)$ and the private-key graph $H_v\in S_v(G)$. Similarly, consider other edge-difference constraint, graceful-difference constraint and felicitous-difference constraint for three colorings $f$, $g_e$ and $g_v$.
\item If the proper total coloring $f$ is a $W$-constraint $(k,d)$-total labeling, \textbf{find} three colorings $f$, $g_e$ and $g_v$.
\item By Definition \ref{defn:v-e-colored-graph-homomorphisms}, \textbf{find} a \emph{public-key graph sequence} $\{H^k_e\}^n_{k=1}$ with $H^k_e\in S_e(G)$ and another private-key graph sequence $\{H^k_v\}^n_{k=1}$ with $H^k_v\in S_v(G)$ such that
\begin{equation}\label{eqa:555555}
{
\begin{split}
&\{H^k_e\}^n_{k=1}\rightarrow ^e_{prop} G:~~H^{n}_e \rightarrow ^e_{prop}H^{n-1}_e \rightarrow ^e_{prop}\cdots \rightarrow ^e_{prop}H^{2}_e\rightarrow ^e_{prop}H^{1}_e \rightarrow ^e_{prop} G,\\
&\{H^k_v\}^n_{k=1}\rightarrow ^v_{prop} G:~~H^{n}_v \rightarrow ^v_{prop}H^{n-1}_v \rightarrow ^v_{prop}\cdots \rightarrow ^v_{prop}H^{2}_v\rightarrow ^v_{prop}H^{1}_v \rightarrow ^v_{prop} G
\end{split}}
\end{equation} Also, we write the \emph{totally-colored graph homomorphism} $\langle H^k_v,H^k_e\rangle \rightarrow ^{ve}_{prop} G$ for $k\in [1,n]$.
\end{asparaenum}
\end{problem}

\begin{problem}\label{question:Pan-ATE-CGH-problem}
\textbf{Pan-ATE-CGH-problem.}
\begin{quote}
Let $S_e(q)$ be the set of all graphs having $q$ edges and admitting proper edge-colorings, and let $S_v(q)$ be the set of all graphs having $q$ edges and admitting proper vertex-colorings. For a connected $(p,q)$-graph $G$, \textbf{find} a graph $H_e\in S_e(q)$ holding $H_e \rightarrow ^e_{prop}G$ defined in Definition \ref{defn:v-e-colored-graph-homomorphisms}, and \textbf{find} a graph $H_v\in S_v(q)$ holding $H_v \rightarrow ^v_{prop}G$ defined in Definition \ref{defn:v-e-colored-graph-homomorphisms}, such that the graph $G$ admits a $W$-constraint proper total coloring induced by the proper edge-coloring of $H_e\in S_e(q)$ and the proper vertex-coloring of $H_v\in S_v(q)$.
\end{quote}
\end{problem}

\begin{problem}\label{question:ATE-To-CGH-problem}
\textbf{ATE-To-CGH-problem (the asymmetric topology encryption related with totally-colored graph homomorphisms).}
\begin{quote}
Let $S_{ve}(q)$ be the set of all graphs having $q$ edges and admitting proper total colorings. For a connected $(p,q)$-graph $H$, \textbf{find} graphs $H^k_{ve}\in S_{ve}(q)$ holding $H^k_{ve}\rightarrow ^{ve}_{prop}H$ defined in Definition \ref{defn:v-e-colored-graph-homomorphisms}, such that the graphs $H$ and $H^k_{ve}$ admit the same $W$-constraint proper total coloring.
\end{quote}
\end{problem}

\begin{figure}[h]
\centering
\includegraphics[width=16.4cm]{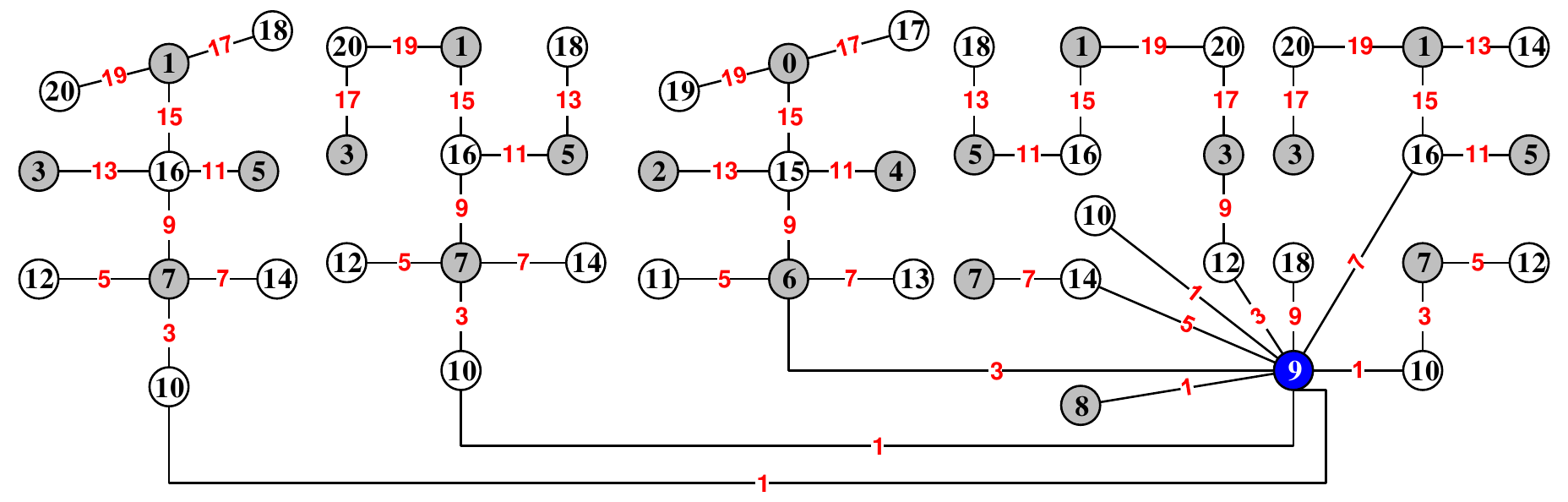}\\
\caption{\label{fig:11-pub-more-pri}{\small One public-key graph corresponding more private-key graphs forms a topological signature authentication $G=T[\odot]^4_{i=1}H_i$; refer to Fig.\ref{fig:top-signature-11}.}}
\end{figure}

\subsubsection{Public-key and private-key graphic lattices}

Let $\textbf{\textrm{T}}=(T_1,T_2,\dots ,T_m)$ be a \emph{tree base}, where mutually vertex-disjoint uncolored trees $T_1,T_2,\dots $, $T_m$ holds $T_i\not \cong T_j$ if $i\neq j$; and let $\textbf{\textrm{H}}=(H_1,H_2,\dots ,H_m)$ be another \emph{tree base}, where vertex-disjoint uncolored trees $H_1,H_2,\dots ,H_m$ holds $H_s\not \cong H_t$ if $s\neq t$. By Eq.(\ref{eqa:proper-operation-tree-base-lattice111}) we have
\begin{equation}\label{eqa:555555}
G=[\odot_{\textrm{prop}}]^m_{k=1}a_kT_k,~G\,'=[\odot_{\textrm{prop}}]^m_{k=1}b_kH_k
\end{equation} after doing the \emph{proper vertex-coinciding operation} ``$[\odot_{\textrm{prop}}]$'', so we get a \emph{tree-base $[\odot_{\textrm{prop}}]$-operation graphic lattice}
\begin{equation}\label{eqa:public-key-graphic-lattice}
\textbf{\textrm{L}}([\odot_{\textrm{prop}}]Z^0\textbf{\textrm{T}})=\big \{[\odot_{\textrm{prop}}]^m_{k=1}a_kT_k:~b_k\in Z^0,T_k\in \textbf{\textrm{T}}\big \}
\end{equation} with $\sum ^m_{k=1}a_k\geq 1$, refer to Eq.(\ref{eqa:proper-operation-tree-base-lattice}), and another \emph{tree-base $[\odot_{\textrm{prop}}]$-operation graphic lattice}
\begin{equation}\label{eqa:private-key-graphic-lattice}
\textbf{\textrm{L}}([\odot_{\textrm{prop}}]Z^0\textbf{\textrm{H}})=\big \{[\odot_{\textrm{prop}}]^m_{k=1}b_kH_k:~b_k\in Z^0,H_k\in \textbf{\textrm{H}}\big \}
\end{equation} with $\sum ^m_{k=1}b_k\geq 1$

\vskip 0.2cm

We consider $\textbf{\textrm{L}}([\odot_{\textrm{prop}}]Z^0\textbf{\textrm{T}})$ defined in Eq.(\ref{eqa:public-key-graphic-lattice}) as a \emph{public-key graphic lattice} and $\textbf{\textrm{L}}([\odot_{\textrm{prop}}]Z^0\textbf{\textrm{H}})$ defined in Eq.(\ref{eqa:private-key-graphic-lattice}) as a \emph{private-key graphic lattice}, where $\textbf{\textrm{T}}$ is a \emph{public-key tree base} and $\textbf{\textrm{H}}$ is a \emph{private-key tree base}.

If there is a topological signature authentication $A_{uth}\langle \textbf{\textrm{T}},\textbf{\textrm{H}}\rangle$ between two tree bases $\textbf{\textrm{T}}$ and $\textbf{\textrm{H}}$, then we get a topological signature authentication \begin{equation}\label{eqa:555555}
A_{uth}\langle G,G^*\rangle=A_{uth}\big \langle [\odot_{\textrm{prop}}]^m_{k=1}c_kT_k,~ [\odot_{\textrm{prop}}]^m_{k=1}c_kH_k\big \rangle
\end{equation} between two graphs $G\in \textbf{\textrm{L}}([\odot_{\textrm{prop}}]Z^0\textbf{\textrm{T}})$ and $G^*\in \textbf{\textrm{L}}([\odot_{\textrm{prop}}]Z^0\textbf{\textrm{H}})$.

\subsection{TKPDRA-center for local area networks}

TKPDRA-center has the functions of topological key production, distribution, replacement and topological signature authentication for local area networks.

In local area networks (LANs) or small communities in a network, it is difficult to require the users on nodes of LAN and communities to make and change various keys frequently. Some important organizations, such as the government and the military, have high security requirements. If some employee leaves or the keys is leaked, then the keys of the network must be revoked immediately. This will cause the entire department or unit to change all keys. This is the kind of problem that our TKPDRA-center can solve by the technique of encrypting the entire network; refer to Fig.\ref{fig:Topological-signature-22}, Fig.\ref{fig:string-graph-group-center-2} and Fig.\ref{fig:string-graph-group-center-11}.

In this new topic of encrypting networks, we employee graphic groups to make topology signatures and Key-pairs for large-scale networks in short time. Every-zero graphic groups can ba used to encrypt every network/graph, namely ``graph-to-graph'', and such graphic groups run in encrypting network/graphs are supported well by coloring theory of graph theory, such as traditional colorings, distinguishing colorings etc. It is noticeable, graphic groups can be used in machine encryption, or AI encryption for networks. We hope algebraic methods may solve some open problems of graph theory, for example, Graceful Tree Conjecture, and try to use algebraic methods for discovering possible regularity among labelings all in confusion \cite{Yao-Mu-Sun-Sun-Zhang-Wang-Su-Zhang-Yang-Zhao-Wang-Ma-Yao-Yang-Xie2019}.

\begin{figure}[h]
\centering
\includegraphics[width=16.4cm]{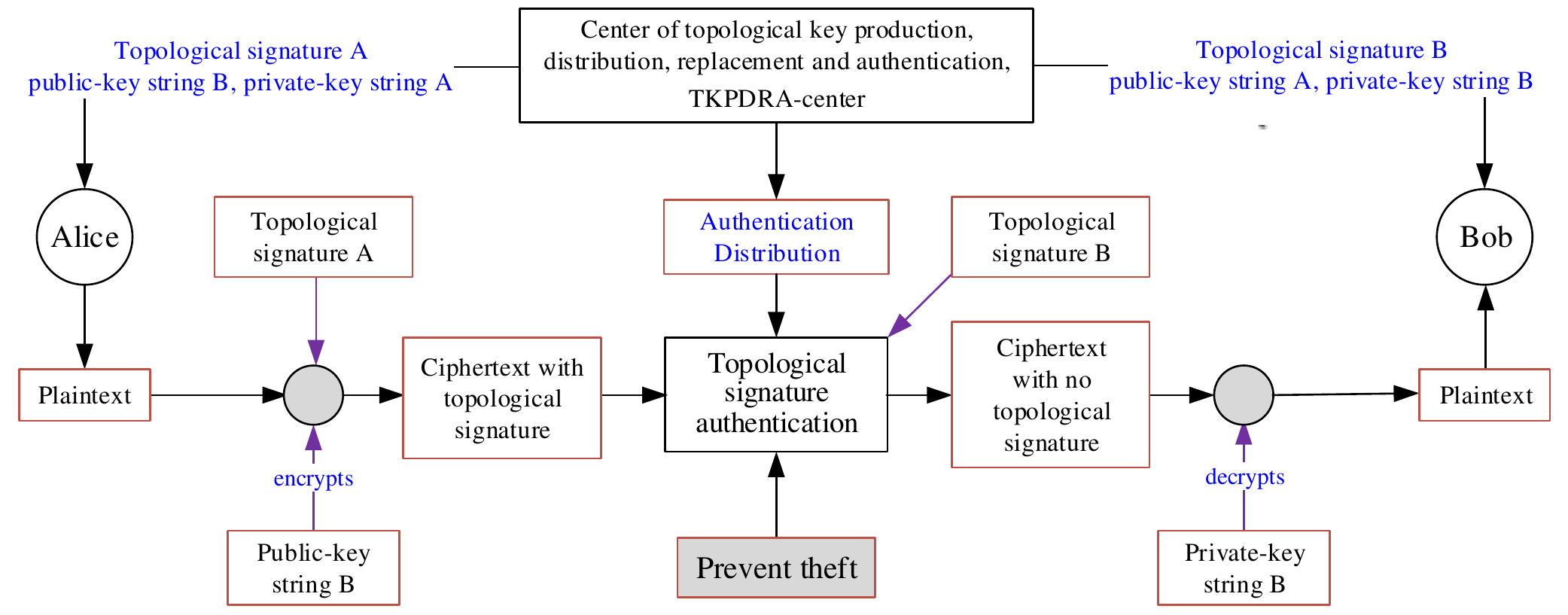}\\
\caption{\label{fig:Topological-signature-22}{\small The center of key-encryption distribution for local area networks and communities.}}
\end{figure}

\subsubsection{Functions of TKPDRA-center}

In general, TKPDRA-center keeps various groups introduced in the previous sections and topological technology support. For example,
\begin{asparaenum}[\textbf{\textrm{Group}} 1.]
\item Accordingly Definition \ref{defn:0-9-string-groups111},
\begin{equation}\label{eqa:555555}
s_{i}[\oplus \ominus_k]s_{j}:=s_{i}[\oplus] s_{j}[\ominus] s_{k}=s_{\lambda}\in S_{tring}(m,n)
\end{equation} and
\begin{equation}\label{eqa:555555}
s_{i}[\ominus \oplus_k]s_{j}:= s_{k}[\ominus]s_{i}[+] s_{j}=s_{\tau}\in S_{tring}(m,n)
\end{equation} there are every-zero $[0,9]$-string groups $\{F_{[+]k}(X);\oplus \ominus(\bmod~A)\}$ and $\{F_{[\ominus]k}(Y);\ominus \oplus (\bmod~B)\}$; refer to Theorem \ref{thm:first-mixed-string-operation} and Theorem \ref{thm:second-mixed-string-operation}.
\item In Theorem \ref{thm:super-string-groups00}, let $B=\max\{(9)_j:j\in [1,n]\}$ in the $(9)$s-number-based super-string set $S^n_{uper}(\{(9)_j\}^n_{j=1})$ defined in Definition \ref{defn:number-based-super-string-def}. Then there are two \emph{every-zero $(9)$s-uniformly super-string groups}
$$
\big \{F_{[+]k}(S^n_{uper}(\{m_j\}^n_{j=1}));\oplus \ominus(\bmod~B)\big \}\textrm{ and }\big \{F_{[\ominus]k}(S^n_{uper}(\{m_j\}^n_{j=1}));\ominus \oplus (\bmod~B)\big \}
$$ for $k\in [1,B]$.
\item In Definition \ref{defn:edge-index-graphic-group222}, a new Abelian additive operation is defined as
\begin{equation}\label{eqa:edge-index-graphic-group222}
T_{i,r}~[\ominus \oplus _k] ~T_{i,j}:=T_{i,k}\ominus (T_{i,r}\oplus T_{i,j})=T_{i,\mu}\in \textbf{\textrm{T}}_i
\end{equation} with the index $\mu=k-(r+j)~(\bmod~M)$ for any preappointed \emph{zero} $T_{i,k}\in \textbf{\textrm{T}}_i$. So, we get an \emph{every-zero edge-index graphic group} $\{F^-_M(\textbf{\textrm{T}}_i);\ominus \oplus\}$ defined by the new Abelian operation $T_{i,r}~[\ominus \oplus _k] ~T_{i,j}$ based on a colored graph base $\textbf{\textrm{T}}_i$, where $\textbf{\textrm{T}}_i$ is a graphic group.
\item In the GROUP-compound algorithm, there are the graphic-group $\{F_m(G,f);\oplus \ominus\}$, Topcode-matrix group $\{F_m(T_{code}(G),f);\oplus \ominus\}$, string-group $\{F_m(S_r);\oplus \ominus\}$.
\item The complex graphic group $C_{om}(G[\oplus \ominus]\{f_{i}\}^n_{i=1})=\bigcup ^n_{i=1}\{F_{m_i}(G,f_i);\oplus \ominus\}$; refer to Eq.(\ref{eqa:complex-graphic-group}).
\item In Eq.(\ref{eqa:compound-string-groups22}), there are compound number-based string groups
$$C_{om}(FTS[\oplus \ominus]\gamma)=\bigcup^n_{i=1} S_{m_i}(T_{code}(G_{i,r}),\gamma)
$$ with $\gamma\in [1,(3q)!]$.
\item The every-zero thing-index group $\{F^+_m(S_{thing});\oplus \ominus\}$ defined in Definition \ref{defn:every-zero-abstract-group}
\item In Remark \ref{rem:infinite-graphic-sequence-group}, $F^*(G,f)=F(\{\{G_{s,k}\}^{+\infty}_{-\infty}\}^{+\infty}_{-\infty};\oplus \ominus;(G,f))$ is an every-zero infinite graphic-sequence group. We have:

\qquad (i) The elements of the group $F^*(G,f)$ can tile fully each point $(x,y)$ of $xOy$-plane.

\qquad (ii) $F^*(G,f)$ contains infinite every-zero graphic groups having finite elements, such as
$$F\big (\{G_{s+i,k}\}^{p_W}_{i=1};[\oplus \ominus];(G,f)\big )\textrm{ and }F\big (\{G_{s,k+j}\}^{q_W}_{i=1};\oplus \ominus;(G,f)\big )$$

\qquad (iii) $F^*(G,f)$ contains infinite every-zero graphic groups of infinite elements.

\item In Definition \ref{defn:every-zero-graphic-group-homomorphism}, we have an every-zero graphic group homomorphism
$$\{F_f(G);\oplus \ominus\}\rightarrow \{F_h(H);\oplus \ominus\}
$$ defined by two every-zero graphic groups $\{F_f(G);\oplus \ominus\}$ based on a graph set $F_f(G)=\{G_i\}^m_1$ and $\{F_h(H);\oplus \ominus\}$ based on a graph set $F_h(H)=\{H_i\}^m_1$, when as there are graph homomorphisms $G_i\rightarrow H_i$ defined by $\theta_i:V(G_i)\rightarrow V(H_i)$ with $i\in [1,m]$.
\end{asparaenum}

\vskip 0.4cm

We list part of functions of TKPDRA-center as follows:
\begin{asparaenum}[\textbf{\textrm{Function}}-1.]
\item TKPDRA-center sends Key-packages to the nodes of a network based on a particular group $G_{roup}$, and adjusts randomly the zeros of the group $G_{roup}$; TKPDRA-center is responsible for the communication topological signature authentication of each pair of nodes in the network.
\item Each node in the network does not need to make its own Key-pairs and topology signatures, so that users on each node do not need to learn the technique of making Key-pairs and topology signatures. For Chinese users, they can get complex keys and topology signatures as long as they speak and write Chinese.
\item TKPDRA-center can makes the protection mechanism of multiple topology authentications; refer to Fig.\ref{fig:string-graph-group-center-2}.
\item Randomly replace the common zeros of a graphic group, or a string group, or other groups for preventing malicious attacks to the networks.
\item TKPDRA-center can timely change the security system of the entire network to resist external damage and attack.
\item TKPDRA-center can customize the authentication mode for some pairs of nodes in the network.
\end{asparaenum}

\begin{figure}[h]
\centering
\includegraphics[width=16.4cm]{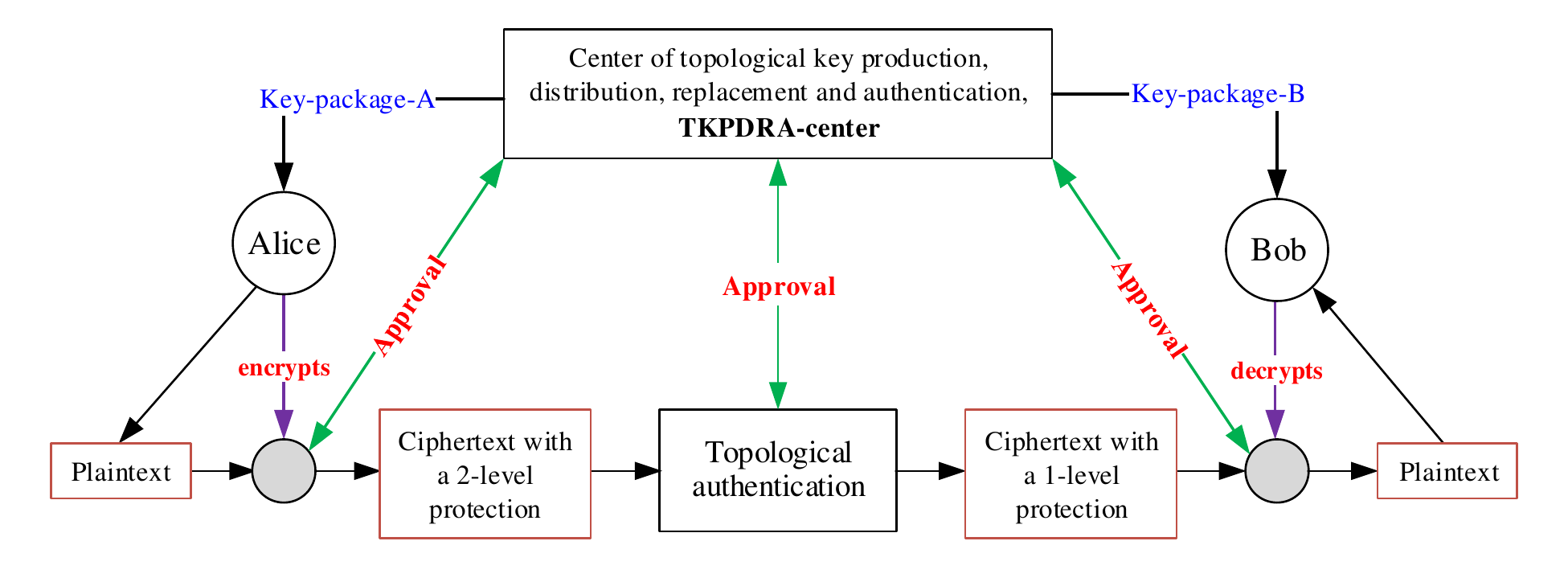}\\
\caption{\label{fig:string-graph-group-center-2}{\small The multiple topological authentication mechanism of TKPDRA-center.}}
\end{figure}

\vskip 0.4cm

\noindent \textbf{TKPDRA-center algorithm.}
\begin{asparaenum}[\textbf{\textrm{TKPDRA}}-1. ]
\item Alice uses her private-key graph $G_{Apri}$ and private-key string $s_{Apri}$ to encrypt a plaintext $F$, the resultant encrypted file is denoted as $F_2$, and Alice sends the encrypted file $F_2$ to TKPDRA-center for asking dealing with $F_2$ and send to Bob.
\item TKPDRA-center encrypts the encrypted file $F_2$ by Bob's public-key graph $G_{Bpub}$ and Bob's public-key string $s_{Bpub}$, the resultant encrypted file is denoted as $F_4$.
\item TKPDRA-center sends to Bob a package including the encrypted file $F_4$, Alice's public-key graph $G_{Apub}$ and Alice's public-key string $s_{Apub}$.
\item Bob, first, decrypts the encrypted file $F_4$ by his private-key graph $G_{Bpri}$, and knows that the encrypted file $F_4$ is really for him, the resultant encrypted file is denoted as $F\,^*_3$.
\item Next, Bob uses Alice's public-key graph $G_{Apub}$ to decrypt the encrypted file $F\,^*_3$, and knows that the encrypted file $F\,^*_3$ is sent from Alice, the resultant encrypted file is denoted as $F\,^*_2$.
\item Then Bob uses Alice's public-key string $s_{Apub}$ to decrypt the encrypted file $F\,^*_2$, the resultant encrypted file is denoted as $F\,^*_1$.
\item Finally, Bob decrypts the encrypted file $F\,^*_1$ by his own private-key string $s_{Bpri}$ for obtaining the original plaintext $F$.
\end{asparaenum}

\vskip 0.4cm

Part of advantages of the TKPDRA-center algorithm are as follows:

(i) TKPDRA-center algorithm has high security, since Alice's public-key graph $G_{Apub}$, Bob's public-key graph $G_{Bpub}$ and Bob's public-key string $s_{Bpub}$ were kept in TKPDRA-center.

(ii) Alice only saves her private-key string and her private-key graph, and the same is true for Bob.

(iii) TKPDRA-center help Alice to encrypt the encrypted file $F_1$, and sends the encrypted files to two or more people if Alice asks for.

\subsubsection{String-string-group Key-pairs}

\textbf{The String-string-group Key-pair algorithm (SSG-KP-algorithm)} is based on an every-zero string group $S_{roup}$ for a preappointed common zero $s^*_{zero}\in S_{roup}$ in TKPDRA-center.
\begin{asparaenum}[\textbf{\textrm{SSG-step-I}} (1)]
\item Alice receives the Key-package-A sent by TKPDRA-center, where Key-package-A contains Alice's private-key string $s_{Apri}$ and Bob's public-key string $s_{Bpub}$, and these strings are selected randomly from an every-zero string group $S_{roup}$ in TKPDRA-center; refer to Fig.\ref{fig:string-graph-group-center-2} and Fig.\ref{fig:string-graph-group-center-11}.
\item Bob receives the Key-package-B sent by TKPDRA-center, where Key-package-B contains Bob's private-key string $s_{Bpri}$ and Alice's public-key string $s_{Apub}$, and these strings are selected randomly from the every-zero string group $S_{roup}$ in TKPDRA-center; refer to Fig.\ref{fig:string-graph-group-center-2} and Fig.\ref{fig:string-graph-group-center-11}.
\item Alice encrypts a plaintext $P_{lan}$ by Alice's private-key string $s_{Apri}$ and Bob's public-key string $s_{Bpub}$, the encrypted file is denoted as $D_{ocum}=\langle P_{lan}$, $s_{Apri}$, $s_{Bpub}\rangle$ with a 2-level protection.
\item First, Bob decrypts the encrypted file $D_{ocum}$ sent from Alice by Alice's public-key string $s_{Apub}$ through the topological signature authentication $S_{tringA}=A_{uth}\langle s_{Apub}$, $s_{Apri}\rangle$ (as the identity authentication), where the result of the operation $s_{Apub}[\oplus]s_{Apri}[\ominus]s^*_{zero}=S_{tringA}\in S_{roup}$ for a preappointed common zero $s^*_{zero}\in S_{roup}$.
\item If $S_{tringA}=A_{uth}\langle s_{Apub},s_{Apri}\rangle$ (has been approved by TKPDRA-center, see Fig.\ref{fig:string-graph-group-center-11}) successfully decrypts the encrypted file $D_{ocum}$ out the first protection, the resultant file is denoted as $D^*_{ocum}$ with one-protection only.
\item Bob uses his private-key string $s_{Bpri}$ to decrypt $D^*_{ocum}$ protected only by Bob's public-key string $s_{Bpub}$. As the string authentication $S_{tringB}=A_{uth}\langle s_{Bpub},s_{Bpri}\rangle$ is approval after the result $s_{Bpub}[\oplus]s_{Bpri}[\ominus]s^*_{zero}=S_{tringB}\in S_{roup}$ for a preappointed common zero $s^*_{zero}\in S_{roup}$, then Bob gets the original plaintext $P_{lan}$ finally.
\end{asparaenum}

\begin{figure}[h]
\centering
\includegraphics[width=16.4cm]{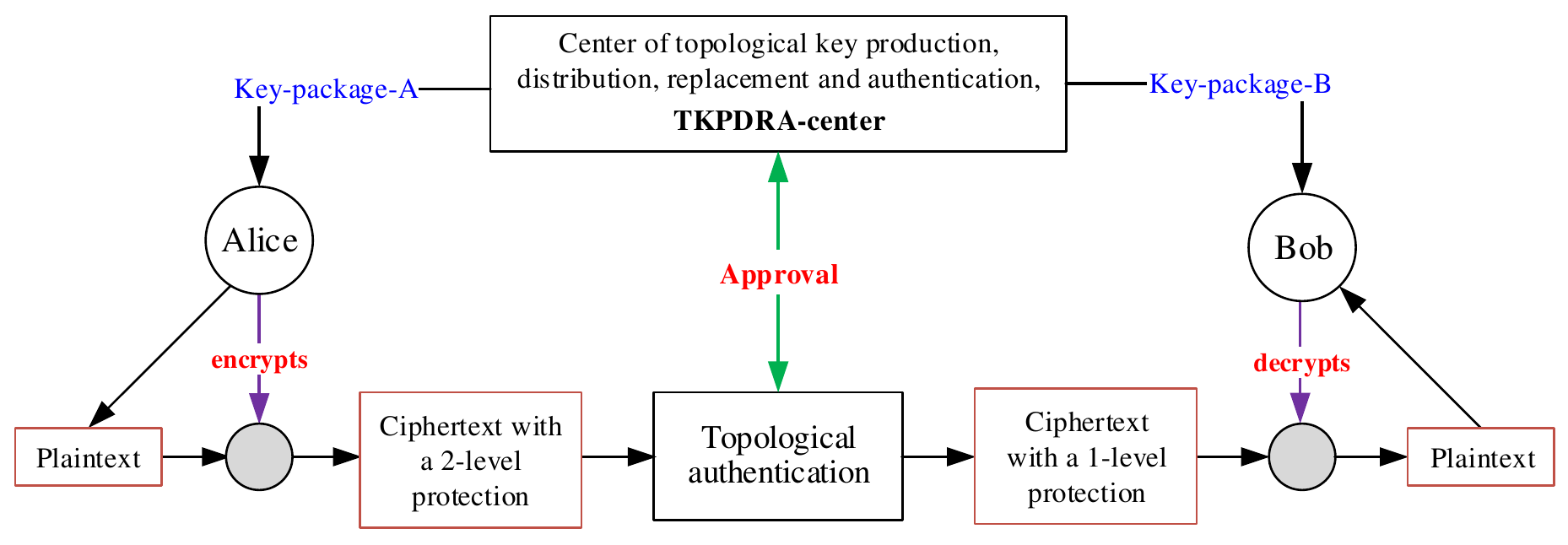}\\
\caption{\label{fig:string-graph-group-center-11}{\small The single topological authentication mechanism of TKPDRA-center.}}
\end{figure}

\subsubsection{Graphic-string-group Key-pairs}

\textbf{Graphic-string-group Key-pair algorithm (GSG-KP-algorithm)} is based on an every-zero graphic group $G_{roup}$ for a preappointed common zero $G^*_{zero}\in G_{roup}$ and an every-zero string group $S_{roup}$ for a preappointed common zero $s^*_{zero}\in S_{roup}$ in TKPDRA-center.
\begin{asparaenum}[\textbf{\textrm{GSG-step-II}} (1)]
\item Alice receives the Key-package-A sent by TKPDRA-center, where the Key-package-A contains Alice's private-key graph $G_{Apri}$, Alice's private-key string $s_{Apri}$ and Bob's public-key string $s_{Bpub}$, where these strings are in an every-zero string group $S_{roup}$ and the public-key graph $G_{Apub}$ is in the every-zero graphic group $G_{roup}$.
\item Bob receives the Key-package-B sent by TKPDRA-center, where Key-package-B contains Alice's public-key graph $G_{Apub}$, Bob's private-key string $s_{Bpri}$ and Alice's public-key string $s_{Apub}$, where these strings are in the every-zero string group $S_{roup}$ and the public-key graph $G_{Bpub}$ is in the every-zero graphic group $G_{roup}$.
\item Alice encrypts a plaintext $P_{lan}$ by Alice's private-key graph $G_{Apri}$ and Bob's public-key string $s_{Bpub}$, the encrypted file is denoted as $D_{ocum}=\langle P_{lan}$, $G_{Apub}$, $s_{Bpub}\rangle$ with a 2-level protection.
\item First, Bob decrypts the encrypted file $D_{ocum}$ sent from Alice by Alice's public-key graph $G_{Apub}$ through the topological signature authentication $T_{sigA}=A_{uth}\langle G_{Apub}$, $G_{Apri}\rangle$ (as the identity authentication), where the result of the operation $G_{Apub}[\oplus]G_{Apri}[\ominus]G^*_{zero}=T_{sigA}\in G_{roup}$ for a preappointed common zero $G^*_{zero}\in G_{roup}$.
\item If $T_{sigA}=A_{uth}\langle G_{Apub},G_{Apri}\rangle$ (has been approved by TKPDRA-center, see Fig.\ref{fig:string-graph-group-center-11}) successfully decrypts the encrypted file $D_{ocum}$ out the \emph{topological signature protection}, the resultant file is denoted as $D^*_{ocum}$ with the string protection only.
\item Bob uses his private-key string $s_{Bpri}$ to decipher $D^*_{ocum}$ protected only by Bob's public-key string $s_{Bpub}$. As TKPDRA-center's Key-pair string authentication $S_{tringB}=A_{uth}\langle s_{Bpub}$, $s_{Bpri}\rangle$ is approval after the operation $s_{Bpub}[\oplus]s_{Bpri}[\ominus]s^*_{zero}=S_{tringB}\in S_{roup}$ for a preappointed common zero $s^*_{zero}\in S_{roup}$, then Bob obtains the original plaintext $P_{lan}$ finally.
\end{asparaenum}

\begin{rem}\label{rem:333333}
TKPDRA-center can customize the individual zeros $s^{A}_{zero}, s^{B}_{zero}, s^{A,B}_{zero}\in S_{roup}$ for Alice and Bob such that
$${
\begin{split}
&s_{Apub}[\oplus]s_{Apri}[\ominus]s^A_{zero}\in S_{roup},~s_{Bpub}[\oplus]s_{Bpri}[\ominus]s^B_{zero}\in S_{roup}\\
&s^A_{zero}[\ominus]s_{Apub}[\oplus]s_{Apri}\in S_{roup},~s^B_{zero}[\ominus]s_{Bpub}[\oplus]s_{Bpri}\in S_{roup}
\end{split}}
$$
or
$${
\begin{split}
&s_{Apub}[\oplus]s_{Apri}[\ominus]s^{A,B}_{zero}\in S_{roup},~s_{Bpub}[\oplus]s_{Bpri}[\ominus]s^{A,B}_{zero}\in S_{roup}\\
&s^{A,B}_{zero}[\ominus]s_{Apub}[\oplus]s_{Apri}\in S_{roup},~s^{A,B}_{zero}[\ominus]s_{Bpub}[\oplus]s_{Bpri}\in S_{roup}
\end{split}}
$$ Moreover, TKPDRA-center can customize the individual zeros $G^{A}_{zero}, G^{B}_{zero}, G^{A,B}_{zero}\in G_{roup}$ for both Alice and Bob, such that $G_{Apub}[\oplus]G_{Apri}[\ominus]G^A_{zero}\in G_{roup}$, or $G_{Bpub}[\oplus]G_{Bpri}[\ominus]G^B_{zero}\in G_{roup}$, or $G_{Apub}[\oplus]G_{Apri}[\ominus]G^{A,B}_{zero}\in G_{roup}$, or $G_{Bpub}[\oplus]G_{Bpri}[\ominus]G^{A,B}_{zero}\in G_{roup}$, and so on.\paralled
\end{rem}

\subsubsection{Graphic-graphic-group Key-pairs}

\textbf{Graphic-graphic-group Key-pairs (GGG-KP-algorithm)} is based on an every-zero graphic group $G_{roup}$ for a preappointed common zero $G^*_{zero}\in G_{roup}$ and another every-zero graphic group $H_{roup}$ for a preappointed common zero $H^*_{zero}\in H_{roup}$ in TKPDRA-center.
\begin{asparaenum}[\textbf{\textrm{GGG-step-III}} (1)]
\item Alice receives Key-package-A sent by TKPDRA-center, where Key-package-A contains Alice's private-key graph $G_{Apri}\in G_{roup}$ and Bob's public-key graph $H_{Bpub}\in H_{roup}$.
\item Bob receives the Key-package-B sent by TKPDRA-center, where Key-package-B contains Alice's public-key graph $G_{Apub}\in G_{roup}$ and Bob's private-key graph $H_{Bpri}\in H_{roup}$.
\item Alice encrypts a plaintext $P_{lan}$ by Alice's private-key graph $G_{Apri}$ and Bob's public-key graph $H_{Bpub}$, the encrypted file denoted as $D_{ocum}=\langle P_{lan}$, $G_{Apri},H_{Bpub}\rangle$ with a 2-level protection consisted of two topological signatures.

\item First, Bob decrypts the encrypted file $D_{ocum}$ sent from Alice by Bob's private-key graph $H_{Bpri}$ through TKPDRA-center's topological signature authentication $T_{sigB}=A_{uth}\langle H_{Bpri}$, $H_{Bpub}\rangle$ made by the operation $H_{Bpub}[\oplus]H_{Bpri}[\ominus]H^*_{zero}=T_{sigB}\in H_{roup}$ for a preappointed common zero $H^*_{zero}\in H_{roup}$, the resultant file is denoted as $D^*_{ocum}$.
\item Bob uses Alice's public-key graph $G_{Apub}$ to decipher the encrypted file $D^*_{ocum}$ protected only by Alice's private-key graph $G_{Apri}$ (as the identity authentication). As the topological signature authentication $T_{sigA}=A_{uth}\langle G_{Apub}$, $G_{Apri}\rangle$ has been approved by the operation $G_{Apub}[\oplus]G_{Apri}[\ominus]G^*_{zero}=T_{sigA}\in G_{roup}$ for a preappointed common zero $G^*_{zero}\in G_{roup}$, then Bob can read the original plaintext $P_{lan}$ directly.
\end{asparaenum}

\begin{rem}\label{rem:333333}
There are many requirements and standards in the topological signature authentications
\begin{equation}\label{eqa:555555}
T_{sigB}=A_{uth}\langle H_{Bpri},H_{Bpub}\rangle,~T_{sigA}=A_{uth}\langle G_{Apub},G_{Apri}\rangle
\end{equation}
and the Abelian additive operations
\begin{equation}\label{eqa:555555}
{
\begin{split}
&H_{Bpub}[\oplus]H_{Bpri}[\ominus]H^*_{zero}\in H_{roup},~G_{Apub}[\oplus]G_{Apri}[\ominus]G^*_{zero}\in G_{roup}\\
&H^*_{zero}[\ominus]H_{Bpub}[\oplus]H_{Bpri}\in H_{roup},~G^*_{zero}[\ominus]G_{Apub}[\oplus]G_{Apri}\in G_{roup}
\end{split}}
\end{equation} related with topological structures and mathematical constraints.\paralled
\end{rem}

\subsection{Miscellaneous topics}

\subsubsection{Curve-attached strings}

We cite a sentence: \emph{A key signature of human intelligence is the ability to make `infinite use of finite means'}, as the beginning of an article entitled ``\emph{Relational inductive biases, deep learning, and graph networks}'' by Battaglia \emph{et al.} in \cite{Battaglia-27-authors-arXiv1806-01261v2}, and they have pointed out: ``\emph{in which a small set of elements (such as words) can be productively composed in limitless ways (such as into new sentences)}'', and they argued that \emph{combinatorial generalization} must be a top priority for AI to achieve human-like abilities, and that structured representations and computations are key to realizing this object. As an example of supporting `infinite use of finite means', \emph{self-similarity} is common phenomena between a part of a complex system and the whole of the system.

In Definition \ref{defn:plane-coordinate-string-sequence}, a \emph{plane-curve-attached string} $s(k,d)$ is induced by a plane curve $\theta(x,y)=0$ and a parameterized Topcode-matrix $P_{(k,d)}(G,F)$ defined in Eq.(\ref{eqa:definition-parameterized-topcode-matrix}) of Definition \ref{defn:bipartite-parameterized-topcode-matrix}.

We have the parameterized Topcode-matrix $P_{(k,d)}(T,F)$ obtained from the Topcode-matrix $T_{code}(T,f)$ of the colored $(p,q)$-graph $T$ shown in Fig.\ref{fig:top-signature-11} as follows:
\begin{equation}\label{eqa:v11111}
\centering
{
\begin{split}
P_{(k,d)}(T,F)=k\cdot I\,^0+d\cdot T_{code}(T,f)
\end{split}}
\end{equation} with both $k$ and $d$ are integers, where $F$ is a parameterized odd-graceful labeling and $f$ is an odd-graceful labeling of the graph $T$.

We take a plane curve $\theta(x,y)=0$ defined on a domain $[\alpha,+\infty]^{r}$ for $0\leq \alpha$, and take integers $k_n,d_n\geq 0$ holding $\theta(k_n,d_n)=0$ for $n\in [1,m]$, then we get a plane-curve-attached string sequence $\big \{s_i(k_n,d_n)\big \}^m_{n=1}$ for $i\in [1,(3q)!]$ based on the parameterized Topcode-matrix $P_{(k,d)}(T,F)$ and the plane curve $\theta(x,y)=0$, so the parameterized number-based strings are as follows
$$s_i(k_n,d_n)=b_{i,1}(k_n,d_n)b_{i,2}(k_n,d_n)\cdots b_{i,3q}(k_n,d_n),~i\in [1,(3q)!]
$$ generated from the parameterized Topcode-matrix $P_{(k_n,d_n)}(T,F)$ and $\theta(k_n,d_n)=0$ for $n\in [1,m]$.

\begin{example}\label{exa:8888888888}
For each curve $y=\sum^n_{k=0}a_kx^k$ with each integer $a_k\geq 0$ and $a_n\neq 0$, we can get infinite integer points $(x,y)$ as $x$ are integers. Notice that curves $x^n+y^n=z^n$ has no solutions of integers when $n\geq 3$.\qqed
\end{example}

Since there are infinite plane curves, and there are infinite non-negative integer points in each plane curve (refer to Theorem \ref{thm:one-encryption-one-time} and Theorem \ref{thm:one-encryption-one-time11}), and there are various $W$-constraint $(k,d)$-total colorings, so we are able to make the \emph{one-encryption one-time} for users and the realization of ``one-time one-secret password'' or ``one-time multi-secret passwords'' or ``multiple multi-secret passwords''.

We have infinite graphic groups for supporting the one-encryption one-time. In the MULTIPLE-JOIN algorithm, for example,

(i) By infinite graphic groups, we get every-zero infinite graphic groups $(\textbf{\textrm{H}}_{total})^{+\infty}_{-\infty}|^2$, and moreover there are every-zero infinite graphic group homomorphism $$(\textbf{\textrm{H}}_{total})^{+\infty}_{-\infty}|^2\rightarrow (\textbf{\textrm{G}}_{total})^{+\infty}_{-\infty}|^2$$

(ii) Since there are infinite graphs $T^r_{s,k}$ admitting total coloring $\theta^r_{i,j}$ equivalent with $h_{i,j}$ defined in $(\textbf{\textrm{H}}_{total})^{+\infty}_{-\infty}|^2$, thus, we have many every-zero infinite graphic groups $(\textbf{\textrm{T}}^r_{total})^{+\infty}_{-\infty}|^2=\{T^r_{s,k}: s,k\in Z\}$ with $r\geq 1$, and get mutually equivalent every-zero infinite graphic groups by $$(\textbf{\textrm{H}}_{total})^{+\infty}_{-\infty}|^2~\sim ~(\textbf{\textrm{T}}^r_{total})^{+\infty}_{-\infty}|^2$$

\begin{thm}\label{thm:666666}
$^*$ By the \emph{$n$-dimension curve-attached $W$-constraint total colorings} defined in Definition \ref{defn:n-dimensioncurve-attached-colorings}, which induce \emph{$n$-dimension curve-attached $W$-constraint number-based strings}, there are infinite $n$-dimension curves for producing $n$-dimension curve-attached $W$-constraint number-based strings, which can be used to realize one-encryption one-time in asymmetric topology encryption.
\end{thm}

\subsubsection{Self-certification topology algorithms}

People customize some simple algorithms to encrypt their data. Usually, these custom algorithms make data more secure because they are different from common encryption algorithms.

For example, the users can replace topological signatures and topological Key-pairs by themselves, not letting the third party know. This is based on topological Key-pairs are irreversible, so that finding the topological signatures from the topological Key-pairs is impossible.

Another advantage of asymmetric topology encryption is its ability to serve as identification. If you encode a known string of data with your private-key string and attach it to a message, anyone who receives that message can decode it with your public-key string. Therefore, the encryption serves as a fingerprint, since only your private-key string could have encrypted the data. As known, the self-certified topology cryptography realizes the properties of no public-key certificate and the unmanaged public-key.

We design some self-certification topology algorithms with no third party and no managed public-key strings in this subsection.

\vskip 0.4cm

\textbf{Self-certification topology algorithm-I. }
\begin{asparaenum}[\textbf{\textrm{Self-1.}}1.]
\item Alice sends to Bob a key-package-A including Alice's public-key string $s_{Apub}$ and Alice's public-key graph $G_{Apub}$ (as the identity authentication).
\item Bob sends to Alice a key-package-B including Bob's public-key string $s_{Bpub}$ and Bob's public-key graph $G_{Bpub}$.
\item Alice encrypts the plaintext $F$ by Bob's public-key string $s_{Bpub}$, and gets the encrypted $F\,^*_1$, and then Alice encrypts the encrypted file $F\,^*_1$ by her private-key graph $G_{Apri}$, the resultant file is denoted as $F\,^*_2$ with a 2-level protection.
\item Bob, first, decrypts the encrypted file $F\,^*_2$ by Alice's public-key graph $G_{Apub}$, as the topological signature authentication $T_{sigA}=A_{uth}\langle G_{Apub}, G_{Apri}\rangle$ passed, the resultant file is denoted as $F\,'_1$ with a one-level protection only.
\item Bob uses his private-key string $s_{Bpri}$ to decrypt the encrypted file $F\,'_1$, as the Key-pair string authentication $A_{uth}\langle s_{Bpub}, s_{Bpri}\rangle$ passed, then Bob gets the original plaintext $F$.
\end{asparaenum}

\begin{rem}\label{rem:333333}
Since Alice's public-key graph $G_{Apub}$ and Bob's private-key string $s_{Bpri}$ are impossible to decipher, the self-certification topology algorithm-I has a high security; refer to Fig.\ref{fig:certification-no-3-party}.\paralled
\end{rem}

\begin{figure}[h]
\centering
\includegraphics[width=16.4cm]{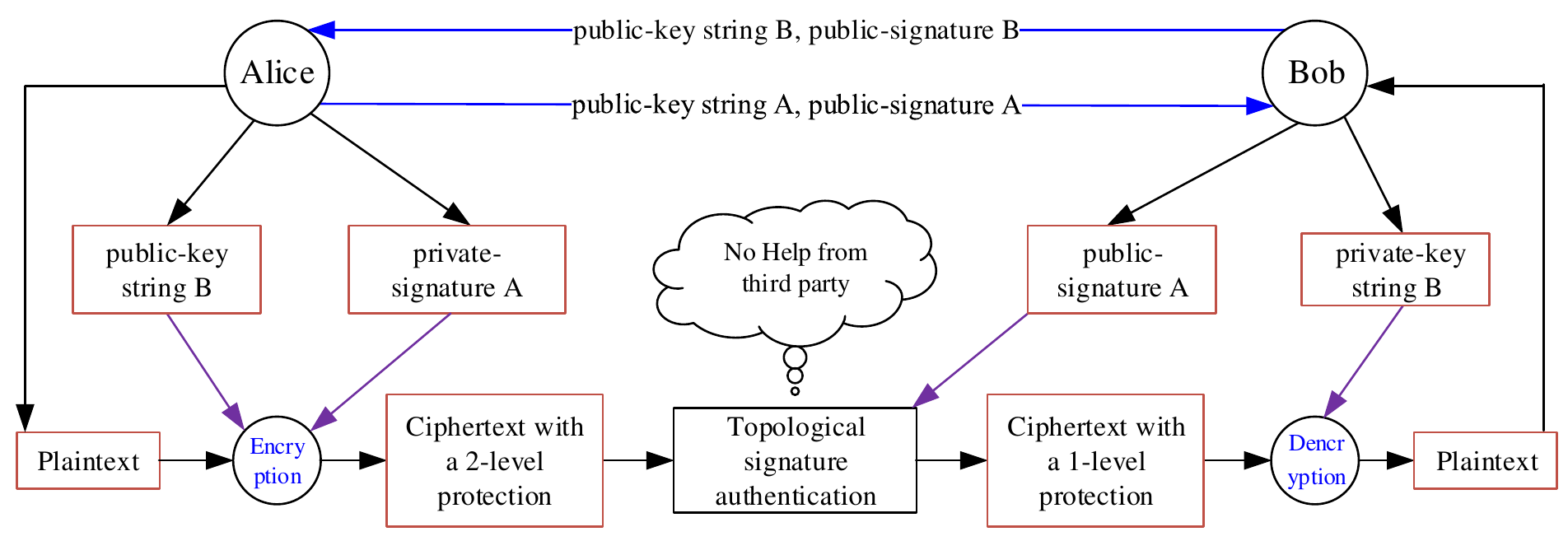}\\
\caption{\label{fig:certification-no-3-party}{\small A scheme for the self-certification algorithm with no third party.}}
\end{figure}

\textbf{Self-certification topology algorithm-II. }
\begin{asparaenum}[\textbf{\textrm{Self-2.}}1.]
\item Alice sends to Bob a key-package-A including Alice's public-key string $s_{Apub}$ and Alice's public-key graph $G_{Apub}$ (as the identity authentication).
\item Bob sends to Alice a key-package-B including Bob's first public-key string $s^1_{Bpub}$, Bob's second public-key string $s^2_{Bpub}$ and Bob's public-key graph $G_{Bpub}$.
\item Alice, first, encrypts the plaintext $F$ by Bob's second public-key string $s^2_{Bpub}$, Alice's private-key string $s_{Apri}$ and Alice's private-key graph $G_{Apri}$, the resultant encrypted file is denoted as $F_{3}$ with a $3$-level protection.
\item Alice makes Alice's public-key string $s^*_{Apub}$ by a assignment technique based one Bob's first public-key string $s^1_{Bpub}$, and encrypts $F_{3}$ by $s^*_{Apub}$, the resultant encrypted file is denoted as $F\,^*_{4}$ with a $4$-level protection.
\item Alice sends to Bob the encrypted file $F\,^*_{4}$ and Alice's public-key string $s^*_{Apub}$.
\item Bob, first, decrypts the encrypted file $F\,^*_{3}$ by Alice's public-key graph $G_{Apub}$, as the topological signature authentication $T_{sigA}=A_{uth}\langle G_{Apub}, G_{Apri}\rangle$ passed, the resultant encrypted file is denoted as $F\,^*_3$ with a 3-level protection.
\item Next, Bob decrypts the encrypted file $F\,^*_{3}$ by Alice's public-key string $s^*_{Apub}$ and Bob's first private-key string $s^1_{Bpri}$, as the Key-pair string authentication $A_{uth}\big \langle s^*_{Apub}, s^1_{Bpub}, s^1_{Bpri}\big \rangle$ passed, the resultant encrypted file is denoted as $F\,^*_2$ with a 2-level protection only.
\item Bob uses his first and second private-key strings $s^1_{Bpub}$ and $s^2_{Bpri}$ to decrypt the encrypted file $F\,^*_2$, after two Key-pair string authentications $A_{uth}\big \langle s^1_{Bpub}, s^1_{Bpri}\big \rangle$ and $A_{uth}\big \langle s^2_{Bpub}, s^2_{Bpri}\big \rangle$ passed, then Bob gets the original plaintext $F$.
\end{asparaenum}

\vskip 0.4cm

\textbf{Self-certification topology algorithm-III based on Key-graph sequences with the same ranks. }
\begin{asparaenum}[\textbf{\textrm{Self-3.}}1.]
\item Alice sends to Bob a key-package-A including Alice's public-key string $s_{Apub}$, and Alice's public-key graph sequence $\{G\,^i_{Apub}\}^m_{i=1}$ (as the identity authentication).
\item Bob sends to Alice a key-package-B including Bob's public-key string $s_{Bpub}$, and Bob's public-key graph sequence $\{G\,^j_{Bpub}\}^m_{j=1}$.
\item Alice encrypts the plaintext $F$ by Bob's public-key string $s_{Bpub}$, and gets the encrypted file $F_1$, and then Alice encrypts the encrypted file $F_1$ by her private-key graph sequence $\big \{G\,^i_{Apri}\big \}^m_{i=1}$ one $G\,^i_{Apri}$ by one $G\,^{i+1}_{Apri}$, the resultant encrypted file is denoted as $F\,^*_{m+1}$ with a $(m+1)$-level protection.
\item Bob, first, decrypts the encrypted file $F\,^*_{m+1}$ by Alice's public-key graph sequence $\{G\,^i_{Apub}\}^m_{i=1}$ one-by-one, as the topological signature authentication $A_{uth}\big \langle \{G\,^i_{Apub}\}^m_{i=1}, \{G\,^i_{Apri}\}^m_{i=1}\big \rangle$ passed, the resultant encrypted file is denoted as $F\,^*_1$ with a one-level protection only.
\item Bob uses his private-key string $s_{Bpri}$ to decrypt the encrypted file $F\,^*_1$, as the Key-pair string authentication $S_{tringB}=A_{uth}\langle s_{Bpub}, s_{Bpri}\rangle$ passed, then Bob gets the original plaintext $F$.
\end{asparaenum}

\vskip 0.4cm

For the security and computational complexity of the self-certification topology algorithm-III we have:
\begin{asparaenum}[\textbf{\textrm{Scta-3.}}1.]
\item Attacking Alice's public-key graph sequence $\{G\,^i_{Apub}\}^m_{i=1}$ and Bob's public-key graph sequence $\big \{G\,^j_{Bpub}\big \}^m_{j=1}$ will meet the problems of topological structures and colorings at $m$ times.
\item There is an order in encrypting the encrypted file $F_1$ by Alice's private-key graph sequence $\big \{G\,^i_{Apri}\big \}^m_{i=1}$, so the topological signature authentication $T\,^{i}_{sigA}=A_{uth}\big \langle G\,^i_{Apub}, G\,^i_{Apri}\big \rangle$ follows the topological signature authentication $T\,^{i+1}_{sigA}=A_{uth}\big \langle G\,^{i+1}_{Apub}, G\,^{i+1}_{Apri}\big \rangle$, one-by-one.
\end{asparaenum}

\vskip 0.4cm

\textbf{Self-certification topology algorithm-IV based on Key-graph sequences with different ranks. }
\begin{asparaenum}[\textbf{\textrm{Self-4.}}1.]
\item Alice sends to Bob a key-package-A including Alice's public-key string $s_{Apub}$, and Alice's public-key graph sequence $\big \{G\,^i_{Apub}\big \}^m_{i=1}$ with $m\geq 1$ (as the identity authentication).
\item Bob sends to Alice a key-package-B including Bob's public-key string $s_{Bpub}$, and Bob's public-key graph sequence $\big \{G\,^j_{Bpub}\big \}^n_{j=1}$ with $n\geq 1$.
\item Alice encrypts the plaintext $F$ by Bob's public-key string $s_{Bpub}$, and gets the encrypted file $F_1$, and then Alice encrypts the encrypted file $F_1$ by her private-key graph sequence $\big \{G\,^i_{Apri}\big \}^m_{i=1}$ one-by-one, the resultant encrypted file is denoted as $F\,^*_{m+1}$ with a $(m+1)$-level protection.

\item Alice encrypts the file $F\,^*_{m+1}$ by Bob's public-key graph sequence $\big \{G\,^j_{Bpub}\big \}^n_{j=1}$ one-by-one, the resultant encrypted file is denoted as $F^{*}_{m+n+1}$ with a $(m+n+1)$-level protection.
\item Bob, first, decrypts the encrypted file $F^{*}_{m+n+1}$ by Alice's public-key graph sequence $\big \{G\,^i_{Apub}\big \}^m_{i=1}$ one-by-one, as the topological signature authentication $$T_{sigA}=A_{uth}\langle \{G\,^i_{Apub}\}^m_{i=1}, \{G\,^i_{Apri}\}^m_{i=1}\rangle$$ passed, the resultant encrypted file is denoted as $F\,^*_{n+1}$ with a $(n+1)$-level protection.
\item Bob uses his private-key string $s_{Bpri}$ and private-key graph sequence $\big \{G\,^j_{Bpri}\big \}^n_{j=1}$ to decrypt the encrypted file $F\,^*_{n+1}$. As the following two authentications
$$S_{tringB}=A_{uth}\langle s_{Bpub}, s_{Bpri}\rangle\textrm{ and }T_{sigB}=A_{uth}\big \langle \{G\,^j_{Bpub}\}^n_{j=1}, \{G\,^j_{Bpri}\}^n_{j=1}\big \rangle
$$ passed, finally, Bob obtains the original plaintext $F$.
\end{asparaenum}
\begin{figure}[h]
\centering
\includegraphics[width=16.4cm]{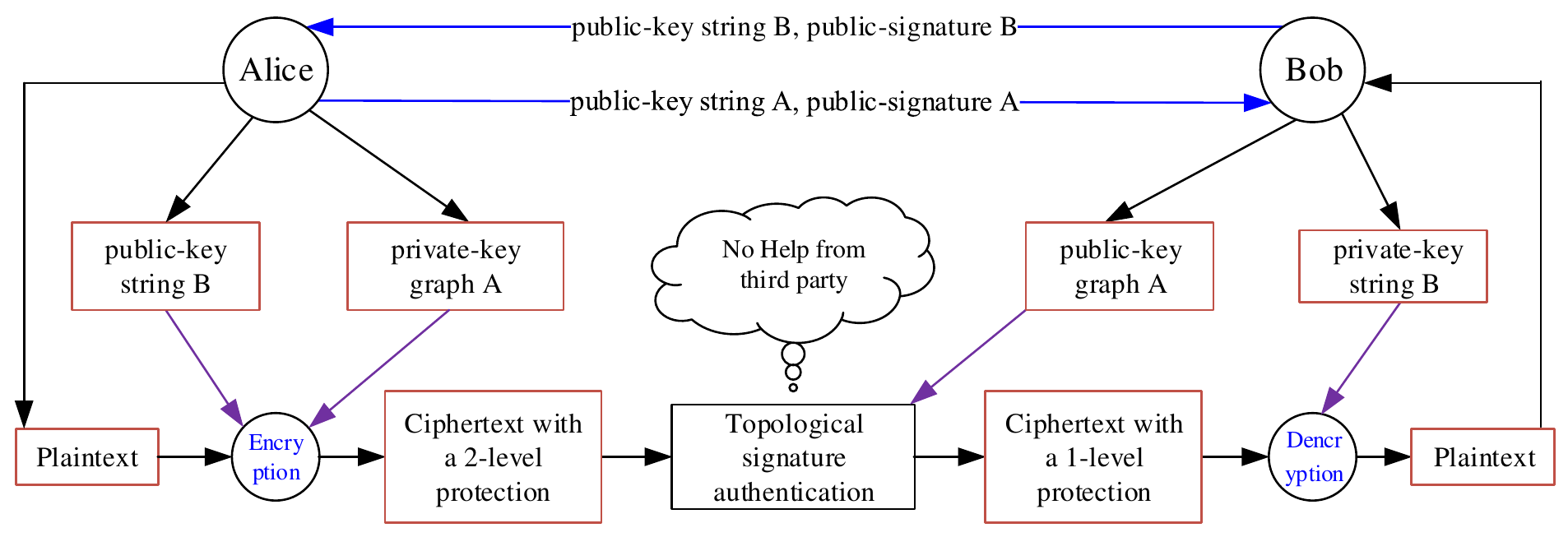}\\
\caption{\label{fig:signature-sequence}{\small A scheme for the self-certification algorithm based on the topological signature sequence.}}
\end{figure}

\textbf{Self-certification topology algorithm-V based on Key-string sequence and Key-graph sequence.}
\begin{asparaenum}[\textbf{\textrm{Self-5.}}1.]
\item Alice sends to Bob a key-package-A including Alice's public-key string sequence $\big \{s\,^i_{Apub}\big \}^{n_a}_{i=1}$ with $n_a\geq 1$, and Alice's public-key graph sequence $\big \{G\,^i_{Apub}\big \}^{m_a}_{i=1}$ with $m_a\geq 1$ (as the identity authentication).
\item Bob sends to Alice a key-package-B, which contains Bob's public-key string sequence $\big \{s\,^i_{Bpub}\big \}^{n_b}_{i=1}$ with $n_b\geq 1$, and Bob's public-key graph sequence $\big \{G\,^j_{Bpub}\big \}^{m_b}_{j=1}$ with $m_b\geq 1$.
\item Alice encrypts the plaintext $F$ by Bob's public-key string sequence $\big \{s\,^i_{Bpub}\big \}^{n_b}_{i=1}$, and gets the encrypted file $F_{n_b}$, and then Alice encrypts the encrypted file $F_{n_b}$ by her private-key graph sequence $\big \{G\,^i_{Apri}\big \}^{m_a}_{i=1}$ one-by-one, the resultant encrypted file is denoted as $F\,^*_{n_b+m_a}$ with a $(n_b+m_a)$-level protection.

\item Alice, again, encrypts the encrypted file $F\,^*_{n_b+m_a}$ by Bob's public-key graph sequence $\big \{G\,^j_{Bpub}\big \}^{m_b}_{j=1}$ one-by-one, the resultant encrypted file is denoted as $F^{*}_{n_b+m_a+m_b}$ with a $(n_b+m_a+m_b)$-level protection.

\item Bob, first, decrypts the encrypted file $F^{*}_{n_b+m_a+m_b}$ by Alice's public-key graph sequence $\big \{G\,^i_{Apub}\big \}^{m_a}_{i=1}$ one-by-one, as the topological signature authentication $$T_{sigA}=A_{uth}\big \langle \{G\,^i_{Apub}\}^{m_a}_{i=1}, \{G\,^i_{Apri}\}^{m_a}_{i=1}\big \rangle$$ passed, the resultant encrypted file is denoted as $F\,^*_{n_b+m_b}$ with a $(n_b+m_b)$-level protection.

\item Bob uses his private-key string sequence $\big \{s\,^i_{Bpri}\big \}^{n_b}_{i=1}$ and private-key graph sequence $\big \{G\,^j_{Bpri}\big \}^{m_b}_{j=1}$ to decrypt the encrypted file $F\,^*_{n_b+m_b}$. As the following two authentications
$$S_{tringB}=A_{uth}\big \langle \big \{s\,^i_{Bpub}\big \}^{n_b}_{i=1}, \big \{s\,^i_{Bpri}\big \}^{n_b}_{i=1}\big \rangle\textrm{ and }T_{sigB}=A_{uth}\big \langle \{G\,^j_{Bpub}\}^{m_b}_{j=1}, \{G\,^j_{Bpri}\}^{m_b}_{j=1}\big \rangle
$$ passed, finally, Bob obtains the original plaintext $F$.
\end{asparaenum}


\section{Conclusion}

In order to approach practical applications, resist the attack of quantum computing, and develop asymmetric topology cryptography, we have introduced a number of ways to induce various string as follows:
\begin{asparaenum}[\textbf{\textrm{Meth}}-1.]
\item Number-based strings made by topological coding are as public-key strings.
\item Topcode-matrices produce number-based strings for public-key strings.
\item Colored graphs have Topcode-matrices, which induce number-based strings.
\item Vectors and sets from topological coding make number-based strings.
\item Compound number-based strings made by topological coding are as public-key graphs.
\item Number-based strings as Key-pairs are made by the integer partition and the integer decomposition.
\item $[0,9]$-strings, and algebraic operations on $[0,9]$-strings, also, \emph{string algebra}.
\item Number-based super-strings.
\item Self-breeding number-based strings.
\item Graph-based strings including total graph-colorings; TOTAL-graph-coloring algorithm-I; TOTAL-graph-coloring algorithm-II; complex number-based strings.
\item Number-based strings with parameters,plane-curve-attached strings, space-coordinate strings, fractional strings.
\item A \emph{public-key graph} $G$ admitting a $W$-constraint parameterized coloring $F$ can be used to make a \emph{private-key graph} $H=\langle G,F,f(x)\rangle$ based on a real-valued function $y=f(x)$.
\item Assignment Topcode-matrices.
\item (abc)-linear colorings.
\item flawed $W$-constraint colorings and labelings.
\item Plane-coordinate string sequences.
\item Graphic lattices for making number-based strings including tree-base graphic lattices; tree-base $[0,9]$-operation graphic lattices.
\item Set-based strings.
\item Number-based strings from indexed-colorings including maximal planar graphic lattices; indexed-colorings; planar graphs and indexed-colorings.
\item Hanzi-based strings including planar graphs and Hanzi-based strings; Hanzi-Lattices; maximal planar graphs and Hanzi-based strings; the $[\odot_{\textrm{plan}}]$-Hanzi-graphic lattices.
\item Every-zero infinite graphic groups, every-zero infinite graphic group homomorphism, mutually equivalent infinite every-zero graphic groups.
\end{asparaenum}

\vskip 0.4cm

We try to apply various strings to the actual scene:
\begin{asparaenum}[\textbf{\textrm{Appl}}-1.]
\item Homogeneous number-based string-colorings, non-homogeneous number-based string-colorings, constructing number-based string-colorings.
\item Multi-level multi-rank strings.
\item Vector-colorings made by strings.
\item Set-colorings made by strings.
\item Topcode-matrices with non-number elements: Topcode-matrices with the elements of strings, vectors and sets; non-number-type Topcode-matrices with parameters.
\item For the overall encryption of networks, or the community encryption of different communities in a network, we have studied very-zero number-based string groups including GROUP-compound algorithm; Number-based sub-string groups; compound number-based string groups; Partial every-zero compound number-based string groups; and every-zero Topcode-matrix group.
\item Topcode-matrix algebra.
\end{asparaenum}

\vskip 0.4cm

There are problems related with strings and string-colorings as follows:
\begin{asparaenum}[\textbf{\textrm{Ques}}-1.]
\item An \emph{$[a,b]$-string} $s=c_1c_2\cdots c_n$ holds $c_j\in [a,b]$ with two integers $a,b$ subject to $0<a<b$. It is may be interesting to generalize the properties of $[0,9]$-strings to $[a,b]$-strings, and \textbf{find} some connections between $[a_k,b_k]$-strings with integers $0<a_k<b_k$.
\item Given a number-based string $s=c_1c_2\cdots c_n$ with $c_i\in [0,9]$ and $c_1\geq 1$, \textbf{find} a positive integer $m$, such that $m$ can be expressed as a sum $m_{i,1}+m_{i,2}+\cdots +m_{i,a_i}$ defined in Eq.(\ref{eqa:intege-m-partitioned-00}) of Definition \ref{defn:integer-partitioned-strings}, and there is $s=s^*_i$ defined in Eq.(\ref{eqa:intege-m-partitioned-11}), or $s=s^j_i$ defined in Eq.(\ref{eqa:intege-m-partitioned-22}), or $s=s_{i,j}$ defined in Eq.(\ref{eqa:intege-m-partitioned-33}).
\item Given a number-based string $s^*=a_1a_2\cdots a_m$ with $a_i\in [0,9]$ and $a_1\geq 1$, \textbf{find} a positive integer $n$, such that $n$ can be expressed as a product $p_{i,1} p_{i,2}\cdots p_{i,b_i}$ defined in Eq.(\ref{eqa:intege-n-decomposed-00}) of Definition \ref{defn:integer-decomposed-strings}, and there is $s^*=r^*_i$ defined in Eq.(\ref{eqa:intege-n-decomposed-11}), or $s^*=r^j_i$ defined in Eq.(\ref{eqa:intege-n-decomposed-22}), or $s^*=r_{i,j}$ defined in Eq.(\ref{eqa:intege-n-decomposed-33}).
\item Let $G_{raph}(p)$ be the set of graphs of $p$ edges. \textbf{Find} a mapping $\theta:V(G)\rightarrow V(H)$ for two graphs $G,H\in G_{raph}(p)$, such that there is a vertex $u\in V(G)$ holding
$G\wedge u\cong H\wedge \theta(u)$.
\item PRONBS-problem defined in Problem \ref{question:PRONBS-problems00} is the parametric reconstitution problem of number-based strings.
\item ATE-CGH-problem defined in Problem \ref{problem:ATE-CGH-problem00} and Pan-ATE-CGH-problem defined in Problem \ref{question:Pan-ATE-CGH-problem} and ATE-To-CGH-problem defined in Problem \ref{question:ATE-To-CGH-problem} are the problems of asymmetric topology encryption related with colored graph homomorphisms.
\end{asparaenum}

\vskip 0.4cm

The definition of ``\emph{zero knowledge proof}'' in the blockchain is that the verifier can make the verifier believe that a certain assertion is correct without providing any useful information to the verifier. \emph{Zero knowledge proof} is a more secure information verification or authentication mechanism. Security and privacy are the values of zero knowledge proof. We hope the techniques of asymmetric topology encryption, such as self-certification topology algorithms, can be applied into blockchains.

Assuming that attackers are very familiar with the technology of topology coding, how to generate Key-pairs and how to distribute Key-pairs without being cracked is one of the focuses of future research, although there is still a hard and challenging way to go from theory to practice.

\section*{Acknowledgment}
The author, \emph{Bing Yao}, was supported by the National Natural Science Foundation of China under grants No. 61163054, No. 61363060 and No. 61662066.

\newpage

\begin{figure}[h]
\centering
\includegraphics[width=16.4cm]{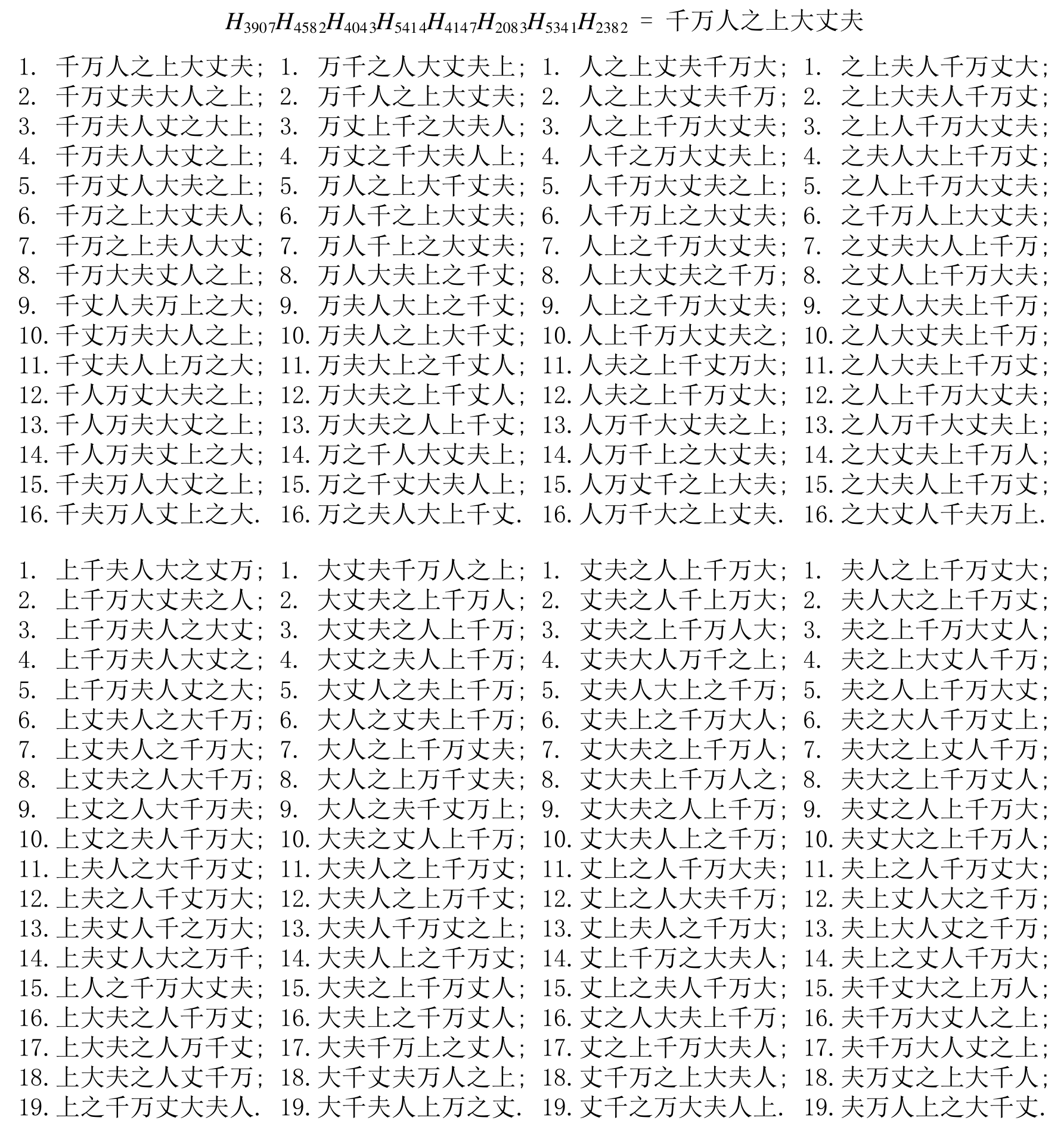}\\
\caption{\label{fig:Hanzi-combinators}{\small A Chinese sentence $H_{3907}H_{4582}H_{4043}H_{5414}H_{4147}H_{2083}H_{5341}H_{2382}$ provides $8!=40320$ different Chinese sentences in total.}}
\end{figure}

\end{document}